\documentclass{amsart}
\usepackage{amsfonts,amssymb}
\usepackage{rotating}
\newtheorem{theorem}{Theorem}[section]
\newtheorem{lemma}[theorem]{Lemma}
\newtheorem{corolla}[theorem]{Corollary}

\numberwithin{equation}{section}
\theoremstyle{definition}
\newtheorem{definition}[theorem]{Definition}
\newtheorem{example}[theorem]{Example}

\theoremstyle{remark}
\newtheorem{remark}[theorem]{Remark}

\newcommand{\re}{\mathop{\mathrm{Re}}}
\newcommand{\im}{\mathop{\mathrm{Im}}}

\begin{document}
\title[Bures Geometry over ${\mathsf C}^*$-algebras ]{Bures Geometry on  ${\mathsf C}^*$-algebraic State Spaces}
\author[P.M.\,Alberti]{Peter M. Alberti}
\address{
Doz. Dr. Peter M. Alberti\\
Topasstra{\ss}e 84, D-04319 Leipzig, Germany
}
\email{peter.alberti@web.de,\,peter.alberti@itp.uni-leipzig.de}
\date{}
\commby{}

\dedicatory{Dedicated to Prof.\,Dr.\,Armin Uhlmann on occasion of the $\it{90}^{\text{th}}$ birthday}
\begin{abstract}
The inner geometry induced by the Bures distance function on the state space of a unital  ${\mathsf{C}}^*$-algebra, or on parts of it, is considered in detail. As a key result the  local dilation function is calculated at a state when the state argument is bound to vary along certain parameterized curves passing through this state. The parameterized curves considered are those possessing differentiable local implementations as vector states relative to some unital  $^*$-representation in the vicinity of the state in question. The space of all tangent forms at a given state  correspondent to this category of curves is specified as normed linear space with a quadratic norm which is depending from the state in a characteristic manner. In analyzing the structure of the local dilation function special emphasis is laid on describing such parts of the state space, in restriction to which the  line element would be of Finslerian type, and which were geodesically convex subsets. All results obtained are illustrated by examples, mainly referring to normal states over the $vN$-algebra of all bounded linear operators on a separable Hilbert space. Especially, certain subsets of normal states corresponding to density operators of full support are considered and the implications obtained are compared to those known from the investigations of Uhlmann et al in the finite dimensional case.
\end{abstract}
\maketitle
\newpage
\tableofcontents
\newpage
\section{Basic settings}\label{bas}
\subsection{The problem}\label{prob}
In this paper the Bures metric on the state space ${\mathcal S}(M)$ of a unital ${\mathsf C}^*$-algebra $M$ is considered. This is the metric coming along with a certain distance $d_B(M|\nu,\varrho)$ between ${\mathsf C}^*$-algebraic states $\nu,\varrho\in {\mathcal S}(M)$ when considered infinitesimally in the vicinity of a state $\nu\in {\mathcal S}(M)$. In the special case if $M$ is a ${\mathsf W}^*$-algebra and the distance is restricted to the normal state space ${\mathcal S}_0(M)$, $d_B$ agrees with
the distance function of Bures \cite{Bure:69}. But the definition of Bures straightforwardly extends to a  ${\mathsf C}^*$-algebraic situation and therefore since the work of Uhlmann \cite{Uhlm:85} this extension is referred to as Bures distance. To prepare for the definition, consider a state $\mu\in {\mathcal S}(M)$ and a unital $^*$-representation $\{\pi,{\mathcal H}_\pi\}$ of $M$ on a Hilbert space ${\mathcal H}_\pi$ with scalar product $\langle\cdot,\cdot\rangle_\pi$. The $\pi$-fibre of $\mu$ is the set
\begin{equation}\label{faser}
{\mathcal S}_{\pi,M}(\mu)=\bigl\{\chi\in {\mathcal H}_\pi:\mu(x)=
\langle\pi(x)\chi,\chi\rangle_\pi,\,\forall x\in M\bigr\}.
\end{equation}
For a given  unital $^*$-representation say that `the $\pi$-fibre of $\mu$ exists' if ${\mathcal S}_{\pi,M}(\mu)$ is non-void.  Thus, the latter is the set of all those unit vectors of ${\mathcal H}_\pi$ by which the state $\mu$ can be implemented within the representation $\pi$. The Bures distance $d_B(M|\nu,\varrho)$
between two states $\nu,\varrho\in {\mathcal S}(M)$ then is defined as follows:
\begin{definition}\label{budi}
$\ \ \displaystyle d_B(M|\nu,\varrho)=\inf_{\varphi\in
{\mathcal S}_{\pi,M}(\nu),\psi\in {\mathcal S}_{\pi,M}(\varrho)}
\|\varphi-\psi\|_\pi$
\end{definition}
In the definition the infimum is to extend
over all  unital $^*$-representations $\pi$ relative to which both $\pi$-fibres exist and,
within each
such representation, $\varphi$ and $\psi$ then may be varied
through all of ${\mathcal S}_{\pi,M}(\nu)$ and ${\mathcal S}_{\pi,M}(\varrho)$, respectively.

Whereas about the Bures distance a lot of general results exist and examples for infinite dimensional $M$ have been thoroughly analyzed too, see e.g.~in \cite{AlUh:00.1,Pelt:00.1,Pelt:00.1a,AlPe:00.2} and the references cited therein, the Bures metric has been considered almost exclusively for finite dimensional cases, that is mainly on sets of $n\times n$-density matrices where nice characterizations reading in differential geometric terms exist \cite{Uhlm:91.2,Uhlm:92.1,Hueb:92,Uhlm:93,Hueb:93,Ditt:94,Uhlm:95,DitUhl:99}. These find various applications in the geometry of quantum states, see e.g.~in \cite{BenZyc:06} for an account on these concepts, and the more recent work \cite{Uhlm:11}. The aim of the present paper is to bridge this gap and to provide new results about the structure of the Bures metric in a general ${\mathsf C}^*$-algebraic context. Also, the known results in finite dimensions will be reclassified when seen in the light of the more general point of view favoured in this paper, and its significance for density operators over a separable Hilbert space exemplarily will be shown.
Clearly, if seen from an analytic-geometric point of view, to clarify the metric associated to some distance in the case at hand requires the structure of the expression
\begin{equation}\label{intro0}
    \limsup_{t\to 0} \frac{d_B(M|\nu_t,\nu)}{|t|}
\end{equation}
which is considered along sufficiently regular curves $\gamma$ in the state space, to be taken into the focus. Thereby, $$\gamma:I\ni t\longmapsto \nu_t\in {\mathcal S}(M)$$ can be any simple parametric $C^1$-curve passing through $\nu$ at $t=0$ ($I$ be some open interval of reals containing zero). That is, in the vicinity of $\nu$ we have $$\nu_t=\nu+t f +{\mathbf o}(t)$$ with tangent form $f$ at $\nu$, which is the hermitian linear form $f=\nu_t^{\,\prime}|_{t=0}$ over $M$ obtained by taking the derivative at $t=0$. Prior to giving an outline of the main results of this paper, some remarks relating conventions will be added now.
\subsection{Conventions}\label{bas1}
Start with some conventions and notations to which in the sequel we will adhere. In the following, for denoting operator algebras, the upper case latins $M,N,R$ will be used. Let $M$ be a unital ${\mathsf C}^*$-algebra and suppose $\{\pi,{\mathcal H}_\pi\}$ is a unital $^*$-representation of $M$ over the Hilbert space ${\mathcal H}_\pi$. That is, $M\ni x\longmapsto \pi(x)\in {\mathsf B}({\mathcal H}_\pi)$ is a $^*$-homomorphism taking the unit of $M$ into the unit operator over ${\mathcal H}_\pi$.  For Hilbert spaces the symbols ${\mathcal H}, {\mathcal K}$ and their derivates will be preferred. The scalar product ${\mathcal H}_\pi\times {\mathcal H}_\pi
\ni\{\chi,\eta\}\,\longmapsto\,\langle\chi,\eta\rangle_\pi\in
{\mathbb{C}}$ on the representation Hilbert space ${\mathcal H}_\pi$ is supposed
to be linear with respect to the first argument $\chi$, and antilinear
in the second argument $\eta$, and maps into the complex
field ${\mathbb{C}}$. Let ${\mathbb{C}}\ni z\mapsto \bar{z}$ be the complex conjugation,
and be $\Re{ z}$, $\Im{z}$, and $|z|$ the real part, imaginary part, and absolute value of $z$, respectively. The norm of $\chi\in
{\mathcal H}_\pi$ is $$\|\chi\|_\pi=\sqrt{\langle\chi,\chi\rangle_\pi}$$
and the norm closure of some subset ${\mathfrak M}\subset {\mathcal H_\pi}$ will be indicated as  $[{\mathfrak M}]$. For a subset ${\mathfrak N}$ of a linear structure like an algebra or an Hilbert space by ${\mathsf{Lin}}({\mathfrak N})$ the linear hull of the elements of ${\mathfrak N}$ within the respective linear structure will be denoted.
Whereas as symbols for vectors in Hilbert spaces the lower-case greek letters $\varphi,\phi,\psi,\chi,\xi,\eta$ and their derivates will be used, for elements of the ${\mathsf C}^*$-algebra $M$ the lower-case latins $a,b,c,p,q,u,v,w,x,y,z$ and their derivates will occur.  For the ${\mathsf C}^*$-norm of an element
$x\in M$ as well as for the operator norm of a concrete bounded linear
operator $x\in {\mathsf B}({\mathcal H}_\pi)$ the same notation
$\|x\|$ will be used, and the involution ($^*$-operation),
respectively the taking of the hermitian conjugate of an element $x$ are
indicated by the transition $x\,\longmapsto\,x^*$. By
$M_{\mathsf{h}}$ and $M_+$ the hermitian and positive elements of $M$ is meant,
respectively. The null
and the unit element/operator in $M$ and ${\mathsf B}({\mathcal
K})$, with ${\mathcal K}$ being a Hilbert space, will be denoted by ${\mathsf 0}$ and ${\mathsf 1}$. A finite sequence $\{x\}=\left\{x_1,x_2,\ldots,x_n\right\}$, $n\in {{\mathbb{N}}}$, of elements of $M$ is termed finite positive decomposition of the unity of $M$ if $x_k\in M_+$ for all $k\leq n$ and $\sum_k x_k={\mathsf 1}$.
Incidentally, relating the operator and ${\mathsf C}^*$-algebraic context, the reader is referred to the standard
monographs,
e.g.~\cite{Dixm:64,Saka:71,KaRi:83}.

For notational purposes mainly, in short now recall some
fundamental facts relating (bounded) linear forms over ${\mathsf C}^*$-algebras which subsequently
might be of concern in context of Definition \ref{budi}. The lower-case latins $f,g,h$ will be in use as symbols for bounded linear forms. Recall that the
topological dual space $M^*$ of $M$ is the
set of all those linear functionals (linear forms) $f$ which are continuous with respect to the
operator norm topology. Equipped with the dual norm
$\|\cdot\|_1$, which is
given by $\|f\|_1=\sup\{|f(x)|\,:\,x\in M,\,\|x\|\leq 1\}$
and which is referred to as the functional norm, $M^*$ is a Banach space. If $R$ is a ${\mathsf C}^*$-subalgebra of $M$, then the restriction of $f\in M^*$ to $R$ is denoted $f|R$.
For each given $f\in M^*$, the hermitian
conjugate functional $f^*\in {M}^*$ is defined by $$f^*(x)=\overline{f(x^*)}$$ for
each $x\in M$. Remind that $f\in M^*$ is hermitian
if $f=f^*$ holds, and $f$ is termed positive if $f(x)\geq 0$ holds,
for each $x\in M_+$. The cone of all bounded positive linear forms over $M$ is denoted $M_+^*$.  Relating positivity, remind that
a bounded linear form over $M$ is positive
if, and only if, $\|f\|_1=f({\mathsf 1})$. The state space ${\mathcal S}(M)$ over $M$ then is the convex (in the usual (affine) sense) closed subset  of all linear forms of $M^*$ which are positive and normalized to one. Hence, by the previous, ${\mathcal S}(M)=\{f\in M_+^*\,:\,f({{\mathsf{1}}})=1\}$. For states we use the lower-case greek symbols $\mu,\nu,\omega,\varrho$ and their derivates. For each $\omega\in {\mathcal S}(M)$, there exists a cyclic $^*$-representation
$\pi$ of $M$ on some Hilbert space ${\mathcal H}_\pi$, with cyclic vector
$\xi\in {\mathcal H}_\pi$, and obeying $\omega(x)=\langle \pi(x)\xi,
\xi\rangle_\pi$, for all $x\in M$ (Gelfand-Naimark-Segal theorem). Suppose now a finite set of states, $\{\nu,\varrho,\ldots,\mu\} \subset {\mathcal S}(M)$, with cardinality $n\in {{\mathbb{N}}}$.
Then, considering that construction in case
with the state $\omega=(\nu+\varrho+\cdots +\mu)/n$ will provide a unital $^*$-representation $\pi$
such that the $\pi$-fibre of each of the states $\nu,\varrho,\ldots,\mu$ exists (omit the details which are standard). Thus e.g., the Definition
\ref{budi} makes sense for any two states.

Relating the usage of such mere differential geometric notions as `curve', `tangent form', `tangent space', `(differential) metric',  and `inner geometry', `inner metric' and `metric convexity' of subsets, `geodesic' and `geodesic convexity' of subsets in context of ${\mathsf C}^*$-algebraic state spaces, and their derivates like `strata' (which are certain convex in the usual (affine) sense subsets of ${\mathcal S}(M)$, see below), in many respects we have chosen to go back to the roots of these intuitive geometric notions when considered on metric spaces with an additional underlying affine structure, see e.g. in  \cite{Whit:32,Whit:35,Buse:55}, and especially \cite{Rino:61,Rino:61a} which explicitly will be referred to. In fact, when seen under the Bures distance, these structures in general will be infinite dimensional, convex (in the usual (affine) sense) metric structures. Also, since in its operational
aspects these notions at least should allow for the minimal needs of an analysis of e.g.~dynamical systems living on state spaces,  the analogy to some ideas and notions from Riemannian geometry and theory of Finsler spaces is highly manifest \cite{Grom:81,Chav:93,AnMa:93} and will be used freely. Finally remark that in the course of the derivation of the results access will be made to some auxiliary facts from theory of real functions, operator theory and functional analysis. These for convenience are sourced out into seven appendices. Thus, when referring to `\eqref{in}' or `Theorem \ref{enddiffex}' e.g., reference to the correspondently named item or result appearing in Appendix \ref{app_a} or Appendix \ref{app_d}, respectively, will be meant.

\subsection{An outline of the main results}\label{mainres}
Some of the results later on derived in this paper will be outlined now. The starting point will be the observation that the following fundamental estimate has to hold
\begin{subequations}\label{est}
\begin{equation}\label{intro1a}
  \liminf_{t\to 0} \frac{d_B(M|\nu_t,\nu)}{|t|}\geq \| f \|_\nu
\end{equation}
and where  $\|\cdot\|_\nu$ is a quadratic norm over the tangent forms at $\nu$ which reads as
\begin{equation}\label{norm}
 \|f\|_\nu=\sup_{\{x\}}{\frac{1}{2}}\,\sqrt{{\sum_{j}}^\prime
\frac{f(x_j)^2}{\nu(x_j)\phantom{^2}}}
\end{equation}
with $\{x\}=\bigl\{x_1,x_2,\ldots,x_n\bigr\}$, $n\in {{\mathbb{N}}}$, running through the set of all finite positive decompositions of the
unity within $M$ (the $^\prime$ at the summation symbol indicates that
only terms with non-vanishing denominators $\nu(x_j)$ have to be summed up). From \eqref{intro1a} obviously the following estimate of the expression \eqref{intro0} can be inferred
\begin{equation}\label{intro1}
  \limsup_{t\to 0} \frac{d_B(M|\nu_t,\nu)}{|t|}\geq \| f \|_\nu
\end{equation}
\end{subequations}
Thereby, on analyzing the expression \eqref{intro0} on special non-commutative ${\mathsf C}^*$-algebras (e.g. full matrix algebras have been considered) and/or with special examples of curves $\gamma$, which are accessible to concrete calculation, we found it quite hard to construct examples which were deviating from the special case
\begin{subequations}\label{intro2}
\begin{equation}\label{intro2a}
  \lim_{t\to 0} \frac{d_B(M|\nu_t,\nu)}{|t|}= \| f  \|_\nu
\end{equation}
of \eqref{intro1}. But such examples are shown to exist and therefore equality occuring in \eqref{intro1} and existence of the strict limit in \eqref{intro0} cannot be the general case. Thus the question arises how the deviation of \eqref{intro1} from the Finslerian law \eqref{intro2a} in the general case could be quantified. As will be shown in this paper, at least as far as only $C^1$-implementable curves for $\gamma$ are admitted for use in \eqref{intro0}, the structure of \eqref{intro0} and the quantitative deviation of \eqref{intro1} from \eqref{intro2a} can be completely clarified. Also remark that in case it is satisfied, owing to $f=\nu_t^{\,\prime}|_{t=0}$ the interpretation of formula  \eqref{intro2a} usually reads in terms of the (infinitesimal) Bures length element ${\mathsf{d}}s$ at $\nu$ along the curve $\gamma$ in quest, in we then equivalently have
\begin{equation}\label{intro2b}
 {\mathsf{d}}s= \| {\mathsf{d}}\nu  \|_\nu
\end{equation}
\end{subequations}
When fulfilled in sufficient generality on suitably chosen subsets of states, this law will yield there a direct analog to the first fundamental form of differential geometry.

To explain these results, agree that a simple parametric $C^1$-curve $\gamma$ is called `locally implementable (around $\nu$)' if it can be implemented by a simple parametric $C^1$-curve $I_\pi \ni t\longmapsto \varphi_t\in {\mathcal S}_{\pi,M}(\nu_t)\subset {\mathcal H}_\pi$, where $I_\pi\subset I$ is an open interval containing $0$ and  $\{\pi,{\mathcal H}_\pi\}$ is a unital $^*$-representation with respect to which such an  implementation $(\varphi_t)$ exists. To such a  $^*$-representation  $\{\pi,{\mathcal H}_\pi\}$ will be  referred to as being `$\gamma$-compliant around $\nu$'. Thus a simple $C^1$-curve $\gamma$ passing through a state $\nu$ belongs to the class mentioned if it admits a $^*$-representation  $\{\pi,{\mathcal H}_\pi\}$ which is $\gamma$-compliant around $\nu$.

Now, suppose $\{\pi,{\mathcal H}_\pi\}$ is $\gamma$-compliant around $\nu$, with differentiable local implementation $(\varphi_t)$ around $\nu$.
Then, with $\varphi=\varphi_0\in {\mathcal S}_{\pi,M}(\nu)$, one has $\varphi_t=\varphi+
t\phi+{\mathbf o}(t)$, with the vector $\phi=(\varphi_t)^{\,\prime}|_{t=0}\in {\mathcal H}_\pi$ which will be called `tangent vector of the local implementation of $\gamma$ at $\nu$'. It depends on the $\gamma$-compliant representation $\{\pi,{\mathcal H}_\pi\}$ and the implementation $(\varphi_t)$ chosen and which both might be highly ambiguous data for one and the same curve $\gamma$ and state $\nu$. All these vectors $\phi=\phi(\pi,{\mathcal H}_\pi,(\varphi_t))$ will be collected into the class ${\mathsf{T}}_\nu(M|\gamma)$ of tangent vectors of $\gamma$ at $\nu$. Clearly, the tangent form $f$ of $\gamma$ at $\nu$ is uniquely determined by each $\phi\in{\mathsf{T}}_\nu(M|\gamma)$, for if $\phi=\phi(\pi,{\mathcal H}_\pi,(\varphi_t))$, then $f$ can be reconstructed via $\pi$ from $\phi$ and $\varphi$. Thereby, the set ${\mathsf{T}}_\nu(M)$ of tangent forms that can occur at $\nu$ in context of the above category of curves can be given a characterization reading in terms of the expression \eqref{norm}
\begin{equation}\label{ts}
  {\mathsf{T}}_\nu(M)=\bigl\{f\in M^*_{\mathsf h}:\,f({\mathsf 1})=0,\,\|f \|_\nu<\infty\bigr\}
\end{equation}
with $M^*_{\mathsf h}$ being the set of all bounded  hermitian linear forms over $M$. For reasons of simplicity here the ${\mathsf 0}$-form is included into the set of all tangent forms at $\nu$.

As will be shown later, the assertion of \eqref{ts} is that a given bounded hermitian linear form $g\not={\mathsf 0}$ on $M$ with $g({\mathsf 1})=0$ can occur as a tangent form at $\nu$ if, and only if, there is a unital $^*$-representation   $\{\pi,{\mathcal H}_\pi\}$ with non-trivial $\pi$-fibre ${\mathcal S}_{\pi,M}(\nu)$ such that, to given $\varphi\in {\mathcal S}_{\pi,M}(\nu)$ there exists $\psi\in {\mathcal H}_\pi$ such that, for any $x\in M$,
\begin{subequations}\label{optprob}
\begin{equation}\label{optprob0}
 g(x)=\langle\pi(x)\psi,\varphi\rangle_\pi+\langle\pi(x)\varphi,\psi\rangle_\pi
\end{equation}
is fulfilled. In line of the proof an alternative to  \eqref{norm} method for evaluating $\|g\|_\nu$ is established and which reads in terms of the data coming along with \eqref{optprob0}.

To explain this, consider ${\mathcal H}_\pi$ as Euclidean space with inner product $\langle\xi,\eta\rangle_{\pi,\mathbb R}=\Re \langle\xi,\eta\rangle_\pi$, and be $[\pi(M)_{\mathsf h}\varphi]$ the closed real linear subspace generated by the vectors $\pi(x)\varphi$, with $x=x^*$, $x\in M$. Let $\hat{\psi}_0\in [\pi(M)_{\mathsf h}\varphi]$ be the (unique) vector obeying
\begin{equation}\label{optprob1}
\|\psi-\hat{\psi}_0\|_\pi=\inf_{\xi\in[\pi(M)_{\mathsf h}\varphi] } \|\psi-\xi\|_\pi
\end{equation}
that is, $\hat{\psi}_0$ is the best approximation of $\psi$ within $ [\pi(M)_{\mathsf h}\varphi]$ in sense of the real Hilbert space structures considered. Incidentally, remark that this is equivalent to the claim that $\xi=\hat{\psi}_0$ be (the unique) solution of the problem
\begin{equation}\label{optprob2}
\xi\in [\pi(M)_{\mathsf h}\varphi]:\  g(x)=\langle\pi(x)\xi,\varphi\rangle_\pi+\langle\pi(x)\varphi,\xi\rangle_\pi,\,\forall x\in M
\end{equation}
\end{subequations}
By means of $\hat{\psi}_0$ the announced evaluation method for $\|g\|_\nu$ then reads
\begin{equation}\label{intro3c}
    \|g\|_\nu=\|\hat{\psi}_0\|_\pi
\end{equation}

Relating the latter  method, it is important to note that if  $\{\pi,{\mathcal H}_\pi\}$ exists  such that \eqref{optprob} and \eqref{intro3c} hold, then such a representation schema of $g$ will exist in respect of each other $^*$-representation $\{\tilde{\pi},{\mathcal H}_{\tilde{\pi}}\}$ too,  provided the $\tilde{\pi}$-fibre of $\nu$ is  non-trivial. That is, depending from the additional information about $\nu$ we have,  the choice of the $^*$-representation which is used in \eqref{intro3c} can be adapted to the needs of an easy determination of the (uniquely determined) solution of problem \eqref{optprob2} as well as to the needs of a potentially existing, comfortable mode of   calculating $\|g\|_\nu$ by formula \eqref{intro3c}.
An often occurring special case of the previous arises if $M$ is a $vN$-algebra acting on a Hilbert space ${\mathcal H}$, and $\nu$ is a vector state on $M$. Thus, $\nu(x)=\langle x\varphi,\varphi\rangle$ holds, for all $x\in M$, with $\varphi\in {\mathcal H}$. Things even simplify: $g\in {\mathsf{T}}_\nu(M)$ then happens if, and only if,
\begin{subequations}\label{optprobvN0}
\begin{equation}\label{optprobvN}
g(x)=\langle x \xi,\varphi\rangle + \langle x\varphi,\xi\rangle,\ \forall x\in M,
\end{equation}
is fulfilled, for some vector $\xi\in [M_{\mathsf h}\varphi]$. Thereby, $\xi$ with this property is uniquely determined, and as a special case of \eqref{intro3c} we infer that
\begin{equation}\label{optprobvN1}
\|g\|_\nu=\|\xi\|
\end{equation}
\end{subequations}

With these preparations in mind, now we are going to state the main results which will be derived and proved subsequently in this paper.
Relating \eqref{intro0}, as a key result the following Pythagorean law will be proved to hold:
\begin{subequations}\label{main}
\begin{equation}\label{intro3}
    \biggl(\limsup_{t\to 0} \frac{d_B(M|\nu_t,\nu)}{|t|}\biggr)^2=\| f \|_\nu^2+\|p_\pi(\varphi)^\perp p_\pi^{\,\prime}(\varphi)^\perp\phi\|_\pi^2
\end{equation}
Note that the formula holds with respect to each $\gamma$-compliant $^*$-representation $\{\pi,{\mathcal H}_\pi\}$ and corresponding local implementation $t\mapsto \varphi_t$ of $\gamma$ around $\nu$, i.e.~with tangent vector $\phi\in{\mathsf{T}}_\nu(M|\gamma)$  and tangent form $f\in {\mathsf{T}}_\nu(M)$ of $\gamma$ at $\nu$ with $f(\cdot)=\langle\pi(\cdot)\phi,\varphi\rangle_\pi+\langle\pi(\cdot)\varphi,\phi\rangle_\pi$,
where $\varphi=\varphi_0$. If $\pi(M)^{\,\prime}$ is the commutant $vN$-algebra of the concret ${\mathsf C}^*$-algebra $\pi(M)$ acting on ${\mathcal H}_\pi$, then $p_\pi^{\,\prime}(\varphi)$ and $p_\pi(\varphi)$  are the orthoprojections within ${\mathsf B}({\mathcal H}_\pi)$ projecting onto the norm closures $[\pi(M)\varphi]$ and $[\pi(M)^{\,\prime}\varphi]$  of the linear subspaces $\pi(M)\varphi$ and $\pi(M)^{\,\prime}\varphi$,  respectively, in the Hilbert space ${\mathcal H}_\pi$. It is known that $p_\pi^{\,\prime}(\varphi)\in \pi(M)^{\,\prime}$,  $p_\pi(\varphi)\in \pi(M)^{\,\prime\prime}$ (double commutant $vN$-algebra of $\pi(M)$ over ${\mathcal H}_\pi$).

The proof of \eqref{intro3}, which subsequently will be presented, at the same footing provides an equivalent characterization reading in terms of the tangent vectors arising from local implementations of $\gamma$ around $\nu$, in there we have
\begin{equation}\label{intro3a}
   \limsup_{t\to 0} \frac{d_B(M|\nu_t,\nu)}{|t|}=\inf_{\phi\in{\mathsf{T}}_\nu(M|\gamma)}\|\phi\|
\end{equation}
\end{subequations}
Thereby, one has to keep in mind that the  right hand side expression of this formula has to be seen as a convenient shortcut notation of the expression
\[
\inf_{\phi(\pi,{\mathcal H}_\pi,(\varphi_t))\in{\mathsf{T}}_\nu(M|\gamma)}\|\phi(\pi,{\mathcal H}_\pi,(\varphi_t))\|_\pi
\]
where $\pi$ is thought to extend over all $\gamma$-compliant representation $\{\pi,{\mathcal H}_\pi\}$ around $\nu$ and which the tangent vector $\phi$ is referring to while it is running through ${\mathsf{T}}_\nu(M|\gamma)$.

On the other hand, formula \eqref{intro3c}  obviously can be applied with $f$ and $\hat{\phi}_0$ instead of $g$ and $\hat{\psi}_0$ with respect to an arbitrarily given, fixed  $\gamma$-compliant representation $\{\pi,{\mathcal H}_\pi\}$, and then both the formulae of \eqref{main} can be  summarized into one relation
\begin{equation}\label{intro3e}
    \inf_{\phi\in{\mathsf{T}}_\nu(M|\gamma)}\|\phi\|^2=\|\hat{\phi}_0\|_\pi^2+\| p_\pi(\varphi)^\perp p_\pi^{\,\prime}(\varphi)^\perp\phi\|_\pi^2
\end{equation}
and where $\hat{\phi}_0$ is the best Euclidean approximation of the tangent vector $\phi$ in $[\pi(M)_{\mathsf h}\varphi]$.
By this relation the Bures geometry on ${\mathsf C}^*$-algebraic state spaces is governed, and the global infimum on the left hand side can be calculated locally within each of the $\gamma$-compliant $^*$-representations and the associated local implementations of $\gamma$ around $\nu$.

Remark that the term $\| p_\pi(\varphi)^\perp p_\pi^{\,\prime}(\varphi)^\perp\phi\|_\pi^2$ occurring on the right hand side of \eqref{intro3} and \eqref{intro3e} is an invariant which is independent from the tangent form $f$ and measures how strongly the Bures geometry on ${\mathcal S}(M)$ within the equivalence class of curves with tangent form $f$ locally at $\nu$ might be deviating from the Finslerian law \eqref{intro2}. Note that as an important consequence of \eqref{est} from \eqref{intro3} the validity of the equivalence
\begin{equation}\label{intro4}
  \lim_{t\to 0} \frac{d_B(M|\nu_t,\nu)}{|t|}= \| f \|_\nu\ \ \Longleftrightarrow\ \ p_\pi(\varphi)^\perp p_\pi^{\,\prime}(\varphi)^\perp\phi=0
\end{equation}
can be followed. Thereby, by the very form of \eqref{intro3}, one infers that $p_\pi(\varphi)^\perp p_\pi^{\,\prime}(\varphi)^\perp\phi$ is vanishing for a $\phi$ if, and only if, it is vanishing for any  $\phi\in{\mathsf{T}}_\nu(M|\gamma)$. Thus, in line with this, we have a necessary and sufficient condition for  the Finslerian law \eqref{intro2} to hold.

As an example we are going to indicate how the validity of \eqref{intro2} for $^*$-automorphic generated $\gamma$ and passing through $\nu$ readily can be seen. In fact, in this case one has $\gamma=(\nu_t)$, with $\nu_t=\nu\circ\alpha_t$ and inner $^*$-automorphism group $(\alpha_t)$. By standard facts $\gamma$ then is implementable with respect to some unital $^*$-representation $\{\pi,{\mathcal H}_\pi\}$ as $\varphi_t=u_t\varphi$, for $\varphi\in {\mathcal S}_{\pi,M}(\nu)$ and with some norm continuous group of unitary operators $u_t$ and belonging to the envelopping $vN$-algebra $\pi(M)^{\,\prime\prime}$. Thus, by the von-Neumann density theorem, $\varphi_t\in [\pi(M) \varphi]$ holds, for all parameters $t$. Especially from the latter then also $\phi=\varphi_t^{\,\prime}\in [\pi(M) \varphi]$ will hold, and therefore $p_\pi^{\,\prime}(\varphi)^\perp\phi=0$ must be fulfilled. In view of \eqref{intro4} then \eqref{intro2} follows.

Another line of investigation of this paper will be to identify non-trivial subsets $\Omega\subset {\mathcal S}(M)$ such that the Finslerian law \eqref{intro2} were fulfilled at each $\nu\in \Omega$ for any implementable curve $\gamma$ specified as above and passing through $\nu$ and evolving completely within $\Omega$. As has been previously mentioned, these then are examples where formula \eqref{intro2b} can be considered to  play a r{\^o}le analogous to the first fundamental form.  To construct a class of examples for $\Omega$, to given state $\mu$ let ${\mathfrak F}(\mu)$ be the set of all uniformly bounded nets $\bigl(x_\alpha\bigr)$ of positive semi-definite elements of $M$ such that $\lim_\alpha \mu(x_\alpha)=0$. For fixed $\mu\in {\mathcal S}(M)$, let
$\Omega(\mu)=\{\omega\in {\mathcal S}(M): {\mathfrak F}(\omega)={\mathfrak F}(\mu)\}$.
The set $\Omega(\mu)$ proves to be an affinely convex subset of ${\mathcal S}(M)$ and subsequently will be referred to as $\mu$-stratum.
As subsequently will be proved, in addition each stratum $\Omega$ possesses the property that, for $\nu\in \Omega$ and implementable curve $\gamma$ passing through $\nu$ and evolving in $\Omega$ exclusively, by a tangent vector $\phi=\phi(\pi,{\mathcal H}_\pi,(\varphi_t))$ necessarily $p_\pi(\varphi)\phi=\phi$ will be fulfilled, with $\varphi=\varphi_0$. Thus, in restricting the Bures metric on a stratum yields that the Finslerian law \eqref{intro2} will be   fulfilled there.  Obviously, the state space decomposes in a unique way as a foliation into leaves
\[\textstyle {\mathcal S}(M)=\bigcup_{i\in \Lambda} \Omega(\mu_i) \]
which are made of mutually disjoint strata $\Omega(\mu_i)$, with an indexed  family $(\mu_i)_{i\in \Lambda}$.
Thus, under the Bures metric, a  stratum resembles some generalized Finsler space.

Another problem arising in this context is to identify the lines of shortest Bures length and which were  connecting two given states $\nu$ and $\varrho$ (either in respect of the whole state space or in restriction to some subset, like a stratum e.g.), and to ask to which extent shortest lines can be uniquely determined. In respect of the whole state space, this problem can be resolved in surprising generality. As it comes out, in infinite dimensional,  non-commutative cases antipodal states normally exist (i.e. more than one shortest line connecting the states in question exist).  If considered in restriction to subsets of states, the criteria obtained are not quite easy to handle. Then, the aim will be to identify those parts of the state space which prove to be geodesically convex. By this notion it is meant that for any two states of such a set there exists (up to the parameterization)  exactly one line of shortest Bures length which is connecting the states considered and thereby is completely evolving in the set in question. For instance, in infinite dimensional non-commutative cases of $M$ and depending from the state $\mu$ the stratum $\Omega(\mu)$ generated by $\mu$ may or may not possess the property of the geodesic convexity. In line with this, it will be a good strategy to look for subsets of strata which are geodesically convex. Clearly, these distinguished subsets of states then are examples of structures where the Finslerian law holds as well as  existence and uniqueness for shortest geodesic lines connecting any two of its points can be granted and thus are structures closely resembling e.g. to  Riemannian manifolds.

As an illustration how the facts previously announced might go   together, the category of full algebras of bounded linear operators on separable Hilbert spaces will be chosen.
To start with, suppose ${\mathcal H}$ to be a separable Hilbert space, and let $\mu$ be a faithful normal state on the ${\mathsf W}^*$-algebra $M={\mathsf B}({\mathcal H})$. Then, the $\mu$-stratum $\Omega(\mu)$ comprises the set of all faithful normal states ${\mathcal S}^{\mathsf{faithful}}_0({\mathsf B}({\mathcal H}))$. As a consequence of the  fact that the  Finslerian law \eqref{intro2} holds on a stratum, on the set of all faithful normal states 
the Bures metric then can be based on the Bures line element ${\mathsf{d}}s$ at $\nu$, and which is  given as
\[
{\mathsf{d}}s=\| {\mathsf{d}}\nu \|_{\nu}
\]
for ${\mathsf{d}}\nu\in {\mathsf{T}}_\nu({\mathsf B}({\mathcal H}))$.
Now, since each faithful normal state $\nu$ via the trace formula
\begin{equation}\label{bi1}
    \nu(x)={{\mathsf{tr}}}\,\varrho x,\,\forall x\in {\mathsf B}({\mathcal H})
\end{equation}
in an affine and one-to-one manner corresponds to a density operator $\varrho$ of full support, the problem can be translated in terms of the affinely convex set of all density operators of full support and of (implementable) $C^1$-curves
$\gamma: I\ni t\longmapsto \varrho_t$ evolving completely in this set and passing through a given density operator $\varrho=\varrho_0$ there. As mentioned above, the Finslerian law \eqref{intro2} is  satisfied and in the context at hand then reads
\begin{subequations}\label{bi2a}
\begin{equation}\label{bi2}
     \lim_{t\to 0} \frac{d_B(\varrho_t,\varrho)}{|t|}= \| \varrho^{\,\prime} \|_\varrho
\end{equation}
with the derivative $\varrho^{\,\prime}=\varrho_t^{\,\prime}|_{t=0}$, and where the expression on the right hand side arises from \eqref{norm} through substituting $f(x_j)$ and $\nu(x_j)$ with ${{\mathsf{tr}}}\,\varrho^{\,\prime} x_j$ and ${{\mathsf{tr}}}\,\varrho x_j$, at each $j$. Equivalent to and instead of \eqref{bi2} the infinitesimal version reading as
\begin{equation}\label{bi2b}
    {\mathsf{d}}s= \| {\mathsf{d}}\varrho \|_\varrho
\end{equation}
\end{subequations}
and representing the Bures length element ${\mathsf{d}}s$ at $\varrho$ in the set of density operators of full support and reading in terms of the differential tangent form ${\mathsf{d}}\varrho$  will be in use. The exercise to be processed now is  chosing a unital $^*$-representation in respect of which the fibre of the state $\nu$ is non-trivial and such that the method of calculating the expression $\| {\mathsf{d}}\varrho \|_\varrho$  based on \eqref{optprob}
and formula \eqref{intro3c} gets  applicable. To this sake, consider now the separable Hilbert space ${\mathsf{ H.S.}}({\mathcal H})$ of all Hilbert-Schmidt operators over ${\mathcal H}$, with scalar product $\langle \xi,\eta\rangle_{{\mathsf{ H.S.}}}={{\mathsf{tr}}}\, \xi\eta^* $, for each $\xi,\eta\in {\mathsf{ H.S.}}({\mathcal H})$. Then, the operators of ${\mathsf B}({\mathcal H})$ isomorphically can be $^*$-represented as acting by left multiplication on Hilbert-Schmidt operators. In this context then each normal state \eqref{bi1} appears as a vector state
$$ \nu(x)={{\mathsf{tr}}}\,x\sqrt{\varrho}\sqrt{\varrho}^{\,*} ,\,\forall x\in {\mathsf B}({\mathcal H})
$$
This $^*$-representation meets all the requirements and will be used. Since $\varrho^{\,\prime}$ corresponds to a tangent form at $\nu$, by \eqref{optprobvN} unique ${\mathsf{d}}\xi\in \bigl[{\mathsf B}({\mathcal H})_{\mathsf h}\sqrt{\varrho}\,\bigr]_{{\mathsf{H.S.}}}$ exists with
\begin{equation}\label{bi3a}
{{\mathsf{tr}}}\,{\mathsf{d}}\varrho\, x={{\mathsf{tr}}}\,x\,{\mathsf{d}}\xi \sqrt{\varrho}+{{\mathsf{tr}}}\,x\sqrt{\varrho}\,{\mathsf{d}}\xi^*,\,\forall x\in {\mathsf B}({\mathcal H})
\end{equation}
This equation is equivalent to the (uniquely solvable) operator identity
\begin{subequations}\label{bi}
\begin{equation}\label{bi3b}
{\mathsf{d}}\varrho ={\mathsf{d}}\xi \sqrt{\varrho}+\sqrt{\varrho}\,{\mathsf{d}}\xi^*
\end{equation}
In view of \eqref{optprobvN1}, \eqref{bi2b}, by means of ${\mathsf{d}}\xi$ the square of the length element ${\mathsf{d}}s$ reads
\begin{equation}\label{bi3c}
 {\mathsf{d}}s^2=\| {\mathsf{d}}\varrho \|_\varrho^2=\|{\mathsf{d}}\xi\|_{{\mathsf{ H.S.}}}^2={{\mathsf{tr}}}\,{\mathsf{d}}\xi {\mathsf{d}}\xi^*
\end{equation}
\end{subequations}
and which relation holds on the set of all density operators of full support ${\mathcal S}^{\mathsf{faithful}}_0({\mathsf B}({\mathcal H}))$.
The differential solution ${\mathsf{d}}\xi$ of \eqref{bi3b} will be given a more  explicit form in case of some exposed category of submanifolds of density operators of full support.\\
To this sake, for $\mu\in {\mathcal S}^{\mathsf{faithful}}_0({\mathsf B}({\mathcal H}))$, consider  the affinely convex subset given as
\begin{equation}\label{introleaf1a}
{\mathcal S}_0(\mu)=\bigl\{\varrho\in {\mathcal S}^{\mathsf{faithful}}_0({\mathsf B}({\mathcal H})):\ \varrho\mu=\mu \varrho\bigr\}
\end{equation}
where in the defining relation the operator multiplication is made use of. The set given by \eqref{introleaf1a} will be referred to as $\mu$-leaf (within the set of   density operators of full support).
Remarkably, each $\mu$-leaf proves to be geodesically convex in respect of  the Bures metric. Note that if a differentiable curve $I\ni t\longmapsto \varrho_t$ evolving within a $\mu$-leaf is considered, then from the definition \eqref{introleaf1a} the  relation $\varrho^{\,\prime}\mu=\mu\, \varrho^{\,\prime}$ follows. Hence, the unique solution ${\mathsf{d}}\xi\in \bigl[{\mathsf B}({\mathcal H})_{\mathsf h}\sqrt{\varrho}\,\bigr]_{{\mathsf{H.S.}}}$ of \eqref{bi3b} is obtained under the auxiliary condition that
\begin{equation}\label{introleaf1b}
\varrho\in {\mathcal S}_0(\mu),\,{\mathsf{d}}\varrho\, \mu=\mu\, {\mathsf{d}}\varrho
\end{equation}
hold. But under these conditions one can prove that a unique self-adjoint (possibly unbounded) operator valued solution ${\mathsf{d}}x$ of the operator equation
\begin{subequations}\label{bi45}
\begin{equation}\label{bi4}
{\mathsf{d}}\varrho ={\mathsf{d}}x \,\varrho+\varrho\,{\mathsf{d}}x
\end{equation}
exists, with domain of definition ${\mathcal D}({\mathsf{d}}x)$ assuring that ${\mathsf{d}}x \sqrt{\varrho}\in {{\mathsf{H.S.}}}({\mathcal H})$ is fulfilled. Also, in this case then ${\mathsf{d}}\xi={\mathsf{d}}x \sqrt{\varrho}$ will be the unique solution of \eqref{bi3b} in $\bigl[{\mathsf B}({\mathcal H})_{\mathsf h}\sqrt{\varrho}\,\bigr]_{{\mathsf{H.S.}}}$. Accordingly, in restriction to the $\mu$-leaf, by means of the self-adjoint solution ${\mathsf{d}}x$ of \eqref{bi4} the square of the Bures length element \eqref{bi3c} reads
\begin{equation}\label{bi5}
 {\mathsf{d}}s^2={\mathsf{tr}}\,\bigl({\mathsf{d}} x\,\sqrt{\varrho}\bigr)\bigl({\mathsf{d}} x\,\sqrt{\varrho}\bigr)^*=\| {\mathsf{d}} x\,\sqrt{\varrho}\,\|_{{\mathsf{H.S.}}}^2
\end{equation}
\end{subequations}
There are special cases of $\mu$-leaves where the ${\mathsf{d}}\varrho$-coordinated $1$-form ${\mathsf{d}} x$ is bounded operator valued. For instance, in case of ${\mathsf{dim}}\,{\mathcal H}<\infty$, if for $\mu$  the equipartition  $\mu=(1/{\mathsf{dim}}\,{\mathcal H})\,{\mathsf 1}$ is chosen, then this will be the case and we have
\[
{\mathcal S}_0(\mu)={\mathcal S}^{\mathsf{faithful}}_0({\mathsf B}({\mathcal H}))
\]
Thus, in this case, ${\mathcal S}^{\mathsf{faithful}}_0({\mathsf B}({\mathcal H}))$ is geodesically convex, and \eqref{bi5} then reads
\[
{\mathsf{d}}s^2={\mathsf{tr}}\,\varrho\,{\mathsf{d}} x^2
\]
 This special finite dimensional branch of formula \eqref{bi45} corresponds to a result observed by Uhlmann \cite{Uhlm:91.2,Uhlm:92.1,Uhlm:93,Uhlm:95,Uhlm:11} and saying that the connection generated by the Bures metric on the manifold of all invertible $n\times n$-density matrices is Riemannian, see \cite{DiRu:92.1,DiRu:92}.
 In this case by standard results, see \cite{Hein:51,Rose:56,LuRo:59},  the ${\mathsf{d}}\varrho$-coordinated $1$-form ${\mathsf{d}} x$ can be written as
 \[
 {\mathsf{d}} x=\int_0^\infty \exp{(-\varrho\,t)}\,{\mathsf{d}}\varrho \exp{(-\varrho\,t)}\ d\/t
 \]
By \eqref{bi45} the functional analytical core of the subject is extended generically from the category of invertible density operators over finite dimensional Hilbert spaces to the category of $\mu$-leaves over separable Hilbert spaces.
\newpage
\section{Bures distance, transition probability and fidelity}\label{burdist}
The aim of this section is to provide some of the  essentials around the distance as coming along with Definition \ref{budi}. Thereby,
note that most of the properties of the Bures distance $d_{\mathrm B}$ follow along with known properties of the functors $P$ or $F$ of Uhlmann's ${\mathsf C}^*$-algebraic transition probability $P$ \cite{Uhlm:76} or the so-called fidelity $F$ \cite{Jozs:94}, respectively. Thereby, between $P$ and $F$ the functional relation $P=F^{\,2}$ is fulfilled, with the value $F(M|\nu,\varrho)$ of the functor $F$ to be defined at two  states $\nu,\varrho\in {\mathcal S}(M)$ of the given unital ${\mathsf C}^*$-algebra $M$ as follows\,:
\begin{definition}\label{genprob.2}\hfill{}
$\displaystyle F(M|\nu,\varrho)=
\sup_{\varphi\in
{\mathcal S}_{\pi,M}(\nu),\psi\in {\mathcal S}_{\pi,M}(\varrho),\, \pi}
|\langle \psi,\varphi\rangle_\pi|\hfill{}\,\ \ \ \ .$
\end{definition}
The range of variables in the supremum is the same as in Definition \ref{budi}.
According to Definition \ref{budi}, in terms of $F$ then the following formula for $d_B$ is obtained\,:
\begin{equation}\label{pcont.1}
d_B(M|\nu,\varrho)=\sqrt{2}\,\sqrt{1-F(M|\nu,\varrho)}
\end{equation}
More generally, for notational simplicity mainly,  for other formulae of Bures geometry too subsequently a (possibly mixed in $F$ and $P$) representation will be preferred in which each free single term of the square root of $P$ at each place of occurrence will be replaced with $F$ by applying the substitution
\begin{equation}\label{pcont.1aa}
F(M|\nu,\varrho)=\sqrt{P(M|\nu,\varrho)}
\end{equation}
The fidelity now plays a fundamental r\^{o}le in context of quantum information processing, see e.g.~\cite{Jozs:94}, and for the terminology and usage in a quantum physical context refer the reader to some of the discussions in \cite{Uhlm:11}.
\subsection{Some basic results}\label{bas2}
Let $\nu,\varrho\in {\mathcal S}(M)$, and be $\{\pi,{\mathcal H}_\pi\}$ a unital $^*$-representation of $M$
such that the $\pi$-fibres of $\nu, \varrho$ both exist. Suppose $\varphi\in {\mathcal S}_{\pi,M}(\nu)$ and
$\psi\in {\mathcal S}_{\pi,M}(\varrho)$. Let
$$\pi(M)^{\,\prime}=\{z\in {\mathsf B}({\mathcal H}_\pi): z\pi(x)=\pi(x)z,\,\forall x\in M\}$$
be the commutant $vN$-algebra of $\pi(M)$, and be ${\mathcal U}(\pi(M)^{\,\prime})$ the group
of unitary operators of $\pi(M)^{\,\prime}$.
Define a linear form $h_{\psi,\varphi}^\pi$ as follows\,:
\begin{equation}\label{hform}
\forall \,z\in\pi(M)^{\,\prime}\,:\  h_{\psi,\varphi}^\pi(z)=\langle z\psi,\varphi\rangle_\pi\,.
\end{equation}
For $\chi\in {\mathcal H}_\pi$ let orthoprojections
$p_\pi(\chi)$ and $p_\pi^{\,\prime}(\chi)$ be defined as the orthoprojections projecting from ${\mathcal H}_\pi$
onto the the closures
$[\pi(M)^{\,\prime}\chi]$ and
$[\pi(M)\chi]$ of the linear subspaces $\pi(M)^{\,\prime}\chi=\{z\chi: z\in\pi(M)^{\,\prime} \}$ and
$\pi(M)\chi=\{y\chi: y\in\pi(M) \}$, respectively.
Note that $p_\pi(\chi)\in \pi(M)^{\,\prime\prime},
p_\pi^{\,\prime}(\chi)\in \pi(M)^{\,\prime}$ with the double commutant
$vN$-algebra $$\pi(M)^{\,\prime\prime}=
\bigl(\pi(M)^{\,\prime}\bigr)^{\,\prime}$$ of $\pi(M)$. By the Kaplansky--von\,Neumann theorem
$\pi(M)$ is strongly dense within $\pi(M)^{\,\prime\prime}$, and therefore
one always has
\begin{equation}\label{bas3a}
[\pi(M)\chi]=[\pi(M)^{\,\prime\prime}\chi]\,,
\end{equation}
which is useful to know. Relating \eqref{faser}, for each
$\chi\in {\mathcal S}_{\pi,M}(\mu)$ the following is true\,:
\begin{equation}\label{bas4}
{\mathcal S}_{\pi,M}(\mu)=\bigl\{v\chi:\,v^*v=p_\pi^{\,\prime}(\chi),\,v\in \pi(M)^{\,\prime}\bigr\}\,.
\end{equation}
We are going to comment now on some of the fundamentals of Bures geometry.
\subsubsection{Basic algebraic facts on fidelity}\label{basfa}
In view of \eqref{pcont.1aa} and \cite[{\sc{Corollary 1}},\,{\sc{Corol\-lary 2}},\,eqs.\,(5),\,(6) and {\sc{Theorem 3}}]{Albe:83} on a  unital ${\mathsf C}^*$-algebra $M$ the fidelity $F$ may be represented as follows.
\begin{lemma}\label{bas3} Whenever $\nu,\varrho\in  {\mathcal S}(M)$, and $\varphi\in {\mathcal S}_{\pi,M}(\nu)$ and
$\psi\in {\mathcal S}_{\pi,M}(\varrho)$, then
\begin{enumerate}
  \item\label{bas3aa} ${F(M|\nu,\varrho)}=\| h_{\psi,\varphi}^\pi \|_1=
\sup_{u\in {\mathcal U}(\pi(M)^{\,\prime})}
| h_{\psi,\varphi}^\pi(u)|\,;$
  \item\label{bas3b} ${F(M|\nu,\varrho)}= \inf_{x\in M_+,\,\text{\rm{invertible}}} \sqrt{\nu(x)\varrho(x^{-1})}\,;$
  \end{enumerate}
\end{lemma}
\begin{proof}
Due to its importance a proof of \eqref{bas3aa} will be included. Let  $N=\pi(M)^{\,\prime\prime}$, and be  $\varepsilon >0$ a real. Let $y$ be chosen from the unit ball of $N^{\,\prime}$, $y\in (N^{\,\prime})_1$, such that
\begin{equation}\label{approx}
{h^\pi_{\psi,\varphi}}(y)=|{h^\pi_{\psi,\varphi}}(y)|\geq\|{h^\pi_{\psi,\varphi}}\|_1-\varepsilon
\end{equation}
is fulfilled. Following \cite{DyRu:66} the unit ball $(A)_1$ in a unital
$C^*$-algebra $A$ is the uniform closure of the convex hull of the unitary
operators ${\mathcal U}(A)$ of $A$, $(A)_1=[{\mathrm{conv}}{\,\mathcal U}(A)]$.
In line with this, \eqref{approx}  can be assumed  to hold
with $y=\sum_{j=1}^n {\lambda}_{j}u_{j}$, for ${\lambda_j} \in {{\mathbb{R}}}_{+}$
with $\sum_{j=1}^n {\lambda_j}=1$, and $u_j\in{\mathcal U}(N^{\,\prime})$, for all $j\leq n$,  with some $n\in {{\mathbb{N}}}$.
Then obviously
\begin{equation}\label{oben}
 {h^\pi_{\psi,\varphi}}(y)={\biggl|\sum_{j=1}^n {\lambda_j}{h^\pi_{\psi,\varphi}}(u_{j})\biggr|}
\leq {\sum_{j=1}^n {\lambda_j}|{h^\pi_{\psi,\varphi}}(u_{j})|}
\end{equation}
is fulfilled. Since from $\psi \in {\mathcal S}_{\pi,M}({\varrho})$ and $u_j\in{\mathcal U}(N^{\,\prime})$
also $u_{j}\psi \in {\mathcal S}_{\pi,M}({\varrho})$ follows, according to the Definition \ref{genprob.2} we see that, for any $j$, the estimate $$|h^\pi_{\psi,\varphi}(u_{j})|=|{\langle u_{j}\psi,\varphi \rangle_\pi}|
\leq F(M|\nu,\varrho)$$ is obtained, and accordingly, (\ref{oben}) can be
continued with the result
$F(M|\nu,\varrho)\geq {h^\pi_{\psi,\varphi}}(y)$. The latter together
with (\ref{approx}) yields
$F(M|\nu,\varrho)\geq \|{h^\pi_{\psi,\varphi}}\|_1-\varepsilon $,
for any $ \varepsilon >0$. Hence, we have arrived at
$F(M|\nu,\varrho)\geq\|h^\pi_{\psi,\varphi}\|_1$, i.e.
$\|{h^\pi_{\psi,\varphi}}\|_1$ is a lower bound for
$F(M|\nu,\varrho)$.
On the other hand, we are going to prove now that $\|{h^\pi_{\psi,\varphi}}\|_1$ also is an upper bound for
$F(M|\nu,\varrho)$.  In fact,  according to Definition \ref{genprob.2} there exist a unital $^*$-representation $\{\pi_\varepsilon,{\mathcal H}_{\pi_\varepsilon}\}$ and   $\psi_\varepsilon \in {\mathcal S}_{\pi_\varepsilon,M}({\varrho}), \varphi_\varepsilon
\in {\mathcal S}_{\pi_\varepsilon,M}({\nu})$ such that
\begin{equation}\label{unten}
\langle\psi_\varepsilon,\varphi_\varepsilon\rangle_{\pi_\varepsilon}\geq F(M|\nu,\varrho)-\varepsilon
\end{equation}
Note that by definition of $\varphi$, $\psi$, $\varphi_\varepsilon$ and $\psi_\varepsilon$, for $x, z\in M$ the mapping  $$\pi(M)\psi\times \pi(M)\varphi\ni \pi(x)\psi \times \pi(z)\varphi\longmapsto \langle \pi_\varepsilon(x)\psi_\varepsilon,\pi_\varepsilon(z)\varphi_\varepsilon\rangle_{\pi_\varepsilon}$$ is sesquilinear and of norm one, and is obeying
\begin{equation}\label{unten1}
 \langle \pi_\varepsilon(w x)\psi_\varepsilon,\pi_\varepsilon(z)\varphi_\varepsilon\rangle_{\pi_\varepsilon}= \langle \pi_\varepsilon(x)\psi_\varepsilon,\pi_\varepsilon(w^*z)\varphi_\varepsilon\rangle_{\pi_\varepsilon}
\end{equation}
for each triplet $x, w, z\in M$.
The unique extension of this map to a sesquilinear mapping over $p^{\,\prime}_{\pi}(\psi){\mathcal H}_{\pi}\times\, p^{\,\prime}_{\pi}(\varphi){\mathcal H}_{\pi}$ has norm one, too, and in view of \eqref{unten1} by standard calculus then can be represented as
\[
\langle \pi_\varepsilon(x)\psi_\varepsilon,\pi_\varepsilon(z)\varphi_\varepsilon\rangle_{\pi_\varepsilon}=
\langle K\pi(x)\psi,\pi(z)\varphi\rangle_\pi
\]
for all $x,z\in M$, with some $K\in (N^{\,\prime})_1$. Thus especially, for $x=z={\mathsf 1}$ one infers that
\[
\langle \psi_\varepsilon,\varphi_\varepsilon\rangle_{\pi_\varepsilon}=
\langle K\psi,\varphi\rangle_\pi=h^\pi_{\psi,\varphi}(K)\leq \|h^\pi_{\psi,\varphi}\|_1
\]
has to be fulfilled.  Thus, in view of  \eqref{unten} we get
 $\|{h^\pi_{\psi,\varphi}}\|_1 \geq
F(M|\nu,\varrho)-\varepsilon$. Hence, since $\varepsilon>0$ can be arbitrarily chosen, we have arrived at the fact that
$\|{h^\pi_{\psi,\varphi}}\|_1$ is an  upper bound of
$F(M|\nu,\varrho)$ as well, and therefore in view of the above  $$F(M|\nu,\varrho)=\|{h^\pi_{\psi,\varphi}}\|_1$$ is seen to hold. Also, since $(N^{\,\prime})_1$ is the closed convex hull of the unitaries of $N^{\,\prime}$,
$$
\|{h^\pi_{\psi,\varphi}}\|_1=
\sup_{u\in {\mathcal U}(\pi(M)^{\,\prime})}
| h_{\psi,\varphi}^\pi(u)|
$$
has to be fulfilled, and thus finally  \eqref{bas3aa} follows.
\end{proof}
Each unitary orbit ${\mathcal U}(\pi(M)^{\,\prime})\chi$ is a
special subset of the $\pi$-fibre \eqref{bas4}. Thus,
the first two items of the following result at once can be seen from Lemma \ref{bas3}\,\eqref{bas3aa}.
\begin{theorem}\label{bas5}
Suppose $\{\pi,{\mathcal H}_\pi\}$ is a unital $^*$-representation such that the $\pi$-fibres of
$\nu,\varrho\in  {\mathcal S}(M)$ exist. For given $\varphi\in {\mathcal S}_{\pi,M}(\nu)$ and
$\psi\in {\mathcal S}_{\pi,M}(\varrho)$ the following hold\textup{:}
\begin{enumerate}
\item\label{bas51}
$d_B(M|\nu,\varrho)=\inf_{u\in {\mathcal U}(\pi(M)^{\,\prime})}
\|u\psi-\varphi\|_\pi$\,;
\item\label{bas52}
$d_B(M|\nu,\varrho)=\inf_{\psi'\in {\mathcal S}_{\pi,M}(\varrho)\phantom{)}}
 \|\psi'-\,\varphi\|_\pi$\,;
\item\label{bas53}
$d_B(M|\nu,\varrho)=\|\psi-\varphi\|_\pi \Longleftrightarrow\ h_{\psi,\varphi}^\pi\geq 0$.
\end{enumerate}
\end{theorem}
Relating \eqref{bas53}, note that $h_{\psi,\varphi}^\pi\geq 0$ is equivalent
to $\langle\psi,\varphi\rangle_\pi=h_{\psi,\varphi}^\pi({\mathsf 1})=\| h_{\psi,\varphi}^\pi \|_1$.
The latter is equivalent to  $d_B(M|\nu,\varrho)=\|\psi-\varphi\|_\pi$, by \eqref{pcont.1} and
Lemma \ref{bas3} \eqref{bas3aa}.
\begin{remark}\label{rem1}
\begin{enumerate}
\item\label{rem10}
According to Lemma \ref{bas3}\,\eqref{bas3aa} and Theorem \ref{bas5}\,\eqref{bas51}, ${F}$ and $d_B$ can be calculated in each $^*$-representation
$\pi$ where both fibres exist.
\item\label{rem10a}
By Lemma \ref{bas3}\,\eqref{bas3b}, ${F(M|\cdot,\cdot)}$ is an infimum of continuous functions. Thus it has to be upper semi-continuous. That is, for
sequences $\{\nu_n\}$ and $\{\varrho_n\}$ of states obeying $\lim_n \nu_n=\nu$ and  $\lim_n \varrho_n=\varrho$ one always has the estimate
\begin{equation}\label{gemhalbstet}
 {F(M|\nu,\varrho)}\geq \limsup_n {F(M|\nu_n,\varrho_n)}
\end{equation}
and for $P$ accordingly.
\item\label{rem11}
For $\mu\in  {\mathcal S}(M)$ and $\pi$ with ${\mathcal S}_{\pi,M}(\mu)\not=\emptyset$
a unique normal state $\mu_\pi$ on
$\pi(M)^{\,\prime\prime}$ with $\mu=\mu_\pi\circ\,\pi$ exists. Let $s(\mu_\pi)\in \pi(M)^{\,\prime\prime}$ be the support orthoprojection of $\mu_\pi$. It is useful to keep in mind that for each  $\xi\in {\mathcal S}_{\pi,M}(\mu)$ one has
\begin{equation}\label{traeger}
   s(\mu_\pi)=p_\pi(\xi)
\end{equation}
Also, by \eqref{bas3a} and \eqref{bas4} one has that
${\mathcal S}_{\mathrm{id},\pi(M)^{\,\prime\prime}}(\mu_\pi)={\mathcal S}_{\pi,M}(\mu)$,
with the trivial representation $\mathrm{id}$ of $\pi(M)^{\,\prime\prime}$ on ${\mathcal H}_\pi$.
\item\label{vN2}
Note that Definition \ref{budi} and Definition \ref{genprob.2} admit  straightforward extensions of $d_B$ and $F$ (and $P$ accordingly) to general positive linear forms. This  extensions   tacitely  will be made use of in all cases if this seems useful. Each  vector $\xi$ of the $\pi$-fibre ${\mathcal S}_{\pi,M}(\mu)$ of  $\mu\in M_+^*$ then is of norm $$\|\xi\|=\sqrt{\mu({\mathsf 1})}$$
\item\label{rem12}
For a $vN$-algebra $M$ acting over some Hilbert space ${\mathcal H}$, with normal state space ${\mathcal S}_0(M)$, suppose $\nu,\varrho\in {\mathcal S}_0(M)$ can be implemented by vectors. By \eqref{rem10}, for calculating ${F}$ and $d_B$ the identity representation $\pi=\mathrm{id}$ may be referred to, and then
${\mathcal S}_M(\nu)={\mathcal S}_{{\mathrm{id}},M}(\nu)$, $h_{\psi,\varphi}=h_{\psi,\varphi}^{\mathrm{id}}$, $v_{\psi,\varphi}=v_{\psi,\varphi}^{\mathrm{id}}$,
$p(\varphi)=p_{\mathrm{id}}(\varphi)$ and $p^{\,\prime}(\varphi)=p_{\mathrm{id}}^{\,\prime}(\varphi)$ will be abbreviated.
\end{enumerate}
\end{remark}
Let $M$ be a unital ${\mathsf C}^*$-algebra and $\nu,\varrho\in {\mathcal S}(M)$. According to Definition \ref{genprob.2}, if we let $\varphi$ and $\psi$ vary through the $\pi$-fibres of $\nu$ and $\varrho$ with respect to any unital $^*$-representation $\{\pi,{\mathcal H}_\pi\}$ of $M$
such that the $\pi$-fibres of $\nu,\varrho$ both exist,  for each finite positive decomposition of the unity $\{x\}\subset M_+$ with the help of the Cauchy-Schwarz inequality the following conclusions can be seen to hold
\begin{eqnarray*}
  {F(M|\nu,\varrho)} &=& \sup_{\psi,\varphi,\pi}|\langle\psi,\varphi\rangle_\pi|=\sup_{\psi,\varphi,\pi}\Bigl|\Bigl\langle\sum_k \pi(x_k) \psi,\varphi\Bigr\rangle_\pi\Bigr| = \sup_{\psi,\varphi,\pi}\Bigl|\sum_k\langle  \pi(x_k)\psi,\varphi\rangle_\pi\Bigr| \\
   &\leq & \sup_{\psi,\varphi,\pi} \sum_k |\langle  \pi(x_k) \psi,\varphi\rangle_\pi|=\sup_{\psi,\varphi,\pi} \sum_k|\langle \sqrt{ \pi(x_k)}\, \psi,\sqrt{ \pi(x_k)}\,\varphi\rangle_\pi|\\
   &\leq & \sup_{\psi,\varphi,\pi}\sum_k \|\sqrt{ \pi(x_k)}\, \psi\|_\pi\,\|\sqrt{ \pi(x_k)}\, \varphi\|_\pi = \sum_k\sqrt{\varrho(x_k)\nu(x_k)}
\end{eqnarray*}
Hence, for each two states  $\nu,\varrho$ and any finite positive decomposition of the unity $\{x\}$ in  a unital ${\mathsf C}^*$-algebra $M$ the following useful estimate is fulfilled
\begin{equation}\label{bas3cc}
  {F(M|\nu,\varrho)}\leq \sum_k\sqrt{\nu(x_k)\varrho(x_k)}
\end{equation}
For ${\mathsf W}^*$-algebras from \eqref{bas3cc} another representation arises \cite[{\sc{Corollary 2}},\,(1)]{AlUh:00.1}.
\begin{lemma}\label{bas3ccc}
Let $N$ be a ${\mathsf W}^*$-algebra, $\nu,\varrho\in {\mathcal S}(N)$. Then one has
$$
    {F(N|\nu,\varrho)}=\inf_{\{x\}}\sum_k \sqrt{\nu(x_k)\varrho(x_k)}
$$
where $\{x\}$ extends over all finite positive decompositions of the unity within $N$.
\end{lemma}
As a useful application from this the following result can be inferred.
\begin{corolla}\label{ctraeger}
Let $\omega,\mu\in  {\mathcal S}(N)$ be normal states in a ${\mathsf W}^*$-algebra $N$, with support orthoprojections $s(\omega), s(\mu)\in N$. Let $p$ be the orthoprojection $p=s(\omega)\vee s(\mu)\in N$, and be $N_p$ the hereditary $vN$-subalgebra $p\,Np$. Then, $$d_B(N|\omega,\mu)= d_B(N_p|\omega_p,\mu_p)$$ holds, with the restrictions of $\omega$ and $\mu$ onto $N_p$, $\omega_p=\omega|N_p$ and $\mu_p=\mu|N_p$
\end{corolla}
\begin{proof}
In fact, each finite positive decomposition $\{y\}$  of the unity  within $N_p$ (containing $n\in {{\mathbb{N}}}$ elements) by triviality extends to a finite positive decomposition $\{x\}$ of the unity within $N$ when setting $x_k=y_k+\bigl(\frac{1}{n}\bigr) p^\perp$. Thereby, one infers that
\begin{equation}\label{gleichomegamu}
  \omega(x_k)\mu(x_k)=\omega_p(y_k)\mu_p(y_k)
\end{equation}
Thus, by Lemma \ref{bas3ccc}, when applied to $N$, $\omega$, $\mu$ and $N_p$, $\omega_p$, $\mu_p$, respectively, we see
\begin{equation*}
     {F(N|\omega,\mu)}\leq  {F(N_p|\,\omega_p,\mu_p)}
\end{equation*}
On the other hand, for each finite positive decomposition $\{x\}$ of the unity in $N$,
by defining $y_k=p \,x_k p$ a finite positive decomposition of the unity in $N_p$ satisfying \eqref{gleichomegamu} can be obtained. Applying Lemma \ref{bas3ccc} with $N$, $\omega$, $\mu$ and $N_p$, $\omega_p$, $\mu_p$ now yields
\begin{equation*}
   {F(N_p|\,\omega_p,\mu_p))} \leq  {F(N|\omega,\mu)}
\end{equation*}
Thus
$
    {F(N|\omega,\mu)}= {F(N_p|\,\omega_p,\mu_p)}
$ holds,
and then by \eqref{pcont.1} the assertion follows.
\end{proof}
Close this paragraph with some special cases where Lemma \ref{bas3} can be made more explicit by an instructive formula.
Let $\tau$ be
a (lower semicontinuous)
trace on $M_+$ (see e.g.~in \cite[6.1.]{Dixm:64}). Then one has $^*$-ideals
${\mathcal L}^2(M,\tau)=\{x\in M:\tau(x^*x)<\infty\}$ and
${\mathcal L}^1(M,\tau)=\{x\in M:\,x=y^*z,\ y,z\in {\mathcal L}^2(M,\tau)\}$.
It is known that, for $z\in {\mathcal L}^2(M,\tau)$ and
$y\in {\mathcal L}^1(M,\tau)$, we have a
unique positive linear form $\tau^z$ and bounded linear form $\tau_y$ on $M$, with
$\tau^z(x)=\tau(z^* xz)$ and $\tau_y(x)=\tau(yx)$, respectively, for each $x\in M$.
Let $I(M,\tau)=\{x\in M:\,\tau(x^*x)=0\}$. This is a $^*$-ideal
in ${\mathcal L}^2(M,\tau)$. Consider the
completion $L^2(M,\tau)$ of the space of
equivalence classes $\eta_x$ modulo $I(M,\tau)$ under the inner product
$\langle \eta_x, \eta_y\rangle=\tau(y^* x)$, $x,y\in {\mathcal L}^2(M,\tau)$. By standard arguments
the latter is well-defined and extends to a scalar product
on $L^2(M,\tau)$, which then is a Hilbert space. Let $\pi: M\ni x\,\longmapsto\,\pi(x)$
be the $^*$-representation
of $M$ over $L^2(M,\tau)$
which can be uniquely given through $\pi(x)\eta_y=\eta_{xy}$
, for each
$y\in {\mathcal L}^2(M,\tau)$ and all $x\in M$.
Then, for each $z\in {\mathcal L}^2(M,\tau)$ and $x\in M$ one has
$\tau^z(x)=\langle \pi(x)\eta_{z},\eta_{z}\rangle$.
In carefully analyzing $\pi$
in case of a $vN$-algebra $M$, and
in applying Lemma \ref{bas3} with the
mentioned $^*$-representation $\pi$ and using that a normal trace is
lower semicontinuous, by specializing from \cite[{\sc Corollary 1}]{Albe:02.1} the following hold.
\begin{example}\label{ex1}
Let $M$ be a $vN$-algebra, with normal trace $\tau$. For $x,y\in {\mathcal L}^2(M,\tau)$ and
$a=|x^*|^2=xx^*$, $c=|y^*|^2=yy^*$, one has $\tau^x=\tau_a$ and $\tau^y=\tau_b$, and
then
\begin{equation}\label{bei1}
{F(M|\tau^x,\tau^y)}=\tau(|x^*y|)=\tau(|\,\sqrt{a}\sqrt{c}\,|)=
{F(M|\tau_{a},\tau_{c})}\,.
\end{equation}
In the special case of $M={\mathsf B}({\mathcal H})$ and $\tau={\mathsf{tr}}$ this according to \eqref{pcont.1aa} is equivalent to Uhlmann's formula  \cite{Uhlm:76} for the algebraic transition probability
between two density operators $a$ and $c$ over some Hilbert space ${\mathcal H}$.
\end{example}
\subsubsection{Implementing fidelity and Bures distance by vectors}\label{met}
From each of the first two items of Theorem \ref{bas5} together with
Definition \ref{budi} and \eqref{bas4}
it is obvious that $d_B$ is symmetric in its arguments, obeys the triangle inequality with respect to any three states of $M$ (apply Theorem \ref{bas5} with the GNS-representation $\pi$ in respect of the sum of the three states in question),
and is vanishing between
$\nu$ and $\varrho$ if, and only if, $\nu=\varrho$. Thus, $d_B$ yields another
distance function on ${\mathcal S}(M)$. By Theorem \ref{bas5}\,\eqref{bas52}, $d_B$ is the natural
distance function, in each space of $\pi$-fibres.
Also, $d_B$ is topologically
equivalent to the functional distance $d_1$, which
is defined as $d_1(\nu,\varrho)=\|\nu-\varrho\|_1$, for $\nu,\varrho\in {\mathcal S}(M)$.
For quantitative estimates,
see \cite[{\sc{Proposition 1}},\,formula (1.2)]{AlUh:00.1}.
We do not prove all relevant facts but instead remark that, in case of normal
states on ${\mathsf W}^*$-algebras, the equivalence in question results from the
estimates given in
\cite{Bure:69,Arak:71,Arak:72}, essentially.
This extends to unital ${\mathsf C}^*$-algebras  by means of the following auxiliary result.
\begin{lemma}\label{eqdist}
Let $\nu,\varrho\in{\mathcal S}(M)$, and be $\{\pi,{\mathcal H}_\pi\}$ such that the $\pi$-fibres of
$\nu$ and $\varrho$ both exist. Then, the following hold\textup{:}
\begin{enumerate}
\item\label{eqdist1}
$\|\nu-\varrho\|_1=\|\nu_\pi- \varrho_\pi\|_1\,;$
\item\label{eqdist2}
$d_B(M|\nu,\varrho)=d_B(\pi(M)^{\,\prime\prime}|\nu_\pi,\varrho_\pi)\,.$
\end{enumerate}
\end{lemma}
\begin{proof}
In order to see \eqref{eqdist1} recall first that, since $\pi$ is a $^*$-homomorphism,
according to basic ${\mathsf C}^*$-theory, the unit ball $M_1$ of $M$ is related to the unit ball
$\pi(M)_1$ of  $\pi(M)$ through $\pi(M)_1=\pi(M_1)$. From this by the Kaplansky density
theorem
$$(\pi(M)^{\,\prime\prime})_1=
[\pi(M)_1]^{\,\mathrm{st}}=[\pi(M_1)]^{\,\mathrm{st}}\text{ (strong closure) }$$
is seen.
Owing to normality of $\nu_\pi$, $\varrho_\pi$, for $\nu_\pi-\varrho_\pi$ we then may conclude as
follows:
\begin{equation*}
\begin{split}
\|\nu-\varrho\|_1 & =\sup_{x\in M,\,\|x\|\leq 1} |\nu_\pi(\pi(x))-\varrho_\pi(\pi(x))|\\
& =\sup_{y\in\pi(M)^{\,\prime\prime},\,\|y\|\leq 1} |\nu_\pi(y)-\varrho_\pi(y)|= \|\nu_\pi- \varrho_\pi\|_1\,.
\end{split}
\end{equation*}
To see \eqref{eqdist2}, apply Theorem \ref{bas5}\,\eqref{bas51} to $\nu_\pi$ and $\varrho_\pi$ in respect of the identity representation
of the $vN$-algebra
$\pi(M)^{\,\prime\prime}$. Under these premises
Theorem \ref{bas5}\,\eqref{bas51} yields
$$d_B(\pi(M)^{\,\prime\prime}|\nu_\pi,\varrho_\pi)=\inf_{u\in {\mathcal U}(\pi(M)^{\,\prime})}\|u\psi-\varphi\|_\pi$$ for each $\psi\in{\mathcal S}_{\mathrm{id},\pi(M)^{\,\prime\prime}}(\varrho_\pi)$ and $\varphi\in
{\mathcal S}_{\mathrm{id},\pi(M)^{\,\prime\prime}}(\nu_\pi)$.
In view of Remark \ref{rem1}\,\eqref{rem11} this is $d_B(M|\nu,\varrho)$ as asserted by
Theorem \ref{bas5}\,\eqref{bas51} when applied to $M$, $\pi$, $\nu$
and $\varrho$.
\end{proof}
For the following the trivial extension of $P$ and $F$ to positive linear forms  according to Remark \ref{rem1}\,\eqref{vN2} is taken into account.
\begin{corolla}\label{subadd}
Let $\nu,\varrho\in M_+^*$, for a unital ${\mathsf C}^*$-algebra $M$. Then,
\begin{enumerate}
\item\label{subadd1}
$P(M|\nu,\varrho)\geq P(M|\nu,\varrho_1)+P(M|\nu,\varrho_2)$ provided $\varrho=\varrho_1+\varrho_2$, with $\varrho_i\in M^*_+$;
\item\label{subadd2}
$F(M|\nu,\varrho)=0$ occurs if, and only if,  $\nu\perp\varrho$ is fulfilled.
\end{enumerate}
\end{corolla}
\begin{proof}
 The superadditivity property \eqref{subadd1} immeditately follows from \eqref{pcont.1aa} and Lemma \ref{bas3}\,\eqref{bas3b}.
 In context of  \eqref{subadd2} remind the notion of orthogonality for two positive linear forms in a ${\mathsf C}^*$-algebra $M$: $\nu,\varrho\in M_+^*$ are termed `mutually orthogonal', in sign $\nu\perp \varrho$, if
 $$
 \|\nu-\varrho\|_1=\|\nu\|_1+\|\varrho\|_1
 $$
 Since Lemma \ref{eqdist}\,\eqref{eqdist1} by triviallity extends to positive linear forms, $\nu\perp \varrho$ happens if, and only if, $\nu_\pi \perp \varrho_\pi$, with respect to any unital $^*$-representation $\{\pi,{\mathcal H}_\pi\}$ such that the $\pi$-fibres of
 $\nu$ and $\varrho$ both exist. On the other hand, for normal forms $\nu_\pi, \varrho_\pi$ over the $vN$-algebra $\pi(M)^{\,\prime\prime}$ orthogonality $\nu_\pi \perp \varrho_\pi$ obviously is equivalent to orthogonality of the support orthoprojections, $s(\nu_\pi)s(\varrho_\pi)={\mathsf 0}$. Hence,  $\nu\perp \varrho$ proves equivalent to $s(\nu_\pi)s(\varrho_\pi)={\mathsf 0}$, too. In order to see  \eqref{subadd2} we therefore may content with showing that $s(\nu_\pi)s(\varrho_\pi)={\mathsf 0}$ if, and only if, $F(M|\nu,\varrho)=0$. In line with this, suppose $\varphi\in {\mathcal S}_{\pi,M}(\nu)$ and $\psi\in {\mathcal S}_{\pi,M}(\varrho)$. Then, since $\pi(M)^{\,\prime}$ is linearily generated by its unitary elements,
$
s(\varrho_\pi){\mathcal H}_\pi=p_\pi(\psi){\mathcal H}_\pi=[{\mathsf{Lin}}({\mathcal U}(\pi(M)^{\,\prime})\psi)]
$ and $s(\nu_\pi){\mathcal H}_\pi=p_\pi(\varphi){\mathcal H}_\pi=[{\mathsf{Lin}}({\mathcal U}(\pi(M)^{\,\prime})\varphi)]$ have to be fulfilled. Hence, in view of  Lemma \ref{bas3}\,\eqref{bas3aa}, $s(\nu_\pi)s(\varrho_\pi)={\mathsf 0}$ if, and only if, $F(M|\nu,\varrho)=0$.
\end{proof}
\begin{remark}\label{vN}
\begin{enumerate}
\item\label{vN1}
In view of \eqref{pcont.1}, Lemma \ref{eqdist} \eqref{eqdist2} equivalently says that
\begin{equation}\label{eqdist3}
{F(M|\nu,\varrho)}={F(\pi(M)^{\,\prime\prime}|\nu_\pi,\varrho_\pi)}
\end{equation}
i.e. calculating the ${\mathsf C}^*$-algebraic fidelity can be traced back to considering the fidelity  between suitably chosen normal states over suitably chosen $vN$-algebras.
\item\label{eingstet}
In making use of \eqref{eqdist3} and Lemma \ref{eqdist}\,\eqref{eqdist1} by means of some arguments arising from  modular theory of $vN$-algebras separate continuity of ${F(M|\cdot,\varrho)}$ can be inferred to hold, see \cite[(2.5)]{AlUh:84} and use \eqref{pcont.1aa}, that is
\[
{F(M|\nu,\varrho)}=\lim_n {F(M|\nu_n,\varrho)}
\]
is fulfilled
for each sequence $\{\nu_n\}$ of states obeying $\|\cdot\|_1-\lim_n \nu_n=\nu$.
\item\label{vN3}
If for each $\mu\in M_+^*$ and given central orthoprojection $c$ of the $vN$-algebra $\pi(M)^{\,\prime\prime}$, $c\in \pi(M)^{\,\prime\prime}\cap \pi(M)^{\,\prime}$,   a linear form $\mu_{\pi,c}\in M_+^*$ is defined by $$\mu_{\pi,c}(x)=\langle \pi(x) c\,\varphi,\varphi\rangle_\pi$$ for all $x\in M$, then in view of \eqref{bas4} and because of $c\in \pi(M)^{\,\prime\prime} \cap  \pi(M)^{\,\prime}$ the relation ${\mathcal S}_{\pi,M}(\mu_{\pi,c})=c\,{\mathcal S}_{\pi,M}(\mu)$ has to be fulfilled. Hence, if $c^\perp={\mathsf 1}-c$, then $\mu=\mu_{\pi,c}+\mu_{\pi,c^\perp}$ and ${\mathcal S}_{\pi,M}(\mu)={\mathcal S}_{\pi,M}(\mu_{\pi,c})+{\mathcal S}_{\pi,M}(\mu_{\pi,c^\perp})$.
\end{enumerate}
\end{remark}
Let $\nu,\varrho\in{\mathcal S}(M)$ and be $\{\pi,{\mathcal H}_\pi\}$ a unital $^*$-representation such both  $\pi$-fibres exist. Then, for $\varphi\in {\mathcal S}_{\pi,M}(\nu), \psi\in {\mathcal S}_{\pi,M}(\varrho)$ and central orthoprojection $c\in\pi(M)^{\,\prime\prime}$, $$c\varphi\in {\mathcal S}_{\pi,M}(\nu_{\pi,c}), \ c^\perp\varphi\in {\mathcal S}_{\pi,M}(\nu_{\pi,c^\perp}),\ c\psi\in {\mathcal S}_{\pi,M}(\varrho_{\pi,c}),\  c^\perp\psi\in {\mathcal S}_{\pi,M}(\varrho_{\pi,c^\perp})$$ hold. Thus, since $c\in \pi(M)^{\,\prime\prime} \cap  \pi(M)^{\,\prime}$ is fulfilled,
$
h_{\varphi,\psi}^\pi=h_{c\varphi,c\psi}^\pi+h_{c^\perp\varphi,c^\perp\psi}^\pi
$
follows. This together with   $$\bigl(\pi(M)^{\,\prime}\bigr)_1=\bigl(\pi(M)^{\,\prime}\,c\bigr)_1+\bigl(\pi(M)^{\,\prime}c^\perp\bigr)_1=
\bigl(\pi(M)^{\,\prime}\bigr)_1\,c+\bigl(\pi(M)^{\,\prime}\bigr)_1\,c^\perp$$
implies that $\|h_{\varphi,\psi}^\pi\|_1=\|h_{c\varphi,c\psi}^\pi\|_1+\|h_{c^\perp\varphi,c^\perp\psi}^\pi\|_1$. By Lemma \ref{bas3}\,\eqref{bas3aa} this means that
\begin{equation}\label{vN4a}
{F(M|\nu,\varrho)}={F(M|\nu_{\pi,c},\varrho_{\pi,c})}+{F(M|\nu_{\pi,c^\perp},\varrho_{\pi,c^\perp})}\,.
\end{equation}
In view of \eqref{eqdist3} the terms on the right hand side of \eqref{vN4a} can be also given as
 $${F\bigl(\pi(M)^{\,\prime\prime}|(\nu_{\pi,c})_\pi,(\varrho_{\pi,c})_\pi\bigr)}\text{ and } {F\bigl(\pi(M)^{\,\prime\prime}|(\nu_{\pi,c^\perp})_\pi,(\varrho_{\pi,c^\perp})_\pi\bigr)}$$ respectively. In order to evaluate the latter two expressions, in case of $0<c<{\mathsf 1}$ one also can refer either to the unital $^*$-representation $\pi(M)^{\,\prime\prime}\ni x\mapsto x\,c$ over the Hilbert spaces $c{\mathcal H}_\pi$ with the positive linear forms $\nu_c,\varrho_c$ given by $\nu_c(x)=\langle x c\varphi,c\varphi\rangle_\pi$ and $\varrho_c(x)=\langle x c\psi,c\psi\rangle_\pi$, for $x\in \pi(M)^{\,\prime\prime}c$, or to the unital $^*$-representation $\pi(M)^{\,\prime\prime}\ni x\mapsto x\,c^\perp$ over the Hilbert space $c^\perp{\mathcal H}_\pi$ with the positive linear forms $\nu_{c^\perp},\varrho_{c^\perp}$ given by $\nu_{c^\perp}(x)=\langle x c^\perp\varphi,c^\perp\varphi\rangle_\pi$ and $\varrho_{c^\perp}(x)=\langle x c^\perp\psi,c^\perp\psi\rangle_\pi$, for $x\in \pi(M)^{\,\prime\prime} c^\perp$, respectively. Having in mind this, and making appeal to the $vN$-algebras $M_c=\pi(M)^{\,\prime\prime} c$ and $M_{c^\perp}=\pi(M)^{\,\prime\prime} c^\perp$, then the relation \eqref{vN4a} conveniently can be rewritten as
\begin{equation}\label{vN4b}
{F(M|\nu,\varrho)}={F(M_c|\nu_c,\varrho_c)}+{F(M_{c^\perp}|\nu_{c^\perp},\varrho_{c^\perp})}\,.
\end{equation}
\begin{definition}[\cite{AlPe:00.2}]\label{twice1}
Let $M$ be a unital ${\mathsf C}^*$-algebra,  $\{\pi,{\mathcal H}_\pi\}$ a unital $^*$-representation such that the $\pi$-fibres of $\nu,\varrho\in {\mathcal S}(M)$  are non-trivial. Define $
    {\mathcal S}_{\pi,M}(\nu|\varrho)\subset {\mathcal S}_{\pi,M}(\nu)$ as follows:
${\mathcal S}_{\pi,M}(\nu|\varrho)=\bigl\{\varphi\in {\mathcal S}_{\pi,M}(\nu): \exists\, \psi\in {\mathcal S}_{\pi,M}(\varrho),\,F(M|\nu,\varrho)=\langle\psi,\varphi\rangle_\pi\bigr\}$.
\end{definition}
\noindent
The remarkable fact with this set of implementing vectors is that it is always non-trivial.
\begin{lemma}\label{nv}
For  $\nu,\varrho\in {\mathcal S}(M)$ with  nontrivial $\pi$-fibres ${\mathcal S}_{\pi,M}(\nu|\varrho)\not=\emptyset$ is fulfilled.
\end{lemma}
\begin{proof}
For the sake of later references this fact will be proved here, using a modified variant of the line of proving given in \cite[Appendix 7]{Albe:92.1}.
Assume $\varphi\in {\mathcal S}_{\pi,M}(\nu)$, $\psi\in {\mathcal S}_{\pi,M}(\varrho)$, and consider
$h_{\psi,\varphi}^\pi$ as defined in \eqref{hform}. Recall first some useful facts about this linear form.
Let the polar decomposition of $h_{\psi,\varphi}^\pi$ be given by
\begin{equation}\label{sub0}
h_{\psi,\varphi}^\pi=
|h_{\psi,\varphi}^\pi|\bigl((\cdot)v_{\psi,\varphi
}^\pi\bigr)
\end{equation}
By the polar decomposition
theorem for normal linear forms in $vN$-algebras, a partial isometry
$v=v_{\psi,\varphi}^\pi$ in and a normal positive linear form
$g=|h_{\psi,\varphi}^\pi|$ over $\pi(M)^{\,\prime}$ and obeying
$h_{\psi,\varphi}^\pi=g((\cdot)v)$ exist, and both
are unique subject to $v^*v=s(g)$, with the support  orthoprojection  $s(g)$
of $g$. Thus especially
\begin{equation}\label{sub0.5}
|h_{\psi,\varphi}^\pi|=h_{\psi,\varphi}^\pi\bigl((\cdot)v_{\psi,\varphi}^{\pi*}\bigr)=
\bigl\langle(\cdot)v_{\psi,\varphi}^{\pi*}\psi,\varphi\bigr
\rangle_\pi\,,
\end{equation}
from which in view of the above the following can be obtained:
\begin{equation}\label{sub1}
v_{\psi,\varphi}^{\pi*}v_{\psi,\varphi}^\pi=
s\bigl(|h_{\psi,\varphi}^\pi|\bigr)\leq p_\pi^{\,\prime}(\varphi)\,.
\end{equation}
Analogously, by polar decomposition of $h_{\psi,\varphi}^{\pi*}=h_{\varphi,\psi}^\pi$,
one has $v_{\varphi,\psi}^\pi=v_{\psi,\varphi}^{\pi*}$ and
\begin{equation}\label{sub1.5}
 |h_{\psi,\varphi}^{\pi*}|=h_{\varphi,\psi}^\pi\bigl((\cdot)v_{\psi,\varphi}^{\pi}\bigr)=
\bigl\langle(\cdot)v_{\psi,\varphi}^{\pi}\varphi,\psi\bigr
\rangle_\pi\,,
\end{equation}
from which analogously
\begin{equation}\label{sub2}
v_{\psi,\varphi}^\pi v_{\psi,\varphi}^{\pi*}=
s\bigl(|h_{\psi,\varphi}^{\pi*}|\bigr)\leq p_\pi^{\,\prime}(\psi)
\end{equation}
is obtained. In particular, with the Murray-von Neumann equivalence `$\sim$' one has
\begin{equation}\label{sub3}
s\bigl(|h_{\psi,\varphi}^\pi|\bigr)\sim
s\bigl(|h_{\psi,\varphi}^{\pi*}|\bigr)\,.
\end{equation}
Now, as a consequence of \eqref{sub0} and \eqref{sub0.5},
$\|h_{\psi,\varphi}^\pi\|_1=\||h_{\psi,\varphi}^\pi|\|_1$ holds. Hence, in the notations from above, since $|h_{\psi,\varphi}^\pi|$ is positive one concludes that
\begin{equation}\label{klar}
   \|h_{\psi,\varphi}^\pi\|_1=\||h_{\psi,\varphi}^\pi|\|_1=|h_{\psi,\varphi}^\pi|({\mathsf 1})=\bigl\langle v_{\psi,\varphi}^{\pi*}\psi,\varphi\bigr
\rangle_\pi=h_{\psi,\varphi}^\pi\bigl(v_{\psi,\varphi}^{\pi*}\bigr)\,.
\end{equation}
Remark that for a proof of the assertion it is  sufficient to find partial
isometries $v,w\in \pi(M)^{\,\prime}$ obeying
$v^*v\geq p_\pi^{\,\prime}(\varphi)$, $w^*w\geq p_\pi^{\,\prime}(\psi)$ and $h_{\psi,\varphi}^\pi\bigl(v_{\psi,\varphi}^{\pi*}\bigr)=
h_{\psi,\varphi}^\pi(v^*w)$. In fact, according to \eqref{bas4}, for
$\psi_0=w\psi$ and $\varphi_0=v\varphi$
one has $\psi_0\in {\mathcal S}_{\pi,M}(\varrho)$ and $\varphi_0\in {\mathcal S}_{\pi,M}(\nu)$.
By \eqref{klar}, in applying Lemma \ref{bas3} twice, we may conclude as follows\,:
\begin{equation}\label{zwie}
\begin{split}
\langle\psi_0,\varphi_0\rangle_\pi=h_{\psi_0,\varphi_0}^\pi({\mathsf 1})  =h_{\psi,\varphi}^\pi(v^*w) & =h_{\psi,\varphi}^\pi\bigl(v_{\psi,\varphi}^{\pi*}\bigr)\\
& =
\|h_{\psi,\varphi}^\pi\|_1={F(M|\nu,\varrho)}=\|h_{\psi_0,\varphi_0}^\pi\|_1
\end{split}
\end{equation}
Hence, $\varphi_0\in {\mathcal S}_{\pi,M}(\nu|\varrho)$, and which then will prove the assertion.

Recall that for a $vN$-algebra $N$ on a Hilbert space ${\mathcal H}_\pi$
there is a largest central orthoprojection
$z\in N\cap N^{\,\prime}$ such that
the $vN$-algebra $(N^{\,\prime})z$ over $z{\mathcal H}_\pi$ is finite.
Thus, in case of $z\not={\mathsf 1}$,
$(N^{\,\prime}){z^\perp}$ is properly infinite over $z^\perp{\mathcal H}_\pi$.
Also, for a central orthoprojection $c$,
$$(N^{\,\prime})c=\bigl(N\,c\bigr)^{\,\prime}$$
holds with respect to the Hilbert space $c{\mathcal H}_\pi$ (the outer commutant
refers to $c{\mathcal H}_\pi$).
Applying this to $N=\pi(M)^{\,\prime}$ with $c=z$ or $c=z^\perp$, and having in mind
Remark \ref{rem1}\,\eqref{rem11}, Remark \ref{vN}\,\eqref{vN3} and relation \eqref{vN4b} we see that we can content ourselves  with constructing the partial
isometries in question only for
$vN$-algebras $M_c$ and $M_{c^\perp}$ in  \eqref{vN4b} with finite or  properly infinite commutant, respectively.
In fact, by the above and in making appeal to the simplified notations according to Remark \ref{rem1}\,\eqref{rem12} and the notions used in context of \eqref{vN4b} we then could be assured that $\varphi_c\in {\mathcal S}_{M_c}(\nu_c)$ and $\psi_c\in {\mathcal S}_{M_c}(\varrho_c)$ with
${F(M_c|\nu_c,\varrho_c)}=\langle \psi_c, \varphi_c\rangle_\pi$, and $\varphi_{c^\perp}\in {\mathcal S}_{M_{c^\perp}}(\nu_{c^\perp})$ and $\psi_{c^\perp}\in {\mathcal S}_{M_{c^\perp}}(\varrho_{c^\perp})$ with
${F(M_{c^\perp}|\nu_{c^\perp},\varrho_{c^\perp})}=\langle \psi_{c^\perp}, \varphi_{c^\perp}\rangle_\pi$ existed. By Remark \ref{vN}\,\eqref{vN3}, in defining vectors $\varphi_0=\varphi_c+\varphi_{c^\perp}$ and $\psi_0=\psi_c+\psi_{c^\perp}$ will yield implementing vectors of the $\pi$-fibres of $\nu$ and $\varrho$, respectively, and which according to \eqref{vN4b} then will obey
${F(M|\nu,\varrho)}=\langle \psi_c, \varphi_c\rangle_\pi+\langle  \psi_{c^\perp},\varphi_{c^\perp}\rangle_\pi=\langle\psi_0,\varphi_0\rangle_\pi$.

In line with this, for the rest of the proof let $M$ be a $vN$-algebra acting on the Hilbert space ${\mathcal H}$, with $M^{\,\prime}$ either finite or
properly infinite, and let $\nu, \varrho$ be positive linear forms on  $M$ implemented by vectors $\varphi, \psi$ in ${\mathcal H}$.
In adhering to the notational conventions of
Remark \ref{rem1}\,\eqref{rem12}, we are going to construct $v$ and $w$ under these premises now.

Suppose $M^{\,\prime}$ to be finite.
By \eqref{sub3} we have
$s(|h_{\psi,\varphi}^*|)\sim
s(|h_{\psi,\varphi}|)$.
In a finite $vN$-algebra, by another standard fact, see \cite[2.4.2.]{Saka:71}, this condition implies
$$
s(|h_{\psi,\varphi}^*|)^\perp
\sim s(|h_{\psi,\varphi}|)^\perp
$$
Hence, there is $m\in M^{\,\prime}$ with $m^*m=s(|h_{\psi,\varphi}|)^\perp$
and  $m m^*=s(|h_{\psi,\varphi}^*|)^\perp$. Define $w=v_{\psi,\varphi}^*+m^*$.
Then, $w$ is unitary, $w\in {\mathcal U}(M^{\,\prime})$, with
$w s(|h_{\psi,\varphi}^*|)=v_{\psi,\varphi}^*$. Since by
polar decomposition $h_{\psi,\varphi}(w)=h_{\psi,\varphi}(v_{\psi,\varphi}^*)$
is fulfilled, $v={\mathsf 1}$ can be chosen, then.

If $M^{\,\prime}$ is properly infinite, then there exists
$p\in M^{\,\prime}$ with $p\sim p^\perp\sim {\mathsf 1}$.
By \eqref{sub1},
$p^{\,\prime}(\varphi)-s(|h_{\psi,\varphi}|)\in M^{\,\prime}$ is an
orthoprojection. Hence
$p^{\,\prime}(\varphi)-s(|h_{\psi,\varphi}|)
\prec p$ and $p'(\psi)\prec p^\perp$ (`$\prec$' be the
Murray-von Neumann comparability relation). Hence,
there are
$v_1,v_2\in M^{\,\prime}$ with
$v_1^*v_1=p^{\,\prime}(\varphi)-s(|h_{\psi,\varphi}|)$,
$v_1v_1^*\leq p$, $v_2^*v_2=p'(\psi)$, $v_2 v_2^*\leq p^\perp$. Define
$w=v_2$, $v=v_1+v_2 v_{\psi,\varphi}$. Since $v_1^*v_2=v_2^*v_1={\mathsf 0}$ holds, the
above and \eqref{sub2} and \eqref{sub1} imply
$$v^*v=v_1^*v_1+v_{\psi,\varphi}^*v_2^*v_2v_{\psi,\varphi}=
p'(\varphi)-s(|h_{\psi,\varphi}|)+
s(|h_{\psi,\varphi}|)=p'(\varphi)$$ and $w^*w=p'(\psi)$.
By \eqref{sub2} from this $v^*w=
v_{\psi,\varphi}^*v_2^*v_2=v_{\psi,\varphi}^*p'(\psi)=
v_{\psi,\varphi}^*s(|h_{\psi,\varphi}^*|)=
v_{\psi,\varphi}^*$ is obtained.
\end{proof}
\begin{theorem}\label{positiv1}
Let $M$ be a unital ${\mathsf C}^*$-algebra,  $\{\pi,{\mathcal H}_\pi\}$ a unital $^*$-representation of $M$ on a Hilbert space ${\mathcal H}_\pi$, with the $\pi$-fibres of
$\nu,\varrho\in {\mathcal S}(M)$ both existing. There are $\varphi_0\in {\mathcal S}_{\pi,M}(\nu)$,
$\psi_0\in {\mathcal S}_{\pi,M}(\varrho)$ with $F(M|\nu,\varrho)=\langle \psi_0,\varphi_0\rangle_\pi$ and  $d_B(M|\nu,\varrho)=\|\psi_0-\varphi_0\|_\pi$, respectively.
\end{theorem}
\begin{proof}
According to the proof of Lemma \ref{nv}, see \eqref{zwie}, we have  $\psi_0\in {\mathcal S}_{\pi,M}(\varrho)$ and $\varphi_0\in {\mathcal S}_{\pi,M}(\nu)$ obeying
 $h_{\psi_0,\varphi_0}^\pi({\mathsf 1})=
\|h_{\psi_0,\varphi_0}^\pi\|_1$, from which
$h_{\psi_0,\varphi_0}^\pi\geq {\mathsf 0}$ follows. In view of Lemma \ref{bas3}\,\eqref{bas3aa} and
Theorem \ref{bas5}\,\eqref{bas53} the assertion then follows.
\end{proof}
Relating the structure of the set defined in Definition \ref{twice1} the following result subsequently will be of importance.
\begin{lemma}\label{king}
Let $M$ be a unital ${\mathsf C}^*$-algebra, and  $\omega,\, \nu,\, \varrho\in {\mathcal S}(M)$ three given states. There exists a unital $^*$-representation $\{\pi,{\mathcal H_\pi}\}$ such that ${\mathcal S}_{\pi,M}(\omega|\nu)\cap{\mathcal S}_{\pi,M}(\omega|\varrho)\not=\emptyset$
\end{lemma}
The proof will be traced back to
a special situation of Theorem \ref{positiv1} and arising if
 $M=R$ is a $vN$-algebra
 acting in standard form on a separable
 Hilbert space ${\mathcal H}$, with cyclic and
 separating vector ${\mathit \Omega}\in {\mathcal H}$, and if the normal state space ${\mathcal S}_0(R)$ is considered. Then, each normal state is a vector state  on ${\mathcal H}$. Thus, the fibre ${\mathcal S}_R({\nu})$
 with respect to the identity representation $\pi={\mathsf{id}}$ is non--void, for each $\nu\in {\mathcal S}_0(R)$. Let $R_{*+}$ be the positive portion of the predual $R_*$ of $R$.
 It is known from modular theory \cite{Arak:74,Conn:74,Haage:75} that there exists a map
 $$
 {R_{*+}} \ni {\nu} \mapsto {\xi_{\nu}} \in
  {{\mathcal P}_{{\mathit{\Omega}}}^{\natural}},
  $$
 into the
       self-dual, natural positive cone   ${\mathcal P}_{{\mathit{\Omega}}}^{\natural}$ in ${\mathcal H}$ associated
       with the pair $\{ R, {\mathit{\Omega}} \}$,
which  is onto and through which
   a homeomorphism is established when both the cones are considered with
  their respective uniform topologies. Thereby, $\xi_{\nu}\in  {\mathcal S}_R({\nu})$ is the unique implementing vector of $\nu$ within the natural positive cone.
  Relating the structure of the set ${\mathcal S}_R(\nu|\varrho)$ defined in Definition \ref{twice1}, in the special context at hand  for any two normal states $\nu, \varrho\in {\mathcal S}_0(R)$ one has the following version of Lemma \ref{nv} proved first in  \cite{Arak:72}:
 \begin{lemma}\label{prinz}
 Let $\omega\in {\mathcal S}_0(R)$ be arbitrarily chosen. For  each normal state $\varrho$ one has
 \begin{equation}\label{stfvNa}
  \xi_\omega\in {\mathcal S}_R(\omega|\varrho)
  \end{equation}
  Thus, there exists $\psi_0\in {\mathcal S}_R(\varrho)$ such that
   $\langle\psi_0,\xi_\omega\rangle=F(R|\omega,\varrho)$
   can be fulfilled.
 \end{lemma}
\begin{proof} We follow along the lines of proving in \cite[Corollary 6.2, \S  7]{Albe:97.1, Albe:96.1}.
Since all states involved are vector state, we will work within the identity representation $\pi={\mathsf {id}}$ and will make use of the simplified notations as agreed in Remark \ref{rem1}\,\eqref{rem12}. As explained above,
$\xi_\omega\in {\mathcal S}_R(\omega)$ is fulfilled.  Let  $\psi\in {\mathcal S}_R(\varrho)$ chosen. Then, by
 the polar decomposition theorem, see  \eqref{sub0}--\eqref{sub2},  on the commutant $R^{\,\prime}$
 $$h_{\psi,\xi_\omega}=
{h_{\psi,\xi_\omega}}{({(\cdot)}s({|h_{\psi,\xi_\omega}^*|}))}$$
is fulfilled. Equivalently,
${h_{\psi,\xi_\omega}}{({(\cdot)}
{s({|h_{\psi,\xi_\omega}^*|})}^{\perp})}={\mathsf 0}$ holds, and for all $ x, y \in R^{\,\prime}$,
$$\bigl\langle x{s({|h_{\psi,\xi_\omega}^*|})}^{\perp}\psi,y
\xi_\omega\bigr\rangle=0$$
follows. Hence,
${p({s({|h_{\psi,\xi_\omega}^*|})}^{\perp}\psi)}
\perp p(\xi_\omega)$, which equivalently can  be expressed as
\begin{equation}\label{ortho0}
{p({s({|h_{\psi,\xi_\omega}^*|})}^{\perp}\psi)}
\leq p(\xi_\omega)^{\perp},
\end{equation}
 Let us consider the positiv linear forms $\varrho_0$, $\varrho_1$ implemented by $s({|h_{\psi,\xi_\omega}^*|})\psi$, $s({|h_{\psi,\xi_\omega}^*|})^{\perp}\psi$ over $R$, respectively, that is,
\[
s({|h_{\psi,\xi_\omega}^*|})\psi\in {\mathcal S}_R(\varrho_0),\,s({|h_{\psi,\xi_\omega}^*|})^{\perp}\psi\in {\mathcal S}_R(\varrho_1)
\]
Then, $\varrho=\varrho_0+\varrho_1$. Note that $v_{\psi,\xi_\omega}\in R^{\,\prime}$ and  ${v_{\psi,\xi_\omega}^{}}v_{\psi,\xi_\omega}^*=
s({|h_{\psi,\xi_\omega}^*|})$ imply that
\[
v_{\psi,\xi_\omega}^*\psi\in {\mathcal S}_R(\varrho_0)
\]
On the other hand, since
$v^*_{\psi,\xi_\omega}{v^{}_{\psi,\xi_\omega}}=
s({|{h_{\psi,\xi_\omega}}|}) \leq p^{\,\prime}(\xi_\omega)$ holds, we
get
\begin{equation}\label{ortho1}
{s({|{h_{\psi,\xi_\omega}}|})}^{\perp}
\geq {p^{\,\prime}(\xi_\omega)}^{\perp}.
\end{equation}
Owing to
${p({s({|h_{\psi,
\xi_\omega}^*|})}^{\perp}\psi)}=
p({\xi}_{\varrho_1})$, where  $\xi_{\varrho_1}\in {\mathcal P}_{\mathit{\Omega}}^{\natural}$ is the uniquely determined implementing vector of $\varrho_1$
within the modular cone, from (\ref{ortho0}) conclude that
\begin{equation}\label{basis}
p({\xi}_{\varrho_1})\leq p(\xi_\omega)^{\perp}
\end{equation}
holds. Let $J$ be the modular conjugation operator associated to the pair $\{R,\mathit{\Omega}\}$. Then, in using Tomita's theorem \cite{Take:70}
saying that $JRJ=R^{\,\prime}$ holds, one easily sees that, for any given vector
$\zeta \in {\mathcal H}$,
the relation $Jp(\zeta)J=p^{\,\prime}(J\zeta)$ is valid.  From this together with the fact that $J\xi=\xi$ holds,
for any vector $\xi$ belonging to the modular cone, we infer that
$$Jp(\xi_{\varrho_1})J=p^{\,\prime}(\xi_{\varrho_1}),\,Jp(\xi_{\omega})J=p^{\,\prime}(\xi_{\omega}) $$
Application to (\ref{basis})
yields
$p^{\,\prime}({\xi}_{\varrho_{1}})\leq
{p^{\,\prime}({\xi}_{\omega})}^{\perp}$, which in view of \eqref{ortho1} implies
\begin{equation}\label{basis1}
{s({|{h_{\psi,\xi_{\omega}}}|})}^{\perp}
\geq {p^{\,\prime}({\xi}_{\varrho_1})}\,.
\end{equation}
From this two things can be seen. Firstly, owing to
$v_{\psi,\xi_{\omega}}^*{v_{\psi,\xi_{\omega}}}=
s({|{h_{\psi,\xi_{\omega}}}|})$ it follows that
$\langle xv_{\psi,\xi_{\omega}}^*\psi,y{\,}{\xi}_{\varrho_1}
\rangle=0$, for all $ x,y \in R$. Hence, for any $x \in R$,
\begin{eqnarray*}
\langle x(v_{\psi,\xi_{\omega}}^*\psi+{{\xi}_{\varrho_1}}
),v_{\psi,\xi_{\omega}}^*\psi+{{\xi}_{\varrho_1}}\rangle & = &
\langle xv_{\psi,\xi_{\omega}}^*\psi,v_{\psi,\xi_{\omega}}^*\psi\rangle +
\langle x{\xi}_{\varrho_1},{\xi}_{\varrho_1}
\rangle\\
& = & {\varrho_0}(x)+{\varrho_1}(x)\\
& = & \varrho(x)
\end{eqnarray*}
This proves that $\psi_0=v_{\psi,\xi_{\omega}}^*\psi+{{\xi}_{\varrho_1}} \in {\mathcal S}_R(\varrho)$. Secondly, from \eqref{basis1}
$s({|{h_{\psi,\xi_{\omega}}}|}){\psi_0}=
v_{\psi,\xi_{\omega}}^*\psi$ follows. Finally, the following conclusions can
be carried out:
$$\langle \psi_0,\xi_{\omega}\rangle=
\langle v_{\psi,\xi_{\omega}}^*\psi,\xi_{\omega}\rangle +
\langle {\xi}_{\varrho_1},
\xi_{\omega}\rangle=\langle v_{\psi,\xi_{\omega}}^*\psi,\xi_{\omega}\rangle=
{|{h_{\psi,\xi_{\omega}}}|}({\mathsf 1})=
{\|\,{|{h_{\psi,\xi_{\omega}}}|}\,\|}_1 =
{\|{h_{\psi,\xi_{\omega}}}\|}_1$$ According to Lemma \ref{bas3}\,\eqref{bas3aa} from the latter
$\langle \psi_0,\xi_{\omega}\rangle=F(R|\omega,\varrho)$ follows.
\end{proof}

\begin{proof}[Proof of Lemma \ref{king}]
Define another state $\mu$ by
$$\mu=\frac{1}{3}(\nu+\omega+\varrho)$$
and consider  a  $^*$-representation
$\{\pi,{\mathcal H_\pi}\}$ of $M$ such that the fibres of the states $\omega,\, \nu,\, \varrho$ and $\mu$ all exist (e.g., the direct sum of the  cyclic representations of  $\omega,\, \nu,\, \varrho$ may be taken). Let $\mathit{\Omega}\in {\mathcal S}_{\pi,M}(\mu)$, and be $N=\pi(M)^{\,\prime\prime}$. Let  $\tilde{\omega},\, \tilde{\nu},\, \tilde{\varrho}$ be the uniquely determined vector states over $N$ such that $\omega(x)=\tilde{\omega}(\pi(x))$,  $\nu(x)=\tilde{\nu}(\pi(x))$,   $\varrho(x)=\tilde{\varrho}(\pi(x))$, and $\mu(x)=\tilde{\mu}(\pi(x))$, for each $x\in M$. By the definition of $\mu$, each of the states $\omega$,  $\nu$, and $\varrho$ is dominated by the positive linear form $3 \mu$. Clearly, with the help of a Kaplansky density argument and due to the normality of the vector states, the relation of domination extends to the vector states, too. Hence, by the Radon-Nikodym theorem of S.Sakai \cite{Saka:71}, in case of $\tilde{\nu}$ e.g., we have $a\in N_+$ such that $\varphi=a\mathit{\Omega}\in {\mathcal S}_{\pi,M}(\nu)$ can be chosen. Thus, the implementing vector $\varphi$ is obeying both $p_\pi(\varphi)=s(\tilde{\nu})\leq s(\tilde{\mu})=
p_\pi(\mathit{\Omega})$ and $p_\pi^{\,\prime}(\varphi)\leq p_\pi^{\,\prime}(\mathit{\Omega})$. Analogously we may argue for implementing vectors $\eta$ and $\psi$ of the states $\omega$ and $\varrho$, respectively,  and therefore can summarize that implementing vectors of $\omega$,  $\nu$, and $\varrho$ can be chosen as to obey $$\varphi,\eta,\psi\in {\mathcal H}=p_\pi^{\,\prime}(\mathit{\Omega})p_\pi(\mathit{\Omega}){\mathcal H}_\pi$$
Let us consider the $vN$-algebra $$R=p_\pi(\mathit{\Omega})Np_\pi(\mathit{\Omega})p_\pi^{\,\prime}(\mathit{\Omega})$$
acting over the Hilbert space ${\mathcal H}$. Obviously, for the action of $R$ over ${\mathcal H}$ the vector $\mathit{\Omega}$ is a cyclic and separating
vector. Let $\nu_0$, $\omega_0$ and $\varrho_0$ the vector states generated over $R$ by $\varphi$, $\eta$ and $\psi$, respectively. Then, Lemma \ref{prinz} may be applied to both pairs $\{\omega_0,\nu_0\}$ and $\{\omega_0,\varrho_0\}$ of states over the standard form $vN$-algebra $R$. In line with this, we then find $\varphi_0\in {\mathcal S}_R(\nu_0)$ and $\psi_0\in {\mathcal S}_R(\varrho_0)$ such that
\begin{equation}\label{king1}
\langle \varphi_0,\xi_{\omega_0}\rangle_\pi=F(R|\omega_0,\nu_0),\, \langle \psi_0,\xi_{\omega_0}\rangle_\pi=F(R|\omega_0,\varrho_0)
\end{equation}
both are fulfilled, with respect to the special implementing vector $ \xi_{\omega_0} \in {\mathcal P}_{\mathit{\Omega}}^{\natural}$ of $\omega_0$. According to Lemma \ref{bas3}\,\eqref{bas3aa} and Theorem \ref{bas5}\,\eqref{bas53} from this we have to conclude that
\begin{equation}\label{king2}
\langle y \varphi_0,\xi_{\omega_0}\rangle_\pi\geq 0,\,\langle y \psi_0,\xi_{\omega_0}\rangle_\pi\geq 0
\end{equation}
holds, for any $y\in R_+^{\,\prime}$. One easily infers that the following holds, for each $x\in N_+^{\,\prime}$,
\[
p_\pi^{\,\prime}(\mathit{\Omega}) x p_\pi^{\,\prime}(\mathit{\Omega})p_\pi(\mathit{\Omega})\in R_+^{\,\prime}
\]
Accordingly, \eqref{king2} then implies that
$
h^\pi_{\varphi_0,\xi_{\omega_0}}(x)\geq 0,\,h^\pi_{\psi_0,\xi_{\omega_0}}(x)\geq 0
$. Thus,
\begin{equation}\label{king3}
h^\pi_{\varphi_0,\xi_{\omega_0}}\geq {\mathsf 0},\, h^\pi_{\psi_0,\xi_{\omega_0}}\geq {\mathsf 0}
\end{equation}
both hold over $N^{\,\prime}$. Also note that
since by $\varphi, \,\varphi_0\in {\mathcal H}$ the same vector states over $R$ are implemented, by definition of $R$ obviously the same vector states over $N$ are implemented. That is, $\varphi_0\in {\mathcal S}_{\pi,M}(\nu)$ has to be fulfilled, too. Analogously we have $\psi_0\in {\mathcal S}_{\pi,M}(\varrho)$ and $\xi_{\omega_0}\in {\mathcal S}_{\pi,M}(\omega)$. The latter and \eqref{king3} by  Lemma \ref{bas3}\,\eqref{bas3aa} and Theorem \ref{bas5}\,
\eqref{bas53} imply
\[
\langle \varphi_0,\xi_{\omega_0}\rangle_\pi=F(M|\omega,\nu),\, \langle \psi_0,\xi_{\omega_0}\rangle_\pi=F(M|\omega,\varrho)
\]
and which is equivalent to the assertion of Lemma \ref{king}.
\end{proof}
\begin{corolla}\label{uppdiff}
 Let $M$ be a unital ${\mathsf C^*}$-algebra. For any $\nu,\varrho\in {\mathcal S}(M)$ one has
 \begin{equation}\label{uppdiff1}
    d_B(M|\nu,\varrho)\leq \sqrt{\|\nu-\varrho\|_1}
 \end{equation}
\end{corolla}
\begin{proof}
Let $\{\pi,{\mathcal H}_\pi\}$ a unital $^*$-representation such that the $\pi$-fibres of both states exist. Consider the $vN$-algebra $N=\pi(M)^{\,\prime\prime}$. By Lemma \ref{eqdist}\,\eqref{eqdist2} one then has
$$d_B(M|\nu,\varrho)=d_B(N|\nu_\pi,\nu_\pi)$$ with the uniquely determined vector states $\nu_\pi$ and $\varrho_\pi$ over $N$ generated by implementing vectors of $\nu$ and $\varrho$ in respect of  $\{\pi,{\mathcal H}_\pi\}$. Corollary \ref{ctraeger} can be applied, and with $p=s(\nu_\pi)\vee s(\varrho_\pi)$, $N_p=p\,Np$ and $\tilde{\nu}={\nu_\pi}|N_p$, $\tilde{\varrho}={\varrho_\pi}|N_p$ we are arriving at
$$
  d_B(N|\nu_\pi,\nu_\pi)=d_B(N_p|\tilde{\nu},\tilde{\varrho})
$$
Note that by construction $N_p$ is a $vN$-algebra over $p\,{\mathcal H}_\pi$ possessing a faithful normal state. Hence it is $^*$-isomorphic to a $vN$-algebra $\hat{N}$ acting in standard form over some Hilbert space.  Accordingly, each normal state $\hat{\omega}$ over $\hat{N}$ can be implemented by a unique vector $\xi_{\hat{\omega}}$ within the natural positive cone of this standard form representation. Then, by modular theory it is known that for each two normal states the following estimate holds
$$
    \|\xi_{\hat{\omega}}-\xi_{\hat{\mu}}\|\leq \sqrt{\|\hat{\omega}-\hat{\mu}\|_1}
$$
This especially applies to the uniquely determined normal states $\hat{\nu}$ and $\hat{\varrho}$ corresponding to $\tilde{\nu}$ and $\tilde{\varrho}$. Clearly, since Lemma \ref{eqdist} can be applied to this situation, together with the previous and Definition \ref{budi} we then may conclude that
\begin{equation*}
  d_B(M|\nu,\varrho)=d_B(N_p|\tilde{\nu},\tilde{\varrho})=d_B(\hat{N}|\hat{\nu},\hat{\varrho})\leq  \|\xi_{\hat{\nu}}-\xi_{\hat{\varrho}}\|\leq \sqrt{\|\hat{\nu}-\hat{\varrho}\|_1} =\sqrt{\|\tilde{\nu}-\tilde{\varrho}\|_1}
\end{equation*}
But note that since obviously $\|\tilde{\nu}-\tilde{\varrho}\|_1=\|\nu_\pi-\varrho_\pi\|_1$ is fulfilled, by an application of Lemma \ref{eqdist}\,\eqref{eqdist1} then $\|\tilde{\nu}-\tilde{\varrho}\|_1=\|\nu-\varrho\|_1$ follows, and thus \eqref{uppdiff1} is seen to hold.
\end{proof}
\subsubsection{Uniqueness results for vectors implementing fidelity}\label{uniass}
Under certain circumstances along with the  pairs of implementing vectors mentioned on in Theorem \ref{positiv1} characteristic uniqueness assertions are arising and which later on will prove useful.
\begin{definition}\label{Fset}
Let $M$ be a unital ${\mathsf C}^*$-algebra.  For given states $\nu,\varrho\in {\mathcal S}(M)$, let ${\mathfrak F}(M|\nu,\varrho)\subset M^*$ be defined as set of all bounded linear forms $f$ over $M$ such that
\begin{enumerate}
\item\label{unique1}
$|f(y^*x)|^2\leq \nu(y^*y)\,\varrho(x^*x),\text{ for all }x,y\in M$;
\item\label{unique2}
$f({\mathsf 1})=F(M|\nu,\varrho)$.
\end{enumerate}
\end{definition}
By Theorem \ref{positiv1}, given any unital $^*$-representation
$\{\pi,{\mathcal H}_\pi\}$ such that the $\pi$-fibres of both states are non-trivial, we can be assured that $\varphi\in {\mathcal S}_{\pi,M}(\nu)$, $\xi\in {\mathcal S}_{\pi,M}(\varrho)$ exist which are obeying  $d_B(M|\nu,\varrho)=\|\xi-\varphi\|_\pi$, or equivalently, for which  \begin{equation}\label{Fset1}
F(M|\nu,\varrho)=\langle\xi,\varphi\rangle_\pi=\|h^\pi_{\xi,\varphi}\|_1
\end{equation}
holds. As we know, the latter is equivalent to $h=h^\pi_{\xi,\varphi}\geq {\mathsf 0}$. Consider $f^\pi_{\xi,\varphi}$ given as
\begin{equation}\label{unique3}
f^\pi_{\xi,\varphi}(x)=\langle\pi(x)\xi,\varphi\rangle_\pi,\ \forall x\in M
\end{equation}
In view of \eqref{Fset1} and in applying the Cauchy-Schwarz inequality with $x,y\in M$,
\[
|\langle\pi(x)\xi,\pi(y)\varphi\rangle_\pi|^2\leq \|\pi(x)\xi\|_\pi^2\, \|\pi(y)\varphi\|_\pi^2=\varrho(x^*x)\,\nu(y^*y)
\]
we infer that each of the linear forms \eqref{unique3} is a solutions of the conditions \eqref{unique1}--\eqref{unique2} of
Definition \ref{Fset}. Thus, ${\mathfrak F}(M|\nu,\varrho)\not=\emptyset$. Also, following the idea in \cite{Arak:72}, in the context of $\pi$, $\varphi$ and $\xi$ as previously specified, and associated to the given $\nu,\varrho$ let us consider another positive linear form $\varrho_\nu$ over $M$ and defined implicitly by
\begin{subequations}\label{diffreteile2aa}
\begin{equation}\label{diffreteile1}
s(h)^\perp\xi\in {\mathcal S}_{\pi,M}(\varrho_\nu)
\end{equation}
That by this $\varrho_\nu$ is unambigously determined is a consequence of the following.
\begin{lemma}\label{diffreteile}
$\mu=\varrho_\nu$  is  the largest positive linear form satisfying $\mu\leq \varrho$ and $\mu\perp\nu$.
\end{lemma}
\begin{proof}
We are going to show that  $\varrho_\nu\leq \varrho$ and $\nu\perp\varrho_\nu$ are fulfilled  and that the conditions $\mu\leq \varrho$ and $\mu\perp\nu$ for $\mu\in M^*_+$ together imply that $\mu\leq \varrho_\nu$. Let $\eta=s(h)^\perp\xi\in {\mathcal H}_\pi$ and consider the $vN$-algebra $N=\pi(M)^{\,\prime\prime}$. By definition of $\eta$ and \eqref{diffreteile1}, since  $s(h)^\perp\in (N^\prime)_1$ holds, $\varrho_\nu\leq \varrho$ has to be fulfilled.
Moreover, since
$$h^\pi_{\eta,\varphi}=h^\pi_{\xi,\varphi}((\cdot) s(h)^\perp)=h((\cdot)s(h)^\perp))={\mathsf 0}$$
holds,
we have $F(M|\nu,\varrho_\nu)=0$, by Lemma \ref{bas3}\,\eqref{bas3aa}. In view of Corollary \ref{subadd}\,\eqref{subadd2} from this $\nu\perp\varrho_\nu$ follows. Now, suppose $\mu\in M^*_+$ such that $\mu\leq \varrho$ and $\mu\perp\nu$. With the help of the von-Neumann density theorem the former condition implies that $\mu_\pi\leq \varrho_\pi$ over $N$. Hence, $\|x\eta\|_\pi\leq \|x\xi\|_\pi$ for all $x\in N$. From this by standard conclusions $\eta=k\xi$ follows, for some $k\in (N^\prime)_1$. Note that then  $$\eta'=|k|\xi\in {\mathcal S}_{\pi,M}(\mu)$$ is satisfied, too. Thus, since by assumption  $\mu\perp\nu$ holds, we have $p_\pi(\eta')\perp p_\pi(\varphi)$. The latter condition owing to $|k|\in N^\prime$ especially implies that
\[
0=\langle |k|\eta',\varphi\rangle_\pi=\langle |k|^2\xi,\varphi\rangle_\pi=h(|k|^2)
\]
Hence, $s(h)|k|^2s(h)={\mathsf 0}$. From this  $|k|s(h)=s(h)|k|={\mathsf 0}$ is inferred to hold. Since $|k|\leq {\mathsf 1}$, from  the previous finally $|k|^2\leq s(h)^\perp$ can be followed. Thus, for $x\in M$,
\begin{eqnarray*}
\mu(x^*x)& =& \|\pi(x)\eta'\|_\pi^2=\||k|\pi(x)\xi\|_\pi^2=\langle|k|^2 \pi(x)\xi,\pi(x)\xi\rangle_\pi\\
&\leq &\langle s(h)^\perp \pi(x)\xi,\pi(x)\xi\rangle_\pi=\|\pi(x)s(h)^\perp\xi\|_\pi^2\\
&\leq &\varrho_\nu(x^*x)
\end{eqnarray*}
follows, that is, $\mu\leq \varrho_\nu$ is seen.
\end{proof}
By symmetry, to $\nu,\varrho\in {\mathcal S}(M)$ also  a largest positive linear form $\mu=\nu_\varrho$ satisfying $\mu\leq \nu$ and $\mu\perp\varrho$ has to exist. In terms of the above notions $\nu_\varrho$ is given by
\begin{equation}\label{diffreteile2ab}
s(h)^\perp\varphi\in {\mathcal S}_{\pi,M}(\nu_\varrho)
\end{equation}
\end{subequations}
Now, uniqueness assertions relating ${\mathfrak F}(M|\nu,\varrho)$ will be in the focus.
To prepare for this,
recall the notion of disjointness of  positive linear forms on a unital ${\mathsf C}^*$-algebras $M$.
According to this notion, $\mu$ and $\omega$ are termed `(mutually) disjoint' if, within each unital $^*$-representation $\{\pi,{\mathcal H}_\pi\}$ of $M$ such that the $\pi$-fibre of  $\mu+\omega$ exists, we have
 \[
 {\mathcal S}_{\pi,M}(\mu) =z{\mathcal S}_{\pi,M}(\mu+\omega),\ {\mathcal S}_{\pi,M}(\omega) =z^\perp{\mathcal S}_{\pi,M}(\mu+\omega)
 \]
for a central orthoprojection $z\in N\cap N^{\prime}$, where $N$ is the $vN$-algebra $N=\pi(M)^{\prime\prime}$. Obviously this is equivalent to the requirement that the following be fulfilled
\[
{\mathcal S}_{\pi,M}(\mu+\omega)={\mathcal S}_{\pi,M}(\mu)+{\mathcal S}_{\pi,M}(\omega)
\]
For $\eta\in {\mathcal H}_\pi$, let $c_\pi(\eta)$ be the orthoprojection onto the closed linear subspace \begin{equation}\label{cdef}
[\pi(M)^{\,\prime}\pi(M)\eta]\subset {\mathcal H}_\pi
\end{equation}
Then, it is easily seen from \eqref{cdef} that $z=c_\pi(\eta)\in N\cap N^\prime$ is the smallest central orthoprojection obeying $z\eta=\eta$, and is also the smallest central orthoprojection obeying $z\geq p_\pi(\eta)$ and  $z\geq p^{\,\prime}_\pi(\eta)$. More generally, at an orthoprojection $q\in N$ or $q\in N^\prime$, the operation $c_\pi[\cdot]$ of going to  the central support of $q$ takes as value $z=c_\pi[q]$, the smallest central orthoprojection $z$ obeying  $z\geq q$. In terms of these notions, for $\eta\in {\mathcal S}_{\pi,M}(\mu)$ and $\chi\in {\mathcal S}_{\pi,M}(\omega)$, owing to $s(\mu_\pi)=p_\pi(\eta)$ and
$
c_\pi(\eta)=c_\pi[p_\pi(\eta)]=c_\pi[p^{\,\prime}_\pi(\eta)]
$
 (and an analogous identity relating to  $\chi$ and $\omega$, respectively)
 disjointness among $\mu$ and $\omega$ equivalently can be stated as orthogonality of the respective central supports, that is ${\mathsf 0}=c_\pi(\eta)\,c_\pi(\chi)$, or  ${\mathsf 0}=c_\pi[p^{\,\prime}_\pi(\eta)]\,c_\pi[p^{\,\prime}_\pi(\chi)]$, or  ${\mathsf 0}=c_\pi[p_\pi(\eta)]\,c_\pi[p_\pi(\chi)]$, which latter is the same as  ${\mathsf 0}=c_\pi[s(\mu_\pi)]\,c_\pi[s(\omega_\pi)]$, with the central supports of the support orthoprojections $s(\mu_\pi)$ and $s(\omega_\pi)$ of $\mu_\pi$ and $\omega_\pi$, respectively.
The main goal of this paragraph is the proof of the following result.
\begin{theorem}\label{unique}
For  $\nu,\varrho\in {\mathcal S}(M)$ the following are mutually equivalent:
\begin{enumerate}
\item \label{u1}
$\varrho_\nu$ and $\nu_\varrho$ are mutually disjoint;
\item \label{u2}
${\mathfrak F}(M|\nu,\varrho)=\{f\}$, with a uniquely determined linear form $f$.
\end{enumerate}
Moreover, the unique linear form $f\in {\mathfrak F}(M|\nu,\varrho)$  can be written as
\begin{enumerate}
\setcounter{enumi}{2}
\item \label{u3}
$f=f^\pi_{\xi,\varphi}$
\end{enumerate}
with respect to any unital $^*$-representation
$\{\pi,{\mathcal H}_\pi\}$ such that  the $\pi$-fibres of $\nu$ and $\varrho$ exist, and with
$\varphi\in {\mathcal S}_{\pi,M}(\nu)$,
$\xi\in {\mathcal S}_{\pi,M}(\varrho)$
chosen such that $F(M|\nu,\varrho)=\langle\xi,\varphi\rangle_\pi$.
\end{theorem}
The proof will be divided into several steps. To start with, we are going to derive a representation formula for the linear forms of ${\mathfrak F}(M|\nu,\varrho)$, for a non-restricted choice of states $\nu,\varrho\in {\mathcal S}(M)$. In accordance with Theorem \ref{positiv1}, suppose $\varphi\in {\mathcal S}_{\pi,M}(\nu)$,
$\xi\in {\mathcal S}_{\pi,M}(\varrho)$
be chosen as to obey  $F(M|\nu,\varrho)=\langle\xi,\varphi\rangle_\pi$, within some arbitrary but fixed  unital $^*$-representation
$\{\pi,{\mathcal H}_\pi\}$  providing non-trivial $\pi$-fibres. Let  $N=\pi(M)^{\,\prime\prime}$.
\begin{lemma}\label{repf}
 For $f\in {\mathfrak F}(M|\nu,\varrho)$.  there exists $K\in (N^{\,\prime})_1$ with the following properties:
 \begin{enumerate}
 \item \label{repf1}
 $f(x)=\langle\pi(x)K\xi,\varphi\rangle_\pi$, $\forall x\in M$;
  \item \label{repf2}
  $p^{\,\prime}_\pi(\varphi)K=Kp^{\,\prime}_\pi(\xi)=K$;
 \item \label{repf3}
 $K s(h)=K^* s(h)=s(h)$.
 \end{enumerate}
\end{lemma}
\begin{proof}
The linear form $f$ has to obey   Definition \ref{Fset}\,\eqref{unique1}--\eqref{unique2}.  By specification of $\{\pi,{\mathcal H}_\pi\}$,  $\varphi\in {\mathcal S}_{\pi,M}(\nu)$ and
$\xi\in {\mathcal S}_{\pi,M}(\varrho)$ in accordance with Theorem \ref{positiv1} we have $$F(M|\nu,\varrho)=\langle\xi,\varphi\rangle_\pi$$
By Definition \ref{Fset}\,\eqref{unique1}, one then has
\[
|f(y^*x)|\leq \|\pi(y)\varphi\|_\pi\|\pi(x)\xi\|_\pi
\]
for all $x,y\in M$. Hence, the mapping
\[
\pi(M)\xi\times \pi(M)\varphi\ni (\pi(x)\xi,\pi(y)\varphi)\longmapsto f(y^*x)
\]
is a bounded by one sesquilinear form, which by continuity extends to a bounded by one sesquilinear form mapping from $p^{\,\prime}_\pi(\xi){\mathcal H}_\pi\times p^{\,\prime}_\pi(\varphi){\mathcal H}_\pi$ into ${\mathbb C}$. By standard conclusions with the help of the Riesz-Theorem we infer that for $x,y\in M$
\[
f(y^*x)=\langle K\pi(x)\xi,\pi(y)\varphi\rangle_\pi
\]
holds, with some $K\in {\mathsf B}({\mathcal H}_\pi)$ obeying $\|K\|\leq 1$ and
$p^{\,\prime}_\pi(\varphi)K p^{\,\prime}_\pi(\xi)=K$. In addition, from this and since $f((yz)^*x)=f(z^*(y^*x))$ is fulfilled, for any $x,y,z\in M$, we infer
\begin{equation}\label{unique4}
f(x)=\langle \pi(x)K\xi,\varphi\rangle_\pi
\end{equation}
with $K\in N^{\,\prime}$,  $\|K\|\leq 1$ and $p^{\,\prime}_\pi(\varphi)K =K p^{\,\prime}_\pi(\xi)=K$ (the same   arguments have been used while proving Lemma \ref{bas3}\,\eqref{bas3aa}). Thus, \eqref{repf1}--\eqref{repf2} hold. Also, by Theorem \ref{bas5}\,\eqref{bas53}, Lemma \ref{bas3}\,\eqref{bas3aa}, $
h=h^\pi_{\xi,\varphi}\geq {\mathsf 0}$ and
$$
h({\mathsf 1})=F(M|\nu,\varrho)=\|h\|_1
$$
have to be fulfilled over  $N^{\,\prime}$. Suppose $F(M|\nu,\varrho)=0$ first. Then,   from the previous $h={\mathsf 0}$ follows. Hence, $s(h)={\mathsf 0}$ for the support orthoprojection $s(h)$ of $h$ in this case, and then \eqref{repf3} is fulfilled in a trivial way.
Now, suppose that $F(M|\nu,\varrho)\not=0$. Then,   from the previous $h\not={\mathsf 0}$ follows, and thus  $s(h)\not={\mathsf 0}$ and
$$
h({\mathsf 1})=F(M|\nu,\varrho)=\|h\|_1 \not=0
$$
Especially, by hermiticity of $h$ (remind that $h$ is positive) and in view of the representation \eqref{unique4} of $f$ the condition Definition \ref{Fset}\,\eqref{unique2} in terms of $h$ and $K$ implies
\begin{subequations}\label{forunique}
\begin{equation}\label{unique5}
h(K)=h(K^*)=F(M|\nu,\varrho)=h({\mathsf 1})=\|h\|_1\not=0
\end{equation}
We are going to show that $Ks(h)=s(h)K$ holds. In fact, by the Cauchy-Schwarz inequality, for the positive linear form $h$ with support $s(h)$ we infer that
\[
\|h\|_1^2=h(K)^2=h(s(h)Ks(h))^2\leq \|h\|_1  h(s(h)K^*s(h)Ks(h))\leq \|h\|_1 h(s(h)K^*Ks(h))
\]
Hence, and since $\|K\|\leq 1$ is fulfilled, the estimate  $$\|h\|_1\leq  h(s(h)K^*s(h)Ks(h))\leq h(s(h)K^*Ks(h))\leq \|h\|_1$$ can be followed. The conclusion is that
$
s(h)=s(h)K^*Ks(h)=s(h)K^*s(h)Ks(h)
$.
Hence $s(h)K^*s(h)^\perp Ks(h)={\mathsf 0}$, and thus $s(h)^\perp Ks(h)={\mathsf 0}$. From this $Ks(h)=s(h)Ks(h)$ follows. On the other hand, owing to \eqref{unique5}, the same line of conclusions has to remain true with $K^*$ instead of $K$. Thus
$
s(h)=s(h)K K^*s(h)=s(h)Ks(h)K^*s(h)
$
and $K^*s(h)=s(h)K^*s(h)$ hold, from which $s(h)K=s(h)Ks(h)$ follows. In summarizing, commutation of both $K$ and $K^*$ with $s(h)$ follow. Thus, from the previous
\begin{equation}\label{unique6}
s(h)=\frac{1}{2}\bigl(K K^*+K^*K\bigr)s(h)
\end{equation}
is obtained. Suppose $K=a+{\mathsf{i}}\, b$, with hermitian $a,b\in (N^{\,\prime})_{\mathsf h}$. As usually, one has $$a=\frac{K+K^*}{2},\,b=\frac{K-K^*}{2\,{\mathsf{i}}},\,a^2+b^2=\frac{1}{2}\bigl(K K^*+K^*K\bigr)$$
Hence, $a$ and $b$ both are commuting with $s(h)$, and therefore in view of \eqref{unique6} the hermitian operators $as(h)$ and $bs(h)$ are obeying
\begin{equation}\label{unique7}
s(h)= s(h)a^2+s(h)b^2
\end{equation}                          By positivity of $h$ and owing to $\|a\|\leq 1$ and in view of \eqref{unique5} from the Cauchy-Schwarz inequality we get the following estimates
\[
\|h\|^2_1=h\Bigl(\frac{K+K^*}{2}\Bigr)^2=h(a)^2=h(as(h))\leq h({\mathsf 1}) h(a^2s(h))=\|h\|_1 h(a^2s(h))\leq \|h\|_1^2
\]
Thus,  $\|h\|_1= h(a^2s(h))$ is seen, and from which owing to $a^2s(h)\leq s(h)$ then $a^2s(h)=s(h)$ follows. Accordingly, we get $b^2s(h)={\mathsf 0}$ by \eqref{unique7}. From this $b s(h)={\mathsf 0}$ is obtained, and therefore $Ks(h)=a s(h)$ is hermitian. Also, from this and \eqref{unique5} we get
\[
h(a s(h))=h(Ks(h))=h(K)=\|h\|_1=h(s(h))
\]
Since $a s(h)$ is hermitian and $\|a\|\leq 1$, $a s(h)\leq s(h)$ follows. Thus, $s(h)-a s(h)\geq {\mathsf 0}$, and in view of the previous  $a s(h)=s(h)$ follows. Hence, the conclusion is
\begin{equation}\label{unique8}
K s(h)=K^*s(h)=s(h)
\end{equation}
\end{subequations}
and which proves \eqref{repf3} in case of $F(M|\nu,\varrho)\not=0$, too.
\end{proof}
\begin{lemma}\label{disj}
Suppose $\nu,\varrho\in {\mathcal S}(M)$ such that $\varrho_\nu$ and $\nu_\varrho$ are not mutually disjoint. Then, within each unital $^*$-representation $\{\pi,{\mathcal H}\}$ of $M$ such that the $\pi$-fibres of $\nu$ and $\varrho$ both exist, there are linear forms $f_1,f_2$ as specified in \eqref{unique3} and obeying $\re f_1\not=\re f_2$.
\end{lemma}
\begin{proof}
For $\{\pi,{\mathcal H}\}$ as mentioned, in line with Lemma \ref{nv} and Theorem \ref{positiv1} there exist    $\varphi\in {\mathcal S}_{\pi,M}(\nu|\varrho)$,  $\xi\in {\mathcal S}_{\pi,M}(\varrho)$ with   $d_B(\varrho,\nu)=\|\xi-\varphi\|_\pi$. By Lemma \ref{diffreteile} and \eqref{diffreteile2aa}, with
\begin{equation}\label{disj0}
h=h^\pi_{\xi,\varphi}\geq {\mathsf 0}
\end{equation}
we have $\eta=s(h)^\perp\xi \in{\mathcal S}_{\pi,M}(\varrho_\nu)$ and $\chi=s(h)^\perp\varphi \in{\mathcal S}_{\pi,M}(\nu_\varrho)$. Note that owing to $s(h)\leq p^{\,\prime}_\pi(\xi)$ and $s(h)\leq p^{\,\prime}_\pi(\varphi)$ then $p^{\,\prime}_\pi(\eta)=p^{\,\prime}_\pi(\xi)-s(h)$ and $p^{\,\prime}_\pi(\chi)=p^{\,\prime}_\pi(\varphi)-s(h)$ hold.  The supposition of non-disjointness between $\varrho_\nu$ and $\nu_\varrho$ then amounts to
\begin{equation}\label{disj1}
c_\pi[p^{\,\prime}_\pi(\xi)-s(h)]\,c_\pi[p^{\,\prime}_\pi(\varphi)-s(h)]\not={\mathsf 0}
\end{equation}
We are going to construct $\xi'\in{\mathcal S}_{\pi,M}(\varrho)$, $\varphi'\in{\mathcal S}_{\pi,M}(\nu)$ with  $h'=h^\pi_{\xi',\varphi'}\geq {\mathsf 0}$ and
\begin{equation}\label{disj2}
(p^{\,\prime}_\pi(\xi')-s(h'))\,(p^{\,\prime}_\pi(\varphi')-s(h'))\not={\mathsf 0}
\end{equation}
Clearly, it might happen that this fact is yet occuring  for the choice $\xi'=\xi$ and $\varphi'=\varphi$. If this however will not be the case, that is, if
\begin{equation}\label{disj3}
(p^{\,\prime}_\pi(\xi)-s(h))\,(p^{\,\prime}_\pi(\varphi)-s(h))={\mathsf 0}
\end{equation}
is fulfilled, let us consider
 a central orthoprojection $z\in N\cap N^\prime$ chosen such that
\begin{equation}\label{disj4}
(p^{\,\prime}_\pi(\xi)-s(h))z\succ (p^{\,\prime}_\pi(\varphi)-s(h))z,\ (p^{\,\prime}_\pi(\varphi)-s(h))z^\perp\succ (p^{\,\prime}_\pi(\xi)-s(h))z^\perp
\end{equation}
By standard facts such $z$ exists. Thereby, $(p^{\,\prime}_\pi(\xi)-s(h))z$ and $(p^{\,\prime}_\pi(\varphi)-s(h))z^\perp$ cannot vanish simultaneously, for otherwise as a consequence of \eqref{disj4} we had $p^{\,\prime}_\pi(\xi)=s(h)$ and $p^{\,\prime}_\pi(\varphi)=s(h)$, but  which contradicted \eqref{disj1}. Let us assume  $(p^{\,\prime}_\pi(\xi)-s(h))z\not={\mathsf 0}$ first. Then, in view of the former relation of \eqref{disj4}, a partial isometry $w\in N^{\,\prime}$ with
\begin{subequations}\label{piso}
\begin{equation}\label{piso1}
w^*w=(p^{\,\prime}_\pi(\varphi)-s(h))z,\ ww^*\leq (p^{\,\prime}_\pi(\xi)-s(h))z
\end{equation}
exists. From this and since $z\in N\cap N^\prime$ holds, by easy inspection one infers that by
\begin{equation}\label{piso2}
u=w+(p^{\,\prime}_\pi(\varphi)-s(h))z^\perp+s(h)
\end{equation}
\end{subequations}
a partial isometry $u\in N^{\,\prime}$ obeying $u^*u=p^{\,\prime}_\pi(\varphi)$ is given. Let us define $\varphi'=u\varphi$ and $\xi'=\xi$. Then, $\varphi'\in {\mathcal S}_{\pi,M}(\nu)$. Also note that \eqref{piso}
implies  $us(h)=s(h)$ and owing to \eqref{disj0} we especially have that  $h^{\pi}_{\xi,\varphi}=h^{\pi*}_{\xi,\varphi}=h^{\pi}_{\varphi,\xi}$. We may conclude as follows:
\[
h^\pi_{\varphi',\xi'}=h^\pi_{\varphi,\xi}((\cdot)u)=h^{\pi}_{\xi,\varphi}((\cdot)u)=h((\cdot)us(h))=h\geq {\mathsf 0}
\]
Thus, by hermiticity, we have
\[
h'=h^\pi_{\xi',\varphi'}=h^\pi_{\varphi',\xi'}=h
\]
Also, in view of \eqref{piso1} and since $s(h')=s(h)$ holds, from the definition of $u$ in \eqref{piso2}
\[
p^{\,\prime}_\pi(\varphi')-s(h')=uu^*-s(h)=ww^*+(p^{\,\prime}_\pi(\varphi)-s(h))z^\perp
\]
is obtained. According to $p^{\,\prime}_\pi(\xi')-s(h')=p^{\,\prime}_\pi(\xi)-s(h)$, $z\in N\cap N^\prime$, \eqref{disj3} and by the last-mentioned relation of  \eqref{piso1}, from the previous one can conclude that
\[
(p^{\,\prime}_\pi(\varphi')-s(h'))\,(p^{\,\prime}_\pi(\xi')-s(h'))=ww^*\not={\mathsf 0}
\]
that is, \eqref{disj2} is seen.  In the contrary case, that is if $(p^{\,\prime}_\pi(\xi)-s(h))z={\mathsf 0}$ is fulfilled, by our preliminary preparations we can be sure that then  $$(p^{\,\prime}_\pi(\varphi)-s(h))z^\perp\not={\mathsf 0}$$ has to hold. Starting from this relation, the previously demonstrated line of reasoning analogously (by interchanging the r\^{o}les of $z$ with $z^\perp$ and $\xi$ with $\varphi$) also can be applied in order to construct implementing vectors $\xi'$ and $\varphi'$ with
\[
h'=h^\pi_{\xi',\varphi'}=h\geq {\mathsf 0}
\]
such that  \eqref{disj2} can be followed (but in contrast to the above, in this case $\xi$ is modified into $\xi'\not=\xi$ and $\varphi$  is kept constant,  $\varphi'=\varphi$). We omit an elaboration of the details, but instead will take for granted now that  non-disjointness of $\varrho_\nu$ and $\nu_\varrho$ in any case will imply existence of $\xi\in{\mathcal S}_{\pi,M}(\varrho)$, $\varphi\in{\mathcal S}_{\pi,M}(\nu)$ with $h\geq {\mathsf 0}$ such that
\begin{equation}\label{disj5}
{\mathsf 0}\not=(p^{\,\prime}_\pi(\xi)-s(h))\,(p^{\,\prime}_\pi(\varphi)-s(h))
\end{equation}
Based on $\xi,\varphi$ with \eqref{disj5} fulfilled, consider now for $\alpha\in {\mathbb R}$
\[
\xi'=\exp{{\mathsf{i}}\alpha}\,(p^{\,\prime}_\pi(\xi)-s(h))\xi+s(h)\xi
\]
Then, also $\xi' \in{\mathcal S}_{\pi,M}(\varrho)$ and obviously  $h^\pi_{\xi',\varphi}=h^\pi_{\xi,\varphi}=h$. On the other hand,  one has
\begin{equation}\label{disj6}
f^\pi_{\xi',\varphi}-f^\pi_{\xi,\varphi}=(\exp{{\mathsf{i}}\alpha}-1) f^\pi_{(p^{\,\prime}_\pi(\xi)-s(h))\xi,(p^{\,\prime}_\pi(\varphi)-s(h))\varphi}
\end{equation}
Note that \eqref{disj5} is equivalent to  $g=f^\pi_{(p^{\,\prime}_\pi(\xi)-s(h))\xi,(p^{\,\prime}_\pi(\varphi)-s(h))\varphi}\not={\mathsf 0}$. Hence, due to
\[
\re{f^\pi_{\xi',\varphi}}-\re{f^\pi_{\xi,\varphi}}=(\cos{\alpha} -1)\re{g}-(\sin{\alpha}) \im{g}
\]
$\alpha$ can be chosen such that $\re{f^\pi_{\xi',\varphi}}-\re{f^\pi_{\xi,\varphi}}\not={\mathsf 0}$. Since both $f_1=f^\pi_{\xi',\varphi}$ and $f_2=f^\pi_{\xi,\varphi}$ are functionals as in \eqref{unique3}, these can be taken to meet our demands.
\end{proof}
\begin{proof}[Proof of Theorem \ref{unique}]
Let $\pi$, $\varphi$ and $\xi$ be chosen as mentioned, and suppose that $\varrho_\nu$ and $\nu_\varrho$ are mutually disjoint. Then we have
\begin{equation}\label{disj1a}
c_\pi[p^{\,\prime}_\pi(\xi)-s(h)]\,c_\pi[p^{\,\prime}_\pi(\varphi)-s(h)]={\mathsf 0}
\end{equation}
Let $z=c_\pi[p^{\,\prime}_\pi(\xi)-s(h)]$. In view of \eqref{disj1a} then
\begin{equation}\label{disj1b}
z\geq p^{\,\prime}_\pi(\xi)-s(h),\ z^\perp\geq p^{\,\prime}_\pi(\varphi)-s(h)
\end{equation}
are fulfilled.
Let $f\in {\mathfrak F}(M|\nu,\varrho)$. According to Lemma \ref{repf} we then have
\begin{equation}\label{disj2a}
f(x)=\langle\pi(x)K\xi,\varphi\rangle_\pi,\ \forall x\in M
\end{equation}
with $K\in N^{\,\prime}$ satisfying Lemma \ref{repf}\,\eqref{repf2}--\eqref{repf3}. From \eqref{repf2} we conclude that
\[
\begin{split}
K=  p^{\,\prime}_\pi(\varphi) K p^{\,\prime}_\pi(\xi)=  (p^{\,\prime}_\pi(\varphi)-&s(h))K(p^{\,\prime}_\pi(\xi)-s(h))+s(h)K(p^{\,\prime}_\pi(\xi)-s(h))\\
& + (p^{\,\prime}_\pi(\varphi)-s(h))K s(h)+s(h)Ks(h)
\end{split}
\]
is fulfilled. By \eqref{repf3} we conclude that $s(h)Ks(h)=s(h)$ and
$$(p^{\,\prime}_\pi(\varphi)-s(h))K s(h)=(p^{\,\prime}_\pi(\varphi)-s(h))s(h)={\mathsf 0}$$
and
$$s(h))K(p^{\,\prime}_\pi(\xi)-s(h))=s(h))(p^{\,\prime}_\pi(\xi)-s(h))={\mathsf 0}$$
have to hold. Accordingly, the conclusion is that
\[
K=(p^{\,\prime}_\pi(\varphi)-s(h))K(p^{\,\prime}_\pi(\xi)-s(h))+s(h)
\]
Moreover, by \eqref{disj1b} with the central orthoprojection $z$ we can  conclude that
\[
\begin{split}
(p^{\,\prime}_\pi(\varphi)-s(h))K(p^{\,\prime}_\pi(\xi)-s(h))=z(p^{\,\prime}_\pi(\varphi)-s(h))K z^\perp (p^{\,\prime}_\pi(\xi)-s(h))={\mathsf 0}
\end{split}
\]
Hence, $K=s(h)$ has to be fulfilled, and \eqref{disj2a} then turns into
\begin{equation}\label{disj3a}
f(x)=\langle\pi(x)s(h)\xi,\varphi\rangle_\pi,\ \forall x\in M
\end{equation}
Finally, note that since the orthoprojection $s(h)$ belongs to $N^{\,\prime}$, for  $x\in M$ we have $$\langle\pi(x)s(h)^\perp\xi,\varphi\rangle_\pi=\langle\pi(x)s(h)^\perp\xi,s(h)^\perp\varphi\rangle_\pi=\langle \pi(x)\eta,\chi\rangle_\pi$$
with implementing vectors $\eta=s(h)^\perp\xi$,  $\chi=s(h)^\perp\varphi$ of $\varrho_\nu$, $\nu_\varrho$  respectively, see  \eqref{diffreteile2aa}. Hence, by disjointness of the latter, with the central orthoprojection $z$ from above $z\eta=\eta$ and $({\mathsf 1}-z) \chi=\chi$ follows. Hence $\langle \pi(x)\eta,\chi\rangle_\pi=0$, and thus accordingly    $$\langle\pi(x)s(h)^\perp\xi,\varphi\rangle_\pi=0$$ for any $x\in M$. By adding this to \eqref{disj3a} yields validity of the formula for $f$ as asserted in Theorem \ref{unique}\,\eqref{u3}. From this also uniqueness in sense of Theorem \ref{unique}\,\eqref{u2} follows.
To see the other way around, assume uniqueness in sense of Theorem \ref{unique}\,\eqref{u2} to be fulfilled. Then, for any two linear forms $f_1,f_2$ constructed in accordance with \eqref{unique3} we especially will find that $\re f_1=\re f_2$. According to the implication arising from Lemma \ref{disj} by negation then the disjointness of $\varrho_\nu$ and $\nu_\varrho$ will follow.
\end{proof}
Let  $\{\pi,{\mathcal H}_\pi\}$ be a unital $^*$-representation  of $M$  such that the $\pi$-fibres of the states $\nu$ and $\varrho$ both exist. In accordance with Theorem \ref{positiv1}, let $\varphi\in {\mathcal S}_{\pi,M}(\nu)$ and $\xi\in {\mathcal S}_{\pi,M}(\varrho)$ be fixed such that $d_B(M|\nu,\varrho)=\|\xi-\varphi\|_\pi$ be fulfilled. Let $N=\pi(M)^{\,\prime\prime}$. We consider the following set ${\mathcal U}_{\pi,M}(\varphi,\xi)$ of partial isometries in $N^{\,\prime}$:
\begin{equation}\label{piso3}
{\mathcal U}_{\pi,M}(\varphi,\xi)=\bigl\{u\in N^{\,\prime}:\ u^*u=p^{\,\prime}_\pi(\xi)\vee p^{\,\prime}_\pi(\varphi) \bigr\}
\end{equation}
Under certain conditions, ${\mathcal S}_{\pi,M}(\nu|\varrho)$ can be described in terms of ${\mathcal U}_{\pi,M}(\varphi,\xi)$.
\begin{corolla}\label{srho}
Suppose $\varrho_\nu$ and $\nu_\varrho$ to be mutually disjoint. Then, the following hold:
 \begin{subequations}\label{relsform}
 \begin{equation}\label{relsform1}
 {\mathcal S}_{\pi,M}(\nu|\varrho)=\bigl\{u\varphi:\ u\in {\mathcal U}_{\pi,M}(\varphi,\xi)\bigr\}
 \end{equation}
 \begin{equation}\label{relsform2}
 {\mathcal S}_{\pi,M}(\varrho|\nu)=\bigl\{u\xi:\ u\in {\mathcal U}_{\pi,M}(\varphi,\xi)\bigr\}
  \end{equation}
 \end{subequations}
\end{corolla}
\begin{proof}
Let $\varphi'\in {\mathcal S}_{\pi,M}(\nu)$ and $\xi'\in {\mathcal S}_{\pi,M}(\varrho)$ be another two  implementing vectors of $\nu, \varrho$ and obeying $d_B(M|\nu,\varrho)=\|\xi'-\varphi'\|_\pi$. The assertion requires to construct  $u\in {\mathcal U}_{\pi,M}(\varphi,\xi)$ such that  $\xi'=u\xi$ and $\varphi'=u\varphi$.
To this sake, let
$v,w\in N^{\,\prime}$ with
\begin{subequations}\label{srho1}
\begin{equation}\label{srho1a}
v^*v =p^{\,\prime}_\pi(\varphi),\ v v^*= p^{\,\prime}_\pi(\varphi'),\ \varphi'=v\varphi
\end{equation}
\begin{equation}\label{srho1b}
w^*w = p^{\,\prime}_\pi(\xi),\ w w^*= p^{\,\prime}_\pi(\xi'),\ \xi'=v\xi
\end{equation}
\end{subequations}
Then, from  $h=h^\pi_{\xi,\varphi}\geq {\mathsf 0}$,  $h'=h^\pi_{\xi',\varphi'}\geq {\mathsf 0}$  with $\|h\|_1=F(M|\nu,\varrho)=\|h'\|_1$  we get
\begin{subequations}\label{srho2}
\begin{equation}\label{srho2a}
h(v^*w)=h(w^*v)=\|h\|_1=F(M|\nu,\varrho)
\end{equation}
\begin{equation}\label{srho2b}
h'(v w^*)=h'(w v^*)=\|h'\|_1=F(M|\nu,\varrho)
\end{equation}
\end{subequations}
As in the proof of Lemma \ref{repf},  with $v^*w$ or $v w^*$ instead of $K$ in \eqref{unique5}, by starting from \eqref{srho2} the conclusions along the lines of \eqref{unique5}--\eqref{unique8} instead yield that
\begin{subequations}\label{srho3}
\begin{equation}\label{srho3a}
v^*w s(h)=w^*v s(h))=s(h)
\end{equation}
\begin{equation}\label{srho3b}
v w^* s(h')=w v^* s(h')=s(h')
\end{equation}
\end{subequations}
Remind that the supposition of disjointness between $\varrho_\nu$ and $\nu_\varrho$ reads
\[
c_\pi[p^{\,\prime}_\pi(\xi)-s(h)]\, c_\pi[p^{\,\prime}_\pi(\varphi)-s(h)]=c_\pi[p^{\,\prime}_\pi(\xi')-s(h)]\, c_\pi[p^{\,\prime}_\pi(\varphi')-s(h)]={\mathsf 0}
\]
Thus, especially we have orthogonality
\[  (p^{\,\prime}_\pi(\xi)-s(h))\, (p^{\,\prime}_\pi(\varphi)-s(h))=(p^{\,\prime}_\pi(\xi')-s(h'))\, (p^{\,\prime}_\pi(\varphi')-s(h'))={\mathsf 0}
\]
and which is equivalent to
\begin{equation}\label{srho4}
p^{\,\prime}_\pi(\xi)\,p^{\,\prime}_\pi(\varphi)=s(h),\  p^{\,\prime}_\pi(\xi')\,p^{\,\prime}_\pi(\varphi')=s(h')
\end{equation}
With the help of \eqref{srho4} from \eqref{srho3} we conclude as follows:
\[ v s(h)=v v^*ws(h)=p^{\,\prime}_\pi(\varphi') w s(h)=p^{\,\prime}_\pi(\varphi')p^{\,\prime}_\pi(\xi') w s(h)=s(h')w s(h) \]
\[
v^* s(h')= v^*v w^* s(h')=p^{\,\prime}_\pi(\varphi) w^* s(h')= p^{\,\prime}_\pi(\varphi)  p^{\,\prime}_\pi(\xi) w^* s(h')=s(h) w^* s(h')
\]
\begin{subequations}\label{srho5}
From this we infer
\begin{equation}\label{srho5a}
v s(h)=s(h')v
\end{equation}
Analogously, with respect to $w$ the conclusion is
\begin{equation}\label{srho5b}
w s(h)=s(h')w
\end{equation}
and by substituting the latter into  $ v s(h)= s(h')w s(h)$, we finally get that
\end{subequations}
\begin{equation}\label{srho6}
v s(h)=w s(h)
\end{equation}
Owing to the latter, and since \eqref{srho4} implies $(p^{\,\prime}_\pi(\xi)-s(h))$, $s(h)$ and $(p^{\,\prime}_\pi(\varphi)-s(h))$ to be mutually orthogonal orthoprojections summing up to $p^{\,\prime}_\pi(\xi)\vee p^{\,\prime}_\pi(\varphi)$, we have
\begin{equation}\label{srho7}
w (p^{\,\prime}_\pi(\xi)-s(h))+v p^{\,\prime}_\pi(\varphi)=w p^{\,\prime}_\pi(\xi)+ v (p^{\,\prime}_\pi(\varphi)-s(h))
\end{equation}
That is, in view of \eqref{srho1}, if $u=w (p^{\,\prime}_\pi(\xi)-s(h))+v p^{\,\prime}_\pi(\varphi)$ is defined, then $u\in N^{\,\prime}$ is a partial isometry with $u^*u=p^{\,\prime}_\pi(\xi)\vee p^{\,\prime}_\pi(\varphi)$ and obeying $u\xi=\xi'$ and $u\varphi=\varphi'$.
\end{proof}
\noindent
Remark that disjointness between $\varrho_\nu$ and $\nu_\varrho$ for $\nu, \varrho\in {\mathcal S}(M)$ in a trivial manner is achieved if $\varrho_\nu={\mathsf 0}$ or $\nu_\varrho={\mathsf 0}$ is occurring.
According to Lemma \ref{diffreteile},  e.g.~the case of $\varrho_\nu={\mathsf 0}$ equivalently means that $\mu={\mathsf 0}$ is the only positive linear form obeying $\mu\leq \varrho$ and $\mu\perp\nu$. We are going to give a characterization of these special cases.

Let  $\{\pi,{\mathcal H}_\pi\}$ be a unital $^*$-representation  of $M$  such that the $\pi$-fibres of the states $\nu$ and $\varrho$ both exist. According to Theorem \ref{positiv1} and Definition \ref{twice1}, there exists $\varphi\in {\mathcal S}_{\pi,M}(\nu|\varrho)$. In line with this, let $\xi\in {\mathcal S}_{\pi,M}(\varrho)$ be chosen as to obey
\begin{subequations}\label{usurr}
\begin{equation}\label{usurr1}
\langle\xi,\varphi\rangle_\pi=F(M|\nu,\varrho)
\end{equation}
As we then know,  over $\pi(M)^{\,\prime}$ the latter is equivalent to
\begin{equation}\label{usurr2}
h_{\xi,\varphi}^\pi\geq {\mathsf 0}
\end{equation}
\end{subequations}
Let  $\psi\in {\mathcal S}_{\pi,M}(\varrho)$, and be $|h_{\psi,\varphi}^\pi|((\cdot)v_{\psi,\varphi}^\pi)$ the polar decomposition of $h_{\psi,\varphi}^\pi$.
\begin{lemma}\label{idem}
For each $\psi\in {\mathcal S}_{\pi,M}(\varrho)$ and $\varphi\in {\mathcal S}_{\pi,M}(\nu|\varrho)$   the following hold true:
\begin{subequations}\label{idem0}
\begin{eqnarray}\label{idem1}
h_{\xi,\varphi}^\pi&=&|h_{\psi,\varphi}^\pi|;\\
\label{idem2}
s(h_{\xi,\varphi}^\pi)\,\xi&=&v_{\psi,\varphi}^{\pi*}\psi\,.
\end{eqnarray}
\end{subequations}
Thereby, $\xi$ can be any vector in ${\mathcal S}_{\pi,M}(\varrho)$  which is obeying \eqref{usurr1}.
\end{lemma}
\begin{proof}
We have $\psi= w\xi$, with $w\in \pi(M)^{\,\prime}$, $w^*w=p_\pi^{\,\prime}(\xi)$. For $h=h_{\xi,\varphi}^\pi$, since $s(h)\leq p_\pi^{\,\prime}(\xi)$ holds, $v=w s(h)$ is a partial isometry with  $v^*v=s(h)$ and
$$
h_{\psi,\varphi}^\pi=h((\cdot)w)=h((\cdot)v)
$$
  In view of \eqref{usurr2} and $w^*\psi=\xi$,  uniqueness of the polar decomposition implies  $h=|h_{\psi,\varphi}^\pi|,\,v=v_{\psi,\varphi}^\pi, \,s(h)\xi=v^*\psi$. This is \eqref{idem0}.
\end{proof}
\begin{corolla}\label{uniqueimp}
Suppose $\nu,\varrho\in {\mathcal S}(M)$ and be $\{\pi,{\mathcal H}_\pi\}$  a unital $^*$-representation  of $M$  such that the $\pi$-fibres of $\nu$ and $\varrho$ both exist. The following assertions are mutually equivalent:
\begin{enumerate}
\item\label{uniqueimp0}
$\mu={\mathsf 0}$ is the only positive linear form $\mu$ over $M$ obeying $\mu\leq \varrho$ and $\mu\perp\nu$;
\item\label{uniqueimp2}
$\xi\in {\mathcal S}_{\pi,M}(\varrho)$ with $\langle\xi,\varphi\rangle_\pi=F(M|\nu,\varrho)$ is uniquely determined by $\varphi\in {\mathcal S}_{\pi,M}(\nu|\varrho)$;
\item\label{uniqueimp3}
$\xi=v_{\psi,\varphi}^{\pi *}\psi$, for given $\varphi\in {\mathcal S}_{\pi,M}(\nu|\varrho)$, $\psi\in {\mathcal S}_{\pi,M}(\varrho)$ and $\xi$ as in  \eqref{uniqueimp2};
\item\label{uniqueimp3a}
$p_\pi^{\,\prime}(\xi)=s(h_{\xi,\varphi}^\pi)$, with $\xi$ as in  \eqref{uniqueimp2}.
\end{enumerate}
Moreover, by each of the above conditions the following are implied to be fulfilled:
\begin{enumerate}
\setcounter{enumi}{4}
\item\label{uniqueimp1}
${\mathcal S}_{\pi,M}(\nu)={\mathcal S}_{\pi,M}(\nu|\varrho)$;
\item\label{uniqueimp1a}
${\mathcal S}_{\pi,M}(\nu)\ni \varphi\longmapsto \xi(\varphi)\in {\mathcal S}_{\pi,M}(\varrho)$ is continuous, with $\xi=\xi(\varphi)$ as in \eqref{uniqueimp2}.
\end{enumerate}
\end{corolla}
\begin{proof}
By Definition \ref{twice1}/Theorem \ref{positiv1}, we find
$\varphi\in {\mathcal S}_{\pi,M}(\nu|\varrho)$, $\xi\in {\mathcal S}_{\pi,M}(\varrho)$ with
\begin{equation}\label{uniqueimp4}
\langle\xi,\varphi\rangle_\pi=F(M|\nu,\varrho)
\end{equation}
Suppose \eqref{uniqueimp0} is fulfilled. Then $\varrho\not\perp\nu$ has to hold. Thus  $F(M|\nu,\varrho)\not=0$, by Corollary \ref{subadd}\,\eqref{subadd2}.
Let $\xi'\in {\mathcal S}_{\pi,M}(\varrho)$ be another vector obeying $\langle\xi',\varphi\rangle_\pi=F(M|\nu,\varrho)$. Then, $\xi'=w\xi$, with $w\in \pi(M)^{\,\prime}$ obeying  $w^*w=p_\pi^{\,\prime}(\xi)$, and due to $h=h_{\xi,\varphi}^\pi\geq {\mathsf 0}$  we have
$$
h(w)=h(w^*)=F(M|\nu,\varrho)=\|h\|_1\not=0
$$
That is, we have arrived at an analogous to \eqref{unique5} relation, which has to be fulfilled under the same assumptions on $\nu$ and $\varrho$, but now with $w$ instead of $K$ there. Accordingly, as in the proof of Lemma \ref{repf}, we will arrive at the same conclusion \eqref{unique8} from there, but reading now in terms of $w$ and telling us that
\begin{equation}\label{unique9}
ws(h)=w^*s(h)=s(h)
\end{equation}
By Lemma \ref{diffreteile}, \eqref{uniqueimp0}
is the same as $\varrho_\nu={\mathsf 0}$. Hence $p^{\,\prime}_\pi(\xi)=s(h)$, and by \eqref{unique9} $$\xi'=w\xi=wp^{\,\prime}_\pi(\xi)\xi=ws(h)\xi=s(h)\xi=\xi$$
and from $s(h)\leq p_\pi^{\,\prime}(\varphi)$ then $s(h)=p_\pi^{\,\prime}(\xi)\leq p_\pi^{\,\prime}(\varphi)$
is seen to hold. The former yields uniqueness \eqref{uniqueimp2} of  $\xi\in{\mathcal S}_{\pi,M}(\varrho)$ satisfying \eqref{uniqueimp4} to the given $\varphi\in{\mathcal S}_{\pi,M}(\nu)$, whereas the last-mentioned estimate includes that \eqref{uniqueimp3a} is fulfilled, and that, for any partial isometry $u\in \pi(M)^{\,\prime}$ obeying $u^*u=p_\pi^{\,\prime}(\varphi)$, $\xi'=u\xi$ is an implementing vector of $\varrho$ satisfying $\langle\xi',u\varphi\rangle_\pi=F(M|\nu,\varrho)$, too. In view of \eqref{bas4} this is \eqref{uniqueimp1}. Also, by Lemma \ref{idem}, see especially \eqref{idem2}, formula \eqref{uniqueimp3} is equivalent to \eqref{uniqueimp3a}. Thus, the implications \eqref{uniqueimp0}$\Rightarrow$\eqref{uniqueimp2}, \eqref{uniqueimp0}$\Rightarrow$\eqref{uniqueimp3a}$\Rightarrow$\eqref{uniqueimp1} and \eqref{uniqueimp3}$\Leftrightarrow$\eqref{uniqueimp3a} are seen.

Suppose \eqref{uniqueimp2}, that is $\xi\in {\mathcal S}_{\pi,M}(\varrho)$ with $\langle\xi,\varphi\rangle_\pi=F(M|\nu,\varrho)$ is uniquely determined. For $h=h_{\xi,\varphi}^\pi$ then $h\geq {\mathsf 0}$, with $p_\pi^{\,\prime}(\xi)\geq s(h)$, $p_\pi^{\,\prime}(\varphi)\geq s(h)$. Let $z=p_\pi^{\,\prime}(\xi)-s(h)$, and  $$\xi^{\,\prime}=-z\xi+s(h)\xi$$ Obviously, $\xi^{\,\prime}\in {\mathcal S}_{\pi,M}(\varrho)$ and $\langle \xi^{\,\prime},\varphi\rangle_\pi=\langle \xi,\varphi\rangle_\pi=F(M|\nu,\varrho)$. By uniqueness $\xi^{\,\prime}=\xi$ has to hold. Hence $\xi^{\,\prime}-\xi=-2 z\xi={\mathsf 0}$. Due to $z\leq p_\pi^{\,\prime}(\xi)$, $z={\mathsf 0}$ follows. Thus $p_\pi^{\,\prime}(\xi)=s(h)$, which is \eqref{uniqueimp3a}. Also note that from this $s(h)^\perp\xi={\mathsf 0}$ follows. Hence $\varrho_\nu={\mathsf 0}$, by \eqref{diffreteile1}. This is  \eqref{uniqueimp0}. Hence,  \eqref{uniqueimp2}$\Rightarrow$\eqref{uniqueimp3a}$\Rightarrow$\eqref{uniqueimp0} holds. From this and the above we may summarize that \eqref{uniqueimp0}, \eqref{uniqueimp2}, \eqref{uniqueimp3} and \eqref{uniqueimp3a} have to be mutually equivalent, and all imply \eqref{uniqueimp1}.

Note that according to all that then the mapping in \eqref{uniqueimp1a} and obeying \eqref{uniqueimp4} with $\xi=\xi(\varphi)$ for each $\varphi\in {\mathcal S}_{\pi,M}(\nu)$ exists and is uniquely determined. We are going to prove continuity. To this sake, suppose $\{\varphi_n\}\subset {\mathcal S}_{\pi,M}(\nu)$ with $\lim_{n\to\infty}\varphi_n=\varphi$.
For the following, let $h=h_{\psi,\varphi}^\pi$, $v=v_{\psi,\varphi}^\pi$ and $h_n=h_{\psi,\varphi_n}^\pi$, $v_n=v_{\psi,\varphi_n}^\pi$ for each $n\in {{\mathbb{N}}}$. Then, $h=\lim_{n\to\infty} h_n$ in norm in the dual space of $R=\pi(M)^{\,\prime}$. From this
\begin{equation}\label{vvh}
\lim_{n\to\infty} v_n^*v=s(|h|)
\end{equation}
with respect to the strong operator topology on $R$ follows. This then will imply \eqref{uniqueimp1a} to be fulfilled. In fact, by definition of $h$ and $v$ and owing to \eqref{uniqueimp3a} and \eqref{idem1}, for $\xi=\xi(\varphi)$  we have $\xi=v^*\psi$ and $v^*v=s(|h|)=p_\pi^{\,\prime}(\xi)$. Hence, $v v^*\psi=v\xi\in {\mathcal S}_{\pi,M}(\varrho)$ is a unit vector. Thus $\psi=v v^*\psi=v\xi$ holds.    Hence, application of \eqref{vvh} to the vector $\xi$ then implies $\lim_{n\to\infty} v_n^*\psi=s(|h|)\xi=\xi$. Thus, provided \eqref{vvh} is valid, then  due to $\xi(\varphi_n)=v_n^*\psi$ from the previous $\lim_{n\to\infty} \xi(\varphi_n)=\xi(\varphi)$ will follow, that is, continuity  \eqref{uniqueimp1a} will hold.\\
We are going to show that \eqref{vvh} is valid. To see this, first note that $R$ can be seen to be  $^*$-isomorphic to some $vN$-algebra $\tau(R)$ acting over some Hilbert space ${\mathcal H}$ such that $|h|$ in respect of this $^*$-isomorphism $\tau$ is implementable by a vector $\chi$, $|h|(z)=\langle\tau(z)\chi,\chi\rangle_\pi$, for each $z\in R$. For each bounded linear form $f$ over $R$, let $f^\tau$ be the unique linear form over $\tau(R)$ such that $f(z)=f^\tau(\tau(z))$, for any $z\in R$. Then $\lim_{n\to\infty} h_n^\tau=h^\tau$ in norm, with $|h^\tau|=\langle(\cdot)\chi,\chi\rangle_\pi$ on $\tau(R)$. Note that  $[\tau(R)^{\,\prime}\chi]=s(|h^\tau|)=\tau(s(|h|))$. Since $v^\tau=\tau(v)$ and $v_n^\tau=\tau(v_n)$ are the partial isometries of the polar decompositions of $h^\tau$ and $h_n^\tau$, respectively, the arguments raised in \cite[Proof of Proposition 4.10.]{Take:79} yield that
$
\lim_{n\to\infty} \tau(v_n^*v)\chi=\chi
$.
Hence  $\lim_{n\to\infty}  \tau(v_n^*v) x\chi=x\chi$ for any $x\in \tau(R)^{\,\prime}$ follows. By uniform boundedness of $\{\tau(v_n^*v)\}$ and by $\tau(v) s(|h^\tau|)=\tau(v)$ from this $\lim_{n\to\infty}  \tau(v_n^*v) =s(|h^\tau|)=\tau(s(|h|))$ (strongly) is obtained. Since $\tau$ is a $^*$-isomorphism of $vN$-algebras \eqref{vvh} follows.
\end{proof}

\subsection{Bures distance and Radon-Nikodym theorems}\label{rn}
Start with a notion extending the term `absolute continuity' from measures to states on ${\mathsf C}^*$-algebras.
\begin{definition}\label{absstet}
Let $M$ be a unital ${\mathsf C}^*$-algebra.
For given states $\nu,\varrho\in {\mathcal S}(M)$, $\varrho$ will be said to be {\em absolutely continuous} in respect of $\nu$, in sign $\varrho\dashv \nu$, if for any bounded net $\bigl\{x_\alpha\bigr\}\subset M$ with $\lim_\alpha \nu(x_\alpha^*x_\alpha)=0$ always
$\lim_\alpha \varrho(x_\alpha^*x_\alpha)=0$ follows.
\end{definition}
Obviously, the binary relation $\dashv$ of the absolute continuity over ${\mathcal S}(M)$ is reflexive and transitive, thus $\{{\mathcal S}(M), \dashv\}$ is a preordered set. By $\varrho \dashv\vdash \nu$ the equivalence relation over ${\mathcal S}(M)$  which is coming along with $\dashv$  will be notified and which has the meaning of `mutual absolute continuity' between two states. That is, for $\varrho,\nu\in {\mathcal S}(M)$ the notion  $\varrho \dashv\vdash \nu$ will mean  that $\varrho \dashv \nu$ and $ \nu\dashv\varrho$ both simultaneously are fulfilled.
\subsubsection{Radon-Nikodym theorems in Bures geometry}\label{RNBG}
Now, in conjunction with  $\dashv$ an intimate relationship to Bures geometry appears.
\begin{theorem}\label{suppabsstet}
For states $\nu,\varrho\in {\mathcal S}(M)$ and a unital $^*$-representation $\{\pi,{\mathcal H}_\pi\}$ such that the $\pi$-fibres of $\nu$ and $\varrho$ both are non-void, the following are mutually equivalent:
\begin{enumerate}
  \item \label{supp1}
  $\varrho\dashv \nu$;
  \item \label{supp2}
  $s(\varrho_\pi)\leq s(\nu_\pi)$;
   \item \label{supp3}
   for $\varphi\in {\mathcal S}_{\pi,M}(\nu)$ there is $\xi\in {\mathcal S}_{\pi,M}(\varrho)$ such that  $\xi=A\varphi$, with  densely defined, positive self-adjoint linear operator $A$ affiliated with $\pi(M)^{\,\prime\prime}$ and $s(\nu_\pi)A=A s(\nu_\pi)$.
\end{enumerate}
Thereby, $\xi\in {\mathcal S}_{\pi,M}(\varrho)$ is uniquely determined by $d_B(M|\varrho,\nu)=\|\xi-\varphi\|_\pi$.
\end{theorem}
\begin{proof}
Suppose \eqref{supp2}. Let $\{x_\alpha\}\subset M$ be a bounded net with $\lim_\alpha \nu(x_\alpha^*x_\alpha)=0$. Then, owing to $\nu(x^*x)=\langle \pi(x^*x)\varphi,\varphi\rangle_\pi=\|\pi(x)\varphi\|_\pi^2$ for any $x\in M$, the relation $\lim_\alpha  \|\pi(x_\alpha)\varphi\|_\pi=0$ is inferred. Hence, $\lim_\alpha \pi(x_\alpha)\varphi=0$. From this $\lim_\alpha \pi(x_\alpha)z\varphi=0$ is seen to hold, for all $z\in \pi(M)^{\,\prime}$.
By uniform boundedness then also $\lim_\alpha \pi(x_\alpha)\eta=0$, for any $\eta\in [\pi(M)^{\,\prime}\varphi]$. This finally yields that ${\mathsf{st}}-\lim_\alpha \pi(x_\alpha)p_\pi(\varphi)=0$ (strong operator convergence). From this and since in line with formula \eqref{traeger} $p_\pi(\varphi)=s(\nu_\pi)$ and $p_\pi(\xi)=s(\varrho_\pi)$ must be fulfilled, \eqref{supp2} then implies ${\mathsf {st}}-\lim_\alpha \pi(x_\alpha)p_\pi(\xi)=0$, and therefore $\lim_\alpha  \pi(x_\alpha)\xi=0$ follows.
Hence, $$\lim_\alpha\varrho(x_\alpha^*x_\alpha)=\lim_\alpha\langle \pi(x_\alpha^*x_\alpha)\xi,\xi\rangle_\pi=\lim_\alpha\|\pi(x_\alpha)\xi\|_\pi^2=0$$ Since $\{x_\alpha\}\subset M$ was an arbitrary bounded net obeying $\lim_\alpha \nu(x_\alpha^*x_\alpha)=0$, by Definition \ref{absstet} this means $\varrho\dashv \nu$. Thus, the implication \eqref{supp2} $\Rightarrow$ \eqref{supp1} is seen to be true. Now, suppose \eqref{supp1} holds. By the Kaplansky density theorem the unit sphere of $\pi(M)$ is strongly dense within the unit sphere of $\pi(M)^{\,\prime\prime}$. Hence, there is a net $\{w_\alpha\}\subset M$, with $\|\pi(w_\alpha)\|\leq 1$ for all $\alpha$, such that ${\mathsf{st}}-\lim_\alpha \pi(w_\alpha)=p_\pi(\varphi)^\perp=s(\nu_\pi)^\perp$. Note that, by standard facts, $\pi(M)$ is isometric isomorphic to the quotient ${\mathsf C}^*$-algebra $M/{\ker \pi}$, that is,
$
\|\pi(w)\|=\inf_{u\in \ker \pi} \|w+u\|
$
holds, for each $w\in M$. Therefore, owing to $\|\pi(w_\alpha)\|\leq 1$, for each $\alpha$ some $u_\alpha\in \ker\pi$ can be chosen such that $x_\alpha=w_\alpha+u_\alpha$ is obeying $\|x_\alpha\|\leq 2$ and $\pi(x_\alpha)=\pi(w_\alpha)$, for each $\alpha$. From this in view of the above
$$\lim_\alpha \nu(x_\alpha^*x_\alpha)=\lim_\alpha  \|\pi(x_\alpha)\varphi\|_\pi^2=\lim_\alpha  \|\pi(w_\alpha)\varphi\|_\pi^2=\|p_\pi(\varphi)^\perp\varphi\|_\pi^2=\|s(\nu_\pi)^\perp\varphi\|_\pi^2=0$$
can be inferred, with the uniformly bounded net $\{x_\alpha\}\subset M$. Since \eqref{supp1} has been supposed, the conclusion is
$$0=\lim_\alpha \varrho(x_\alpha^*x_\alpha)=\lim_\alpha  \|\pi(x_\alpha)\xi\|_\pi^2=\|p_\pi(\varphi)^\perp\xi\|_\pi^2$$
That is, $p_\pi(\varphi)\xi=\xi$ holds. From this $p_\pi(\varphi)\geq p_\pi(\xi)$ follows, which is \eqref{supp2} when formula \eqref{traeger} is taken into account, i.e.~\eqref{supp1} $\Rightarrow$ \eqref{supp2} is true. Hence  \eqref{supp1} $\Leftrightarrow$ \eqref{supp2} holds.

Now we are going to show that \eqref{supp2} $\Leftrightarrow$ \eqref{supp3}. To this sake, let $\{\pi,{\mathcal H}_\pi\}$ to be chosen such that the $\pi$-fibres of both $\nu$ and $\varrho$ are non-void. Suppose $\tilde{\varphi}\in {\mathcal S}_{\pi,M}(\nu)$, $\tilde{\xi}\in {\mathcal S}_{\pi,M}(\varrho)$ with
\begin{equation}\label{absstetpos}
h_{\tilde{\xi},\tilde{\varphi}}^\pi\geq 0
\end{equation}
over $\pi(M)^{\,\prime}$. According to Theorem \ref{positiv1} this choice always can be accomplished. Especially, since the condition \eqref{absstetpos} is fulfilled, an application of Lemma \ref{affipos} provides existence of some densely defined, positive self-adjoint linear operator $A$ affiliated to $\pi(M)^{\,\prime\prime}$ and obeying $p_\pi(\tilde{\varphi})A= A p_\pi(\tilde{\varphi})$ such that
\begin{subequations}\label{absstetpos0}
\begin{equation}\label{absstetpos1}
p_\pi(\tilde{\varphi})\tilde{\xi}=A\tilde{\varphi}
\end{equation}
Suppose \eqref{supp2} to be fulfilled. By formula \eqref{traeger} then $p_\pi(\tilde{\xi})=s(\varrho_\pi)\leq s(\nu_\pi)= p_\pi(\tilde{\varphi})$ holds. Hence, \eqref{absstetpos1} can be extended into
\begin{equation}\label{absstetpos2}
\tilde{\xi}=A\tilde{\varphi}
\end{equation}
Now, let $\varphi\in {\mathcal S}_{\pi,M}(\nu)$ be fixed but arbitrarily chosen. Let $v\in \pi(M)^{\,\prime}$, with $v^*v= p_\pi^{\,\prime}(\tilde{\varphi})$, be the partial isometry in $\pi(M)^{\,\prime}$ and transforming $\tilde{\varphi}$ into the given $\varphi$. Then, since for each $z\in \pi(M)^{\,\prime}$ one has $z A\subset A z$, from \eqref{absstetpos2} we get $v\tilde{\xi}=v A\tilde{\varphi}=A v\tilde{\varphi}=A\varphi$. On the other hand, since \eqref{absstetpos2} is equivalent to $\tilde{\xi}\in [{\pi(M)^{\,\prime\prime}}_+\tilde{\varphi}]$, one infers that $p_\pi^{\,\prime}(\tilde{\xi})\leq p_\pi^{\,\prime}(\tilde{\varphi})$ is fulfilled. Hence $\xi=v\tilde{\xi}\in {\mathcal S}_{\pi,M}(\varrho)$, and thus finally we have
\begin{equation}\label{absstetpos3}
\xi=A\varphi
\end{equation}
\end{subequations}
for some  $\xi\in {\mathcal S}_{\pi,M}(\varrho)$, and with a densely defined, positive self-adjoint linear operator $A$ affiliated to $\pi(M)^{\,\prime\prime}$, and which owing to $s(\nu_\pi)=p_\pi(\varphi)=p_\pi(\tilde{\varphi})$ in addition is obeying $s(\nu_\pi)A= A s(\nu_\pi)$, that is, the implication \eqref{supp2} $\Rightarrow$ \eqref{supp3} is true. The validity of the reverse implication \eqref{supp3} $\Rightarrow$ \eqref{supp2} is true for, if in line with \eqref{supp3} the operator
$A$ with the mentioned properties is obeying \eqref{absstetpos3}, then from this $p_\pi(\varphi)\xi=\xi$ follows. Hence,
 the relation $p_\pi(\xi)\leq p_\pi(\varphi)$ is fulfilled for the special pair of implementing vectors $\xi\in {\mathcal S}_{\pi,M}(\varrho)$ and $\varphi\in {\mathcal S}_{\pi,M}(\nu)$ figuring in \eqref{supp3}. In view of \eqref{traeger} this is \eqref{supp2}.
Finally, if $\xi$, $\varphi$ and the operator $A$ are supposed as specified in \eqref{supp3}, then \eqref{absstetpos3} is fulfilled. From this by Lemma \ref{affipos} \eqref{affipos.1} $\Leftrightarrow$ \eqref{affipos.2} the relation  $h=h_{\xi,\varphi}^\pi\geq 0$ on $\pi(M)^{\,\prime}$ is implied. By Theorem \ref{bas5} \eqref{bas53} then  $$d_B(M|\varrho,\nu)=\|\xi-\varphi\|_\pi$$
 and which by Lemma \ref{bas3}\,\eqref{bas3aa} equivalently says that $\xi\in  {\mathcal S}_{\pi,M}(\varrho)$ is a solution of
 \[
 {F(M|\varrho,\nu)}=\langle \xi,\varphi\rangle_\pi
 \]
 Now, since \eqref{supp2} is fulfilled we have that each $\mu\in M_+^*$ obeying $\mu\leq \varrho$ and $\mu\perp\nu$ has to be vanishing. Hence, Corollary \ref{uniqueimp}\,  \eqref{uniqueimp0}$\Leftrightarrow$\eqref{uniqueimp2}$\Rightarrow$\eqref{uniqueimp1} applies, and thus $\xi$ in \eqref{supp3} then has to be the unique  $\xi\in {\mathcal S}_{\pi,M}(\varrho)$ to given  $\varphi\in {\mathcal S}_{\pi,M}(\nu)$ and obeying $d_B(M|\varrho,\nu)=\|\xi-\varphi\|_\pi$.
\end{proof}
Relating Theorem \ref{suppabsstet}\,\eqref{supp2}, if $M$ is a ${\mathsf{W}}^*$-algebra, then in restriction to the normal state space ${\mathcal{S}}_0(M)$  the relation  $\dashv$ can be characterized by the following addendum.
\begin{corolla}\label{standform}
Suppose $M$ is a ${\mathsf{W}}^*$-algebra, and  $\nu,\varrho\in {\mathcal{S}}_0(M)$. Then, $\varrho\dashv\nu$ holds if, and only if,  $s(\varrho)\leq s(\nu)$ is fulfilled  with the respective support orthoprojections.
\end{corolla}
\begin{proof}
By standard results, see \cite{Saka:56},  \cite[1.16.7.\,Theorem]{Saka:66}, for a  ${\mathsf{W}}^*$-algebra there exists a unital $^*$-representation $\{\pi,{\mathcal{H}}\}$ of $M$ such that $\pi: M\rightarrow N$ is a unital $^*$-isomorphism between $M$ and a  $vN$-algebra $N=\pi(M)=\pi(M)^{\,\prime\prime}$ acting over the Hilbert space ${\mathcal{H}}$ and  such that the $\pi$-fibre of each normal state of $M$ is non-trivial. In line with this,   suppose $\eta\in {\mathcal{S}}_{\pi,M}(\mu)$, for a normal state $\mu$ over $M$. Consider the vector state $\mu_\pi$ over $N$, that is  $\mu_\pi(z)=\langle z\eta,\eta\rangle_\pi$ holds, for each $z\in N$. Clearly, $\mu_\pi$ is a normal state on $N$. The support orthoprojection $s(\mu_\pi)$ of the vector state $\mu_\pi$ over $N$ is $s(\mu_\pi)=p_\pi(\eta)$. On the other hand, due to $\mu(x)=\mu_\pi(\pi(x))$, for all $x\in M$, $q=\pi(s(\mu))\in N$ is an orthoprojection of $N$ which is obeying $\mu_\pi(q)=1$, and therefore $q\geq s(\mu_\pi)$ holds. But since $\pi$ as a $^*$-isomorphism maps faithfully onto $N$, there is a unique orthoprojection $p=\pi^{-1}(s(\mu_\pi))\in M$. Obviously, we have $\mu(p)=\mu_\pi(\pi(p))=\mu_\pi(s(\mu_\pi))=1$. Hence $p\geq s(\mu)$ is fulfilled. From this we infer that $s(\mu_\pi)=\pi(p)\geq \pi(s(\mu))=q$. Accordingly, since in view of the above $q\geq s(\mu_\pi)$ holds, we may summarize all this into $s(\mu_\pi)=\pi(p)\geq \pi(s(\mu))\geq s(\mu_\pi)$. That is, the conclusion is $\pi(s(\mu))=s(\mu_\pi)$. Thus, under the assumption that the $\pi$-fibres of all normal states are non-trivial, $\pi(s(\mu))=s(\mu_\pi)$ holds, for each $\mu\in {\mathcal{S}}_0(M)$. Hence, since $\pi$ is a $^*$-isomorphism, from this relation and for $\nu,\varrho\in {\mathcal{S}}_0(M)$ we conclude that $s(\varrho)\leq s(\nu)$ holds if, and only if, $s(\varrho_\pi)\leq s(\nu_\pi)$.  The equivalence between $\varrho\dashv\nu$ and  $s(\varrho)\leq s(\nu)$ for normal states $\nu,\varrho\in {\mathcal{S}}_0(M)$ then is a consequence of Theorem \ref{suppabsstet} when applied in respect of the above mentioned special  ${\mathsf{W}}^*$-representation $\{\pi,{\mathcal{H}}\}$, for instance.
\end{proof}
Another question arising is under what  circumstances not only $\xi\in {\mathcal S}_{\pi,M}(\varrho)$ in context of  Theorem \ref{suppabsstet}\,\eqref{supp3} is  uniquely determined by the fact that it can be obtained as the image of $\varphi$ under  an  affiliated to $N=\pi(M)^{\,\prime\prime}$ self-adjoint, positive linear operator $A$, but whether or not the latter can be unique in some sense. One criterion comes along with the following result (refer to the notations used in context of Theorem \ref{suppabsstet}). Consider the dense in $p_\pi(\varphi){\mathcal H}_\pi$ domain
\begin{subequations}\label{ope0}
\begin{equation}\label{opea}
{\mathcal D}_\pi(\varphi)=\{z\varphi: z\in N^{\,\prime}\}
\end{equation}
and a linear operator  $A_0$ over $p_\pi(\varphi){\mathcal H}_\pi$ with domain of definition ${\mathcal D}_\pi(\varphi)$ given by
\begin{equation}\label{ope}
A_0:{\mathcal D}_\pi(\varphi)\ni z\varphi\longmapsto z\xi
\end{equation}
\end{subequations}
Let $p=p_\pi(\varphi)=s(\nu_\pi)$. It is obvious that $A_0$ is affiliated to the $vN$-algebra  $N_p=p\, N p$ acting on $p{\mathcal H}_\pi$. Remark that since $A_0$ is positive and symmetric, there exists the positive self-adjoint Friedrichs extension $A_F$ of $A_0$. Note that the restriction to $p{\mathcal H}_\pi$ of each operator $A$ considered in context of Theorem \ref{suppabsstet}\,\eqref{supp3} has to be a self-adjoint extension of this $A_0$.
\begin{lemma}\label{supp5}
The following assertions are mutually equivalent:
\begin{enumerate}
\item \label{supp5.1}
$A_0$ is essentially self-adjoint;
\item \label{supp5.2}
$p_\pi(\xi+\varphi)=p_\pi(\varphi)$.
\end{enumerate}
\end{lemma}
\begin{proof}
With $p=p_\pi(\varphi)$, note that  owing to Theorem \ref{suppabsstet}\,\eqref{supp2}, $\xi\in p{\mathcal H}_\pi$ holds. Hence, $p_\pi(\xi+\varphi)\leq p$ is fulfilled. Let $q=p-p_\pi(\xi+\varphi)$. Then, since for $\eta\in p{\mathcal H}_\pi$ we have $\eta\in q{\mathcal H}_\pi$ if, and only if,  for each $z\in \pi(M)^{\,\prime}$,
\[
0=\langle \eta,z\xi+z\varphi\rangle_\pi=\langle\eta,(A_0+p)z\varphi\rangle_\pi=\langle(A_0^*+p)\eta,z\varphi\rangle_\pi
\]
is fulfilled,
one infers that
\begin{equation}\label{supp5.3}
 {{\mathsf{ker}}}(A_0^*+p)=q{\mathcal H}_\pi\subset {\mathcal D}_\pi^*(\varphi)
\end{equation}
with the adjoint to $A_0$ operator $A_0^*$ over $p{\mathcal H}_\pi$, with dense domain of definition ${\mathcal D}_\pi^*(\varphi)$ of $A_0^*$. Assume $q\not={\mathsf 0}$.  In view of \eqref{supp5.3} then there exists  $\eta\in {\mathcal D}_\pi^*(\varphi)\backslash \{{\mathsf 0}\}$ with  $A_0^*\eta=-\eta$. Accordingly, $A_0^*$ cannot be positive. Thus, since any self-adjoint extension of $A_0$ has to be a restriction of $A_0^*$, due to positivity of $A_F$ we then have a proper inclusion
$
A_F\subset A_0^*
$
to hold. Thus, $A_0^*$ then especially cannot be a self-adjoint operator, for by standard facts two mutually comparable self-adjoint extensions of one and the same operator cannot exist. Hence, under the condition of $p_\pi(\xi+\varphi)< p$ the operator $A_0$ cannot be essentially self-adjoint. By negation the validity of the implication \eqref{supp5.1}\,$\Rightarrow$\eqref{supp5.2} now follows.

To see the other direction, suppose \eqref{supp5.2}. By \eqref{supp5.3} one then infers that
\begin{equation}\label{supp5.4}
 {{\mathsf{ker}}}(A_0^*+p)=\{{\mathsf 0}\}
\end{equation}
Let $\psi\in p{\mathcal H}_\pi$ be  arbitrarily chosen. Also  $q={\mathsf 0}$ implies existence of a sequence $\{\varphi_n\}\subset {\mathcal D}_\pi(\varphi)$, with $\varphi_n=z_n\varphi\in p {\mathcal{H}}_\pi$ for some $z_n\in \pi(M)^{\,\prime}$,  such that $$\psi=\lim_{n\to\infty} z_n(\xi+\varphi)=\lim_{n\to\infty} (A_0+p)\varphi_n$$
Let $\psi_n=(A_0+p)\varphi_n$. Then, as a converging sequence $\{\psi_n\}$ is a Cauchy sequence.  From $\{\varphi_n\}\subset{\mathcal D}_\pi(\varphi)\subset  p {\mathcal{H}}_\pi$ and by positivity of $A_0$, for arbitrary $n,m\in {\mathbb{N}}$ the conclusion is
\begin{eqnarray*}
  \|\psi_n-\psi_m\|_\pi^2 &=& \|A_0(\varphi_n-\varphi_m\|_\pi^2+2 \Re \langle A_0(\varphi_n-\varphi_m),p(\varphi_n-\varphi_m)\rangle_\pi + \\
  & & \phantom{\|A_0(\varphi_n-\varphi_m\|_\pi^2+2 \Re \langle A_0(\varphi_n-\varphi_m),}+\|p(\varphi_n-\varphi_m)\|_\pi^2 \\
   &=& \|A_0(\varphi_n-\varphi_m\|_\pi^2+2\langle A_0(\varphi_n-\varphi_m),(\varphi_n-\varphi_m)\rangle_\pi+\|\varphi_n-\varphi_m\|_\pi^2 \\
   &\geq & \|\varphi_n-\varphi_m\|_\pi^2
\end{eqnarray*}
That is, $\{\varphi_n\}$ is  a Cauchy sequence, too. Also, by standard facts, since  $A_0$ is a symmetric operator, the closure exists and reads $\overline{A}_0=A_0^{**}$. Hence, for the closed linear operator $B=A_0^{**}+p$ and Cauchy sequence $\{\psi_n\}$ we have
\[
\psi=\lim_{n\to\infty} \psi_n=\lim_{n\to\infty} B\varphi_n
\]
Hence, by closeness of $B$, if $\varphi_0=\lim_{n\to\infty} \varphi_n$ then $\varphi_0$ has to belong to the domain of definition ${\mathcal D}_\pi(B)$ of
$B$, but which equals the domain of definition ${\mathcal D}_\pi^{**}(\varphi)$ of $A_0^{**}$, and $\psi=B\varphi_0$. Thus, the conclusion is that each $\psi\in p{\mathcal H}_\pi$ is in the range of $B$,
\begin{equation}\label{supp6}
p{\mathcal H}_\pi=B {\mathcal D}_\pi(B)
\end{equation}
Now, by symmetry we have $A_0\subseteq A_0^{**}\subseteq A_0^*$. Let $\psi^*\in {\mathcal D}_\pi^*(\varphi)$ be arbitrarily chosen. Then, in view of \eqref{supp6} there exists $\chi\in {\mathcal D}_\pi(B)$ such that
$
B\chi=(A_0^*+p)\psi^*
$.
On the other hand, owing to $B=A_0^{**}+p$ and since $A_0^*$ extends $A_0^{**}$ we have ${\mathcal D}_\pi(B)\subseteq {\mathcal D}_\pi^*(\varphi)$, and therefore
$
B\chi=(A_0^{**}+p)\chi=(A_0^*+p)\chi
$
holds, too. Hence, $\psi^*-\chi\in {{\mathsf{ker}}}(A_0^*+p)$ is seen. In view of \eqref{supp5.4} $\psi^*=\chi\in {\mathcal D}_\pi(B)$ follows. Hence, ${\mathcal D}_\pi^*(\varphi)\subseteq {\mathcal D}_\pi(B)={\mathcal D}_\pi^{**}(\varphi)$. Thus  $A_0^*=A_0^{**}$ is seen. This is \eqref{supp5.1}.
\end{proof}
\begin{remark}\label{saka6}
A densely defined positive linear operator
$a_0$ over a Hilbert space ${\mathcal H}$ is self-adjoint (resp.\,essentially self-adjoint) if, and only if, the range ${\mathcal{R}}({\mathsf{1}}+a_0)$ of ${\mathsf{1}}+a_0$ is satisfying ${\mathcal H}={\mathcal{R}}({\mathsf{1}}+a_0)$ (resp.\,${\mathcal{R}}({\mathsf{1}}+a_0)$ is dense in ${\mathcal H}$). The facts follow by literally the same kind of arguments as those raised in context of $A_0$ in the last part of the proof of Lemma \ref{supp5}.
\end{remark}
\begin{example}\label{suppfin}
If $A$ occuring in Theorem \ref{suppabsstet}\,\eqref{supp3} is bounded, or $s(\nu_\pi)=p_\pi(\varphi)$ is a finite  orthoprojection of the $vN$-algebra $N=\pi(M)^{\,\prime\prime}$, then $p_\pi(\xi+\varphi)=p_\pi(\varphi)$.
\end{example}
\begin{proof}
If  $A$ is bounded, it is the unique  self-adjoint extension of $A_0$ in \eqref{ope}. Thus, $A_0$ is essentially self-adjoint. On the other hand, if $p=p_\pi(\varphi)$ is finite, the positive symmetric operator $A_0$ in \eqref{ope} is affiliated to the finite $vN$-algebra $N_p$ acting on the Hilbert space $p{\mathcal H}_\pi$. For each symmetric operator affiliated to a finite type $vN$-algebra one knows that the  closure is the unique self-adjoint extension, see e.g.~in  \cite[9.8]{StZs:79}. Thus, in this context again  $A_0$ is essentially self-adjoint.  In both cases Lemma \ref{supp5} applies and   yields the result.
\end{proof}
Suppose $\varrho, \nu$ are states with $\varrho\dashv \nu$, and $\{\pi,{\mathcal H}_\pi\}$ and $\varphi\in {\mathcal S}_{\pi,M}(\nu)$, $\xi\in {\mathcal S}_{\pi,M}(\varrho)$ chosen in accordance with the settings of  Theorem \ref{suppabsstet}, and be  $N=\pi(M)^{\,\prime\prime}$. The following almost trivial criterion can prove useful for constructing counterexamples.
\begin{corolla}\label{supp5neg}
The operator $A_0$ in \eqref{ope} is not essentially self-adjoint if, and only if,
\begin{equation}\label{supp5neg0}
q\xi=-q\varphi
\end{equation}
is fulfilled, for some orthoprojection $q\in N$ obeying ${\mathsf{0}}<q\leq p_\pi(\varphi)$.
\end{corolla}
\begin{proof}
Under the given suppositions, by negation from Lemma \ref{supp5} one infers that the operator  $A_0$ is not essentially self-adjoint iff $p_\pi(\xi+\varphi)<p_\pi(\varphi)$ is fulfilled. From the latter condition clearly $q\xi=-q\varphi\not= {\mathsf 0}$ follows, for each orthoprojection  $q\in N\backslash\{{\mathsf 0}\}$ obeying $q\leq p_\pi(\varphi)-p_\pi(\xi+\varphi)$. On the other hand,  for each non-trivial orthoprojection $q\in N$ with
$q\xi=-q\varphi$ one has  $q(\xi+\varphi)={\mathsf 0}$. From this and ${\mathsf 0}<q\leq p_\pi(\varphi)$ owing to $p_\pi(\xi)\leq p_\pi(\varphi)$ we get $q\leq p_\pi(\varphi)-p_\pi(\xi+\varphi)$. That is,  $p_\pi(\xi+\varphi)<p_\pi(\varphi)$ follows.
\end{proof}
\subsubsection{The stratum of a state}\label{strat}
For a given state $\nu$, let
\begin{subequations}\label{stratum}
\begin{equation}\label{vorstratum}
\Omega^0_M(\nu)=\bigl\{\varrho\in {\mathcal S}(M):\ \varrho\dashv \nu\bigr\}
\end{equation}
be the subset of all states which are absolutely continuous in respect of $\nu$, and be
\begin{equation}\label{stratum0}
\Omega_M(\nu)=\bigl\{\varrho\in {\mathcal S}(M):\ \varrho\dashv\vdash \nu\bigr\}
\end{equation}
\end{subequations}
Subsequently, each maximal subset of mutually absolutely continuous states will be referred to as `stratum'. Obviously, $\Omega_M(\nu)$ is the unique stratum to which $\nu$ belongs, and which therefore as $\nu$-stratum henceforth will be referred to.
\begin{lemma}\label{setprop}
For each state $\nu$ the following properties of $\Omega^0_M(\nu)$ and $\Omega_M(\nu)$ hold:
\begin{enumerate}
  \item \label{setprop1}
  $\Omega^0_M(\nu)$ and  $\Omega_M(\nu)$ both are sets which are convex in the affine sense;
  \item \label{setprop2}
  $\Omega^0_M(\nu)$ is closed;
  \item \label{setprop3}
  $\Omega_M(\nu)$ is a dense subset of $\Omega^0_M(\nu)$.
\end{enumerate}
Moreover, if $M$ is a ${\mathsf{W}}^*$-algebra with normal state space ${\mathcal{S}}_0(M)$, then for $\nu\in {\mathcal{S}}_0(M)$
\begin{enumerate}
\setcounter{enumi}{3}
  \item \label{setprop4}
$\Omega^0_M(\nu)$ and  $\Omega_M(\nu)$ both are subsets of normal states.
\end{enumerate}
\end{lemma}
\begin{proof}
Let $\bigl\{x_\alpha\bigr\}\subset M$ be a bounded net with $\lim_\alpha \nu(x_\alpha^*x_\alpha)=0$. For $\varrho,\omega\in \Omega^0_M(\nu)$ by Definition \ref{absstet} one has $\lim_\alpha \varrho(x_\alpha^*x_\alpha)=0$ and $\lim_\alpha \omega(x_\alpha^*x_\alpha)=0$. By linearity then it is trivial that $\lim_\alpha \mu(x_\alpha^*x_\alpha)=0$ follows, for each affine convex combination  $\mu=\lambda\,\varrho+(1-\lambda)\,\omega$, with $0< \lambda< 1$. Thus, convexity in the usual (affine) sense of the set defined in \eqref{vorstratum} is obvious. Also, in case if $\varrho,\omega\in \Omega_M(\nu)$, the previous argument shows that $\mu\dashv \nu$. On the other hand, if $\{y_\alpha\}\subset M$ is any bounded net with $\lim_\alpha \mu(y_\alpha^*y_\alpha)=0$, owing to $\varrho\leq \lambda^{-1}\mu$ then $\lim_\alpha \varrho(y_\alpha^*y_\alpha)=0$ follows. Hence, $\varrho\dashv \mu$ must be fulfilled. Due to the assumption $\varrho\in \Omega_M(\nu)$ and by transitivity of $\dashv$ from this $\nu\dashv \mu$ is seen. In view of the above stated $\mu\dashv \nu$ we conclude that $\mu\dashv\vdash \nu$, which means that $\mu\in \Omega_M(\nu)$. Thus the $\nu$-stratum is convex in the affine sense. This completes the proof of \eqref{setprop1}.\\
To see \eqref{setprop2}, suppose $\{\varrho_n\}$ is a sequence of states with $\varrho_n\dashv \nu$, for each $n$, and such that $\|\cdot\|_1-\lim_n \varrho_n=\varrho$. We are going to show that then $\varrho\dashv \nu$ holds. To this sake, suppose  $\{x_\alpha\}\subset M$ is a bounded net with $\lim_\alpha \nu(x_\alpha^*x_\alpha)=0$. We may suppose that $\sup \|x^*_\alpha x_\alpha\|=\beta>0$. Let $\varepsilon>0$ be an arbitrarily chosen real, and fix $m\in {{\mathbb{N}}}$ such that $\|\varrho-\varrho_m\|_1\leq \varepsilon/(2\beta )$. Also, let $\alpha_0$ be chosen such that $\varrho_m(x^*_\alpha x_\alpha)\leq \varepsilon/2$, for any $\alpha\geq \alpha_0$. Because of $\varrho_m\dashv \nu$ this choice of $\alpha_0$ is possible. Then, for $\alpha\geq \alpha_0$ the conclusion is
\[
\varrho(x_\alpha^*x_\alpha) \leq | \varrho(x_\alpha^*x_\alpha)-\varrho_m(x_\alpha^*x_\alpha)|+\varrho_m(x_\alpha^*x_\alpha)\leq \|\varrho-\varrho_m\|_1 \beta+ \varrho_m(x_\alpha^*x_\alpha)\leq \varepsilon/2 + \varepsilon/2
\]
Hence, $\varrho(x_\alpha^*x_\alpha) \leq\varepsilon$, for all $\alpha\geq \alpha_0$. Since such a subscript $\alpha_0$ can be chosen to any $\varepsilon>0$, the conclusion is $\lim_\alpha \varrho(x_\alpha^*x_\alpha)=0$. Hence, $\varrho\dashv \nu$ and thus \eqref{setprop2} is seen.

For the proof of \eqref{setprop3}, suppose $\varrho\dashv \nu$, that is, $\varrho\in \Omega^0_M(\nu)$. Let $\{\pi,{\mathcal H}_\pi\}$ be chosen such that the $\pi$-fibres of $\nu$ and $\varrho$ both exist. Let  $\varphi\in {\mathcal S}_{\pi,M}(\nu)$ and be $\xi\in {\mathcal S}_{\pi,M}(\varrho)$ the unique implementing vector chosen in accordance with Theorem \ref{absstet}. Let $p=p_\pi(\varphi)$, $N=\pi(M)^{\,\prime\prime}$ and consider the $vN$-algebra $N_p=p\,Np$ acting over $p{\mathcal H}_\pi$. By Theorem \ref{absstet}\,\eqref{supp3} there is  $\{A_n\}\subset (N_p)_+$ with $\lim_{n\to\infty} A_n\varphi=\xi$. Consider vectors  $
\eta_n=A_n\varphi+\varphi/n$.
By positivity of each $A_n$, $\|\eta_n\|\not=0$  and $\lim_{n\to\infty} \|\eta_n\|=1$ hold. Also, since $n A_n+p\in N_p$ is a bounded invertible operator over $p{\mathcal H}_\pi$, its range has to be full, and thus maps  $p{\mathcal H}_\pi=[\pi(M)^{\,\prime}\varphi]$ onto $p{\mathcal H}_\pi=[\pi(M)^{\,\prime}\eta_n]$. Hence, $p_\pi(\eta_n)=p_\pi(\varphi)$, and if states $\varrho_n$ are defined by
$
\xi_n=\eta_n/\|\eta_n\|\in {\mathcal S}_{\pi,M}(\varrho_n)
$
then we have that $s((\varrho_n)_\pi)=p_\pi(\xi_n)=p_\pi(\eta_n)=p_\pi(\varphi)=s(\nu_\pi)$ is fulfilled, for each $n$. That is, in view of \eqref{stratum0} and in applying Theorem \ref{suppabsstet}\,\eqref{supp2} accordingly,  $\varrho_n\in \Omega_M(\nu)$.
Finally, from  $\lim_{n\to\infty}\xi_n=\xi$ we infer $\lim_{n\to\infty}\varrho_n=\varrho$ in norm.

To see \eqref{setprop4}, suppose $\varrho$ to be a state with $\varrho\dashv \nu$. Consider a descendingly directed system $\{p_\alpha\}\subset M$ of orthoprojections satisfying $\mathsf{g.l.b.} \{p_\alpha\}={\mathsf{0}}$. Since  $\nu\in {\mathcal{S}}_0(M)$ by assumption, we have $\lim_\alpha \nu(p_\alpha)=0$. Due to $\varrho\dashv \nu$ from this $\lim_\alpha \varrho(p_\alpha)=0$ follows, for  each such system $\{p_\alpha\}$. This implies $\varrho\in {\mathcal{S}}_0(M)$, by known  facts. Hence $\Omega^0_M(\nu)\subset {\mathcal{S}}_0(M)$, from which due to  $\Omega_M(\nu) \subset \Omega^0_M(\nu)$ the result follows.
\end{proof}
\begin{remark}\label{saka}
\begin{enumerate}
\item\label{saka0}
Relating `absolute continuity' in a ${\mathsf C}^*$-algebraic context see  \cite{Gudd:79}, \cite{Hiai:84}, and \cite[{\sc{Definition 1.1}}\,(iii)]{Nies:83}. The latter setting is used in Definition \ref{absstet}.
  \item \label{saka1}
  A special case of $\varrho\dashv\nu$ is based on the domination of positive linear forms. Namely, for  $\nu,\varrho\in {\mathcal S}(M)$ let $\varrho\ll \nu$ mean that $\lambda>0$ exists with $\lambda\,\varrho(x^*x)\leq \nu(x^*x)$, for any $x\in M$.
Then  $\varrho\ll \nu$ implies $\varrho\dashv\nu$. For ${\mathsf{dim}} M<\infty$ also the reverse holds.
\item \label{saka2} Let $\{\pi,{\mathcal H}_\pi\}$ be a unital $^*$-representation such that the $\pi$-fibres of both $\varrho$ and $\nu$ are non-void. By  Sakai's Radon-Nikodym theorem \cite{Saka:65,Saka:71}, $\varrho\ll \nu$ implies existence of $A\in \pi(M)_+^{\,\prime\prime}$ such that, with $\varphi\in {\mathcal S}_{\pi,M}(\nu)$,
\begin{equation}\label{srnt}
\varrho(x)=\langle \pi(x)A\varphi,A\varphi\rangle_\pi
\end{equation}
holds, for any $x\in M$, with $A$ uniquely determined provided $Ap_\pi(\varphi)=A$.
\item  \label{saka3}
By Theorem \ref{suppabsstet},  under the premises of \eqref{saka2}, $\varrho\dashv \nu$ is equivalent to \eqref{srnt} with positive self-adjoint $A$ affiliated with $\pi(M)^{\,\prime\prime}$ and obeying $p_\pi(\varphi)A=A p_\pi(\varphi)$. This generalization of Sakai's RN-theorem is due to Araki \cite{Arak:72}.
\item\label{saka5}
Obviously, Theorem \ref{supp5}/Example \ref{suppfin} remain true if in place of $\xi+\varphi$ the vector $\alpha \xi+\beta\varphi$ is used, with arbitrary reals obeying $\alpha\geq 0$, $\beta>0$.
\end{enumerate}
\end{remark}
\subsection{The special case of all bounded linear  operators}\label{beimatrix}
The ${\mathsf C}^*$-algebra  $M={\mathsf B}({\mathcal H})$ of bounded linear operators  on a  separable Hilbert space ${\mathcal H}$ will be considered. The shortcut $x\smile y$ is used to indicate commutation $x y=y x$ between $x,y\in {\mathsf B}({\mathcal H})$.
\subsubsection{Basic facts about density operators}\label{beimatrix0}
As $^*$-representation $\{\pi,{\mathcal H}_\pi\}$ consider the left action of bounded linear operators $x$
\begin{subequations}\label{hs0}
\begin{equation}\label{hs1}
{\mathsf B}({\mathcal H})\ni x\longmapsto \pi(x)a=x a\in {\mathsf{ H.S.}}({\mathcal H})
\end{equation}
on Hilbert-Schmidt operators $a\in {\mathsf{ H.S.}}({\mathcal H})$ over ${\mathcal H}$. Thus,
\begin{equation}\label{hs2}
{\mathcal H}_\pi={\mathsf{ H.S.}}({\mathcal H})=\{a\in{\mathsf B}({\mathcal H}): {{\mathsf{tr}}}\, a a^*<\infty\}
\end{equation}
is taken as representation Hilbert space, with the  scalar product defined as
\begin{equation}\label{hs3}
\langle b,a\rangle_{{\mathsf{ H.S.}}}={{\mathsf{tr}}}\,b a^*
\end{equation}
with the usual trace ${{\mathsf{tr}}}$, see \cite{Scha:70} for operator theoretical reasonings, and Example \ref{ex1} for the general context. Remark that by the very definition in \eqref{hs1}, the unital $^*$-representation $\pi$ under consideration is a $^*$-isomorphism onto the $vN$-algebra of all operators of left-multiplications (by bounded linear operators) to Hilbert-Schmidt operators.
The vector states with respect to $\{\pi,{\mathcal H}_\pi\}$ are given as
\begin{equation}\label{nstate}
\varrho(x)=\langle \pi(x)\, a,a\rangle_{{\mathsf{ H.S.}}}={{\mathsf{tr}}}\,x\, a a^*,\text{ with } \|a\|^2_{{\mathsf{ H.S.}}}={{\mathsf{tr}}}\,a a^*=1,
\end{equation}
\end{subequations}
for all $x\in{\mathsf B}({\mathcal H})$, and exactly correspond to the normal states ${\mathcal S}_0({\mathsf B}({\mathcal H}))$ over $M={\mathsf B}({\mathcal H})$. Thereby, the correspondence between normal states $\varrho$ and normalized positive trace class operators  $ \sigma_\varrho=a a^*$ in formula \eqref{nstate} is one-to-one. The unique $\sigma_\varrho$ is the density operator of $\varrho$, and    $\varrho(x)={{\mathsf{tr}}}\,\sigma_\varrho\, x$ holds, at each $x\in {\mathsf B}({\mathcal H})$.

Also note that  $\pi({\mathsf B}({\mathcal H}))^{\,\prime}$ is implemented by the right-actions of bounded linear operators to Hilbert-Schmidt operators. Thus, for $Z\in \pi({\mathsf B}({\mathcal H}))^{\,\prime}$ there is a unique $z\in {\mathsf B}({\mathcal H})$ such that $Z\, c=cz$, for each $c\in {\mathsf{H.S.}}({\mathcal H})$. Thereby, the map ${\mathsf B}({\mathcal H})\ni z\longmapsto Z\in \pi({\mathsf B}({\mathcal H}))^{\,\prime}$ is a  $^*$-(anti)morphism.
Accordingly, the $\pi$-fibre \eqref{faser} of $\varrho$ is easily seen to be
\begin{equation}\label{nfibre}
   {\mathcal S}_{\pi,{\mathsf B}({\mathcal H})}(\varrho)=\bigl\{b\in {\mathsf{ H.S.}}({\mathcal H}): b=\sqrt{\sigma_\varrho}w^*,\ w^*w=s(\varrho),\,w\in {\mathsf B}({\mathcal H}) \bigr\}
\end{equation}
with the support orthoprojection $s(\varrho)$ of $\varrho$ in ${\mathsf B}({\mathcal H})$, which is the same as the support $s(\sigma_\varrho)$ of the density operator $\sigma_\varrho$. Since according to \eqref{nfibre} the $\pi$-fibre to each normal state is non-trivial, in context of fidelity between normal states over $M={\mathsf B}({\mathcal H})$ the special $^*$-representation $\{\pi,{\mathsf{H.S.}}({\mathcal H})\}$ constantly will be used.

Remind that a normal state $\varrho$ with full support orthoprojection $s(\varrho)={\mathsf 1}$ is called faithful. Since ${\mathcal H}$ is supposed to be separable, faithful normal states exist. The set ${\mathcal S}_0^{\mathsf{faithful}}({\mathsf B}({\mathcal H}))$ of all faithful normal states is convex (in the affine sense) and uniformly dense in the convex set ${\mathcal S}_0({\mathsf B}({\mathcal H}))$. In contrast to a faithful state, a normal state $\nu$ with minimal support orthoprojection $s(\nu)$ is extremal in ${\mathcal S}_0({\mathsf B}({\mathcal H}))$, and uniquely corresponds to the density operator $\sigma_\nu=s(\nu)=p_\psi$, with $p_\psi$ being the minimal orthoprojection projecting onto the one-dimensional complex space spanned by a unit vector $\psi$. As a special case of \eqref{nstate}, for an extremal normal state $\nu$ we have
\begin{equation}\label{exstate}
\nu(x)={{\mathsf{tr}}}\,x\,p_\psi =\langle x \psi,\psi\rangle
\end{equation}
with unit vector $\psi\in {\mathcal{H}}$ which is uniquely determined by $\nu$ up to a complex multiple of modulus one. In the situation \eqref{exstate} often a notation like $\nu_\psi$ will be used.
The set ${\mathsf{ex}}\,{\mathcal S}_0({\mathsf B}({\mathcal H}))$ of all extremal (or equivalently: pure) normal states is of fundamental importance in the geometry of the normal state space,  for each $\varrho\in{\mathcal S}_0({\mathsf B}({\mathcal H}))$ can be represented as a convex combination of extremal normal states $\nu_{\psi_j}$
\begin{subequations}
\begin{equation}\label{conde}
 \varrho=\sum_j \lambda_j \nu_{\psi_j}
\end{equation}
(finite or countably infinite), with strictly positive reals $\lambda_j$ summing up to one. Generally, such decompositions may be highly non-unique. To reduce this a little, subsequently will refer to \eqref{conde} as {\it{minimal}} decomposition of $\varrho$ if
\begin{equation}\label{conde1}
s(\varrho)\not= \bigvee_{j\not=k} p_{\psi_j}, \forall k
\end{equation}
\end{subequations}
Representatively, minimal decompositions will be dealt with:
\newpage
\begin{corolla}\label{minidec}
Let $\varrho\in {\mathcal S}_0^{\mathsf{faithful}}({\mathsf B}({\mathcal H}))$, with density operator  $\sigma_\varrho$, over a separable Hilbert space ${\mathcal H}$ with ${\mathsf{dim}}\,{\mathcal{H}}=\infty$, and be  $\{\psi_k\}_{k\in {\mathbb{N}}}$ a system of unit vectors. A convex combination $$\sum_{k=1}^\infty \lambda_k \nu_{\psi_k}$$   is a minimal decomposition of $\varrho$ if, and only if, the following conditions are fulfilled:
\begin{enumerate}
  \item\label{minidec1}
  $\psi_k \in {\mathcal{R}}(\sigma_\varrho)$,\ $\forall\,k\in {\mathbb{N}}$;
  \item\label{minidec2}
  $\bigl\{\sqrt{\sigma_\varrho}^{-1} \psi_k \bigr\}$ is a maximal systems of mutually orthogonal vectors in ${\mathcal{H}}$.
\end{enumerate}
Moreover, under these conditions $\lambda_k=\bigl\|\sqrt{\sigma_\varrho}^{-1} \psi_k\bigr\|^{-2}$ holds, for all $k\in {\mathbb{N}}$.
\end{corolla}
\begin{proof}
Remark that $\sqrt{\sigma_\varrho}^{-1} $ is self-adjoint over the dense domain ${\mathcal{D}}={\mathcal{R}}(\sqrt{\sigma_\varrho})$. Hence ${\mathcal{R}}(\sigma_\varrho)\subset {\mathcal{D}}$, and thus the condition \eqref{minidec1} guarantees that the settings within \eqref{minidec2} make sense. Furthermore, due to the definition of the term `minimal decomposition', the problem equivalently asks for necessary and sufficient conditions which assure that a convex combination $\sum_{k=1}^\infty \lambda_k p_{\psi_k}$ will be constituting a minimal decomposition of the  density operator $\sigma_\varrho$. As a problem about positive trace-class operators, this is affirmatively answered by Lemma \ref{minidec0} (put $\Lambda=\sigma_\varrho$ there).
\end{proof}
Remark that for each faithful $\varrho$ the stratum $\Omega_{{\mathsf B}({\mathcal H})}(\varrho)$ as defined in \eqref{stratum0} is
\begin{equation}\label{equi2a}
\Omega_{{\mathsf B}({\mathcal H})}(\varrho)={\mathcal S}_0^{\mathsf{faithful}}({\mathsf B}({\mathcal H}))
\end{equation}
In fact, due to $s(\varrho)={\mathsf 1}$ this follows by Lemma \ref{setprop}\,\eqref{setprop4} and Example \ref{Bop}\,\eqref{equi2}.

Finally, what happens if the relation $\dashv$ enters will be exemplified.
\begin{example}\label{Bop} Suppose $\varrho,\nu\in {\mathcal S}_0({\mathsf B}({\mathcal H}))$, with support projections $s(\nu), s(\varrho)$. Let $a\in {\mathcal S}_{\pi,{\mathsf B}({\mathcal H})}(\nu)$, $b\in {\mathcal S}_{\pi,{\mathsf B}({\mathcal H})}(\varrho)$ and be
$a^*b=u|a^*b|$ the polar decomposition of $a^*b$.
\begin{enumerate}
  \item \label{equi1} $ \varrho\dashv \nu\,\,\,\, \ \Longleftrightarrow\  s(\varrho)\leq s(\nu)$;
  \item \label{equi2} $ \varrho\dashv\vdash \nu\ \Longleftrightarrow\  s(\varrho)=s(\nu)$;
  \item \label{equi3} $ \varrho\dashv \nu\,\,\,\,\,\ \Longrightarrow\ \xi(a)=b u^*\in {\mathcal S}_{\pi,{\mathsf B}({\mathcal H})}(\varrho)$,
  \end{enumerate}
  where $\xi(\cdot)$ is the map ${\mathcal S}_{\pi,{\mathsf B}({\mathcal H})}(\nu)\ni c \longmapsto \xi(c)\in {\mathcal S}_{\pi,{\mathsf B}({\mathcal H})}(\varrho)$ of Corollary \ref{uniqueimp}\,\eqref{uniqueimp1a}.
\end{example}
\begin{proof}
Since ${\mathsf B}({\mathcal H})$ is a $vN$-algebra, application of
Corollary \ref{standform} yields that \eqref{equi1} and \eqref{equi2} hold true, for normal states $\nu, \varrho$. Therefore, since $\varrho\dashv\nu$ implies Corollary \ref{uniqueimp}\,\eqref{uniqueimp0} to be fulfilled, in view of Corollary \ref{uniqueimp}\,\eqref{uniqueimp3} and Remark \ref{polspez} the mapping  $${\mathcal S}_{\pi,{\mathsf B}({\mathcal H})}(\nu)\ni a \longmapsto \xi(a)\in {\mathcal S}_{\pi,{\mathsf B}({\mathcal H})}(\varrho)$$ reads with
$\xi(a)={v_{b,a}^\pi}^* b =b u^*\in {\mathcal S}_{\pi,{\mathsf B}({\mathcal H})}(\varrho)$.
\end{proof}
\subsubsection{Fidelity between density operators}\label{beimatrix1}
When seen in terms of Example \ref{ex1}, for the fidelity between normal states in the context at hand formula \eqref{bei1} will be of relevance, exclusively. Therefore, we are going to retrace in detail along the lines of the general ideas made use of in Lemma \ref{nv} and Theorem \ref{positiv1} how the special case of \eqref{bei1}, that is Uhlmann's formula, will come about.
\begin{lemma}\label{nvbas}
Let $\varrho,\nu\in {\mathcal S}_0({\mathsf B}({\mathcal H}))$. Then the following holds:
\begin{equation}\label{nvbasi}
{{\mathsf{tr}}}\,|\sqrt{\sigma_\nu}\sqrt{\sigma_\varrho}\,|=\max_{a\in {\mathcal S}_{\pi,{\mathsf B}({\mathcal H})}(\nu),\,b\in {\mathcal S}_{\pi,{\mathsf B}({\mathcal H})}(\varrho)} |{{\mathsf{tr}}}\,b a^*|
\end{equation}
\end{lemma}
\begin{proof}
We are going to show that for $a\in {\mathcal S}_{\pi,{\mathsf B}({\mathcal H})}(\nu)$ and $b\in {\mathcal S}_{\pi,{\mathsf B}({\mathcal H})}(\varrho)$
\begin{equation}\label{nvspez1}
 |{{\mathsf{tr}}}\,b a^*|\leq {{\mathsf{tr}}}\,|\sqrt{\sigma_\nu}\sqrt{\sigma_\varrho}\,|
\end{equation}
holds. In line with \eqref{nfibre}, let partial isometries $v$ and $w$ with $v^*v=s(\nu)$ and $w^*w=s(\varrho)$ be chosen such that $a=\sqrt{\sigma_\nu}\,v^*$ and $b=\sqrt{\sigma_\varrho}\,w^*$ are fulfilled. Let
$
\sqrt{\sigma_\nu}\sqrt{\sigma_\varrho}=u|\sqrt{\sigma_\nu}\sqrt{\sigma_\varrho}\,|
$
be the polar decomposition of $\sqrt{\sigma_\nu}\sqrt{\sigma_\varrho}$. Then we can conclude that
\begin{eqnarray*}
 |{{\mathsf{tr}}}\,b a^*| &=& |{{\mathsf{tr}}}\,\sqrt{\sigma_\varrho}\,w^* v\sqrt{\sigma_\nu}|= |{{\mathsf{tr}}}\,\sqrt{\sigma_\nu}\sqrt{\sigma_\varrho}\,w^* v| \\
  &=& |{{\mathsf{tr}}}\,u|\sqrt{\sigma_\nu}\sqrt{\sigma_\varrho}|\,w^* v| \\
  &=& |{{\mathsf{tr}}}\,w^*vu|\sqrt{\sigma_\nu}\sqrt{\sigma_\varrho}| \,|= |{{\mathsf{tr}}}\,|\sqrt{\sigma_\nu}\sqrt{\sigma_\varrho}| (w^*vu) \,|\\
  &\leq & \bigl\| |\sqrt{\sigma_\nu}\sqrt{\sigma_\varrho}|\bigr\|_1={{\mathsf{tr}}}\,|\sqrt{\sigma_\nu}\sqrt{\sigma_\varrho}|,
\end{eqnarray*}
where the latter estimate follows by standard facts and since the operator $w^*vu$ belongs to the unit ball of  ${\mathsf{B}}({\mathcal{H}})$. That is, \eqref{nvspez1} is seen to be true. Obviously, from the latter the assertion \eqref{nvbasi} will follow if the partial isometries $v$ and $w$ featuring in the expressions of $a$ and $b$ could be specialized in such a way that equality
\begin{equation}\label{nvspez2}
 {{\mathsf{tr}}}\,b a^*={{\mathsf{tr}}}\,|\sqrt{\sigma_\nu}\sqrt{\sigma_\varrho}\,|
\end{equation}
occurred. To see that this can be accomplished, first remind that the characteristic partial isometry $u$
appearing in the polar decomposition of $\sqrt{\sigma_\nu}\sqrt{\sigma_\varrho}$ is satisfying $u^*u=s(|\sqrt{\sigma_\nu}\sqrt{\sigma_\varrho}\,|)$, $u u^*=s(|\sqrt{\sigma_\varrho}\sqrt{\sigma_\nu}\,|)$ and $|\sqrt{\sigma_\varrho}\sqrt{\sigma_\nu}\,|=u|\sqrt{\sigma_\nu}\sqrt{\sigma_\varrho}\,|u^*$.

In case of ${\mathsf{dim}}\,{\mathcal{H}}<\infty$, $w={\mathsf{1}}$ and $v=u_0 s(\nu)$ with a unitary extension $u_0$ of $u^*$ is chosen. Then $a=\sqrt{\sigma_\nu} v^*\in {\mathcal S}_{\pi,{\mathsf B}({\mathcal H})}(\nu)$, and since $s(|\sqrt{\sigma_\varrho}\sqrt{\sigma_\nu}\,|)\leq s(\nu)$ holds, the relation
$s(|\sqrt{\sigma_\nu}\sqrt{\sigma_\varrho}\,|)v=u^*s(\nu)=u^*$ is fulfilled. We may conclude as follows:
\begin{eqnarray*}
  {{\mathsf{tr}}}\,b a^* &=& {{\mathsf{tr}}}\,\sqrt{\sigma_\varrho} v\sqrt{\sigma_\nu}=  {{\mathsf{tr}}}\,\sqrt{\sigma_\nu}\sqrt{\sigma_\varrho} v \\
   &=& {{\mathsf{tr}}}\,u|\sqrt{\sigma_\nu}\sqrt{\sigma_\varrho}| v= {{\mathsf{tr}}}\,u|\sqrt{\sigma_\nu}\sqrt{\sigma_\varrho}|s(|\sqrt{\sigma_\nu}\sqrt{\sigma_\varrho}\,|) v\\
   &=& {{\mathsf{tr}}}\,u|\sqrt{\sigma_\nu}\sqrt{\sigma_\varrho}|u^*=
   {{\mathsf{tr}}}\,u^*u|\sqrt{\sigma_\nu}\sqrt{\sigma_\varrho}|=
   {{\mathsf{tr}}}\,|\sqrt{\sigma_\nu}\sqrt{\sigma_\varrho}|,
\end{eqnarray*}
that is, \eqref{nvspez2} is satisfied, for some $a\in {\mathcal S}_{\pi,{\mathsf B}({\mathcal H})}(\nu)$ and $b\in {\mathcal S}_{\pi,{\mathsf B}({\mathcal H})}(\varrho)$.

In case of separable ${\mathcal{H}}$ with ${\mathsf{dim}}\,{\mathcal{H}}=\infty$, let us fix an  orthoprojection $p$
obeying ${\mathsf{dim}}\,p{\mathcal H}={\mathsf{dim}}\,p^\perp {\mathcal H}=\infty$. Then, since we are in a factor $vN$-algebra, for any two orthoprojections $q$ and $q^\prime$ over $\mathcal H$ we have comparability in the $vN$-sense $p\succ q$ and $p^\perp\succ q^\prime$. Thus especially, with  $q=s(\nu)-s(|\sqrt{\sigma_\varrho}\sqrt{\sigma_\nu}\,|)$ and $q^\prime=s(\varrho)$, partial isometries  $v_1$, $v_2$ exist  such that $v_1 v_1^*\leq p\phantom{\perp}$, $v_1^*v_1 =s(\nu)-s(|\sqrt{\sigma_\varrho}\sqrt{\sigma_\nu}\,|)$ and $v_2 v_2^*\leq p^\perp$, $ v_2^*v_2 = s(\varrho)$.
Clearly, from this $v_2^*v_1={\mathsf{0}}=v_1^*v_2$ follows.
Define $w=v_2$ and $v=v_1+v_2 u^*$. Then, by the previous $w^*w=s(\varrho)$ is fulfilled. We are going to show that $v^*v=s(\nu)$. In fact, $s(\varrho)\geq s(|\sqrt{\sigma_\nu}\sqrt{\sigma_\varrho}\,|)$ implies $u s(\varrho)u^*=u s(|\sqrt{\sigma_\nu}\sqrt{\sigma_\varrho}\,|)u^*=uu^*u u^*=s(|\sqrt{\sigma_\varrho}\sqrt{\sigma_\nu}\,|)$, and thus  one can conclude as follows:
\begin{eqnarray*}
v^*v &=&(v_1^*+u v_2^*)( v_1+v_2 u^*)=v_1^*v_1+u v_2^*v_2u^*\\
&=& s(\nu)-s(|\sqrt{\sigma_\varrho}\sqrt{\sigma_\nu}\,|)+u s(\varrho)u^*\\
&=& s(\nu)-s(|\sqrt{\sigma_\varrho}\sqrt{\sigma_\nu}\,|)+s(|\sqrt{\sigma_\varrho}\sqrt{\sigma_\nu}\,|)= s(\nu)
\end{eqnarray*}
Thus $a=\sqrt{\sigma_\nu}\,v^*\in {\mathcal S}_{\pi,{\mathsf B}({\mathcal H})}(\nu)$,  $b=\sqrt{\sigma_\varrho}\,w^*\in {\mathcal S}_{\pi,{\mathsf B}({\mathcal H})}(\varrho)$. We are going to calculate
$b a^*=\sqrt{\sigma_\varrho}\,w^*v\sqrt{\sigma_\nu}$.
By definition of $v, w$ and $v_1, v_2$ we see that $w^*v=v_2^*(v_1+v_2 u^*)=v_2^*v_2 u^*=s(\varrho) u^*$. Hence, $b a^*=\sqrt{\sigma_\varrho}\,s(\varrho) u^*\sqrt{\sigma_\nu}=\sqrt{\sigma_\varrho}\, u^*\sqrt{\sigma_\nu}$. But then
${{\mathsf{tr}}}\,b a^*={{\mathsf{tr}}}\,\sqrt{\sigma_\varrho}\, u^*\sqrt{\sigma_\nu}= {{\mathsf{tr}}}\,\sqrt{\sigma_\nu}\sqrt{\sigma_\varrho}\, u^*={{\mathsf{tr}}}\,u|\sqrt{\sigma_\nu}\sqrt{\sigma_\varrho}|\, u^*={{\mathsf{tr}}}\,u^*u|\sqrt{\sigma_\nu}\sqrt{\sigma_\varrho}|={{\mathsf{tr}}}\,|\sqrt{\sigma_\nu}\sqrt{\sigma_\varrho}|$. Thus \eqref{nvspez2} can be accomplished in each case of a separable ${\mathcal{H}}$.
\end{proof}
By means of Lemma \ref{nvbas} formula \eqref{bei1} in the case at hand will follow, with the extra feature that the supremum in the definition of $F$ in fact is a maximum.
\begin{corolla}\label{0Bop}
Suppose $\varrho,\nu\in {\mathcal S}_0({\mathsf B}({\mathcal H}))$, $a\in {\mathcal S}_{\pi,{\mathsf B}({\mathcal H})}(\nu)$, $b\in {\mathcal S}_{\pi,{\mathsf B}({\mathcal H})}(\varrho)$. Then
\begin{equation}\label{equi0}
 F({\mathsf B}({\mathcal H})|\varrho,\nu)={{\mathsf{tr}}}\,|a^*b\,|=\max_{\tilde{a}\in {\mathcal S}_{\pi,{\mathsf B}({\mathcal H})}(\nu),\,\tilde{b}\in {\mathcal S}_{\pi,{\mathsf B}({\mathcal H})}(\varrho)} |{{\mathsf{tr}}}\,\tilde{a}^*\tilde{b} |
\end{equation}
\end{corolla}
\begin{proof}
To see \eqref{equi0}, let us consider the linear form $ h_{b,a}^\pi$ given over $\pi({\mathsf B}({\mathcal H}))^{\,\prime}$ by  $$\pi({\mathsf B}({\mathcal H}))^{\,\prime}\ni Z\longmapsto h_{b,a}^\pi(Z)=\langle Z\,b,a\rangle_{{\mathsf{ H.S.}}}$$  Let $a^*b=u|a^*b|$ be the polar decomposition of $a^*b$, and be $U$ the partial isometry in $\pi({\mathsf B}({\mathcal H}))^{\,\prime}$ implemented by multiplication of Hilbert-Schmidt operators with $u$ from the right. Then, $U^*U$ is implemented by multiplication from the right with the support orthoprojection $uu^*$ of the trace class operator $u|a^*b| u^*$ within ${\mathsf B}({\mathcal H})$, and by which  the support of   $\sqrt{u|a^*b| u^*}\in {\mathsf{H.S.}}({\mathcal H})$ with respect of $\pi({\mathsf B}({\mathcal H}))^{\,\prime}$ is implemented. Hence, for each $Z\in \pi({\mathsf B}({\mathcal H}))^{\,\prime}$ one can conclude as follows:
\begin{eqnarray*}
h_{b,a}^\pi(Z)&=& \langle Z\,b,a\rangle_{{\mathsf{ H.S.}}}= \langle b z,a\rangle_{{\mathsf{ H.S.}}} ={{\mathsf{tr}}}\, b z a^* ={{\mathsf{tr}}}\, a^*b  z ={{\mathsf{tr}}}\, u|a^*b| z \\
&=& {{\mathsf{tr}}}\, u|a^*b| u^*u z ={{\mathsf{tr}}}\, \sqrt{u|a^*b| u^*} u z (\sqrt{u|a^*b| u^*})^*\\
&=& \bigl\langle \sqrt{u|a^*b| u^*}\, u z, \sqrt{u|a^*b| u^*}\bigr\rangle_{{\mathsf{ H.S.}}}\\
&=& \bigl\langle Z\,U\sqrt{u|a^*b| u^*}, \sqrt{u|a^*b| u^*}\bigr\rangle_{{\mathsf{ H.S.}}}
\end{eqnarray*}
Since $u^*u=s(|a^*b|)$ holds, from this obviously the following estimate is obtained $$\|h_{b,a}^\pi\|_1\leq \bigl\langle \sqrt{u|a^*b| u^*}, \sqrt{u|a^*b| u^*}\bigr\rangle_{{\mathsf{ H.S.}}}={{\mathsf{tr}}}\,u|a^*b| u^*={{\mathsf{tr}}}\,u^*u|a^*b|={{\mathsf{tr}}}\,|a^*b|$$
On the other hand, since $h_{b,a}^\pi(U^*)=\bigl\langle \sqrt{u|a^*b| u^*}, \sqrt{u|a^*b| u^*}\bigr\rangle_{{\mathsf{ H.S.}}}$ and $\|U^*\|=1$ hold, the upper bound ${{\mathsf{tr}}}\,|a^*b|$ is attained. Hence we find that $\|h_{b,a}^\pi\|_1={{\mathsf{tr}}}\,|a^*b|$ holds and which implies $ F({\mathsf B}({\mathcal H})|\varrho,\nu)={{\mathsf{tr}}}\,|a^*b|$, by Lemma \ref{bas3}\,\eqref{bas3aa}. The latter has to be true for each choice of $a\in {\mathcal S}_{\pi,{\mathsf B}({\mathcal H})}(\nu)$, $b\in {\mathcal S}_{\pi,{\mathsf B}({\mathcal H})}(\varrho)$. Thus, especially also
$
F({\mathsf B}({\mathcal H})|\varrho,\nu)={{\mathsf{tr}}}\,|a^*b|={{\mathsf{tr}}}\,|\sqrt{\sigma_\nu}\sqrt{\sigma_\varrho}\,|
$. By Lemma \ref{nvbas} this implies
\eqref{equi0}.
\end{proof}
\begin{remark}\label{polspez}
 In terms of the notations of \ref{basfa}, $U$ introduced in the proof of Corollary \ref{0Bop} equals the partial isometry $v_{b,a}^\pi\in \pi({\mathsf B}({\mathcal H}))^{\,\prime}$ of the polar decomposition of $h_{b,a}^\pi$.
\end{remark}
\subsubsection{The hypothesis of the intermediate faithfulness}\label{beimatrix3}
For faithful normal states  with $\varrho\not=\nu$, and in referring to the density operators $\sigma_\nu$ and $\sigma_\varrho$ of the states in question, we are going to exemplify Lemma \ref{supp5}, now. In the notations of Example \ref{Bop}, let $a, b$ be chosen as $a=\sqrt{\sigma_\nu}$, $b=\sqrt{\sigma_\varrho}$, and define the normal positive linear form $\omega$ satisfying
\begin{equation}\label{fquest}
\xi(\sqrt{\sigma_\nu})+ \sqrt{\sigma_\nu}\in {\mathcal S}_{\pi,{\mathsf B}({\mathcal H})}(\omega)
\end{equation}
According to  Example \ref{Bop}\,\eqref{equi2}, since $s(\nu)=s(\varrho)={\mathsf 1}$ holds we have $\varrho\dashv\vdash \nu$ and the partial isometry $u$ of Example \ref{Bop}\,\eqref{equi3} is unitary.
Then, if $x\in {\mathsf{H.S.}}({\mathcal H})$ is given by
\begin{subequations}\label{fquest0}
\begin{equation}\label{fquest1}
x=\sqrt{\sigma_\varrho} u^*+ \sqrt{\sigma_\nu}
\end{equation}
from \eqref{fquest} we get $\omega(\cdot)={{\mathsf{tr}}}\, \sigma_\omega (\cdot)$, with positive trace-class operator $\sigma_\omega$ given by
\begin{equation}\label{fquest2}
\sigma_\omega=x x^*=|x^*|^2
\end{equation}
\end{subequations}
\begin{remark}\label{invarc}
Note that equally well $\sigma_\omega=|y^*|^2$ holds, with $y$ represented by
$$y=x u=\xi(\sqrt{\sigma_\varrho})+\sqrt{\sigma_\varrho}=\sqrt{\sigma_\nu} u+\sqrt{\sigma_\varrho}\in {\mathcal S}_{\pi,{\mathsf B}({\mathcal H})}(\omega) $$
which formally is arising by a r\^{o}le reversal of $\nu$ with $\varrho$ in \eqref{fquest1}. Note that the latter requires substituting $u$ with $u^*$ for, the polar decomposition of $\sqrt{\sigma_\varrho} \sqrt{\sigma_\nu}$ reads
\begin{equation}\label{polzrev}
\sqrt{\sigma_\varrho}\sqrt{\sigma_\nu}=u^*\bigl|\sqrt{\sigma_\varrho}\sqrt{\sigma_\nu}\,\bigr|
\end{equation}
\end{remark}
In respect of $\{\pi,{\mathcal{H}}_\pi\}$, the question of interest will be whether or not by the affine superposition \eqref{fquest} of the special two implementing vectors of the given faithful normal states $\nu$, $\varrho$ a faithful $\omega$ is implemented again. In case of an affirmative answer, we will say that the {\em hypothesis of the intermediate faithfulness} is satisfied in respect of $\nu$, $\varrho$.

Relating the latter, note that in view of $x$ and $\sigma_\omega$ defined in \eqref{fquest}  we have  $s(\omega)=s(\sigma_\omega)=s(|x^*|)$, and if the polar decomposition of $x$ reads
\begin{subequations}\label{ifl0}
\begin{equation}\label{ifl1}
x=w|x|
\end{equation}
then $|x^*|=w |x| w^*$ and $s(\omega)=s(|x^*|)=w w^*$, $s(|x|)=w^* w$. Since both $\sigma_\nu$ and  $\sigma_\varrho$ are density operators of full support and $\sqrt{\sigma_\nu} \sqrt{\sigma_\varrho}=u |\sqrt{\sigma_\nu} \sqrt{\sigma_\varrho}|$ holds, we get
\begin{eqnarray*}
|x|^2=x^* x &=& u \sigma_\varrho u^*+ \sigma_\nu + u \sqrt{\sigma_\varrho} \sqrt{\sigma_\nu} + \sqrt{\sigma_\nu} \sqrt{\sigma_\varrho} u^*\\
&= & u \sigma_\varrho u^*+ \sigma_\nu +  2  u |\sqrt{\sigma_\nu} \sqrt{\sigma_\varrho}| u^*
\end{eqnarray*}
Obviously $|x|^2$ is of full support. Hence, from this in view of the above,
\begin{equation}\label{if1}
  w w^*=s(\omega),\ w^*w={\mathsf{1}}
\end{equation}
\end{subequations}
follows, that is, $w$ is an isometry accomplishing the equivalence $s(\omega)\sim {\mathsf{1}}$ in the sense of the $vN$-comparability of orthoprojections. Therefore, as a consequence of \eqref{ifl0} we may note the following criterion.
\begin{lemma}\label{if}
Intermediate faithfulness holds if, and only if, $w$ is unitary.
\end{lemma}
Apply this to identify some further situations where the fulfilment of the hypothesis can be settled rather elementary.
\begin{example}\label{arcsupp}
For $\omega$ and $\sigma_\omega$ given by \eqref{fquest} and \eqref{fquest2} the following hold:
\begin{enumerate}
  \item\label{arcsupp1}
  if $\sigma_\omega$ is of finite rank, then ${\mathsf{dim}}\, {\mathcal{H}} <\infty$;
  \item\label{arcsupp2}  if ${\mathsf{dim}}\, {\mathcal{H}}<\infty$ is fulfilled, then $\sigma_\omega$ is of full rank;
  \item\label{arcsupp3} if $\sigma_\nu \smile \sigma_\varrho$ holds, then $\sigma_\omega$ is of full rank.
\end{enumerate}
\end{example}
\begin{proof}
In view of $s(\omega)\sim {\mathsf{1}}$ both \eqref{arcsupp1} and \eqref{arcsupp2} follow from elementary dimension theory in ${\mathsf{B}}({\mathcal{H}})$. Relating \eqref{arcsupp3}, remark that for  mutually commuting density operators $\sigma_\nu$ and $\sigma_\varrho$ obviously $u={\mathsf{1}}$ is fulfilled. Hence, as a sum of two strictly positive trace-class operators $x$ has to be strictly positive, too. Thus $\sigma_\omega=x^2$ is of full support.
\end{proof}
\begin{example}\label{Bop1}
Let $\varrho$ and $\nu$ be faithful normal states such that there exists a finite or countably infinite, strictly  ascending, directed system $\{P_k\}$ of orthoprojections over ${\mathcal H}$ with ${\mathsf{dim}}\, P_k {\mathcal H}<\infty$, for all $k$, and obeying the following two conditions:
\begin{enumerate}
\item\label{Bop1a}
${\mathsf{l.u.b.}} \{P_k\}={\mathsf{1}}$;
\item\label{Bop1b}
$\sigma_\nu \smile P_k$, $\sigma_\varrho  \smile P_k$, for all $k$.
\end{enumerate}
Then, the normal positive linear form $\omega$ implemented in accordance with formula \eqref{fquest}
is faithful, and is obeying  $\sigma_\omega \smile P_k $  again, for all $k$.
\end{example}
\begin{proof}
Let $\Delta P_1=P_1$ and $\Delta P_k=P_k-P_{k-1}$, for $k> 1$. By \eqref{Bop1a} $\{\Delta P_k\}$ is a finite or countably infinite decomposition of the unity into mutually orthogonal orthoprojections $\Delta P_k$ of finite rank. Note that \eqref{Bop1b} implies that each $\Delta P_k$ is commuting with the unitary $u$ figuring in \eqref{fquest1}, too. Hence, $\Delta P_k x=x \Delta P_k$, for each $k$. But then each $\Delta P_k$ is commuting with $|x|$, $w$ and $|x^*|$.  By \eqref{fquest2} this implies commutation of $\sigma_\omega$ with each $P_k$. Also, according to \eqref{if1},  $w_k=w\Delta P_k=\Delta P_k w$ is a partial isometry obeying $w_k w_k^*=w w^*\Delta P_k \leq w^*w \Delta P_k= w_k^* w_k=\Delta P_k$, for all $k$. Since each $\Delta P_k$ is of finite rank, the latter estimate implies  $w^*w \Delta P_k=w w^* \Delta P_k=\Delta P_k$, for each subscript. Hence, by summing up over $k$ this yields $w w^*=w^* w={\mathsf{1}}$. That is, under the given conditions $w$ proves to be unitary. By Lemma \ref{if} then $s(\omega)={\mathsf{1}}$ follows.
\end{proof}
\begin{corolla}\label{pobei00}
The premises of \textup{Example \ref{Bop1}} are satisfied if, and only if, $\sigma_\nu\smile \sigma$ and $\sigma_\varrho\smile \sigma$ are fulfilled, with some density operator $\sigma$ of full support. Especially, for ${\mathsf{dim}}\,{\mathcal{H}}<\infty$, the premises are satisfied for any two faithful normal states $\nu$, $\varrho$.
\end{corolla}
\begin{proof}
Suppose $\{P_k\}$ as in Example \ref{Bop1}. Let $\{\lambda_k\}$, with $\lambda_1> \lambda_2> \lambda_3 > \cdots >0$, be a sequence of reals and $\sum_k \lambda_k {\mathsf{dim}} \,\Delta P_k=1$. Then  $\sigma= \sum_k (\lambda_k-\lambda_{k+1})P_k=\sum_k \lambda_k \Delta P_k $ is a faithful density operator obeying $\sigma_\nu\smile \sigma$ and $\sigma_\varrho\smile \sigma$. On the other hand, let  $\sigma=\sum_k \lambda_k q_k $ be the spectral representation of a faithful density operator, with sequence $\{\lambda_k\}$ of different proper values and decomposition $\{q_k\}$ of the unity into the mutually orthogonal orthoprojections $q_k$ projecting onto the characteristic subspaces $q_k{\mathcal{H}}$ to the corresponding proper values $\lambda_k$. By spectral theory,  $\sigma_\nu\smile \sigma$ and $\sigma_\varrho\smile \sigma$ imply $\sigma_\nu\smile q_k$ and  $\sigma_\varrho\smile q_k$, for each $k$, and $P_k=\sum_{j\leq k} q_j$ can be taken in
Example \ref{Bop1}. Obviously, for ${\mathsf{dim}}\,{\mathcal{H}}<\infty$, $\sigma=(1/{\mathsf{dim}}\,{\mathcal{H}})\,{\mathsf{1}}$ can be chosen, for each $\nu$ and $\varrho$.
\end{proof}
Whereas according to Corollary \ref{pobei00} for ${\mathsf{dim}}\, {\mathcal{H}}<\infty$ the hypothesis of the intermediate faithfulness is true in any case of two faithful normal states, in
the infinite dimensional case the premises of Example \ref{Bop1} can fail to be satisfied.
\begin{example}\label{Bop20}
Suppose ${\mathcal{H}}$ is separable with ${\mathsf{dim}}\,{\mathcal{H}}=\infty$. Let $\varrho$ be a faithful normal state over ${\mathsf{B}}({\mathcal{H}})$. Let $\sigma_\varrho=\sum_k \varrho_k q_k$ be the spectral representation of the corresponding density operator $\sigma_\varrho$, where the $\varrho_k$'s are the different proper values of $\sigma_\varrho$ and the $q_k$'s are the corresponding spectral projections. Let $\psi_1$, $\psi_2$ be unit vectors with $\psi_1\in {\mathcal{R}}(\sigma_\varrho)$ and $\psi_2\in {\mathcal{R}}(\sqrt{\sigma_\varrho})$, and obeying $q_k\psi_j\not=0$, for all $k\in {\mathbb{N}}$ and $j=1,2$. Let $p_1$, $p_2$ be the minimal orthoprojections such that $p_j\psi_j=\psi_j$, for $j=1,2$. Consider those faithful normal states $\nu$ which are given by the following  proportionalities:
\begin{enumerate}
  \item\label{Bop20a} $\sigma_\nu\sim \sigma_\varrho + \lambda p_1$, for fixed real $\lambda> -1/\langle \sigma_\varrho^{-1} \psi_1,\psi_1\rangle$;
  \item\label{Bop20b} $\sqrt{\sigma_\nu}\sim \sqrt{\sigma_\varrho}+\lambda p_2$, for fixed real $\lambda> -1/\bigl\langle \sqrt{\sigma_\varrho}^{-1} \psi_2,\psi_2\bigr\rangle$.
\end{enumerate}
In both cases, the premises of Example \ref{Bop1} do not apply in respect of $\nu$ and $\varrho$. Moreover, for any $\varepsilon>0$, the respective $\lambda$ can be chosen such that $\|\varrho-\nu\|_1<\varepsilon$.
\end{example}
\begin{proof}
In order to see that for $\nu$ and $\varrho$ in  accordance with \eqref{Bop20a} or \eqref{Bop20b} the premises of Example \ref{Bop1} cannot be satisfied, apply Corollary \ref{hilfcomm} from the appendix,  with $\Lambda=\sigma_\varrho$ and $\Lambda=\sqrt{\sigma_\varrho}$, respectively. The last assertion follows since $\sigma_\nu$ is continuously depending on the parameter $\lambda$, and $\lambda=0$ is an inner point of the respective parameter range.
\end{proof}
\subsubsection{When does the hypothesis of the intermediate faithfulness hold?}\label{beimatrix4}
In operator theoretical notions, the hypothesis of the intermediate faithfulness is equivalent to the assertion that the left support of $x$ be full, that is, $\ell(x)={\mathsf{1}}$. The latter is the same as asserting the range ${\mathcal{R}}(x)$ of $x$ to be dense in ${\mathcal{H}}$. Therefore, if a linear operator $a_0$ operating from the dense linear submanifold ${\mathcal{D}}(a_0)=\sqrt{\sigma_\nu}{\mathcal{H}}$  into ${\mathcal{H}}$ is defined as
\begin{subequations}\label{a0def}
\begin{equation}\label{a0mapneu}
  a_0:\ {\mathcal{D}}(a_0)\ni \sqrt{\sigma_\nu}\varphi\longmapsto \sqrt{\sigma_\varrho}\,u^*\varphi
\end{equation}
then we have ${\mathcal{R}}(x)=\bigl(\sqrt{\sigma_\varrho} u^*+ \sqrt{\sigma_\nu}\bigr){\mathcal{H}}=(a_0+{\mathsf{1}}) {\mathcal{D}}(a_0)={\mathcal{R}}(a_0+{\mathsf{1}})$. Also, it is easily inferred that instead of \eqref{a0mapneu} equally well
\begin{equation}\label{essa3neu}
a_0=\sqrt{\sigma_\varrho}\,u^*\left(\sqrt{\sigma_\nu}\right)^{-1}
\end{equation}
\end{subequations}
can be written. Since the inverse of $\sqrt{\sigma_\nu}$ is self-adjoint over ${\mathcal{D}}(a_0)$ and $\sqrt{\sigma_\varrho}\,u^*$ is bounded, from  \eqref{essa3neu} by standard procedures the data of the adjoint read as
\begin{subequations}\label{essa4neu}
\begin{eqnarray}
\label{essa4aneu}
  a_0^* &=& \left(\sqrt{\sigma_\nu}\right)^{-1} u \sqrt{\sigma_\varrho} \\
\label{essa4bneu}
{\mathcal{D}}(a_0^*) &=& \Bigl\{\xi\in {\mathcal{H}}: u\sqrt{\sigma_\varrho}\, \xi\in {\mathcal{D}}(a_0) \Bigr\}
\end{eqnarray}
\end{subequations}
Also, from \eqref{polzrev} and since  $\bigl|\sqrt{\sigma_\varrho}\sqrt{\sigma_\nu}\,\bigr|$ is self-adjoint, we infer that
\begin{equation}\label{a0x01}
u\sqrt{\sigma_\varrho}\sqrt{\sigma_\nu}=\bigl|\sqrt{\sigma_\varrho}\sqrt{\sigma_\nu}\,\bigr|=\sqrt{\sigma_\nu}\sqrt{\sigma_\varrho}u^*
\end{equation}
Accordingly, for each $\varphi\in {\mathcal{H}}$ we have $$ u\sqrt{\sigma_\varrho} \bigl(\sqrt{\sigma_\nu}\,\varphi\bigr)=\bigl|\sqrt{\sigma_\varrho} \sqrt{\sigma_\nu}\bigr|\,\varphi=\sqrt{\sigma_\nu} \bigl(\sqrt{\sigma_\varrho} u^*\varphi\bigr)$$
Hence ${\mathcal{D}}(a_0)\subset {\mathcal{D}}(a_0^*)$ and $a_0^* \sqrt{\sigma_\nu}\,\varphi=\sqrt{\sigma_\varrho} u^*\varphi=a_0 \sqrt{\sigma_\nu}\,\varphi$ are seen. Thus $a_0^*\supset a_0$ holds, which means that $a_0$ is a symmetric linear operator. Also, for each $\varphi\in {\mathcal{H}}$
\begin{eqnarray*}
  \langle a_0 \sqrt{\sigma_\nu}\,\varphi,\sqrt{\sigma_\nu}\,\varphi\rangle &=& \langle \sqrt{\sigma_\varrho}\,u^*\varphi,\sqrt{\sigma_\nu}\,\varphi\rangle=\langle \varphi,u\sqrt{\sigma_\varrho}\sqrt{\sigma_\nu}\,\varphi\rangle \\
   &=& \langle \varphi,|\sqrt{\sigma_\varrho}\sqrt{\sigma_\nu}|\,\varphi\rangle=\langle |\sqrt{\sigma_\varrho}\sqrt{\sigma_\nu}|\,\varphi,\varphi\rangle \geq  0
\end{eqnarray*}
holds. Thus $a_0$ is positive symmetric. As mentioned above,  $[{\mathcal{R}}(a_0+{\mathsf{1}})]={\mathcal{H}}$ is fulfilled if, and only if,  the hypothesis of the intermediate faithfulness is satisfied. But by the facts mentioned in Remark \ref{saka6}, for a densely defined, positive linear operator $a_0$ the condition $[{\mathcal{R}}(a_0+{\mathsf{1}})]={\mathcal{H}}$ equivalently means that $a_0$ has to be essentially self-adjoint. The latter proves that \eqref{supp5full1} and \eqref{supp5full2} of the following result are mutually equivalent.
\begin{lemma}\label{supp5full}
Suppose $\sigma_\nu, \sigma_\varrho$ are of full rank. Then, $a_0$ is a positive symmetric operator and the following are mutually equivalent:
\begin{enumerate}
  \item \label{supp5full1}
  the hypothesis of the intermediate faithfulness is satisfied in respect of $\nu$, $\varrho$;
  \item \label{supp5full2}
  $a_0$ is essentially self-adjoint.
\end{enumerate}
Moreover, in case \eqref{supp5full2} holds, the closure $\bar{a}_0$ is the unique solution $y$ of the equation
\begin{equation}\label{supp5full3}
\sqrt{\sigma_\varrho}\,u^*= y \sqrt{\sigma_\nu}
\end{equation}
provided $y$ is taken from the set of all densely defined,  self-adjoint linear operators.
\end{lemma}
\begin{proof}
For essential self-adjoint $a_0$ the closure $y=\bar{a}_0$ is  self-adjoint and is obeying
\begin{equation}\label{supp5fullzak}
a_0 \sqrt{\sigma_\nu}=y \sqrt{\sigma_\nu}
\end{equation}
Thus, due to \eqref{essa3neu}, the equation \eqref{supp5full3} possesses a special solution $y=\bar{a}_0$. On the other hand, let $y$ be a linear operator with $y=y^*$ satisfying equation \eqref{supp5full3}. Then, by the structure of the latter, the domain of definition ${\mathcal{D}}(y)$ of $y$ obviously has to obey ${\mathcal{D}}(y)\supset {\mathcal{R}}\bigl(\sqrt{\sigma_\nu}\bigr)={\mathcal{D}}(a_0)$.
In view of \eqref{supp5full3} and \eqref{essa3neu} we infer that \eqref{supp5fullzak} has to be fulfilled, for the $y$ at hand, too. That is, $y$ is self-adjoint and extends $a_0$, $y\supset a_0$. Accordingly, since owing to self-adjointness $y$ is closed, $y\supset\bar{a}_0$ has to be fulfilled. By essential self-adjointness of $a_0$ from this then  $y=\bar{a}_0$ follows. That is, the closure is  the unique solution of \eqref{supp5full3}.
\end{proof}
Conclude with a couple of conditions assuring the essential self-adjointness of $a_0$.
\begin{corolla}\label{boundkato}
Each of the following assumptions implies $a_0$ to be essentially self-adjoint:
\begin{enumerate}
\item\label{boundkato1}
  $a_0$ is bounded;
\item\label{boundkato2}
  $\sigma_\varrho^{-1}+\bigl|\sqrt{\sigma_\nu}\sqrt{\sigma_\varrho}\bigr|^{-1}$ is essentially self-adjoint.
\end{enumerate}
\end{corolla}
\begin{proof}
If $a_0$ is bounded, then the symmetry of $a_0$ implies the closure $\bar{a}_0$ to be a bounded, self-adjoint linear operator on $\mathcal H$.
If \eqref{boundkato2} holds, this relates to the positive linear operator
\begin{equation*}
  c_0=\sigma_\varrho^{-1}+\bigl|\sqrt{\sigma_\nu}\sqrt{\sigma_\varrho}\bigr|^{-1}
\end{equation*}
with ${\mathcal{D}}(c_0)=\sigma_\varrho {\mathcal{H}}\cap \bigl|\sqrt{\sigma_\nu}\sqrt{\sigma_\varrho}\bigr| {\mathcal{H}}$ being dense.
It then follows that ${\mathcal{R}}(c_0)$ is dense. In fact, for $\xi\in {\mathcal{H}}$ one has $\langle c_0\psi,\xi\rangle=0$ for all $\psi\in {\mathcal{D}}(c_0)$ if, and only if, $\xi\in {\mathcal{D}}(c_0^*)$ and
$\langle \psi,c_0^*\xi\rangle=0$ for all $\psi\in {\mathcal{D}}(c_0)$. Hence $c_0^*\xi={\mathsf{0}}$ must hold.
Since $c_0$ is essentially self-adjoint, $c_0^*=\bar{c}_0$, $\xi\in{\mathcal{D}}(\bar{c}_0)$ and $\bar{c}_0\xi={\mathsf{0}}$ hold. Thus, there is a sequence $\{\xi_n\}\subset {\mathcal{D}}(c_0)$ with $\xi_n\rightarrow \xi$ and $c_0 \xi_n  \rightarrow \bar{c}_0\xi={\mathsf{0}}$. From this we infer that $\lim_{n\to\infty}\langle c_0 \xi_n,\xi_n\rangle = 0$.
Now,
$$\langle c_0 \xi_n,\xi_n\rangle= \langle \sigma_\varrho^{-1} \xi_n,\xi_n\rangle+\langle \bigl|\sqrt{\sigma_\nu}\sqrt{\sigma_\varrho}\bigr|^{-1} \xi_n,\xi_n\rangle$$
is fulfilled, with $\langle \sigma_\varrho^{-1} \xi_n,\xi_n\rangle\geq 0$ and $\langle \bigl|\sqrt{\sigma_\nu}\sqrt{\sigma_\varrho}\bigr|^{-1} \xi_n,\xi_n\rangle\geq 0$, by positivity of  the two characteristic constituents of $c_0$. Accordingly, e.g.\,also
$$\lim_{n\to\infty}\langle \bigl|\sqrt{\sigma_\nu}\sqrt{\sigma_\varrho}\bigr|^{-1} \xi_n,\xi_n\rangle = 0$$ has to be fulfilled.  Now, there has to exist a sequence $\{\varphi_n\}$ such that $\xi_n=\bigl|\sqrt{\sigma_\nu}\sqrt{\sigma_\varrho}\bigr| \,\varphi_n$, for all $n\in {\mathbb{N}}$. Thus $\lim_{n\to\infty}\langle \varphi_n,\bigl|\sqrt{\sigma_\nu}\sqrt{\sigma_\varrho}\bigr|\,\varphi_n\rangle = 0$, and which is the same as
$$\lim_{n\to\infty} \bigl|\sqrt{\sigma_\nu}\sqrt{\sigma_\varrho}\bigr|^{\frac{1}{2}}\, \varphi_n={\mathsf{0}}$$ But then also $\lim_{n\to\infty} \bigl|\sqrt{\sigma_\nu}\sqrt{\sigma_\varrho}\bigr|\, \varphi_n=\lim_{n\to\infty} \xi_n={\mathsf{0}}$ is fulfilled. Thus $\xi={\mathsf{0}}$ is seen, and therefore ${\mathsf{0}}$ is the only vector orthogonal to $c_0\psi$, for all $\psi\in {\mathcal{D}}(c_0)$. That is, ${\mathcal{R}}(c_0)$ has to be dense in ${\mathcal{H}}$, and therefore the following density relation holds
\begin{equation}\label{den}
\bigl[\sqrt{\sigma_\varrho} {\mathcal{R}}(c_0)\bigr]={\mathcal{H}}
\end{equation}
Now, let ${\mathcal{D}}_0=\bigl\{\varphi\in {\mathcal{H}}: \bigl|\sqrt{\sigma_\nu}\sqrt{\sigma_\varrho}\bigr| \varphi \in {\mathcal{D}}(c_0)\bigr\}$. For each $\varphi\in {\mathcal{D}}_0$ there exists $\xi(\varphi)\in {\mathcal{H}}$ with $ \bigl|\sqrt{\sigma_\nu}\sqrt{\sigma_\varrho}\bigr| \varphi =\sigma_\varrho \xi(\varphi)$. Thus, since $\sqrt{\sigma_\varrho} \sqrt{\sigma_\nu} =\bigl|\sqrt{\sigma_\nu}\sqrt{\sigma_\varrho}\bigr| u^*$ holds, we infer that
 for each $\varphi\in {\mathcal{D}}_0$ the following is fulfilled
\begin{equation*}
\bigl|\sqrt{\sigma_\nu}\sqrt{\sigma_\varrho}\bigr|\varphi= \sqrt{\sigma_\varrho} \sqrt{\sigma_\nu} u \varphi=\sigma_\varrho \xi(\varphi)
\end{equation*}
Since $\sqrt{\sigma_\varrho}$ as a map is injective, from the latter relation to
$
\sqrt{\sigma_\nu} u \varphi=\sqrt{\sigma_\varrho} \xi(\varphi)
$
can be concluded, for $\varphi\in {\mathcal{D}}_0$. Accordingly, in view of the above relations, we see that
\begin{eqnarray*}                                                                                         \sqrt{\sigma_\varrho}\bigl( c_0 \bigl|\sqrt{\sigma_\nu}\sqrt{\sigma_\varrho}\bigr|\varphi\bigr) &=& \sqrt{\sigma_\varrho}\bigl(\sigma_\varrho^{-1} \bigl|\sqrt{\sigma_\nu}\sqrt{\sigma_\varrho}\bigr|\varphi+ \bigl|\sqrt{\sigma_\nu}\sqrt{\sigma_\varrho}\bigr|^{-1} \bigl|\sqrt{\sigma_\nu}\sqrt{\sigma_\varrho}\bigr|\varphi\bigr) \\
&=& \sqrt{\sigma_\varrho} (\xi(\varphi)+ \varphi)= \sqrt{\sigma_\varrho} \xi(\varphi) + \sqrt{\sigma_\varrho} \varphi=\sqrt{\sigma_\nu} u \varphi + \sqrt{\sigma_\varrho} \varphi\\
& =& (\sqrt{\sigma_\nu}+ \sqrt{\sigma_\varrho} u^*) u\varphi = ({\mathsf{1}}+ a_0)\sqrt{\sigma_\nu}u\varphi
\end{eqnarray*}
where in the last step \eqref{essa3neu} has been made use of. Thus, in view of \eqref{den} the subset $\sqrt{\sigma_\nu}u {\mathcal{D}}_0\subset {\mathcal{D}}(a_0)$ by ${\mathsf{1}}+ a_0$ is mapped onto a dense subset of ${\mathcal{H}}$. But then even more $[{\mathcal{R}}(({\mathsf{1}}+ a_0))]={\mathcal{H}}$ has to hold. Thus, by our standard argumentation for densely defined, positive linear operators, $a_0$ is essentially self-adjoint.
\end{proof}
\subsubsection{Special cases where the hypothesis holds}\label{beimatrix5}
As noted in Corollary \ref{boundkato}\,\eqref{boundkato1}, $a_0$ is essentially self-adjoint in the bounded case, and thus then by Lemma \ref{supp5full} in respect of $\nu$, $\varrho$  intermediate faithfulness will occur. Boundedness of  $a_0$ can be characterized in terms of domination $\ll$, see Remark \ref{saka}, \eqref{saka0}-\eqref{saka2}.
\begin{lemma}\label{supp5bound}
For  density operators $\sigma_\nu, \sigma_\varrho$ of full rank, the following are equivalent:
\begin{enumerate}
  \item\label{supp5bound1} $a_0$ is bounded on ${\mathcal{D}}(a_0)$;
  \item\label{supp5bound2} $u\sigma_\varrho u^*\ll \sigma_\nu$, that is, $u\sigma_\varrho u^*\leq \gamma \sigma_\nu$, for some  real $\gamma >0$;
  \item\label{supp5bound3} $\bigl|\sqrt{\sigma_\varrho}\sqrt{\sigma_\nu}\bigr|\ll\sigma_\nu$, that is, $\bigl|\sqrt{\sigma_\varrho}\sqrt{\sigma_\nu}\bigr|\leq \gamma \sigma_\nu$, for some  real $\gamma >0$.
\end{enumerate}
Thereby, in \eqref{supp5bound2} and \eqref{supp5bound3}, there is a minimal bound $\gamma$ with this property fulfilled.
\end{lemma}
\begin{proof}
Owing to ${\mathcal{D}}(a_0)=\sqrt{\sigma_\nu} {\mathcal{H}}$ and \eqref{essa3neu}, for each $\eta\in {\mathcal{H}}$ we have
\begin{equation*}
  \|a_0 \sqrt{\sigma_\nu}\eta \|^2=\|\sqrt{\sigma_\varrho}\,u^*\eta\|^2=\langle u\sigma_\varrho u^* \eta,\eta\rangle
\end{equation*}
Hence, since $\|\sqrt{\sigma_\nu}\eta \|^2=\langle \sigma_\nu \eta,\eta \rangle$ holds, boundedness of $a_0$, that is, the condition
$$ \|a_0 \sqrt{\sigma_\nu}\eta \|\leq \alpha \|\sqrt{\sigma_\nu}\eta \|, \,\forall\,\eta\in {\mathcal{H}},$$
with real $\alpha > 0$, is equivalent to
$
  \langle u\sigma_\varrho u^* \eta,\eta\rangle \leq \alpha^2  \langle \sigma_\nu \eta,\eta \rangle
$
for all $\eta\in {\mathcal{H}}$. The latter is equivalent to $u\sigma_\varrho u^*\leq \gamma \sigma_\nu$, with $\gamma=\alpha^2\geq \|a_0\|^2$.
Hence, $\gamma= \|a_0\|^2$ is the minimal possible value. This is \eqref{supp5bound1} $\Leftrightarrow$  \eqref{supp5bound2}.

To see \eqref{supp5bound1} $\Leftrightarrow$  \eqref{supp5bound3} remind that, for a positive symmetric operator $a_0$, boundedness equivalently is stated through the condition asserting existence of $\alpha>0$ such that
$$ \bigl\langle a_0 \sqrt{\sigma_\nu}\eta,\sqrt{\sigma_\nu}\eta\bigr\rangle \leq \alpha \bigl\|\sqrt{\sigma_\nu}\eta \bigr\|^2$$
gets satisfied, for all $\eta\in {\mathcal{H}}$. This is the same as asserting that for any $\eta\in {\mathcal{H}}$
\begin{equation*}
\begin{split}
   \bigl\langle\bigl|\sqrt{\sigma_\varrho}\sqrt{\sigma_\nu}\bigr|\, \eta,\eta\bigr\rangle & =
\bigl\langle u\bigl|\sqrt{\sigma_\nu}\sqrt{\sigma_\varrho}\bigr| u^* \eta,\eta\bigr\rangle=
\bigl\langle \sqrt{\sigma_\nu}\sqrt{\sigma_\varrho} u^* \eta,\eta\bigr\rangle=\bigl\langle \sqrt{\sigma_\varrho} u^* \eta,\sqrt{\sigma_\nu}\eta\bigr\rangle\\
     &=\bigl\langle a_0 \sqrt{\sigma_\nu}\eta,\sqrt{\sigma_\nu}\eta\bigr\rangle\leq \alpha \bigl\langle  \sqrt{\sigma_\nu}\eta,\sqrt{\sigma_\nu}\eta\bigr\rangle=\alpha \bigl\langle  \sigma_\nu \eta,\eta\bigr\rangle
\end{split}
\end{equation*}
has to be fulfilled. Hence, \eqref{supp5bound1} is equivalent to $\bigl|\sqrt{\sigma_\varrho}\sqrt{\sigma_\nu}\bigr|\leq \alpha \sigma_\nu$, for some $\alpha>0$. Clearly, if $a_0$ is bounded, then $\alpha$ cannot be smaller than $\|a_0\|$, but $\alpha=\|a_0\|$ is possible. Thus, any $\gamma\geq \|a_0\|$ is possible in \eqref{supp5bound3}.
\end{proof}
Apply Lemma \ref{supp5bound} to the well-known context of domination of states.
\begin{corolla}\label{domi}
Suppose $\varrho\ll \nu$ or $\nu\ll \varrho$, for $\varrho,\nu\in {\mathcal S}_0^{\mathsf{faithful}}({\mathsf B}({\mathcal H}))$. Then, the hypothesis of the intermediate faithfulness is fulfilled.
\end{corolla}
\begin{proof}
The assumption $\varrho\ll \nu$ is equivalent to  $\sigma_\varrho\leq \lambda \,\sigma_\nu$, for real $\lambda>0$. Hence,
\begin{equation*}
\left|\sqrt{\sigma_\varrho} \sqrt{\sigma_\nu} \right|^2= \sqrt{\sigma_\nu} \sigma_\varrho \sqrt{\sigma_\nu}\leq  \lambda \,\sigma_\nu^2
\end{equation*}
By operator monotony \cite{Dono:74} of the square root function $|\sqrt{\sigma_\varrho} \sqrt{\sigma_\nu} |\leq \sqrt{\lambda} \,\sigma_\nu$ is implied. Hence condition \eqref{supp5bound3} of Lemma \ref{supp5bound} is fulfilled. By  Lemma \ref{supp5bound} this guarantees for boundedness of the positive symmetric $a_0$. Hence, $\bar{a}_0$ is self-adjoint. By Lemma \ref{supp5full} then  intermediate faithfulness holds. In case of  $\nu\ll \varrho$ the same conclusions follow with the state arguments seen in reversed order. By Remark \ref{invarc} however, intermediate faithfulness in respect of $\nu$ and $\varrho$ is the same as that in respect of $\varrho$ and $\nu$.
\end{proof}
\begin{example}\label{domicount}
For a faithful normal state $\varrho$, let $\nu$  be any faithful normal state constructed in accordance with Example \ref{Bop20}\,\eqref{Bop20a}. Then, whereas in respect of $\nu$ and $\varrho$ the premises of Example \ref{Bop1} are not fulfillable, nevertheless the hypothesis of the intermediate faithfulness holds. Moreover, for ${\mathsf{dim}}\,{\mathcal{H}}=\infty$, depending from the choice of the parameter $\lambda$ in Example \ref{Bop20}\,\eqref{Bop20a}, the operator $a_0$ may be bounded or unbounded.
\end{example}
\begin{proof}
Relating the construction schema in Example \ref{Bop20}\,\eqref{Bop20a}, if $\nu$ corresponds to an admitted parameter value with $\lambda>0$, then $\sigma_\varrho\leq \alpha\, \sigma_\nu$ follows, for some $\alpha>0$. Hence, $\varrho\ll \nu$ is fulfilled. On the other hand, if the parameter is chosen with $\lambda<0$, then $\nu\ll \varrho$ follows.
In view of Corollary \ref{domi} now the result follows.
\end{proof}
\begin{example}\label{inklus}
Let $\nu$ and $\varrho$ be faithful normal states such that one of the relations $${\mathcal{R}}(\sigma_\nu)\subset {\mathcal{R}}(\sigma_\varrho), {\mathcal{R}}(\sigma_\varrho)\subset {\mathcal{R}}(\sigma_\nu), {\mathcal{R}}(\sqrt{\sigma_\nu})\subset {\mathcal{R}}(\sqrt{\sigma_\varrho}), {\mathcal{R}}(\sqrt{\sigma_\varrho})\subset {\mathcal{R}}(\sqrt{\sigma_\nu})$$
is fulfilled.
Then the hypothesis of the intermediate faithfulness holds.
\end{example}
\begin{proof}
Suppose $\sigma_\nu{\mathcal{H}}\subset \sigma_\varrho{\mathcal{H}}$. Then  $\sigma_\varrho^{-1}\sigma_\nu$ is defined on all of ${\mathcal{H}}$, and owing to
$$\sigma_\varrho^{-1}\sigma_\nu=\bigl(\sigma_\nu \sigma_\varrho^{-1}\bigr)^*$$ has to be closed.
Hence, by the closed graph theorem $y=\sigma_\varrho^{-1}\sigma_\nu$ is bounded on ${\mathcal{H}}$. That is,
$\sigma_\varrho y=\sigma_\nu$ holds, with $y\in {\mathsf{B}}({\mathcal{H}})\backslash \{0\}$. From this the conclusion is
\begin{equation}\label{mon}
\sigma_\nu^2 =\sigma_\varrho yy^*\sigma_\varrho \leq \|yy^*\| \sigma_\varrho^2=\|y\|^2 \sigma_\varrho^2
\end{equation}
Hence, since the square root function is operator monotonous, $\sigma_\nu\leq \|y\| \sigma_\varrho$ follows. In case of ${\mathcal{R}}(\sqrt{\sigma_\nu})\subset {\mathcal{R}}(\sqrt{\sigma_\varrho})$, the conclusions leading to \eqref{mon} remain true with the density operators replaced with their  respective square roots. That is, instead of \eqref{mon} we now have $\sigma_\nu\leq \|y\|^2 \sigma_\varrho$.
Thus, in both cases $\nu \ll \varrho$ is seen to hold.
The remaining cases follow from the previous by interchanging the r\^{o}le of the state arguments, that is $\varrho\ll \nu$  then holds. In either case, in view of Corollary \ref{domi} the assertion follows.
\end{proof}
\begin{remark}\label{Bop1c}
\begin{enumerate}
\item\label{Bop1ca}
Example \ref{arcsupp}\,\eqref{arcsupp2} and \eqref{arcsupp3} are special cases of Example \ref{Bop1}.
\item\label{Bop1cc}
The results of this subsection
remain true for all those normal positive linear forms $\omega$ which are implemented if instead of $x$ in  \eqref{fquest2} $$x=\alpha \sqrt{\sigma_\varrho} u^*+ \beta \sqrt{\sigma_\nu}$$ is used there,
with arbitrary reals with $\alpha\geq 0$, $\beta>0$ (the expression $x$ occurring in \eqref{fquest1} corresponds to the special choice $\alpha=\beta=1$).
\end{enumerate}
\end{remark}
\newpage
\section{The Bures distance along a curve}\label{kurv}
\subsection{The local dilation function in Bures geometry}\label{kurv1}
In the following, the infinitesimal behavior of the Bures distance along the states on a (continuous) curve $\gamma$ in the vicinity of one of the states $\nu$ it passes through will be analyzed.
A curve (oriented curve) $\gamma$ in the state space ${\mathcal S}(M)$ and passing through $\nu$ will be supposed to be given by
an injective continuous map $\Phi_\gamma\,: I\ni t\longmapsto \nu_t\in {\mathcal S}(M)$ acting from some non-trivial interval $I\subset {\mathbb R}$ (equipped with the usual
distance function given by $|x-y|$) into the metric space
$\{{\mathcal S}(M),d_1\}$ such that $\nu=\nu_{t_0}$, for some inner point $t_0\in I$. The curve $\gamma$ passing through $\nu$ then will be understood as an equivalence class modulo homeomorphisms
(orientation preserving homeomorphisms) of such kind of parameterization.  The supporting set $|\gamma|=
\Phi_\gamma(I)$ is a metric subspace of $\{{\mathcal S}(M),d_1\}$ which is
homeomorphic with the interval
$I$, and therefore as a metric structure can be identified with the curve
$\gamma$ (which often tacitly will be done).  The curve $\gamma$ is called oriented, if also the orientation of the given parameterization (which
is determined by the foregiven $I$ and $\Phi_\gamma$) is taken
care of. In
order to simplify in the notations, $\gamma$ and
one of its foregiven (homeomorphic) parameterizations $\Phi_\gamma$ and homeomorphy
interval $I$ will be clashed together into a map-like notation
$\gamma:I\ni t\,\longmapsto\,\nu_t\in {\mathcal S}(M)$, with
$\gamma(t)=\nu_t$, and if
we do not even care about any details,
the abbreviating notation $\gamma\subset{\mathcal S}(M)$ will be
in use. If $\nu$ is a state of special interest the curve is passing through, then without loss of
generality for local considerations in the vicinity of $\nu$ a defining parameterization with $\gamma(0)=\nu$ and compact  parameter interval $I$ of fixed but arbitrary small length and with inner point $0$ will be considered, and $\gamma$ for short then will be referred to as `parameterized curve passing through $\nu$'. For these and subsequent details on general settings from curve theory on metric spaces, see
e.g.~\cite[\S\S 11,12,\,especially \S 12,\,5.-7.]{Rino:61}.
\subsubsection{Local dilation function of a curve}\label{dil0}
According to curve theory, and since $d_B$ is topologically equivalent with
$d_1$ on the state space, each curve $\gamma$ also is
$d_B$-continuous and, since $I$ is compact, as a metric space will be complete. Hence, the local dilation function $I\ni t\,\longmapsto\,
{\mathop{\mathrm{dil}}}_t^B \gamma\in {\overline{{\mathbb{R}}}}_+$ of a
curve $\gamma:I\ni t\,\longmapsto\,\gamma(t)\in {\mathcal S}(M)$ with respect to the Bures distance
reads as
\begin{equation}\label{3b.1}
{\mathop{\mathrm{dil}}}_t^B \gamma=\limsup_{|s-u| \downarrow 0}
\biggl\{ \frac{d_B(M|\nu_s,\nu_u)}{|s-u|}:
u\leq t\leq s, u\not = s,\, u,s\in I\biggr\}.
\end{equation}
Also, since for each inner point $t\in I$, and
$\alpha,\beta>0$ with $u=t-\alpha,s=t+\beta\in I$, by
the triangle inequality $d_B(M|\nu_s,\nu_u)\leq d_B(M|\nu_s,\nu_t)+d_B(M|\nu_t,\nu_u)$ has to be fulfilled, dividing this inequality by $|s-u|=\alpha+\beta$ will yield that
$$
\frac{d_B(M|\nu_s,\nu_u)}{|s-u|}\leq \lambda\, \frac{d_B(M|\nu_s,\nu_t)}{|s-t|}+(1-\lambda)\, \frac{d_B(M|\nu_t,\nu_u)}{|t-u|}
$$
with $\lambda=\frac{|s-t|}{|s-u|}=\frac{\beta}{(\alpha+\beta)}\in \bigl(0,1\bigr)$. Hence, for $u<t<s$
\[
\frac{d_B(M|\nu_s,\nu_u)}{|s-u|} \leq \max\left\{\frac{d_B(M|\nu_s,\nu_t)}{|s-t|}, \frac{d_B(M|\nu_t,\nu_u)}{|t-u|}\right\}
\]
From this in view of
\eqref{3b.1}
\begin{equation}\label{3b.2}
{\mathop{\mathrm{dil}}}_t^B \gamma=\limsup_{\delta\to 0}
\frac{d_B(M|\nu_{t+\delta},\nu_t)}{|\delta|}\,,
\end{equation}
with the admissible increments $\delta\in {{\mathbb{R}}}\backslash \{0\}$
obeying $t+\delta\in I$. Note that since the functions $I\ni t\,\longmapsto\,
|\delta|^{-1}d_B(M|\nu_{t+\delta},\nu_t)$ are continuous at each $t\in I$,
for all $\delta\not=0$ with $t+\delta\in I$, according to \eqref{3b.2} the
local dilation function at each point is a $\limsup$ of continuous functions.
Hence, it is a Lebesgue-measurable function.

As for the meaning of this function, suppose $I\ni t\,\mapsto\,\gamma(t)\in {\mathcal S}(M)$, with compact interval $I$, defines a curve $\gamma$. Then, the
following is the most important special case of a result which
usually is referred to as Rinow's theorem:
\begin{lemma}\label{rinow0}
Assume \eqref{3b.2} as a point function is bounded, i.e.~
$\max_{t\in I} {\mathop{\mathrm{dil}}}_t^B \gamma<\infty$.
Then, $\gamma$ is Bures-rectifiable with Bures-length $\varTheta[\gamma]$ given by
\begin{equation}\label{rinow}
\varTheta[\gamma]=\int_I {\mathop{\mathrm{dil}}}_t^B \gamma\ d\/t\,.
\end{equation}
\end{lemma}
What will be done is to study systematically the structure of the local dilation function at a state $\nu$ under certain plausible assumptions about the way a curve $\gamma$ might be passing through this state, with $\gamma(0)=\nu$, that is, the expression
\begin{equation}\label{3b.3}
{\mathop{\mathrm{dil}}}_0^B \gamma={{\mathop{\mathrm{dil}}}_t^B \gamma}\big|_{t=0}=\limsup_{t\to 0}
\frac{d_B(M|\nu_{t},\nu)}{|t|}\,
\end{equation}
will be under inspection. Start with a fact on  the similarity of parameterized curves.
\begin{corolla}\label{nuequiv}
Let $\gamma:I\ni t\,\longmapsto\,\nu_t\in {\mathcal S}(M)$ and $\tilde{\gamma}:\tilde{I}\ni t\,\longmapsto\,\tilde{\nu}_t\in {\mathcal S}(M)$ be two parameterized curves $\gamma$ and $\tilde{\gamma}$ which both are passing through $\nu$ at $t=0$. Suppose
\begin{equation}\label{nuequiv0}
 \limsup_{t\to 0}
\frac{d_B(M|\nu_{t},\tilde{\nu}_{t})}{|t|}\leq \varepsilon<\infty
\end{equation}
Then, either ${\mathop{\mathrm{dil}}}_0^B \gamma={\mathop{\mathrm{dil}}}_0^B \tilde{\gamma}=\infty$ or both are finite and obey
$
|{\mathop{\mathrm{dil}}}_0^B \gamma-{\mathop{\mathrm{dil}}}_0^B \tilde{\gamma}|\leq \varepsilon
$.
\end{corolla}
\begin{proof}
Since ${\mathcal S}(M)\times{\mathcal S}(M)\ni\{\varrho,\omega\}\mapsto d_B(M|\varrho,\omega)$ is a distance, by the triangle inequality and by symmetry, for each $t\in I\cap \tilde{I}$,
$d_B(M|\nu_{t},\nu) \leq  d_B(M|\nu_{t},\tilde{\nu}_{t})+d_B(M|\tilde{\nu}_{t},\nu)$ and
$d_B(M|\tilde{\nu}_{t},\nu) \leq  d_B(M|\nu_{t},\tilde{\nu}_{t})+d_B(M|\nu_{t},\nu)$
hold. Hence
\begin{equation}\label{gleich1}
 \biggl|\frac{d_B(M|\nu_{t},\nu)}{|t|}-\frac{d_B(M|\tilde{\nu}_{t},\nu)}{|t|}\biggr|\leq \frac{d_B(M|\nu_{t},\tilde{\nu}_{t})}{|t|}
\end{equation}
Suppose $\lim_{n\to\infty} t_n=0$ with ${\mathop{\mathrm{dil}}}_0^B \gamma=\lim_{n\to \infty} d_B(M|\nu_{t_n},\nu)/|t_n|$ and such that
$$\lim_{n\to \infty} \frac{d_B(M|\nu_{t_n},\tilde{\nu}_{t_n})}{|t_n|}$$ exists. By \eqref{nuequiv0} the latter limit must be finite. Thus, in case if ${\mathop{\mathrm{dil}}}_0^B \gamma=\infty$ from   \eqref{nuequiv0}  and \eqref{gleich1}, $\lim_{n\to \infty} d_B(M|\tilde{\nu}_{t_n},\nu)/|t_n|=\infty$ follows. In view of \eqref{3b.3} then ${\mathop{\mathrm{dil}}}_0^B \tilde{\gamma}=\infty$ follows. On the other hand, in case of ${\mathop{\mathrm{dil}}}_0^B \gamma<\infty$ the conclusion by \eqref{gleich1} is that the sequence $d_B(M|\tilde{\nu}_{t_n},\nu)/|t_n|$ is bounded and therefore a converging subsequence exists, say with arguments  $s_k=t_{n_k}$. Application of \eqref{gleich1} with $t=s_k$ by supposition yields
$$
{\mathop{\mathrm{dil}}}_0^B \gamma-\lim_{k\to\infty} \frac{d_B(M|\tilde{\nu}_{s_n},\nu)}{|s_k|} \leq \lim_{k\to \infty} \frac{d_B(M|\nu_{s_k},\tilde{\nu}_{s_k})}{|s_k|}\leq \varepsilon
$$
By $\lim_{k\to\infty} d_B(M|\tilde{\nu}_{s_k},\nu)/|s_k|\leq {\mathop{\mathrm{dil}}}_0^B \tilde{\gamma}$ then
$
{\mathop{\mathrm{dil}}}_0^B\gamma-{\mathop{\mathrm{dil}}}_0^B \tilde{\gamma}\leq \varepsilon$ follows. Now, by reversing the r{\^o}les of $\gamma$ and $\tilde{\gamma}$, let $\lim_{n\to\infty} t_n=0$ with ${\mathop{\mathrm{dil}}}_0^B \tilde{\gamma}=\lim_{n\to \infty} d_B(M|\tilde{\nu}_{t_n},\nu)/|t_n|$ and $\lim_{n\to \infty} d_B(M|\nu_{t_n},\tilde{\nu}_{t_n})/|t_n|$ existing. Clearly, ${\mathop{\mathrm{dil}}}_0^B \tilde{\gamma}=\infty$ cannot hold, for otherwise by \eqref{gleich1} and supposition also ${\mathop{\mathrm{dil}}}_0^B\gamma=\infty$ had to be followed, but  which contradicts our case. Hence, ${\mathop{\mathrm{dil}}}_0^B \tilde{\gamma}<\infty$ must hold. Concluding analogously as above, there exists (with $s_k=t_{n_k}$) a subsequence with $\lim_{k\to\infty} d_B(M|\nu_{s_k},\nu)/|s_k|$ existing and
$${\mathop{\mathrm{dil}}}_0^B \tilde{\gamma}-\lim_{k\to\infty} \frac{d_B(M|\nu_{s_k},\nu)}{|s_k|}\leq\lim_{n\to \infty} \frac{d_B(M|\nu_{s_k},\tilde{\nu}_{s_k})}{|s_k|} \leq \varepsilon$$
From this owing to $\lim_{k\to\infty} d_B(M|\nu_{s_k},\nu)/|s_k|\leq {\mathop{\mathrm{dil}}}_0^B \gamma$ then ${\mathop{\mathrm{dil}}}_0^B \tilde{\gamma}-{\mathop{\mathrm{dil}}}_0^B \gamma\leq \varepsilon$ follows.
Thus, together with the above, what we have shown is that ${\mathop{\mathrm{dil}}}_0^B \gamma=\infty$ implies ${\mathop{\mathrm{dil}}}_0^B \tilde{\gamma}=\infty$, and ${\mathop{\mathrm{dil}}}_0^B \gamma<\infty$ implies
${\mathop{\mathrm{dil}}}_0^B \tilde{\gamma}<\infty$ with
$|{\mathop{\mathrm{dil}}}_0^B \tilde{\gamma}-{\mathop{\mathrm{dil}}}_0^B \gamma|\leq \varepsilon$ . By symmetry of the arguments the validity of the  assertion then follows.
\end{proof}

\subsubsection{Curves with differentiable parameterization}\label{ppath}
Start the investigations of \eqref{3b.3} with the case of $\gamma$  passing through $\nu$ and thereby possessing a tangent.
Let $I\ni t\,\mapsto\,\nu_t\in
{\mathcal S}(M)$ be a parameterized curve $\gamma$ passing through a state $\nu$ at  $t=0$. $\gamma$ will be
called `curve with
differentiable parameterization at $\nu$', or  `$C^1$-curve through $\nu$' for short, if
\begin{equation}\label{tangent}
   f=\|\cdot\|_1-\lim_{t\to 0}\frac{\nu_t-\nu}{t}=\frac{d}{d\/t}\,\nu_t\bigg|_{t=0}=\nu_t^{\,\prime}\bigl|_{t=0}
\end{equation}
exists. Then, refer to the limit $f\in M_{\mathsf h}^*$ as tangent form of $\gamma$ at $\nu$.

Note that throughout this paper the notion of a `parameterized curve passing through a state' requires the map to be   injective, at least around the state $\nu$ of interest. An elementary but often useful criterion in order to get verified that injectivity is fulfilled around a state $\nu$  is based on the subsequently described auxiliary technical tool. \begin{lemma}\label{constr}
Suppose $\gamma:I\ni t\,\longmapsto\nu_t\in {\mathcal S}(M)$ to be differentiable within $I$. Let $\nu=\nu_s$, for some inner
point $s\in I$. If the derivation map $I\ni t\longmapsto\nu_t^{\,\prime}$
is continuous at $t=s$ and is obeying $\nu_t^{\,\prime}|_{t=s}=f$, with $f\in M_{\mathsf h}^*\backslash \{{\mathsf 0}\}$, then
there exists some neighborhood $J$ of $s$ within $I$ such that the map $\gamma|J$ is injective and thus then defines a curve passing through $\nu$.
\end{lemma}
\begin{proof}
Let $\{s_k\},
\{t_k\}\subset I$ be sequences, with $s_k<t_k$, for
each $k\in
{{\mathbb{N}}}$, and $\lim_{k\to\infty} s_k=\lim_{k\to\infty} t_k=s$
fulfilled.  Assume $\nu_{s_k}=\nu_{t_k}$ holds, for each
$k\in {{\mathbb{N}}}$. By assumption, and since $\gamma$ maps into
hermitian linear forms over $M$ and the taking of the derivative
respects linearity, for each $x\in
M_h$ the real-valued function $I\ni t\longmapsto\nu_t(x)$ is obeying
$\nu_{s_k}(x)=\nu_{t_k}(x)$
for all $k\in
{{\mathbb{N}}}$, and is
differentiable, with continuous derivative at $t=s$, which obeys
$$\bigl(\nu_t(x)\bigr)^{\,\prime}=\nu_t^
{\,\prime}(x)$$ Let $x\in M_h$ be fixed. In line with
this, by the mean
value theorem of differential calculus it follows that, for each $k$, there is
$u_k(x)\in [s_k,t_k]$ obeying
$$\nu_{u_k(x)}^{\,\prime}(x)=0$$
Also, note that by assumption
on the sequences $\{s_k\},
\{t_k\}\subset I$ one has
$\bigl\{u_k(x)\}\subset I$, $\lim_{k\to\infty} u_k(x)=s$.
Continuity of $\nu_t^{\,\prime}(x)$\, at $t=s$ in view of the
previous then yields $$\nu_t^{\,\prime}(x)|_{t=s}=f(x)=0$$
Since the latter
conclusion works for each $x\in M_{\mathsf{h}}$, in view of
$f\in M_{\mathsf{h}}^*$ this violates the supposition $f\not={\mathsf 0}$ about $f$.
Thus, in summarizing, $f\not={\mathsf 0}$ implies that there cannot exist  sequences $\{s_k\},
\{t_k\}\subset I$ with the properties mentioned in context
with the continuously differentiable at $t=s$ map $\gamma$. But then, there has to exist some  neighborhood $J$ of $s$ within $I$ such that the restriction of $\gamma$ to $J$, $\gamma|J$ is injective. Since $t=s$ is an inner point of $I$,
this means that by $\gamma|J$ a parameterized curve passing through $\nu$ is given.
\end{proof}
\begin{remark}\label{f0}
By Lemma \ref{constr} the condition $f\not={\mathsf 0}$ in case of continuous differentiability of $\gamma$ at $\nu$ is sufficient to guarantee injectivity of $\gamma$ around the state in question, but it is not neccessary at all. Thus, its main use will be for constructing examples.
\end{remark}
\begin{definition}\label{e.1}
Let an extended positive, subadditive and (real) absolutely homogenous function $M_{\mathsf h}^*\ni f\longmapsto\|f\|_\nu\in \overline{{\mathbb R}}_+$ be defined by
\begin{equation}\label{tangentnorm}
\|f\|_\nu=\sup_{\{x\}}{\frac{1}{2}}\,\sqrt{{\sum_{j}}^\prime
\frac{f(x_j)^2}{\nu(x_j)\phantom{^2}}}
\end{equation}
with $\{x\}$ running through the set of all finite positive decompositions of the
unity within $M$. The $^\prime$ indicates that the summation has to be extended
only about those terms with non-vanishing denominator $\nu(x_j)$.
\end{definition}
In case if $f$ is a tangent form at $\nu$, $\|f\|_\nu$ provides a lower bound for \eqref{3b.3}.
\begin{lemma}\label{lowerbound}
Suppose $I\ni t\,\mapsto\,\nu_t\in
{\mathcal S}(M)$ with $\nu=\nu_0$ and $f=\frac{d}{d\/t}\,\nu_t\bigl|_{t=0}$. Then
\begin{equation}\label{lbound}
{\mathop{\mathrm{dil}}}_0^B \gamma\geq \liminf_{t\to 0}
\frac{d_B(M|\nu_{t},\nu)}{|t|}  \geq   \|f\|_\nu,
\end{equation}
\end{lemma}
\begin{proof} Let $\{x\}=\left\{x_1,x_2,\ldots,x_n\right\}$, $n\in {{\mathbb{N}}}$, be a finite positive decomposition of the unity of $M$. That is, $x_j\in M_+$ for all $j\leq n$ and $\sum_j x_j={\mathsf 1}$. Let $t\in I$, $t\not=0$, be fixed but arbitrarily chosen. For each $j\in \{1,2,\ldots,n\}$, let us define $\xi_j=\nu_t(x_j)$ and $\eta_j=\nu(x_j)$. Then, $\xi_j\geq 0$ and $\eta_j\geq 0$ for each subscript $j$, and $\sum_j \xi_j=1$, $\sum_j \eta_j=1$. Let $J=\{j:\,\eta_j\not=0\}$. Then $J\not=\emptyset$. According to \eqref{bas3cc} the estimate
$${F(M|\nu_t,\nu)}\leq \sum_k \sqrt{\xi_k\eta_k}$$ has to be fulfilled. Thus, by formula \eqref{pcont.1}
\begin{equation*}
d_B(M|\nu_t,\nu)^2\geq 2\Bigl(1- \sum_k \sqrt{\xi_k\eta_k}\Bigr)
\end{equation*}
Now, by rearranging the auxiliary estimate \eqref{4.1}, which for probability vectors $\vec{\xi}$ and $\vec{\eta}$  according to Lemma \ref{vektoren} has to be fulfilled, the following estimate of the expression on the right hand side can be inferred to hold:
\begin{equation*}
 2\Bigl(1- \sum_k \sqrt{\xi_j\eta_j}\Bigr)\geq  {\frac{1}{4}}\,\sum_{j\in J}\frac{\left(\xi_j-\eta_j\right)^2}{\max\left\{
\xi_j,\eta_j\right\}}
\end{equation*}
Combining both estimates provides
\[
d_B(M|\nu_t,\nu)^2\geq {\frac{1}{4}}\,\sum_{j\in J}\frac{\left(\xi_j-\eta_j\right)^2}{\max\left\{
\xi_j,\eta_j\right\}}
\]
Re-substituting $\xi_j=\nu_t(x_j)$ and $\eta_j=\nu(x_j)$ and dividing by $t^2$ yields the estimate
\begin{equation}\label{pcont.1a}
\biggl(\frac{d_B(M|\nu_t,\nu)}{|t|}\biggr)^2\geq {\frac{1}{4}}\,\sum_{j\in J}\frac{\left(\frac{(\nu_t-\nu)}{t}(x_j)\right)^2}{\max\left\{
\nu_t(x_j),\nu(x_j)\right\}}
\end{equation}
for any $t\in I$, with $t\not=0$. Since by assumptions $\lim_{t\to 0} \nu_t(x_j)=\nu(x_j)$ and $ f(x_j)=\lim_{t\to 0}\frac{\nu_t-\nu}{t}(x_j)$, the expression on the right hand side possesses
a limit
\[
\lim_{t\to 0} {\frac{1}{4}}\,\sum_{j\in J}\frac{\left(\frac{(\nu_t-\nu)}{t}(x_j)\right)^2}{\max\left\{
\nu_t(x_j),\nu(x_j)\right\}}  ={\frac{1}{4}}\,\sum_{j\in J}\frac{f(x_j)^2}{\nu(x_j)\phantom{^2}} ={\frac{1}{4}}\,{\sum_{j}}^\prime
\frac{f(x_j)^2}{\nu(x_j)\phantom{^2}}
\]
From this by taking square roots in \eqref{pcont.1a} \[ \liminf_{t\to 0}
\frac{d_B(M|\nu_{t},\nu)}{|t|}\geq {\frac{1}{2}}\,\sqrt{{\sum_{j}}^\prime
\frac{f(x_j)^2}{\nu(x_j)\phantom{^2}}}
\]
follows, which inequality according to the above has to be fulfilled for any  finite positive decomposition $\{x\}=\left\{x_1,x_2,\ldots,x_n\right\}$ of the unity of $M$. From this in view of \eqref{tangentnorm} the asserted estimate \eqref{lowerbound} is obtained.
\end{proof}

\subsubsection{Locally implementable curves}\label{ipath}
Let $I\ni t\,\mapsto\,\nu_t\in
{\mathcal S}(M)$ be a parameterized curve $\gamma$ passing through a state $\nu$, at parameter value $t=0$. Suppose the parameterization is such that, for some unital $^*$-representation $\{\pi,{\mathcal H}_\pi\}$ and open interval $I_\pi\subset I$ around $0$, the $\pi$-fibres of $\nu_t$ for all $t\in I_\pi$ are nontrivial and therefore an implementing family $I_\pi\ni t\mapsto \varphi_t\in {\mathcal H}_\pi$ with $\varphi_t\in {\mathcal S}_{\pi,M}(\nu_t)$ exists. In such situation $\gamma$ is said to be `locally implementable at $\nu$'.  In addition,
$\gamma$ will be said to admit a `differentiable local implementation around $\nu$', or equivalently is said to be a `$C^1$-implementable curve around $\nu$', if the implementing family  $I_\pi\ni t\mapsto \varphi_t\in {\mathcal H}_\pi$ can be chosen such that the derivative
\begin{subequations}\label{ifam}
\begin{equation}\label{locabl}
  \psi= \frac{d}{d\/t}\,\varphi_t\Bigl|_{t=0}=\|\cdot\|_\pi-\lim_{t\to 0} \frac{\varphi_t-\varphi_0}{t}
\end{equation}
exists. This latter is equivalent with requiring existence of a local implementation  $I_\pi\ni t\mapsto \varphi_t\in {\mathcal S}_{\pi,M}(\nu_t)$ and a vector $\psi\in {\mathcal H}_\pi$ such that, with $\varphi=\varphi_0$,
\begin{equation}\label{ifamily}
    \varphi_t=\varphi+t\,\psi+{\mathbf o}(t)
\end{equation}
\end{subequations}
holds, for all $t\in I_\pi$, with the Landau symbol ${\mathbf o}(t)\in {\mathcal H}_\pi$ obeying $$\displaystyle\lim_{t\to 0}\, \frac{\|{\mathbf o}(t)\|_\pi}{|t|}=0$$
Let $\gamma$ be a parameterized curve passing through $\nu$. A unital $^*$-representation $\{\pi,{\mathcal H}_\pi\}$ with respect to which a $C^1$-implementation $(\varphi_t)\subset {\mathcal H}_\pi $ of $\gamma$ around $\nu$ exists will be said to be   `$\gamma$-compliant around $\nu$'. In this paper, mainly parametrized curves $\gamma$  passing through a given state $\nu$ and admitting a $\gamma$-compliant representation $\{\pi,{\mathcal H}_\pi\}$ around $\nu$  will be considered.

\begin{remark}\label{iimpdi}
\begin{enumerate}
\item\label{iimpdi0}
In case of existence the implementing family $I_\pi\ni t\mapsto \varphi_t\in {\mathcal H}_\pi$ with $\varphi_t\in {\mathcal S}_{\pi,M}(\nu_t)$ is highly non-uniquely determined and, although the parameterization $I\ni t\,\mapsto\,\nu_t\in
{\mathcal S}(M)$ by definition is continuous, this does not necessarily imply  continuity of the implementing family at $t=0$. However, arguments borrowed from the modular theory of $vN$-algebras show that a continuous implementation can always be chosen (see Lemma \ref{wurzeleins} in case of $M={\mathsf B}({\mathcal H})$ with separable  ${\mathcal H}$).
\item\label{iimpdi1} Suppose $\gamma$ admits a $C^1$-implementation around $\nu$. Then, differentiability of the parameterization $I\ni t\,\mapsto\,\nu_t$ at $\nu$ is implied. In fact, as a consequence of \eqref{locabl}  then $f=\frac{d}{d\/t}\,\nu_t|_{t=0}\in M_{\mathsf h}^*$ exists and at $x\in M$ reads
      \begin{equation}\label{iimpdi1a}
        f(x)=\langle \pi(x)\psi,\varphi\rangle_\pi+ \langle \pi(x)\varphi,\psi\rangle_\pi
      \end{equation}

  \item\label{iimpdi2} Relating \eqref{iimpdi1a}, it is easily inferred that $\psi$ at each place where it occurs can be replaced with $ \hat{\psi}\in  [\pi(M)^{\,\prime\prime}\varphi]$ given by $ \hat{\psi}=p_\pi^{\,\prime}(\varphi)\psi$. That is, at each $x\in M$  equivalently one also finds that
      \begin{equation}\label{iimpdi1b}
        f(x)=\langle \pi(x)\hat{\psi},\varphi\rangle_\pi+ \langle \pi(x)\varphi,\hat{\psi}\rangle_\pi
      \end{equation}
      is fulfilled, with the same linear form $f$ as given by formula \eqref{iimpdi1a}.
       \item\label{iimpdi3} Remark that if $$\langle \xi,\chi\rangle_{\pi,{\mathbb R}}=\Re\langle \xi,\chi\rangle_\pi=\frac{1}{2}\,\bigl( \langle \xi,\chi\rangle_\pi +\langle\chi,\xi\rangle_\pi\bigr)$$ is considered for $\xi$ and $\chi$ varying through all of ${\mathcal H}_\pi$, then by $\langle \cdot,\cdot\rangle_{\pi,\mathbb R}$ a real scalar product on ${\mathcal H}_\pi$ is given provided the latter is viewed as a real linear space. In the following, whenever the $\langle \cdot,\cdot\rangle_{\pi,\mathbb R}$ corresponding real Euclidean structure on ${\mathcal H}_\pi$, or on some other closed complex linear subspace ${\mathcal K}\subset {\mathcal H}_\pi$, will be concerned with, then this change of view will be notified by attaching a subscript ${\mathbb R}$ to the symbols ${\mathcal H}_\pi$ or ${\mathcal K}$, respectively, and $({\mathcal H}_\pi)_{\mathbb R}$ or ${\mathcal K}_{\mathbb R}$ will be used instead.
        \item\label{iimpdi3a}
        For the tangent form $f$ of \eqref{iimpdi1a} obviously $f({\mathsf 1})=0$ has to hold. Hence, in line with \eqref{iimpdi3}, the tangent vector $\psi$ figuring in \eqref{ifam} satisfies the relation
       \begin{subequations}\label{iimpd1aa}
       \begin{equation}\label{iimpd3b}
        \langle \psi,\varphi\rangle_{\pi,\mathbb R}=0
        \end{equation}
That is, $\psi\in [\varphi]_{\mathbb R}^\perp$, where $[\varphi]_{\mathbb R}^\perp$ is the $\langle \cdot,\cdot\rangle_{\pi,\mathbb R}$-orthogonal complement of the real (one dimensional) subspace generated by $\varphi$ within
$({\mathcal H}_\pi)_{\mathbb R}$. Hence, according to \eqref{iimpdi1b} and with $N=\pi(M)^{\,\prime\prime}$, one then also has
\begin{equation}\label{iimpdi4}
\hat{\psi}\in [{N\varphi}_{\,\mathbb R}]\cap [\varphi]_{\mathbb R}^\perp
\end{equation}
       \end{subequations}
for $\hat{\psi}=p_\pi^{\,\prime}(\varphi)\psi$ occurring in \eqref{iimpdi1b}
\end{enumerate}
\end{remark}
Note that the closure $[N_{\mathsf h}\varphi]$ of the subset $N_{\mathsf h}\varphi =\{x\varphi:\,x^*=x,\,x\in N\}$ of $[N\varphi]$ in a natural manner is a real linear subspace of $[{N\varphi}_{\,\mathbb R}]$, $[N_{\mathsf h}\varphi]\subset [{N\varphi}_{\,\mathbb R}]$. For the orthogonal complement
$[{N\varphi}_{\,\mathbb R}]\ominus [N_{\mathsf h}\varphi]$ of this real linear subspace with respect to  $[{N\varphi}_{\,\mathbb R}]$ one has the following result (${\mathsf{i}}$ stands for the imaginary unit):
\begin{lemma}\label{ocomp}
\hfill{}$[{N\varphi}_{\,\mathbb R}]\ominus [N_{\mathsf h}\varphi]= {\mathsf{i}}[N^{\,\prime}_{\mathsf h}\varphi] \cap [{N\varphi}_{\,\mathbb R}]={\mathsf{i}}[p_\pi^{\,\prime}(\varphi) N^{\,\prime}_{\mathsf h}\varphi]$\hfill{}\phantom{}
\end{lemma}
\begin{proof} For $\varphi,\chi\in {\mathcal H}_\pi$ let a linear form $f_{\chi,\varphi}$ over the $vN$-algebra $N$ be defined by
$$f_{\chi,\varphi}(x)=\langle x \chi,\varphi\rangle_\pi$$ for all $x\in N$. Suppose $\xi\in [N\varphi]$. The condition $\langle \chi,\xi\rangle_{\pi,\mathbb R}=0$ for all $\chi\in [N_{\mathsf h}\varphi]$ then will be shown to be equivalent to the hermiticity of the linear form $f_{\,{\mathsf{i}}\xi,\varphi}$ on $N$.
In fact, by supposition on $\xi$ the condition is equivalent with $\langle x\varphi,\xi\rangle_{\pi,\mathbb R}=0$, for each $x\in N$ with $x^*=x$. Hence, by definition of $\langle \cdot,\cdot\rangle_{\pi,\mathbb R}$ from this equivalence with
$\langle x\varphi,\xi\rangle_\pi=-\langle \xi,x\varphi\rangle_\pi=-\langle x\xi,\varphi\rangle_\pi $ can be followed, for any hermitian element $x\in N_{\mathsf h}$. Thus, $\langle x\varphi,\xi\rangle_{\pi,\mathbb R}=0$, for all $x\in N_{\mathsf h}$, is the same as  $\langle x\varphi,\xi\rangle_\pi=-\langle x\xi,\varphi\rangle_\pi $ for hermitian $x$. From this owing to ${\mathsf{i}}\langle x\varphi,\xi\rangle_\pi=-\langle x\varphi,{\mathsf{i}}\xi\rangle_\pi$ and ${\mathsf{i}}\langle x\xi,\varphi\rangle_\pi=\langle x\,{\mathsf{i}}\xi,\varphi\rangle_\pi$ obviously $\langle x\varphi,{\mathsf{i}}\xi\rangle_\pi=\langle x\,{\mathsf{i}}\xi,\varphi\rangle_\pi$ follows, for hermitian $x\in N$, and vice versa. Thus, the condition $\langle \chi,\xi\rangle_{\pi,\mathbb R}=0$ for all $\chi\in [N_{\mathsf h}\varphi]$ holds if, and only if,
$f_{\varphi,{\mathsf{i}}\xi}(x)=f_{\,{\mathsf{i}}\xi,\varphi}(x)$ holds, for all hermitian  $x\in N$. Clearly, since each $x\in N$ uniquely decomposes as $x=a+{\mathsf{i}}b$, with hermitian $a,b\in N$,  if $f_{\varphi,{\mathsf{i}}\xi}$ and  $f_{\,{\mathsf{i}}\xi,\varphi}$ equal to each other on hermitian elements, this by complex linearity of the scalar product extends to all of $N$. Thus $\langle \chi,\xi\rangle_{\pi,\mathbb R}=0$ for all $\chi\in [N_{\mathsf h}\varphi]$ holds if, and only if $f_{\varphi,{\mathsf{i}}\xi}=f_{\,{\mathsf{i}}\xi,\varphi}$, which condition is equivalent to hermiticity of $f_{\,{\mathsf{i}}\xi,\varphi}$.
Since according to the supposition
${\mathsf{i}}\xi\in[N\varphi]$ is fulfilled, ${\mathsf{i}}\xi$ is invariant under the action of the projection operator $p_\pi^{\,\prime}(\varphi)$, $p_\pi^{\,\prime}(\varphi)\,{\mathsf{i}}\xi={\mathsf{i}}\xi$. But according to Corollary \ref{hermform} hermiticity of $f_{\,{\mathsf{i}}\xi,\varphi}$ on the $vN$-algebra $N$ holds if, and only if, ${\mathsf{i}}\xi\in [N^{\,\prime}_{\mathsf h}\varphi]$, or equivalently $\xi\in{\mathsf{i}} [N^{\,\prime}_{\mathsf h}\varphi]$ holds.  Thus ${\mathsf{i}}[N^{\,\prime}_{\mathsf h}\varphi]\cap [{N\varphi}_{\mathbb R}]$ is the orthogonal complement of $[N_{\mathsf h}\varphi]$ within $ [{N\varphi}_{\mathbb R}]$. Since
$p_\pi^{\,\prime}(\varphi){\mathcal H}_\pi=[N\varphi]$ holds, one finally has ${\mathsf{i}}[p_\pi^{\,\prime}(\varphi) N^{\,\prime}_{\mathsf h}\varphi]=p_\pi^{\,\prime}(\varphi) {\mathsf{i}}[N^{\,\prime}_{\mathsf h}\varphi]={\mathsf{i}}[N^{\,\prime}_{\mathsf h}\varphi]\cap [{N\varphi}_{\mathbb R}]$.
\end{proof}

By Lemma \ref{ocomp},   $\psi\in {\mathcal H}_\pi$ can be decomposed as follows (cf. Remark \ref{iimpdi} \eqref{iimpdi2}):
\begin{equation}\label{decompo}
  \psi =p_\pi^{\,\prime}(\varphi)\psi+ p_\pi^{\,\prime}(\varphi)^\perp\psi
   = \hat{\psi} + p_\pi^{\,\prime}(\varphi)^\perp\psi
   =\hat{\psi}_0+ \hat{\psi}_1  + p_\pi^{\,\prime}(\varphi)^\perp\psi\,,
\end{equation}
where $\hat{\psi}_0$ and $\hat{\psi}_1$ are the uniquely determined components of $\hat{\psi}=p_\pi^{\,\prime}(\varphi)\psi$ in respect of the decomposition of $[{N\varphi}_{\,\mathbb R}]$ into $[N_{\mathsf h}\varphi]$ and its
$\langle\cdot,\cdot\rangle_{\pi,\mathbb R}$-orthogonal complement. That is, one has $\hat{\psi}_0\in [N_{\mathsf h}\varphi]$ and  $\hat{\psi}_1\in {\mathsf{i}}[p_\pi^{\,\prime}(\varphi) N^{\,\prime}_{\mathsf h}\varphi]$.

In addition assume that $\langle \psi,\varphi\rangle_{\pi,\mathbb R}=0$ holds. According to \eqref{iimpdi4} one then has $$\langle \hat{\psi},\varphi\rangle_{\pi,\mathbb R}=0$$ Especially,  $\varphi\in[N_{\mathsf h}\varphi]$ is $\langle\cdot,\cdot\rangle_{\pi,\mathbb R}$-orthogonal to the orthogonal complement  ${\mathsf{i}}[p_\pi^{\,\prime}(\varphi) N^{\,\prime}_{\mathsf h}\varphi]$ of $[N_{\mathsf h}\varphi]$. Hence, $\langle \hat{\psi}_1,\varphi\rangle_{\pi,\mathbb R}=0$. But then $$\langle \hat{\psi}_0,\varphi\rangle_{\pi,\mathbb R}=\langle \hat{\psi}-\hat{\psi}_1,\varphi\rangle_{\pi,\mathbb R}=
\langle \hat{\psi},\varphi\rangle_{\pi,\mathbb R}-\langle \hat{\psi}_1,\varphi\rangle_{\pi,\mathbb R} =0$$ follows.
Note that for each vector $\chi\in [N_{\mathsf h}\varphi]$ obviously $\langle\chi,\varphi\rangle_{\pi,\mathbb R}= \langle\chi,\varphi\rangle_\pi$ holds. In fact, this follows since for each $x\in N$ with $x=x^*$, $\langle x \varphi,\varphi\rangle_\pi = \langle \varphi,x\varphi\rangle_\pi$ is fulfilled. The fact holds on a dense subset, and by continuity extends to all of $ [N_{\mathsf h}\varphi]$. Thus,
\begin{equation}\label{echtorth}
\langle \hat{\psi}_0,\varphi\rangle_\pi=0
\end{equation}
for any $\psi\in {\mathcal H}_\pi$ obeying \eqref{iimpd3b}. Then, the following auxiliary fact holds.
\begin{lemma}\label{auxx}
Let $N=\pi(M)^{\,\prime\prime}$, and $\psi\in {\mathcal H}_\pi$ with $\langle\psi,\varphi\rangle_{\pi,\mathbb R}=0$. Then, either $p_\pi(\varphi)\in N$ is a minimal orthoprojection of $N$, in which case $p_\pi(\varphi)^\perp \hat{\psi}_0=\hat{\psi}_0$ is fulfilled, or one has
\begin{equation}\label{decompoapprox1}
 \hat{\psi}_0=\lim_{n\to\infty} x_n\varphi
\end{equation}
with  a sequence $\{x_n\}\subset N_{\mathsf h}$ obeying $\langle x_n \varphi,\varphi\rangle_\pi=0$ and $p_\pi(\varphi)x_n \varphi\not=0$, for all $ n\in {{\mathbb{N}}}$.
\end{lemma}
\begin{proof}
Let $\{a_n\} \subset N$, with $a_n=a_n^*$ for all $n\in {{\mathbb{N}}}$, be chosen such that $\hat{\psi}_0=\lim_{n\to\infty} a_n\varphi$. Because of $\hat{\psi}_0\in [N_{\mathsf h}\varphi]$ this can always be accomplished. If $p_\pi(\varphi)\in N$ is minimal, then $p_\pi(\varphi)a_n p_\pi(\varphi)=\langle a_n\varphi,\varphi\rangle_\pi p_\pi(\varphi)$ is fulfilled, for each $n\in {{\mathbb{N}}}$. Hence, owing to
 $p_\pi(\varphi)a_n\varphi=p_\pi(\varphi)a_n p_\pi(\varphi)\varphi$ and \eqref{echtorth} one infers that $$p_\pi(\varphi)\hat{\psi}_0=\lim_{n\to\infty} p_\pi(\varphi)a_n\varphi=\lim_{n\to\infty} \langle a_n\varphi,\varphi\rangle_\pi \varphi=\langle \hat{\psi}_0,\varphi\rangle_\pi \varphi=0$$
 On the other hand, if $p_\pi(\varphi)\in N$ is not minimal, there has to exist $b\in N_{\mathsf h}$ such that
 $p_\pi(\varphi)b p_\pi(\varphi)\not=\langle b\varphi,\varphi\rangle_\pi p_\pi(\varphi)$. It is obvious that by choosing a null-sequence $(\varepsilon_n)$ of non-negative reals appropriately, each $a_n$ can be perturbed to $b_n=a_n+\varepsilon_n b$ such that
$p_\pi(\varphi)b_n p_\pi(\varphi)\not=\langle b_n\varphi,\varphi\rangle_\pi p_\pi(\varphi)$ can be accomplished. Note that then  $\hat{\psi}_0=\lim_{n\to\infty} b_n\varphi$ is fulfilled. Defining $x_n=b_n-\langle b_n\varphi,\varphi\rangle_\pi{\mathsf 1}$, for each $n\in {{\mathbb{N}}}$ we then have $\langle x_n\varphi,\varphi\rangle_\pi=0$. Moreover, $p_\pi(\varphi)x_n\varphi\not=0$ must be fulfilled. So assuming ad absurdum that $p_\pi(\varphi)x_n\varphi=0$, this implies   $p_\pi(\varphi)x_np_\pi(\varphi)=0$. However, the latter is equivalent to   $p_\pi(\varphi)b_np_\pi(\varphi)-\langle b_n\varphi,\varphi\rangle_\pi p_\pi(\varphi)=0$, which contradicts the assumptions on $b_n$.
\end{proof}
Note that owing to $ \hat{\psi}_1\in {\mathsf{i}}[p_\pi^{\,\prime}(\varphi) N^{\,\prime}_{\mathsf h}\varphi]$ the vector $ \hat{\psi}_1$ admits an approximating sequence constructed by means of hermitian operators $z_n\in N^{\,\prime}$ as follows:
\begin{equation}\label{decompoapprox2}
 \hat{\psi}_1=\lim_{n\to\infty} {\mathsf{i}}\, z_n\varphi,\text{ with $z_n=z_n^*$, and $z_n=p_\pi^{\,\prime}(\varphi)z_n$ , for all  $n\in {{\mathbb{N}}}$}\,.
\end{equation}
\begin{remark}\label{agg}
\begin{enumerate}
\item\label{agg1}
Note that $ \hat{\psi}_0=\lim_{n\to\infty} x_n\varphi$, with $\{x_n\}\subset N_{\mathsf h}$, $\langle x_n \varphi,\varphi\rangle_\pi=0$ for all  $n\in {{\mathbb{N}}}$, can be accomplished in any case. The crucial point of Lemma \ref{auxx} is that $p_\pi(\varphi)x_n \varphi\not=0$  can be  satisfied, under certain conditions.
\item\label{agg2}
As a trivial consequence of the previous and \eqref{decompoapprox2} we see that, on a $vN$-algebra $N$ acting on a Hilbert space ${\mathcal H}$, for each vector $\varphi\in {\mathcal H}$ with $\langle\psi,\varphi\rangle_{\mathbb R}=0$ the following useful fact is true: the set
of all aggregates of the form $x\varphi+{\mathsf{i}} z\varphi$, with hermitian operators $x\in N_{\mathsf h}$ and $z\in N^{\,\prime}_{\mathsf h}$ such that $\langle x\varphi,\varphi\rangle=0$ and $z=p^{\,\prime}(\varphi)z $ are satisfied, have to be dense within $$[N\varphi]\cap [\varphi]_{\mathbb R}^\perp$$
\item\label{agg3}
 For a  tangent vector $\psi$ as in \eqref{locabl} let us consider the decomposition in accordance with \eqref{decompo}. In view of Lemma \ref{ocomp} and \eqref{decompoapprox2} it follows that \eqref{iimpdi1a} and \eqref{iimpdi1b} can be complemented by the fact that for each $x\in M$
\begin{equation}\label{agg3a}
    f(x)=\langle\pi(x)\hat{\psi}_0,\varphi\rangle_\pi+\langle \pi(x)\varphi,\hat{\psi}_0\rangle_\pi
\end{equation}
is fulfilled, and that the vector $\hat{\psi}_0$ in this situation is the unique vector within $[N_{\mathsf h}\varphi]$ satisfying \eqref{agg3a} for the tangent form $f$ of \eqref{iimpdi1a}.
\end{enumerate}
\end{remark}

\subsubsection{Upper bounds for the local dilation function}\label{uppb}
With the help of the decomposition in \eqref{decompo} we are going to derive upper bounds of the local dilation function of a curve $$\gamma: I\ni t\,\mapsto\,\nu_t\in
{\mathcal S}(M) $$ which is passing through $\nu$ at $t=0$ and which is admitting a $\gamma$-compliant  unital $^*$-representation $\{\pi,{\mathcal H}_\pi\}$ around $\nu$. In line with this, let  $$I_\pi\ni t\mapsto \varphi_t\in {\mathcal S}_{\pi,M}(\nu_t)$$
be a corresponding local implementation of $\gamma$ around $\nu$ and which is differentiable at $t=0$. Then, there is a vector $\psi\in {\mathcal H}_\pi$ such that, with $\varphi=\varphi_0$, the conditions of \eqref{ifam} are fulfilled. Let $N=\pi(M)^{\prime\prime}$. We consider $\beta_\pi(\psi)$ defined by
\begin{equation}\label{beta}
   \beta_\pi(\psi)=\sqrt{\|\hat{\psi}_0\|_\pi^2+\|p_\pi(\varphi)^\perp p_\pi^{\,\prime}(\varphi)^\perp\psi\|_\pi^2}
\end{equation}
and where $\hat{\psi}_0$ is the best $\langle\cdot,\cdot\rangle_{\pi,\mathbb R}$-approximation of $\psi\in ({\mathcal H}_\pi)_{\mathbb R}$ in $[N_{\mathsf h}\varphi]$.
\begin{lemma}\label{upper} In this context, and with $\beta_\pi(\psi)$ from \eqref{beta}, the following holds:
\begin{equation}\label{upperest}
   {\mathop{\mathrm{dil}}}_0^B \gamma\leq \beta_\pi(\psi)\leq   \|\psi\|_\pi
\end{equation}
\end{lemma}
\begin{proof}
For $z\in N_{\mathsf h}^{\,\prime}$ arbitrarily chosen but fixed, and $t\in {\mathbb R}$ let $V_t$ be defined as $V_t=\exp({\mathsf{i}}zt)$. Owing to $V_t\in {\mathcal U}(N^{\,\prime})$ and Theorem \ref{bas} \eqref{bas51} we then can estimate
\begin{equation}\label{upperneu1}
    d_B(M|\nu_{t},\nu)\leq \|\varphi_t-V_t\varphi\|_\pi
\end{equation}
Note that from $V_t\varphi=\varphi+{\mathsf{i}}t z\varphi+{\mathbf{o}}(t)$ and by condition \eqref{ifamily} one gets
\begin{equation*}
   \varphi_t-V_t\varphi=t(\psi-{\mathsf{i}} z\varphi)+{\mathbf{o}}(t)
\end{equation*}
From this $\lim_{t\to 0} \frac{\varphi_t-V_t\varphi}{t}=\psi-{\mathsf{i}} z\varphi$ is seen. Since $z\in N_{\mathsf h}^\prime$ could have been chosen at will, this together with \eqref{upperneu1} and by continuity yields
\begin{equation}\label{dilober}
   {\mathop{\mathrm{dil}}}_0^B \gamma= \limsup_{t\to 0}
\frac{d_B(M|\nu_{t},\nu)}{|t|}\leq \inf_{\xi\in [N_{\mathsf h}^\prime\varphi]} \|\psi-{\mathsf{i}}\xi\|_\pi\leq \|\psi\|_\pi
\end{equation}
We are going to calculate the infimum within \eqref{dilober}. We consider the decomposition
\begin{equation}\label{upperneu2}
\psi =\hat{\psi}_0+ \hat{\psi}_1  + p_\pi^{\,\prime}(\varphi)^\perp\psi
\end{equation}
of $\psi$ according to \eqref{decompo}. By its very definition the component $\hat{\psi}_0$ of $\psi$ then is the best $\langle\cdot,\cdot\rangle_{\pi,\mathbb R}$-approximation of $\psi\in ({\mathcal H}_\pi)_{\mathbb R}$ within $[N_{\mathsf h}\varphi]$, whereas $$\hat{\psi}_1\in {\mathsf{i}}[p_\pi^{\,\prime}(\varphi) N^{\,\prime}_{\mathsf h}\varphi]\subset {\mathsf{i}}[N_{\mathsf h}^\prime\varphi]$$
Occasionally note that with $\xi\in[N_{\mathsf h}^\prime\varphi]$ and in view of the specific of \eqref{upperneu2} one has
\begin{equation*}
    \psi-{\mathsf{i}}\xi=\hat{\psi}_0+ (\hat{\psi}_1 -{\mathsf{i}}p_\pi^{\,\prime}(\varphi)\xi)  + (p_\pi^{\,\prime}(\varphi)^\perp\psi-{\mathsf{i}}p_\pi^{\,\prime}(\varphi)^\perp\xi)
\end{equation*}
with the terms $\hat{\psi}_0$, $(\hat{\psi}_1 -{\mathsf{i}}p_\pi^{\,\prime}(\varphi)\xi)$ and  $ (p_\pi^{\,\prime}(\varphi)^\perp\psi-{\mathsf{i}}p_\pi^{\,\prime}(\varphi)^\perp\xi)$ to be mutually $\langle\cdot,\cdot\rangle_{\pi,\mathbb R}$-orthogonal. Thus one has
\begin{subequations}\label{infi10a}
\begin{equation}\label{infi1}
  \| \psi-{\mathsf{i}}\xi\|_\pi^2=\|\hat{\psi}_0\|_\pi^2+ \|\hat{\psi}_1 -{\mathsf{i}}p_\pi^{\,\prime}(\varphi)\xi\|_\pi^2  + \|p_\pi^{\,\prime}(\varphi)^\perp\psi-{\mathsf{i}}p_\pi^{\,\prime}(\varphi)^\perp\xi\|_\pi^2
  \end{equation}
Moreover, because of ${\mathsf{i}}p_\pi^{\,\prime}(\varphi)\xi\in[N^\prime\varphi]$, one has $p_\pi(\varphi){\mathsf{i}}p_\pi^{\,\prime}(\varphi)\xi={\mathsf{i}}p_\pi^{\,\prime}(\varphi)\xi$, and thus the last term on the right hand side of \eqref{infi1} decomposes as
\begin{multline}\label{infi10}
\|p_\pi^{\,\prime}(\varphi)^\perp\psi-{\mathsf{i}}p_\pi^{\,\prime}(\varphi)^\perp\xi\|_\pi^2 =\|p_\pi(\varphi)p_\pi^{\,\prime}(\varphi)^\perp\psi-{\mathsf{i}}p_\pi^{\,\prime}(\varphi)^\perp\xi\|_\pi^2\\
  + \|p_\pi(\varphi)^\perp p_\pi^{\,\prime}(\varphi)^\perp\psi\|_\pi^2
\end{multline}
\end{subequations}
Note that owing to $p_\pi(\varphi)\psi\in [N^\prime\varphi]$ there exists a sequence $\{c_n\}\subset N^\prime $ such that $p_\pi(\varphi)\psi=\lim_{n\to\infty} c_n\varphi$. From this
\begin{equation}\label{line4}
  p_\pi(\varphi)p_\pi^{\,\prime}(\varphi)^\perp\psi=\lim_{n\to\infty} p_\pi^{\,\prime}(\varphi)^\perp c_n\varphi
\end{equation}
follows. For each $n\in {{\mathbb{N}}}$, let us define elements $y_n\in N^\prime$ by
$$y_n=-{\mathsf{i}}\bigl(p_\pi^{\,\prime}(\varphi)^\perp c_n-c^*_n p_\pi^{\,\prime}(\varphi)^\perp\bigr)$$
Obviously $y_n^*=y_n$ and $y_n\varphi=-{\mathsf{i}}p_\pi^{\,\prime}(\varphi)^\perp c_n\varphi$ are fulfilled.
According to \eqref{line4} we therefore can be assured that
\begin{equation}\label{ypsi}
    p_\pi(\varphi)p_\pi^{\,\prime}(\varphi)^\perp\psi=\lim_{n\to\infty}{\mathsf{i}} y_n\varphi
\end{equation}
can be accomplished, with  a sequence $\{y_n\}\subset N^\prime_{\mathsf h}$. Hence,
$$p_\pi(\varphi)p_\pi^{\,\prime}(\varphi)^\perp\psi\in {\mathsf{i}}p_\pi^{\,\prime}(\varphi)^\perp[N^\prime_{\mathsf h}\varphi]\subset {\mathsf{i}}[N^\prime_{\mathsf h}\varphi]$$ Especially,  if $\hat{\xi}=-{\mathsf{i}}(\hat{\psi}_1+p_\pi(\varphi)p_\pi^{\,\prime}(\varphi)^\perp\psi)$ is chosen, then we have $\hat{\xi}\in [N_{\mathsf h}^\prime\varphi]$ with $$p_\pi^{\,\prime}(\varphi)\hat{\xi}=-{\mathsf{i}}\hat{\psi}_1,\
p_\pi^{\,\prime}(\varphi)^\perp\hat{\xi}=-{\mathsf{i}}p_\pi(\varphi)p_\pi^{\,\prime}(\varphi)^\perp\psi$$
That is, ${\mathsf{i}}p_\pi^{\,\prime}(\varphi)\hat{\xi}=\hat{\psi}_1$ and ${\mathsf{i}}p_\pi^{\,\prime}(\varphi)^\perp\hat{\xi}=p_\pi(\varphi)p_\pi^{\,\prime}(\varphi)^\perp\psi$ are fulfilled. In substituting \eqref{infi10} into \eqref{infi1} and choosing $\xi=\hat{\xi}$ in the resulting equation then yields that
\begin{equation*}
\| \psi-{\mathsf{i}}\hat{\xi}\|_\pi^2=\|\hat{\psi}_0\|_\pi^2+\|p_\pi(\varphi)^\perp p_\pi^{\,\prime}(\varphi)^\perp\psi\|_\pi^2\leq \| \psi-{\mathsf{i}}\xi\|_\pi^2
\end{equation*}
has to be fulfilled, for any $\xi\in[N_{\mathsf h}^\prime\varphi]$. What has been shown in view of \eqref{beta} is
\begin{equation}\label{min}
   \inf_{\xi\in [N_{\mathsf h}^\prime\varphi]} \|\psi-{\mathsf{i}}\xi\|_\pi=\beta_\pi(\psi)
\end{equation}
In view of \eqref{dilober} the latter provides the validity of the estimate \eqref{upperest}.
\end{proof}
Let $\gamma$ be a parameterized curve $\gamma$ passing through $\nu$, and be  $I_\pi\ni t\mapsto \varphi_t\in {\mathcal S}_{\pi,M}(\nu_t)$ a $C^1$-implementation of $\gamma$ around $\nu$. Then,
neither the latter, nor the  $\gamma$-compliant representation $\{\pi,{\mathcal H}_\pi\}$ with respect to which $(\varphi_t)$ is defined, will be uniquely determined. Accordingly, from each such context of an individual implementation around $\nu$ a  derivative $\psi=\varphi_t^{\,\prime}|_{t=0}\in {\mathcal H}_\pi$ arises. The latter object is referred to as `tangent vector at $\nu$' of the implementation in question. All the tangent vectors at $\nu$ for a given curve $\gamma$ that can occur in any such context will be collected into the class ${\mathsf{T}}_\nu(M|\gamma)$ of tangent vectors of $\gamma$ at $\nu$. The notation $$\psi\in{\mathsf{T}}_\nu(M|\gamma)$$ for a given fixed differentiable at $\nu$ curve $\gamma\subset {\mathcal S}(M)$ will mean that a  $\gamma$-compliant $^*$-representation $\{\pi,{\mathcal H}_\pi\}$ of $M$ exists such that $\gamma$ around $\nu$ can be implemented by a differentiable at $t=0$ family $I_\pi\ni t\mapsto \varphi_t\in {\mathcal S}_{\pi,M}(\nu_t)$ with derivative $\psi\in {\mathcal H}_\pi$. If we want to emphasize that context, a notation like $\psi=\psi(\pi,{\mathcal H}_\pi,(\varphi_t))$ will be chosen. Thereby, the class ${\mathsf{T}}_\nu(M|\gamma)$ obeys a remarkable invariance property. Namely,
\begin{equation}\label{invari}
\psi=\psi(\pi,{\mathcal H}_\pi,(\varphi_t))\in{\mathsf{T}}_\nu(M|\gamma)\ \Longrightarrow\ \psi-{\mathsf{i}}\xi\in{\mathsf{T}}_\nu(M|\gamma),\,\forall\xi\in \pi(M)_{\mathsf h}^{\,\prime}\varphi
\end{equation}
In fact, if $\xi=b\,\varphi$ with $b\in  \pi(M)_{\mathsf h}^{\,\prime}$, then $\tilde{\psi}=\psi-{\mathsf{i}}\xi$ obeys
$\tilde{\psi}=\tilde{\psi}(\pi,{\mathcal H}_\pi,(\tilde{\varphi}_t))$, with the implementation  $(\tilde{\varphi}_t)$  of $\gamma$ at $\nu$ given by $\tilde{\varphi}_t=\exp(-{\mathsf{i}}\,b\,t)\varphi_t$ at any $t\in I_\pi$. As a consequence of \eqref{invari} we then even get the following strengthening of  Lemma \ref{upper}.
\begin{corolla}\label{upperbest}
Let $\gamma\subset{\mathcal S}(M) $ be a parameterized curve  passing through $\nu$ and admitting $\gamma$-compliant representations around $\nu$. Then,
\begin{subequations}\label{invari10}
\begin{equation}\label{invari1}
    {\mathop{\mathrm{dil}}}_0^B \gamma\leq\inf_{\phi\in{\mathsf{T}}_\nu(M|\gamma)}\|\phi\|\leq \beta_\pi(\psi)\leq \|\psi\|_\pi
\end{equation}
for any $\psi=\psi(\pi,{\mathcal H}_\pi,(\varphi_t))\in {\mathsf{T}}_\nu(M|\gamma)$. The infimum within \eqref{invari1} is to be understood as shortcut notation of the expression
\begin{equation}\label{invari11}
\inf_{\phi\in{\mathsf{T}}_\nu(M|\gamma)}\|\phi\|:=\inf_{\phi(\pi,{\mathcal H}_\pi,(\varphi_t))\in{\mathsf{T}}_\nu(M|\gamma)}\|\phi(\pi,{\mathcal H}_\pi,(\varphi_t))\|_\pi
\end{equation}
\end{subequations}
Thereby, the occurring $\gamma$-compliant representation $\{\pi,{\mathcal H}_\pi\}$ around $\nu$ is the representation to which the tangent vector $\phi$ is referring while running through the elements of ${\mathsf{T}}_\nu(M|\gamma)$.
\end{corolla}
\begin{proof}
In fact, if $\psi=\psi(\pi,{\mathcal H}_\pi,(\varphi_t))\in{\mathsf{T}}_\nu(M|\gamma)$, then owing to \eqref{invari} and due to continuity, with respect to the $\gamma$-compliant representation $\pi$ we obviously have that
\[
\inf_{\phi\in{\mathsf{T}}_\nu(M|\gamma)}\|\phi\|\leq \|\psi-\mathsf{i}\xi\|_\pi
\]
for each $\xi\in [\pi(M)_{\mathsf h}^{\,\prime}\varphi]$.
From this in view of \eqref{min} with $N=\pi(M)^{\,\prime\prime}$ we infer that
\[
\inf_{\phi\in{\mathsf{T}}_\nu(M|\gamma)}\|\phi\|\leq \inf_{\xi\in [\pi(M)_{\mathsf h}^{\,\prime}\varphi]}\|\psi-\mathsf{i}\xi\|_\pi=\beta_\pi(\psi)
\]
By Lemma \ref{upper},   ${\mathop{\mathrm{dil}}}_0^B\gamma\leq \|\phi\|_{\hat{\pi}}$, for any $\phi=\phi(\hat{\pi},\hat{{\mathcal H}}_\pi,({\hat{\varphi}}_t))\in {\mathsf{T}}_\nu(M|\gamma)$. Hence
\[
{\mathop{\mathrm{dil}}}_0^B\gamma\leq\inf_{\phi\in{\mathsf{T}}_\nu(M|\gamma)}\|\phi\|
\]
follows, see \eqref{invari11}.
The last two estimates together with  \eqref{upperest} give \eqref{invari1}.
\end{proof}
\subsection{Calculating the local dilation function}\label{ipathex0}
Next we are going to construct a class of curves through $\nu$ such that, for some $\psi=\psi(\pi,{\mathcal H}_\pi,(\varphi_t))\in {\mathsf{T}}_\nu(M|\gamma)$,
\begin{equation}\label{gleich}
     {\mathop{\mathrm{dil}}}_0^B \gamma=\inf_{\phi\in{\mathsf{T}}_\nu(M|\gamma)}\|\phi\|=\|\psi\|_\pi
\end{equation}
will be satisfied. That is, in \eqref{upperest} and \eqref{invari1} then equality occurs.
\subsubsection{The basic construct}\label{bascon}
To start with, let $\nu,\varrho\in {\mathcal S}(M)$ be fixed, with $\varrho\not=\nu$, and which is equivalent to $0\leq F(M|\varrho,\nu)< 1$.
Let $P=P(M|\varrho,\nu)$, $F=F(M|\varrho,\nu)$, and be $\{\pi,{\mathcal H}_\pi\}$ a unital $^*$-representation of $M$ chosen such that the $\pi$-fibres of both states exist. Let $N=\pi(M)^{\prime\prime}$. By Theorem \ref{positiv1} there are $\varphi\in {\mathcal S}_{\pi,M}(\nu)$ and
$\zeta\in {\mathcal S}_{\pi,M}(\varrho)$ with $d_B(M|\nu,\varrho)=\|\zeta-\varphi\|_\pi$ or equivalently, with $\langle \zeta,\varphi\rangle_\pi={F}$. By Theorem \ref{bas5} \eqref{bas53} this is equivalent with the positivity of $h^\pi_{\zeta,\varphi}$ over $N^\prime$. For $t\in {\mathbb R}$ with $|t|\leq 1$ consider vectors
\begin{subequations}\label{geo0}
\begin{equation}\label{geo0a}
    \varphi_t=t\,\zeta+\lambda(t)\,\varphi\text{, with }\lambda(t)=-t\,{F}+\sqrt{1-t^2(1-P)}
\end{equation}
By the aforementioned facts and \eqref{pcont.1aa} it is easily checked that the vectors $\varphi_t$ are of unit length. Let $\nu_t$ be the state implemented by
$\varphi_t$ over $M$ with the help of $\pi$, $\varphi_t\in {\mathcal S}_{\pi,M}(\nu_t)$. In the sequel let $\Gamma(M|\nu,\varrho)$ be the set of all differentiable curves $\gamma$
\begin{equation}\label{geo0b}
   \gamma:\ \bigl[-1,1\bigr]\ni t\longmapsto \nu_t
\end{equation}
\end{subequations}
arising along such a construction \eqref{geo0a} for given $\nu$ and $\varrho$, unital $^*$-representation $\{\pi,{\mathcal H}_\pi\}$ of $M$  and representing vectors $\varphi$  and $\zeta$, with the latter being chosen as to obey  $d_B(M|\nu,\varrho)=\|\zeta-\varphi\|_\pi$ with respect to $\pi$.  Obviously, $\varphi_t=\varphi+t\psi+{\mathbf{o}}(t)$ is fulfilled, with the derivative of $\varphi_t$ at $t=0$ given by
\begin{subequations}\label{geo10}
\begin{equation}\label{geo1}
    \psi=\zeta-{F}\,\varphi\in {\mathsf{T}}_\nu(M|\gamma)
\end{equation}
With this tangent vector $\psi$, the map $\gamma$ of \eqref{geo0b} at $\nu$ is seen to have tangent form
\begin{equation}\label{geo1b}
f(\cdot)=\nu_t^{\,\prime}|_{t=0}=\langle\pi(\cdot)\psi,\varphi\rangle_\pi+\langle\pi(\cdot)\varphi,\psi\rangle_\pi
\end{equation}
\end{subequations}
With the help of $f$ one easily finds the following explicite formula to hold:
\begin{equation}\label{geo0c}
\nu_t=t^2\varrho+(1-t^2)\nu+t\lambda(t)\,f
\end{equation}
In case of $f\not={\mathsf 0}$, by continuous differentiability of \eqref{geo0c}, Lemma \ref{constr} at $t=0$ can be applied and implies the map $\gamma$ to be injective around $\nu$. Hence, each map $\gamma\in \Gamma(M|\nu,\varrho)$ with $\nu_t^{\,\prime}|_{t=0}\not={\mathsf 0}$ then is a parameterized curve passing through $\nu$ in the sense of the definition, and by construction admits the  $\gamma$-compliant representation $\{\pi,{\mathcal H}_\pi\}$ around $\nu$. Also, if $f\not={\mathsf 0}$ is not a multiple of $(\varrho-\nu)$, then $\gamma$ will be injective as a whole. In fact,  assuming $\nu_t=\nu_s$, for $t\not=s$ with $s,t\in [-1,1]$, this according to \eqref{geo0c} is equivalent to
\[ {\mathsf 0}=(t^2-s^2)(\varrho-\nu)+(t\lambda(t)-s\lambda(s))f \]
By assumptions on $f$, from this $t\lambda(t)=s\lambda(s)$ follows and ${\mathsf 0}=(t^2-s^2)(\varrho-\nu)$ has to be fulfilled. Note that since $F<1$ is fulfilled, one has $\varrho\not=\nu$. From this and by the previous then $t^2=s^2$ has to be followed. That is, $s=- t$, with $t\in [-1,1]\backslash\{0\}$. But in view of the definition of $\lambda(t)$ in \eqref{geo0a} this contradicts to the fact that the solution of $t\lambda(t)=s\lambda(s)$ under the condition $s=-t$ with $t\not=0$
would read $$t=\pm \frac{1}{\sqrt{1-P}}$$ and which both definitely were not in $ [-1,1]\backslash\{0\}$. Thus,  $\gamma$ has to be injective as a whole, then. On the other hand, in view of \eqref{geo0c} the case $f={\mathsf 0}$ is occurring if, and only if,
\[
\nu_t=t^2\varrho+(1-t^2)\nu
\]
is fulfilled. Hence, in this case $\gamma$ is degenerated, for then $\nu_t=\nu_{-t}$ holds, and then $\gamma$ cannot be injective around $\nu$, and then $\gamma$ is not a parameterized curve passing through $\nu$ in the sense of the definition. Nevertheless, owing to $\varrho\not=\nu$ (which is a consequence of $F<1$), e.g.\,the restriction  $\gamma|[0,1]$ of $\gamma$ to the unit interval is injective as a whole, and thus is passing through any of its states $\nu_t$ with $t\in\, ]0,1[$, in the sense of the definition. Similar to the case with vanishing $f$, the cases where $f$ proves to be a non-vanishing real multiple of $(\varrho-\nu)$ will lead into situations where degeneracies can occur around a (finite) number of exceptional states of $\gamma$, the curve then is not passing through in the very sense of our definition (we omit the details of the analysis). But interestingly, as will be shown below, even in those cases the property of injectivity of $\gamma|[0,1]$ will persist to be fulfilled.
Before proving this, start with the following auxiliary fact which holds if $\nu_t$ from \eqref{geo0c} is considered in restriction to the unit interval.
\begin{lemma}\label{auxin}
Provided $\varrho\not=\nu$, then for each $t\in\, ]0,1[$ one has $\nu_t\not=\nu$ and $\nu_t\not=\varrho$.
\end{lemma}
\begin{proof}
Note that $\varrho\not=\nu$ implies $F(M|\varrho,\nu)=F<1$. Also note that $[0,1]\ni t\longmapsto\lambda(t)$ is obeying $\lambda(t)^{\,\prime}<0$ and $\lambda(1)=0$. Thus $\lambda(t)$ is non-negative over $[0,1]$.

Assume $\nu_t=\nu$, for some $t\in [0,1]$. Then, since $t\geq 0$ and $\lambda(t)\geq 0$ hold, according to \eqref{geo0a} one infers that over $\pi(M)^{\,\prime\prime}$
\[
h^\pi_{\varphi_t,\varphi}=t\,h^\pi_{\zeta,\varphi}+\lambda(t)\,h^\pi_{\varphi,\varphi}\geq {\mathsf 0}
 \]
has to be fulfilled. Hence, by Lemma \ref{bas3}\,\eqref{bas3aa}, $$1=F(M|\nu_t,\nu)=\|h^\pi_{\varphi_t,\varphi}\|_1=h^\pi_{\varphi_t,\varphi}({\mathsf 1})=t\,F+\lambda(t)$$
By definition of $\lambda(t)$ the latter means $\sqrt{1-t^2(1-P)}=1$. Owing to $0\leq P<1$ the only solution $t$ of the latter equation in $[0,1]$ is $t=0$. Accordingly, for each $t\in\, ]0,1[$,  $\nu_t\not=\nu$.

Analogously, $\nu_t=\varrho$ for some $t\in [0,1]$ implies $F(M|\nu_t,\varrho)=1$. Hence, owing to
\[
h^\pi_{\varphi_t,\zeta}=t\,h^\pi_{\zeta,\zeta}+\lambda(t)\,h^\pi_{\varphi,\zeta}\geq {\mathsf 0}
 \]
and in view of  Lemma \ref{bas3}\,\eqref{bas3aa} the relation
$
1=t+\lambda(t)\,F
$
can be inferred to hold. With help of the relation $F^2=P$ and since $0\leq P<1$ holds, the above mentioned relation easily can be seen to be equivalent to $t^2-2\,t+1=0$. Hence, $t=1$ is the only solution within $[0,1]$. Accordingly, for each $t\in\, ]0,1[$,  $\nu_t\not=\varrho$.
\end{proof}
\begin{corolla}\label{auxin1}
Provided $\varrho\not=\nu$,  then $\gamma|[0,1]$ is injective, for each $\gamma\in \Gamma(M|\nu,\varrho)$.
\end{corolla}
\begin{proof}
The cases not yet covered by our preliminary considerations arise if $f=\mu\,(\varrho-\nu)$ is fulfilled, with $\mu\in {\mathbb R}\backslash\{0\}$. Thus, we put this case. In this case, \eqref{geo0c} is equivalent to
\begin{subequations}\label{geoc1}
\begin{equation}\label{geo0cc}
\nu_t=g(t)\,\varrho+(1-g(t))\,\nu
\end{equation}
with $g(t)=t^2+\mu\, t\lambda(t)$. At the bounds, $g(t)$ and $g^{\,\prime}(t)$ take the following values
\begin{equation}\label{geo0cd}
g(0)=0,\,g(1)=1,\,g^{\,\prime}(0)=\mu,\,g^{\,\prime}(1)=2-\frac{\mu}{F}
\end{equation}
\end{subequations}
Assume $\mu<0$. From \eqref{geo0cd} then $g^{\,\prime}(0)<0$ is seen, and thus $g(t)<0$ for all $t>0$ sufficiently small. Due to $g(1)>0$ there exists $t$ with $0<t<1$ and $g(t)=0$. According to \eqref{geo0cc}, $\nu_t=\nu$ had to be followed, for some $t\not=0$. This contradicts the facts stated by Lemma \ref{auxin}. Hence, $f=\mu\,(\varrho-\nu)$ with $\mu<0$ cannot occur.

Assume the case $\mu>2\,F$. Then, $g(1)=1$ and $g^{\,\prime}(1)<0$. Hence, $g(t)>1$ for all $t<1$ sufficiently close to one. Thus, owing to  $g(0)=0$, $t\in \,]0,1[$ obeying $g(t)=1$ has to exist. Hence, $\nu_t=\varrho$, for some $t\not=1$, by \eqref{geo0cc}. This contradicts   another of the facts stated by Lemma \ref{auxin}. Hence, the case $f=\mu\,(\varrho-\nu)$ with $\mu>2\,F$ cannot occur.

Also note that all $t\in [0,1]$ satisfying
\begin{equation}\label{geo0ce}
t^2=\frac{1}{2(1-P)}\,\biggl(1\pm \frac{|1-\mu\,F|	}{\sqrt{\mu^2-2\,\mu F+1}}\biggr)
\end{equation}
exactly correspond to the solutions of the  equation $g^{\,\prime}(t)=0$ there. Thus, in view of \eqref{geo0ce} and \eqref{geo0cd}, there can exist two critical points $t_\pm$ of $g(t)$ within $]0,1]$ at most, and which correspond to the $\pm$-signs  in formula \eqref{geo0ce}, respectively. An easy calculation is showing that $t_+^2\geq 1$ has to be fulfilled whenever $0< \mu \leq 2F$ is satisfied. Hence, if any, in this case $g(t)$ within $]0,1[$ can possess the one critical point $t=t_-$ at most. Thus, since in this case $$g^{\,\prime}(0)=\mu>0,\, g^{\,\prime}(1)=2-\frac{\mu}{F}\geq 0$$ holds (with `$=0$' occurring only for $\mu=2F$), and one critical point exists at most, the derivative of $g(t)$ within  $]0,1[$ cannot change signs. Thus, $g(t)$ as an increasing continuous function with boundary values $g(0)=0$ and $g(1)=1$ is mapping $[0,1]$ injectively onto itself. Hence, in view of \eqref{geo0cc} and since $\varrho\not=\nu$ holds, also in case of $f=\mu\,(\varrho-\nu)$ with $0< \mu \leq 2F$ the injectivity of $\gamma|[0,1]$ follows. Obviously, the situations with $\mu<0$, $\mu>2 F$ and $0<\mu\leq 2 F$ cover all cases of $\mu\in {\mathbb R}\backslash\{0\}$.
\end{proof}
By Corollary \ref{auxin1}, $\gamma|[0,1]$ for any $\gamma\in \Gamma(M|\nu,\varrho)$ can be considered to be a curve in the sense of our definition and which is passing through any of its states corresponding to parameter values $t\in\,]0,1[$. But
in either case of $\gamma$, the conditions of \eqref{ifam} are fulfilled locally and thus the estimate \eqref{upperest} can be applied formally, with respect to the $^*$-representation $\pi$.
\subsubsection{Processing basic examples}\label{ipathex00}
By the construction schema of \eqref{geo0} together with the results of section \ref{burdist} the main data of $\gamma\in \Gamma(M|\nu,\varrho)$ can be determined.
\begin{example}\label{ex2}
Let $\nu\in {\mathcal S}(M)$ be an arbitrarily given, but fixed state. Let $\varrho \in {\mathcal S}(M)$ be chosen with  $P=P(M|\varrho,\nu)$ obeying $0\leq P<1$. Let $\gamma\in \Gamma(M|\nu,\varrho)$, $\psi\in {\mathsf{T}}_\nu(M|\gamma)$, with
$
\psi=\psi(\pi,(\varphi_t),{\mathcal H}_\pi)
$
and $\pi,(\varphi_t),{\mathcal H}_\pi$ referring to the construction schema \eqref{geo0}.
Then, the estimate \eqref{upperest} can be strengthened as following
\begin{subequations}\label{geo2}
\begin{equation}\label{geo2a}
    {\mathop{\mathrm{dil}}}_0^B \gamma =\limsup_{t\to 0}\frac{d_B(M|\nu_t,\nu)}{|t|}=\|\psi\|_\pi=\sqrt{1-P}
\end{equation}
More generally, for $1\geq s> 0$, one has
\begin{equation}\label{geo2b}
    {\mathop{\mathrm{dil}}}_s^B \gamma =\lim_{t\to s}\phantom{sup}\frac{d_B(M|\nu_t,\nu_s)}{|t-s|}
    =\sqrt{\frac{1-P}{1-s^2(1-P)}}\,,
\end{equation}
\end{subequations}
that is, for $s>0$ the $\limsup$ is even a $\lim$.
\end{example}
\begin{proof} In the special case of $\gamma\in \Gamma(M|\nu,\varrho)$ the parameters $\|\psi\|_\pi$ and $ {\mathop{\mathrm{dil}}}_0^B \gamma$ can be calculated explicitly.
Due to \eqref{geo1} and since $\langle \zeta,\varphi\rangle_\pi={F}$ is fulfilled,
\begin{subequations}\label{geo100}
\begin{equation}\label{geo1a}
    \|\psi\|_\pi=\sqrt{1-P(M|\varrho,\nu)}
\end{equation}
is easily seen. As mentioned above,  over $N^\prime$ the positivity condition $h^\pi_{\zeta,\varphi}\geq 0$ holds. In order to evaluate the dilation function of $\gamma$, remark that
for $t\geq 0$ the function $\lambda(t)$ in formula \eqref{geo0a} can take only non-negative values. Hence, the above mentioned positivity of $h^\pi_{\zeta,\varphi}$ over $N^\prime$ implies $h^\pi_{\varphi_t,\varphi_s}$ to be positive   for $s,t\geq 0$ on $N^\prime$, too, i.e.
\begin{equation}\label{hpos1}
   h^\pi_{\varphi_t,\varphi_s}=ts\,h^\pi_{\zeta,\zeta}+ \bigl(t\lambda(s)+s\lambda(t)\bigr)h^\pi_{\zeta,\varphi}+\lambda(t)\lambda(s)\,h^\pi_{\varphi,\varphi}\geq 0
\end{equation}
Accordingly, by Theorem \ref{bas5}\,\eqref{bas53}, for $s=0$ especially one infers that
\begin{equation*}
    d_B(M|\nu_t,\nu)=\|\varphi_t-\varphi\|_\pi=\|\,t\,\psi+{\mathbf{o}}(t)\|_\pi
\end{equation*}
must be fulfilled, for $t\geq 0$. The conclusion from this is
\begin{equation}\label{geo1c}
    \lim_{t\to 0+} \frac{d_B(M|\nu_t,\nu)}{|t|}=\|\psi\|_\pi
\end{equation}
\end{subequations}
This in view of \eqref{invari1} gives \eqref{geo2a}. Now, let us consider $s$ with $1\geq s>0$. Then, according to  \eqref{hpos1} and Theorem \ref{bas5}\,\eqref{bas53} for all $t$ of some full neighborhood of the considered $s>0$ within $[-1,1]$ one obviously has
\begin{equation}\label{allg1}
\begin{split}
    d_B(M|\nu_t,\nu_s)^2& =\|\varphi_t-\varphi_s\|_\pi^2=\|(t-s)\zeta+\bigl(\lambda(t)-\lambda(s)\bigr)\varphi\|_\pi^2\\
    &=(t-s)^2+2(t-s)\bigl(\lambda(t)-\lambda(s)\bigr){F}+\bigl(\lambda(t)-\lambda(s)\bigr)^2
\end{split}
\end{equation}
From this and \eqref{geo0a} then the following relations are easily inferred to hold
\begin{subequations}\label{allg2}
\begin{eqnarray}\label{allg2a}
    \lim_{t\to s} \frac{d_B(M|\nu_t,\nu_s)}{|t-s|}& =& \sqrt{1+2\lambda^\prime(s){F}+\lambda^\prime(s)^2}\\
  \label{allg2b}
   \lambda^\prime(s)& =& -{F}-\frac{s(1-P)}{\sqrt{1-s^2(1-P)}}
\end{eqnarray}
\end{subequations}
By substituting \eqref{allg2b} into \eqref{allg2a} and suitably processing the expression under the square root in \eqref{allg2a}, after subsequent simplifications and by suitably rearranging the arising terms, this finally will give that
\begin{equation*}
     \lim_{t\to s} \frac{d_B(M|\nu_t,\nu_s)}{|t-s|}
    =\sqrt{\frac{1-P}{1-s^2(1-P)}}
\end{equation*}
This is \eqref{geo0b}.
\end{proof}
For a better understanding of the estimates occurring in Lemma \ref{upper} and Corollary \ref{upperbest}, and for later use,  consider now the following example of $\gamma\in \Gamma(M|\nu,\varrho)$.
\begin{example}\label{exupper}
Let $M$ be a unital ${\mathsf C}^*$-algebra admitting a non-factorial  $^*$-representation $\{\pi,{\mathcal H}_\pi\}$. Let $N=\pi(M)^{\prime\prime}$, and be $0<z<{\mathsf 1}$ a central orthoprojection of $N$. Let $\varphi\in z{\mathcal H}_\pi$ and $\phi\in  z^\perp{\mathcal H}_\pi$ be unit vectors, and $a\in N_+$ with  $\|a\varphi\|=1$. Define the unit vector $\zeta$ by $$\zeta=\frac{a\varphi+\phi}{\sqrt{2}}$$
Let $\nu$ and $\varrho$ be those states on $M$ which are implemented by $\varphi$ and $\zeta$ via $\pi$, respectively. Then, in terms of the notions used in Lemma \ref{upper} and Example \ref{ex2}, for $\gamma\in \Gamma(M|\nu,\varrho)$ constructed along with \eqref{geo0}, in addition to \eqref{geo2a} one has
\begin{equation}\label{geo3}
  \|\hat{\psi}_0\|_\pi<{\mathop{\mathrm{dil}}}_0^B \gamma =\|\psi\|_\pi
\end{equation}
\end{example}
\begin{proof}
Since $z\geq p_\pi^{\,\prime}(\varphi)$ and $z\phi=0$ are fulfilled, by definition of $\zeta$ one finds that \begin{equation}\label{gl}
p_\pi^{\,\prime}(\varphi)\zeta=\frac{a\varphi}{\sqrt{2}},\ p_\pi^{\,\prime}(\varphi)^\perp\zeta=\frac{\phi}{\sqrt{2}}\not=0
\end{equation}
Also, $z\in N\cap N^\prime$ implies that  $$h^\pi_{\zeta,\varphi}=\frac{1}{\sqrt{2}}\, h^\pi_{a\varphi,\varphi}=\frac{1}{\sqrt{2}}\, h^\pi_{\sqrt{a}\varphi,\sqrt{a}\varphi}$$ Hence, $h^\pi_{\zeta,\varphi}\geq 0$ on $N^\prime$. Thus $$F=F(M|\nu,\varrho)=\langle \zeta,\varphi\rangle_\pi=\frac{1}{\sqrt{2}}\,\langle a\varphi,\varphi\rangle_\pi$$ and then $\gamma\in \Gamma(M|\nu,\varrho)$ arising along with \eqref{geo0} may be considered.  By \eqref{geo1} this implementation of $\gamma$ admits the tangent vector $\psi=\zeta-{F}\varphi$. From this and \eqref{gl} then $$p_\pi^{\,\prime}(\varphi)\psi=\frac{1}{\sqrt{2}}\,\bigl(a-\langle a\varphi,\varphi\rangle_\pi\bigr)\varphi\in N_{\mathsf{h}}\varphi,\ p_\pi^{\,\prime}(\varphi)^\perp\psi\not=0$$
follow. Hence, $\hat{\psi}_0=p_\pi^{\,\prime}(\varphi)\psi$, and thus   $\|\hat{\psi}_0\|_\pi<\|\hat{\psi}_0+p_\pi^{\,\prime}(\varphi)^\perp \psi\|_\pi=\|\psi\|_\pi$.
Together with \eqref{geo2a} this is \eqref{geo3}.
\end{proof}
Conclude with some auxiliary results about the construction schema \eqref{geo0}. Relating notations, if orthoprojections $p,q$ of a $vN$-algebra $N$ are given, then by $p\vee q$ the least orthoprojection of $N$ with $p\vee q\geq p$ and $p\vee q\geq q$ is meant.
\begin{lemma}\label{porter}
    Let $\{\pi,{\mathcal H}_\pi\}$ be a unital $^*$-representation, $\zeta,\varphi\in {\mathcal H}_\pi$ with  $h_{\zeta,\varphi}^\pi\geq 0$ on $\pi(M)^{\,\prime}$. Let $\phi=r\,\zeta+s\,\varphi$ and  $\hat{\phi}=\hat{r}\,\zeta+\hat{s}\,\varphi$, with reals $r,\hat{r},s,\hat{s}>0$. Then, $h^\pi_{\phi,\hat{\phi}}\geq {\mathsf 0}$, with \begin{equation}\label{porter0}
    s(h^\pi_{\phi,\hat{\phi}})=p_\pi^{\,\prime}(\zeta)\vee p_\pi^{\,\prime}(\varphi)=p_\pi^{\,\prime}(\phi)=p_\pi^{\,\prime}(\hat{\phi})
    \end{equation} Moreover, $s(h^\pi_{\phi,\varphi})= p_\pi^{\,\prime}(\varphi)$ and $s(h^\pi_{\zeta,\phi})=p_\pi^{\,\prime}(\zeta)$ hold.
\end{lemma}
\begin{proof}
Owing to $h_{\varphi,\zeta}^\pi=h_{\zeta,\varphi}^\pi\geq 0$ and by the strict positivity of all coefficients
\begin{equation}\label{porter1}
    h_{\phi,\hat{\phi}}^\pi=r \hat{r} h_{\zeta,\zeta}^\pi+(r\hat{s}+\hat{r} s)\, h_{\zeta,\varphi}^\pi+s \hat{s} h_{\varphi,\varphi}^\pi\geq {\mathsf 0}
\end{equation}
is fulfilled, where all $h^\pi$-functionals occuring are positive on $\pi(M)^{\,\prime}$, with $\|h_{\phi,\hat{\phi}}^\pi\|_1=\langle\phi,\hat{\phi}\rangle_\pi$, $\|h_{\varphi,\varphi}^\pi\|_1=\|\varphi\|_\pi^2$, $\|h_{\zeta,\zeta}^\pi\|_1=\|\zeta\|_\pi^2$ and $\|h_{\zeta,\phi}^\pi\|_1=\langle\zeta,\varphi\rangle_\pi\geq 0$.
Hence, considering \eqref{porter1} at the ${\mathsf 1}$-operator, we obtain
\begin{subequations}\label{porter2}
\begin{equation}\label{porter2a}
 \|h_{\phi,\hat{\phi}}^\pi\|_1=r \hat{r} \|h_{\zeta,\zeta}^\pi\|_1+(r\hat{s}+\hat{r} s)\,\|h_{\zeta,\phi}^\pi\|_1+s \hat{s}\|h_{\varphi,\varphi}^\pi\|_1 \end{equation}
Let
$p^{\,\prime}=s(h^\pi_{\phi,\hat{\phi}})\in \pi(M)^{\,\prime}$ be the support orthoprojection of $h_{\phi,\hat{\phi}}^\pi$ in $\pi(M)^{\,\prime}$. Hence,
$h_{\phi,\hat{\phi}}^\pi(p^{\,\prime})=\|h_{\phi,\hat{\phi}}^\pi\|_1$, and thus from \eqref{porter1}  the following is seen
\begin{equation}\label{porter2b}
    \|h_{\phi,\hat{\phi}}^\pi\|_1=r \hat{r} h_{\zeta,\zeta}^\pi(p^{\,\prime})+(r\hat{s}+\hat{r} s)\, h_{\zeta,\varphi}^\pi(p^{\,\prime})+s \hat{s} h_{\varphi,\varphi}^\pi(p^{\,\prime})
\end{equation}
\end{subequations}
with $0\leq h_{\zeta,\zeta}^\pi(p^{\,\prime})\leq \|h_{\zeta,\zeta}^\pi\|_1$, $ 0\leq h_{\zeta,\varphi}^\pi(p^{\,\prime})\leq \|h_{\zeta,\phi}^\pi\|_1$ and $0\leq h_{\varphi,\varphi}^\pi(p^{\,\prime})\leq \|h_{\varphi,\varphi}^\pi\|_1$.
Note that by assumption the coefficients appearing on the right hand side of \eqref{porter1} are strictly positive.
Hence, under these conditions, by comparing \eqref{porter2a} with \eqref{porter2b} we see that in the previous estimates equality in fact has to be enforced to hold, that is, $h_{\zeta,\varphi}^\pi(p^{\,\prime})= \|h_{\zeta,\phi}^\pi\|_1$, $ h_{\varphi,\varphi}^\pi(p^{\,\prime})=\|h_{\varphi,\varphi}^\pi\|_1$ and $h_{\zeta,\zeta}^\pi(p^{\,\prime})= \|h_{\zeta,\zeta}^\pi\|_1$ have to be fulfilled. This implies $p^{\,\prime}\geq p_\pi^{\,\prime}(\varphi)$ and $p^{\,\prime}\geq p_\pi^{\,\prime}(\zeta)$, for, $p_\pi^{\,\prime}(\varphi)$ and $p_\pi^{\,\prime}(\zeta)$ are the support orthoprojections of $h_{\varphi,\varphi}^\pi$ and $h_{\zeta,\zeta}^\pi$ in  $\pi(M)^{\,\prime}$.
Hence $$p^{\,\prime}\geq p_\pi^{\,\prime}(\zeta)\vee p_\pi^{\,\prime}(\varphi)$$ On the other hand, let $q\in \pi(M)^{\,\prime}$ be an orthoprojection, with $q\leq  (p_\pi^{\,\prime}(\zeta)\vee p_\pi^{\,\prime}(\varphi))^\perp$. Since \eqref{porter1} implies $h_{\phi,\hat{\phi}}^\pi(q)=0$,
$q\leq {p^{\,\prime}}^\perp$
 follows. Clearly, under the supposition $$p^{\,\prime}> p_\pi^{\,\prime}(\zeta)\vee p_\pi^{\,\prime}(\varphi)$$ this argumentation also could be applied to $q=p^{\,\prime}-p_\pi^{\,\prime}(\zeta)\vee p_\pi^{\,\prime}(\varphi)$, with the result $$p^{\,\prime}-p_\pi^{\,\prime}(\zeta)\vee p_\pi^{\,\prime}(\varphi)\leq {p^{\,\prime}}^\perp$$ which is a contradiction. Thus $$s(h^\pi_{\phi,\hat{\phi}})=p^{\,\prime}=p_\pi^{\,\prime}(\zeta)\vee p_\pi^{\,\prime}(\varphi)$$ has to be  true. Obviously, in the special case if $\phi=\hat{\phi}$ the latter remains true with  $$p_\pi^{\,\prime}(\zeta)\vee p_\pi^{\,\prime}(\varphi)=s(h^\pi_{\phi,\hat{\phi}})=p_\pi^{\,\prime}(\phi)=p_\pi^{\,\prime}(\hat{\phi})$$
Then, on combining both relations  \eqref{porter0} is obtained. The additional two relations can be followed by applying similar arguments (the obvious details are omitted).
\end{proof}
\begin{corolla}\label{porter3}
    Let $\gamma\in \Gamma(M|\nu,\varrho)$ be as in \textup{Example \ref{ex2}}, implemented by $[-1,1]\ni t\longmapsto \varphi_t\in {\mathcal S}_{\pi,M}(\nu_t)$ of \eqref{geo0a}. Then,  for  $t,u\in]0,1[$,  $h^\pi_{\varphi_t,\varphi_u}$, $h^\pi_{\varphi_t,\varphi}$ and $h^\pi_{\zeta,\varphi_u}$ are positive and
    \begin{subequations}\label{porter4}
    \begin{eqnarray}\label{porter4a}
    p_\pi^{\,\prime}(\varphi_t)& =& s(h^\pi_{\varphi_t,\varphi_u})=p_\pi^{\,\prime}(\varphi_u)= p_\pi^{\,\prime}(\zeta)\vee p_\pi^{\,\prime}(\varphi)  \\
    \label{porter4b} p_\pi^{\,\prime}(\varphi)& = & s(h^\pi_{\varphi_t,\varphi})  \\
    \label{porter4c} p_\pi^{\,\prime}(\zeta)& = & s(h^\pi_{\zeta,\varphi_u})
    \end{eqnarray}
    \end{subequations}
    \end{corolla}
   \begin{proof}
    Apply Lemma \ref{porter} with $r=t$, $s=\lambda(t)$, $\hat{r}=u$, $\hat{s}=\lambda(u)$.
   \end{proof}
   \subsubsection{A formula for the local dilation function}\label{iline}
   We are ready now to derive an exact expression for the local dilation function of a parameterized curve
    $$\gamma: I\ni t\,\mapsto\,\nu_t\in
      {\mathcal S}(M)$$
   which is passing through a state $\nu$ at $t=0$  and is admitting a $\gamma$-compliant  unital $^*$-representation  $\{\pi,{\mathcal H}_\pi\}$ around $\nu$. In line with this, suppose
    $$I_\pi\ni t\mapsto \varphi_t\in {\mathcal S}_{\pi,M}(\nu_t)$$
   is a  differentiable at $t=0$ implementation of $\gamma$ around $\nu$. We are going to show that  Corollary \ref{upperbest} can be strengthened so as to equality occurring in formula \eqref{invari1}.
 \begin{theorem}\label{line0}
   If $\psi=\psi(\pi,{\mathcal H}_\pi,(\varphi_t))\in{\mathsf{T}}_\nu(M|\gamma)$ is the tangent vector  at $\nu$,  then
   \begin{equation}\label{line1}
       {\mathop{\mathrm{dil}}}_0^B \gamma=\inf_{\phi\in{\mathsf{T}}_\nu(M|\gamma)}\|\phi\|=\sqrt{\|\hat{\psi}_0\|_\pi^2+\|p_\pi(\varphi)^\perp p_\pi^{\,\prime}(\varphi)^\perp\psi\|_\pi^2}
   \end{equation}
   is fulfilled, with $\hat{\psi}_0$ being the best $\langle\cdot,\cdot\rangle_{\pi,{\mathbb R}}$-approximation of $\psi\in ({\mathcal H}_\pi)_{\mathbb R}$ in $[\pi(M)_{\mathsf{h}}\varphi]$.
   \end{theorem}
   \begin{proof}
   Throughout let $N=\pi(M)^{\,\prime\prime}$. Also, remind that according to \eqref{beta} the expression on the right hand side of formula \eqref{line1} will be abbreviated as $\beta_\pi(\psi)$.  In the special case where $\beta_\pi(\psi)=0$ is fulfilled the validity of \eqref{line1} is an immediate consequence  of Lemma \ref{upper}. Thus, we will suppose  $\beta_\pi(\psi)\not=0$. It is easy to infer from \eqref{beta} that then
   \begin{equation}\label{beta0}
    \hat{\psi}_0+p_\pi(\varphi)^\perp p_\pi^{\,\prime}(\varphi)^\perp\psi\not=0
   \end{equation}
   According to \eqref{decompo} the tangent vector $\psi$ uniquely decomposes as
   $\psi=\hat{\psi}_0+\hat{\psi}_1+p_\pi^{\,\prime}(\varphi)^\perp\psi$, with $\hat{\psi}_0\in [N_{\mathsf{h}}\varphi]$ and $\hat{\psi}_1\in {\mathsf{i}}p_\pi^{\,\prime}(\varphi)[N^\prime_{\mathsf{h}}\varphi]$, with $\varphi=\varphi_0$. Let sequences $\{x_n\}\subset N_{\mathsf{h}}$,
   $\{z_n\}\subset N^\prime_{\mathsf{h}}$, and  $\{y_n\}\subset N^\prime_{\mathsf{h}}$  be chosen in accordance with \eqref{decompoapprox1}, \eqref{decompoapprox2} and \eqref{ypsi}, i.e. such that
   \begin{subequations}\label{lines}
   \begin{eqnarray}\label{line2}
   \hat{\psi}_0 & = \lim_{n\to\infty} x_n\varphi\\ \label{line3}
    \hat{\psi}_1 & = \lim_{n\to\infty}{\mathsf{i}} z_n\varphi\\ \label{line2a}
    p_\pi(\varphi)p_\pi^{\,\prime}(\varphi)^\perp\psi& =\lim_{n\to\infty} {\mathsf{i}} y_n\varphi
   \end{eqnarray}
   \end{subequations}
   Also note that the following two identities are fulfilled:
   \begin{equation}\label{zerl}
   \begin{array}{rcrcrcr}
     p_\pi(\varphi)\psi & = & p_\pi(\varphi)\hat{\psi}_0&+ & \hat{\psi}_1&+& p_\pi(\varphi)p_\pi^{\,\prime}(\varphi)^\perp\psi \\
     p_\pi(\varphi)^\perp\psi &= &  p_\pi(\varphi)^\perp\hat{\psi}_0 & & &+& p_\pi(\varphi)^\perp p_\pi^{\,\prime}(\varphi)^\perp\psi
   \end{array}
   \end{equation}
   Let us consider the sequence $\{a_n\}\subset N^\prime_{\mathsf{h}}$  defined by $a_n=y_n+z_n$ for all subscripts. By \eqref{line2a} and \eqref{line3} we have
   \begin{equation}\label{line2b}
     \hat{\psi}_1 +  p_\pi(\varphi)p_\pi^{\,\prime}(\varphi)^\perp\psi=\lim_{n\to\infty} {\mathsf{i}} a_n\varphi
   \end{equation}
   Also, for each $n\in {\mathbb{N}}$, in defining at any $t\in I_\pi$ implementing vectors $\varphi_t^{(n)}$ of $\nu_t$ by
   \begin{equation}\label{line5}
     \varphi_t^{(n)}=\exp(-{\mathsf{i}}\,a_n t)\,\varphi_t\in {\mathcal S}_{\pi,M}(\nu_t)
   \end{equation}
   local implementations $(\varphi_t^{(n)})$ of $\gamma$ are given.
   Let $\psi^{(n)}=\psi\bigl(\pi,{\mathcal H}_\pi,(\varphi_t^{(n)})\bigr)\in{\mathsf{T}}_\nu(M|\gamma)$ be the tangent vector of this implementation. By \eqref{line5} then $\psi^{(n)}=\psi-{\mathsf{i}}\,a_n\varphi$, and thus from  \eqref{zerl} one infers that
   \begin{equation}\label{line6}
      \psi^{(n)}=p_\pi(\varphi)\hat{\psi}_0+\bigl(\hat{\psi}_1 +  p_\pi(\varphi)p_\pi^{\,\prime}(\varphi)^\perp\psi-{\mathsf{i}}\,a_n\varphi\bigr)+p_\pi(\varphi)^\perp \bigl(\hat{\psi}_0+p_\pi^{\,\prime}(\varphi)^\perp\psi\bigr)
   \end{equation}
   Then, if  another sequence of vectors $\tilde{\psi}^{(n)}$ is defined by the setting
   \begin{equation}\label{line7}
      \tilde{\psi}^{(n)}=p_\pi(\varphi)x_n\varphi+p_\pi(\varphi)^\perp \bigl(\hat{\psi}_0+p_\pi^{\,\prime}(\varphi)^\perp\psi\bigr)
   \end{equation}
   by \eqref{line2} together with \eqref{line6} and \eqref{line2b} it is easily seen that
   \begin{equation}\label{line8}
       \lim_{n\to\infty} \bigl(\psi^{(n)}-\tilde{\psi}^{(n)}\bigr)=0
   \end{equation}
   Note that $\langle \tilde{\psi}^{(n)},\varphi\rangle_\pi=0$ is fulfilled, for each $n$.
   In all cases where  $p_\pi(\varphi)$ is not a minimal orthoprojection of $N$, by Lemma \ref{auxx} the $x_n$ figuring in \eqref{line2} can be chosen to satisfy $p_\pi(\varphi)x_n\varphi\not=0$, for each $n\in {{\mathbb{N}}}$. In this case
   then $\tilde{\psi}^{(n)}\not=0$ is guaranteed by \eqref{line7}. Hence, the vectors $\tilde{\varphi}_t$ defined by
   \begin{equation}\label{line9}
      \tilde{\varphi}_t^{(n)}=t\,\tilde{\psi}^{(n)}+\sqrt{1-t^2\|\tilde{\psi}^{(n)}\|^2}\,\varphi, \ |t|\leq \|\tilde{\psi}^{(n)}\|^{-1},
   \end{equation}
   are unit vectors, and if $J_n=\bigl\{t\in {\mathbb R}: |t|\leq \|\tilde{\psi}^{(n)}\|^{-1}\bigr\}$ is defined, by
   \begin{subequations}\label{def}
    \begin{equation}\label{def1}
    J_n\ni t\longmapsto\tilde{\varphi}_t^{(n)}\in {\mathcal S}_{\pi,M}(\tilde{\nu}_t^{(n)})
   \end{equation}
    a local implementation of some parameterized curve
    \begin{equation}\label{def2}
    \tilde{\gamma}^{(n)}: J_n\ni t\longmapsto \tilde{\nu}_t^{(n)}\in {\mathcal S}(M)
    \end{equation}
   \end{subequations}
    at the state $ \tilde{\nu}_0^{(n)}=\nu$ is given.
   Note that by Definition \ref{budi} one has
   \begin{equation*}
   \begin{split}
       d_B\bigl(M|\nu_t,\tilde{\nu}_t^{(n)}\bigr)& \leq \| \varphi_t^{(n)}-\tilde{\varphi}_t^{(n)}\|_\pi= \| (\varphi+t \psi^{(n)}+{\mathbf o}(t))-(\varphi+t\tilde{\psi}^{(n)}+{\mathbf o}(t))\|_\pi\\
       &\leq  \| t( \psi^{(n)}-\tilde{\psi}^{(n)})+{\mathbf o}(t)\|_\pi
   \end{split}
   \end{equation*}
   Obviously, from this the following relation is obtained:
   \begin{equation}\label{line10}
      \limsup_{t\to 0} \frac{d_B\bigl(M|\nu_t,\tilde{\nu}_t^{(n)}\bigr)}{|t|}\leq \| \psi^{(n)}-\tilde{\psi}^{(n)}\|_\pi
   \end{equation}
   Now, due to \eqref{invari1}, ${\mathop{\mathrm{dil}}}_0^B \gamma<\infty$ holds. Hence, owing to \eqref{line10} an application of Corollary \ref{nuequiv} with $\varepsilon=\| \psi^{(n)}-\tilde{\psi}^{(n)}\|_\pi$ will show that
   \begin{equation*}
       |\,{\mathop{\mathrm{dil}}}_0^B \gamma-{\mathop{\mathrm{dil}}}_0^B \tilde{\gamma}^{(n)}|\leq \| \psi^{(n)}-\tilde{\psi}^{(n)}\|_\pi
   \end{equation*}
   for each $n$. Owing to \eqref{line8}, by taking the limit $\lim_{n\to\infty}$ from the latter
   \begin{equation}\label{line11}
      {\mathop{\mathrm{dil}}}_0^B \gamma=\lim_{n\to\infty} {\mathop{\mathrm{dil}}}_0^B \tilde{\gamma}^{(n)}
   \end{equation}
   is obtained. We are going to calculate the value of ${\mathop{\mathrm{dil}}}_0^B \tilde{\gamma}^{(n)}$, now. Thus, let $n\in {{\mathbb{N}}}$ be fixed but arbitrarily chosen. With the $x_n$ figuring in \eqref{line2}, we fix a real $\alpha_n$ in accordance with $\alpha_n\geq \|p_\pi(\varphi)x_np_\pi(\varphi)\|$. Note that by Lemma \ref{auxx} in the case at hand owing to $\langle x_n\varphi,\varphi\rangle_\pi=0$ and $p_\pi(\varphi)x_n\varphi\not=0$ we can be assured that for each $n\in {{\mathbb{N}}}$ a unit vector $\zeta_n\in{\mathcal H}_\pi$ and a state $\varrho^{(n)}\in {\mathcal S}(M)$ can be defined by the setting
   \begin{equation}\label{line12}
   \zeta_n=\frac{\tilde{\psi}^{(n)}+\alpha_n\varphi}{\|\tilde{\psi}^{(n)}+\alpha_n\varphi\|_\pi}\in {\mathcal S}_{\pi,M}(\varrho^{(n)})
   \end{equation}
   Note that $\tilde{\psi}^{(n)}+\alpha_n\varphi=b_n\varphi$ holds, with $b_n=p_\pi(\varphi)x_n p_\pi(\varphi)+\alpha_n p_\pi(\varphi)\in N_{\mathsf h}$. By definition of $\alpha_n$, $b_n$ is positive semi-definite. But then, we can be assured that $$h^\pi_{\zeta_n,\varphi}=h^\pi_{\sqrt{b_n}\varphi,\sqrt{b_n}\varphi}\geq 0$$ is fulfilled on $N^\prime$. From this and since $\langle \tilde{\psi}^{(n)},\varphi\rangle_\pi=0$ one easily infers that
   \begin{equation}\label{line13}
   {F}={F\bigl(M|\varrho^{(n)},\nu\bigr)}=\langle \zeta_n,\varphi\rangle_\pi=\frac{\alpha_n}{\|\tilde{\psi}^{(n)}+\alpha_n\varphi\|_\pi}=\frac{\alpha_n}{\sqrt{\|\tilde{\psi}^{(n)}\|_\pi^2+\alpha_n^2}}
   \end{equation}
  holds. Hence, in line with Example \ref{ex2} a parameterized curve $\gamma^{(n)}\in \Gamma\bigl(M|\nu,\varrho^{(n)}\bigr)$, $$\gamma^{(n)}: [-1,1]\ni t\longmapsto \nu^{(n)}_t$$ can be considered which arises from the implementation given by equations \eqref{geo0} for $\zeta=\zeta_n$.  A view on \eqref{line9}/\eqref{def} is showing that the parameterized curve $\tilde{\gamma}^{(n)}$ around $t=0$ is obtained from  $\gamma^{(n)}\in \Gamma\bigl(M|\nu,\varrho^{(n)}\bigr)$ by a simple re-parametrization:
   \begin{equation}\label{repar}
       \tilde{\nu}^{(n)}_t=\nu^{(n)}_{t\sqrt{\|\tilde{\psi}^{(n)}\|_\pi^2+\alpha_n^2}}
   \end{equation}
   Application of Example \ref{ex2}, \eqref{geo2a}, with $\gamma=\gamma^{(n)}$ in view of \eqref{line13} then implies
   \begin{equation*}
    {\mathop{\mathrm{dil}}}_0^B \gamma^{(n)} =\limsup_{t\to 0}\frac{d_B(M|\nu^{(n)}_t,\nu)}{|t|}=\sqrt{1-P}=\frac{\|\tilde{\psi}^{(n)}\|_\pi}{\sqrt{\|\tilde{\psi}^{(n)}\|_\pi^2+\alpha_n^2}}
   \end{equation*}
   Hence, in making use of  \eqref{repar} and the latter we can conclude as follows:
   \begin{equation*}
       \begin{split}
       {\mathop{\mathrm{dil}}}_0^B \tilde{\gamma}^{(n)}& =\limsup_{t\to 0}\frac{d_B\Bigl(M\Big|\nu^{(n)}_{t\sqrt{\|\tilde{\psi}^{(n)}\|_\pi^2+\alpha_n^2}},\nu\Bigr)}{|t|}\\
       &=\Bigl(\sqrt{\|\tilde{\psi}^{(n)}\|_\pi^2+\alpha_n^2}\Bigr)\limsup_{t\to 0}\frac{d_B\bigl(M\Big|\nu^{(n)}_{t\sqrt{\|\tilde{\psi}^{(n)}\|_\pi^2+\alpha_n^2}},\nu\bigr)}{|t|\sqrt{\|\tilde{\psi}^{(n)}\|_\pi^2+\alpha_n^2}}
       \\
       &=\bigl(\sqrt{\|\tilde{\psi}^{(n)}\|_\pi^2+\alpha_n^2}\bigr)\frac{\|\tilde{\psi}^{(n)}\|_\pi}{\sqrt{\|\tilde{\psi}^{(n)}\|_\pi^2+\alpha_n^2}}  \\
       & =\|\tilde{\psi}^{(n)}\|_\pi
        \end{split}
   \end{equation*}
   That is, we have been arriving at the following relation:
   \begin{equation}\label{line14}
       \begin{split}
       {\mathop{\mathrm{dil}}}_0^B \tilde{\gamma}^{(n)}=\|\tilde{\psi}^{(n)}\|_\pi
        \end{split}
   \end{equation}
   Clearly, since $n$ could have been chosen at will, \eqref{line14} holds for any $n\in {{\mathbb{N}}}$.
   Thus, the limit for $n\to\infty$ can be considered. By \eqref{line2} and \eqref{line7} we may follow that
   \begin{equation}\label{line15}
   \lim_{n\to\infty} \tilde{\psi}^{(n)}=p_\pi(\varphi)\hat{\psi}_0+p_\pi(\varphi)^\perp \bigl(\hat{\psi}_0+p_\pi^{\,\prime}(\varphi)^\perp\psi\bigr)=\hat{\psi}_0+p_\pi(\varphi)^\perp p_\pi^{\,\prime}(\varphi)^\perp\psi
   \end{equation}
   must hold. In accordance with \eqref{line11} and \eqref{line14} from  \eqref{line15} the relation
   \begin{equation*}
     {\mathop{\mathrm{dil}}}_0^B \gamma= \lim_{n\to\infty} {\mathop{\mathrm{dil}}}_0^B \tilde{\gamma}^{(n)}= \|\hat{\psi}_0+ p_\pi(\varphi)^\perp p_\pi^{\,\prime}(\varphi)^\perp\psi\|_\pi
   \end{equation*}
   follows. Finally, since owing to $ p_\pi^{\,\prime}(\varphi)\hat{\psi}_0=\hat{\psi}_0$ the vector $\hat{\psi}_0$ has to be
   orthogonal to $p_\pi^{\,\prime}(\varphi)^\perp p_\pi(\varphi)^\perp \psi= p_\pi(\varphi)^\perp p_\pi^{\,\prime}(\varphi)^\perp\psi$, from the previous $${\mathop{\mathrm{dil}}}_0^B \gamma=\sqrt{\|\hat{\psi}_0\|_\pi^2+ \|p_\pi(\varphi)^\perp p_\pi^{\,\prime}(\varphi)^\perp\psi\|_\pi^2}=\beta_\pi(\psi)$$ is seen. Thus, in view of Corollary \ref{upperbest} the validity of \eqref{line1} is obtained in the special case if $p_\pi(\varphi)$ is not a minimal orthoprojection of $N$.

   Finally, let us consider now the case with $p_\pi(\varphi)$ being  a minimal orthoprojection of $N$. Then, by Lemma \ref{auxx}, $p_\pi(\varphi)^\perp \hat{\psi}_0= \hat{\psi}_0$ holds. In this case, instead of \eqref{line6}
   \begin{equation}\label{line6add}
      \psi^{(n)}=\bigl(\hat{\psi}_1 +  p_\pi(\varphi)p_\pi^{\,\prime}(\varphi)^\perp\psi-{\mathsf{i}}\,a_n\varphi\bigr)+p_\pi(\varphi)^\perp \bigl(\hat{\psi}_0+p_\pi^{\,\prime}(\varphi)^\perp\psi\bigr)
   \end{equation}
   is fulfilled, and then with $\tilde{\psi}=\hat{\psi}_0+p_\pi(\varphi)^\perp p_\pi^{\,\prime}(\varphi)^\perp\psi\not=0$ instead of \eqref{line7} one has
    \begin{equation}\label{line7add}
      \lim_{n\to\infty}\bigl(\psi^{(n)}-\tilde{\psi}\bigr)=0
   \end{equation}
   for $p_\pi(\varphi)^\perp \bigl(\hat{\psi}_0+p_\pi^{\,\prime}(\varphi)^\perp\psi\bigr)=\hat{\psi}_0+p_\pi(\varphi)^\perp p_\pi^{\,\prime}(\varphi)^\perp\psi$ is fulfilled.
   Note that $ \tilde{\psi}\not=0$ and $\langle \tilde{\psi},\varphi\rangle_\pi=0$ hold. Hence, those vectors $\tilde{\varphi}_t$ which are defined by
   \begin{equation}\label{line9add}
      \tilde{\varphi}_t=t\,\tilde{\psi}+\sqrt{1-t^2\|\tilde{\psi}\|_\pi^2}\,\varphi, \ |t|\leq \|\tilde{\psi}\|_\pi^{-1},
   \end{equation}
   are unit vectors, and if $J=\{t\in {\mathbb R}: |t|\leq \|\tilde{\psi}\|_\pi^{-1}\}$ is defined, by
   \begin{subequations}\label{defadd}
    \begin{equation}\label{def1add}
    J\ni t\longmapsto\tilde{\varphi}_t\in {\mathcal S}_{\pi,M}(\tilde{\nu}_t)
   \end{equation}
    a local implementation of some parameterized curve
    \begin{equation}\label{def2add}
    \tilde{\gamma}: J\ni t\longmapsto \tilde{\nu}_t\in {\mathcal S}(M)
    \end{equation}
   \end{subequations}
    at the state $ \tilde{\nu}_0=\nu$ is given.
   Note that by Definition \ref{budi} one has
   \begin{equation*}
   \begin{split}
       d_B\bigl(M|\nu_t,\tilde{\nu}_t\bigr)& \leq \| \varphi_t^{(n)}-\tilde{\varphi}_t\|_\pi= \| (\varphi+t \psi^{(n)}+{\mathbf o}(t))-(\varphi+t\tilde{\psi}+{\mathbf o}(t))\|_\pi\\
       &\leq  \| t( \psi^{(n)}-\tilde{\psi})+{\mathbf o}(t)\|_\pi
   \end{split}
   \end{equation*}
   Obviously, from this the following relation is obtained:
   \begin{equation}\label{line10add}
      \limsup_{t\to 0} \frac{d_B(M|\nu_t,\tilde{\nu}_t)}{|t|}\leq \| \psi^{(n)}-\tilde{\psi}\|_\pi
   \end{equation}
   Now, due to \eqref{invari1}, ${\mathop{\mathrm{dil}}}_0^B \gamma<\infty$ holds. Hence, owing to \eqref{line10add} an application of Corollary \ref{nuequiv} with $\varepsilon=\| \psi^{(n)}-\tilde{\psi}\|_\pi$ will show that
   \begin{equation*}
       |\,{\mathop{\mathrm{dil}}}_0^B \gamma-{\mathop{\mathrm{dil}}}_0^B \tilde{\gamma}|\leq \| \psi^{(n)}-\tilde{\psi}\|_\pi
   \end{equation*}
   In taking the limit for $n\to\infty$ of the latter yields
   \begin{equation}\label{line11add}
      {\mathop{\mathrm{dil}}}_0^B \gamma={\mathop{\mathrm{dil}}}_0^B \tilde{\gamma}
   \end{equation}
  We are going to calculate ${\mathop{\mathrm{dil}}}_0^B \tilde{\gamma}$ now. Note first that $h_{\tilde{\psi},\varphi}^\pi=0$ on $N^{\,\prime}$, by definition of $\tilde{\psi}$. Hence, by definition in \eqref{line9add} also $$h_{\tilde{\varphi}_t,\varphi}^\pi=t^2 h_{\tilde{\psi},\tilde{\psi}}^\pi+\bigl(1-t^2\|\tilde{\psi}\|_\pi^2\bigr)h_{\varphi,\varphi}^\pi\geq 0$$ on $N^{\,\prime}$. Due to  Theorem \ref{bas5}\,\eqref{bas53} and since $\tilde{\varphi}=\varphi+t\tilde{\psi}+{\mathbf o}(t)$ holds, this yields $$d_B(M|\tilde{\nu}_t,\nu)=\|\tilde{\varphi}_t-\varphi\|_\pi=\| t\tilde{\psi}+{\mathbf o}(t)\|_\pi$$
   From this we can conclude as follows:
   \[
   \begin{split}
   {\mathop{\mathrm{dil}}}_0^B \tilde{\gamma}& =\limsup_{t\to 0} \frac{d_B(M|\tilde{\nu}_t,\nu)}{|t|}=\| \tilde{\psi}\|_\pi=
   \|\hat{\psi}_0+p_\pi(\varphi)^\perp p_\pi^{\,\prime}(\varphi)^\perp\psi\|_\pi\\ &=\sqrt{\|\hat{\psi}_0\|_\pi^2+\|p_\pi(\varphi)^\perp p_\pi^{\,\prime}(\varphi)^\perp\psi\|_\pi^2}
   \end{split}
   \]
   Thus, by \eqref{line11add}, if $p_\pi(\varphi)\in N$ is minimal, then the assertion remains true, too.
   \end{proof}
   \begin{remark}\label{analog}
       \begin{enumerate}
         \item\label{analog1}
         On the one hand, by Theorem \ref{line0}, the dilation function
         \begin{equation}\label{global}
          {\mathop{\mathrm{dil}}}_0^B \gamma=\inf_{\phi\in{\mathsf{T}}_\nu(M|\gamma)}\|\phi\|
         \end{equation}
           is represented in terms of the class ${\mathsf{T}}_\nu(M|\gamma)$ of all possible tangent vectors coming along with the local differentiable implementations of $\gamma$ at $\nu$.
         \item\label{analog2}
         On the other hand, the invariant $\beta=\beta_\pi(\psi)$ can be calculated locally by
         \begin{equation}\label{local}
          {\mathop{\mathrm{dil}}}_0^B \gamma=\sqrt{\|\hat{\psi}_0\|_\pi^2+\|p_\pi(\varphi)^\perp p_\pi^{\,\prime}(\varphi)^\perp\psi\|_\pi^2}=\beta_\pi(\psi)
         \end{equation}
         within any $\gamma$-compliant $^*$-representation $\{\pi,{\mathcal H}_\pi\}$ and corresponding local implementations of $\gamma$ around $\nu$.
         Thereby, note that the best $\langle\cdot,\cdot\rangle_{\pi,\mathbb R}$-approximation  $\hat{\psi}_0$ of $\psi=\psi(\pi,{\mathcal H}_\pi,(\varphi_t))$ within $[\pi(M)^{\prime\prime}_{\mathsf{h}}\varphi]=[\pi(M)_{\mathsf{h}}\varphi]$ depends from the $\gamma$-compliant $^*$-representation $\{\pi,{\mathcal H}_\pi\}$ as well as from the  differentiable local implementation $(\varphi_t)$ chosen, that is, one has that $\hat{\psi}_0=\hat{\psi}_0(\pi,{\mathcal H}_\pi,(\varphi_t))$ as well.
       \end{enumerate}
   \end{remark}
\newpage
\section{The length structure of the Bures geometry}\label{lstruc}
\subsection{Geodesic arcs}\label{lstrucarcs}
In restricting the curves of $\Gamma(M|\nu,\varrho)$ to the part extending from $\nu$ to $\varrho$ we will arrive at another important class of implementable curves.
\begin{definition}\label{allg2c0}
For states $\nu\not=\varrho$, consider $\gamma\in \Gamma(M|\nu,\varrho)$ parameterized in accordance with \eqref{geo0}. Each parameterized curve $\gamma^{\,\prime}=\gamma|[0,1]$ obtained from such $\gamma$ by restricting parameters to the unit interval  $[0,1]$  is called  `geodesic arc' (extending from $\nu$ to $\varrho$). ${\mathcal{C}}^{\nu,\varrho}(M)$ be the set of all geodesic arcs extending from $\nu$ to $\varrho$.
\end{definition}
\begin{remark}\label{auxin2}
Note that according to Corollary \ref{auxin1} each geodesic arc $\gamma\in {\mathcal{C}}^{\nu,\varrho}(M)$ is injective as a map and therefore defines a simple curve. When required this fact in the sequel often tacitly and without explicit reference will be made use of.
\end{remark}
\begin{lemma}\label{geoarclength}
Let $\nu,\varrho\in {\mathcal S}(M)$. The Bures length of each $\gamma\in {\mathcal{C}}^{\nu,\varrho}(M)$ is
\begin{equation}\label{gal}
     \varTheta[\gamma]=\arcsin\sqrt{1-P(M|\nu,\varrho)}=\arccos{F(M|\nu,\varrho)}
\end{equation}
Thus, all geodesic arcs extending from $\nu$ to $\varrho$ have the same Bures length.
\end{lemma}
\begin{proof}
By Lemma \ref{rinow0} and Example \ref{ex2}, the Bures length $\varTheta[\gamma]$ of $\gamma\in {\mathcal{C}}^{\nu,\varrho}(M)$ is
 \begin{equation*}
 \varTheta[\gamma] =\int_0^1 {\mathop{\mathrm{dil}}}_s^B \gamma\ d\/s
= \int_0^1 \sqrt{\frac{1-P}{1-s^2(1-P)}}\ d\/s=\arcsin\sqrt{1-P},
\end{equation*}
with $P=P(M|\nu,\varrho)$.
\end{proof}
We remark that depending from the specifics of the two states $\nu$ and $\varrho$ a geodesic arc extending from $\nu$ to $\varrho$ might not be uniquely determined and then more than one curve might be in ${\mathcal{C}}^{\nu,\varrho}(M)$, see  Definition \ref{allg2c0}.  Uniqueness is describes as follows.
\begin{theorem}\label{einga}
 $\gamma\in {\mathcal{C}}^{\nu,\varrho}(M)$ is unique if, and only if, $\varrho_\nu$ is disjoint to $\nu_\varrho$.
\end{theorem}
\begin{proof}
Each geodesic arc $\gamma\in {\mathcal{C}}^{\nu,\varrho}(M)$ can be  implemented in accordance with the parametrization of \eqref{geo0}. Hence, $\gamma: [0,1]\ni t\longmapsto\nu_t\in {\mathcal S}(M)$ is given by
\begin{equation}\label{einga1}
\nu_t=t^2\,\varrho+\lambda(t)^2\nu+ 2\,t\lambda(t)\re f^\pi_{\zeta,\varphi}
\end{equation}
for $\varphi\in {\mathcal S}_{\pi,M}(\nu)$, $\zeta\in  {\mathcal S}_{\pi,M}(\varrho)$  with    $F(M|\nu,\varrho)=\langle\zeta,\varphi\rangle_\pi$. By Theorem \ref{unique} and \eqref{einga1}, for disjoint $\varrho_\nu$ and $\nu_\varrho$ uniqueness follows, and Lemma \ref{disj} for the reverse.
\end{proof}
A particularly interesting special case is described in the following.
\begin{example}\label{geouni}
Let $\varrho\dashv \nu$, and be $\gamma: [0,1]\ni t\longmapsto \nu_t\in {\mathcal S}(M)$ a geodesic arc, $\gamma\in {\mathcal{C}}^{\nu,\varrho}(M)$. Up to the parametrization, $\gamma$ is unique, and $|\gamma|\subset \Omega^0_M(\nu)$ is fulfilled.
\end{example}
\begin{proof}
Suppose $\varrho\dashv \nu$. According to Theorem \ref{suppabsstet} and Corollary \ref{uniqueimp} then $\mu={\mathsf 0}$ is the only positive linear form obeying $\mu\leq \varrho$ and $\mu\perp\nu$. Hence, $\varrho_\nu={\mathsf 0}$, which is  disjoint to $\nu_\varrho$ in a trivial way. Application of Theorem \ref{einga} then gives uniqueness of $\gamma$.
 Now, let $[0,1]\ni t\longmapsto \varphi_t$ be the implementing family of $\gamma$, with $\varphi_t=t\,\zeta+\lambda(t)\,\varphi$, see \eqref{geo0a}. Then, since according to Theorem \ref{suppabsstet}\,\eqref{supp2} one has  $p_\pi(\zeta)=s(\varrho_\pi)\leq s(\nu_\pi)=p_\pi(\varphi)$, by construction of $\varphi_t$, $p_\pi(\varphi)\varphi_t=\varphi_t$ follows. Hence $s((\nu_t)_\pi)=p_\pi(\varphi_t)\leq p_\pi(\varphi)=s(\nu_\pi)$ has to hold, that is, $\nu_t\dashv \nu$, by Theorem \ref{suppabsstet}\,\eqref{supp1}--\eqref{supp2}. From this $|\gamma|\subset \Omega^0_M(\nu)$ follows.
\end{proof}
\begin{corolla}\label{geoconv}
Let $\varrho,\nu,\omega\in {\mathcal S}(M)$ with $\nu\not=\varrho$. Provided both $\varrho$ and $\nu$ belong to  $\Omega^0_M(\omega)$, then   $|\gamma|\subset \Omega^0_M(\omega)$, for  each $\gamma\in {\mathcal{C}}^{\nu,\varrho}(M)$.
\end{corolla}
\begin{proof}
Let $[0,1]\ni t\longmapsto \varphi_t$ be the implementing family of $\gamma$, with $\varphi_t=t\,\zeta+\lambda(t)\,\varphi$, see \eqref{geo0a}, with $\varphi\in {\mathcal S}_{\pi,M}(\nu)$ and $\zeta\in  {\mathcal S}_{\pi,M}(\varrho)$. By the Cauchy-Schwarz inequality and since $t$ and $\lambda(t)$ are non-negative, for each $x\in M$ one infers that
\[ \begin{split}
\nu_t(x^*x)= \|x\varphi_t\|_\pi^2 & = t^2\varrho(x^*x)+\lambda(t)^2\nu(x^*x)+ 2\,t\lambda(t)\Re\langle x\zeta,x\varphi\rangle_\pi\\
&\leq t^2\varrho(x^*x)+\lambda(t)^2\nu(x^*x)+ 2\,t\lambda(t) \sqrt{\varrho(x^*x)} \sqrt{\nu(x^*x)}
\end{split} \]
Hence, the conclusion is that, for any $t\in [0,1]$ and any $x\in M$,
\begin{equation}\label{geoconv3}
    \sqrt{\nu_t(x^*x)}\leq t\sqrt{\varrho(x^*x)}+ \lambda(t) \sqrt{\nu(x^*x)}
\end{equation}
is fulfilled. Now, suppose $\varrho,\nu\in \Omega^0_M(\omega)$, i.e. $\varrho\dashv\omega$ and $\nu\dashv\omega$ hold.
In view of \eqref{geoconv3} and by Definition \ref{absstet} from the supposition $\nu_t\dashv \omega$  follows. Hence $|\gamma|\subset \Omega^0_M(\omega)$.
\end{proof}
\begin{definition}\label{geodisj}
A subset $\Omega\subset {\mathcal S}(M)$ is termed `arc-determining set' if for any two states $\varrho, \nu\in \Omega$,  $\nu\not=\varrho$,  the geodesic arc between $\nu$ and $\varrho$ is uniquely determined.
\end{definition}
As a consequence of  Theorem \ref{einga} the following is true for arc-determining sets.
\begin{corolla}\label{geodisj1}
A non-void subset $\Omega\subset {\mathcal S}(M)$ is arc-determining if, and only if, for  any two states $\varrho, \nu\in \Omega$, with $\nu\not=\varrho$, $\varrho_\nu$ and $\nu_\varrho$ are mutually disjoint. \end{corolla}
\begin{example}\label{geodisj2}
\begin{enumerate}
\item \label{geodisj2a}
For a geodesic arc $\gamma$ and each choice of $a,b\in {\mathbb R}$,  $0\leq a<b\leq 1$ and $b-a<1$, the supporting set $|\gamma|[a,b]|$ of $\gamma|[a,b]$ is arc-determining.
\item \label{geodisj2b}
The supporting set $ |\gamma|$ of $\gamma\in {\mathcal{C}}^{\nu,\varrho}(M)$ is arc-determining if, and only if, $\gamma$ is uniquely determined in $ {\mathcal{C}}^{\nu,\varrho}(M) $, i.e. if  $\varrho_\nu$ and $\nu_\varrho$ are mutually disjoint.
\item \label{geodisj2c}
Each stratum $\Omega$ is arc-determining.
\end{enumerate}
\end{example}
\begin{proof}
Let $\gamma:[0,1]\ni t\longmapsto \nu_t$ be a geodesic arc, with $\nu\not=\varrho$, and $\nu_0=\nu$, $\nu_1=\varrho$. In line with Definition \ref{allg2c0}, suppose $\gamma$ to be implemented by  $[0,1]\ni t\longmapsto \varphi_t\in {\mathcal S}_{\pi,M}(\nu_t)$, with $\varphi_t$ given in accordance with  \eqref{geo0a}. Then, according to Corollary \ref{porter3},  \eqref{porter4a} and Corollary \ref{uniqueimp}\,\eqref{uniqueimp0}$\Leftrightarrow$\eqref{uniqueimp3a}, by assumptions on $a,b$, for each choice of $s,t$ with $a\leq s<t\leq b$ at least one of $(\nu_t)_{\nu_s}$ or $(\nu_s)_{\nu_t}$ has to be vanishing. Anyhow, this especially means that $(\nu_t)_{\nu_s}$ and $(\nu_s)_{\nu_t}$ are mutually disjoint. By Theorem \ref{einga} the geodesic arc connecting the two  states $\nu_t$ and $\nu_s$ of $|\gamma|$ then has to be uniquely determined. Thus $|\gamma|[a,b]|$ is an arc-determining set. This is \eqref{geodisj2a}. From this in view of Theorem \ref{einga} and  Corollary \ref{uniqueimp}\,\eqref{uniqueimp0}$\Leftrightarrow$\eqref{uniqueimp3a} together with \eqref{porter4b}--\eqref{porter4c} then also the validity of \eqref{geodisj2b} follows.
Finally, \eqref{geodisj2c} holds since for $\nu,\varrho\in \Omega$ we have $s(\nu_\pi)=s(\varrho_\pi)$, with respect to any $^*$-representation $\{\pi,{\mathcal H}_\pi\}$ with respect to which the $\pi$-fibres of both states are non-void. This e.g.~implies that $\mu\in M^*_+$ with $\mu\leq \varrho$ and $\mu\perp\nu$ has to vanish. Thus, $\varrho_\nu={\mathsf 0}$, and which makes that $\varrho_\nu$ and $\nu_\varrho$ are mutually disjoint in a trivial way. Then \eqref{geodisj2c} follows, by Theorem \ref{einga}.
\end{proof}
\subsection{The inner metric distance in Bures geometry}\label{geodistbu}
The importance of the curves of the class   ${\mathcal{C}}^{\nu,\varrho}(M)$  is that they yield special examples of continuous lines of shortest Bures length  connecting $\nu$ and $\varrho$. In line with the proof, the inner metric distance coming along with the Bures metric will be  under investigation. Thereby, except for possible  reparameterizations, the geodesic arcs of ${\mathcal{C}}^{\nu,\varrho}(M)$ will prove to be the only continuous lines of shortest Bures length in ${\mathcal S}(M)$ and doing the job.
\subsubsection{Minimizing the Bures length of a curve between two
states}\label{minbu}
Following Lemma \ref{geoarclength}, formula \eqref{gal}, depending from $\nu,\varrho\in {\mathcal S}(M)$ the Bures length of each geodesic arc extending from $\nu$ to $\varrho$ is given by the expression
\begin{equation}\label{geodist}
        {\mathfrak d}_B(M|\nu,\varrho)=\arcsin\sqrt{1-P(M|\nu,\varrho)}=\arccos{F(M|\nu,\varrho)}
\end{equation}
\begin{lemma}\label{dreieck}
 ${\mathcal S}(M)\times {\mathcal S}(M)\ni \{\nu,\varrho\}\longmapsto {\mathfrak d}_B(M|\nu,\varrho)$ is a metric distance.
\end{lemma}
\begin{proof}
First note that ${\mathfrak d}_B(M|\nu,\varrho)\geq 0$, and in view of \eqref{pcont.1} and \eqref{pcont.1aa}  one has
\begin{equation}\label{metzu}
    \sqrt{1-P(M|\nu,\varrho)}=d_B(M|\nu,\varrho)\sqrt{\frac{1+{F(M|\nu,\varrho)}}{2}}
\end{equation}
By \eqref{geodist} and \eqref{metzu}, ${\mathfrak d}_B(M|\nu,\varrho)=0$ is equivalent to $d_B(M|\nu,\varrho)=0$.  Thus ${\mathfrak d}_B(M|\nu,\varrho)=0$ iff $\nu=\varrho$. Obviously, since both $d_B(M|\nu,\varrho)$ and $P(M|\nu,\varrho)$ are symmetric, ${\mathfrak d}_B(M|\nu,\varrho)$ is symmetric, too. It remains to be shown that the triangle inequality
\begin{equation}\label{drei}
{\mathfrak d}_B(M|\nu,\varrho)\leq {\mathfrak d}_B(M|\nu,\omega)+{\mathfrak d}_B(M|\omega,\varrho)
\end{equation}
holds, for any triplet of states $\nu,\varrho,\omega\in {\mathcal S}(M)$.
We are going to deal with the special case where $\nu\dashv\omega$ and $\varrho\dashv\omega$ first. Let $\{\pi,{\mathcal H}_\pi\}$ be a unital $^*$-representation such that all three states have non-void $\pi$-fibres. Let $\varphi\in {\mathcal S}_{\pi,M}(\omega)$. By Theorem \ref{suppabsstet} there exist unique implementing vectors $\varphi_1\in  {\mathcal S}_{\pi,M}(\varrho)$ and $\varphi_0\in  {\mathcal S}_{\pi,M}(\nu)$
and obeying $d_B(M|\nu,\omega)=\|\varphi_0-\varphi\|_\pi$ and $d_B(M|\omega,\varrho)=\|\varphi-\varphi_1\|_\pi$.
Then, let $[0,1]\ni t\longmapsto \hat{\varphi}_t\in {\mathcal H}_\pi$ be the implementing family of  the geodesic arc $\hat{\gamma}$ extending from $\nu$ to $\omega$ and constructed  with the help of
$\varphi_0\in  {\mathcal S}_{\pi,M}(\nu)$ and $\varphi\in {\mathcal S}_{\pi,M}(\omega)$. Analogously, be  $[0,1]\ni t\longmapsto \tilde{\varphi}_t\in {\mathcal H}_\pi$ the implementing family of
the geodesic arc $\tilde{\gamma}$ extending from $\omega$ to $\varrho$ and constructed by means of
$\varphi\in  {\mathcal S}_{\pi,M}(\omega)$ and $\varphi_1\in {\mathcal S}_{\pi,M}(\varrho)$.
Define an implementation $[0,1]\ni t\longmapsto \varphi_t\in {\mathcal H}_\pi$ for a curve $\gamma: [0,1]\ni t\longmapsto \nu_t$ extending from $\nu_0=\nu$ over $\nu_{1/2}=\omega$ to $\nu_1=\varrho$ as follows
\begin{align}\label{famdef}
\varphi_t & =
\begin{cases}
   \hat{\varphi}_{2t}& \text{for $0\leq t\leq \frac{1}{2}$}\\
   &\\
    \tilde{\varphi}_{2t-1}&\text{for $\frac{1}{2}\leq t\leq 1$}
\end{cases}
\end{align}
Note that $(\varphi_t)$ is satisfying $h^\pi_{\varphi_t,\varphi_s}\geq 0$ on $\pi(M)^{\,\prime}$ provided $t,s\in [0,0.5]$ or $t,s\in [0.5,1]$. By Theorem \ref{bas5}\,\eqref{bas53}, for the Bures length of $\gamma$ we may conclude as follows
\begin{eqnarray*}
  \varTheta[\gamma] &=& \varTheta[\hat{\gamma}]+\varTheta[\tilde{\gamma}] \\
   &=& \lim_{\tau'} \sum_{j=0}^{n_{\tau'}} d_B(M|\nu_{{r}_j},\nu_{{r}_{j+1}}) + \lim_{\tau''} \sum_{j=0}^{n_{\tau''}} d_B(M|\nu_{s_j},\nu_{s_{j+1}})\\
   &=& \lim_{\tau'} \sum_{j=0}^{n_{\tau'}} \|\varphi_{r_j} -\varphi_{r_{j+1}}\|_\pi + \lim_{\tau''} \sum_{j=0}^{n_{\tau''}} \|\varphi_{s_j} -\varphi_{s_{j+1}}\|_\pi\\
   & = & \lim_{\tau_0} \sum_{j=0}^{n_{\tau_0}} \|\varphi_{t_j} -\varphi_{t_{j+1}}\|_\pi
\end{eqnarray*}
Thereby, whereas $\tau'=\{r_0,r_1,\ldots,r_{n+1}\}$ with $n=n_{\tau'}\in {{\mathbb{N}}}\cap \{0\}$ and $0=r_0<r_1<\cdots <r_n<r_{n+1}=0.5$ is running through the ascending net of the finite partitions of the interval $[0,0.5]$, and  $\tau''=\{s_0,s_1,\ldots,s_{n+1}\}$ with $n=n_{\tau''}\in {{\mathbb{N}}}\cap \{0\}$ and $0.5=s_0<s_1<\cdots <s_n<s_{n+1}=1$ is running through the directed net of the finite partitions of the interval $[0.5,1]$, the variable $\tau_0$ is thought to run through the ascending net of those finite partitions of $[0,1]$ containing $0.5$. Now, if a finite partition $\tau$ of the unit interval is given, then $\tau_0=\tau\cup \{0.5\}$ is a finite partition with $\tau\subset \tau_0$ and therefore according to the triangle inequality for $\|\cdot\|_\pi$ we find that
 $$\sum_{j=0}^{n_{\tau}} \|\varphi_{t_j} -\varphi_{t_{j+1}}\|_\pi\leq \sum_{j=0}^{n_{\tau_0}} \|\varphi_{t_j} -\varphi_{t_{j+1}}\|_\pi$$ has to hold. Hence, the conclusion is that
 \begin{equation*}
    \lim_{\tau} \sum_{j=0}^{n_{\tau}} \|\varphi_{t_j} -\varphi_{t_{j+1}}\|_\pi=\lim_{\tau_0} \sum_{j=0}^{n_{\tau_0}} \|\varphi_{t_j} -\varphi_{t_{j+1}}\|_\pi
 \end{equation*}
 has to be fulfilled, with $\tau$ running through the ascending net of all finite partitions of the unit interval.
 Thus, in view of the above we see that the following is fulfilled
 \begin{equation*}
    \varTheta[\gamma]=\varTheta[\hat{\gamma}]+\varTheta[\tilde{\gamma}]=\lim_{\tau} \sum_{j=0}^{n_{\tau}} \|\varphi_{t_j} -\varphi_{t_{j+1}}\|_\pi
 \end{equation*}
 with continuous map $[0,1]\ni t\longmapsto \varphi_t\in {\mathcal H}_\pi$, with boundary conditions $\varphi_0\in  {\mathcal S}_{\pi,M}(\nu)$, $\varphi_1\in  {\mathcal S}_{\pi,M}(\varrho)$. Application of Lemma \ref{kugel} under these conditions yields
 \begin{equation*}
    \arcsin{\sqrt{1-|\langle\varphi_0,\varphi_1\rangle_\pi|^2}}\leq
\lim_\tau {\sum_{j=0}^{n_\tau}
{\|\varphi_{t_j}-\varphi_{t_{j+1}}\|_\pi}}=\varTheta[\hat{\gamma}]+\varTheta[\tilde{\gamma}]
 \end{equation*}
By $P(M|\nu,\varrho)\geq |\langle\varphi_0,\varphi_1\rangle_\pi|^2$ then
$\arcsin{\sqrt{1-P(M|\nu,\varrho)}}\leq\arcsin{\sqrt{1-|\langle\varphi_0,\varphi_1\rangle_\pi|^2}}$ follows.
Hence, in view of the previous and \eqref{geodist} the conclusion is
\begin{equation*}
{\mathfrak d}_B(M|\nu,\varrho)=\arcsin{\sqrt{1-P(M|\nu,\varrho)}}\leq \varTheta[\hat{\gamma}]+\varTheta[\tilde{\gamma}]
\end{equation*}
Thus, since  ${\mathfrak d}_B(M|\nu,\omega)=\varTheta[\hat{\gamma}]$ and ${\mathfrak d}_B(M|\omega,\varrho)=\varTheta[\tilde{\gamma}]$ are fulfilled, the validity of formula \eqref{drei} is seen provided both $\nu\dashv\omega$ and $\varrho\dashv \omega$ are supposed. To release from the latter condition, let $\nu,\varrho,\omega$ be any triplet of states. For $n\in {{\mathbb{N}}}$, define $\omega_n$ by
\begin{equation*}
    \omega_n=\frac{n}{n+1}\,\omega+\frac{1}{n+1}\,\biggl(\frac{\nu+\varrho}{2}\biggr)
\end{equation*}
Then, in the sense of domination of positive linear forms, see Remark \ref{saka}\,\eqref{saka1}, we have $\nu\ll \omega_n$ and $\varrho\ll \omega_n$, and which implies $\nu\dashv \omega_n$ and $\varrho\dashv \omega_n$, for any $n$. By the just proven, for any $n\in {{\mathbb{N}}}$ we therefore have
\begin{equation}\label{appdrei}
   {\mathfrak d}_B(M|\nu,\varrho)\leq {\mathfrak d}_B(M|\nu,\omega_n)+{\mathfrak d}_B(M|\omega_n,\varrho)
\end{equation}
 Also, note that $\lim_n \omega_n=\omega$ is fulfilled. Now, in line with Lemma \ref{bas3}\,\eqref{bas3b}, let  positive, invertible $a_n\in M_+$ be chosen such that $$\nu(a_n)\,\omega_n(a_n^{-1})\leq P(M|\nu,\omega_n)+\frac{1}{n}$$
 From  this and $\frac{n}{n+1}\,\omega\leq  \omega_n$ together with Lemma \ref{bas3}\,\eqref{bas3b} and \eqref{pcont.1aa}, we infer that
\begin{equation*}
  \frac{n}{n+1}\, P(M|\nu,\omega)  \leq \frac{n}{n+1}\,\nu(a_n)\,\omega(a_n^{-1})\leq \nu(a_n)\,\omega_n(a_n^{-1})\leq P(M|\nu,\omega_n)+\frac{1}{n}
\end{equation*}
Hence, $P(M|\nu,\omega)  \leq \liminf_n P(M|\nu,\omega_n)$ follows. By upper semi-continuity of $P(M|\cdot,\cdot)$, see  \eqref{gemhalbstet} and \eqref{pcont.1aa}, then $P(M|\nu,\omega)=\lim_n P(M|\nu,\omega_n)$ is seen.\\
Analogously, $P(M|\omega,\varrho)=\lim_n P(M|\omega_n,\varrho)$ holds. By definition in  \eqref{geodist} from this
${\mathfrak d}_B(M|\nu,\omega)=\lim_n {\mathfrak d}_B(M|\nu,\omega_n)$ and ${\mathfrak d}_B(M|\omega,\varrho)=\lim_n {\mathfrak d}_B(M|\omega_n,\varrho)$ are implied. Now, on going to the limit for $n\to \infty$ in \eqref{appdrei}, the triangle inequality in respect of the triplet $\nu,\omega,\varrho$ will follow.
\end{proof}
Consider now a parameterized continuous curve ${\mathcal{C}}:[0,1]\ni t\longmapsto \nu_t\in {\mathcal S}(M)$. Then, ${\mathcal{C}}$ is said to connect the states $\nu$ and $\varrho$ if $\nu_0=\nu$  and $\nu_1=\varrho$, and ${\mathcal{C}}$ is said to be implementable, if there exists  a unital $^*$-representation $\{\pi,{\mathcal H}_\pi\}$ and a continuous map $[0,1]\ni t\longmapsto \varphi_t\in {\mathcal H}_\pi$ such that $\varphi_t\in {\mathcal S}_{\pi,M}(\nu_t)$, for any $t$.

Since by Lemma \ref{geoarclength} the value ${\mathfrak d}_B(M|\nu,\varrho)$ appears as the Bures length of a parameterized curve connecting $\nu$ and $\varrho$, the following lower estimate must hold:
\begin{subequations}\label{arcest}
\begin{equation}\label{lowerarc}
  {d}_B(M|\nu,\varrho)\leq {\mathfrak d}_B(M|\nu,\varrho)
\end{equation}
Now, suppose $\nu,\varrho$ are states obeying $\|\nu-\varrho\|_1\leq 1$. Then, owing to Corollary  \ref{uppdiff} $d_B(M|\nu,\varrho)\leq \sqrt{\|\nu-\varrho\|_1}$ and by formula \eqref{metzu}, $\sqrt{1-P(M|\nu,\varrho)}\leq d_B(M|\nu,\varrho)\leq 1$
can be inferred to hold. Hence, by monotony of the $\arcsin$-function one has
\[
 {\mathfrak d}_B(M|\nu,\varrho)=\arcsin \sqrt{1-P(M|\nu,\varrho)}\leq \arcsin d_B(M|\nu,\varrho)
\]
Combined with the elementary estimate $\arcsin x\leq x(1 + x^2)$, for $0\leq x\leq 1$, in view of formula \eqref{uppdiff1} for $x=d_B(M|\nu,\varrho)$ and under the above suppositions we arrive at
\begin{equation}\label{metzu1}
   {\mathfrak d}_B(M|\nu,\varrho)\leq d_B(M|\nu,\varrho)\bigl(1+\|\nu-\varrho\|_1\bigr)
\end{equation}
\end{subequations}
Now we are going to apply estimates of type \eqref{arcest} to a continuous curve  ${\mathcal{C}}:[0,1]\ni t\longmapsto \nu_t$  of  length $\varTheta[{\mathcal{C}}]$ and connecting the states $\nu$ and $\varrho$ in the state space of a unital ${\mathsf C}^*$-algebra $M$. Let $1\geq \varepsilon>0$, and be $\tau_\varepsilon$ a finite partition $\{t^\varepsilon_0,t^\varepsilon_1,\ldots,t^\varepsilon_n,t^\varepsilon_{n+1}\}$ of $[0,1]$, with $n=n_{\tau_\varepsilon}\in {{\mathbb{N}}}\cup \{0\}$ and  $0=t^\varepsilon_0<t^\varepsilon_1<\cdots< t^\varepsilon_n<t^\varepsilon_{n+1}=1$,
such that, for all $0\leq i\leq n_{\tau_\varepsilon}$,  $$\|\nu_{t^\varepsilon_i}-\nu_{t^\varepsilon_{i+1}}\|_1\leq \varepsilon$$
 As a continuous map  over a compact interval, the implementing family $(\nu_t)$ is uniformly continuous, and thus such a partition $\tau_\varepsilon$ has to exist. Hence, \eqref{arcest} applies to all pairs $\{\nu_{t_i},\nu_{t_{i+1}}\}$ of each partition $\tau$ of $[0,1]$ obeying $\tau\supset\tau_\varepsilon$.  The result is
\[
d_B(M|\nu_{t_i},\nu_{t_{i+1}})\leq  {\mathfrak d}_B(M|\nu_{t_i},\nu_{t_{i+1}})\leq d_B(M|\nu_{t_i},\nu_{t_{i+1}})\bigl(1+\varepsilon\bigr)
\]
for $0\leq i\leq n_\tau$. By summing up from this we get
\[
\sum_{j=0}^{n_{\tau}}d_B(M|\nu_{t_j},\nu_{t_{j+1}})\leq \sum_{j=0}^{n_{\tau}} {\mathfrak d}_B(M|\nu_{t_j},\nu_{t_{j+1}})\leq \Biggl(\sum_{j=0}^{n_{\tau}}d_B(M|\nu_{t_j},\nu_{t_{j+1}})\Biggr)\bigl(1+\varepsilon\bigr)
\]
Now, by definition of a length, obviously
\[
\varTheta[{\mathcal{C}}]=\lim_\tau \sum_{j=0}^{n_\tau} { d}_B(M|\nu_{t_j},\nu_{t_{j+1}})=\lim_{\tau\supset\tau_\varepsilon} \sum_{j=0}^{n_\tau} { d}_B(M|\nu_{t_j},\nu_{t_{j+1}})
\]
Also, by Lemma \ref{dreieck}, since ${\mathfrak d}_B$ is a metric distance, similarly we have existence of
\[
\lim_\tau \sum_{j=0}^{n_\tau} {\mathfrak d}_B(M|\nu_{t_j},\nu_{t_{j+1}})=\lim_{\tau\supset\tau_\varepsilon} \sum_{j=0}^{n_\tau} {\mathfrak d}_B(M|\nu_{t_j},\nu_{t_{j+1}})
\]
and which in view of the above is satisfying the following estimates
\[
\varTheta[{\mathcal{C}}]\leq \lim_\tau \sum_{j=0}^{n_\tau} {\mathfrak d}_B(M|\nu_{t_j},\nu_{t_{j+1}})\leq \varTheta[{\mathcal{C}}]\bigl(1+\varepsilon\bigr)
\]
Since the latter has to be true with $\varepsilon>0$ arbitrarily small, from this the validity of
\begin{equation}\label{summe}
 \varTheta[{\mathcal{C}}]=\lim_\tau \sum_{j=0}^{n_\tau} {\mathfrak d}_B(M|\nu_{t_j},\nu_{t_{j+1}})
\end{equation}
can be inferred to hold, where $\tau$ is running through the ascendingly directed net of the finite partitions of the unit interval.
The formula \eqref{summe} is the key tool for proving an important fact about continuous curves of states connecting $\nu$ and $\varrho$.
\begin{theorem}\label{innergeodist}
\begin{subequations}\label{geoabstand}
Let ${\mathcal{C}}$ be a parameterized continuous curve of Bures length $\varTheta[{\mathcal{C}}]$ and  connecting the states $\nu$ and $\varrho$ in the state space of a unital ${\mathsf C}^*$-algebra $M$. Then
\begin{equation}\label{mindist}
   {\mathfrak d}_B(M|\nu,\varrho)\leq \varTheta[{\mathcal{C}}]
\end{equation}
Especially, each geodesic arc $\gamma\in {\mathcal{C}}^{\nu,\varrho}(M)$ is an implementable curve of minimal Bures length and connecting $\nu$ and $\varrho$. Moreover, a continuous curve ${\mathcal{C}}:[0,1]\ni t\longmapsto \nu_t\in {\mathcal S}(M)$ connecting $\nu$ and $\varrho$ is of minimal Bures length if, and only if
\begin{equation}\label{geocon}
 {\mathfrak d}_B(M|\nu,\varrho)=\sum_{j=0}^{n_\tau} {\mathfrak d}_B(M|\nu_{t_j},\nu_{t_{j+1}})
\end{equation}
\end{subequations}
for any finite partition $\tau=\{t_0,t_1,\ldots,t_n,t_{n+1}\}$ of the unit interval, with $n=n_{\tau}\in {{\mathbb{N}}}\cup \{0\}$ and  $0=t_0<t_1<\cdots< t_n<t_{n+1}=1$.
\end{theorem}
\begin{proof}
By formula \eqref{summe} and in view of Lemma \ref{dreieck} we conclude that
\[
{\mathfrak d}_B(M|\nu,\varrho)\leq\sum_{j=0}^{n_{\tau}} {\mathfrak d}_B(M|\nu_{t_i},\nu_{t_{i+1}})\leq \lim_\tau \sum_{j=0}^{n_\tau} {\mathfrak d}_B(M|\nu_{t_j},\nu_{t_{j+1}})= \varTheta[{\mathcal{C}}]
\]
and which is \eqref{mindist}. Thus, ${\mathfrak d}_B(M|\nu,\varrho)$ is the minimal Bures length a continuous curve connecting $\nu$ and $\varrho$ can have. Moreover, if $\varTheta[{\mathcal{C}}]={\mathfrak d}_B(M|\nu,\varrho)$ is fulfilled for such a curve, according to formula \eqref{summe} and the triangle inequality for ${\mathfrak d}_B$
\[
{\mathfrak d}_B(M|\nu,\varrho)=\lim_\tau \sum_{j=0}^{n_\tau} {\mathfrak d}_B(M|\nu_{t_j},\nu_{t_{j+1}})\geq \sum_{j=0}^{n_{\tau}} {\mathfrak d}_B(M|\nu_{t_i},\nu_{t_{i+1}})\geq {\mathfrak d}_B(M|\nu,\varrho)
\]
follows. That is, \eqref{geocon} has to hold, for each finite partition $\tau$ of the unit interval.
On the other hand, if the latter is supposed to hold for any finite partition $\tau$, then in taking the net-limit
over $\tau$ within \eqref{geocon} will yield that also
\[
{\mathfrak d}_B(M|\nu,\varrho)=\lim_\tau \sum_{j=0}^{n_\tau} {\mathfrak d}_B(M|\nu_{t_j},\nu_{t_{j+1}})
\]
is fulfilled. By formula \eqref{summe} this is equivalent to ${\mathfrak d}_B(M|\nu,\varrho)=\varTheta[{\mathcal{C}}]$ and which in view of \eqref{mindist} means that ${\mathcal{C}}$ is of minimal Bures length between $\nu$ and $\varrho$.
\end{proof}
\begin{remark}\label{allg2c}
\begin{enumerate}
 \item\label{allg2c0a}
By Theorem \ref{innergeodist}, as a metric space $\{{\mathcal S}(M),d_B(M|\cdot,\cdot)\}$ has the property that, for any $\nu,\varrho\in {\mathcal S}(M)$, the following is fulfilled:
\begin{equation}\label{lestr}
{\mathfrak d}_B(M|\nu,\varrho)=\inf_{\mathcal{C}} \varTheta[{\mathcal{C}}]
\end{equation}
Thereby, $\varTheta[{\mathcal{C}}]$ is the Bures-length of the curve ${\mathcal{C}}$, and the infimum extends over all Bures-rectifiable continuous curves connecting $\nu$ with $\varrho$.
Accordingly, ${\mathfrak d}_B(M|\cdot,\cdot)$ is the geodesic distance or `inner metric' coming along with the Bures distance. For generalities see \cite[\S 15.]{Rino:61}. Under the metric distance ${\mathfrak d}_B$ the state space becomes a so-called `length space', with length structure on  ${\mathcal S}(M)$ given by the pair $\{d_B(M|\cdot,\cdot),{\mathfrak d}_B(M|\cdot,\cdot)\}$.
   \item\label{allg2c1} Relating Theorem \ref{innergeodist},  for vector states $\nu, \varrho$ on a $vN$-algebra,  \eqref{mindist} and the fact  that each $\gamma\in {\mathcal{C}}^{\nu,\varrho}(M)$ under all implementable curves extending from $\nu$ to $\varrho$ is a curve of minimal Bures length, are known \cite{Albe:92.2, Albe:96.1, Pelt:00.1a}.
  \item\label{allg2b3} In special cases, the length structure  $\{d_B(M|\cdot,\cdot),{\mathfrak d}_B(M|\cdot,\cdot)\}$ on the state space ${\mathcal S}(M)$ is well-known. Especially, on a commutative ${\mathsf W}^*$-algebra $M\simeq L^\infty(X,m)$, $(X,m)$ a localizable measure space, with operators acting by multiplication on elements of ${\mathcal H}_\pi=L^2(X,m)$, if $\nu$ and $\varrho$ are normal states, then we will find that
\[
d_B(M|\nu,\varrho)=\biggl\|\sqrt{\frac{d\nu}{dm}}-\sqrt{\frac{d\varrho}{dm}}\,\biggr\|_2
\]
with $d\nu/dm$ and $d\varrho/dm$ being the absolutely continuous to $m$ probability distributions which are corresponding to  $\nu$ and $\varrho$, respectively.
In this case
the pair $\{d_B(M|\cdot,\cdot),{\mathfrak d}_B(M|\cdot,\cdot)\}$ of distance functions up to a factor is known from statistics as Hellinger-distance between probability distributions and the corresponding geodesic distance which by \eqref{geodist} then reads
\[
{\mathfrak d}_B(M|\nu,\varrho)=\arccos{F(M|\nu,\varrho)}=\arccos{\int_X d\/m\,\sqrt{\left(\frac{d\nu}{dm}\right)\left(\frac{d\varrho}{dm}\right)}}
\]
see e.g.~\cite{Chen:72} and \cite{Ragg:84}.
 \item\label{allg2b4}
In the
special case $M\simeq {\mathsf B}({\mathcal H})$\,
the r{\^o}le of the expression
$\arccos{{F(M|\nu,\varrho)}}$ of \eqref{geodist} as a geodesic distance analogous to the Hellinger-distance  case has been yet mentioned
in \cite[p.199, footnote]{Khol:76}, see also \cite{Khol:82}.
\end{enumerate}
\end{remark}
Close with an auxiliary result telling what happens if equality occurs in the triangle inequality for the inner metric distance ${\mathfrak d}_B$, under certain circumstances.
\begin{lemma}\label{besonders}
Let $\varrho\not=\nu$ be states, and be $\omega$ another state such that the relations
\begin{equation}\label{besonders0}
{\mathfrak d}_B(M|\nu,\omega)=\lambda\,{\mathfrak d}_B(M|\nu,\varrho),\ {\mathfrak d}_B(M|\omega,\varrho)=(1-\lambda)\,{\mathfrak d}_B(M|\nu,\varrho)
\end{equation}
hold, with some real  $\lambda\in ]0,1[$. Then, $\omega\in |\gamma|$ for a geodesic arc $\gamma\in {\mathcal{C}}^{\nu,\varrho}(M)$.
\end{lemma}
\begin{proof}
By Lemma \ref{king} there exists a unital $^*$-representation $\{\pi,{\mathcal H}_\pi\}$ such that
\[
{\mathcal S}_{\pi,M}(\omega|\nu)\cap{\mathcal S}_{\pi,M}(\omega|\varrho)\not=\emptyset
\]
In view of Definition \ref{twice1} the latter  means that $\xi\in {\mathcal S}_{\pi,M}(\omega)$ exists such that, for some $\zeta\in {\mathcal S}_{\pi,M}(\varrho)$ and  $\varphi\in  {\mathcal S}_{\pi,M}(\nu)$ one has
\begin{subequations}\label{redf0}
\begin{equation}\label{bes}
\langle \zeta,\xi\rangle_\pi=F(M|\omega,\varrho)=F_1,\, \langle \varphi,\xi\rangle_\pi=F(M|\omega,\nu)=F_2
\end{equation}
Let $P=F^2$, $P_1=F_1^2$ and $P_2=F_2^2$ be the transition probabilities  corresponding to the pairs $\{\nu,\varrho\}$,  $\{\omega,\varrho\}$ and $\{\nu,\omega\}$, respectively.  Also, note that \eqref{besonders0} amounts to
\[
{\mathfrak d}_B(M|\nu,\omega)+{\mathfrak d}_B(M|\omega,\varrho)={\mathfrak d}_B(M|\nu,\varrho)
\]
Then, in view of \eqref{geodist} by trigonometric calculus one easily infers that
 \begin{equation}\label{redf}
 F=F_1 F_2 -\sqrt{1-P_1} \sqrt{1-P_2}
 \end{equation}
 \end{subequations}
 Let us define $\varphi_f=\varphi-F_2\xi$ and $\zeta_f=\zeta-F_1\xi$. We have $\varphi=\varphi_f+F_2\xi$, $\zeta=\zeta_f+F_1\xi$,
and due to \eqref{bes} both $\varphi_f$ and $\zeta_f$ are orthogonal to $\xi$. Hence, we get
\begin{equation}\label{redf2}
\langle \zeta,\varphi\rangle_\pi=F_1 F_2+\langle\varphi_f,\zeta_f\rangle_\pi
\end{equation}
and see that $\|\zeta_f\|=\sqrt{1-P_1}$, $\|\varphi_f\|=\sqrt{1-P_2}$ with
Cauchy-Schwarz inequality
\begin{equation}\label{redf3}
|\langle \varphi_f,\zeta_f\rangle_\pi|\leq \| \varphi_f\|_\pi \|\zeta_f\|_\pi=\sqrt{1-P_1}\sqrt{1-P_2}
\end{equation}
From \eqref{redf2}, \eqref{redf3} and \eqref{redf} we may conclude as follows:
\[
|\langle \zeta,\varphi\rangle_\pi|\geq F_1 F_2-|\langle\varphi_f,\zeta_f\rangle_\pi|\geq F_1 F_2 -\sqrt{1-P_1} \sqrt{1-P_2}=F
\]
Since owing to $\varphi\in  {\mathcal S}_{\pi,M}(\nu),\,\zeta\in {\mathcal S}_{\pi,M}(\varrho)$ the estimate
$
|\langle \zeta,\varphi\rangle_\pi|\leq F
$ has to be fulfilled, the conclusion from the previous is
that
\begin{equation}\label{redf3a}
|\langle \zeta,\varphi\rangle_\pi|= F_1 F_2-|\langle\varphi_f,\zeta_f\rangle_\pi|=F_1 F_2 -\sqrt{1-P_1} \sqrt{1-P_2}=F
\end{equation}
Hence, in \eqref{redf3} equality has to occur: $$|\langle\varphi_f,\zeta_f\rangle_\pi|=\| \varphi_f\|_\pi \|\zeta_f\|_\pi$$
Thus, since both vectors involved are non-zero, there has to exist  $\mu\in {\mathbb C}\backslash \{0\}$ with
\begin{equation}\label{redf1}
|\mu|=\sqrt{\frac{1-P_1}{1-P_2}}
\end{equation}
such that $\zeta_f=\mu\,\varphi_f$, and therefore, for $\mu$ obeying \eqref{redf1} we get
\begin{equation}\label{refd4}
\zeta=\mu\varphi+(F_1-\mu F_2)\,\xi
\end{equation}
Let $\mu=\exp {\mathsf{i}}\alpha\,|\mu|$, $\alpha\in {\mathbb R}$. Also, in view of \eqref{bes} we may conclude that
\begin{eqnarray*}
\langle \zeta,\varphi\rangle_\pi & =& \mu+(F_1-\mu F_2) F_2=F_1 F_2+\mu(1-P_2)\\
& = & F_1 F_2+\exp {\mathsf{i}}\alpha\,\sqrt{1-P_1} \sqrt{1-P_2}=\\
& = & F+\bigl(1+\exp {\mathsf{i}}\alpha\bigr)\sqrt{1-P_1} \sqrt{1-P_2}
\end{eqnarray*}
Thus, from this and \eqref{redf3a} the conclusion is
\begin{eqnarray*}
F^2 & = & |\langle \zeta,\varphi\rangle_\pi|^2=\bigl|F+(1+\exp {\mathsf{i}}\alpha)\sqrt{1-P_1} \sqrt{1-P_2}\,\bigr|^2\\
& =& F^2+2(1+\cos \alpha)\bigl(F\,\sqrt{1-P_1} \sqrt{1-P_2}+(1-P_1) (1-P_2)\bigr)
\end{eqnarray*}
Hence, $\cos \alpha=-1$ has to be fulfilled. But then, in view of the previous $\zeta\in {\mathcal S}_{\pi,M}(\varrho)$, $\varphi\in {\mathcal S}_{\pi,M}(\nu)$ obey $\langle\zeta,\varphi\rangle_\pi=F$. Hence,  in view of Lemma \ref{bas3} and Theorem \ref{bas5},
\begin{subequations}\label{refd4a}
\begin{equation}\label{refd4b}
h^\pi_{\zeta,\varphi}\geq {\mathsf 0}
\end{equation}
has to be fulfilled. On the other hand, since $\mu=-|\mu|$ holds, in view of \eqref{refd4}
\begin{equation}\label{refd5}
\xi=\frac{\zeta+|\mu|\varphi}{(F_1+|\mu|F_2)}
\end{equation}
\end{subequations}
is fulfilled. That is, $\xi$ is arising as a linear combination with positive coefficients of the implementing vectors $\zeta$ and $\varphi$ and which are obeying \eqref{refd4b}. Accordingly, if the construction schema \eqref{geo0} is considered, then \eqref{refd5}  is showing that $\xi=\varphi_s$ is the implementing vector of the state $\nu_s=\omega$ of a particular geodesic arc $\gamma : [0,1]\ni t\longmapsto\nu_t$  connecting $\nu$ and $\varrho$, with the  parameter value $t=s$ being the (unique) solution of ${\mathfrak d}_B(M|\nu,\nu_t)=\lambda\,{\mathfrak d}_B(M|\nu,\varrho)$ in $[0,1]$, for this see Definition \ref{allg2c0}, too.
\end{proof}
\begin{remark}\label{refd6}
An often occurring case of Lemma \ref{besonders} is $\lambda=1/2$. In this case  $\omega$ is asserted to be  the midpoint state of some geodesic arc connecting $\nu$ and $\varrho$.
\end{remark}
\subsubsection{The geodesic structure of the state space}\label{geoconstr}
Suppose the continuous map
\begin{subequations}\label{Bpara}
\begin{equation}\label{Bpara-1}
{\mathcal{C}}: I \ni t\longmapsto \omega_t\in {\mathcal S}(M)
\end{equation}
with compact interval $I\subset {\mathbb R}$, is a Bures-rectifiable curve ${\mathcal{C}}$ in the state space, with Bures length $\theta_0=\varTheta[\mathcal C]$. For simplicity we may assume that $\min_{t\in I} t=0$. Then, according to Lemma \ref{rinow0}, by means of the continuous invertible  function
 \begin{equation}\label{Bpara0}
I \ni t\longmapsto \theta(t)\in \left[ 0,\theta_0\right]
\end{equation}
which is given by
\begin{equation}\label{Bpara1}
\theta(t)=\int_0^t {\mathop{\mathrm{dil}}}_t^B {\mathcal{C}}\ d\/t
\end{equation}
a re-parameterization of the curve ${\mathcal{C}}$ can be achieved, which is reading in terms of the Bures length $\theta(t)$ counted along the curve ${\mathcal{C}}$ when evolving from the starting state $\omega_0=\nu$ to the state $\omega_t$ up to which a Bures length of $\theta=\theta(t)$ has been evolved. Let $t(\theta)$ be the inverse to \eqref{Bpara0} function. For the  re-parameterization of \eqref{Bpara-1} when reading in terms of the Bures length  $\theta$ along ${\mathcal{C}}$ the notation
\begin{equation}\label{Bpara3}
{\mathcal{C}}: \left[ 0,\theta_0\right]\ni\theta\longmapsto \omega_\theta
\end{equation}
will be used, with $\omega_\theta$ taken as an abbreviating notation of the indirect function  $\omega_{t(\theta)}$ arising in line of the  parameterization  \eqref{Bpara-1} originally used.  Note that parameterizations like \eqref{Bpara3} are distinguished by the property that for $\theta_1<\theta_2$
\begin{equation}\label{Bpara4}
\varTheta\left[{\mathcal{C}}|[ \theta_1,\theta_2] \right] = \theta_2-\theta_1
\end{equation}
\end{subequations}
As usually, in the case of occurance this (up to the orientation uniquely determined) parameterization reading in terms of the natural length-parameter and satisfying \eqref{Bpara4} will be referred to as `normal parameterization' of ${\mathcal{C}}$.

In context of Theorem \ref{innergeodist},  in a twofold manner uniqueness problems relating  continuous curves of minimal Bures length and connecting two given states arise. Here comes the main result about this subject and saying, that geodesic arcs are the only examples possible.
\begin{theorem}\label{gearcuni}
Suppose $\varrho,\nu\in {\mathcal S}(M)$, $\varrho\not=\nu$, and let ${\mathcal{C}}\subset {\mathcal S}(M)$ be any continuous curve of minimal Bures length evolving from $\nu$ to $\varrho$. Then, the following holds:
\begin{enumerate}
\item \label{gac1}
 $|{\mathcal{C}}|=|\gamma|$, for some $\gamma\in {\mathcal{C}}^{\nu,\varrho}(M)$.
\item \label{gac2}
$|{\mathcal{C}}|$ is uniquely determined if, and only if, $\varrho_\nu$ and $\nu_\varrho$ are mutually disjoint.
\end{enumerate}
Especially, both $\gamma$ and ${\mathcal{C}}$ then possess identical normal parameterizations.
\end{theorem}
\begin{proof}
In context of \eqref{gac1} note that by Theorem \ref{innergeodist}  each geodesic arc $\gamma\in {\mathcal{C}}^{\nu,\varrho}(M)$ is a special continuous curve of minimal Bures length connecting $\nu$ and $\varrho$. Also, provided \eqref{gac1} holds, then validity of \eqref{gac2} follows along with Theorem \ref{einga}.  It remains to be shown that if ${\mathcal{C}}$ is any continuous curve of minimal Bures length  connecting $\nu$ and $\varrho$, then up to the parameterization  ${\mathcal{C}}$ equals some $\gamma\in {\mathcal{C}}^{\nu,\varrho}(M)$. This will be the case if ${\mathcal{C}}$ and $\gamma$  can be shown to have the same normal parameterization. This we are going to do now.

By Theorem \ref{innergeodist} and  preliminary considerations we know that the Bures length of ${\mathcal{C}}$ reads    $\varTheta({\mathcal{C}})={\mathfrak d}_B(M|\nu,\varrho)$. We may assume the normal parameterization \eqref{Bpara} $${\mathcal{C}}: [0,\theta_0]\ni \theta\longmapsto \omega_\theta\in S(M)$$  to be given, with  $\theta_0=\varTheta({\mathcal{C}})={\mathfrak d}_B(M|\nu,\varrho)$ and with the parameter $\theta$ with $0\leq \theta\leq \theta_0$ having the meaning of the Bures length counted from the starting state $\omega_0=\nu$ along the curve to the state $\omega_\theta$ up to which a Bures length of $\theta$ has been evolved. Note that the  final state of ${\mathcal{C}}$ is $\varrho=\omega_{\theta_0}$.  What will be shown is that the midpoint $\omega=\omega_{\theta_0/2}$ of ${\mathcal{C}}$ equals the midpoint $\nu_{\theta_0/2}$ of some  geodesic arc $\gamma\in {\mathcal{C}}^{\nu,\varrho}(M)$. By Lemma \ref{geoarclength} the normal parameterization of this $\gamma$ can be written as
$$\gamma: [0,\theta_0]\ni \theta\longmapsto \nu_\theta\in S(M)$$
and will be shown to satisfy the needs of \eqref{gac1}.
In fact, by the assumptions on $\omega$ and since ${\mathcal{C}}$ is of minimal Bures length,
\[
{\mathfrak d}_B(M|\nu,\omega)={\mathfrak d}_B(M|\omega,\varrho)
\]
has to be fulfilled, and by
Theorem \ref{innergeodist} we have
\[
{\mathfrak d}_B(M|\nu,\omega)+{\mathfrak d}_B(M|\omega,\varrho)={\mathfrak d}_B(M|\nu,\varrho)
\]
Thus, the assumptions for an application of Lemma \ref{besonders} in the special version with $\lambda=1/2$ are given, see Remark \ref{refd6}. According to the latter, there is a geodesic arc $\gamma$  connecting $\nu$ with $\varrho$ and $\omega_{\theta_0/2}=\omega=\nu_{\theta_0/2}$, that is, $\omega$ is the midpoint of $\gamma$, too.
Now, to given $n\in {{\mathbb{N}}}$, let
$${\mathfrak B}_n=\left\{k\,\theta_0/2^n,\,0\leq k\leq 2^n\right\}$$
be the partition points of the binary equipartition of $n$-th order of $[0,\theta_0]$.  Then, by the above we are yet knowing that $\omega_\theta=\nu_\theta$ in fact is true for all $\theta\in {\mathfrak B}_1$.
Assume now we would know that $\omega_\theta=\nu_\theta$ were fulfilled for all $\theta\in {\mathfrak B}_n$. Then, let us consider those parts ${\mathcal{C}}_{n,k}$, $k=0,1,\ldots,2^n$, of the curve ${\mathcal{C}}$ connecting $\omega_{(k-1)\,\theta_0/2^n}$ with $\omega_{k\,\theta_0/2^n}$. Since ${\mathcal{C}}$ is connecting $\nu$ and $\omega$ with minimal Bures length, application of Theorem \ref{innergeodist}, especially according to formula \eqref{geocon} with $\tau={\mathfrak B}_n$,  will give that each ${\mathcal{C}}_{n,k}$ is a continuous curve of minimal Bures length $\theta_0/2^n$ and connecting $\omega_{(k-1)\,\theta_0/2^n}=\nu_{(k-1)\,\theta_0/2^n}$ and $\omega_{k\,\theta_0/2^n}=\nu_{k\,\theta_0/2^n}$. On the other hand, in view of Example \ref{geodisj2}\,\eqref{geodisj2a}, both $|\gamma|[0,\theta_0/2]|$ and $|\gamma|[\theta_0/2,\theta_0]|$ have to be arc-determining sets. Accordingly, especially also the geodesic arc ${\gamma}_{n,k}$ connecting the above mentioned two states has to be uniquely determined.
Hence, when applied to both the cases  ${\mathcal{C}}_{n,k}$ and $\gamma|[(k-1)\,\theta_0/2^{n},k\,\theta_0/2^n]$ of continuous curves which are connecting the same two states, by literally the same conclusions by means of  Lemma \ref{besonders} and which were used to show that $\omega_\theta=\nu_\theta$ on ${\mathfrak B}_1$ is true, both the midpoint $\omega_{(2k-1)\,\theta_0/2^{n+1}}$ of ${\mathcal{C}}_{n,k}$ and the midpoint $\nu_{(2k-1)\,\theta_0/2^{n+1}}$  of $\gamma|[(k-1)\,\theta_0/2^{n},k\,\theta_0/2^n]$ have to equal the midpoint of the (now uniquely determined)  geodesic arc ${\gamma}_{n,k}$ connecting  $\nu_{(k-1)\,\theta_0/2^n}$ and $\nu_{k\,\theta_0/2^n}$. Thus, by transitivity  we conclude that  $\omega_{(2k-1)\,\theta_0/2^{n+1}}=\nu_{(2k-1)\,\theta_0/2^{n+1}}$ has to be fulfilled, for all integers $k$ with $0\leq k\leq 2^n$. That is,  whenever for $n\in {{\mathbb{N}}}$ one has $\omega_\theta=\nu_\theta$ for all $\theta\in {\mathfrak B}_n$, then $\omega_\theta=\nu_\theta$ for all $\theta\in {\mathfrak B}_{n+1}$ follows. Since $\omega_\theta=\nu_\theta$ on ${\mathfrak B}_1$ holds, by the induction principle $\omega_\theta=\nu_\theta$ has to hold on all binary partition points  of $[0,\theta_0]$, that is, for all
$$\textstyle \theta\in \bigcup_n {\mathfrak B}_n$$
Thus, ${\mathcal{C}}$ and $\gamma$ have in common a dense set of points. By continuity then  $\omega_\theta=\nu_\theta$ for any $\theta\in [0,\theta_0]$ follows, that is, the normal parameterizations are the same.
\end{proof}
Within the above proof implicitly we have made access to  the  re-parameterization schema \eqref{Bpara} in context of curves of minimal Bures lenght. Now, in making use of Theorem \ref{gearcuni}, in this special case this procedure can be made more explicit. Suppose a curve ${\mathcal{C}}$ of minimal Bures length and connecting two given states $\omega_0=\nu$ and $\omega_1=\varrho$ is given by a continuous map  ${\mathcal{C}}: [0,1]\ni t\longmapsto \omega_t\in {\mathcal S}(M)$ . Let \begin{subequations}\label{para}
\begin{equation}\label{para1}
\theta_0={\mathfrak d}_B(M|\nu,\varrho)
\end{equation}
 The reparameterization $\theta:[0,1]\ni t\longmapsto \theta(t)\in [0,\theta_0]$ of \eqref{Bpara1} is given by
\begin{equation}\label{para2}
\theta(t)={\mathfrak d}_B(M|\nu,\omega_t)
\end{equation}
With the inverse function $t(\theta)$ of \eqref{para2}, the normal parameterization of ${\mathcal{C}}$ reads
\begin{equation}\label{para3}
{\mathcal{C}}: [0,\theta_0]\ni \theta\longmapsto \nu_\theta=\omega_{t(\theta)}
\end{equation}
\end{subequations}
More precisely, with respect to \eqref{para3} the following is fulfilled.
\begin{corolla}\label{normalpara}
There is a $^*$-representation $\{\pi,{\mathcal H}_\pi\}$ of $M$,  $\varphi\in {\mathcal S}_{\pi,M}(\nu|\varrho)$,  $\zeta\in {\mathcal S}_{\pi,M}(\varrho)$ with $F(M|\nu,\varrho)=\langle\zeta,\varphi\rangle_\pi$, with normal parameterization of ${\mathcal{C}}$ reading
\begin{equation}\label{para4}
\nu_\theta=\biggl(\frac{\sin \theta\phantom{_0}}{\sin \theta_0}\biggr)^2\,\varrho+\biggl(\frac{\sin (\theta_0-\theta)}{\sin \theta_0}\biggr)^2\,\nu+\frac{2\sin \theta \sin (\theta_0-\theta)}{\sin^2 \theta_0}\,\re f^\pi_{\zeta,\varphi}
\end{equation}
\end{corolla}
\begin{proof}
By Theorem \ref{gearcuni} there is $\gamma\in {\mathcal{C}}^{\nu,\varrho}(M)$ with $|{\mathcal{C}}|=|\gamma|$. Clearly, both $\gamma$ and ${\mathcal{C}}$ then have to possess the same normal parameterizations. It suffices to arrange for the normal parametrization of  $\gamma$. Let $\gamma: [0,1]\ni t\longmapsto\omega_t\in {\mathcal S}(M)$ be given by
\begin{equation}\label{para5}
\omega_t=t^2\,\varrho+\lambda(t)^2\nu+ 2\,t\lambda(t)\re f^\pi_{\zeta,\varphi}
\end{equation}
with some $\varphi\in {\mathcal S}_{\pi,M}(\nu|\varrho)$, $\zeta\in  {\mathcal S}_{\pi,M}(\varrho)$ satisfying  $F(M|\nu,\varrho)=\langle\zeta,\varphi\rangle_\pi$, see \eqref{einga1}, \eqref{geo0} and Definition \ref{allg2c0}. Also, in line with \eqref{geo2},  the reparameterization function $\theta(t)$ of $\gamma$ given in accordance with \eqref{Bpara1} and with $P=P(M|\nu,\varrho)$ reads
 \begin{eqnarray}\nonumber
\theta(t) & =& \int_0^t  {\mathop{\mathrm{dil}}}_s^B \gamma\, d\/s
 =  \int_0^t \sqrt{\frac{1-P}{1-s^2(1-P)}}\ d\/s\\
 \nonumber
 &=&\arcsin t\sqrt{1-P}
\end{eqnarray}
Especially, in view of \eqref{para1} and with $F=F(M|\nu,\varrho)$ we then have $$\theta_0=\arcsin \sqrt{1-P}=\arccos F$$ and thus in the case of the geodesic arc $\gamma$ the inverse to \eqref{para2} function is
\begin{subequations}\label{para6}
\begin{equation}\label{para6a}
t(\theta)=\frac{\sin \theta\phantom{_0}}{\sin \theta_0}
\end{equation}
Inserting the latter into the formula of $\lambda(t)$ (see \eqref{geo0a}) and using that $F=\cos \theta_0$ holds, in the result of some trigonometric calculations we arrive at the expression
\begin{equation}\label{para6b}
\lambda(t(\theta))=\frac{\sin \,(\theta_0-\theta)}{\sin \theta_0}
\end{equation}
\end{subequations}
Applying formulae \eqref{para6} to  formula \eqref{para5} in view of \eqref{para3} then yields  \eqref{para4}.
\end{proof}
A closed Bures rectifiable  continuous curve ${\mathcal{C}}$ extending with minimal Bures length between its boundary states/endpoints will be termed `shortest path'.
We are going to discuss how the supporting sets of any two shortest paths can be correlated with each other. Assume that ${\mathcal{C}}$,  $\hat{{\mathcal{C}}}$ and $\hat{{\mathcal{C}}}^{\,\prime}$ can stand for shortest paths.
\begin{lemma}\label{bifurc}
For any two shortest paths ${\mathcal{C}}$ and $\hat{{\mathcal{C}}}$ the following can happen:
\begin{itemize}
\item[(a)] \label{bif1}
$|{\mathcal{C}}|\cap |\hat{{\mathcal{C}}}|=\emptyset$;
\item[(b)] \label{bif2}
$|{\mathcal{C}}|\cap |\hat{{\mathcal{C}}}|$ is a one-point set;
\item[(c)] \label{bif3}
$|{\mathcal{C}}|\cap |\hat{{\mathcal{C}}}|=\{\nu,\varrho\}$, with common boundary states $\nu$, $\varrho$;
\item[(d)] \label{bif4}
$ |{\mathcal{C}}|\subset |\hat{{\mathcal{C}}}|$ or $ |\hat{{\mathcal{C}}}|\subset |{\mathcal{C}}|$;
\item[(e)] \label{bif5}
$|{\mathcal{C}}|\cap |\hat{{\mathcal{C}}}|\supset |\hat{{\mathcal{C}}}^{\,\prime}|$, with $\hat{{\mathcal{C}}}^{\,\prime}$ having as the one endpoint one of the endpoints of ${\mathcal{C}}$, whereas the other endpoint is one of the endpoints of $\hat{{\mathcal{C}}}$, respectively.
\end{itemize}
Thereby, in all cases where  $|{\mathcal{C}}|\cap |\hat{{\mathcal{C}}}|$ consists of at least three states one of which is a common boundary state, then only \textup{(d)} can occur.
\end{lemma}
\begin{proof}
Let ${\mathcal{C}}$ and $\hat{{\mathcal{C}}}$ be evolving from $\nu$ to $\varrho$, and from $\hat{\nu}$ to $\hat{\varrho}$, respectively.
To be non-trivial, we may suppose that $|{\mathcal{C}}|\cap |\hat{{\mathcal{C}}}|$ consists of at least two states. Under this supposition we are going to show that then one of the cases (c)--(e)  occurs.

Start with the case that ${\mathcal{C}}$ and $\hat{{\mathcal{C}}}$ have one of the  endpoints in common, and have yet in common another state  $\omega\in   |{\mathcal{C}}|\cap |\hat{{\mathcal{C}}}|$ which however is not a common endpoint.

We proceed the case with  $\nu=\hat{\nu}$ (with $\varrho=\hat{\varrho}$ procedures will run  accordingly). Assume first  $\theta_0={\mathfrak d}_B(M|\nu,\varrho)\leq {\mathfrak d}_B(M|\nu,\hat{\varrho})=\theta_1$ to hold. Let $\hat{\omega}$ be the uniquely determined state on  $\hat{{\mathcal{C}}}$ such that ${\mathfrak d}_B(M|\nu,\hat{\omega})=\theta_0$.  Let $\hat{{\mathcal{C}}}^{\,\prime}$ be the part of $\hat{{\mathcal{C}}}$ extending from $\nu$ to $\hat{\omega}$. By Theorem \ref{gearcuni} and Example \ref{geodisj2}\,\eqref{geodisj2a},  $\hat{{\mathcal{C}}}^{\,\prime}$ is minimizing the Bures length between $\nu$ and $\hat{\omega}$. By  Corollary \ref{normalpara} the normal parameterization $[0,\theta_1]\ni\theta\longmapsto \hat{\nu}_\theta$ of $\hat{{\mathcal{C}}}$ for $\theta \in [0,\theta_0]$ with a $\hat{f}\in {\mathfrak F}(M|\nu,\hat{\omega})$ then equally well can be given as
\begin{subequations}\label{teilbogen}
\begin{equation}\label{teilbogen1}
\hat{\nu}_\theta=\biggl(\frac{\sin \theta\phantom{_0}}{\sin \theta_0}\biggr)^2\,\hat{\omega}+\biggl(\frac{\sin (\theta_0-\theta)}{\sin \theta_0}\biggr)^2\,\nu+\frac{2\sin \theta \sin (\theta_0-\theta)}{\sin^2 \theta_0}\,\re \hat{f}
\end{equation}
whereas the normal parameterization $[0,\theta_0]\ni \theta\longmapsto\nu_\theta$ of ${\mathcal{C}}$ reads
\begin{equation}\label{teilbogen2}
\nu_\theta=\biggl(\frac{\sin \theta\phantom{_0}}{\sin \theta_0}\biggr)^2\,\varrho+\biggl(\frac{\sin (\theta_0-\theta)}{\sin \theta_0}\biggr)^2\,\nu+\frac{2\sin \theta \sin (\theta_0-\theta)}{\sin^2 \theta_0}\,\re f
\end{equation}
\end{subequations}
for some $f\in {\mathfrak F}(M|\nu,\varrho)$. Now, by assumption on $\omega$, for some $\vartheta$ with $0<\vartheta\leq  \theta_0$, $\omega=\nu_\vartheta=\hat{\nu}_\vartheta$ has to be fulfilled. Once more again, by means of Theorem \ref{gearcuni} and Example \ref{geodisj2}\,\eqref{geodisj2a} we conclude that, since by assumption $\omega$ cannot be an endpoint of $\hat{\mathcal{C}}$ and thus has to be an inner point of a geodesic arc $\gamma$ with $|\gamma|=|\hat{\mathcal{C}}|$,  the supporting sets of those  parts of ${\mathcal{C}}$ and $\hat{{\mathcal{C}}}^{\,\prime}$  extending between $\nu$ and $\omega$ have to equal the supporting set of the uniquely determined geodesic arc $\gamma\in {\mathcal{C}}^{\nu,\omega}(M)$. Hence, for all $\theta$ with $0< \theta\leq \vartheta$,  $\nu_\theta=\hat{\nu}_\theta$ has to hold.  In view of \eqref{teilbogen} this  means that
\begin{equation}\label{bogengl}
\hat{\omega}-\varrho=\frac{\sin(\theta_0-\theta)}{\sin \theta}\,\re\,(f-\hat{f})
\end{equation}
for any $0<\theta<\vartheta\leq \theta_0$. From this $\hat{\omega}=\varrho$ and $\re f=\re \hat{f}$ follow. That is, $|{\mathcal{C}}|=|\hat{{\mathcal{C}}}^{\,\prime}|\subset |\hat{{\mathcal{C}}}|$ is seen, provided ${\mathfrak d}_B(M|\nu,\varrho)\leq {\mathfrak d}_B(M|\nu,\hat{\varrho})$ has been supposed. It is plain to see that in case of ${\mathfrak d}_B(M|\nu,\varrho)\geq {\mathfrak d}_B(M|\nu,\hat{\varrho})$ by tacitly the same arguments but with  interchanging  the r\^{o}les of ${\mathcal{C}}$ and
$\hat{{\mathcal{C}}}$ we will end up with $|{\mathcal{C}}|\supset |\hat{{\mathcal{C}}}|$.  Hence, if $|{\mathcal{C}}|\cap |\hat{{\mathcal{C}}}|$ contains a common endpoint and another state which is not a common endpoint, one arrives at (d).
Now, suppose ${\mathcal{C}}$ and $\hat{{\mathcal{C}}}$ have the same endpoints in common, i.e.  $\nu=\hat{\nu}$ and $\varrho=\hat{\varrho}$ hold. Then, in setting $\hat{\omega}=\varrho$ in \eqref{teilbogen1} and \eqref{bogengl}, in \eqref{teilbogen} the normal parameterizations of ${\mathcal{C}}$ and $\hat{{\mathcal{C}}}$ are given. Now, suppose $\nu_\theta=\hat{\nu}_{\theta'}$, for some $\theta,\theta'\in ]0,\theta_0[$. Then, owing to $\theta={\mathfrak d}_B(M|\nu,\nu_\theta)$ and $\theta'={\mathfrak d}_B(M|\nu,\hat{\nu}_{\theta'})$, we  conclude that $\theta=\theta'$. Hence, from \eqref{bogengl}
$0=\re\,(f-\hat{f})$ follows, and then \eqref{teilbogen} yield identical normal parameterizations.  Thus, if ${\nu,\varrho}\in  |{\mathcal{C}}|\cap|\hat{{\mathcal{C}}}|$, then always either $|{\mathcal{C}}|=|\hat{{\mathcal{C}}}|$, and which is a special case of (d) again, or  $\{\nu,\varrho\}=  |{\mathcal{C}}|\cap|\hat{{\mathcal{C}}}|$, which is case (c). Especially, this implies that if $ |{\mathcal{C}}|\cap|\hat{{\mathcal{C}}}|$ consist of at least three states, one of which is a common boundary state, then only (d) can occur.

Finally, let us consider those cases with at least two states in $|{\mathcal{C}}|\cap|\hat{{\mathcal{C}}}|$ which are not yet covered by the above special cases. In line with this, assume $\omega_1,\omega_2\in|{\mathcal{C}}|\cap|\hat{{\mathcal{C}}}|$ such that $\omega_1\not=\omega_2$, but now  both being inner states in respect of ${\mathcal{C}}$ and $\hat{{\mathcal{C}}}$, respectively. To start with, let us consider those connected components ${\mathcal{C}}_1$ and $\hat{{\mathcal{C}}}_1$ of ${\mathcal{C}}$ and $\hat{{\mathcal{C}}}$, respectively, possessing $\omega_1$ as an endpoint and containing $\omega_2$ as an inner point. Correspondingly, let ${\mathcal{C}}_2$ and $\hat{{\mathcal{C}}}_2$ be the connected components of ${\mathcal{C}}$ and $\hat{{\mathcal{C}}}$, respectively, possessing $\omega_2$ as an endpoint and containing $\omega_1$ as an inner point. Then, the continuous curves ${\mathcal{C}}_i$ and $\hat{{\mathcal{C}}}_i$ for $i=1,2$ are minimizing the Bures length between their respective endpoints, and reproduce the type of situation among two shortest paths with more than one point in common which has been analyzed above. That is, for ${\mathcal{C}}_i, \hat{{\mathcal{C}}}_i$ with $i=1,2$ we are in case (d). That is, we have $ |{\mathcal{C}}_1|\subset |\hat{{\mathcal{C}}}_1|$ or $ |\hat{{\mathcal{C}}}_1|\subset |{\mathcal{C}}_1|$, and $|{\mathcal{C}}_2|\subset |\hat{{\mathcal{C}}}_2|$ or $ |\hat{{\mathcal{C}}}_2|\subset |{\mathcal{C}}_2|$. Thus, if e.g.
$|{\mathcal{C}}_1|\subset |\hat{{\mathcal{C}}}_1|$ and $ |\hat{{\mathcal{C}}}_2|\subset |{\mathcal{C}}_2|$ happens, we have
$|{\mathcal{C}}|\cap|\hat{{\mathcal{C}}}|=|{\mathcal{C}}_1|\cup |\hat{{\mathcal{C}}}_2| $.
Thereby, in view of our standard argumentation by means of Theorem \ref{gearcuni} and Example \ref{geodisj2}\,\eqref{geodisj2a}, we can be sure that $|{\mathcal{C}}_1|\cup |\hat{{\mathcal{C}}}_2| $ has to be the supporting set $|\hat{{\mathcal{C}}}^{\,\prime}|$ of some shortest path $\hat{{\mathcal{C}}}^{\,\prime}$ with the one endpoint being endpoint of ${\mathcal{C}}$, whereas the other endpoint is an endpoint of $\hat{{\mathcal{C}}}$. That is, then we end up in case (e). On the other hand, if e.g.
$|{\mathcal{C}}_1|\subset |\hat{{\mathcal{C}}}_1|$ and $ |{\mathcal{C}}_2|\subset |\hat{{\mathcal{C}}}_2|$ happens, then  we have
$|{\mathcal{C}}|\cap|\hat{{\mathcal{C}}}|=|{\mathcal{C}}_1|\cup |{\mathcal{C}}_2| $, with the set on the right-hand side being the supporting set of ${\mathcal{C}}$, that is, we will end up in case (d). And analogously, also for the remaining other two combinations one is arriving at either (e) or (d).
\end{proof}
\begin{remark}\label{simpcu}
From Theorem \ref{gearcuni} it follows that each shortest path in ${\mathcal S}(M)$ has to be a simple curve for, each geodesic arc $\gamma$ is so, by Corollary \ref{auxin1}.
\end{remark}
In the following two notions relating topological pecularities of points in metrical spaces are adapted to the case of states under the Bures metric, refer to \cite[\S 19.]{Rino:61}.
\begin{definition}\label{bipunkt}
A state $\omega\in {\mathcal S}(M)$ is termed
\begin{enumerate}
\item\label{pass}
`transition state' if there exists a shortest path ${\mathcal{C}}$ passing through $\omega$;
\item \label{bipunkt0}
`bifurcation state' if it is  endpoint of three shortest paths ${\mathcal{C}}_1$, ${\mathcal{C}}_2$ and ${\mathcal{C}}_3$ such that the following properties hold:
\begin{enumerate}
\item \label{noexbi1}
$|{\mathcal{C}}_1|\cup |{\mathcal{C}}_2|$ and $|{\mathcal{C}}_1|\cup |{\mathcal{C}}_3|$ are supporting sets  of shortest paths;
\item \label{noexbi2}
  $\omega\not\in|\gamma|$, for any shortest path  $\gamma$ obeying  $|\gamma|\subset |{\mathcal{C}}_2|\cap |{\mathcal{C}}_3|$.
\end{enumerate}
\end{enumerate}
\end{definition}
Relating properties of the set of all transition states, let $\omega\in {\mathcal S}(M)$, and take some $\varrho\not=\omega$. Consider $\gamma\in {\mathcal{C}}^{\omega,\varrho}(M)$. According to Theorem \ref{innergeodist} $\gamma$ is a shortest path. By  Corollary \ref{auxin1} and in view of Definition \ref{bipunkt}\,\eqref{pass}, each inner state of $\gamma$ is a transition state. Hence, since $\omega$ is endpoint of $\gamma$,  $\omega$ can be arbitrarily well approximated by transition states. Thus we have denseness of the set of transition states within ${\mathcal S}(M)$ and, since this especially also applies if $\omega$ is a transition state by itself, in addition we infer that the set of transition states cannot admit any  isolated points. These are general properties of metric spaces where any two points can be connected by at least one shortest path. More specific is the following.
\begin{corolla}\label{nobipass}
The set of bifurcation states is empty.
\end{corolla}
\begin{proof}
To see this,
suppose $\omega$ to be endpoint of three shortest paths ${\mathcal{C}}_1$, ${\mathcal{C}}_2$ and ${\mathcal{C}}_3$. Then, in line with
Definition \ref{bipunkt}\,\eqref{noexbi2}, suppose $|\gamma|\subset |{\mathcal{C}}_2|\cap |{\mathcal{C}}_3|$, for some shortest path $\gamma$. Accordingly, $|{\mathcal{C}}_2|\cap |{\mathcal{C}}_3|$ then contains at least three mutually different states, one of which can be assumed   to be the common endpoint $\omega$. Thus, we can apply Lemma \ref{bifurc} with ${\mathcal{C}}={\mathcal{C}}_2$ and $\hat{\mathcal{C}}={\mathcal{C}}_3$. The result is that then only $|{\mathcal{C}}_2|\subset |{\mathcal{C}}_3|$ or  $|{\mathcal{C}}_3|\subset |{\mathcal{C}}_2|$ can occur. In the first case, $\gamma={\mathcal{C}}_2$ is a shortest path obeying $\omega\in |\gamma|\subset |{\mathcal{C}}_2|\cap |{\mathcal{C}}_3|$. In the second case, $\gamma={\mathcal{C}}_3$ is a shortest path obeying $\omega\in |\gamma|\subset |{\mathcal{C}}_2|\cap |{\mathcal{C}}_3|$. Hence, in either case, this is showing that the condition \eqref{noexbi2} for $\omega$ never can be met.
\end{proof}
\begin{remark}\label{allgbi}
In a general metric space the analogous to (a)--(e) of Lemma \ref{bifurc} conditions are necessary and sufficient to assure that no bifurcation points exist.
\end{remark}
Let ${\mathcal{C}}$ be a shortest path extending from endpoint $\nu$ to endpoint $\varrho\not=\nu$. A shortest path $\hat{\mathcal{C}}$ is said to extend ${\mathcal{C}}$ beyond $\varrho$ if
$|{\mathcal{C}}|\subsetneq |\hat{\mathcal{C}}|$ and $\nu$ is a common endpoint. Whether or not an extension of ${\mathcal{C}}$ beyond $\varrho$ exists is given by the following result.
\begin{lemma}\label{Bext}
Let $\{\pi,{\mathcal H}_\pi\}$ be a unital $^*$-representation of $M$ with  non-trivial $\pi$-fibres of $\nu, \varrho$. An extension of ${\mathcal{C}}$ beyond $\varrho$ exists if, and only if, on $\pi(M)^{\,\prime}$
\begin{equation}\label{Bext1}
h^\pi_{\varphi,\varphi}\ll h^\pi_{\zeta,\varphi}
\end{equation}
holds, for $\zeta\in {\mathcal S}_{\pi,M}(\varrho)$, $\varphi\in {\mathcal S}_{\pi,M}(\nu)$ chosen in accordance with $\langle \zeta,\varphi\rangle_\pi=F(M|\nu,\varrho)$.
\end{lemma}
\begin{proof}
Let $\{\pi,{\mathcal H}_\pi\}$,   $\zeta\in {\mathcal S}_{\pi,M}(\varrho)$, $\varphi\in {\mathcal S}_{\pi,M}(\nu)$ be such that  $\langle \zeta,\varphi\rangle_\pi=F(M|\nu,\varrho)=F$.
In view of Theorem \ref{gearcuni} we may be content with arguing for the  case of extending a geodesic arc beyond of one of its endpoints.

In line with this, suppose the geodesic arc $\hat{\gamma}\in {\mathcal{C}}^{\nu,\omega}(M)$ to be  an extension beyond $\varrho$ of the geodesic arc $\gamma\in {\mathcal{C}}^{\nu,\varrho}(M)$. In accordance with \eqref{geo0a}, the geodesic arc $\hat{\gamma}:[0,1]\ni t\longmapsto \omega_t$ in respect of some unital $^*$-representation $\{\hat{\pi},\hat{\mathcal H}_{\hat{\pi}}\}$ can be implemented by vectors  $${\hat{\varphi}}_t=t\,\hat{\xi}+\hat{\lambda}(t)\,\hat{\varphi}\in {\mathcal S}_{\hat{\pi},M}(\omega_t)$$ with $\hat{\xi}\in {\mathcal S}_{\hat{\pi},M}(\omega)$,  $\hat{\varphi}\in {\mathcal S}_{\hat{\pi},M}(\nu)$ obeying   $\langle \hat{\xi},\hat{\varphi}\rangle_{\hat{\pi}}=F(M|\nu,\omega)=\hat{F}$, and function
$$\hat{\lambda}(t)=-t\hat{F}+\sqrt{1-t^2(1-\hat{P})} $$
with $\hat{P}=\hat{F}^2$. By choice of $\hat{\xi}$ and $\hat{\varphi}$ equivalently  $h^{\hat{\pi}}_{{\hat{\xi}},{\hat{\varphi}}}\geq {\mathsf 0}$ is fulfilled  on ${\hat{\pi}}(M)^{\,\prime}$.
Since  $\hat{\gamma}$ is passing through $\omega$, $s\in ]0,1[$ has to exist such that $\varrho=\omega_s$, and the vector
\begin{equation}\label{Bext2}
\hat{\zeta}=s\,\hat{\xi}+\hat{\lambda}(s)\,\hat{\varphi}\in {\mathcal S}_{\hat{\pi},M}(\varrho)
\end{equation}
by the previous then is obeying $\langle \hat{\zeta},\hat{\varphi}\rangle_{\hat{\pi}}=F(M|\nu,\varrho)=F$.
 Accordingly, from \eqref{Bext2}
\begin{equation}\label{Bext3}
F=s\,\hat{F}+\hat{\lambda}(s)
\end{equation}
is obtained. Also, since $\varrho$ is an inner point of the geodesic arc $\hat{\gamma}$, by applying Corollary \ref{porter3}, \eqref{porter4b} and Corollary \ref{uniqueimp} one infers that $\nu_\varrho={\mathsf 0}$.
Note that from \eqref{Bext2} the relation
$$
\hat{\xi}=s^{-1}\,\left(\hat{\zeta}-\hat{\lambda}(s)\,\hat{\varphi}\right)
$$
can be inferred to hold. Accordingly,
we get that $\omega$ can be represented as follows
\begin{equation}\label{Bext3a}
\omega=s^{-2}\,\left( \varrho+\hat{\lambda}(s)^2\,\nu- 2\,\hat{\lambda}(s)\,\re f^{\hat{\pi}}_{\hat{\zeta},\hat{\varphi}}\right)
\end{equation}
In respect of $\{\pi,{\mathcal H}_\pi\}$, let us consider a vector $\xi\in {\mathcal H}_\pi$ defined by
\begin{equation}\label{Bext3aa}
\xi=s^{-1}\,\left(\zeta-\hat{\lambda}(s)\,\varphi\right)
\end{equation}
Let $\hat{\omega}$ be the positive linear form implemented by $\xi$ via $\pi$. Then, we find that
\begin{equation}\label{Bext3b}
\hat{\omega}=s^{-2}\,\left( \varrho+\hat{\lambda}(s)^2\,\nu- 2\,\hat{\lambda}(s)\,\re f^\pi_{\zeta,\varphi}\right)
\end{equation}
Since  $\nu_\varrho={\mathsf 0}$ is fulfilled, by Theorem \ref{unique}  with the unique $f\in {\mathfrak F}(M|\nu,\varrho)$ we have
\[
f=f^\pi_{\zeta,\varphi}=f^{\hat{\pi}}_{\hat{\zeta},\hat{\varphi}}
\]
Hence, on comparing \eqref{Bext3b} with \eqref{Bext3a}, $\hat{\omega}=\omega$ has to be followed,  i.e.~$\xi\in {\mathcal S}_{\pi,M}(\omega)$.
In view of \eqref{Bext3} and by the choice of $\zeta$ and $\varphi$, from the latter the relation
\begin{equation}\label{Bext4}
\langle\xi,\varphi\rangle_\pi=s^{-1}\,\left(F-\hat{\lambda}(s)\right)=\hat{F}=F(M|\nu,\omega)
\end{equation}
follows. Hence, $h^\pi_{\xi,\varphi}\geq {\mathsf 0}$ on $\pi(M)^{\,\prime}$ has to be fulfilled. From this together with \eqref{Bext3aa}
\[
s^{-1}\,h^\pi_{\xi,\varphi}=h^\pi_{\zeta,\varphi}-\hat{\lambda}(s)\,h^\pi_{\varphi,\varphi}\geq {\mathsf 0}
\]
is obtained. Owing to  $\hat{\lambda}(s)>0$ this proves \eqref{Bext1}. Thus, the latter condition is a necessary condition for an extension $\hat{\gamma}$ of $\gamma$ beyond $\varrho$ to exist.

To see sufficiency, let $\gamma: [0,1]\ni t\longmapsto\nu_t$ be a geodesic arc $\gamma\in {\mathcal{C}}^{\nu,\varrho}(M)$, and  let \eqref{Bext1} hold, for some implementing vectors $\zeta$ and $\varphi$ as asserted, for $\varrho$ and $\nu$, respectively. Then $\langle\zeta,\varphi\rangle_\pi=F(M|\nu,\omega)=F$, and  $\varepsilon$ with $1>\varepsilon >0$ exists such that
\begin{equation}\label{Bext5}
h^\pi_{\zeta,\varphi}\geq \varepsilon\,h^\pi_{\varphi,\varphi}
\end{equation}
is fulfilled on $\pi(M)^{\,\prime}$. Now, consider the unit vector $\xi$ defined by
$
\xi=\alpha^{-1}\left(\zeta-\varepsilon\,\varphi\right)
$, with the shortcut  $\alpha=\|\zeta-\varepsilon\,\varphi\|_\pi$. In view of \eqref{Bext5} $\xi$ is obeying $h^\pi_{\xi,\varphi}\geq {\mathsf 0}$. Let $\omega$ be the state implemented by $\xi$, i.e.~$\xi\in {\mathcal S}_{\pi,M}(\omega)$. We then consider the geodesic arc $\hat{\gamma}\in {\mathcal{C}}^{\nu,\omega}(M)$,  $\hat{\gamma}:[0,1]\ni t\longmapsto\omega_t$,  implemented by the vectors
$
\psi_t=t\,\xi+\hat{\lambda}(t)\,\varphi\in {\mathcal S}_{\pi,M}(\omega_t)
$, with $\hat{\lambda}(t)$ referring to \begin{equation}\label{Bext5a}
\langle\xi,\varphi\rangle_\pi=F(M|\nu,\omega)=\hat{F}=\alpha^{-1}(F-\varepsilon)
\end{equation}
Obviously, the implementing family of $\hat{\gamma}$ equivalently can be read also as following:
\begin{subequations}\label{Bext6a}
\begin{equation}\label{Bext6}
\psi_t=\left(\frac{t}{\alpha}\right)\,\zeta+\left(\hat{\lambda}(t)-\left(\frac{t}{\alpha}\right)\,\varepsilon\right)\,\varphi
\end{equation}
From $\zeta=\alpha\,\xi+\varepsilon\,\varphi$ we conclude that
$
1=\alpha^2+\varepsilon^2+2\,\alpha\,\varepsilon \hat{F}
$ holds.
Hence, $$\alpha=-\varepsilon\,\hat{F} +\sqrt{\varepsilon^2\hat{F}^2+1-\varepsilon^2}=-\varepsilon\,\hat{F} +\sqrt{1-\varepsilon^2(1-\hat{P})}=\hat{\lambda}(\varepsilon)$$ can be inferred. Thus especially, since for  $0<\varepsilon<1$ always   $0<\hat{\lambda}(\varepsilon)<1$ follows,  $0<\alpha<1$ has to be fulfilled. Also, since all vectors in \eqref{Bext6} are unit vectors, and since  $\langle\zeta,\varphi\rangle_\pi=F(M|\nu,\omega)=F$ is fulfilled, for $0\leq t\leq \alpha$ we have to conclude to
\begin{equation}\label{Bext7}
\lambda\left(\frac{t}{\alpha}\right)=\hat{\lambda}(t)-\left(\frac{t}{\alpha}\right)\,\varepsilon
\end{equation}
\end{subequations}
Accordingly, by \eqref{Bext6a} the restriction $\hat{\gamma}|[0,\alpha]$ of $\hat{\gamma}$ to the interval $[0,\alpha]$ in a one-to-one manner  corresponds to the geodesic arc $\gamma'\in {\mathcal{C}}^{\nu,\varrho}(M)$, $\gamma':[0,1]\ni t\longmapsto\nu_t$, implemented via $\pi$ by the family $[0,1]\ni t\longmapsto\varphi_t=t\,\zeta+\lambda(t)\,\varphi$. Thereby, the mentioned correspondence is given by means of the reparameterization $\nu_t=\omega_{\alpha\, t}$, for $0\leq t\leq 1$. Hence, $\hat{\gamma}$ is an extension beyond $\varrho$ of the geodesic arc  $\gamma'$. Finally, since $\varrho=\omega_\alpha$ is an inner point of the geodesic arc $\hat{\gamma}$, in line with   Corollary \ref{porter3}, \eqref{porter4b} and Corollary \ref{uniqueimp}  we have that the following is satisfied: $$\left(\omega_0\right)_{\omega_\alpha}=\nu_\varrho={\mathsf 0}$$
Hence, by Theorem \ref{einga}, there is exactly one geodesic arc in ${\mathcal{C}}^{\nu,\varrho}(M)$, and therefore $\gamma=\gamma'$. That is, the constructed geodesic arc $\hat{\gamma}$ is an extension of $\gamma$ beyond $\varrho$.
\end{proof}
\begin{definition}\label{geExt}
A shortest path ${\mathcal{C}}$ with endpoint $\varrho$ admits a geodesic extension beyond $\varrho$ if there is a shortest path $\hat{\mathcal{C}}$ starting at a state $\omega\in |{\mathcal{C}}|$ and passing through $\varrho$ such that $|{\mathcal{C}}|\cap |\hat{\mathcal{C}}|$ is supporting a shortest path with endpoints $\omega, \varrho$.
\end{definition}
The extension of a shortest path beyond an endpoint as done in Lemma \ref{Bext} is a special case of a geodesic extension. A state $\varrho$ such that any shortest path ending in $\varrho$ admits a geodesic extension beyond $\varrho$ is termed `polydirectional transition state'.
\begin{theorem}\label{geoBext}
Let ${\mathcal{C}}$ be a shortest path with endpoints $\nu$ and  $\varrho\not=\nu$. A geodesic extension of ${\mathcal{C}}$ beyond $\varrho$ exists if, and only if, over $\pi(M)^{\,\prime}$ one has
\begin{equation}\label{geoBext1}
h^\pi_{\varphi,\varphi}\ll h^\pi_{\zeta,\varphi}+h^\pi_{\zeta,\zeta}
\end{equation}
for $\zeta\in {\mathcal S}_{\pi,M}(\varrho)$, $\varphi\in {\mathcal S}_{\pi,M}(\nu)$ chosen with $\langle \zeta,\varphi\rangle_\pi=F(M|\nu,\varrho)$, with respect to any unital $^*$-representation $\{\pi,{\mathcal H}_\pi\}$ of $M$ where this choice can be accomplished.
\end{theorem}
\begin{proof}
By Theorem \ref{gearcuni} we may assume  ${\mathcal{C}}=\gamma$, with $\gamma\in {\mathcal{C}}^{\nu,\varrho}(M)$, $\gamma: [0,1]\ni t\longmapsto \nu_t$, implemented with respect to some unital $^*$-representation $\{\pi,{\mathcal H}_\pi\}$ by
$
[0,1]\ni t\longmapsto \varphi_t=t\,\zeta+\lambda(t)\,\varphi
$
with $\zeta\in {\mathcal S}_{\pi,M}(\varrho)$, $\varphi\in {\mathcal S}_{\pi,M}(\nu)$ and $\langle \zeta,\varphi\rangle_\pi=F(M|\nu,\varrho)$.

According to Definition \ref{geExt}, ${\mathcal{C}}$ admits a geodesic extension beyond $\varrho$ if, and only if, there is $0<t<1$ such that with $\omega=\nu_t$ the unique geodesic arc $\hat{\gamma}\in {\mathcal{C}}^{\omega,\varrho}(M)$ can be extended as a geodesic arc beyond $\varrho$. Since the latter is equivalent to
\begin{equation}\label{geoBext2}
h^\pi_{\varphi_t,\varphi_t}\ll h^\pi_{\zeta,\varphi_t}
\end{equation}
for some $t\in ]0,1[$ by Lemma  \ref{Bext},  equivalence of \eqref{geoBext1} and  \eqref{geoBext2} is in quest.
By positivity of $h^\pi_{\zeta,\zeta}$, $h^\pi_{\zeta,\varphi}$, $h^\pi_{\varphi,\varphi}$, and since  $t\in ]0,1[$ and $\lambda(t)\in ]0,1[$ are fulfilled, from
 \begin{equation}\label{geoBext3}
h^\pi_{\varphi_t,\varphi_t}=t^2\,h^\pi_{\zeta,\zeta}+\lambda(t)^2\,h^\pi_{\varphi,\varphi}+2\,t\lambda(t)\,h^\pi_{\zeta,\varphi}
\end{equation}
we see    $\lambda(t)^2\,h^\pi_{\varphi,\varphi} \leq h^\pi_{\varphi_t,\varphi_t}$, that is, $h^\pi_{\varphi,\varphi}\ll h^\pi_{\zeta,\varphi_t} $ is fulfilled. Hence, from \eqref{geoBext2}
\[
h^\pi_{\varphi,\varphi}\ll h^\pi_{\zeta,\varphi_t}=t\,h^\pi_{\zeta,\zeta}+\lambda(t)\,h^\pi_{\zeta,\varphi}\leq  h^\pi_{\zeta,\zeta}+h^\pi_{\zeta,\varphi}
\]
follows. Thus condition  \eqref{geoBext1} is necessary for \eqref{geoBext2} to be fulfilled. To see sufficiency, suppose \eqref{geoBext1} to be satisfied. Then, there exists $\varepsilon >0$ such that $$\varepsilon\, h^\pi_{\varphi,\varphi}\leq h^\pi_{\zeta,\zeta}+h^\pi_{\zeta,\varphi}$$
From this together with \eqref{geoBext3} conclude that for each fixed $t$ with $0<t<1$
\begin{eqnarray}
h^\pi_{\varphi_t,\varphi_t} & \leq & \left(t^2+\frac{\lambda(t)^2}{\varepsilon}\right)\,h^\pi_{\zeta,\zeta}+\left(2\,t\lambda(t)+\frac{\lambda(t)^2}{\varepsilon}\right)\,h^\pi_{\zeta,\varphi}\\
&=& \left(t+\frac{\lambda(t)^2}{t\varepsilon}\right)t\,h^\pi_{\zeta,\zeta}+\left(2\,t+\frac{\lambda(t)}{\varepsilon}\right)\lambda(t)\,h^\pi_{\zeta,\varphi}\\
&\leq & \alpha \left(t\,h^\pi_{\zeta,\zeta}+\lambda(t)\,h^\pi_{\zeta,\varphi}\right)=\alpha\,h^\pi_{\zeta,\varphi_t}
\end{eqnarray}
with $\alpha$ at least as large as the maximum of $(t+\lambda(t)/t\varepsilon)$ and $(2\,t+\lambda(t)/\varepsilon)$. This is the same as \eqref{geoBext2} and assures that a geodesic extension of $\gamma$ beyond $\varrho$ exists.
\end{proof}
Suppose ${\mathcal C } :I\ni \theta\longmapsto \nu_\theta\in {\mathcal S}(M)$, with compact interval $I$, to be the normal parameterization of a Bures-rectifiable continuous curve. ${\mathcal C }$ is termed `geodesic curve' if to each $\theta_0\in I$ there exists $\varepsilon>0$ such that ${\mathcal C }|[\theta_0-\varepsilon,\theta_0+\varepsilon]\cap I$ is a shortest path. Thus especially, each geodesic curve can be decomposed into finitely many shortest paths. Let us consider a continuous curve ${\mathcal C }$ with normal parameterization
$${\mathcal C }: [0,\beta\,[\ni \theta\longmapsto \nu_\theta\in {\mathcal S}(M)$$
and where $\beta=+\infty$ is possible. Suppose  $\nu_0=\nu$ and assume that ${\mathcal C }|[0,\alpha]$ for each $\alpha$ with  $0\leq \alpha<\beta$ is a geodesic curve. In case that either $\beta=+\infty$ or $0<\beta<\infty$ is the  maximal possible value such that these properties hold,  the curve ${\mathcal C }$ will be referred to as `halfgeodesic' with boundary point $\nu$ and length $\beta$.
We remark that in case if any state proves to be a polydirectional transition state, then each halfgeodesic is of infinite length.
As a special case, one has a `geodesic loop', which is a halfgeodesic which is recurrent, that is, there exists a smallest $\theta_0\in ]0,\beta[$ such that  $\nu_\theta=\nu_{\theta_0+\theta}$, for all $\theta\geq 0$ satisfying $\theta_0+\theta< \beta$. Clearly, in this case due to the maximality requirement in the definition the length of a geodesic loop has to be re-defined as $\beta=\infty$, whereas the length of one cycle is $\theta_0$. Note that as a consequence of Corollary \ref{nobipass}, for a shortest path $\gamma\in {\mathcal{C}}^{\nu,\varrho}(M)$, either there exists a unique halfgeodesic ${\mathcal{C}}$ with endpoint $\nu$ and length $\beta>{\mathfrak d}_B(M|\nu,\varrho)$ and obeying $|\gamma|\subset | {\mathcal{C}}|$, or $|\gamma|\backslash\{\varrho\}$ is the supporting set of a halfgeodesic of the finite length $\beta={\mathfrak d}_B(M|\nu,\varrho)$. In the former case, refer to ${\mathcal{C}}$ as  maximal geodesic extension of $\gamma$ beyond $\varrho$, whereas in the latter case $\gamma$  does  not admit a geodesic extension beyond $\varrho$. Obviously, Theorem \ref{geoBext} provides the tool to decide whether or not this case occurs. We are going to give some applications.
\begin{example}\label{notrans}
If mutually disjoint states $\nu$ and $\varrho$ exist, then neither $\nu$ nor $\varrho$ can be  polydirectional transition states,  and halfgeodesics of finite length exist.
\end{example}
\begin{proof}
Let $\{\pi,{\mathcal H}_\pi\}$ be a unital $^*$-representation with non-trivial $\pi$-fibres of $\nu$ and $\varrho$. Then, we may choose $\varphi\in {\mathcal S}_{\pi,M}(\nu)$ and $\zeta\in {\mathcal S}_{\pi,M}(\varrho)$ with $\langle\zeta,\varphi\rangle_\pi=F(M|\nu,\varrho)$. By mutual disjointness of $\nu$ and $\varrho$, also $\nu_\varrho$ and $\varrho_\nu$ are mutually disjoint in a trivial way (since then $\nu_\varrho={\mathsf 0}=\varrho_\nu$). Thus, the geodesic arc $\gamma\in {\mathcal{C}}^{\nu,\varrho}(M)$ is uniquely determined. Therefore each shortest path with endpoints $\nu$ and $\varrho$ has the same supporting set $|\gamma|$. By disjointness, $p_\pi(\varphi)$ and $p_\pi(\zeta)$ as well as $p^{\,\prime}_\pi(\varphi)$ and $p^{\,\prime}_\pi(\zeta)$ have to be mutually orthogonal orthoprojections in $N=\pi(M)^{\,\prime\prime}$ and $N^{\,\prime}$, respectively. Thus $\nu\perp\varrho$, which implies $F(M|\nu,\varrho)=0$, by Corollary \ref{subadd}\,\eqref{subadd2},
and $h^\pi_{\zeta,\zeta}\perp h^\pi_{\varphi,\varphi}$ and  $h^\pi_{\zeta,\varphi}=h^\pi_{\varphi,\zeta}={\mathsf 0}$   on $N^{\,\prime}$. Hence,
\[
h^\pi_{\varphi,\varphi}\not\ll h^\pi_{\zeta,\varphi}+h^\pi_{\zeta,\zeta},\ h^\pi_{\zeta,\zeta}\not\ll h^\pi_{\varphi,\zeta}+h^\pi_{\varphi,\varphi}
\]
that is, the condition \eqref{geoBext1} cannot be satisfied neither in respect of $\varrho$ nor in respect of $\nu$. By Theorem \ref{Bext}, $\gamma$ does not admit geodesic extensions beyond $\varrho$ or $\nu$, and therefore both states cannot be  polydirectional states, and  $|\gamma|\backslash \{\varrho\}$ or  $|\gamma|\backslash \{\nu\}$, respectively, are supporting sets of halfgeodesics with boundary point $\nu$ or $\varrho$, respectively, and according to \eqref{geodist} are of finite length $\beta={\mathfrak d}_B(M|\nu,\varrho)=\pmb{\pi}/2$.
\end{proof}
\begin{example}\label{loop}
For a unital $^*$-representation $\{\pi,{\mathcal H}_\pi\}$ of $M$, if orthoprojections $p,q\in \pi(M)^{\,\prime\prime}$ with $p\perp q$, $p\sim q$ exist, then geodesic loops with cycle-length $\pmb{\pi}$ exist.
\end{example}
\begin{proof}
Let $N=\pi(M)^{\,\prime\prime}$. Consider a unit vector $\varphi\in {\mathcal H}_\pi$ with $p\varphi=\varphi$ and a partial isometry $v\in N$  with $v^*v=p$, $vv^*=q$.  Define states $\nu$ and $\varrho$ by $\varphi\in {\mathcal S}_{\pi,M}(\nu)$, $\zeta=v\varphi\in  {\mathcal S}_{\pi,M}(\varrho)$. Then   $$h^\pi_{\zeta,\zeta}=h^\pi_{v\varphi,v\varphi}=h^\pi_{\varphi,\varphi}$$ on $N^{\,\prime}$. Due to $p\perp q$ we can be assured that $\nu\perp\varrho$. Hence, $F(M|\nu,\varrho)=0$, by Corollary \ref{subadd}\,\eqref{subadd2}, and thus  $h^\pi_{\zeta,\varphi}={\mathsf 0}$. For $0\leq \theta\leq \frac{\pi}{2}$ let unit vectors given by
\begin{equation}\label{uvec}
\varphi_\theta=\zeta\,\sin \theta+\varphi\,\cos \theta
\end{equation}
Then, according to formulae \eqref{para6} and since due to  $\nu\perp\varrho$ one has  $\theta_0=\pmb{\pi}/2$, by the family $\{\varphi_\theta: 0\leq \theta\leq \pmb{\pi}/2 \}$ the normal parameterization of the  uniquely determined geodesic arc $\gamma\in {\mathcal{C}}^{\nu,\varrho}(M)$ is implemented. Clearly, formula \eqref{uvec} yields unit vectors for $\theta\geq \pmb{\pi}/2$ as well, and then is cyclic with period $\pi$. We will show that by the states $\nu_\theta$ given by $\varphi_\theta\in {\mathcal S}_{\pi,M}(\nu_\theta)$ for $0\leq \theta<\infty$ a geodesic loop is given. In fact, the condition \eqref{geoBext1} is satisfied in a trivial way, and thus the geodesic arc $\gamma\in {\mathcal{C}}^{\nu,\varrho}(M)$ according to Theorem \ref{geoBext} admits a geodesic extension beyond $\varrho$. Also, owing to
\begin{subequations}\label{uvec0}
\begin{eqnarray}\label{uvec0a}
h^\pi_{\varphi_{3\pmb{\pi}/4},\varphi_{\pmb{\pi}/4}}\ &=&h^\pi_{(\zeta-\varphi)/\sqrt{2},(\zeta+\varphi)/\sqrt{2}}=\frac{1}{2}\,\bigl(h^\pi_{\zeta,\zeta}-h^\pi_{\varphi,\varphi}\bigr)={\mathsf 0}\\
\label{uvec0b}
h^\pi_{\varphi_{\pmb{\pi}/4},\varphi_{\pmb{\pi}/4}}\ \ &=&h^\pi_{(\zeta+\varphi)/\sqrt{2},(\zeta+\varphi)/\sqrt{2}}=\frac{1}{2}\,\bigl(h^\pi_{\zeta,\zeta}+h^\pi_{\varphi,\varphi}\bigr)\\
\label{uvec0c}
h^\pi_{\varphi_{3\pmb{\pi}/4},\varphi_{3\pmb{\pi}/4}}&=&h^\pi_{(\zeta-\varphi)/\sqrt{2},(\zeta-\varphi)/\sqrt{2}}=\frac{1}{2}\,\bigl(h^\pi_{\zeta,\zeta}+h^\pi_{\varphi,\varphi}\bigr)
\end{eqnarray}
\end{subequations}
from \eqref{uvec0a} and Lemma \ref{bas3}\,\eqref{bas3b} we find $F\bigl(M|\nu_{3\pmb{\pi}/4},\nu_{\pmb{\pi}/4}\bigr)=0$, and thus by
\[
\zeta_\vartheta=\frac{(\zeta-\varphi)}{\sqrt{2}}\,\sin\vartheta+ \frac{(\zeta+\varphi)}{\sqrt{2}}\,\cos\vartheta=\zeta\,\sin (\vartheta+\pmb{\pi}/4)+\varphi\,\cos (\vartheta+\pmb{\pi}/4)=\varphi_{(\vartheta+\pmb{\pi}/4)}
\]
for $0\leq \vartheta\leq \pmb{\pi}/2$ a geodesic arc $\gamma' \in {\mathcal{C}}^{\nu_{\pmb{\pi}/4},\nu_{3\pmb{\pi}/4}}(M)$ of length $\pmb{\pi}/2$ is implemented, and which is passing through $\varrho$ at $\vartheta=\pmb{\pi}/4$, and with the help of which obviously Definition \ref{geExt} is fulfilled, with  ${{\mathcal{C}}}=\gamma$ and ${\hat{\mathcal{C}}}=\gamma'$ and $\omega=\nu_{\pmb{\pi}/4}$. Obviously, $|\gamma|\cup |\gamma'|$ then is the supporting set of a geodesic curve of length $3\pmb{\pi}/4$. Note that according to Theorem \ref{geoBext} from  \eqref{uvec0} we infer that $\gamma'$ admits a geodesic extension beyond $\nu_{3\pmb{\pi}/4}$. To construct such an extension practically, consider that
\begin{subequations}\label{uvec1}
\begin{eqnarray}\label{uvec1a}
h^\pi_{\varphi_{\pi},\varphi_{\pmb{\pi}/2}} \ \  &=&h^\pi_{(-\varphi),\zeta}\ \ \ \ \, =-h^\pi_{\varphi,\zeta}=-h^\pi_{\zeta,\varphi}={\mathsf 0}\\
\label{uvec1b}
h^\pi_{\varphi_\pi,\varphi_\pi}\ \ \ \ \, &=&h^\pi_{(-\varphi),(-\varphi)}=\phantom{-}h^\pi_{\varphi,\varphi}\\
\label{uvec1c}
h^\pi_{\varphi_{\pmb{\pi}/2},\varphi_{\pmb{\pi}/2}}&=&h^\pi_{\zeta,\zeta}\ \ \ \ \ \ \ \ \  =\phantom{-}h^\pi_{\varphi,\varphi}
\end{eqnarray}
\end{subequations}
Hence, from \eqref{uvec1a} we infer that $F\bigl(M|\nu_\pi,\nu_{\pmb{\pi}/2}\bigr)=0$ is fulfilled, and thus by
 \begin{equation}\label{uvec2}
\eta_\vartheta=-\varphi\,\sin\vartheta+\zeta\,\cos\vartheta=\zeta\,\sin (\vartheta+\pmb{\pi}/2)+\varphi\,\cos (\vartheta+\pmb{\pi}/2)=\varphi_{(\vartheta+\pmb{\pi}/2)}
\end{equation}
for $0\leq \vartheta\leq \pmb{\pi}/2$ a geodesic arc $\gamma'' \in {\mathcal{C}}^{\nu_{\pmb{\pi}/2},\nu_{\pi}}(M)$ of length $\pmb{\pi}/2$ is implemented, and which is passing through $\nu_{3\pmb{\pi}/4}$ at $\vartheta=\pmb{\pi}/4$, and Definition \ref{geExt} is satisfied, with  ${{\mathcal{C}}}=\gamma'$, ${\hat{\mathcal{C}}}=\gamma''$ and $\omega=\nu_{\pmb{\pi}/2}=\varrho$. Thus, $|\gamma|\cup|\gamma'|\cup|\gamma''|$ is seen to be the supporting set of a geodesic curve of length $\pi$ evolving from $\nu$ to $\nu_\pi=\nu$. In view of \eqref{uvec1} it is clear now that the latter curve admits a geodesic extension beyond $\nu$, too. Such an extension is constructed by considering the uniquely determined geodesic arc $\gamma'''\in {\mathcal{C}}^{\nu_{3\pmb{\pi}/4},\nu_{5\pmb{\pi}/4}}(M)$. According to \eqref{uvec1a} we have
$
h^\pi_{\varphi_{5\pmb{\pi}/4},\varphi_{3\pmb{\pi}/4}}=-h^\pi_{\varphi_{3\pmb{\pi}/4},\varphi_{\pmb{\pi}/4}}={\mathsf 0}
$.
Hence,  by the vectors
\begin{eqnarray}\nonumber
\psi_\vartheta&=&-\frac{(\zeta+\varphi)}{\sqrt{2}}\,\sin\vartheta+ \frac{(\zeta-\varphi)}{\sqrt{2}}\,\cos\vartheta\\
\label{uvec3}
&=&\phantom{-}\zeta\,\cos (\vartheta+3\pmb{\pi}/4)+\varphi\,\sin \,(\vartheta+3\pmb{\pi}/4)\\
\nonumber
&=&\phantom{-}\varphi_{(\vartheta+3\pmb{\pi}/4)}
\end{eqnarray}
for $0\leq \vartheta\leq \pmb{\pi}/2$ the geodesic arc $\gamma'''$  is implemented.
On the other hand, by a little  trigonometric calculation from this one infers that
\begin{equation*}
\psi_\vartheta=
\begin{cases}
\phantom{-}\eta_{(\vartheta+\pmb{\pi}/4)} & \text{for }\phantom{\pmb{\pi}/}0\leq\vartheta\leq \pmb{\pi}/4\\
&\\
-\varphi_{(\vartheta-\pmb{\pi}/4)}& \text{for }\pmb{\pi}/4\leq\vartheta\leq \pmb{\pi}/2
\end{cases}
\end{equation*}
and which in view of \eqref{uvec2} and \eqref{uvec3} is showing that $|\gamma'''|\subset|\gamma|\cup|\gamma''|$. Therefore, we have that by  $[0,5\pmb{\pi}/4]\ni \theta\longmapsto \varphi_\theta$
a geodesic curve ${\mathcal{C}}$ of length $5\pmb{\pi}/4$ with initial state $\nu$ is implemented which appears as a geodesic extension of $\gamma$ beyond the other endpoint $\varrho$ of $\gamma$. Thereby, the geodesic extension ${\mathcal{C}}$ has been obtained in the demonstrated way by successively overlapping and sticking together the members of the family $\{\gamma,\gamma',\gamma'',\gamma'''\}$ of geodesic arcs. As we have seen, $\theta_0=\pi$ is minimal with the property that for the states $\nu_\theta$ implemented by $\varphi_\theta$ over $M$ the property $\nu_{\theta_0+\vartheta}=\nu_\vartheta$ is fulfilled, for all $\vartheta$ with $0\leq \vartheta\leq \pmb{\pi}/4$. Also, the terminal state of ${\mathcal{C}}$ then is the midpoint state  of $\gamma$. The extension proceedure with the help of $\{\gamma,\gamma',\gamma'',\gamma'''\}$ now can be repeated again and again ad infinitum and thus inductively is proving that the by ${\mathbb R}_+ \ni \theta\longmapsto \varphi_\theta$ in accordance with \eqref{uvec} given  (uniquely determined) maximal geodesic curve ${\mathcal{C}}(\gamma)$ with starting state $\nu$ and obeying  $|\gamma|\subset |{\mathcal{C}}(\gamma)|=\{\nu_\theta: 0\leq \theta\leq \pi\}$ is a geodesic loop,  with length $\pi$ of the cycle.
\end{proof}
Finally, remind the notion of `geodesic'. Suppose that the  continuous map
\begin{equation}\label{geodesic}
{\mathcal C }: ]\alpha,\beta\,[\ni \theta\longmapsto \nu_\theta\in {\mathcal S}(M)
\end{equation}
is defining a normal parameterization of
a continuous curve without endpoints,  where $\alpha=-\infty$ and/or   $\beta=+\infty$ are possible.
That is, for each choice of reals $\theta_1, \theta_2$ obeying  $\alpha<\theta_1\leq \theta_2<\beta$ the condition \eqref{Bpara4} is supposed to be satisfied. If  ${\mathcal{C}}|I$ is a geodesic curve for each compact subinterval $I\subset  ]\alpha,\beta\,[$, and if ${\mathcal{C}}$ is maximal with this property, then ${\mathcal{C}}$ is said  to be a `geodesic'.  It is easily seen that to each shortest path $\gamma\in {\mathcal{C}}^{\nu,\varrho}(M)$ there exists a geodesic ${\mathcal{C}}$ to which the inner of $\gamma$ belongs, that is, which is obeying $|\gamma|\backslash \{\nu,\varrho\}\subset |{\mathcal{C}}|$. Especially, if $\gamma$ is admitting a geodesic extension beyond both its endpoints $\nu$ and $\varrho$, then $\gamma$ can be extended to a geodesic. Thereby, according to Corollary \ref{nobipass}, in any case the geodesic associated to $\gamma$ will be uniquely determined. We remark that the study of a geodesic  \eqref{geodesic} is equivalent to the study of the two halfgeodesics ${\mathcal{C}}|]\alpha,\delta]$ and ${\mathcal{C}}|[\delta,\beta\,[$, for each $\delta$ with $\alpha<\delta<\beta$.

Note that locally around each state $\nu_\theta$ a geodesic \eqref{geodesic} looks like a geodesic arc connecting some neighbouring states $\nu_{\theta_1}$ and $\nu_{\theta_2}$, with $\alpha<\theta_1<\theta<\theta_2<\beta$. Accordingly, the normal parameterization \eqref{para4} implies that $\nu_\theta$ has to obey the dynamical equation
\begin{equation}\label{para7a}
\frac{1}{4}\,\nu^{\,\prime\prime\prime}_\theta +\nu^{\,\prime}_\theta={\mathsf 0}
\end{equation}
 on ${\mathcal S}(M)$, for each $\theta\in ]\alpha,\beta\,[$.
The first integral of this equation then reads
\begin{equation}\label{para7b}
\frac{1}{4}\,\nu^{\,\prime\prime}_\theta +\nu_\theta=\mu_{\mathcal{C}}
\end{equation}
with a state $\mu_{\mathcal{C}}\in {\mathcal S}(M)$ given by (see Corollary \ref{normalpara} for the context)
\begin{equation}\label{para8}
\mu_{\mathcal{C}}=\frac{1}{2\,\sin^2\Delta\theta}\Bigl(\nu+\varrho-2\,\cos\Delta\theta\, \re f^\pi_{\zeta,\varphi}\Bigr)
\end{equation}
and where $\nu=\nu_{\theta_1}$, $\varrho=\nu_{\theta_2}$ and $\Delta\theta=\theta_2-\theta_1$ may be chosen in accordance  with  $\alpha<\theta_1<\theta<\theta_2<\beta$ and ${\mathcal{C}}|[\theta_1,\theta_2]$ to be a shortest path. Note that under these circumstances   $F=F(M|\nu,\varrho)=\cos\Delta\theta$ holds. The fact that \eqref{para8} is a state follows since $0\leq F< 1$ holds and
$ f^\pi_{\zeta,\varphi}\in {\mathfrak F}(M|\nu,\varrho)$ is satisfied,
see Definition \ref{Fset}.
\begin{remark}\label{osc}
\begin{enumerate}
\item \label{osc1}
Refer to  \eqref{para8} as `osculating center' of the geodesic ${\mathcal{C}}$, and which in view of \eqref{para7b} extends the notion of the `center of the osculating circle' known from classical curve theory to a context of states.
\item\label{osc3}
In line with Lemma \ref{bifurc}\,(c) and formula \eqref{para8}, shortest paths with the same two endpoints $\nu,  \varrho$ may be uniquely distinguished by the respective osculating centers of the geodesics which are associated to the respective shortest paths.
\end{enumerate}
\end{remark}

\subsubsection{Convexity and geodesic convexity}\label{geoconvexity}
There exist various notions of convexity for subsets in context of spaces with inner metrics, see \cite[\S 18.]{Rino:61} and the references cited therein, and which terms are dating back to works of K. Menger \cite{Meng:30}, W. Wald \cite{Wald:33} and A.D. Alexandrow \cite{Alex:55}, essentially. Also, in modified form some of these notions are common in context of Riemannian manifolds  \cite{Whit:32,Chav:93}. Those types of convexities have to be well-distinguished from the property  of (affine) convexity of a set and referring to a possibly underlying affine-linear structur of the metric space in question. In view of Theorem  \ref{gearcuni},  Lemma \ref{bifurc}\,(c)--(d) and Definition \ref{geodisj}, in case of the Bures metric, since geodesic arcs instead of shortest paths  synonymously can be referred to, the various types of convexities may be expressed in terms of geodesic arcs and then definitions read as follows.
\begin{definition}\label{geoconvex}
A  subset  $\Omega\subset {\mathcal S}(M)$ consisting of more than one state is termed
\begin{enumerate}
\item \label{geoconvex00}
 `midpoint convex' if for any choice $\nu, \varrho\in \Omega$ with $\nu\not=\varrho$ $$\ \ \ \ \ \ \ \ \ \ \ \ \ \ \, \exists \mu\in \Omega :\ {\mathfrak d}_B(M|\nu,\mu)=\frac{1}{2}\,{\mathfrak d}_B(M|\nu,\varrho),\ {\mathfrak d}_B(M|\mu,\varrho)=\frac{1}{2}\,{\mathfrak d}_B(M|\nu,\varrho);\ \ \ \ \ \ \ $$
\item \label{geoconvex0}
 `metrically convex' if for any choice $\nu, \varrho\in \Omega$ with $\nu\not=\varrho$ $$\ \ \ \ \ \ \ \ \ \ \ \ \ \ \, \exists \mu\in \Omega :\ {\mathfrak d}_B(M|\nu,\mu)=\lambda\,{\mathfrak d}_B(M|\nu,\varrho),\ {\mathfrak d}_B(M|\mu,\varrho)=(1-\lambda)\,{\mathfrak d}_B(M|\nu,\varrho)\ $$
 for some $\lambda\in ]0,1[$;
\item \label{geoconvex1}
 `convex' if for any choice $\nu, \varrho\in \Omega$ with $\nu\not=\varrho$
\[
\exists \gamma\in {\mathcal{C}}^{\nu,\varrho}(M) :\ |\gamma|\subset \Omega;\phantom{\ \ \ \lambda\, ,\ {\mathfrak d}_B(M|\mu,\varrho)=(1-\lambda)\,{\mathfrak d}_B(M|\nu,\varrho)\ }
\]
\item \label{geoconvex2}
`totally convex' if  for any $\nu, \varrho\in \Omega$ with $\nu\not=\varrho$
\[
\forall \gamma\in{\mathcal{C}}^{\nu,\varrho}(M):\ |\gamma|\subset \Omega;\phantom{\ \ \ \lambda\, ,\ {\mathfrak d}_B(M|\mu,\varrho)=(1-\lambda)\,{\mathfrak d}_B(M|\nu,\varrho)\ }
\]
\item \label{geoconvex3}
`geodesically convex' if $\Omega$ is a convex, arc-determining subset.
\end{enumerate}
\end{definition}
Typical examples of these notions come along with the considerations around Radon-Nikodym theorems, see \eqref{stratum} and subsection \ref{lstrucarcs}.
\begin{example}\label{exccconv}
\begin{enumerate}
\item \label{exccconv2}
The inner of each geodesic arc is    geodesically convex.
\item \label{exccconv1}
The closure  of the  $\tau$-stratum $\Omega_M(\tau)$ is totally convex, for any  state $\tau$.
\end{enumerate}
\end{example}
\begin{proof}
 The inner $|\gamma|\backslash \{\nu,\varrho\}$  of a geodesic arc $\gamma\in {\mathcal{C}}^{\nu,\varrho}(M)$ is  satisfying
$
|\gamma|\backslash \{\nu,\varrho\}=\cup_{n\in {{\mathbb{N}}}} |\gamma_n|
$
with $\gamma_n=\gamma|[1/n,1-1/n]$, for each $n\in {{\mathbb{N}}}$. In view of Example \ref{geodisj2}\,\eqref{geodisj2a} the  validity of  \eqref{exccconv2} follows.
Note that, by Lemma \ref{setprop}, each stratum $\Omega_M(\tau)$ is obeying  $$[\Omega_M(\tau)]=\Omega^0_M(\tau)$$ see \eqref{stratum} for  definitions. In view of Definition \ref{geoconvex}\,\eqref{geoconvex2} and Corollary \ref{geoconv} the validity of \eqref{exccconv1} then is following at once.
\end{proof}
Consider the following non-trivial example of geodesic convexity  for $M={\mathsf B}({\mathcal H})$.
\begin{example}\label{Bopex}
Let ${\mathcal P}=\{P_k\}$ be a finite or countably infinite, strictly ascending directed system of finite orthoprojections $P_k$ over ${\mathcal H}$, with  ${\mathsf{l.u.b.}} \{P_k\}={\mathsf 1}$. Then,
\begin{equation*}
\Omega_{\mathcal P}({\mathcal H})=\bigl\{\varrho\in {\mathcal S}_0^{\mathsf{faithful}}({\mathsf B}({\mathcal H})): \sigma_\varrho P_k=P_k \sigma_\varrho,\,\forall k \bigr\}
\end{equation*}
is a geodesically convex subset of states over $M={\mathsf B}({\mathcal H})$ (by $\sigma_\varrho$ the density operator over ${\mathcal H}$ which is uniquely associated with $\varrho$ is meant).
\end{example}
\begin{proof}
By Corollary \ref{geodisj1}, since $\Omega_{\mathcal P}({\mathcal H})$ consists of faithful normal states, it has to be an arc-determining set. Accordingly, for each two states  $\varrho,\nu\in\Omega_{\mathcal  P}({\mathcal H})$, when considering the $^*$-representation $\pi$ with  bounded linear operators acting by left multiplication on Hilbert-Schmidt operators over ${\mathcal H}$, then in respect of $\pi$ the states of the unique geodesic arc $\gamma\in {\mathcal{C}}^{\nu,\varrho}({\mathsf B}({\mathcal H}))$ can be implemented by the (normalized) positive linear combinations of the Hilbert-Schmidt operators $\sqrt{\sigma_\nu}$ and $\xi\bigl(\sqrt{\sigma_\nu}\bigr)$ figuring in context of Example \ref{Bop}. Especially, the state $\omega$ of $\gamma$ which gets  implemented by the vector  proportional to $$\alpha\,\xi\bigl(\sqrt{\sigma_\nu}\bigr)+\beta\sqrt{\sigma_\nu}$$
with $\alpha\geq 0$, $\beta>0$, by assumption about $\varrho$ and $\nu$ and by Example \ref{Bop1} and according to Remark \ref{Bop1c}\,\eqref{Bop1cc} has full support and its density operator has to commute with each $P_k$ again.   Thus, each state $\omega\in |\gamma|$ has to belong to $\Omega_{\mathcal P}({\mathcal H})$ again. But then, $|\gamma|\subset \Omega_{\mathcal P}({\mathcal H})$ for any $\nu, \varrho\in \Omega_{\mathcal P}({\mathcal H})$, i.e.~ geodesic convexity follows.
\end{proof}
\begin{remark}\label{Bopex1}
$\Omega_{\mathcal P}({\mathcal H})$ is an example of a non-closed, affinely convex set, which in addition is geodesically convex in respect of the Bures metric. Thereby, the dimensionality of $\Omega_{\mathcal P}({\mathcal H})$ is the same as the dimensionality of the underlying Hilbert space ${\mathcal H}$ is (for this see the remark at the beginning of the proof of Example \ref{Bop1}).
\end{remark}
There is a natural gradation in strength among these types of convexities. In particular, each implication which is in accordance with the following graph
\begin{equation*}
\begin{array}{ccccccc}
&&&&&& \text{midpoint convexity}\\
 \text{geodesic convexity} & \Rightarrow &\text{total convexity}& \Rightarrow &\text{convexity}& \stackrel{\rotatebox{45}{$\Rightarrow$}}{\rotatebox{-45}{$\Rightarrow$}}& \Downarrow \\
 &&&&&& \text{metric convexity}
\end{array}
\end{equation*}
 holds true.
When equipped with the Bures metric the state space ${\mathcal S}(M)$ itself is a convex metric space (and is totally convex as a set even), and which then especially implies metric convexity. More generally, it is  known that in a complete metric space with inner metric a closed subset is convex if, and only if, it is metrically convex (or equivalently, iff it is midpoint convex).
Also, in certain cases of non-closed subsets the mentioned equivalence between convexity and midpoint convexity may persist to hold. As an interesting case the example of a stratum $\Omega_M(\tau)$, $\tau\in {\mathcal S}(M)$, will be considered now.

 To start with, suppose $\gamma\in {\mathcal{C}}^{\nu,\varrho}(M)$, with  $\nu,\varrho\in \Omega_M(\tau)$, and be $\{\pi,{\mathcal H}_\pi\}$ a unital $^*$-representation with non-trivial $\pi$-fibres of $\nu$ and $\varrho$. Let $\varphi\in {\mathcal S}_{\pi,M}(\nu)$ and $\xi\in {\mathcal S}_{\pi,M}(\varrho)$ such that $\langle\xi,\varphi\rangle_\pi=F(M|\nu,\varrho)$. Let $N=\pi(M)^{\,\prime\prime}$, and be $A_0$ the densely defined over $[N^{\,\prime}\varphi]$ linear operator given by $A_0:  z\varphi\longmapsto z\xi$, for all $z\in N^{\,\prime}$. The following holds.
\begin{lemma}\label{beicconv1} Suppose $\gamma\in {\mathcal{C}}^{\nu,\varrho}(M)$, with  $\nu,\varrho\in \Omega_M(\tau)$. Then,
\begin{equation}\label{geoalter}
|\gamma|\cap \Omega_M(\tau)=
\begin{cases}
\phantom{\,\nu}|\gamma| & \text{ if $A_0$ is essentially self-adjoint;}\\ \\
\{\nu,\varrho\} & \text{else.}
\end{cases}
\end{equation}
\end{lemma}
\begin{proof}
Suppose $\nu,\varrho\in \Omega_M(\tau)$, $\nu\not=\varrho$. Following Example \ref{geodisj}\,\eqref{geodisj2c},  the geodesic arc $\gamma\in {\mathcal{C}}^{\nu,\varrho}(M)$ in quest has to be uniquely determined. Let  $\{\pi,{\mathcal H}_\pi\}$ be a unital $^*$-representation  and $\xi\in {\mathcal S}_{\pi,M}(\varrho)$, $\varphi\in {\mathcal S}_{\pi,M}(\nu)$ obeying $\langle\xi,\varphi\rangle_\pi=F(M|\nu,\varrho)=F$. Then, with respect to $\pi$ the geodesic arc  $\gamma:[0,1]\ni t\longmapsto\nu_t$ at each parameter value $t$ gets implemented by the vectors  $\varphi_t=t\,\xi+\lambda(t)\,\varphi$. Let $\omega\in |\gamma|$ be the midpoint state of $\gamma$. Then $\omega=\nu_s$ has to be fulfilled,  with the unique $s\in ]0,1[$ satisfying $s=\lambda(s)$, that is, $\varphi_s=\alpha (\xi+\varphi)$ holds, with $\alpha^2=1/2(1-F)$. Note that for a stratum by definition \eqref{stratum0} and in view of Theorem \ref{absstet} we have $\omega\in \Omega_M(\tau)$ if, and only if $p_\pi(\xi+\varphi)=p_\pi(\varphi)=p_\pi(\xi)$. By Lemma \ref{supp5} it is known that this can happen if, and only if, the densely defined, positive linear operator $A_0$ given in   \eqref{ope} is essentially self-adjoint. Trivially, the latter is the same as requiring $\varepsilon A_0$ to be  essentially self-adjoint, for any given fixed $\varepsilon > 0$. Thus, by Lemma \ref{supp5} once more again, $p_\pi(\xi+\varphi)=p_\pi(\varphi)$ is equivalent to  $p_\pi(\varepsilon\,\xi+\varphi)=p_\pi(\varphi)$, for any given fixed  $\varepsilon > 0$. Note that since each of the vectors $\varphi_t$, with $0<t<1$, is proportional to $\varepsilon\,\xi+\varphi$, for some suitably chosen $\varepsilon >0$, this means that $p_\pi(\varphi_s)=p_\pi(\varphi)$ is equivalent to   $p_\pi(\varphi_t)=p_\pi(\varphi)$.  Thus, $\omega\in \Omega_M(\tau)$ and $\nu_t\in \Omega_M(\tau)$ are mutually equivalent assertions, for any $t\in ]0,1[$. Hence, if  $|\gamma|\cap \Omega_M(\tau)\supsetneq
\{\nu,\varrho\}$ then $|\gamma|\cap \Omega_M(\tau)=|\gamma|$, and the latter happens iff $\omega\in \Omega_M(\tau)$. As mentioned, the latter condition is the same as claiming  $A_0$ to be essentially self-adjoint. Thus \eqref{geoalter} is seen to be true.
\end{proof}
\begin{corolla}\label{beicconv}
$\Omega_M(\tau)$ is convex iff it is metrically \textup{(}or midpoint\textup{)} convex.
\end{corolla}
\begin{proof}
By the previous, if $\Omega=\Omega_M(\tau)$ is chosen in Definition \ref{geoconvex}, we have to show that the implication  $\eqref{geoconvex0}\Rightarrow\eqref{geoconvex1}$ is true.
In line with this, suppose  $\Omega_M(\tau)$ to be metrically convex. For $\nu, \varrho\in \Omega_M(\tau)$ with $\nu\not=\varrho$, let $\mu\in \Omega_M(\tau)$ be chosen in accordance with  the defining condition for metric convexity by Definition \ref{geoconvex}\, \eqref{geoconvex0}. Then,  by Lemma \ref{besonders} and due to Example \ref{geodisj}\,\eqref{geodisj2c}, $\mu\in |\gamma|\backslash \{\nu,\varrho\}$ has to be fulfilled, for the unique geodesic arc $\gamma$ in $ {\mathcal{C}}^{\nu,\varrho}(M)$. Thus, since  $|\gamma|\cap\, \Omega_M(\tau)\not=\{\nu,\varrho\}$ is seen,  in view of Lemma \ref{beicconv1} $|\gamma|\subset \Omega_M(\tau)$ follows, and which is condition \eqref{geoconvex1} of  Definition \ref{geoconvex}. Thus, in supposing metric convexity for  $\Omega_M(\tau)$, convexity in the sense of Definition \ref{geoconvex}\,\eqref{geoconvex1} follows.
\end{proof}
In order to see examples of convex strata it is useful to introduce a class of states on a ${\mathsf C}^*$-algebra $M$ called `finite-type states'.
Consider a state $\tau$ together with corresponding cyclic  $^*$-representation $\pi_\tau$  generated by $\tau$ and acting over the cyclic representation  Hilbert space  ${\mathcal H}_{\pi_\tau}=[\pi_\tau(M)\psi]$, with cyclic vector $\psi_\tau\in {\mathcal S}_{\pi_\tau,M}(\tau)$. $\tau$ is said to be a finite-type state (resp.~a state of finite type), if $p_{\pi_\tau}(\psi_\tau)$ is a finite orthoprojection in $\pi_\tau(M)^{\,\prime\prime}$. Remind that this is the same as requiring that the unique $\tau$-implementing vector state $\tau_{\pi_\tau}$ over $N_\tau$ has finite support $s(\tau_{\pi_\tau})=p_{\pi_\tau}(\psi_\tau)$.
\begin{lemma}\label{ft}
If $\tau$ is a finite-type state,  then each state of $\Omega_M(\tau)$ is of finite type, too. Especially, for $\varrho,\nu\in \Omega_M(\tau)$,  with finite-type $\tau$, there are a unital $^*$-representation $\{\pi,{\mathcal H}_\pi\}$ and cyclic vectors $\varphi$ and $\xi$ within ${\mathcal H}_\pi$ obeying  $\varphi\in {\mathcal S}_{\pi,M}(\nu)$, $\xi\in {\mathcal S}_{\pi,M}(\varrho)$, with $p_\pi(\varphi)=p_\pi(\xi)$ being finite  and  $\langle\xi,\varphi\rangle_\pi=F(M|\varrho,\nu)$.
\end{lemma}
\begin{proof}
Let $\{\pi_0,{\mathcal H}_0\}$ be a  unital $^*$-representation such that the $\pi_0$-fibres of $\tau$, $\nu$ and $\varrho$ are non-trivial. Let $N_0=\pi_0(M)^{\,\prime\prime}$. Also, since $\tau, \nu, \varrho$ belong to the same stratum,  $\varrho_\nu=\nu_\varrho={\mathsf 0}$ and $\tau_\nu=\nu_\tau={\mathsf 0}$ hold, and to given $\varphi\in {\mathcal S}_{\pi_0,M}(\nu)$ there exist uniquely determined $\xi\in {\mathcal S}_{\pi_0,M}(\varrho)$ and $\psi\in {\mathcal S}_{\pi_0,M}(\tau)$ obeying  $\langle\xi,\varphi\rangle_{\pi_0}=F(M|\varrho,\nu)$ and  $\langle\psi,\xi\rangle_{\pi_0}=F(M|\tau,\varrho)$, respectively, see Lemma \ref{diffreteile}, Corollary \ref{uniqueimp}, \eqref{stratum0} and Example \ref{geodisj2}\,\eqref{geodisj2c}.  Especially, on applying Corollary  \ref{uniqueimp}  with respect to the pairs of states $\{\varrho, \nu\}$ and $\{\tau, \nu\}$ respectively, on   $N_0^\prime$ the following can be inferred to hold:
$$p_{\pi_0}^{\,\prime}(\xi)=s(h^{\pi_0}_{\xi,\varphi})=p_{\pi_0}^{\,\prime}(\varphi)\text{ and } p_{\pi_0}^{\,\prime}(\psi)=s(h^{\pi_0}_{\psi,\varphi})=p_{\pi_0}^{\,\prime}(\varphi)$$
Hence, we have  $p_{\pi_0}^{\,\prime}(\xi)=p_{\pi_0}^{\,\prime}(\varphi)=p_{\pi_0}^{\,\prime}(\psi)$. Let us denote the common orthoprojection by $z\in N_0^\prime$ and define ${\mathcal H}_\pi=z{\mathcal H}_0$ and $$\pi: M\ni x\longmapsto \pi_0(x)z$$ Then,  $\pi$ is a unital $^*$-representation of $M$ over ${\mathcal H}_\pi$ with cyclic vectors $\varphi,\xi, \psi\in {\mathcal H}_\pi$. Also, by the $vN$-density theorem  $N=\pi(M)^{\,\prime\prime}=N_0 z$ is seen. Especially, since $\psi\in {\mathcal S}_{\pi,M}(\tau)$ holds, the linear mapping
$ \pi(x)\psi\longmapsto\pi_\tau(x)\psi_\tau$,
with $x$ running through all of $M$, owing to ${\mathcal H}_\pi=[\pi(M)\psi]$ and ${\mathcal H}_{\pi_\tau}=[\pi_\tau(M)\psi_\tau]$   uniquely extends to an unitary mapping $U\in {\mathsf B}({\mathcal H}_\pi,{\mathcal H}_{\pi_\tau})$ such that $U\pi\, U^*=\pi_\tau$.
Therefore, $U N U^*=N_\tau$, and $U p_\pi(\psi) U^*=p_{\pi_\tau}(\psi_\tau)$. Since by assumption   $p_{\pi_\tau}(\psi_\tau)$ is a finite orthoprojection of $N_\tau$, $p_\pi(\psi)$ is a finite orthoprojection of $N$, too.  Finally, owing to $\psi\in {\mathcal S}_{\pi,M}(\tau)$, $\varphi\in {\mathcal S}_{\pi,M}(\nu)$ and $\xi\in {\mathcal S}_{\pi_0,M}(\varrho)$, and since $\tau$, $\nu$ and $\varrho$ are in the same stratum,
$
p_\pi(\xi)=p_\pi(\varphi)=p_\pi(\psi)
$
holds, with the finite orthoprojection $p_\pi(\psi)
$. This completes the proof of the assertions relating $\varrho$ and $\nu$. Moreover, since $\varrho$ and $\nu$  can be arbitrarily chosen from $\Omega_M(\tau)$, this in view of the above is showing that each other state there is of finite type, too.
\end{proof}
\begin{corolla}\label{fingeoconv}
Let $M$ be a unital ${\mathsf C}^*$-algebra. Then the following facts hold true.
\begin{enumerate}
\item \label{fingeoconv0}
Each pure state over $M$ is of  finite type.
\item \label{fingeoconv2}
For finite dimensional $M$, each state is finite-type.
\item \label{fingeoconv1}
The stratum $\Omega_M(\tau)$ of each finite-type state $\tau$ is geodesically convex.
\end{enumerate}
\end{corolla}
\begin{proof}
By standard theory, ${\mathcal{S}}(M)$ admits extremal elements (pure states). Let $\tau\in {{\mathsf{ex}}}\,{\mathcal S}(M)$, that is $\tau$ be a pure state. Then, it is well-known that the  $\tau$-associated cyclic $^*$-representation $\pi_\tau$ of $M$ is irreducibly acting, that is, $\pi_\tau(M)^{\,\prime\prime}={\mathsf B}({\mathcal H}_{\pi_\tau})$ holds. Therefore, the $\tau$-associated vector state has a minimal orthoprojection as support, and thus \eqref{fingeoconv0} is seen. To see \eqref{fingeoconv2}, let  $M$ be dimensionally finite. Then, the $vN$-algebra $N=\pi(M)^{\,\prime\prime}$ generated by any $^*$-representation will be  dimensionally finite, too. Thus, any state  is of finite type. Finally, in order to see \eqref{fingeoconv1},
let $\tau$ be a finite-type state. To be non-trivial, suppose $\#\,\Omega_M(\tau)>1$. Consider $\nu,\varrho\in \Omega_M(\tau)$, with $\varrho\not=\nu$.  Let  a unital $^*$-representation $\{\pi,{\mathcal H}_\pi\}$ be chosen as in Lemma \ref{ft}. By assumption $p_\pi(\varphi)$ is a finite orthoprojection of  $N=\pi(M)^{\,\prime\prime}$. But then, following Example \ref{suppfin} we may conclude that  $p_\pi(\xi+\varphi)=p_\pi(\varphi)$ (refer to the notations used in Lemma \ref{ft}). Hence, the midpoint-state of the (unique) geodesic arc $\gamma$ connecting $\nu$ and $\varrho$ belongs to $\Omega_M(\tau)$, too. Since $\nu\not=\varrho$ could be chosen at will from $\Omega_M(\tau)$, the stratum is midpoint convex. By Corollary \ref{beicconv} convexity of $\Omega_M(\tau)$ follows. Since by Example \ref{geodisj2} strata are arc-determining sets, \eqref{fingeoconv1} is seen.
\end{proof}
Relating non-triviality of strata, it is useful to take notice of the following facts.
\begin{lemma}\label{multistrat0}
If $\nu\in  {\mathsf{ex}}\,{\mathcal{S}}(M)$, then $\#\,\Omega_M(\nu)=1$.
\end{lemma}
\begin{proof}
Suppose $\varrho\in \Omega_M(\nu)$, see \eqref{stratum0}, and let $\{\pi,{\mathcal{H}}_\pi\}$, $\varphi\in {\mathcal{S}}_{\pi,M}(\nu),\,\xi \in {\mathcal{S}}_{\pi,M}(\varrho)$ be chosen as in Theorem \ref{suppabsstet}. Since both $\varrho\dashv \nu$ and  $\nu\dashv \varrho$ are fulfilled, in respect of $N=\pi(M)^{\,\prime\prime}$ and $N^{\,\prime}$, respectively, from Theorem \ref{suppabsstet}\,\eqref{supp2}, \eqref{supp3} we infer that
\begin{equation}\label{multistrat0a}
p_\pi(\varphi)=p_\pi(\xi),\ p^{\,\prime}_\pi(\varphi)=p^{\,\prime}_\pi(\xi)
\end{equation}
are fulfilled. Let ${\mathcal{H}}_1=p^{\,\prime}_\pi(\varphi){\mathcal{H}}_\pi=[\pi(M)\varphi]$, and be $\pi_1$ the unital $^*$-representation obtained by restricting the action of $\pi$ onto ${\mathcal{H}}_1$. Thus $\varphi\in {\mathcal{H}}_1$ is cyclic under the action of $\{\pi_1,{\mathcal{H}}_1\}$ and is implementing $\nu$. But since $\nu$ is pure, $\nu\in  {\mathsf{ex}}\,{\mathcal{S}}(M)$, we must have that $\pi_1$ is irreducibly acting over ${\mathcal{H}}_1$,  $\pi_1(M)^{\,\prime\prime}={\mathsf{B}}({\mathcal{H}}_1)$, and the minimal orthoprojection $p_\pi(\varphi)p^{\,\prime}_\pi(\varphi)$ in ${\mathsf{B}}({\mathcal{H}}_1)$ is the support of the vector state $\nu_{\pi_1}$. Also, according to the second of the relations in \eqref{multistrat0a}, we have $\xi\in {\mathcal{H}}_1$, and $\xi$ is a cyclic implementing vector for $\varrho$ in respect of $\pi_1$. Accordingly, also $\varrho$ has to be a pure state. Note that the vector state $\varrho_{\pi_1}$ then has minimal in ${\mathsf{B}}({\mathcal{H}}_1)$ orthoprojection $p_\pi(\xi)p^{\,\prime}_\pi(\xi)$ as support. Since owing to \eqref{multistrat0a} we have $p_\pi(\xi)p^{\,\prime}_\pi(\xi)=p_\pi(\varphi)p^{\,\prime}_\pi(\varphi)$, by both $\varphi$ and $\xi$ the same vector state is  implemented over ${\mathsf{B}}({\mathcal{H}}_1)$, that is, $\varrho_{\pi_1}=\nu_{\pi_1}$ has to hold. From this $\varrho=\nu$ follows.
\end{proof}
\begin{corolla}\label{multistrat}
$\#\,\Omega>1$ holds if, and only if, $\Omega=\Omega_M(\nu)$, with $\nu\not\in {\mathsf{ex}}\,{\mathcal{S}}(M)$.
\end{corolla}
\begin{proof}
Note that according to Lemma \ref{multistrat0}, in supposing $\#\,\Omega > 1$ for a stratum $\Omega=\Omega_M(\nu)$ we have to conclude that $\nu\not\in {\mathsf{ex}}\,{\mathcal{S}}(M)$.
To see the other way around, suppose  $\Omega=\Omega_M(\nu)$, for $\nu\not\in {\mathsf{ex}}\,{\mathcal{S}}(M)$. Then, $\nu=(\varrho+\mu)/2$, for states $\varrho, \mu\in \mathsf{S}(M)$ with $\varrho\not=\nu$.
Hence $\varrho\leq 2\nu$, from which $\varrho\dashv \nu$ follows, see Remark \ref{saka}\,\eqref{saka1}. In view of \eqref{vorstratum} this is the same as $\varrho\in \Omega^0_M(\nu)$. Thus $\#\,\Omega^0_M(\nu)>1$ holds. But by  Lemma \ref{setprop}\,\eqref{setprop3}  $\Omega_M(\nu)$ is dense in  $\Omega^0_M(\nu)$. Thus $\#\,\Omega >1$ follows.
\end{proof}

\section{Tangent forms}\label{TF}
The aim of this section is to identify the set of those linear forms over a unital ${\mathsf C}^*$-algebra $M$ that, at a given fixed state $\nu$, can arise as tangent forms along parameterized curves which are passing through this state. However, as usually  only such curves $\gamma$ will be taken into account which were admitting $\gamma$-compliant unital $^*$-representations of $M$ around $\nu$. As will be shown, when making reference to Definition \ref{e.1} and formula \eqref{tangentnorm}, the subsequently defined set will meet the requirements, essentially.
\begin{definition}\label{trdef}
Let $M$ be a unital ${\mathsf C}^*$-algebra. For $\nu\in {\mathcal S}(M)$ consider the set
\begin{equation}\label{tspace}
{\mathsf{T}}_\nu(M)=\bigl\{f\in M^*_h:\,f({\mathsf 1})=0,\,\|f\|_\nu<\infty\bigr\}
\end{equation}
and which henceforth will be referred to  as `tangent space at $\nu$'.
\end{definition}
To prepare for the proof some preliminary considerations will be useful.
\subsection{Basic facts about tangent forms}\label{tforms} Let $I\ni t\,\mapsto\,\nu_t\in
{\mathcal S}(M)$ be a parameterized curve $\gamma$ passing through a state $\nu$ at parameter value
$t=0$ and thereby admitting a $\gamma$-compliant unital *-representation $\{\pi,{\mathcal H}_\pi\}$ around $\nu$. Let $I_\pi\ni t\mapsto \varphi_t\in {\mathcal S}_{\pi,M}(\nu_t)$ be a corresponding differentiable at $t=0$ implementation of $\gamma$, with tangent vector at $\nu$
$$\psi=\psi(\pi,{\mathcal H}_\pi,(\varphi_t))\in {\mathsf{T}}_\nu(M|\gamma)$$
Following Remark \ref{iimpdi}\,\eqref{iimpdi1}
the tangent form $f=\nu^{\,\prime}_t|_{t=0}$ of $\gamma$ at $\nu$ exists and
reads as
\begin{equation}\label{tform1}
    f(x)=\re f_{\xi,\varphi}(\pi(x))\,,\text{ with }\xi=2\psi\,,
\end{equation}
for all $x\in M$. For two vectors $\xi,\chi\in {\mathcal H}_\pi$, $f_{\xi,\chi}$ is the
normal (i.e.~ultrastrongly continuous) linear form that is defined on the $vN$-algebra $N=\pi(M)^{\,\prime\prime}$ as
$$f_{\xi,\chi}(x)=
\langle x\xi,\chi\rangle_\pi$$ for all $x\in N$. Moreover, whenever $g$ is a bounded linear form,
by $\re g$ the hermitian linear form of the real part of $g$ will be meant: $\re g=(g+g^*)/2$.
 Thus especially, $$\re f_{\xi,\varphi}(x)=\frac{1}{2}\bigl(\langle x\xi,\varphi\rangle_\pi+ \langle x\varphi,\xi\rangle_\pi\bigr)$$ for any $x\in N$.
 From this and \eqref{iimpdi1a} the validity of \eqref{tform1} then is evident.
Note that for each hermitian $x\in N_{\mathsf h}$ one has
$$\re f_{\xi,\varphi}(x)=\frac{1}{2}\bigl(\langle \xi,x\varphi\rangle_\pi+ \langle x\varphi,\xi\rangle_\pi\bigr)$$
 Thus, since for $x=x^*$ one has $\langle x\varphi,\xi\rangle_\pi=\overline{\langle \xi,x\varphi\rangle}_\pi$ (complex conjugate),  and since according to the Cauchy-Schwarz inequality the estimate
 $$|\langle \xi,x\varphi\rangle_\pi|\leq \|\xi\|_\pi\,\sqrt{\langle x^2\varphi,\varphi\rangle_\pi}$$ is fulfilled,
 one has the following fundamental estimate:
\begin{subequations}\label{CSI0}
\begin{equation}\label{CSI}
  |\re f_{\xi,\varphi}(x)|\leq \|\xi\|_\pi\,\sqrt{\langle x^2\varphi,\varphi\rangle_\pi}\,,\text{ for each }x=x^*,\, x\in N
\end{equation}
On the other hand, in case of positive $x$, $x\in N_+$, application
of the Cauchy-Schwarz inequality is showing that $$|\langle
\xi,x\varphi\rangle_\pi|=|\langle
\sqrt{x}\xi,\sqrt{x}\varphi\rangle_\pi|\leq
\|\sqrt{x}\,\xi\|_\pi\,\|\sqrt{x}\,\varphi\|_\pi$$ Hence, \eqref{CSI} can be
supplemented by
\begin{equation}\label{CSI1}
  |\re f_{\xi,\varphi}(x)|\leq \|\sqrt{x}\,\xi\|_\pi\,\sqrt{\langle x\,\varphi,\varphi\rangle_\pi}\,,\text{ for each }x\in
  N_+\ \
\end{equation}
\end{subequations}
As subsequently will be shown, in
the sense of Remark \ref{rem1}\,\eqref{rem11},  the structure of the candidates
\eqref{tform1} for tangent forms at $\nu$ over the ${\mathsf C}^*$-algebra $M$ completely can be explained by  the structure of the hermitian linear forms of the type $\re f_{\xi,\varphi}$
with fixed $\varphi\in {\mathcal S}_{N}(\nu_\pi)={\mathcal S}_{\pi,M}(\nu)$ over the $vN$-algebra $N=\pi(M)^{\,\prime\prime}$. Note that $\nu_\pi$ is a (uniquely determined by $\nu$)  vector state over $N$. Hence, if seen in more  abstract settings, the problem is as follows: provided a vector state $\nu$ over a general $vN$-algebra $N$ acting on a Hilbert space ${\mathcal H}$ with scalar product $\langle\cdot,\cdot\rangle$ is given, then for fixed $\varphi\in {\mathcal S}_{N}(\nu)$ the structure of the set of all forms  $\re f_{\xi,\varphi}$ with $\xi\in {\mathcal H}$ will be the object of investigation.
Clearly, when omitting the subscript $\pi$ within  each of the above mentioned expressions referring to $N=\pi(M)^{\,\prime\prime}$ and ${\mathcal H}_\pi$ exclusively, see  e.g.\,\eqref{CSI0}, a true expression in this more abstract context will be obtained. Keeping in mind the latter, the analysis of the problem will be started now.
\subsubsection{Hermitian forms
associated with a vector state on a $vN$-algebra}\label{euk} Let $N$ be a $vN$-algebra acting
on a Hilbertspace ${\mathcal H}$, with scalar product $\langle\cdot,\cdot\rangle$. Let $\nu$ be a vector state on
$N$, that is, a state  with non-trivial fibre ${\mathcal
S}_N(\nu)={\mathcal S}_{{\mathop{id}},N}(\nu)$.  Refer to a
hermitian linear form $g$ on $N$ as `associated
with $\nu$' if $g=\re f_{\psi,\varphi}$, for $\psi\in {\mathcal H}$
and  $\varphi\in{\mathcal S}_N(\nu)$. Assume now also $\chi\in
{\mathcal S}_N(\nu)$. By \eqref{bas4}, since there exists a partial
isometry $v\in N^{\,\prime}$ with $v^*v\geq p^{\,\prime}(\chi)$ and
$\varphi=v\chi$, one then has that $g=\re f_{\psi,\chi}=\re
f_{\psi,v\varphi}=\re f_{v^*\psi,\varphi}$. Thus, in order to
consider the real linear space ${\mathfrak H}_\nu(N)$ of all hermitian linear
forms over $N$ which are associated with $\nu$, we can fix $\varphi\in
{\mathcal S}_N(\nu)$ and then find that
\begin{subequations}\label{euk0}
\begin{equation}\label{euk.1}
{\mathfrak H}_\nu(N)=\bigl\{ g=\re f_{\psi,\varphi}:\,\forall\,\psi\in
{\mathcal H}\bigr\}\,.
\end{equation}
On this space a norm
$|\cdot|_\nu$ can be introduced as follows:
\begin{equation}\label{euk.2}
{\mathfrak H}_\nu(N)\ni g\longmapsto|g|_\nu=\inf\bigl\{\|\psi\| :\, \psi\in {\mathcal H}\mbox{ with }g=\re
f_{\psi,\varphi}\bigr\}\,.
\end{equation}
\end{subequations}
Since absolute homogeneity is obvious, to see that $|\cdot|_\nu$ is
a norm, one is required to show strict positivity and subadditivity.
Suppose $|g|_\nu=0$. Then, in accordance with \eqref{euk.2} there
exist a null-sequence $\{\psi_n\}\subset {\mathcal H}$ with $g=\re
f_{\psi_n,\varphi}$, for each $n\in {{\mathbb{N}}}$. Hence, for each
$x\in N_{\mathrm{h}}$, according to \eqref{CSI}, $$|g(x)|\leq
\|\psi_n\|\,\sqrt{\nu(x^2)}$$ is fulfilled. Upon going to the
limit $n\to\infty$ the relation $|g(x)|=0$ is inferred. Hence, since
the hermitian elements of $N$ separate the linear forms, $g=0$
follows, and strict positivity is seen.  To see subadditivity, let
$g,f\in {\mathfrak H}_\nu(N)$, and be $\varepsilon >0$. In line with
\eqref{euk.2} let $\psi,\tilde{\psi}\in {\mathcal H}$ and obeying
$g=\re f_{\psi,\varphi}$, $f=\re f_{\tilde{\psi},\varphi}$ and
$\|\psi\|\leq |g|_\nu+\varepsilon$, $\|\tilde{\psi}\|\leq
|f|_\nu+\varepsilon$, respectively. Then $g+f=\re
f_{\psi+\tilde{\psi},\varphi}$ follows. Hence, by subadditivity of
$\|\cdot\|$ and \eqref{euk.2},
$$|g|_\nu+|f|_\nu+2\varepsilon\geq \|\psi\|+ \|\tilde{\psi}\|\geq
\|\psi+\tilde{\psi}\|\geq |g+f|_\nu$$
follows.
Since $\varepsilon >0$ can be
chosen at will, subadditivity is seen. Thus,
\eqref{euk.2} defines a norm on the real linear space
\eqref{euk.1}.
\begin{lemma}\label{euk1}
Let us consider the real linear
space ${\mathfrak H}_\nu(N)$ as a normed space under
the norm $|\cdot|_\nu$. Then, the linear mapping
\begin{equation}\label{euk.3}
\iota: [N_{\mathrm{h}}\varphi]\ni \psi\,\longmapsto\,\re
f_{\psi,\varphi}\in  {\mathfrak H}_\nu(N)
\end{equation}
gets an isometric isomorphism from the Euclidean subspace
$[N_{\mathrm{h}}\varphi]\subset {\mathcal H}_{{\mathbb{R}}}$ onto
the normed real linear space ${\mathfrak H}_\nu(N)$.
\end{lemma}
\begin{proof}
Let $g=\re f_{\psi,\varphi}$, with arbitrary $\psi\in {\mathcal H}$,
and be $x\in N_{\mathrm{h}}$. Then, from the definition \eqref{euk.2} of
$|g|_\nu$ together with the estimate \eqref{CSI} the estimate
\begin{subequations}\label{fund}
\begin{equation}\label{fund1}
 |g(x)|\leq |g|_\nu\|x\varphi\|
\end{equation}
is seen to hold. Thus $N_{\mathrm{h}}\varphi\ni x\varphi
\,\longmapsto\,g(x)\in {{\mathbb{R}}}$ is continuous on the real linear
subspace $N_{\mathrm{h}}\varphi\subset {\mathcal H}_{{\mathbb{R}}}$. Hence, there exists
unique $\psi_0\in [N_{\mathrm{h}}\varphi]$ and obeying
\begin{equation}\label{fund2}
g(x)=\langle \psi_0,x\varphi\rangle_{{\mathbb{R}}}=
\re \langle x\psi_0,\varphi\rangle,
\end{equation}
\end{subequations}
for each $x\in N_{\mathrm{h}}$. Since the
hermitian elements separate the linear forms, $$g=\re f_{\psi_0,\varphi}$$ follows,
with a unique $\psi_0\in [N_{\mathrm{h}}\varphi]$. Hence, the map \eqref{euk.3} is surjective.

Suppose now that $g\not=0$. Then, $\psi_0\not=0$.
Let $\psi_0=
\lim_{n\to\infty} x_n\varphi$, with a sequence $\{x_n\}\subset N_{\mathrm{h}}$.
Then, upon taking the limit within the estimates arising from  \eqref{fund2} for $x=x_n$,  $$\lim_{n\to\infty}|g(x_n)|=\|\psi_0\|^2$$
can be inferred.
On the other hand, from the estimates  arising from  \eqref{fund1} for $x=x_n$, $$\lim_{n\to\infty}|g(x_n)|\leq
|g|_\nu \lim_{n\to\infty} \|x_n\varphi\|= |g|_\nu\|\psi_0\|$$ is inferred. Taking together both estimates
shows that $\|\psi_0\|\leq |g|_\nu$ has to be fulfilled. By construction of $\psi_0$ and in view of
\eqref{euk.2} it is clear now that equality $\|\psi_0\|=|g|_\nu$ has to occur. Thus, the map $\iota$ defined in \eqref{euk.3} is an isometry.
\end{proof}
Remind that from the canonical decomposition of a hermitian linear form $g$
especially follows that $\|g\|_1=\sup\{|g(x)|: x\in (N_{\mathrm{h}})_1\}$
holds. Having in mind this, from \eqref{fund1} the following can be
inferred to hold:
\begin{equation}\label{euk.4}
\forall\,g\in {\mathfrak H}_\nu(N)\,:\ \|g\|_1\leq |g|_\nu \sqrt{\|\nu\|_1}\,.
\end{equation}
\begin{remark}\label{scalar}
\begin{enumerate}
\item\label{scalar1}
According to Lemma \ref{euk1} the norm $|\cdot|_\nu$ is quadratic
and ${\mathfrak H}_\nu(N)$ is complete under this norm. For the
corresponding inner product on ${\mathfrak H}_\nu(N)$ the notation
$(\cdot,\cdot)_\nu$ will be used. For $x\in N_{\mathrm{h}}$, let a
linear form $\nu_x$ on $N$ be defined by $\nu_x=\re
f_{x\varphi,\varphi}=\re \nu((\cdot)x)$. The real linear space of
all these elements will be denoted by ${\mathfrak K}_\nu(N)$,
${\mathfrak K}_\nu(N)=\{\nu_x: x\in N_{\mathrm{h}}\}$.
 \item\label{scalar2}  Note that ${\mathfrak K}_\nu(N)\subset {\mathfrak H}_\nu(N)$
and according to Lemma \ref{euk1} the Euclidean space ${\mathfrak
H}_\nu(N)$ equivalently can be considered as the completition of
the real linear subspace ${\mathfrak K}_\nu(N)$ under the inner
product $\bigl(\nu_x,\nu_y\bigr)_\nu=\Re\,\nu(xy)$.
 \item\label{scalar3} Note that \eqref{euk.4} says that the  $|\cdot|_\nu$-topology  is at least as strong as the functional norm topology on ${\mathfrak H}_\nu(N)$. As a result of
this, ${\mathfrak H}_\nu(N)$  in the infinite dimensional case in
general
will not be $\|\cdot\|_1$-complete.
\end{enumerate}
\end{remark}
\begin{lemma}\label{normformel}
$\ \forall\,g\in {\mathfrak H}_\nu(N)\,:\ \displaystyle\frac{1}{2}\,|g|_\nu=\|g\|_\nu\,.$
\end{lemma}
\begin{proof}
Suppose $g=\re f_{\psi,\varphi}$, with $\psi\in
[N_{\mathrm{h}}\varphi]$, $\varphi\in{\mathcal S}_N(\nu)$.
Then, by \eqref{CSI1}
\[
|g(x)|\leq
\|\sqrt{x}\psi\|\,\sqrt{\nu(x)}=\sqrt{\langle x\psi,\psi\rangle}\,\sqrt{\nu(x)}
\]
follows, for each $x\in N_+$. Since $g$ is hermitian, in this case $g(x)\in {\mathbb R}$. Hence, for $x\in N_+$ with $\nu(x)\not=0$, according to the previous
\begin{equation}\label{CSI2}
    \frac{g(x)^2}{\nu(x)\phantom{^2}}\leq \langle x\psi,\psi\rangle
\end{equation}
Let  $\{y\}$ be a finite
positive decomposition of the unity in $N$. From \eqref{CSI2} then
\[
{\sum_{k}}^\prime
\frac{g(y_k)^2}{\nu(y_k)\phantom{^2}}\leq \|\psi\|^2
\]
is inferred to hold, with the $^\prime$ indicating that the summation has to be extended
only about those terms with non-vanishing denominator $\nu(y_k)$.
In view of Definition \ref{e.1} and formula \eqref{tangentnorm}, when applied over the $vN$-algebra $M=N$, from the previous
\begin{equation}\label{CSI3}
\forall\,g\in {\mathfrak H}_\nu(N)\,:\ \|g\|_\nu\leq
\frac{1}{2}\,|g|_\nu
\end{equation}
can be followed.
Suppose now $x\in N_{\mathrm{h}}$,
with spectrum consisting of finitely many spectral values only, say
${\mathsf{spec}}(x)=\{\lambda_1,\ldots,\lambda_m\}$ with $m\in {{\mathbb{N}}}$,
and spectral decomposition $x=\sum_j \lambda_j p_j$, with
corresponding orthogonal decomposition $\{p\}$ of the unity into
spectral projections. Let us consider the hermitian linear form $\nu_x\in {\mathfrak H}_\nu(N)$ defined in accordance with Remark \ref{scalar}\,\eqref{scalar1}.
Note that since $\nu_x(p_j)=\lambda_j\nu(p_j)$
holds, for each $j\leq m$, one can conclude as follows:
$$\|\nu_x\|_\nu^2\geq \frac{1}{4}\,{\sum_{j}}^\prime
\frac{\nu_x\left(p_j\right)^2}{\nu(p_j)\phantom{^2}}=\frac{1}{4}\,{\sum_{j}}^\prime
\lambda_j^2\nu(p_j)=\frac{1}{4}\,\nu(x^2)=\frac{1}{4}\,\|x\varphi\|^2=\frac{1}{4}\,|\nu_x|_\nu^2$$
Thereby, according to Remark \ref{scalar}\,\eqref{scalar2} it has been used that $|\nu_x|_\nu=\sqrt{\nu(x^2)}$ has to be fulfilled. By taking square roots one gets  $\|\nu_x\|_\nu\geq \frac{1}{2}\,|\nu_x|_\nu$. In view
to \eqref{CSI3} then equality has to occur, that is, for
$x\in N_{\mathrm{h}}$ with finite spectrum
\begin{equation}\label{CSI4}
    \|\nu_x\|_\nu=
\frac{1}{2}\,|\nu_x|_\nu
\end{equation}
 has to hold. Since each hermitan element of
$N$ can be approximated arbitrarily well in the uniform sense by
hermitian elements of
$N$ with finite spectrum, to given $\psi\in
[N_{\mathrm{h}}\varphi]$ especially one
finds such a sequence $\{x_n\} \subset N_{\mathrm{h}}$ and
obeying $\psi=\lim_{n\to\infty} x_n\varphi$. Hence, in view of Lemma
\ref{euk1} the following is inferred to hold:
\begin{subequations}\label{CSI5}
\begin{equation}\label{CSI5a}
|\cdot|_\nu-\lim_{n\to\infty}
\nu_{x_n}=g=\re f_{\psi,\varphi}
\end{equation}
From the latter and by \eqref{CSI4} we then see that
\begin{equation}\label{CSI5b}
 \frac{1}{2}\,|g|_\nu=
\lim_{n\to\infty} \|\nu_{x_n}\|_\nu\,.
\end{equation}
\end{subequations}
On the other hand, by \eqref{CSI3} and by subadditivity of $\|\cdot\|_\nu$, for any $n$, conclude that
$$
\|\nu_{x_n}\|_\nu=\|g+(\nu_{x_n}-g) \|_\nu\leq \|g\|_\nu+\|\nu_{x_n}-g\|_\nu\leq \|g\|_\nu+\frac{1}{2}\,|\nu_{x_n}-g|_\nu
$$
is fulfilled. Hence,  for all $n$,
$$
\|\nu_{x_n}\|_\nu\leq \|g\|_\nu+\frac{1}{2}\,|\nu_{x_n}-g|_\nu$$
Now, if both the relations of \eqref{CSI5} are taken into account, in the limit for $n\to \infty$  these estimates result in
$\frac{1}{2}\,|g|_\nu\leq \|g\|_\nu$. This together with \eqref{CSI3} yields $$\|g\|_\nu=\frac{1}{2}\,|g|_\nu$$ for each
$g\in{\mathfrak H}_\nu(N)$.
\end{proof}
Now, let $x$ be a hermitian linear operator $x\in N_h$, with spectral representation
\begin{subequations}\label{specrep0}
\begin{equation}\label{specrep}
    x=\int_{\alpha}^{\beta^+} \lambda\,E_x(d\/\lambda)
\end{equation}
with left-continuous (strong operator topology) spectral resolution $E_x$, which for $\lambda\in {\mathbb R}$ is the support of the positive part of
the element $\lambda{\mathsf 1}-x$ within $ N$,
\begin{equation}\label{specrep1}
E_x(\lambda)=s((\lambda{\mathsf 1}-x)_+)
\end{equation}
\end{subequations}
and with $[\alpha,\beta]\subset {\mathbb R}$ compact with ${\mathsf{spec}}(x)\subset [\alpha,\beta]$. Let $\{x_n\}$ be a sequence of approximating hermitian linear operators for $x$ of the form
\[
x_n=\sum_i \lambda_i^{(n)}\,E_x(\Delta \lambda_i^{(n)})
\]
with $ \lambda_1^{(n)}>\lambda_2^{(n)}>\cdots > \lambda_{m_n}^{(n)}$, $m_n\in {{\mathbb{N}}}$, $ \lambda_1^{(n)}> \beta>\alpha\geq  \lambda_{m_n}^{(n)}$, and orthoprojections $$E_x(\Delta \lambda_i^{(n)})=E_x( \lambda_i^{(n)})-E_x( \lambda_{i+1}^{(n)})$$ for $i=1,2,\ldots,m_n-1$. Then $\{E_x( \Delta\lambda_i^{(n)})\}\subset N$  for each $n$ is a finite decomposition of the unity into mutually orthogonal orthoprojections. Also, if $$\lim_{n\to\infty}\max_i |\lambda_i^{(n)}-\lambda_{i+1}^{(n)}|=0$$ is supposed, one can be assured that $x\varphi=\lim_{n\to\infty} x_n\varphi$ is fulfilled. Hence, $\nu_x=\|\cdot\|_\nu-\lim_{n\to\infty} \nu_{x_n}$ and $\nu_x=\|\cdot\|_1-\lim_{n\to\infty} \nu_{x_n}$. From the latter, and since the arguments used in deriving formula \eqref{CSI4} apply to each $x_n$ and then show that
\[
\|\nu_{x_n}\|_\nu=\frac{1}{2}\,\sqrt{{\sum_i}^\prime\frac{\nu_{x_n}\bigl(E_x( \Delta\lambda_i^{(n)})\bigr)^2}{\nu\bigl(E_x( \Delta\lambda_i^{(n)})\bigr)}}
\]
is fulfilled for each $n$, we conclude that
\begin{equation}\label{orthoeuk1}
\frac{1}{2}\,\sqrt{\nu(x^2)}= \|\nu_{x}\|_\nu=\sup_{{\mathcal E}_x}{\frac{1}{2}\,\sqrt{{\sum_i}^\prime\frac{\nu_{x}\bigl(E_x( \Delta\lambda_i)\bigr)^2}{\nu\bigl(E_x( \Delta\lambda_i)\bigr)}}}
\end{equation}
with the supremum extending over the set ${\mathcal E}_x$ of all finite decompositions of the unity into mutually orthogonal orthoprojections $E_x( \Delta\lambda_i)=E_x( \lambda_i)-E_x( \lambda_{i+1})$ which can be obtained for any choice of finitely many reals $\lambda_j$, with $ \lambda_1>\lambda_2>\cdots > \lambda_{m}$, $m\in {{\mathbb{N}}}$ arbitrarily chosen, and thereby obeying $ \lambda_1>\beta>\alpha\geq \lambda_m$.
\begin{remark}\label{orthoeuk}
In view of Remark \ref{scalar}\,\eqref{scalar2},  on a $vN$-algebra $N$ from formula \eqref{orthoeuk1}
\[
\frac{1}{2}\,|f|_\nu=\|f\|_\nu =\sup_{\{e\}}\frac{1}{2}\,\sqrt{{\sum_{j}}^\prime
\frac{f(e_j)^2}{\nu(e_j)}}
\]
is obtained, for each $f\in{\mathfrak H}_\nu(N)$, with the supremum extending over
the set of all finite
decompositions of the unit into mutually
orthogonal orthoprojections of $N$.
\end{remark}

\subsubsection{Density results for the tangent norm}\label{poseuk}
Let ${\mathfrak A}$ be a ${\mathsf C}^*$-algebra of bounded linear operators acting on the Hilbert space ${\mathcal H}$. Suppose ${\mathsf 1}\in {\mathfrak A}$ and $N={\mathfrak A}^{\,\prime\prime}$. Relating formula \eqref{tangentnorm}, if seen in context of $N$ or ${\mathfrak A}$ respectively, in respect of elements of ${\mathfrak H}_\nu(N)$ and vector states on $N$ the following density result holds.
\begin{lemma}\label{dichte}
Suppose $\nu\in {\mathcal S}(N)$  is a vector state and $f\in {\mathfrak H}_\nu(N)$. Then
\begin{equation}\label{dens}
    \|f\|_\nu=\|f|{\mathfrak A}\|_{\nu|{\mathfrak A}}
\end{equation}
where $\nu|{\mathfrak A}$ and $f|{\mathfrak A}$ are the restrictions of $\nu$ and $f$ from $N$ to ${\mathfrak A}$.
\end{lemma}
\begin{proof}
Due to $N\supset {\mathfrak A}$ the set of all finite positive decompositions of the unity within ${\mathfrak A}$ is a subset of all  finite positive decompositions of the unity within $N$. When applying Definition \ref{e.1} to the situation with $\nu$ and $f$, or $\nu|{\mathfrak A}$ and $f|{\mathfrak A}$ respectively,
\begin{equation}\label{dens0}
    \|f\|_\nu\geq \|f|{\mathfrak A}\|_{\nu|{\mathfrak A}}
\end{equation}
is seen.
First we are going to show that \eqref{dens} holds for $f=\nu_x$ with $x\in{\mathfrak A}_{\mathsf h}$, i.e.~
\begin{equation}\label{dens1}
    \|\nu_x\|_\nu=\|\nu_x|{\mathfrak A}\|_{\nu|{\mathfrak A}}
\end{equation}
see Remark \ref{scalar}\,\eqref{scalar1}.
Remark that owing to the property that $ \|\nu_{\lambda x}\|_\nu=\lambda \, \|\nu_x\|_\nu$ holds for any real $\lambda>0$, one may content with verifying \eqref{dens1} for $x\in{\mathfrak A}_{\mathsf h}$ with $\|x\|<1/2$. Suppose this case, and let $\alpha$ and $\beta$ be the lower and upper bound of ${\mathsf{spec}}(x)$, respectively. Let real $\varepsilon > 0$ be chosen. Because of $x\in {\mathfrak A}_{\mathsf h}\subset N_{\mathsf h}$ a spectral representation \eqref{specrep} of $x$ with respect to $N$ exists and formula \eqref{orthoeuk1} then can be applied and shows that reals $ \lambda_1>\lambda_2>\cdots > \lambda_{m}$, $m\in {{\mathbb{N}}}$, with $\lambda_m\leq\beta$ and $\lambda_1 >\alpha$, can be chosen such that
\begin{equation}\label{appl2}
   \|\nu_x\|_\nu \geq {\frac{1}{2}\,\sqrt{{\sum_i}^\prime\frac{\nu_{x}\bigl(E_x( \Delta\lambda_i)\bigr)^2}{\nu\bigl(E_x( \Delta\lambda_i)\bigr)}}}\geq \|\nu_x\|_\nu- \varepsilon
\end{equation}
with $E_x\bigl( \Delta\lambda_i\bigr)=E_x\bigl( \lambda_i\bigr)-E_x\bigl( \lambda_{i+1}\bigr)$, for all $i< m$. By assumption on $x$ we may suppose that $\lambda_1 < 1/2$ and $\lambda_m > -1/2$ hold.
Now, for each $n\in {{\mathbb{N}}}$, let us define
\begin{equation}\label{spec}
    a_i^{(n)}=\sqrt[n]{(\lambda_i\,{\mathsf 1}-x)_+}
\end{equation}
for $i=1,2,\ldots,m$. Remind that if $y\in {\mathfrak A}_{\mathsf h}$ holds, then also the positive and negative parts $y_\pm$ of $y$ both belong to ${\mathfrak A}$ again, and for a positive operator $y\in {\mathfrak A}_+$, the  $n$-th root of $y$ belongs to ${\mathfrak A}_+$ again. Thus, $$\bigl\{a_1^{(n)},a_2^{(n)},\ldots,a_m^{(n)}\bigr\}\subset {\mathfrak A}_+$$ On the other hand, by spectral calculus, with respect to the $vN$-algebra $N$ one has
\[
a_i^{(n)}=\int_{\alpha}^{\lambda_i} \sqrt[n]{\lambda_i-t}\ E_x(d\/t)
\]
From this formula and since  $\frac{1}{2}>\lambda_1>\alpha>-\frac{1}{2}$ are  fulfilled, for each $n\in {{\mathbb{N}}}$
\begin{subequations}\label{specc}
\begin{equation}\label{spec1}
   {\mathsf 1}\geq a_1^{(n)} \geq  a_2^{(n)}\geq \cdots\geq a_m^{(n)}
\end{equation}
and existence of the strong limit of $ a_i^{(n)}$ as $n\to\infty$ with
\begin{equation}\label{spec2}
E_x(\lambda_i)\xi=\lim_{n\to \infty} a_i^{(n)}\xi\,, \text{ for any }\xi\in {\mathcal H}
\end{equation}
\end{subequations}
from \eqref{specrep1} and \eqref{spec} can be followed. Hence, in defining
\begin{subequations}\label{speccc}
\begin{eqnarray}
\label{spec3}
  x_0^{(n)} &=& {\mathsf 1}- a_1^{(n)} \\
\label{spec4}
  x_i^{(n)} &=& a_i^{(n)}- a_{i+1}^{(n)}\,,\text{ for }i=1,2,\ldots,m-1
\end{eqnarray}
\end{subequations}
from \eqref{spec1} and \eqref{spec3} and since,  due to $\lambda_m<\alpha$,   $a_m^{(n)}={\mathsf 0}$ is fulfilled, we see that $$X^{(n)}=\{x_0^{(n)},x_1^{(n)},\ldots,x_{m-1}^{(n)}\}\subset {\mathfrak A}_+$$ has to be a finite positive decomposition of the unity within ${\mathfrak A}$, for each $n\in {{\mathbb{N}}}$. On the other hand, as a consequence of \eqref{spec2} from
\eqref{speccc} one concludes that
\begin{subequations}\label{speclim}
\begin{eqnarray}
\label{spec5}
  {\mathsf 0} &=& \lim_{n\to \infty} x_0^{(n)}\xi \\
\label{spec6}
  E_x(\Delta\lambda_i)\xi &=& \lim_{n\to \infty} x_i^{(n)}\xi\,,\text{ for }i=1,2,\ldots,m-1
\end{eqnarray}
\end{subequations}
and any $\xi\in {\mathcal H}$. But then, owing to  $\nu_x=\re{f_{x\varphi,\varphi}}$, with $\varphi\in {\mathcal S}(\nu)$,
by \eqref{spec6}
\begin{eqnarray*}
  \nu\bigl(E_x(\Delta\lambda_i)\bigr)&=& \lim_{n\to\infty} \nu(x_i^{(n)}) \\
  \nu_x\bigl(E_x(\Delta\lambda_i)\bigr) &=& \lim_{n\to\infty} \nu_x(x_i^{(n)})
\end{eqnarray*}
for $i=1,2,\ldots,m-1$ follow. Hence, in view of Definition \ref{e.1} and by \eqref{appl2}, for $n$ sufficiently large, say for $n\geq n(\varepsilon)$ with some $n(\varepsilon)$, we can be assured that
\begin{equation}\label{app3}
\|\nu_x|{\mathfrak A}\|_{\nu|{\mathfrak A}} \geq \frac{1}{2}\,\sqrt{{\sum_{i\geq 0}}^\prime\frac{\nu_{x}\bigl( x_i^{(n)}\bigr)^2}{\nu\bigl( x_i^{(n)}\bigr)}} \geq  \frac{1}{2}\,\sqrt{{\sum_{i\geq 1}}^\prime\frac{\nu_{x}\bigl( x_i^{(n)}\bigr)^2}{\nu\bigl( x_i^{(n)}\bigr)}}\geq \|\nu_x\|_\nu- 2\,\varepsilon
\end{equation}
According to \eqref{dens0} we have $ \|\nu_x\|_\nu\geq\|\nu_x|{\mathfrak A}\|_{\nu|{\mathfrak A}}$.  In view of \eqref{app3} this implies that $$
\|\nu_x\|_\nu\geq \|\nu_x|{\mathfrak A}\|_{\nu|{\mathfrak A}} \geq \|\nu_x\|_\nu- 2\,\varepsilon$$
Since $\varepsilon>0$ can be chosen at will, from this \eqref{dens1} follows, for all $x\in {\mathfrak A}_{\mathsf h}$ with $\|x\|<\frac{1}{2}$. Since each element of ${\mathfrak A}_{\mathsf h}$ is a positive multiple of such ones, from this  \eqref{dens1} in general follows. Now, let  $f\in {\mathfrak H}_\nu(N)$. By the von-Neumann and Kaplansky density theorems, $N_{\mathsf h}$ is the strong closure of ${\mathfrak A}_{\mathsf h}$. Hence, the set $\{\nu_x:\,x\in{\mathfrak A}_{\mathsf h}\}$ is a $\|\cdot\|_\nu$-dense subset of ${\mathfrak H}_\nu(N)$. Now, let $\{x_n\}\subset {\mathfrak A}_{\mathsf h}$ with $\|\cdot\|_\nu-\lim_{n\to\infty} \nu_{x_n}=f$. By the just verified \eqref{dens1} and by subadditivity of $\|\cdot\|_{\nu|{\mathfrak A}}$ we then get
\[
\|\nu_{x_n}\|_{\nu}=\|\nu_{x_n}|{\mathfrak A}\|_{\nu|{\mathfrak A}}=\|f|{\mathfrak A}+(\nu_{x_n}-f)|{\mathfrak A}\|_{\nu|{\mathfrak A}}\leq \|f|{\mathfrak A}\|_{\nu|{\mathfrak A}}+\|(\nu_{x_n}-f)|{\mathfrak A}\|_{\nu|{\mathfrak A}}
\]
In view of \eqref{dens0} (applied with $\nu_{x_n}-f$ instead of $f$ there) from the previous estimate
\[
\|\nu_{x_n}\|_{\nu}\leq \|f|{\mathfrak A}\|_{\nu|{\mathfrak A}}+\|\nu_{x_n}-f\|_{\nu}
\]
is obtained, for any $n\in {{\mathbb{N}}}$. In the limit for $n\to\infty$ from the latter one is arriving at the estimate $\|f\|_\nu\leq  \|f|{\mathfrak A}\|_{\nu|{\mathfrak A}}$. In view of this and \eqref{dens0} then \eqref{dens} follows.
\end{proof}
\begin{remark}\label{comm1}
According to \eqref{spec}/\eqref{speccc}, each $X^{(n)}$ is a commuting system of elements of ${\mathfrak A}_+$. As a consequence of this, formula \eqref{tangentnorm} can be simplified
\begin{equation}\label{tangentnorm1}
\|f|{\mathfrak A}\|_{\nu|{\mathfrak A}}=\sup_{\{x\}\subset {\mathfrak A}_+,\text{ commuting}}{\frac{1}{2}}\,\sqrt{{\sum_{j}}^\prime
\frac{f(x_j)^2}{\nu(x_j)\phantom{^2}}}
\end{equation}
 \end{remark}
Now, from \eqref{tangentnorm} (with  $M=N$) it easily follows that
\begin{equation}\label{traum}
{\mathfrak T}_\nu(N)=\bigl\{g\in N_{\mathsf h}^*\,:\
\|g\|_\nu<\infty\bigr\}
\end{equation}
is a linear subspace of hermitian linear forms over $N$, in restriction to which
$\|\cdot\|_\nu$ is a norm and, as will be shown now, for a vector state yields another form of \eqref{euk.1}.
\begin{lemma}\label{neben}
${\mathfrak T}_\nu(N)={\mathfrak H}_\nu(N)$ holds, at each vector state $\nu\in {\mathcal S}(N)$.
\end{lemma}
\begin{proof}
Suppose $x\in N_{\mathsf h}\backslash\{{\mathsf 0}\}$, with spectrum consisting of finitely many spectral values, say $n$. Then, $x=\sum_j^n \lambda_j p_j$,
with $\mathsf{spec}(x)=\{\lambda_1,\lambda_2,\ldots,\lambda_n\}\subset {\mathbb R}$, and $\{p_1,p_2,\ldots,p_n\}$ being a decomposition of the unit operator ${\mathsf 1}$ into mutually orthogonal orthoprojections. Let $R$ be the commutative $vN$-subalgebra of $N$ generated by $x$ and ${\mathsf 1}$. Let $f\in {\mathfrak T}_\nu(N)$, $f\not=0$. We are going to show that then $f\in {\mathfrak H}_\nu(N)$. In fact, if $f(p_j)\not=0$, then owing to $\|f\|_\nu<\infty$ from \eqref{tangentnorm} the estimate  $|f(p_j)|\leq 2 \,\nu(p_j)\,\|f\|_\nu$ follows. Hence, if $f(p_j)\not=0$ for some $j$, then $\nu(p_j)\not=0$. Define $s_j\in {\mathbb R}$ by
\[
s_j=
\begin{cases}
\frac{f(p_j)}{\nu(p_j)} & \text{ for }f(p_j)\not=0,\\
&\\
\hfill{} 0\hfill{} & \text{ else. }
\end{cases}
\]
For $\varphi\in {\mathcal S}_N(\nu)$, define $\psi=\Bigl(\sum_j^n s_j p_j\Bigr)\varphi$. Then, the restriction $f|R$ of $f$ to $R$ reads
\[
f|R=\re f_{\psi,\varphi} |R
\]
 with $\psi\in R_{\mathsf h}\varphi$. Hence, Lemma \ref{euk1} and Lemma \ref{normformel} with respect to $R$ can be applied and the corresponding equivalent to \eqref{fund1} then reads as
\[
|f(x)|\leq |f|R|_{\nu|R}\, \|x\varphi\|=2\,\|f|R\|_{\nu|R} \|x\varphi\|
\]
But note that from formula \eqref{tangentnorm} owing to $R\subset N$ obviously $$\|f|R\|_{\nu|R}\leq \|f\|_\nu$$ follows. Thus, for the chosen $x\in N_{\mathsf h}$ thus $|f(x)|\leq 2\,\|f\|_\nu\, \|x\varphi\|$ is seen to hold. But since for  $x\in N_{\mathsf h}\backslash\{{\mathsf 0}\}$ any hermitian element of $N$ with spectrum consisting of finitely many spectral values could have been chosen, and these elements form a uniformly dense subset of $N_{\mathsf h}$, by continuity we may conclude that
\begin{equation}\label{fazit}
  |f(x)|\leq 2\,\|f\|_\nu \,\|x\varphi\|
\end{equation}
holds at each $x\in N_{\mathsf h}$, and any given $f\in {\mathfrak T}_\nu(N)$, $f\not=0$. Since for $f=0$ this  trivially remains valid,
\eqref{fazit} shows
that at each $f\in {\mathfrak T}_\nu(M)$,
the prescription $N_{\mathrm{h}}\varphi\ni x\varphi\,\longmapsto\,f(x)\in {{\mathbb{R}}}$
is a well-defined bounded map.
Obviously the mentioned map
is linear, hence is a bounded real linear form on the real pre-Hilbert space
$N_{\mathrm{h}}\varphi
\subset {\mathcal H}_{{\mathbb{R}}}$. But then there has to exist unique
$\phi\in [N_{\mathrm{h}}\varphi]$ such that
$f(x)=\langle \phi,x\varphi\rangle_{{\mathbb{R}}}=\langle x\phi,\varphi\rangle_{{\mathbb{R}}}$ is
fulfilled,
for each $x\in N_{\mathrm{h}}$. Since
$f\in N^*_{\mathrm{h}}$ holds and the hermitian elements of $N$ separate the points of
$N^*$, the previous formula for $f$ finally gives that $f=\re f_{\phi,\varphi}$ has to be fulfilled,
with $\varphi\in {\mathcal S}_N(\nu)$. According to \eqref{euk.1} $f\in {\mathfrak H}_\nu(N)$
follows. Hence,
${\mathfrak T}_\nu(N)\subset {\mathfrak H}_\nu(N)$ is seen. The assertion follows from this and
since according to \eqref{traum} and Lemma \ref{normformel} the other inclusion
${\mathfrak T}_\nu(N)\supset {\mathfrak H}_\nu(N)$ is valid.
\end{proof}
An important consequence of Lemma \ref{neben} is the following result about \eqref{tangentnorm}.
\begin{corolla}\label{haupt}
Let $N$ be a $vN$-algebra acting on a Hilbert space ${\mathcal H}$. Suppose ${\mathfrak A}$ is a ${\mathsf C}^*$-algebra with ${\mathsf 1}\in {\mathfrak A}\subset N$ and ${\mathfrak A}^{\,\prime\prime}=N$. For each  vector state $\nu\in {\mathcal S}(N)$  and any ultrastrongly continuous hermitian linear form $f\in (N_*)_{\mathsf h}$, the following holds
\begin{equation}\label{densend}
    \|f\|_\nu=\|f|{\mathfrak A}\|_{\nu|{\mathfrak A}}
\end{equation}
with $\nu|{\mathfrak A}$ and $f|{\mathfrak A}$ being the restrictions of $\nu$ and $f$ from $N$ to ${\mathfrak A}$, respectively.
\end{corolla}
\begin{proof}
Let  $\nu\in {\mathcal S}(N)$ be a vector state over $N$. In case if $\|f\|_\nu<\infty$, Lemma \ref{neben} and Lemma \ref{dichte} together imply \eqref{densend} to be fulfilled. Suppose $f\in (N_*)_{\mathsf h}$ with $\|f\|_\nu=\infty$. In view of Lemma \ref{neben} then $f\not\in {\mathfrak H}_\nu(N)$. Note that the arguments raised at the beginning of the proof of Lemma \ref{neben} are saying that if there were a constant  $\beta\geq 0$ with $|f(p)|\leq \beta \|p\,\varphi\|$, for any orthoprojection $p\in N$, this would imply that $f\in {\mathfrak H}_\nu(N)$. Hence, in the case at hand, in view of Lemma \ref{neben} such a constant cannot exist. In line with this, either there is an orthoprojection $p$ with $\|p\,\varphi\|=0$ and $f(p)\not=0$, or there exists a sequence of orthoprojections $\{p_n\}\subset N$ such that $|f(p_n)|> n  \|p_n\varphi\|$, for any $n\in {{\mathbb{N}}}$. Therefore, by a Kaplansky density argument, we can find nets $\{x_\alpha\}$ or $\{x_{n,\alpha}\}$ of positive operators in the unit ball of ${\mathfrak A}$ and strongly converging to $p$ or $p_n$, respectively. Since both $f$ and $\nu$ are normal linear forms, one then has that either $\lim_\alpha |f(x_\alpha)|^2/\nu(x_\alpha)=\infty$, or
that $\lim_\alpha |f(x_{n,\alpha})|^2/\nu(x_{n,\alpha})> n^2$, for each $n\in {{\mathbb{N}}}$. In either case, from this  more than ever $\|f|{\mathfrak A}\|_{\nu|{\mathfrak A}}=\infty$ follows. Hence,
in case of $\|f\|_\nu=\infty$ \eqref{densend} holds, too.
\end{proof}
\subsection{Characterizing tangent forms in terms of tangent spaces}\label{finaltanallg}
For a unital ${\mathsf C}^*$-algebra $M$ we are ready now to explain structure and basic meaning of the elements of ${\mathsf{T}}_\nu(M)$ as announced in context of Definition \ref{trdef}.
\subsubsection{Tangent forms as elements of tangent spaces}\label{tanallg}
Start with the special case if $M$ is a $vN$-algebra $N$ acting on a Hilbert space ${\mathcal H}$, with scalar product $\langle\cdot,\cdot\rangle$, and where $\nu$ is a vector state. Then things around ${\mathsf{T}}_\nu(M)$ simplify.
 \begin{corolla}\label{vNT}
 Let $N$ be a  $vN$-algebra acting on a Hilbert space ${\mathcal H}$, and let  $\nu$ be a vector state over $N$. For each $\varphi\in {\mathcal S}_N(\nu)$ the following are mutually equivalent:
 \begin{enumerate}
 \item\label{vNT1} $f\in {\mathsf{T}}_\nu(N)$;
 \item\label{vNT2} $f(\cdot)=\langle(\cdot)\hat{\psi}_0,\varphi\rangle+\langle(\cdot)\varphi,\hat{\psi}_0\rangle$, with unique $\hat{\psi}_0\in[N_{\mathsf h}\varphi]$ obeying $\Re\langle\hat{\psi}_0,\varphi\rangle=0$.
 \end{enumerate}
 Moreover, if \eqref{vNT1} is fulfilled, then $\|f\|_\nu=\|\hat{\psi}_0\|$ with the vector  $\hat{\psi}_0$ figuring in \eqref{vNT2}.
 \end{corolla}
 \begin{proof}
 Since $N$ is a $vN$-algebra acting on the  Hilbert space ${\mathcal H}$, the Definition \ref{trdef}, when applied
 together with \eqref{traum} and Lemma \ref{neben} in case of $N$, amounts to
 \[
 {\mathsf{T}}_\nu(N)={\mathfrak H}_\nu(N) \cap \{g\in N_{\mathsf h}^*:\,g({\mathsf 1})=0\}
 \]
 By the latter and Lemma \ref{euk1}, \eqref{vNT1} holds if, and only if, $f$ admits a representation as $f(\cdot)=\re \langle(\cdot)\psi_0,\varphi\rangle $, with $\psi_0\in[N_{\mathsf h}\varphi]$ and $\Re\langle\psi_0,\varphi\rangle=0$, which is uniquely determined and is obeying
 $\|\psi_0\|=|f|_\nu$. In defining $\hat{\psi}_0=\psi_0/2$, in view of  Lemma \ref{normformel} from the previous equivalence of \eqref{vNT1} and \eqref{vNT2} follows, with  $\|f\|_\nu=\|\hat{\psi}_0\|$.
 \end{proof}
\begin{theorem}\label{Cstern}
  $M$ be a unital ${\mathsf C}^*$-algebra, $\nu\in {\mathcal S}(M)$. The following are equivalent.
 \begin{enumerate}
 \item\label{Cstern1} $f\in {\mathsf{T}}_\nu(M)$;
 \item\label{Cstern1a}    $f(\cdot)=\langle\pi(\cdot)\psi,\varphi\rangle_\pi+\langle\pi(\cdot)\varphi,\psi\rangle_\pi$, with $\varphi\in {\mathcal S}_{\pi,M}(\nu)$ and $\psi\in{\mathcal H}_\pi$ obeying $\Re\langle\psi,\varphi\rangle_\pi=0$, for some unital $^*$-representation $\{\pi,{\mathcal H}_\pi\}$;
 \item\label{Cstern2}    $f(\cdot)=\langle\pi(\cdot)\hat{\psi}_0,\varphi\rangle_\pi+\langle\pi(\cdot)\varphi,\hat{\psi}_0\rangle_\pi$, with $\varphi\in {\mathcal S}_{\pi,M}(\nu)$ and unique $\hat{\psi}_0\in[\pi(M)_{\mathsf h}\varphi]$ obeying $\Re\langle\hat{\psi}_0,\varphi\rangle_\pi=0$, for some unital $^*$-representation $\{\pi,{\mathcal H}_\pi\}$;
 \item\label{Cstern3} $f\not={\mathsf 0}$ is the tangent form at $\nu$ of a parameterized curve $\gamma\subset {\mathcal S}(M)$ passing through $\nu$ and admitting a $\gamma$-compliant $^*$-representation around there.
 \end{enumerate}
 The vector $\hat{\psi}_0$ figuring in \eqref{Cstern2} is satisfying $\|f\|_\nu=\|\hat{\psi}_0\|_\pi$, and the curve $\gamma$ in \eqref{Cstern3} can be chosen to obey $\|f\|_\nu=\|\phi\|_\pi$ for a  tangent vector $\phi=\phi(\pi,{\mathcal H}_\pi,(\varphi_t))\in {\mathsf{T}}_\nu(M|\gamma)$.
 \end{theorem}
\begin{proof}
Suppose $f\in {\mathsf{T}}_\nu(M)\backslash\{{\mathsf 0}\}$. Let $f=f_+-f_-$ be the canonical decomposition of $f\in M^*_{\mathsf h}$ as difference of positive linear forms $f_\pm$. Let $\{\pi,{\mathcal H}_\pi\}$ be a unital $^*$-representation of $M$ such that $\nu$, $f_+$ and $f_-$ can be implemented by vectors $\varphi\in {\mathcal S}_{\pi,M}(\nu)$, $\eta_\pm\in {\mathcal S}_{\pi,M}(f_\pm)$. Over $N=\pi(M)^{\,\prime\prime}$, let a vector state $\nu_\pi$ over $N$ and a hermitian linear form $f_\pi\in (N_*)_{\mathsf h}$ be defined by $\nu_\pi(x)=\langle x\varphi,\varphi\rangle_\pi$ and $f_\pi(x)=\langle x\eta_+,\eta_+\rangle_\pi-\langle x\eta_-,\eta_-\rangle_\pi$, respectively, for all  $x\in N$. Then $\nu=\nu_\pi\circ \pi$ and $f=f_\pi\circ\pi$ hold, and therefore in view of Corollary \ref{haupt}, see especially \eqref{densend}, when applied in respect of ${\mathfrak A}=\pi(M)$ and $f_\pi$, $\nu_\pi$ provides that
\begin{equation}\label{idi}
\|f\|_\nu=\|f_\pi|{\pi(M)}\|_{{\nu_\pi}|{\pi(M)}}=\|f_\pi\|_{\nu_\pi}
\end{equation}
Since by assumption $\|f\|_\nu<\infty$ and $f({\mathsf 1})=0$ hold, $f_\pi({\mathsf 1})=0$ and in view of \eqref{idi}  $$\|f_\pi\|_{\nu_\pi}<\infty$$
follows, with
 the vector state $\nu_\pi$ implemented by $\varphi$ over the $vN$-algebra $N=\pi(M)^{\,\prime\prime}$. Thus, $f_\pi\in {\mathsf{T}}_{\nu_\pi}(N)$.  For $N$ acting on ${\mathcal H}_\pi$  Corollary \ref{vNT} can be applied with the result
\[
f_\pi(\cdot)=\langle(\cdot)\hat{\psi}_0,\varphi\rangle_\pi+\langle(\cdot)\varphi,\hat{\psi}_0\rangle_\pi
\]
with unique $\hat{\psi}_0\in[N_{\mathsf h}\varphi]$, and such that
\[
\|f_\pi\|_{\nu_\pi}=\|\hat{\psi}_0\|_\pi
\]
is satisfied. Now, remind that by the  usual density argumentation one has $$[N_{\mathsf h}\varphi]=[\pi(M)^{\,\prime\prime}_{\mathsf h}\varphi]=[\pi(M)_{\mathsf h}\varphi]$$  Hence, $\hat{\psi}_0\in[\pi(M)_{\mathsf h}\varphi]$. But then, owing to \eqref{idi} and since $f=f_\pi\circ\pi$ holds, from the above mentioned  relations the validity of \eqref{Cstern2} follows.
Thus \eqref{Cstern1} with $f\not={\mathsf 0}$ implies \eqref{Cstern2}.
Thereby note that $f\not={\mathsf 0}$ and $f({\mathsf 1})=0$ imply $$\langle \hat{\psi}_0,\varphi\rangle_\pi+\langle \varphi,\hat{\psi}_0\rangle_\pi=0$$ but $\hat{\psi}_0\not=0$. Especially,  $\hat{\psi}_0$ then cannot be a real multiple of $\varphi$. Let a symmetric around zero interval $I\subset {\mathbb R}$ be given with length $|I|<1$. Consider a map $\gamma: I\ni t\longmapsto \nu_t\in {\mathcal S}(M)$ implemented by the simple $C^1$-curve $I\ni t\longmapsto \varphi_t\in {\mathcal S}_{\pi,M}(\nu_t)$, with unit vector
\begin{subequations}\label{gff0}
\begin{equation}\label{gff1}
 \varphi_t=\frac{\varphi+t\,\hat{\psi}_0}{\sqrt{1+t^2\|\hat{\psi}_0\|_\pi^2}}
\end{equation}
The tangent vector reads $\phi=\varphi_t^{\,\prime}|_{t=0}=\hat{\psi}_0$. Thus, by $\nu_t(\cdot)=\langle \pi(\cdot)\varphi_t,\varphi_t\rangle_\pi$ and \eqref{Cstern2}
\begin{equation}\label{fform}
\nu_t^{\,\prime}(\cdot)|_{t=0}=\bigl(\langle\pi(\cdot)\varphi_t^{\,\prime},\varphi_t\rangle_\pi+
\langle\pi(\cdot)\varphi_t,\varphi_t^{\,\prime}\rangle_\pi\bigr){{\bigl|}_{t=0}}=f(\cdot)\not={\mathsf 0}
\end{equation}
is implied over $M$. Also, if $\omega$ is the positive linear form implemented by $\hat{\psi}_0$, $\hat{\psi}_0\in {\mathcal S}_{\pi,M}(\omega)$, then $\nu_t$ can be given according to the following explicit formula
\begin{equation}\label{gff3}
\nu_t=\frac{\nu+tf+t^2\omega}{1+t^2\omega({\mathsf 1})}
\end{equation}
\end{subequations}
Obviously, $t\longmapsto\nu_t$ is continuously differentiable.
In line with \eqref{gff0}, the conditions for an application of Lemma \ref{constr} are given. The conclusion is that
the map $\gamma$ when restricted to some sufficiently small interval around $t=0$ in fact defines a curve of the admissible type which is passing through $\nu$ at parameter value $t=0$ and is relating to  $\{\pi,{\mathcal H}_\pi\}$ as  $\gamma$-compliant $^*$-representation around $\nu$. In addition, according to \eqref{fform} the curve $\gamma$ possesses the tangent form $f$ at $\nu$.
Hence, each $f\in {\mathsf{T}}_\nu(M)\backslash\{{\mathsf 0}\}$ can arise as a tangent form at $\nu$ along some curve of the admissable class. Thereby, the tangent vector of the special curve constructed, $\phi\in {\mathsf{T}}_\nu(M|\gamma)$, is obeying $\phi=\hat{\psi}_0$, and thus in addition  $\|f\|_\nu=\|\phi\|_\pi$ is  fulfilled for the $\gamma$ in question and which is  passing through $\nu$ with tangent form $f$.  Thus, for $f\not={\mathsf 0}$, the net of implications \eqref{Cstern1}$\Rightarrow$\eqref{Cstern2}$\Rightarrow$\eqref{Cstern3}
has to hold, with the conditions  $\|f\|_\nu=\|\hat{\psi}_0\|_\pi$ within \eqref{Cstern2}, and $\|f\|_\nu=\|\phi\|_\pi$ for some $\phi\in {\mathsf{T}}_\nu(M|\gamma)$ within \eqref{Cstern3} fulfilled.

On the other hand, let $\gamma:I\ni t\longmapsto\nu_t\in {\mathcal S}(M)$ be a parameterized curve passing through $\nu$ at $t=0$, with tangent form $f\not={\mathsf 0}$ at $\nu$. Suppose a $\gamma$-compliant $^*$-representation $\{\pi,{\mathcal H}_\pi\}$ around $\nu$ exists. In line with this, let $I_\pi\in t\longmapsto \varphi_t\in {\mathcal S}_{\pi,M}(\nu_t)$ be a differentiable at $t=0$ implementation of $\gamma$ around $\nu$.
Then, with $\varphi_t^{\,\prime}|_{t=0}=\phi=\phi(\pi, {\mathcal H}_\pi,(\varphi_t))\in {\mathsf{T}}_\nu(M|\gamma)$ and $\varphi=\varphi_0$, for the tangent form $f$ at $\nu$ one finds
$
f(\cdot)=\langle \pi(\cdot)\phi,\varphi\rangle_\pi+\langle \pi(\cdot)\varphi,\phi\rangle_\pi
$, with $f({\mathsf{1}})=0$ fulfilled.
Thus, the hermitian linear form $f_\pi(\cdot)=\re\langle (\cdot)2\,\phi,\varphi\rangle_\pi$ on $N=\pi(M)^{\,\prime\prime}$ is obeying $f_\pi\circ\pi=f$ and $f_\pi\in {\mathfrak H}_{\nu_\pi}(N)$, with the vector state $\nu_\pi$ implemented by $\varphi$ on $N$. Obviously  $f_\pi\in (N_*)_{\mathsf h}$ holds, with $f_\pi({\mathsf 1})=0$ and $\|f_\pi\|_{\nu_\pi}<\infty$, due to \eqref{euk0} and Lemma \ref{normformel}. With ${\mathfrak A}=\pi(M)$, Corollary \ref{haupt} and  \eqref{densend} can be applied and yield
$$
\|f_\pi|{\pi(M)}\|_{\nu_\pi|{\pi(M)}}=\|f_\pi\|_{\nu_\pi}
$$
According to $\|f_\pi\|_{\nu_\pi}<\infty$ from the previous and since always
\[
\|f\|_\nu=\|f_\pi|{\pi(M)}\|_{\nu_\pi|{\pi(M)}}
\]
has to be fulfilled, we see that the hermitian linear form $f$ has to satisfy $\|f\|_\nu<\infty$, $f({\mathsf 1})=0$, and thus in view of Definition \ref{trdef} $f\in {\mathsf{T}}_\nu(M)$ follows. Hence, in case of $f\not={\mathsf 0}$, the implication \eqref{Cstern3}$\Rightarrow$\eqref{Cstern1} is seen to hold. Together with the above in case of $f\not={\mathsf 0}$ we thus infer that the net of implications  \eqref{Cstern1}$\Rightarrow$\eqref{Cstern2}$\Rightarrow$\eqref{Cstern3}$\Rightarrow$\eqref{Cstern1} takes place, and therefore in this case we have equivalence among \eqref{Cstern1}, \eqref{Cstern2} and \eqref{Cstern3}. The implication \eqref{Cstern2}$\Rightarrow$\eqref{Cstern1a} is trivially valid. To see the other way around, suppose \eqref{Cstern1a}. Then,
$$
f_\pi=\re \langle(\cdot) 2\psi,\varphi\rangle_\pi
$$
over the $vN$-algebra $N=\pi(M)^{\,\prime\prime}$. In accordance with Lemma \ref{euk1}, see \eqref{euk.3}, we may consider the linear isometry
$\iota_\pi: [N_{\mathsf h}\varphi]\rightarrow {\mathfrak H}_{\nu_\pi}(N)$
and in defining
\[
\hat{\psi}_0=\frac{1}{2}\,\iota_\pi^{-1}(f_\pi)\in [N_{\mathsf h}\varphi]
\]
we can be assured that $\hat{\psi}_0$ is the only vector in $[N_{\mathsf h}\varphi]$ obeying over $N$
\[
f_\pi(\cdot)=\langle(\cdot)\hat{\psi}_0,\varphi\rangle_\pi+\langle(\cdot)\varphi,\hat{\psi}_0\rangle_\pi
\]
In view of Lemma \ref{normformel} then  $\|f_\pi\|_{\nu_\pi}=\|\hat{\psi}_0\|_\pi$ is fulfilled.  Owing to $f_\pi \circ \pi=f$ over $M$ and since \eqref{idi} holds, from the previous \eqref{Cstern2} follows.  That is, the implication \eqref{Cstern1a}$\Rightarrow$\eqref{Cstern2} is seen to hold, too. Thus, by  the above proven, for $f\not={\mathsf 0}$, \eqref{Cstern1}--\eqref{Cstern3} are mutually equivalent. For $f={\mathsf 0}$ equivalence among \eqref{Cstern1}--\eqref{Cstern2} is trivially valid. This completes the proof.
\end{proof}
With a unital $^*$-representation $\{\pi,{\mathcal H}_\pi\}$ consider the subset of states given as
\begin{equation}\label{pimann1}
{\mathcal{S}}^\pi(M)=\bigl\{\varrho\in {\mathcal S}(M): {\mathcal S}_{\pi,M}(\varrho)\not=\emptyset\bigr\}
\end{equation}
That is, ${\mathcal{S}}^\pi(M)$ is  the set of all states arising as vector states relative to $\{\pi,{\mathcal H}_\pi\}$. With this setting in mind,
a useful supplement to Theorem \ref{Cstern}  will be added.
\begin{corolla}\label{tanfolg}
Let $\nu\in {\mathcal S}(M)$, $f\in {\mathsf{T}}_\nu(M)$, and be   $\{\pi_0,{\mathcal H}_{\pi_0}\}$ a   unital $^*$-representation of $M$ with non-trivial $\pi_0$-fibre of $\nu$. Then, the following facts hold.
\begin{enumerate}
\item \label{tanfolg01}
For given $\varphi_0\in {\mathcal S}_{\pi_0,M}(\nu)$, over $M$ we have
\begin{equation}\label{tanfolg0}
f(\cdot)=\langle\pi_0(\cdot)\hat{\xi}_0,\varphi_0\rangle_{\pi_0}+\langle\pi_0(\cdot)\varphi_0,\hat{\xi}_0\rangle_{\pi_0}
\end{equation}
with a unique $\hat{\xi}_0\in[\pi_0(M)_{\mathsf h}\varphi_0]$ obeying  $\|f\|_\nu=\|\hat{\xi}_0\|_{\pi_0}$.
\item \label{tanfolg02}
Define $\omega(\cdot)=\langle \pi_0(\cdot) \hat{\xi}_0,\hat{\xi}_0\rangle_{\pi_0}$ over $M$. For $f\not={\mathsf 0}$, the parameterized curve
\begin{equation}\label{tanfolg0b}
\gamma:\,[-1,1]\ni t\longmapsto \nu_t=\frac{\nu+t f + t^2 \omega}{1+t^2\|f\|_\nu^2}
\end{equation}
is passing through $\nu$ and is exhibiting $f$ as tangent form at $\nu$. Moreover, $\gamma\subset {\mathcal{S}}^\pi(M)$ holds, for  each unital $^*$-representation $\{\pi,{\mathcal H}_\pi\}$ such that $\nu\in {\mathcal{S}}^\pi(M)$. Thus, each  $\{\pi,{\mathcal H}_\pi\}$ obeying  ${\mathcal S}_{\pi,M}(\nu)\not=\emptyset$ is $\gamma$-compliant around $\nu$.
\end{enumerate}
\end{corolla}
\begin{proof}
There are $\{\pi,{\mathcal H}_\pi\}$, $\varphi\in {\mathcal S}_{\pi,M}(\nu)$ and $\hat{\psi}_0\in [\pi(M)_{\mathsf h}\varphi]$ with respect to which \eqref{Cstern2} of Theorem \ref{Cstern} holds. Since $\varphi_0\in {\mathcal S}_{\pi_0,M}(\nu)$ holds,  for each $x\in M$ we have $$\nu(x^* x)=\|\pi(x)\varphi\|_\pi^2=\|\pi_0(x)\varphi_0\|_{\pi_0}^2$$ Hence, the map defined by $ \pi(x)\varphi\longmapsto \pi_0(x)\varphi_0$ for each $x\in M$ uniquely extends to an isometry $U$ acting from  $p_\pi^{\,\prime}(\varphi){\mathcal H}_\pi=[\pi(M)\varphi]$ onto $p_{\pi_0}^{\,\prime}(\varphi_0){\mathcal H}_{\pi_0}=[\pi_0(M)\varphi_0]$.
Let $\hat{\xi}_0=U\hat{\psi}_0$. Owing to $\hat{\psi}_0\in [\pi(M)_{\mathsf h}\varphi]$ we have  $\hat{\xi}_0\in [\pi_0(M)_{\mathsf h}\varphi_0]$. From this together with $U\varphi=\varphi_0$ and the fact that $U$ is an isometry we infer that, for each $y\in M$,
\[
\langle \pi_0(y)\varphi_0,\hat{\xi}_0\rangle_{\pi_0}=\langle U\pi(y)\varphi,U\hat{\psi}_0\rangle_{\pi_0}=\langle \pi(y)\varphi,\hat{\psi}_0\rangle_\pi
\]
holds. From the representation \eqref{Cstern2} of Theorem \ref{Cstern} we then conclude that
\begin{eqnarray*}
f(x)&=&\langle\pi(x)\hat{\psi}_0,\varphi\rangle_\pi+\langle\pi(x)\varphi,\hat{\psi}_0\rangle_\pi\\
&=& \overline{\langle\pi(x^*)\varphi,\hat{\psi}_0\rangle_\pi} + \langle\pi(x)\varphi,\hat{\psi}_0\rangle_\pi\\
&=& \overline{\langle\pi_0(x^*)\varphi_0,\hat{\xi}_0\rangle_{\pi_0}} + \langle\pi_0(x)\varphi_0,\hat{\xi}_0\rangle_{\pi_0}\\
&=& \langle\pi_0(x)\hat{\xi}_0,\varphi_0\rangle_{\pi_0}+\langle\pi_0(x)\varphi_0,\hat{\xi}_0\rangle_{\pi_0}
\end{eqnarray*}
for each $x\in M$. This is  formula \eqref{tanfolg0}.  Uniqueness of $\hat{\xi}_0\in [\pi_0(M)_{\mathsf h}\varphi_0]$ with $\|f\|_\nu=\|\hat{\xi}_0\|_{\pi_0}$ follows from the uniqueness of $\hat{\psi}_0\in [\pi(M)_{\mathsf h}\varphi]$ with $\|f\|_\nu=\|\hat{\psi}_0\|_{\pi_0}$ in Theorem \ref{Cstern} \eqref{Cstern2} and since $U$ is an isometry. This proves \eqref{tanfolg01}.
Also remark that if $\{x_n\}\subset M_{\mathsf{h}}$ is chosen such that $\hat{\psi}_0 = \lim_{n\to \infty} \pi(x_n)\varphi$, then  $
\hat{\xi}_0=U\hat{\psi}_0=\lim_{n\to \infty}\pi_0(x_n)\varphi_0
$ follows. Hence, for each $x\in M$, we conclude as follows:
\begin{eqnarray*}
\langle\pi(x)\hat{\psi}_0,\hat{\psi}_0\rangle_{\pi} &=& \lim_{n\to \infty} \langle\pi(x x_n)\varphi,\pi(x_n)\varphi\rangle_{\pi}=\lim_{n\to \infty} \langle U\pi(x x_n)\varphi,U\pi(x_n)\varphi\rangle_{\pi_0}=\\
&=&\lim_{n\to \infty} \langle \pi_0(x x_n)\varphi_0,\pi_0(x_n)\varphi_0\rangle_{\pi_0}=\langle \pi_0(x) \hat{\xi}_0,\hat{\xi}_0\rangle_{\pi_0}
\end{eqnarray*}
Thus, the positive linear form $\omega$ implemented by $\hat{\xi}_0$ in respect of $\pi_0$ is the same as the positive linear form implemented by $\hat{\psi}_0$ in respect of $\pi$. The special construction \eqref{gff0} used in proving Theorem \ref{Cstern} \eqref{Cstern2} $\Rightarrow$ \eqref{Cstern3} with input data $\pi$, $\varphi$ and $\hat{\psi}_0$, equally well can be applied when starting from the input data $\pi_0$, $\varphi_0$ and $\hat{\xi}_0$.
When inserting  these input data into formula \eqref{gff1}, the facts mentioned about $f, \omega$ imply that both the families  $$[-1,1]\ni t\longmapsto (\varphi+t\hat{\psi}_0)/\biggl(\sqrt{1+t^2\|f\|_\nu^2}\biggr),\ [-1,1]\ni t\longmapsto (\varphi_0+t\hat{\xi}_0)/\biggl(\sqrt{1+t^2\|f\|_\nu^2}\biggr)$$
when considered  relative to $\pi$ and $\pi_0$, respectively, pointwise are implementations of the same states $\nu_t$. That is, the resulting curve $\gamma$ then will be the same in both cases, and is given by \eqref{tanfolg0b}. It is passing through $\nu$ at $t=0$, and is exhibiting $f$ as tangent form at $\nu$. In addition, relating \eqref{pimann1}, we have $\gamma\in {\mathcal S}^\pi(M)\cap {\mathcal S}^{\pi_0}(M)$. Thus  especially both $\{\pi,{\mathcal H}_\pi\}$ and   $\{\pi_0,{\mathcal H}_{\pi_0}\}$ are  $\gamma$-compliant. Finally, since $\pi_0$ can stand for any unital $^*$-representation $\pi$ with non-trivial $\pi$-fibre of  $\nu$, the proof of \eqref{tanfolg02} is complete now.
\end{proof}
\begin{remark}\label{bestapprox}
\begin{enumerate}
\item\label{bestapprox1}
Usually, $f\in {\mathsf{T}}_\nu(M)$ will be given as in Theorem \ref{Cstern}\,\eqref{Cstern1a} by
\begin{subequations}\label{allgf}
\begin{equation}\label{allgf1}
f(\cdot)=\langle\pi(\cdot)\psi,\varphi\rangle_\pi+\langle\pi(\cdot)\varphi,\psi\rangle_\pi
\end{equation}
with the help of $\varphi\in {\mathcal S}_{\pi,M}(\nu)$ and some tangent vector
$\psi=\psi(\pi,{\mathcal H}_\pi,(\varphi_t)) \in {\mathsf{T}}_\nu(M|\gamma)$ referring to some implementation of a curve $\gamma$ and $\gamma$-compliant representation $\{\pi,{\mathcal H}_\pi\}$. And $\hat{\psi}_0$ arises from $\psi$ by the following procedure:
\begin{equation}\label{allgf2}
\hat{\psi}_0=\mathfrak{P}_\pi \psi
\end{equation}
where $\mathfrak{P}_\pi$ is the linear projection operator over ${\mathcal H}_{\pi,{\mathbb R}}$ projecting onto the real linear subspace $[\pi(M)_{\mathsf h}\varphi]$. In fact, since for   $\psi$  and $\hat{\psi}_0$ in respect of $f$,  \eqref{allgf1} and Theorem \ref{Cstern}\,\eqref{Cstern2}  are satisfied, for each $x\in \pi(M)_{\mathsf h}$ especially
\begin{equation}\label{allgf3}
0=\langle(\psi-\hat{\psi}_0),\pi(x)\varphi\rangle_{\pi,{\mathbb R}}
\end{equation}
\end{subequations}
is seen to hold. That is, $\psi-\hat{\psi}_0\in [\pi(M)_{\mathsf h}\varphi]^\perp$ holds, with respect to the inner product $\langle\cdot,\cdot\rangle_{\pi,{\mathbb R}}$. This together with $\hat{\psi}_0\in [\pi(M)_{\mathsf h}\varphi]$ implies \eqref{allgf2}.
\item\label{bestapprox2}
By Euklidean calculus,  \eqref{allgf3} together with  $\hat{\psi}_0\in [\pi(M)_{\mathsf h}\varphi]$ means that $\hat{\psi}_0$  is the unique vector in $[\pi(M)_{\mathsf h}\varphi]$ minimizing the distance to $\psi$,
\begin{equation}\label{allgf4}
\|\psi-\hat{\psi}_0\|_\pi=\inf_{\xi\in[\pi(M)_{\mathsf h}\varphi]} \|\psi-\xi\|_\pi
\end{equation}
Thus, $\hat{\psi}_0$ is the best $\langle\cdot,\cdot\rangle_{\pi,\mathbb R}$-approximation of $\psi$ within  $[\pi(M)_{\mathsf h}\varphi]\subset {\mathcal H}_{\pi,{\mathbb R}}$.  This interpretation of $\hat{\psi}_0$ yet has been mentioned in Remark \ref{iimpdi}\,\eqref{iimpdi3}, Lemma \ref{ocomp} and Remark \ref{agg}\,\eqref{agg3} in context of the decomposition \eqref{decompo} of $\psi\in {\mathcal H}_\pi$.
\item\label{bestapprox3}
Let $f\in {\mathsf{T}}_\nu(M)\backslash\{{\mathsf 0}\}$, and be  $\{\pi,{\mathcal H}_\pi\}$ a unital $^*$-representation with non-trivial $\pi$-fibre at $\nu$. Then, according to Corollary \ref{tanfolg} and  Theorem \ref{Cstern}, $\{\pi,{\mathcal H}_\pi\}$ appears as a $\gamma$-compliant $^*$-representation around $\nu$ for  some  parameterized curve $\gamma$ passing through $\nu$ and which has $f$ as  tangent form at $\nu$. Accordingly, each unital $^*$-representation with non-trivial fibre at $\nu$ is equally well  qualified as a carrier of a sufficient large set of implementable curves passing through $\nu$ and thereby displaying at $\nu$ the full range  ${\mathsf{T}}_\nu(M)\backslash\{{\mathsf 0}\}$ of tangent forms that can appear there. An example for such a set of curves is given by \eqref{tanfolg0b}.
\item\label{bestapprox4}
Note that  \eqref{bestapprox3}  does not provide any answer to the question whether or not a given unital $^*$-representation with non-trivial fibre at $\nu$ for given individual implementable curve $\gamma$ passing through $\nu$ can be $\gamma$-compliant around there. In yet sharpened form one may wonder whether there might exist some  unital $^*$-representation with non-trivial fibre at $\nu$ and which were $\gamma$-compliant around there, for any implementable curve $\gamma$ passing through $\nu$ and belonging to some well-defined class. Thus e.g., if  $M$ is a finite factor of type $I$ and the class of all implementable curves of states of full support and passing through $\nu$ is considered, the latter question can be answered in the  affirmative, see Theorem \ref{enddiffex} and Example \ref{gegendiff}. From this to other cases can be extrapolated.
\end{enumerate}
\end{remark}
\subsubsection{Further properties of the tangent space}\label{proptannorm}
Add some useful information about the structure of the tangent space ${\mathsf{T}}_\nu(M)$, for a state $\nu\in {\mathcal S }(M)$. In particular, we ask for conditions assuring ${\mathsf{T}}_\nu(M)$ to be non-trivial.
Start with some useful auxiliary results in context of Definition \ref{trdef} and formula \eqref{tangentnorm}.
\begin{lemma}\label{suppnorm}
Let $f\in {\mathsf{T}}_\nu(M)$, $a\in M$ with ${\mathsf 0}\leq a \leq {\mathsf 1}$. Then, the following hold.
\begin{enumerate}
\item \label{suppnorma}
$f(a)\not=0\Rightarrow\,\nu(a)\not=0$;
\item \label{suppnormb}
$\,\nu(a)=0\Rightarrow f(a)=0$.
\end{enumerate}
\end{lemma}
\begin{proof}
Note that the assertion of \eqref{suppnormb} is the contraposition of \eqref{suppnorma}. Thus, the proof will be complete if \eqref{suppnorma} is shown to hold. We are going to prove the latter now.
For $1>\varepsilon>0$, let $a_\varepsilon=(1-\varepsilon)a+\varepsilon {\mathsf 1}$. Then,  ${\mathsf 0}\leq a_\varepsilon \leq {\mathsf 1}$, and since especially $f({\mathsf 1})=0$ holds,  $f(a_\varepsilon)=(1-\varepsilon) f(a)\not=0$ follows. Thus, by Definition \ref{e.1} we then see that
\[
\infty >\|f\|^2_\nu\geq \frac{1}{4}\,\frac{f(a_\varepsilon)^2}{\nu(a_\varepsilon)}=\frac{1}{4}\,\frac{(1-\varepsilon)^2 }{(1-\varepsilon)\nu(a)+\varepsilon }\,f(a)^2
\]
has to be fulfilled for any $\varepsilon\in ]0,1[$. Especially, due to $f(a)\not=0$ we conclude that the $\varepsilon$-depending factor (which is positive) for $\varepsilon \to 0$ has to remain bounded from above. This only can be possible provided  $\nu(a)\not=0$ is fulfilled.
\end{proof}
\begin{corolla}\label{normabsch}
$\|f\|_1\leq 2\,\|f\|_\nu,\,\forall f\in {\mathsf{T}}_\nu(M)$.
\end{corolla}
\begin{proof}
Let $f\not={\mathsf 0}$, $f\in {\mathsf{T}}_\nu(M)$. Then, for ${\mathsf 0}< a < {\mathsf 1}$ with $f(a)\not=0$ we have $f({\mathsf 1}-a)=-f(a)\not=0$. By Lemma \ref{suppnorm}, $\nu(a)\not=0$ as well as $\nu({\mathsf 1}-a)\not=0$ then have to be fulfilled,  and  by definition of $\|f\|_\nu$ we conclude that the following holds:
\[
\|f\|^2_\nu\geq \frac{1}{4}\,\biggl(\frac{f(a)^2}{\nu(a)}+\frac{f({\mathsf 1}-a)^2}{\nu({\mathsf 1}-a)}\biggr)=f(a)^2 \frac{1}{4 \,\nu(a)(1-\nu(a))}\geq f(a)^2
\]
For the last estimate the fact has been taken into account that  $0<4\,\nu(a)(1-\nu(a))\leq 1$ is implied by $0<\nu(a)<1$. It is known that for a bounded hermitian linear form $f$ on a unital  ${\mathsf C}^*$-algebra the numerical value $\|f_+\|_1$ given by the functional norm of the positive part $f_+$ in the canonical decomposition of $f$ belongs to the closure of the numerical range  $\{f(a):{\mathsf 0}< a < {\mathsf 1} \}$. Hence, from the above estimate
\[
\|f\|_\nu\geq \|f_+\|_1
\]
follows. Since $f({\mathsf 1})=0$ implies $\|f\|_1=2\,\|f_+\|_1$, from this the assertion follows.
\end{proof}
\begin{corolla}\label{compT}
 When equipped with the $\|\cdot\|_\nu$-norm ${\mathsf{T}}_\nu(M)$ is a Banach space.
\end{corolla}
\begin{proof}
As mentioned in Definition \ref{e.1}, the extended positive function $M^*_{\mathsf h}\ni f\mapsto \|f\|_\nu$ is subadditive and (real) absolutely homogeneous. Thus, by Definition \ref{trdef}, ${\mathsf{T}}_\nu(M)$ is a real linear space with seminorm $ {\mathsf{T}}_\nu(M)\ni f\mapsto \|f\|_\nu$. For given $f\in {\mathsf{T}}_\nu(M)$, let $f$ be represented as asserted by  \eqref{tanfolg0}, with regard to a unital $^*$-representation $\pi_0$ with respect to which the fibre of $\nu$ is non-trivial. Then, by Corollary \ref{tanfolg}, the  unique $\hat{\xi}_0\in [\pi_0(M)_{\mathsf h}\varphi]$ figuring in context of \eqref{tanfolg0}  is obeying  $$\|f\|_\nu=\|\hat{\xi}_0\|_{\pi_0}$$
Hence, $\|f\|_\nu=0$ implies  $\hat{\xi}_0={\mathsf 0}$. By formula \eqref{tanfolg0} then $f={\mathsf 0}$ follows. This proves that $ {\mathsf{T}}_\nu(M)\ni f\mapsto \|f\|_\nu$ is a norm.
Now, let $\{f_n\}\subset {\mathsf{T}}_\nu(M)$ be a Cauchy-sequence in respect of the norm $\|\cdot\|_\nu$. As in Corollary \ref{tanfolg}, let $\hat{\xi}(n)_0$ be the uniquely determined element of $ [\pi_0(M)_{\mathsf h}\varphi]$ and corresponding to $f_n$  in  accordance to formula \eqref{tanfolg0}. By linearity the assertion of Corollary \ref{tanfolg} then amounts to $$\|f_n-f_m\|_\nu =\| \hat{\xi}(n)_0-\hat{\xi}(m)_0\|_{\pi_0}$$
for all $n,m\in {{\mathbb{N}}}$. Hence, $\{\hat{\xi}(n)_0\}$ is a Cauchy-sequence in $ [\pi_0(M)_{\mathsf h}\varphi]$, and thus has to converge to some $\hat{\xi}_0\in  [\pi_0(M)_{\mathsf h}\varphi]$. The linear form $f$ given in accordance with formula \eqref{tanfolg0} then obviously will satisfy  $f=\|\cdot\|_\nu-\lim_{n\to\infty} f_n$. Hence, under the $\|\cdot\|_\nu$-norm the normed real linear space ${\mathsf{T}}_\nu(M)$ will be complete.
\end{proof}
Remind that $\nu\in {\mathcal S}(M)$ is termed `mixed state', if $\varrho, \mu\in {\mathcal S}(M)$ exist with $\varrho\not=\mu$ and
$\nu=(\varrho+\mu)/2$. In a unital ${\mathsf{C}}^*$-algebra $M$, the set  of mixed states reads
\begin{equation}\label{ccomb0}
{\mathcal S}^{\mathsf{mix}}(M)={\mathcal S}(M)\backslash {\mathsf{ex}}\,{\mathcal S}(M)
\end{equation}
At a mixed state, the tangent space will be non-trivial.
\begin{lemma}\label{extriv}
Suppose $\nu\in{\mathcal S}^{\mathsf{mix}}(M)$. Then ${\mathsf{T}}_\nu(M)\not=\{{\mathsf 0}\}$.
\end{lemma}
\begin{proof}
Since $\nu$ is a mixed state, there exist two  states $\varrho$ and $\mu$ with $\varrho\not=\mu$ and $\nu=(\varrho+\mu)/2$. Consider the bounded linear form $f=\varrho-\mu$. Obviously, $f$ is hermitian and is obeying $f({\mathsf 1})=0$ and $f\not={\mathsf 0}$. Also, due to  $\varrho(a)^2+\mu(a)^2-2\,\varrho(a)\mu(a)\leq \varrho(a)^2+\mu(a)^2+2\,\varrho(a)\mu(a)$ for $a\in M$ with ${\mathsf 0}\leq a\leq {\mathsf 1}$
we see $f(a)^2\leq 4\,\nu(a)^2$. Thus, if $\{a_k\}\subset M_+$ is a finite positive decomposition of the unity, we find
\[
{\sum_k}^{\,\prime}\, \frac{f(a_k)^2}{\nu(a_k)}\leq \sum_k\,4\,\nu(a_k)= 4
\]
By Definition \ref{e.1}, see formula \eqref{tangentnorm}, $\|f\|_\nu\leq 1$ follows. Thus  $f\in {\mathsf{T}}_\nu(M)\backslash \{{\mathsf 0}\}$.
\end{proof}
Close by some auxiliary results which will prove useful either when relating  mixed states to pure states in terms of implementable curves or when embedding curves of states over certain subalgebras as curves in the state space.
\begin{lemma}\label{auxpure}
For a unital $^*$-representation $\{\pi,{\mathcal H}_\pi\}$, let unit vectors $\varphi, \xi\in {\mathcal H}_\pi$ obey $\langle \xi,\varphi\rangle_\pi\not=0$. Suppose $\varphi\in {\mathcal{S}}_{\pi,M}(\nu)$, with $\nu\in {\mathsf{ex}}\,{\mathcal S}(M)$. Then $p_\pi^{\,\prime}(\varphi)\leq p_\pi^{\,\prime}(\xi)$, with equality  occurring if the state $\varrho$ with $\xi\in {\mathcal S}_{\pi,M}(\varrho)$ is obeying $\varrho\in {\mathsf{ex}}\,{\mathcal S}(M)$, too.
\end{lemma}
\begin{proof}
Owing to $p_\pi^{\,\prime}(\varphi){\mathcal H}_\pi=[\pi(M)\varphi]$ and since the state $\nu$ implemented by $\varphi$ is pure, the action of $\pi$ in restriction to $[\pi(M)\varphi]$ has to be irreducible. On the other hand, the subspace $p_\pi^{\,\prime}(\varphi)[\pi(M)\xi]\subset [\pi(M)\varphi]$  is left invariant under the action of $\pi$. Hence, either $p_\pi^{\,\prime}(\varphi)[\pi(M)\xi]=\{{\mathsf{0}}\}$ or $p_\pi^{\,\prime}(\varphi)[\pi(M)\xi]= [\pi(M)\varphi]$ must be fulfilled. However, due to $\langle \xi,\varphi\rangle_\pi\not=0$ only the latter options can occur. This is the same as $[\pi(M)\xi]\supset[\pi(M)\varphi]$, or equivalently, $p_\pi^{\,\prime}(\varphi)\leq p_\pi^{\,\prime}(\xi)$. Finally, if the state $\varrho$ implemented by $\xi$ via $\pi$ is pure, then by symmetry also  $p_\pi^{\,\prime}(\xi)\leq p_\pi^{\,\prime}(\varphi)$ holds.
\end{proof}
\begin{corolla}\label{auxpure1}
Let $\gamma=(\nu_t)\subset {\mathcal S}(M)$ be a parameterized curve passing through $\nu=\nu_0$, and let $\{\pi,{\mathcal H}_\pi\}$ be $\gamma$-compliant  around $\nu$. Provided $\nu\in {\mathcal S}^{\mathsf{mix}}(M)$, then also $\nu_t\in {\mathcal S}^{\mathsf{mix}}(M)$ holds, for all parameter values of sufficiently small modulus.
\end{corolla}
\begin{proof}
Under the condition $\nu\in {\mathcal S}^{\mathsf{mix}}(M)$, assume the contrary of the assertion to hold. That is, in view of \eqref{ccomb0}, there are parameter values $\{t_k\}$ with $\lim_{k\to \infty} t_k=0$ but $\nu_{t_k}\in {\mathsf{ex}}\,{\mathcal S}(M)$. We are going to show that this will lead into contradiction.\\
Since $\{\pi,{\mathcal H}_\pi\}$ is $\gamma$-compliant,  a differentiable at $t=0$ implementation  $\varphi_t\in {\mathcal S}_{\pi,M}(\nu_t)$  of $\gamma$ exists. Thus, by continuity, for all $t$ with sufficiently small modulus  we may suppose that $\|\varphi-\varphi_t\|_\pi^2< 1/2$ for $\varphi=\varphi_0$. Also, for any two values $s,t$ of the mentioned specification we then have $\|\varphi-\varphi_t\|_\pi^2<2,\,\|\varphi_s-\varphi_t\|_\pi^2< 2$.
That is, $\langle \varphi,\varphi_t\rangle_\pi\not=0$ and $\langle \varphi_s,\varphi_t\rangle_\pi\not=0$ are fulfilled there. Thus, we may even suppose that the terms of the sequence $\{t_k\}$ are chosen accordingly, from the very beginning. Put this case now. Then, $\langle \varphi,\varphi_{t_k}\rangle_\pi\not=0$ and $\langle \varphi_{t_1},\varphi_{t_k}\rangle_\pi\not=0$, for all $k\in {\mathbb{N}}$.
Application of Lemma \ref{auxpure} to the pairs $\{\varphi,\varphi_{t_k}\}$ and $\{\varphi_{t_k},\varphi_{t_1}\}$ under the suppositions at hand imply $p_\pi^{\,\prime}(\varphi)\geq p_\pi^{\,\prime}(\varphi_{t_k})$ and  $p_\pi^{\,\prime}(\varphi_{t_k})=p_\pi^{\,\prime}(\varphi_{t_1})$,  for all $k\in {{\mathbb{N}}}$. Due to $\varphi=\lim_{k\to \infty} \varphi_{t_k}$ from this  $p_\pi^{\,\prime}(\varphi_{t_1})\varphi=\varphi$ is inferred to hold. Hence we have  $p_\pi^{\,\prime}(\varphi_{t_1})\geq p_\pi^{\,\prime}(\varphi)$. In summarizing, we even see that $p_\pi^{\,\prime}(\varphi_{t_1})= p_\pi^{\,\prime}(\varphi)$ is fulfilled. Now, by assumption $\nu_1$ is pure and therefore the action of $\pi$ on $p_\pi^{\,\prime}(\varphi_{t_1}){\mathcal H}_\pi$ is irreducible. Then, owing to $p_\pi^{\,\prime}(\varphi_{t_1})= p_\pi^{\,\prime}(\varphi)$ the same has to be true for the cyclic representation generated by $\pi$ over $p_\pi^{\,\prime}(\varphi){\mathcal H}_\pi=[\pi(M)\varphi]$ with   cyclic vector $\varphi$. This however  contradicts the fact that $\nu\not\in \mathsf{ex }\,{\mathcal S}(M)$ by supposition.
\end{proof}
Complementary to Corollary \ref{auxpure1} is the following addendum to Corollary \ref{tanfolg}.
\begin{corolla}\label{auxpure2}
Let $\nu\in {\mathsf{ex}}\,{\mathcal S}(M)$, with ${\mathsf{T}}_\nu(M)\not=\{{\mathsf{0}}\}$.
For each $f\in{\mathsf{T}}_\nu(M)\backslash \{{\mathsf{0}}\}$, consider the   parameterized curve $\gamma=(\nu_t)$ given in accordance with formula \eqref{tanfolg0b}. Then, $\nu_t\in \mathsf{ex }\,{\mathcal S}( M)$ with ${\mathsf{T}}_{\nu_t}(M)\not=\{{\mathsf{0}}\}$ holds,  for all $t$ of sufficiently small modulus.
\end{corolla}
\begin{proof}
Consider the cyclic $^*$-represention  $\{\pi,{\mathcal H}_\pi\}$ generated by $\nu$, with cyclic vector $\varphi\in {\mathcal S}_{\pi,M}(\nu)$. Note that since $\nu$ is pure we then must have that $N=\pi(M)^{\,\prime\prime}={\mathsf{B}}({\mathcal H}_\pi)$. But then, each other unit vector $\xi\in {\mathcal H}_\pi$ has to be a cyclic vector as well in respect of the action of $\pi$. Accordingly,  if $\xi\in {\mathcal S}_{\pi,M}(\varrho)$ is fulfilled, the state $\varrho$ implemented by $\xi$ has to be a pure state, too.  Now,
 since ${\mathcal S}_{\pi,M}(\nu)\not=\emptyset$ is fulfilled, the parameterized curve $\gamma=(\nu_t)$ given by formula \eqref{tanfolg0b},  by Corollary \ref{tanfolg}\,\eqref{tanfolg02} for each $f\in{\mathsf{T}}_\nu(M)\backslash \{{\mathsf{0}}\}$ can be implemented by a differentiable parameterized curve $[-1,1]\ni t\longmapsto \varphi_t\in {\mathcal S}_{\pi,M}(\nu_t)$,  relative to $\pi$. Thus, by the preliminary consideration  $\nu_t\in {\mathsf{ex}}\,{\mathcal S}(M)$ is fulfilled. Obviously, the derivative $\nu_t^{\,\prime}=f_t$ of \eqref{tanfolg0b} is continuous in $t$, with $f_0=f\not={\mathsf{0}}$. Hence, for all $t$ of sufficiently small modulus $f_t\not={\mathsf{0}}$ follows. But then, according to Lemma \ref{constr}, the parameterized curve $\gamma$ arising from the differentiable implementation $(\varphi_t)$ relative to $\pi$ around $\nu_t$ is passing through any of these states $\nu_t$, too. Hence, by Theorem \ref{Cstern}, $f_t\in {\mathsf{T}}_{\nu_t}(M)$, and therefore ${\mathsf{T}}_{\nu_t}(M)\not=\{{\mathsf{0}}\}$, for all $t$ with sufficiently small modulus.
\end{proof}
Now, consider a unital ${\mathsf{C}}^*$-algebra $M$ and a  ${\mathsf{C}}^*$-subalgebra $N$ containing the unit ${\mathsf{1}}$ of  $M$, with a projection map $\Phi$ of norm $1$ onto $N$ existing, that is, a linear map
$
\Phi: M\ni x\mapsto \Phi(x)\in N
$
obeying $\Phi(x)=x$, for all $x\in N$, and $\|\Phi(x)\|\leq \|x\|$, for all $x\in M$. As a unital positive linear mapping $\Phi$ is satisfying
\begin{equation}\label{condcp}
\Phi(a x b)=a\Phi(x) b,\,\forall a,b\in N,\,\forall x\in M
\end{equation}
see \cite[III,3.]{Take:79}. The latter condition implies $\Phi$ to be completely positive  \cite[IV,3.]{Take:79}. By means of $\Phi$ the state space  of the subalgebra $N$ can be identified (affinely) isometrically with the subset
\begin{equation}\label{condinv}
{\mathcal S}^\Phi(M)=\bigl\{\nu\in {\mathcal S}(M): \nu\circ \Phi=\nu\bigr\}={\mathcal S}(N)\circ\Phi
\end{equation}
of all states of ${\mathcal S}(M)$ which are left invariant under the dual action of the map $\Phi$.
\begin{lemma}\label{condexpect}
${\mathsf{T}}_\nu(M)|N={\mathsf{T}}_{\nu|N}(N)$, for each $\nu\in {\mathcal S}^\Phi(M)$.
\end{lemma}
\begin{proof}
Let $f\in {\mathsf{T}}_\nu(M)$. In view of  $N\subset M$, by Definition \ref{e.1}, see formula \eqref{tangentnorm}, obviously $\| f|N\|_{\nu|N}\leq \|f\|_\nu<\infty$ follows. Thus we have ${\mathsf{T}}_\nu(M)|N\subset{\mathsf{T}}_{\nu|N}(N)$.
On the other hand, for $g\in {\mathsf{T}}_{\nu|N}(N)$ consider $f=g\circ\Phi$. Then, $f|N=g$, and therefore
\begin{eqnarray*}
\|f\|_\nu^2 &=& \sup_{\{x\}\subset M_+} \frac{1}{4}\,{\sum_k}^{\,\prime} \frac{f(x_k)^2}{\nu(x_k)}=\sup_{\{x\}\subset M_+} \frac{1}{4}\,{\sum_k}^{\,\prime} \frac{g(\Phi(x_k))^2}{\nu(\Phi(x_k))}\\
&\leq& \sup_{\{y\}\subset N_+} \frac{1}{4}\,{\sum_k}^{\,\prime} \frac{g(y_k)^2}{\nu(y_k)}=\|g\|_{\nu|N}^2<\infty
\end{eqnarray*}
holds. Thus, for $g\in {\mathsf{T}}_{\nu|N}(N)$ there exists $f\in {\mathsf{T}}_\nu(M)$ with $f|N=g$. Accordingly,
${\mathsf{T}}_{\nu|N}(N)\subset{\mathsf{T}}_\nu(M)|N$ holds.
In view of the above then the assertion follows.
\end{proof}
\begin{corolla}\label{condexpect0}
Let $\nu\in {\mathcal S}^\Phi(M)$ and be $\gamma=(\nu_t)\subset {\mathcal S}(N)$  a parameterized curve passing through $\nu_0=\nu|N$ with $\gamma$-compliant $^*$-representation $\{\pi,{\mathcal H}_\pi\}$ around $\nu|N$.
\begin{enumerate}
\item \label{condexpect0a}
Then, $\gamma\circ\Phi=(\nu_t\circ\Phi)\subset {\mathcal S}^\Phi(M)$ is a parameterized curve passing through $\nu$ and admitting a $\gamma\circ\Phi$-compliant unital $^*$-representation  around $\nu$;
\item\label{condexpect0b}
 Suppose $ {\mathsf{T}}_{\nu}(M)\not=\{{\mathsf{0}}\}$. Then, for $f\in {\mathsf{T}}_{\nu}(M)\backslash \{{\mathsf{0}}\}$ a parameterized curve $\hat{\gamma}=(\varrho_t)\subset  {\mathcal S}^\Phi(M)$ passing through $\varrho_0=\nu$ with $\varrho_t^{\,\prime}|_{t=0}=f$ and a $\hat{\gamma}$-compliant unital $^*$-representation around $\nu$ exist if, and only if, $f\circ\Phi=f$.
 \end{enumerate}
\end{corolla}
\begin{proof}
By supposition, we have $\nu_t(x)=\langle \pi(x)\varphi_t,\varphi_t\rangle_\pi$, for all $x\in N$, with $\nu=\nu_0$ and $t\mapsto \varphi_t\in {\mathcal S}_{\pi,N}(\nu_t)$ being a differentiable at $t=0$ implementation of $\gamma$. Now, since $\pi$ as a unital $^*$-representation is a unital completely positive linear map acting from $N$ into ${\mathsf{B}}({\mathcal H}_\pi)$, and $\Phi$ as a map from $M$ into $N$ is unital and completely positive, then the composed mapping $\pi\circ \Phi: M\rightarrow {\mathsf{B}}({\mathcal H}_\pi)$ has to be a unital, completely positive linear mapping, too. Hence, by Stinespring's theorem \cite{stine:55} there exists a unital $^*$-representation $\{\pi_1,{\mathcal H}_1\}$ of $M$ and an isometry $V: {\mathcal H}_\pi \rightarrow {\mathcal H}_1$ such that
$
\pi(\Phi(x))=V^*\pi_1(x)V,\,\forall x\in M
$ is fulfilled. From this it follows that   $t\mapsto V\varphi_t\in {\mathcal S}_{\pi_1,M}(\nu_t\circ\Phi)$ is a differentiable at $t=0$ implementation of the parameterized family $\gamma\circ\Phi=(\nu_t\circ\Phi)$ of states over $M$. Clearly, due to $\nu_t\circ\Phi|N=\nu_t$ for all $t$ and since $\gamma$ as a parameterized curve passing through $\nu|N$ is injective around $t=0$, it follows that $\gamma\circ\Phi$ has to be injective around $t=0$, hence is a parameterized curve passing through $\nu$. Also, in view of the above $\{\pi_1,{\mathcal H}_1\}$ is a $\gamma\circ\Phi$-compliant unital $^*$-representation of $M$ around $\nu$. Finally, in view of \eqref{condinv} the projection property $\Phi^2=\Phi$ implies that $\gamma\circ\Phi\subset {\mathcal S}^\Phi(M)$ is satisfied. Thus \eqref{condexpect0a} is seen.

Relating \eqref{condexpect0b}, to see sufficiency, suppose  $ {\mathsf{T}}_{\nu}(M)\not=\{{\mathsf{0}}\}$ and $f\in {\mathsf{T}}_{\nu}(M)\backslash \{{\mathsf{0}}\}$ with $f\circ \Phi=f$. Then, $f\not={\mathsf{0}}$ and $f|N\not={\mathsf{0}}$ follow. Hence, by Lemma \ref{condexpect}, for $g=f|N$ we have $g\in {\mathsf{T}}_{\nu|N}(N)\backslash \{{\mathsf{0}}\}$ and $g\circ\Phi=f$. By Corollary \ref{tanfolg}\,\eqref{tanfolg02},   formula \eqref{tanfolg0b} when applied at $\nu|N$ in context of $N$, ${\mathcal S}(N)$ and $g$ instead of $M$, ${\mathcal S}(M)$ and $f$, respectively,  in accordance with the suppositions will yield a parameterized curve $\gamma=(\nu_t)\subset {\mathcal S}(N)$ passing through $\nu|N$ and obeying $\nu_t^{\,\prime}|_{t=0}=g$. By \eqref{condexpect0a} the curve $\hat{\gamma}=\gamma\circ\Phi$ can be taken and with $\varrho_t=\nu_t\circ\Phi$ is obeying $\varrho_t^{\,\prime}|_{t=0}=g\circ \Phi=f$.

To see necessity of \eqref{condexpect0b}, note that each curve $\hat{\gamma}=(\varrho_t)\subset {\mathcal S}^\Phi(M)$ passing through $\varrho_0=\nu$ and admitting a $\hat{\gamma}$-compliant unital $^*$-representation around $\nu$ has to obey $\varrho_t^{\,\prime}|_{t=0}=f$, for some $f\in {\mathsf{T}}_\nu(M)$. Also, due to $\varrho_t\in{\mathcal S}^\Phi(M)$ there is $\nu_t\in {\mathcal S}(N)$ such that $\varrho_t=\nu_t\circ\Phi$. From this $\nu_t=\varrho_t|N$ and  $\nu_t^{\,\prime}|_{t=0}=f|N$ are seen.
Due to  $N=\Phi(M)$ then
$
f=\varrho_t^{\,\prime}|_{t=0}=(\varrho_t|N)\circ\Phi^{\,\prime}|_{t=0}=(\varrho_t|N)^{\,\prime}\circ\Phi|_{t=0}=\nu_t^{\,\prime}|_{t=0}\circ\Phi=f|N\circ\Phi=f\circ\Phi
$ follows. Thus, if $\hat{\gamma}=(\varrho_t)\subset {\mathcal S}^\Phi(M)$ holds with   $\varrho_t^{\,\prime}|_{t=0}=f$, this implies $f\in {\mathsf{T}}_\nu(M)$ with $f=f\circ\Phi$.
\end{proof}
\subsubsection{When is the tangent space trivial?}\label{trivtsp}
By Lemma \ref{extriv},  at a  mixed state the tangent space is non-trivial. Thus, if occurring at all, a trivial tangent space requires the state to be extremal.  More precisely, for a trivial tangent space at $\nu$,  the state in addition has to be a character state of the unital ${\mathsf C}^*$-algebra $M$, i.e.~
\begin{equation}\label{char}
\nu(a b)=\nu(a)\nu(b),\, \forall a,b\in M
\end{equation}
has to be satisfied. Remind some of the canonical properties of character states.
\begin{lemma}\label{char1}
On a unital ${\mathsf C}^*$-algebra $M$, $\nu \in{\mathcal S}(M)$ is a character state if, and only if, ${\mathsf{dim}}\,{\mathcal H}_\pi=1$ holds, for the cyclic $^*$-representation  $\{\pi,{\mathcal H}_\pi\}$  induced by $\nu$. Thus especially, if character states over $M$ exist these are special pure states.
\end{lemma}
\begin{proof}
Let $\varphi\in {\mathcal S}_{\pi,M}(\nu)$. Then, if ${\mathsf{dim}}\,{\mathcal H}_\pi=1$ is fulfilled, by cyclicity and owing to $\varphi\in {\mathcal S}_{\pi,M}(\nu)$, we must have that $\pi(x)\varphi=\nu(x)\varphi$, for each $x\in M$. Hence, $\nu(x y)=\nu(x)\nu(y)$ follows, that is, $\nu$ is a character of $M$. On the other hand, if $\nu$ is a character state, then
\[
\|\pi(z)\varphi\|_\pi=\sqrt{\nu(z^*z)}=\sqrt{\overline{\nu(z)}\,\nu(z)}=|\nu(z)|=|\langle\pi(z)\varphi,\varphi\rangle_\pi|
\]
for each $z\in M$. This is only possible if $\pi(z)\varphi$ is a scalar multiple of $\varphi$, for each $z\in M$, by elementary properties of the Cauchy-Schwarz estimate.  Thus, by cyclicity ${\mathsf{dim}}\,{\mathcal H}_\pi=1$ follows. In  summarizing, property \eqref{char} proves equivalent to the condition that  ${\mathsf{dim}}\,{\mathcal H}_\pi=1$ be fulfilled. If the latter condition is fulfilled however, this in a trivial way  implies the condition  $N=\pi(M)^{\,\prime\prime}={\mathsf B}({\mathcal H}_\pi)$ to be fulfilled. Hence, the cyclic $^*$-representation of $\nu$ is acting  over ${\mathcal H}_\pi$ in an  irreducible way. By standard theory then $\nu\in {\mathsf{ex}}\,{\mathcal S}(M)$ must hold.
\end{proof}
\begin{lemma}\label{char1/2}
Suppose $\nu$ is a normal character state over a ${\mathsf W}^*$-algebra $M$. Then, the support orthoprojection $z=s(\nu)$ of $\nu$ is a minimal orthoprojection and
\begin{enumerate}
\item \label{char1/2a}
$x z=z x$;
\item \label{char1/2b}
$x z=z x z=\nu(x) z$.
\end{enumerate}
are fulfilled, for each $x\in M$.
\end{lemma}
\begin{proof}
For a normal state $\nu$, the support $z=s(\nu)$ is an orthoprojection in $M$. By the meaning of $z$ we have $\nu(z^\perp)=0$, and since $\nu$ is a  character state, for each $x\in M$ we have $\nu(zx^*z^\perp xz)=\nu(z)\nu(x^*)\nu(z^\perp)\nu(x^*)\nu(z)=0$. Hence, $zx^*z^\perp xz={\mathsf{0}}$, or equivalently  $z^\perp xz={\mathsf{0}}$ follows. For hermitian $x\in M$ from this $xz=z x z=(z x z)^*=z x$ is following.  But in a ${\mathsf C}^*$-algebraic situation, this yields commutation of $z$ with each $x\in M$. This is  \eqref{char1/2a}. We are going to show that $z$ is a minimal orthoprojection of $M$. For an orthoprojection $q\in M$ obeying  ${\mathsf{0}}<q\leq z$ we have $0<\nu(q)\leq 1$. Thus, since $\nu$ is a character state, for the normal state $\varrho$  defined by $$\varrho(x)=\frac{\nu(q x q)}{\nu(q)}=\nu(q)^2 \frac{\nu(x)}{\nu(q)}=\nu(q)\nu(x)$$
for all $x\in M$ we infer that $1=\varrho({\mathsf{1}})=\nu(q)\nu({\mathsf{1}})=\nu(q)$. Hence, $q\geq z$ has to hold. In view of the assumption made on $q$ then $q=z$ follows. That is, $z$ is a minimal orthoprojection in $M$. But then, for hermitian $x$ we infer that $z x z=\lambda z$, for some $\lambda\in {\mathbb R}$. Since $\nu(x)=\nu(z x z)=\lambda \nu(z)=\lambda$ has to be fulfilled, $z x z=\nu(x) z$ follows, for hermitian $x\in M$. In a ${\mathsf{C}}^*$-algebraic situation according to this and \eqref{char1/2a} then  \eqref{char1/2b} can be followed, for all $x\in M$.
\end{proof}
\begin{lemma}\label{char2}
  Let $M$ be a unital ${\mathsf C}^*$-algebra admitting a character state $\nu$. Then, $\nu$ and $\varrho$ are mutually disjoint states, for each $\varrho\in {\mathsf{ex}}\,{\mathcal S}(M)$ with $\nu\not=\varrho$.
  \end{lemma}
  \begin{proof}
   Suppose $\varrho\in {\mathsf{ex}}\,{\mathcal S}(M)$ with $\nu\not=\varrho$, and be
  $\{\pi,{\mathcal H}_\pi\}$ a unital $^*$-representation of $M$ such that the $\pi$-fibres of both $\nu$ and $\varrho$ are non-void. Let $N=\pi(M)^{\,\prime\prime}$. The claim is that a central orthoprojection $z\in N\cap N^{\,\prime}$ exists such that both the relations
  \begin{equation}\label{chardis}
  z {\mathcal S}_{\pi,M}(\nu)={\mathcal S}_{\pi,M}(\nu),\ z^\perp {\mathcal S}_{\pi,M}(\varrho)={\mathcal S}_{\pi,M}(\varrho)
  \end{equation}
  are fulfilled. We are going to construct the orthoprojection $z$ now.
  Let $\varphi\in {\mathcal S}_{\pi,M}(\nu)$.  Consider the vector state $\nu_\pi$ defined on $N$ by $\nu_\pi(x)=\langle x\varphi,\varphi\rangle_\pi$, for all $x\in N$. Since according to \eqref{char} for the normal state $\nu_\pi$ we have $\nu_\pi(\pi(a)\pi(b))=\nu(ab)=\nu(a)\nu(b)=\nu_\pi(\pi(a)) \nu_\pi(\pi(b))$, for any $a,b\in M$, $\nu_\pi$ behaves multiplicatively over the strongly dense $^*$-subalgebra $\pi(M)$ of $N$. Since operator multiplication is continuous on bounded spheres, by a Kaplansky density argument $\nu_
  \pi$ is seen to be a normal character state on the $vN$-algebra $N$. Thus, Lemma \ref{char1/2} can be applied with $N$ (instead of $M$).  Let $z=s(\nu_\pi)$ be the support orthoprojection of $\nu_\pi$ within $N$. According to Lemma \ref{char1/2}\,\eqref{char1/2a} we have $z\in N\cap N^{\,\prime}$, and owing to  $0=\nu_\pi(z^\perp)=\|z^\perp\varphi\|_\pi$ we have  $z\varphi=\varphi$. From this the first of the relations of \eqref{chardis} follows. On the other hand, Lemma \ref{char1/2}\,\eqref{char1/2b} implies $z$ to be minimal in $N$, that is, $N=\mathbb{ C}z+Nz^\perp$ holds. In line with this we infer $$\pi(x)=\pi(x)z+\pi(x)z^\perp=\nu(x)z+\pi(x)z^\perp$$ for all $x\in M$. Accordingly, for each $\eta\in {\mathcal S}_{\pi,M}(\varrho)$ and each $x\in M$ the conclusion is
 \[
 \varrho(x^*x)=\langle\pi(x^*x)\eta,\eta\rangle_\pi=\| z\eta\|_\pi^2 \nu(x^*x)+\| z^\perp\pi(x)\eta\|_\pi^2 \geq \| z\eta\|_\pi^2 \nu(x^*x)
 \]
Thereby, the case $\| z\eta\|_\pi^2=1$ cannot occur since then we had  $\nu(x^*x)\leq \varrho(x^*x)$, for all $x\in M$. For states the latter would imply
  $\nu(x^*x)=\varrho(x^*x)$, i.e.~we had $\nu=\varrho$, a contradiction. On the other hand, for $0<\| z\eta\|_\pi^2<1$  from the above estimate
  \[
  \varrho=\| z\eta\|_\pi^2\nu+\bigl(1-\| z\eta\|_\pi^2\bigr)\,\omega
  \]
  had to be followed,
  with the state
  $
  \omega=\bigl(1-\| z\eta\|_\pi^2\bigr)^{-1}\bigl(\varrho-\| z\eta\|_\pi^2\nu\bigr)
  $. Since $\varrho$ is pure by assumption, from the previous we had $\nu=\omega=\varrho$, which were contradicting to the other assumption $\nu\not=\varrho$ again.
  Hence, $z\eta={\mathsf{0}}$ has to hold,  for each $\eta\in {\mathcal S}_{\pi,M}(\varrho)$.
  \end{proof}
\begin{remark}\label{char2rem}
Let $\{\pi,{\mathcal H}_\pi\}$ be a unital $^*$-representation of $M$ such that the $\pi$-fibres of both $\nu$ and $\varrho$ are non-void. Then, if  $\nu_\pi$ is the normal state implemented by each $\varphi\in {\mathcal S}_{\pi,M}(\nu)$ over $N=\pi(M)^{\,\prime\prime}$, according to the previous proof in the disjointness condition \eqref{chardis}, $z\in N\cap N^{\,\prime}$ is the support orthoprojection $z=s(\nu_\pi)$.
\end{remark}
\begin{corolla}\label{extan}
Let $M$ be a unital ${\mathsf C}^*$-algebra, and $\nu \in{\mathcal S}(M)$. Then, ${\mathsf{T}}_\nu(M)=\{{\mathsf 0}\}$ if, and only if, $\nu$ is a character state of $M$.
Especially, ${\mathsf{T}}_\nu(M)=\{{\mathsf 0}\}$ is fulfilled at each $\nu \in {\mathsf{ex}}\,{\mathcal S}(M)$ if, and only if, $M$ is   commutative.
\end{corolla}
\begin{proof}
It is useful to refer to the cyclic $^*$-representation  $\{\pi,{\mathcal H}_\pi\}$  induced by the state $\nu$ in question, with $N=\pi(M)^{\,\prime\prime}$ and cyclic implementing vector $\varphi\in {\mathcal S}_{\pi,M}(\nu)$.

Suppose $\nu$ is a character state of $M$. Let $f\in {\mathsf{T}}_\nu(M)$. Application of Corollary \ref{tanfolg} with $\pi_0=\pi$ and $\varphi_0=\varphi$ yields $\hat{\xi}\in [N_{\mathsf h}\varphi]\subset {\mathcal H}_\pi$, with $\|f\|_\nu=\|\hat{\xi}\|_\pi$ and
\[
f(x)=\langle\pi(x)\hat{\xi},\varphi\rangle_\pi+\langle\pi(x)\varphi,\hat{\xi}\rangle_\pi
\]
for any $x\in M$. Now, by Lemma \ref{char1},   ${\mathsf{dim}}\,{\mathcal H}_\pi=1$ holds. Hence,  $\hat{\xi}=\lambda \varphi$, with $\lambda\in {\mathbb R}$. That is, according to the above formula and by choice of $\varphi$, we get  $f=2 \lambda\, \nu$. Thus, $f({\mathsf 1})=2\lambda$ follows. Since $f({\mathsf 1})=0$ holds, for each $f\in {\mathsf{T}}_\nu(M)$, $\lambda=0$ follows. Hence, for a character state $\nu$, $\hat{\xi}=\lambda \varphi={\mathsf{0}}$, and from $\|f\|_\nu=\|\hat{\xi}\|_\pi$ then  $f= {\mathsf 0}$ is obtained, for each $f\in {\mathsf{T}}_\nu(M)$. That is,  ${\mathsf{T}}_\nu(M)=\{{\mathsf{0}}\}$ is fulfilled.

Suppose $\nu$ is not a character state. By Lemma \ref{char1} then  ${\mathsf{dim}}\,{\mathcal H}_\pi>1$ must hold.  We are going to show that then ${\mathsf{T}}_\nu(M)\not=\{{\mathsf{0}}\}$ is fulfilled.
Two cases can appear in this context: either $\nu$ is a mixed state, or $\nu\in {\mathsf{ex}}\,{\mathcal S}(M)$. According to Lemma \ref{extriv} at a mixed state ${\mathsf{T}}_\nu(M)\not=\{{\mathsf{0}}\}$ follows, and we are done.
If $\nu$ is extremal, then the action of $\pi$ on ${\mathcal H}_\pi$ has to be irreducible, that is,  $N=\pi(M)^{\,\prime\prime}={\mathsf B}({\mathcal H}_\pi)$ has to be fulfilled.  In this case let an orthoprojection $p\in {\mathsf B}({\mathcal H}_\pi)$ be chosen such that $$0< \|p\varphi\|_\pi < 1$$
Due to ${\mathsf{dim}}\,{\mathcal H}_\pi>1$ this choice is possible. Let $\nu_\pi$ be the vector state on $N=\pi(M)^{\,\prime\prime}={\mathsf B}({\mathcal H}_\pi)$ implemented by $\varphi$. By the previous this  implies
$
0<\nu_\pi(p)<1
$.
We may consider a hermitian element $a\in N$ defined with the help of $p$ as following
\begin{equation*}
a=p-\frac{\nu_\pi(p)}{\nu_\pi(p^\perp)}\,p^\perp
\end{equation*}
From this we get that
\begin{subequations}\label{psq}
\begin{eqnarray}\label{psqa}
\nu_\pi(a) &=& 0\\
\label{psqb}
\nu_\pi(a^2)& =&\nu_\pi(p)\,\biggl(1+\frac{\nu_\pi(p)}{\nu_\pi(p^\perp)}\biggr)\not=0
\end{eqnarray}
\end{subequations}
With the help of $a$ define a linear form $f_\pi$ over $N$ by
\[
f_\pi(\cdot)=(\nu_\pi)_a(\cdot)=\re\nu_\pi((\cdot)a)
\]
By definition, $f_\pi$ is hermitian and according to \eqref{psqa} is obeying $f_\pi({\mathsf 1})=0$. Also, according to Remark \ref{scalar} and Lemma \ref{normformel}, we have
\[
\|f_\pi\|_{\nu_\pi}=\frac{1}{2}\,\sqrt{\nu_\pi(a^2)}
\]
In view of Definition \ref{trdef}, $f_\pi\in {\mathsf{T}}_{\nu_\pi}(N)$ holds, with $f_\pi\not={\mathsf 0}$ due to \eqref{psqb} and by the previous formula. Define $f\in M_{\mathsf h}^*$ by $f(\cdot)=f_\pi(\pi(\cdot))$. Then, Corollary \ref{haupt}, see especially \eqref{densend}, when applied in respect of ${\mathfrak A}=\pi(M)$, $f_\pi$, $\nu_\pi$ provides that
\begin{equation*}
\|f\|_\nu=\|f_\pi|{\pi(M)}\|_{{\nu_\pi}|{\pi(M)}}=\|f_\pi\|_{\nu_\pi}\not=0
\end{equation*}
Hence, we have $f\not={\mathsf 0}$ belonging to ${\mathsf{T}}_{\nu}(M)$.
This proves that ${\mathsf{T}}_\nu(M)\not=\{{\mathsf 0}\}$ also in case that $\nu$ is an extremal, but non-character state. Thus, in summarizing,  ${\mathsf{T}}_\nu(M)=\{{\mathsf 0}\}$ holds if, and only if, $\nu$ is a character state on $M$.
Finally, note that for a commutative $M$ and pure state $\nu$, the correspondent cyclic $^*$-representation by extremality  has to obey  $\pi(M)^{\,\prime\prime}={\mathsf B}({\mathcal H}_\pi)$, but due to  commutativity of $M$ has to be commutative, too. Hence, ${\mathsf{dim}}\,{\mathcal H}_\pi=1$ has to hold, for each pure state.  By Lemma \ref{char1}, the state $\nu$ has to be a character, which by the just proven implies ${\mathsf{T}}_\nu(M)=\{{\mathsf{0}}\}$, at each pure state $\nu$.  On the other hand, suppose ${\mathsf{T}}_\nu(M)=\{{\mathsf{0}}\}$, at each $\nu\in {\mathsf{ex}}\,{\mathcal{S}}(M)$. By the just proven each pure state has to be a character state, that is behaves   multiplicatively. Since the set of states is separating the elements of $M$, for each $x\in M\backslash \{{\mathsf{0}}\}$  there is $\nu\in {\mathsf{ex}}\,{\mathcal S}(M)$ such that $\nu(x)\not=0$. But then, since for $x,y\in M$ we have  $\nu(xy-yx)=\nu(x)\nu(y)-\nu(y)\nu(x)=0$, for any $\nu\in {\mathsf{ex}}\,{\mathcal S}(M)$, conclude that  $xy-yx={\mathsf 0}$ must hold, for any choice of $x,y\in M$. Hence, provided ${\mathsf{T}}_\nu(M)=\{\mathsf 0\}$ holds at  each pure state $\nu$, commutativity of $M$ follows.
\end{proof}
\subsection{Characterizing local tangent spaces}\label{finaltan}
Let $M$ be a unital ${\mathsf C}^*$-algebra with state space ${\mathcal S}(M)$. The aim will be to introduce notions for `submanifold' (of states) and `local tangent space' (at a state of a submanifold) which are suitable in context of Bures geometry. As it will turn out, in order to be applicable to the full state space, the notion as a tacit assumption requires the underlying unital  ${\mathsf{C}}^*$-algebra $M$ to be restricted by the condition ${\mathsf{dim}}\,M\geq 3$. It is easy to see that the latter condition is equivalent to the condition $\#\, {\mathsf{ex}}\,{\mathcal S}(M)>2$. In particular, each non-commutative unital ${\mathsf{C}}^*$-algebra is conform to this constraint.
\subsubsection{Submanifolds of states \textup{(}generalities\textup{)}}\label{finaltanspez} Start with the definition.
\begin{definition}\label{submann}
A non-void subset of states  $\Omega\subset {\mathcal S}(M)$ is termed  `submanifold of states' if for each state $\nu\in \Omega$ the set of all  parameterized curves $\gamma$  passing through $\nu$ and thereby evolving completely within $\Omega$, $\gamma\subset \Omega$, and admitting a  $\gamma$-compliant unital $^*$-representation of $M$ around $\nu$, is non-void.
 \end{definition}
A class of examples, and at the same time a key tool for identifying further examples of submanifolds, is given by the subsets ${\mathcal{S}}^\pi(M)$ defined in  \eqref{pimann1}, for any given unital $^*$-representation $\{\pi,{\mathcal H}_\pi\}$ of $M$. This will be described now.
Relating to these settings, and with   $N=\pi(M)^{\,\prime\prime}$, the following result holds.
\begin{lemma}\label{submannele}
Let $M$ be a unital ${\mathsf{C}}^*$-algebra. The subset ${\mathcal{S}}^\pi(M)$ is a submanifold of states if, and only if, ${\mathsf{dim}}\,N>2$.
\end{lemma}
\begin{proof}
We may content with showing that under the condition mentioned through each state $\nu\in {\mathcal S}^\pi(M)$ there is passing a parameterized curve $\gamma=(\nu_t)$, with $\nu_0=\nu$, and admitting  $\{\pi,{\mathcal H}_\pi\}$ as a $\gamma$-compliant unital $^*$-representation around there. In fact, from $\gamma$-compliance of $\pi$ around $\nu$ the fact that $\gamma|J\subset {\mathcal{S}}^\pi(M)$ holds, for some symmetric interval $J$ around zero, will follow automatically, and thus the subset ${\mathcal{S}}^\pi(M)$ at $\nu$ will be conform to requirements of Definition \ref{submann}, with $\gamma|J$ considered  instead of $\gamma$.
Prior to starting the construction of $\gamma$, remark that if ${\mathsf{dim}}\,N\leq 2$ holds, existence of a parameterized curve $\gamma$  of the mentioned specification with  $\gamma\subset {\mathcal{S}}^\pi(M)$ and which were passing through an extremal state of ${\mathcal{S}}^\pi(M)$ will fail. In fact, by definition the notion of  `passing through' for a parameterized curve $\gamma$ requires injectivity of the map in some vicinity of the state of interest. But injectivity around an extremal state of ${\mathcal{S}}^\pi(M)$ necessarily will be violated since  ${\mathcal S}^\pi(M)$  either will be a single extremal state, or will be the (affinely convex) interval spanned by two extremal states. This follows since the $vN$-algebra $N$ then is commutative and generated by one or two minimal orthoprojections. Hence, the state space ${\mathcal S}(N)$ either consists of the only character state over $N$, or is the (affinely convex) interval spanned by the two character states over $N$. At the same time, these then are the only states which can be implemented by vectors over $N$.  Hence, in view of \eqref{pimann1},  ${\mathcal{S}}^\pi(M)$ then either can consist of a single character state, or is the (affinely convex) interval spanned by two character states of $M$.  Thus, the condition ${\mathsf{dim}}\,N>2$ is necessary for the assertion be true.

Suppose ${\mathsf{dim}}\,N> 2$ holds (remark that in order that such case can occur the condition ${\mathsf{dim}}\,M> 2$ has to be satisfied). As mentioned above, we are going to show that to any $\nu\in {\mathcal{S}}^\pi(M)$ a parameterized curve $\gamma$ exists which is passing through $\nu$ and is admitting  $\{\pi,{\mathcal H}_\pi\}$ as a $\gamma$-compliant $^*$-representation. To start with, note that if $\nu$ is a state obeying ${\mathsf{T}}_\nu(M)
\not=\{{\mathsf{0}}\}$, then to each $f\in {\mathsf{T}}_\nu(M)
\backslash\{{\mathsf{0}}\}$ the special parameterized curve $\gamma=(\nu_t)\subset {\mathcal S}(M)$ passing through $\nu$ and obeying $\nu_t^{\,\prime}|_{t=0}=f$ which is given by formula \eqref{tanfolg0b} can be considered.  Thereby, in line with Corollary \ref{tanfolg}\,\eqref{tanfolg02} this curve is admitting $\{\pi,{\mathcal H}_\pi\}$ as a $\gamma$-compliant $^*$-representation, since it is even satisfying $\gamma\subset {\mathcal{S}}^\pi(M)$.  Assume now $\nu$ to be a state obeying ${\mathsf{T}}_\nu(M)=\{{\mathsf{0}}\}$. By Corollary \ref{extan}, the state $\nu$ then has to be a character state of $M$. Let $\varphi\in {\mathcal S}_{\pi,M}(\nu)$.  Consider the vector state $\nu_\pi$ defined on $N$ by $\nu_\pi(x)=\langle x\varphi,\varphi\rangle_\pi$, for all $x\in N$. Since $\nu$ is a character state, as in the proof of Lemma \ref{char2} we infer that $\nu_\pi$ is a  normal character state on the $vN$-algebra $N$. Thus, Lemma \ref{char1/2} can be applied with $N$ (instead of $M$).  Let $z=s(\nu_\pi)$ be the support orthoprojection of $\nu_\pi$ within $N$. According to Lemma \ref{char1/2}\,\eqref{char1/2a} we have $z\in N\cap N^{\,\prime}$, and owing to  $0=\nu_\pi(z^\perp)=\|z^\perp\varphi\|^2_\pi$ we have  $z\varphi=\varphi$.  On the other hand, by Lemma \ref{char1/2}\,\eqref{char1/2b} we have minimality of $z$ within $N$, that is, $N=\mathbb{ C}z+Nz^\perp$ has to be fulfilled. Therefore, owing to ${\mathsf{dim}}\,N> 2$ the $vN$-algebra $N z^\perp$ acting over $z^\perp {\mathcal H}_\pi$ has to be at least two-dimensional. Thus, there are unit vectors $\eta, \xi\in {\mathcal H}_\pi$ obeying $z\eta=z\xi={\mathsf{0}}$ and such that $\varrho_\pi\not=\mu_\pi$ holds, for the vector states $\varrho_\pi$ and $\mu_\pi$ implemented by $\xi$ and $\eta$ over $N$. Accordingly, if $\varrho, \mu\in {\mathcal S}(M)$ are defined by $\xi\in {\mathcal S}_{\pi,M}(\varrho)$ and $\eta\in {\mathcal S}_{\pi,M}(\mu)$, then  $\varrho\not=\mu$ follows and the following relations are fulfilled:
\begin{subequations}\label{mannig2n}
\begin{equation}\label{mannig2cn}
  z\varphi=\varphi,\ z\xi=z\eta={\mathsf{0}}
  \end{equation}
with the orthoprojection  $z\in N\cap N^{\,\prime}$ obeying ${\mathsf{0}}<z<{\mathsf{1}}$. From this obviously
\begin{equation}\label{charbedn}
  \langle \pi(x)\varphi,\xi\rangle_\pi=\langle \pi(x)\xi,\varphi\rangle_\pi=0=\langle \pi(x)\varphi,\eta\rangle_\pi=\langle \pi(x)\eta,\varphi\rangle_\pi
\end{equation}
follows,  for each $x\in M$. From \eqref{charbedn} the orthogonality relations $\xi\perp \varphi$ and $\eta\perp \varphi$ are seen to hold. Thus, if we let a  map  $[-1,1]\ni t\longmapsto\varphi_t\in {\mathcal H}_\pi$ be defined by \begin{equation}\label{mannig2an} \varphi_t =
  \begin{cases}
  t\,\xi +\sqrt{1-t^2}\,\varphi & \text{ for }\phantom{\ \,-}0\leq t\leq 1\\
  & \\
  t\,\eta +\sqrt{1-t^2}\,\varphi & \text{ for }-1\leq t\leq 0
  \end{cases}
\end{equation}
then each $\varphi_t$ is a unit vector. Obviously, $(\varphi_t)$ is differentiable, with $\varphi_t^{\,\prime}|_{t=0}={\mathsf{0}}$.
In addition, due to \eqref{charbedn} the state $\nu_t$ implemented by $\varphi_t$ in respect of $\pi$ then reads
\begin{equation}\label{mannig2bn}
  \nu_t =
  \begin{cases}
  t^2\varrho+(1-t^2)\,\nu & \text{ for }\phantom{\ \,-}0\leq t\leq 1\\
  & \\
  t^2\mu +(1-t^2)\,\nu & \text{ for }-1\leq t\leq 0
  \end{cases}
\end{equation}
\end{subequations}
Let $(\nu_t)_\pi$ be the vector state implemented by $\varphi_t$ on $N$. Then, the map
$\gamma_\pi: [-1,1]\ni t\longmapsto (\nu_t)_\pi$
is injective. In fact, assuming  $(\nu_t)_\pi=(\nu_s)_\pi$ for $s,t\in [-1,1]$ and $s\not= t$, owing to $z=s(\nu_\pi)\in N$ and \eqref{mannig2cn} and \eqref{mannig2an} would imply
$(1-t^2)=(\nu_t)_\pi(z)=(\nu_s)_\pi(z)=(1-s^2)$. Thus $t=-s\not=0$ had to be followed. Hence, since $(\nu_\alpha)_\pi(\pi(x))=\nu_\alpha(x)$ holds, for all $x\in M$ and any $\alpha\in [-1,1]$, in view of \eqref{mannig2bn} from the previous $\varrho=\mu$ had to be followed. This   contradicted the choice of $\varrho$ and $\mu$. Therefore, $(\nu_t)_\pi\not=(\nu_s)_\pi$ for all $s,t\in [-1,1]$ with $s\not= t$, and thus  $\gamma_\pi$ has to be  injective. Now, due to normality, $(\nu_t)_\pi$ over $N$ is uniquely determined by its restriction onto the strongly dense $^*$-subalgebra $\pi(M)$. Hence, injectivity of $\gamma: [-1,1]\ni t\longmapsto \nu_t$ follows.
Thus, by construction  $\gamma\subset {\mathcal S}^\pi(M)$ is a parameterized curve passing through $\nu$, and with  $\{\pi,{\mathcal H}_\pi\}$ being $\gamma$-compliant around $\nu$.
\end{proof}
\begin{theorem}\label{mannig}
${\mathcal S}(M)$ is a submanifold of states if, and only if, $\#\, {\mathsf{ex}}\,{\mathcal{S}}(M)>2$.
\end{theorem}
\begin{proof}
In view of Definition \ref{submann}, one has to show that under the condition mentioned through each state $\nu\in {\mathcal S}(M)$ there is passing a parameterized curve $\gamma$ admitting a $\gamma$-compliant unital $^*$-representation $\{\pi,{\mathcal H}_\pi\}$ around there.

First note that
if the condition $\#\, {\mathsf{ex}}\,{\mathcal S}(M)\leq 2$ holds, then existence of a parameterized curve $\gamma$ of the mentioned specification and passing through a state $\nu$ will fail in case $\nu$ is extremal. In fact, by definition the notion of  `passing through' for a parameterized curve $\gamma$ requires injectivity of the map in some vicinity of the state. But around an extremal state injectivity necessarily is violated since  ${\mathcal S}(M)$  then either will be a single extremal state, or will be the (affinely convex) interval spanned by two extremal states. Thus, the condition $\#\, {\mathsf{ex}}\,{\mathcal S}(M)> 2$ is necessary for the assertion be true.

Suppose the mentioned condition is fulfilled. Then, the assertion holds at each $\nu$ with ${\mathsf{T}}_\nu(M)\not=\{{\mathsf{0}}\}$. In fact, by
Corollary \ref{tanfolg}\,\eqref{tanfolg02}, to each $f\in {\mathsf{T}}_\nu(M)\backslash \{{\mathsf{0}}\}$ the parameterized curve $\gamma$ of \eqref{tanfolg0b} is passing through $\nu$ and exhibiting $f$ as tangent form at $\nu$. Also, each unital $^*$-representation $\{\pi,{\mathcal H}_\pi\}$ with non-trivial $\pi$-fibre of $\nu$ then is $\gamma$-compliant.

Suppose a state $\nu$ with ${\mathsf{T}}_\nu(M)=\{{\mathsf{0}}\}$ to exist. Let $\nu$ be such a state. Then, by Corollary \ref{extan}, $\nu$ is a character state of $M$. Since $\#\, {\mathsf{ex}}\,{\mathcal S}(M)> 2$ is fulfilled,  another two extremal states $\varrho,\mu\in {\mathsf{ex}}\,{\mathcal S}(M)$ exist with $\nu\not=\varrho$ and $\nu\not=\mu$. Consider a unital $^*$-representation  $\{\pi,{\mathcal H}_\pi\}$ such that the $\pi$-fibres of the three states are non-trivial. Let $N=\pi(M)^{\,\prime\prime}$. By Lemma \ref{char2},  a central orthoprojection $z\in N\cap N^{\,\prime}$ exists such that $N={\mathbb{C}}\,z+N z^\perp$ and
\[
z {\mathcal S}_{\pi,M}(\nu)={\mathcal S}_{\pi,M}(\nu),\ z^\perp {\mathcal S}_{\pi,M}(\varrho)={\mathcal S}_{\pi,M}(\varrho),\ z^\perp {\mathcal S}_{\pi,M}(\mu)={\mathcal S}_{\pi,M}(\mu)
\]
are fulfilled. Hence, the $vN$-algebra $N$ at least admits two vector states $\varrho_\pi\not=\mu_\pi$ with  supports obeying $s(\varrho_\pi)\leq z^\perp$ and $s(\mu_\pi)\leq z^\perp$. Hence, $N z^\perp$ is of dimension two at least, and therefore ${\mathsf{dim}}\,N>2$ has to hold.
Thus, in view of Lemma \ref{submannele}, ${\mathcal S}^\pi(M)$ is a submanifold obeying $\nu\in {\mathcal S}^\pi(M)\subset {\mathcal S}(M)$, and  a parameterized curve passing through $\nu$ and admitting $\{\pi,{\mathcal H}_\pi\}$ as $\gamma$-compliant $^*$-representation exists.
\end{proof}
To given state $\nu$ of a unital ${\mathsf{C}}^*$-algebra $M$ with  $M\not=\mathbb{C
}{\mathsf{1}}$, let us consider the stratum  $\Omega_M(\nu)$ of the state $\nu$ as defined by \eqref{stratum0}.  For strata the following holds.
\begin{theorem}\label{substrat}
$\Omega=\Omega_M(\nu)$  is a submanifold of states  if, and only if, $\nu\not\in {\mathsf{ex}}\,{\mathcal{S}}(M)$.
\end{theorem}
\begin{proof}
Clearly, in order that $\Omega$ is a submanifold of states requires $\#\,\Omega > 1$, at least. From Corollary \ref{multistrat} we know that the latter condition is equivalent to $\Omega=\Omega_M(\nu)$, with $\nu\not\in  {\mathsf{ex}}\,{\mathcal{S}}(M)$. Thus, the proof will be done if we can show that a stratum $\Omega=\Omega_M(\nu)$ obeying $\#\,\Omega > 1$ is a submanifold of states. This will be done now.

For $\varrho\in\Omega_M(\nu)$ we have $\Omega_M(\varrho)=\Omega_M(\nu)$. Thus, $\nu$ can stand for any state $\varrho$ in $\Omega$, and one may content oneself with constructing a parameterized curve $\gamma\subset \Omega_M(\nu)$ which is passing through $\nu$ and is admitting a $\gamma$-compliant unital  $^*$-representation around there. In line with this, suppose $\varrho\in \Omega_M(\nu)$ with $\varrho\not=\nu$ exists. By \eqref{stratum}, we then especially have $\varrho \dashv \nu$ and $\nu \dashv \varrho$, see Definition \ref{absstet}. Also, by Theorem \ref{suppabsstet}, to given unital $^*$-representation $\{\pi,{\mathcal{H}}_\pi\}$ such that the $\pi$-fibres of $\varrho$ and $\nu$ are non-void, for chosen $\varphi\in {\mathcal{S}}_{\pi,M}(\nu)$ unique  $\zeta \in {\mathcal{S}}_{\pi,M}(\varrho)$ exists  which is obeying $h^\pi_{\zeta,\varphi}\geq {\mathsf{0}}$ on $N^{\,\prime}$, with   $N=\pi(M)^{\,\prime\prime}$. Then, $F(M|\nu,\varrho)=\|h^\pi_{\zeta,\varphi}\|_1=h^\pi_{\zeta,\varphi}({\mathsf{1}})=\langle\zeta,\varphi\rangle_\pi$ by Lemma \ref{bas3}, and  in view of Theorem \ref{suppabsstet}, \eqref{supp2} and \eqref{supp3}, the conditions
\begin{subequations}\label{substrat1}
\begin{eqnarray}\label{substrat1a}
p_\pi(\zeta)& =& p_\pi(\varphi)\\
\label{substrat1b}
\zeta\phantom{)}&\in &[N_+\varphi] =[\pi(M)_+\varphi]
\end{eqnarray}
\end{subequations}
are satisfied. Note that  $\varrho\not=\nu$ implies $F=F(M|\nu,\varrho)<1$. Then, in considering the vector $\psi$ given by  $\psi=\zeta-F \varphi$, one can be assured that $\psi\not={\mathsf{0}}$ and $\Re\langle\psi,\varphi\rangle_\pi=0$ hold, and due to the conditions of \eqref{substrat1} then $\psi$ is satisfying
\begin{subequations}\label{substrat12}
\begin{eqnarray}\label{substrat1c}
p_\pi(\varphi)\psi&=&\psi\\
\label{substrat1d}
\psi&\in& [N_{\mathsf{h}}\varphi]=[\pi(M)_{\mathsf{h}}\varphi]
\end{eqnarray}
As a consequence of Theorem \ref{Cstern}, if a form $f$ generated by $\psi$ in accordance with
\begin{equation}\label{substrat1e}
f(\cdot)=\langle\pi(\cdot)\psi,\varphi\rangle_\pi+\langle\pi(\cdot)\varphi,\psi\rangle_\pi
\end{equation}
\end{subequations}
is considered, then in view of \eqref{substrat1d} and $\Re\langle\psi,\varphi\rangle_\pi=0$ we have $f\in {\mathsf{T}}_\nu(M)$, with  $\|f\|_\nu=\|\psi\|_\pi\not=0$. Note that, in terms of the notions used in Theorem \ref{Cstern}, we see that    $\psi=\hat{\psi}_0$ is fulfilled in the case at hand. Hence, $f\in {\mathsf{T}}_\nu(M)\backslash \{{\mathsf{0}}\}$, and then, as has been commented on in context of \eqref{gff0}, one can be assured that $\gamma=(\nu_t)$ defined as in \eqref{gff3} is a parameterized curve passing through $\nu$ such that  $\{\pi,{\mathcal{H}}_\pi\}$ is $\gamma$-compliant and $f$ is the tangent form at $\nu$ along $\gamma$. Thereby, note that according to condition \eqref{substrat1c},
the construction procedure \eqref{gff1} for the $\nu_t$ implementing vectors $\varphi_t$ assures that  $p_\pi(\varphi_t)\leq p_\pi(\varphi)$ is fulfilled, for all $t$ with $-1\leq t \leq 1$. Since this is the same as $s((\nu_t)_\pi)\leq s(\nu_\pi)$, from applying Theorem \ref{suppabsstet} to $\nu_t,\,\nu$ we infer that $\nu_t\dashv \nu$, for the above range of $t$. Hence, $\gamma\subset \Omega^0_M(\nu)$ follows, see \eqref{vorstratum}. Now, let us consider the combined family $(\mu_t)$ of states given by
\begin{subequations}\label{substrat23}
\begin{equation}\label{substrat2}
\mu_t= t^2\nu+(1-t^2)\nu_t
\end{equation}
It is easy to see that $t\mapsto \mu_t$ is continuously differentiable, with $\mu_0=\nu$ and $\mu_t^{\,\prime}|_{t=0}=f$.
Thus, by Lemma \ref{constr}, around $t=0$ by the map $t\mapsto \mu_t$ a parameterized curve passing through $\nu$ and exhibiting $f$ as tangent form at $\nu$ is defined. Moreover, by Lemma \ref{setprop}\,\eqref{setprop1} the set $\Omega^0_M(\nu)$ is affinely convex. Thus, due to $\nu,\nu_t\in \Omega^0_M(\nu)$ and in view of \eqref{substrat2}, we can be sure that $\mu_t \in\Omega^0_M(\nu)$, that is, $\mu_t\dashv \nu$ holds. On the other hand, \eqref{substrat2} at $t\not=0$ implies $t^2\nu\leq \mu_t$. According to Remark \ref{saka}\,\eqref{saka1} domination of states is a special case of $\dashv$. Thus $\nu\dashv \mu_t$ follows. Hence, in summarizing, we infer that $\mu_t\dashv\vdash\nu$, at each $t\not=0$. From the latter and since $\mu_0=\nu$ holds, we may conclude that $\mu_t\in \Omega_M(\nu)$, for all $t\in [-1,1]$. Thus, the parameterized curve $(\mu_t)$ is passing through $\nu$ and evolving in $\Omega_M(\nu)$. Define vectors $\xi_t\in {\mathcal H}_\pi \oplus {\mathcal H}_\pi $ by
\begin{equation}\label{substrat3}
\xi_t=\sqrt{1-t^2}\,\varphi_t\oplus t\,\varphi
\end{equation}
\end{subequations}
Then, in respect of the direct sum representation $\pi\oplus \pi$ acting over the direct sum Hilbert space  ${\mathcal H}_\pi \oplus {\mathcal H}_\pi$, for all parameter values sufficiently small by the map $t\mapsto \xi_t\in {\mathcal{S}}_{\pi,M}(\mu_t)$ a continuously differentiable implementation of $(\mu_t)$ around $\nu=\mu_0$ is given. Hence,  the parameterized curve $(\mu_t)\subset \Omega_M(\nu)$ is passing through $\nu$ and in addition admits a compliant unital $^*$-representation around $\nu$.
\end{proof}

On assuring existence of a parameterized curve $\gamma\subset \Omega_M(\nu)$ passing through $\nu$ and admitting a $\gamma$-compliant unital $^*$-representation around $\nu$, the previous proof provides some more general and useful insight.
\begin{corolla}\label{extsubstrat}
Let $\nu\in \Omega$, with $\Omega=\Omega_M(\omega)$, $\omega\not\in {\mathsf{ex}}\,{\mathcal{S}}(M)$.  Suppose $f\in {\mathsf{T}}_\nu(M)\backslash \{{\mathsf{0}}\}$ can be given by formula \eqref{substrat1e},  with $\varphi\in {\mathcal{S}}_{\pi,M}(\nu)$, $\psi\in {\mathcal{H}}_\pi\backslash \{{\mathsf{0}}\}$ obeying \eqref{substrat1c}-\eqref{substrat1d}. Then, there exists a parameterized curve $\gamma=(\mu_t)$ passing through $\nu=\mu_0$ and admitting a $\gamma$-compliant unital $^*$-representation around there such that in addition  $\gamma\subset \Omega$ holds and $f=\mu_t^{\,\prime}|_{t=0}$ is fulfilled.
\end{corolla}
\begin{proof}
Obviously, the result follows if in respect of $f$ and under the given assumptions the constructions of the previous proof starting from  \eqref{substrat23} on are applied.
\end{proof}
\subsubsection{More specific examples of submanifolds of states}\label{spezsubmann}
In the following example the subset ${\mathcal S}^{\mathsf{mix}}(M)$ of mixed states as defined in \eqref{ccomb0} will be referred to.
\begin{example}\label{mixedmann}
${\mathcal S}^{\mathsf{mix}}(M)$ is a submanifold of states if, and only if, $M\not=\mathbb{C}\,{\mathsf{1}}$.
\end{example}
\begin{proof}
Obviously, in view of \eqref{ccomb0}, $M\not=\mathbb{C}\,{\mathsf{1}}$ is necessary for ${\mathcal S}^{\mathsf{mix}}(M)\not=\emptyset$. Thus, in order to be a non-void set,  ${\mathsf{dim}} M >1$ must hold. Hence $\#\, {\mathsf{ex}}\,{\mathcal{S}}(M)> 1$ and therefore ${\mathcal S}(M)$ contains a proper (affinely convex) interval of states. Thus, a continuum of mixed states exists. Let $\nu\in {\mathcal S}^{\mathsf{mix}}(M)$. By Lemma \ref{extriv},  we have ${\mathsf{T}}_\nu(M)\not=\{\mathsf 0\}$. Let $f\not={\mathsf{0}}$, $f\in {\mathsf{T}}_\nu(M)$, and be $\{\pi,{\mathcal H}_\pi\}$ a unital $^*$-representation with ${\mathcal S}_{\pi,M}(\nu)\not=\emptyset$. According to Corollary \ref{tanfolg}, $\{\pi,{\mathcal H}_\pi\}$ is $\gamma$-compliant around $\nu$ for the parameterized curve $\gamma=(\nu_t)$ passing through $\nu=\nu_0$ and obeying $\nu_t^{\,\prime}|_{t=0}=f$  which is given by formula \eqref{tanfolg0b}. By Corollary \ref{auxpure1}, $\gamma|J\subset {\mathcal S}^{\mathsf{mix}}(M)$ will happen, with a suitable subinterval $J$ around $0$.
\end{proof}
Another set of interest is the set of the extremal states which are not characters:
\begin{subequations}\label{edef0}
\begin{equation}\label{edef2}
{\mathcal S}^{{\mathsf{ex}}}(M)=\bigl\{ \nu\in {\mathsf{ex}}\,{\mathcal S}(M):\ \exists x,y\in M, \nu(x y)\not=\nu(x)\nu(y)\bigr\}
\end{equation}
\begin{example}\label{puremann}
Let $M$ be a unital ${\mathsf{C}}^*$-algebra. Provided $M$ is noncommutative, then ${\mathcal S}^{{\mathsf{ex}}}(M)$
is a submanifold of states.
\end{example}
\begin{proof}
According to Lemma \ref{char1} and Corollary \ref{extan},  $M$ is noncommutative if, and only if, ${\mathcal S}^{{\mathsf{ex}}}(M)$ is non-void. Hence, under the assumption of the assertion ${\mathcal S}^{{\mathsf{ex}}}(M)\not=\emptyset$.
 By Corollary \ref{extan}, an equivalent definition of ${\mathcal S}^{{\mathsf{ex}}}(M)$ reads
\begin{equation}\label{edef}
{\mathcal S}^{{\mathsf{ex}}}(M)=\bigl\{ \nu\in {\mathsf{ex}}\,{\mathcal S}(M): {\mathsf{T}}_\nu(M)\not=\{{\mathsf{0}}\}\bigr\}
\end{equation}
Thus, to given $\nu\in {\mathcal S}^{{\mathsf{ex}}}(M)$, let us fix some $f\in {\mathsf{T}}_\nu(M)\backslash\{{\mathsf{0}}\}$. According to Corollary \ref{auxpure2}, when restricting  the parameterized curve given by \eqref{tanfolg0b}
to all $t$ of sufficiently small modulus, then in view of \eqref{edef} an example of a parameterized curve $(\nu_t)$ obeying  $(\nu_t)\subset  {\mathcal S}^{{\mathsf{ex}}}(M)$ and passing through $\nu$ with $\nu_t^{\,\prime}|_{t=0}=f$ and  admitting a $\gamma$-compliant unital $^*$-representation around there is obtained.
\end{proof}
\end{subequations}
In context of formula \eqref{condinv} a canonical lifting procedure for submanifolds of states defined over some subalgebra of a ${\mathsf{C}}^*$-algebra arises. In order to fix a relevant context,
let $N\subset M$ be a non-trivial  ${\mathsf{C}}^*$-subalgebra of a unital ${\mathsf{C}}^*$-algebra $M$, both with common unit, such that $N$ can be obtained as image of $M$ under a projection map $\Phi$ of norm $1$.
\begin{lemma}\label{mannmap}
Let $\Omega\subset {\mathcal S}(N)$ be a submanifold of states over $N$. Then  \begin{equation}\label{mannmap0}
\Omega\circ\Phi=\bigl\{\nu\in {\mathcal S}^\Phi(M): \nu|N\in \Omega\bigr\}
\end{equation}
is fulfilled, and $\Omega\circ\Phi$
is a submanifold of states over $M$.
\end{lemma}
\begin{proof}
Formula \eqref{mannmap0} is a consequence of \eqref{condinv}.
Let $\nu\in\Omega\circ\Phi$ be arbitrarily chosen. Then, $\nu|N\in \Omega$ and $\nu\circ\Phi=\nu$ hold. Since $\Omega$ is a submanifold of states over $N$, there exists a parameterized curve $\gamma=(\nu_t)\subset \Omega\subset {\mathcal S}(N) $ passing through $\nu|N=\nu_0$ and admitting a $\gamma$-compliant unital $^*$-representation of $N$ around there. By the previous and in view of Corollary \ref{condexpect0}\,\eqref{condexpect0a}, $\hat{\gamma}=\gamma\circ\Phi$ then is a  parameterized curve passing through $\nu$ and admitting a $\hat{\gamma}$-compliant unital $^*$-representation of $M$ around there. By construction $\hat{\gamma}$ is evolving completely within $\Omega\circ\Phi$. Therefore, the requirements of Definition \ref{submann} are satisfied.
\end{proof}
\begin{example}\label{imbed}
Let $\Phi: M\rightarrow N$ be a projection map of norm $1$ projecting onto a ${\mathsf{C}}^*$-subalgebra $N$ of $M$ with  ${\mathsf{dim}}\,N>2$. Then, the set   ${\mathcal S}^\Phi(M)$ of states over $M$ left invariant under the dual action of $\Phi$ is a submanifold of states over $M$.
\end{example}
\begin{proof}
Let $\Omega={\mathcal S}(N)$. Application of Theorem \ref{mannig} in respect of $N$ yields that $\Omega$ is a submanifold of states over $N$. On the other hand, by formula \eqref{mannmap0}, in the special case at hand  $\Omega\circ\Phi={\mathcal S}^\Phi(M)$ is seen.  By Lemma \ref{mannmap} the assertion follows.
\end{proof}
\subsubsection{Submanifolds of normal states}\label{finaltanspeznorm}
In case if $M$ is a ${\mathsf{W}}^*$-algebra the   normal state space ${\mathcal S}_0(M)$ will be of special interest. In addition, some subsets of normal states like the set of faithful normal states, which is defined by
\begin{equation}\label{fstate}
\begin{split}
{\mathcal S}_0^{\mathsf{faithful}}(M) &=\bigl\{\varrho\in {\mathcal S}_0(M): \varrho(x^*x)>0,\,\forall x\in M\backslash \{{\mathsf{0}}\}\bigr\}\\& =\bigl\{\varrho\in {\mathcal S}_0(M): s(\varrho)={\mathsf{1}}\bigr\}
\end{split}
\end{equation}
and the normal variants $${\mathcal S}_0^{\mathsf{mix}}(M)={\mathcal S}^{\mathsf{mix}}(M)\cap {\mathcal S}_0(M),\ {\mathcal S}^{{\mathsf{ex}}}_0(M)={\mathcal S}^{{\mathsf{ex}}}(M)\cap {\mathcal S}_0(M)$$  of the subsets ${\mathcal S}^{\mathsf{mix}}(M)$ and ${\mathcal S}^{{\mathsf{ex}}}(M)$ which were  defined in context of Example \ref{mixedmann} and Example \ref{puremann}, respectively, will be considered.
On a  ${\mathsf{W}}^*$-algebra faithful normal states needn't exist. But in case of existence,  the set of faithful normal states on $M$ is a uniformly dense, affinely convex subset  of ${\mathcal S}_0(M)$.
Start with an auxiliary result which is an essential ingredient for all what follows in the ${\mathsf{W}}^*$-case.
\begin{lemma}\label{wstar}
Assume $M$ is a ${\mathsf{W}}^*$-algebra, with predual space $M_*$, and normal state space ${\mathcal S}_0(M)\subset M_*$. Then  ${\mathsf{T}}_\nu(M)\subset M_*$, for each
$\nu\in {\mathcal S}_0(M)$. Moreover, in case of non-trivial ${\mathsf{T}}_\nu(M)$,  for each $f\in {\mathsf{T}}_\nu(M)\backslash \{{\mathsf{0}}\}$ the  parameterized curve  passing through $\nu$ and obeying $\nu_t^{\,\prime}|_{t=0}=f$ and specified in accordance with  formula \eqref{tanfolg0b} is obeying $\gamma\subset {\mathcal S}_0(M)$.
\end{lemma}
\begin{proof}
We suppose the context as in Corollary \ref{tanfolg}, with $f\in {\mathsf{T}}_\nu(M)$, but relating to  $\nu\in {\mathcal S}_0(M)$. In view of  \eqref{tanfolg0}, the Cauchy-Schwarz inequality for $x\in M$ implies
\[
| f(x)|\leq 2\,\|f\|_\nu \, \sqrt{\nu(x^*x)}
\]
By normality of $\nu$, for each decreasingly ordered directed net $\{p_\lambda\}$ of orthoprojections $p_\lambda\in M$ with greatest lower bound $\mathsf{g.l.b.}\{p_\lambda\}={\mathsf{0}}$ one has $\lim_\lambda \nu(p_\lambda)=0$. Thus, by the above estimate, $\lim_\lambda |f(p_\lambda)|=0$, for each such net. Hence $f\in M_*$.\\ Let $\omega$ be the positive linear form implemented by $\hat{\xi}_0$ via $\pi_0$. Owing to  $\hat{\xi}_0\in [\pi_0(M)_{\mathsf{h}}\varphi_0]$ there exists a sequence $\{x_n\}\subset M_{\mathsf{h}}$ such that $\hat{\xi}_0=\lim_{k\to\infty}\pi_0(x_k)\varphi_0$.
From this  $\omega=\|\cdot\|_1-\lim_{k\to\infty} \nu(x_k^*(\cdot)x_k)$
is obtained. Now, by normality of $\nu$, each of the positive linear forms $\nu(x_k^*(\cdot)x_k)$ has to be normal, too. Thus, since $M_*$ is a  $\|\cdot\|_1$-closed linear subspace of $M^*$ normality of $\omega$ follows. Hence, due to normality of $\nu$, $f$ and $\omega$, each of the parameterized curves  $\gamma=(\nu_t)$ constructed according to formula \eqref{tanfolg0b} yields a  parameterized curve  $\gamma$  passing through $\nu$ and exhibiting  $f\not={\mathsf{0}}$ as tangent form at $\nu$, and which  in addition is obeying  $\gamma\subset {\mathcal S}_0(M)$.
\end{proof}
\begin{theorem}\label{wstar1}
Let $M$ be a ${\mathsf W}^*$-algebra. ${\mathcal S}_0(M)$ is a submanifold of states if, and only if, $\#\, {\mathsf{ex}}\,{\mathcal S}(M)>2$ is fulfilled.
\end{theorem}
\begin{proof}
A unital ${\mathsf C}^*$-algebra $M$ violating the condition $\#\, {\mathsf{ex}}\,{\mathcal S}(M)>2$ is a commutative ${\mathsf W}^*$-algebra with ${\mathsf{dim}}\,M\leq 2$. Since any state there is normal, the assertion will fail at each $\nu\in {\mathsf{ex}}\,{\mathcal S}(M)\subset {\mathcal S}_0(M)$, by literally the same reasoning as given at the beginning of the proof of Theorem \ref{mannig}. Thus, the condition is necessary.

 Suppose $M$ is a ${\mathsf{W}}^*$-algebra obeying $\#\, {\mathsf{ex}}\,{\mathcal S}(M)>2$. By Lemma \ref{wstar}, through any normal state $\nu\in {\mathcal S}_0(M)$ with ${\mathsf{T}}_\nu(M)\not=\{{\mathsf{0}}\}$ a parameterized curve $\gamma\subset {\mathcal S}_0(M)$ is passing which is admitting a $\gamma$-compliant unital $^*$-representation around $\nu$. Thus,
 the outstanding cases are states  $\nu\in {\mathcal S}_0(M)$ with ${\mathsf{T}}_\nu(M)=\{{\mathsf{0}}\}$. Put that case at $\nu$. By Corollary \ref{extan}, $\nu$ is a normal character state.
 By Lemma \ref{char1/2}, the support orthoprojection  $z_0=s(\nu)$ is minimal in $M$,  with  $x z_0=z_0 x$, for all $x\in M$. Thus,
 \begin{equation}\label{minop}
 M={\mathbb{C}} z_0+Mz_0^\perp
 \end{equation} Clearly, by assumption there have to exist at least another two pure states on $M$. Necessarily, these pure states have to be vanishing on $z_0$. Hence, $Mz_0^\perp$ has to be a non-trivial  ${\mathsf{W}}^*$-algebra. Thus, we can be assured that over $M$ even another two normal states $\mu\not=\varrho$ exist, with support orthoprojections obeying $s(\mu)\leq z_0^\perp$ and $s(\varrho)\leq z_0^\perp$, respectively. Now, suppose $\{\pi,{\mathcal H}_\pi\}$ to be a unital $^*$-representation of $M$ such that implementing vectors $\varphi\in {\mathcal S}_{\pi,M}(\varrho)$, $\eta \in {\mathcal S}_{\pi,M}(\mu)$ and $\varphi\in {\mathcal S}_{\pi,M}(\nu)$ exist. Let $N=\pi(M)^{\,\prime\prime}$, and be $\nu_\pi$ the state implemented by $\varphi$ on $N$. Then, since $\nu$ is a character state on $M$,  $\nu_\pi$ is a normal character state over $N$. By Lemma \ref{char1/2}, $s(\nu_\pi)$ is a minimal orthoprojection of $N$. In view of \eqref{minop} and the previous,  $z=\pi(z_0)$ is an orthoprojection obeying $z\in N\cap N^{\,\prime}$,   ${\mathsf{0}}<z<{\mathsf{1}}$,  and   $\pi(M)={\mathbb{C}}z+\pi(M)z^\perp$. Hence, $N={\mathbb{C}}z+N z^\perp$ holds. That is, $z$ is a minimal orthoprojection of $N$.  Note that owing to $\nu_\pi(z)=\nu(z_0)=1$ we have $z\geq s(\nu_\pi)$. By minimality of $z$ and $s(\nu_\pi)$ from this  $z=s(\nu_\pi)$ follows. Thus we have $z\varphi=\varphi$, and from $0=\varrho(z_0)$ then $0=\varrho(z_0)=\|\pi(z_0)\xi\|_\pi^2=\|z\xi\|_\pi^2$ is obtained, that is, $z\xi={\mathsf{0}}$ is fulfilled. Analogously, we see that $z\eta={\mathsf{0}}$. Thus, we have arrived at the condition  \eqref{mannig2cn} again. Therefore the construction procedure for implementing vectors given by  \eqref{mannig2an} can be applied in the situation at hand as well.  Since the states $\nu$, $\varrho$ and $\mu$ fed into formula  \eqref{mannig2bn} are normal states on $M$, the states $\nu_t$ implemented all are normal states, and thus a parameterized curve $\gamma$ passing through $\nu$ is generated, but which now in addition is obeying $\gamma\subset {\mathcal S}_0(M)$,  with  $\gamma$-compliant representation $\{\pi,{\mathcal H}_\pi\}$ around $\nu$.
 \end{proof}
 \newpage
 \begin{corolla}\label{wstar2}
 Let $M\not={\mathbb{C}}\,{\mathsf{1}}$ be a ${\mathsf{W}}^*$-algebra.
 The following subsets of normal states are submanifolds:
 \begin{enumerate}
 \item \label{wstar2a}
 ${\mathcal S}_0^{\mathsf{mix}}(M)$;
 \item \label{wstar2b}
   ${\mathcal S}_0^{\mathsf{ex }}(M)$, provided this set is non-void \textup{(}then $M$ is non-commutative\textup{)};
 \item \label{wstar2c}
 ${\mathcal S}^{\mathsf{faithful}}_0(M)$, provided faithful normal states exist;
 \item \label{wstar2d}
 $\Omega_M(\varrho)$ if, and only if, $\varrho\in {\mathcal S}_0^{\mathsf{mix}}(M)$.
 \end{enumerate}
 \end{corolla}
 \begin{proof}
For $\nu\in {\mathcal S}_0^{\mathsf{mix}}(M)$ and   $f\in {\mathsf{T}}_\nu(M)\backslash \{{\mathsf{0}}\}$ by  Example \ref{mixedmann} we have $\gamma|J\subset {\mathcal S}^{\mathsf{mix}}(M)$, with $\gamma$ being the special  parameterized curve \eqref{tanfolg0b} and some parameter interval $J$ around $0$. But since according to Lemma \ref{wstar} for $\nu\in {\mathcal S}_0(M)$ we even have $\gamma\subset {\mathcal S}_0(M)$, in view of the previous $\gamma|J\subset {\mathcal S}_0^{\mathsf{mix}}(M)$ follows.  This is \eqref{wstar2a}.

Remark that, in view of the facts  mentioned at the beginning of the proof of Example \ref{puremann}, ${\mathcal S}_0^{\mathsf{ex }}(M)\not=\emptyset$ will imply  $M$ to be non-commutative. Also, if $\nu\in {\mathcal S}_0^{\mathsf{ex }}(M)$, then ${\mathsf{T}}_\nu(M)\not=\{{\mathsf{0}}\}$, by \eqref{edef}. Thus,  \eqref{wstar2b} follows from Lemma \ref{wstar}.

To see \eqref{wstar2c}, suppose $\nu$ to be a faithful normal state on $M$, that is,   $\nu$ is a normal state with support orthoprojection $s(\nu)={\mathsf{1}}$. Note that  $\nu$ cannot be a character state, for otherwise from Lemma \ref{char1/2}\,\eqref{char1/2b} we had to conclude to $x=\nu(x){\mathsf{1}}$, for each $x\in M$. This would contradict the assumption  $M\not={\mathbb{C}}\,{\mathsf{1}}$. Thus, by Corollary \ref{extan}, ${\mathsf{T}}_\nu(M)\not=\{{\mathsf{0}}\}$ follows. Let us fix a unital $^*$-representation $\{\pi,{\mathcal H}_\pi\}$ with non-trivial $\pi$-fibre of $\nu$.   Let $f\in {\mathsf{T}}_\nu(M)\backslash \{{\mathsf{0}}\}$, arbitrarily chosen, and as mentioned in Lemma \ref{wstar}, be $\gamma$ the special parameterized curve $\gamma=(\nu_t)\subset {\mathcal S}_0(M)$ passing through $\nu$ and obeying $\nu_t^{\,\prime}|_{t=0}=f$. Then,  $\{\pi,{\mathcal H}_\pi\}$ is $\gamma$-compliant, i.e.~ there is a differentiable map $t\mapsto\varphi_t\in {\mathcal S}_{\pi,M}(\nu_t)$.
Application of the construction schema \eqref{substrat23} mentioned on in the proof of Theorem \ref{substrat} with input data $\nu,\nu_t$ and $\varphi,\varphi_t$  now yields normal states $\mu_t$ and unit vectors $\xi_t$.  Thus, by $(\mu_t)$ now a parameterized curve of normal states passing through $\nu$ and exhibiting $f$ as tangent form at $\nu$ is defined. Also, since $\nu$ is faithful, we even have $(\mu_t)\subset {\mathcal S}^{\mathsf{faithful}}_0(M)$. Moreover, in respect of the direct sum representation $\{\pi,{\mathcal H}_\pi\}\oplus \{\pi,{\mathcal H}_\pi\}$, by $t\mapsto \xi_t$  a continuously differentiable implementation of $(\mu_t)$ around $\nu=\mu_0$ is defined.  Hence, there exists a parameterized curve  $(\mu_t)$ of faithful normal states passing through $\nu=\mu_0$ and obeying  $\mu_t^{\,\prime}|_{t=0}=f$ and which is admitting a $\gamma$-compliant unital $^*$-representation around $\nu$. This especially means that the assertion of \eqref{wstar2c} holds true.

Relating \eqref{wstar2d}, note that according to Lemma \ref{setprop}\,\eqref{setprop4} the inclusion $\Omega_M(\varrho)\subset {\mathcal{S}}_0(M)$ takes place iff $\varrho\in {\mathcal{S}}_0(M)$ is fulfilled. Hence, the assertion then appears as an immediate consequence  of Theorem \ref{substrat}.
\end{proof}
A situation as follows often will be met: for two ${\mathsf{W}}^*$-algebras $M$ and $N$ obeying ${\mathsf{1}}\in N\subset M$ we suppose there exists a normal projection map  $\Phi$ of norm $1$ obeying $\Phi(M)=N$.  Due to normality, in addition to \eqref{condinv}, let us consider now the set
\begin{equation}\label{condinv0}
{\mathcal S}_0^\Phi(M)=\bigl\{\nu\in {\mathcal S}_0(M): \nu\circ \Phi=\nu\bigr\}={\mathcal S}_0(N)\circ\Phi
\end{equation}
and which then is a subset of normal states.
In addition to Lemma \ref{mannmap}, a  ${\mathsf{W}}^*$-variant referring to normal states and to \eqref{condinv0} exists and reads as following.
\begin{lemma}\label{mannmapnor}
Let $\Omega\subset {\mathcal S}_0(N)$ be a submanifold of normal states over $N$. Then  \begin{equation}\label{mannmapnor0}
\Omega\circ\Phi=\bigl\{\nu\in {\mathcal S}_0^\Phi(M): \nu|N\in \Omega\bigr\}
\end{equation}
is fulfilled, and $\Omega\circ\Phi$
is a submanifold of normal states over $M$.
\end{lemma}
The proof
gets obvious if the following remark is taken into account.
\begin{remark}\label{condexpectrem}
Lemma \ref{wstar} implies the assertions of
Lemma \ref{condexpect} and Corollary \ref{condexpect0} to remain true in respect of normal states when in the suppositions ${\mathcal S}^\Phi(M)$ is replaced by ${\mathcal S}_0^\Phi(M)$, and
${\mathcal S}^\Phi(M)$, ${\mathcal S}(N)$ are replaced by ${\mathcal S}_0^\Phi(M)$, ${\mathcal S}_0(N)$, respectively.
\end{remark}
In the above situation, if the normal projection map $\Phi$ in addition is faithful, this is of greatest  significance. Here, `faithful' refers to the property that, for $x\in M$,
  \begin{equation}\label{fproj}
  \Phi(x^*x)={\mathsf{0}}\,\Longleftrightarrow \,x={\mathsf{0}}
  \end{equation}
holds. Subsequently, refer to such $\Phi$
  as `conditional expectation' from $M$ onto $N$.
\begin{corolla}\label{imbed0}
Let $M$, $N$  be  ${\mathsf{W}}^*$-algebras, with $M\supset N\ni {\mathsf{1}}$, together with a conditional expectation $\Phi$ acting from $M$ onto $N$.
Suppose $M$ admits faithful normal states, and $N$ is non-trivial, that is   $N\not=\mathbb{C}\,{\mathsf{1}}$. Then, the subset
\begin{equation}\label{imbed00}
{\mathcal{S}}_0^{\Phi(M)}\cap {\mathcal S}_0^{\mathsf{faithful}}(M)=\bigl\{\nu\in {\mathcal S}_0^{\mathsf{faithful}}(M): \nu\circ\Phi=\nu\bigr\}
\end{equation}
is a submanifold of faithful normal states of $M$.
\end{corolla}
\begin{proof}
Since ${\mathcal S}_0^{\mathsf{faithful}}(M)\not=\emptyset$ holds and $N$ is non-trivial, we have $$\emptyset\not={\mathcal S}_0^{\mathsf{faithful}}(M)|N\subset {\mathcal S}_0^{\mathsf{faithful}}(N)$$
Also, from this and $\Phi\circ\Phi=\Phi$ one concludes that
\[
{\mathcal{S}}_0^\Phi(M)\cap {\mathcal S}_0^{\mathsf{faithful}}(M)\subset \big\{\nu\in {\mathcal{S}}_0(M): \nu\circ\Phi=\nu,\, \nu|N\in
{\mathcal S}_0^{\mathsf{faithful}}(N)\bigr\}={\mathcal S}_0^{\mathsf{faithful}}(N)\circ\Phi
\]
Thus, by Corollary \ref{wstar2}\,\eqref{wstar2c} when applied with respect to $N$,  $\Omega={\mathcal S}_0^{\mathsf{faithful}}(N)$ is a submanifold of normal states of $N$. Hence, application of Lemma \ref{mannmapnor} yields that
\[
{\mathcal S}_0^{\mathsf{faithful}}(N)\circ\Phi=\big\{\nu\in {\mathcal{S}}_0(M): \nu\circ\Phi=\nu,\, \nu|N\in
{\mathcal S}_0^{\mathsf{faithful}}(N)\bigr\}
\]
is a submanifold of normal states of $M$. Note that, for each $\nu\in {\mathcal S}_0^{\mathsf{faithful}}(N)\circ\Phi$, for  $x\in M$ with $\nu(x^*x)=0$ we infer $\nu(\Phi(x^*x))=0$. From this and owing to  $\Phi(x^*x)\in N_+$ and $\nu|N\in
{\mathcal S}_0^{\mathsf{faithful}}(N)$ we infer that  $\Phi(x^*x)={\mathsf{0}}$. By faithfulness, $x={\mathsf{0}}$ follows. Thus,
\[
{\mathcal S}_0^{\mathsf{faithful}}(N)\circ\Phi\subset {\mathcal S}_0^{\mathsf{faithful}}(M)
\]
has to be fulfilled. On the other hand, from this we have that for $\nu\in {\mathcal S}_0(M)$ with $\nu\circ\Phi=\nu$ and $\nu|N\in
{\mathcal S}_0^{\mathsf{faithful}}(N)$ also $\nu=\nu\circ\Phi\in {\mathcal S}_0^{\mathsf{faithful}}(M)$ is fulfilled.  That is,
\[
{\mathcal S}_0^{\mathsf{faithful}}(N)\circ\Phi\subset {\mathcal{S}}_0^\Phi(M)\cap {\mathcal S}_0^{\mathsf{faithful}}(M)
\]
Since in view of the above also the reverse inclusion holds, equality has to occur
\[
{\mathcal S}_0^{\mathsf{faithful}}(N)\circ\Phi= {\mathcal{S}}_0^\Phi(M)\cap {\mathcal S}_0^{\mathsf{faithful}}(M)
\]
between the subset of states defined in  \eqref{imbed00} and the subset ${\mathcal S}_0^{\mathsf{faithful}}(N)\circ\Phi$  which has been seen to be a submanifold of states.
\end{proof}
As an application,
let us consider a finite or countably infinite, strictly ascendingly directed system  ${\mathcal P}=\{P_k\}$ of orthoprojections $P_k$ over a Hilbert space ${\mathcal H}$, with ${\mathsf{dim}}\,P_k {\mathcal H}<\infty$ for all $k\in {\mathbb{N}}$, and obeying ${\mathsf{l.u.b.}} \{P_k\}={\mathsf 1}$.  We then are going to consider the subset
\[
\Omega_{{\mathcal P}}({\mathcal H})\subset {\mathcal S}_0^{\mathsf{faithful}}({\mathsf B}({\mathcal H}))
\]
dealt with in Example \ref{Bopex}.
On $M={\mathsf{B}}({\mathcal H})$, a linear mapping $\Phi$ be defined by
\begin{subequations}\label{imb}
\begin{equation}\label{imbed00ex}
\Phi(x)=\sum_k \Delta_k P\, x\, \Delta_k P,\,\forall\,x\in M,
\end{equation}
with $\Delta_1 P=P_1$ and $\Delta_k P=P_{k}-P_{k-1}$, for $k>1$ in case $\#\, \mathcal{P}>1$. $\Phi$ is a conditional expectation with fixed point $vN$-subalgebra
$
N_{\mathcal P}={\sum_k}^\oplus {\mathsf{B}}(\Delta_k P \,{\mathcal H})
$, and thus is satisfying
\begin{equation}\label{imbed1ex}
\Omega_{{\mathcal P}}({\mathcal H})={\mathcal S}_0^{\mathsf{faithful}}(N_{\mathcal{P}})\circ\Phi={\mathcal S}_0^{\mathsf{faithful}}(M)\circ\Phi
\end{equation}
\end{subequations}
in respect of $\Omega_{{\mathcal P}}({\mathcal H})$. Therefore, since in case of ${\mathsf{dim}}\, {\mathcal{H}}>1$ and  for each ${\mathcal P}$ the premises of Corollary \ref{imbed0} in respect of $M={\mathsf{B}}({\mathcal H})$ and  $N=N_{\mathcal P}$ are satisfied, in view of the previous and  Example \ref{Bopex}, the application of  Corollary \ref{imbed0}  yields the following result.
\begin{example}\label{imbed0ex}
In respect of the Bures metric the affinely convex subset $\Omega_{{\mathcal P}}({\mathcal H})$ is a geodesically convex submanifold of states over ${\mathsf{B}}({\mathcal H})$, with ${\mathsf{dim}}\, {\mathcal{H}}>1$.
\end{example}
\subsubsection{Local tangent spaces  relative to a  submanifold of states}\label{finaltanspez1}
 From now on suppose ${\mathsf{dim}}\,M > 2$ for the ${\mathsf{C}}^*$-algebra $M$ under consideration. Then, according to Theorem \ref{mannig}/Theorem \ref{wstar1}, the full state space, respectively the normal state space if $M$ is a ${\mathsf{W}}^*$-algebra, are submanifolds of states in the sense of Definition \ref{submann}.
\begin{definition}\label{loctan}
Let $\Omega$ be a submanifold of states, $\nu\in \Omega$. The `local tangent space at $\nu$ (relative to $\Omega$)' ${\mathsf{T}}\Omega_\nu$ is defined to contain the ${\mathsf{0}}$-form and each bounded linear form over $M$ arising as a tangent form at $\nu$ of a  parameterized curve $\gamma\subset \Omega$  passing through $\nu$  and admitting a $\gamma$-compliant unital $^*$-representation around there.
\end{definition}
\begin{lemma}\label{fulltsp}
Let $M$ be a  ${\mathsf{C}}^*$-algebra as above. The following hold:
\begin{enumerate}
\item\label{fulltsp1}
${{\mathsf{T}}{\mathcal S}(M)_\nu}\phantom{_0}={\mathsf{T}}_\nu(M)$, for all $\nu\in  {\mathcal S}(M)$;
\item\label{fulltsp2}
${\mathsf{T}}{\mathcal S}_0(M)_\nu\,={\mathsf{T}}_\nu(M)$, for all $\nu\in  {\mathcal S}_0(M)$, if $M$ is a ${\mathsf{W}}^*$-algebra.
\end{enumerate}
\end{lemma}
\begin{proof}
By Theorem \ref{mannig}, ${\mathcal S}(M)_\nu$
is a submanifold of states. Also, as mentioned in the proof, for $\nu\in  {\mathcal S}(M)$ satisfying ${\mathsf{T}}_\nu(M)\not=\{{\mathsf{0}}\}$,
to each $f\in {\mathsf{T}}_\nu(M)\backslash \{{\mathsf{0}}\}$ a parameterized curve $\gamma$ passing through $\nu$ and  exhibiting $f$ as tangent form at $\nu$ exists, such that each unital $^*$-representation $\{\pi,{\mathcal H}_\pi\}$ with non-trivial $\pi$-fibre of $\nu$ is $\gamma$-compliant. This in view of Definition \ref{loctan} and Theorem \ref{Cstern}\,\eqref{Cstern1}\,$\Leftrightarrow$\eqref{Cstern3} implies
$
{\mathsf{T}}{\mathcal S}(M)_\nu\backslash \{{\mathsf{0}}\}={\mathsf{T}}_\nu(M)\backslash \{{\mathsf{0}}\}
$. Therefore, since ${\mathsf{T}}_\nu(M)$ is a linear space and ${\mathsf{0}}$ by definition belongs to ${\mathsf{T}}{\mathcal S}(M)_\nu$, we have
$
{\mathsf{T}}{\mathcal S}(M)_\nu={\mathsf{T}}_\nu(M)
$ for all $\nu$ obeying ${\mathsf{T}}_\nu(M)\not=\{{\mathsf{0}}\}$. Also, by Corollary \ref{extan}, if a state $\nu$ with ${\mathsf{T}}_\nu(M)=\{{\mathsf{0}}\}$ exists,  then  ${\mathsf{T}}{\mathcal S}(M)_\nu\supset{\mathsf{T}}_\nu(M)=\{{\mathsf{0}}\}$ has to be fulfilled. However, in assuming existence of $f\not={\mathsf{0}}$  such that  $f\in {\mathsf{T}}{\mathcal S}(M)_\nu$ by Theorem \ref{Cstern} would have to imply that $f\in {\mathsf{T}}_\nu(M)$, a contradiction. Hence,
$
{\mathsf{T}}{\mathcal S}(M)_\nu={\mathsf{T}}_\nu(M)
$
for those cases, too. Thus \eqref{fulltsp1} holds.  Moreover, if Lemma \ref{wstar} and Theorem \ref{wstar1} are taken into account, an analogous line of reasoning (we omit the details) will show  that \eqref{fulltsp2} holds.
\end{proof}
\begin{remark}\label{fulltansp1}
According to Corollary \ref{compT}, when equipped with the $\|\cdot\|_\nu$-norm each ${\mathsf{T}}{\mathcal S}(M)_\nu$ (resp.\,${\mathsf{T}}{\mathcal S}_0(M)_\nu$) is a (real) Banach space.
\end{remark}
Further examples of local tangent spaces to some elementary but important special cases of submanifolds $\Omega\subset {\mathcal S}(M)$ are listed in the following:
\begin{example}\label{furthex}
Let $M$ be a unital ${\mathsf{C}}^*$-algebra, $\{\pi,{\mathcal H}_\pi\}$  a unital $^*$-representation of $M$. Suppose $\Omega$ can be  each of ${\mathcal S}(M)$, ${\mathcal S}^\pi(M)$,
${\mathcal S}^{\mathsf{mix}}(M)$, and
${\mathcal S}^{\mathsf{ex}}(M)$, provided the latter are submanifolds of states, by respective conditions. Then,
\begin{equation}\label{furthex1}
{\mathsf{T}}\Omega_\nu= {\mathsf{T}}_\nu(M),\ \forall \nu\in \Omega,
\end{equation}
is fulfilled.
If $M$ is a ${\mathsf{W}}^*$-algebra,  in addition to the above \eqref{furthex1} remains true also if  $\Omega$ can stand for each of ${\mathcal S}_0(M)$,
${\mathcal S}_0^{\mathsf{mix}}(M)$,
${\mathcal S}_0^{\mathsf{ex }}(M)$, and  ${\mathcal S}_0^{\mathsf{faithful}}(M)$, provided $M$ by respective conditions is such that the latter are  submanifolds of normal states.
\end{example}
\begin{proof}
For ${\mathcal S}(M)$, ${\mathcal S}_0(M)$ with ${\mathsf{dim}}\,M> 2$ the result follows from Lemma \ref{fulltsp}. For the other cases and conditions, under which \eqref{furthex1} holds, see Lemma \ref{submannele}, Example \ref{mixedmann}, Example \ref{puremann} and the respective proofs, and Corollary \ref{wstar2} in the  ${\mathsf{W}}^*$-case, respectively. In each case, the line of reasoning goes along the argument by means of Theorem \ref{Cstern} as raised in the proof of Lemma \ref{fulltsp}.
\end{proof}
Let ${\mathsf{1}}\in N\subset M$, with ${\mathsf{C}}^*$-algebras $N$ and $M$ with common unit ${\mathsf{1}}$, and with  $N=\Phi(M)$, for a projection map $\Phi$ of norm one. Then the following holds.
\begin{lemma}\label{proex}
Let $\Omega\subset {\mathcal{S}}(N)$ be a submanifold of states over $N$. Then,
\begin{equation}\label{proex1}
{\mathsf{T}}(\Omega\circ\Phi)_\nu={\mathsf{T}}\Omega_{\nu|N}\circ\Phi
\end{equation}
is fulfilled, for each $\nu\in \Omega\circ\Phi$.
\end{lemma}
\begin{proof}
By assumption about $\Omega$ and Lemma \ref{mannmap}, $\Omega\circ\Phi\subset  {\mathcal{S}}(M)$ is a submanifold of states on $M$. Also note that the following is fulfilled
\begin{equation}\label{proex2}
 \Omega\circ\Phi=\bigl\{\varrho\in {\mathcal{S}}^\Phi(M): \varrho|N\in \Omega\bigr\}
\end{equation}
 with ${\mathcal{S}}^\Phi(M)$ defined in \eqref{condinv}. Suppose $\nu\in \Omega\circ\Phi$. Then, by Corollary \ref{condexpect0}\,\eqref{condexpect0a} and in view of Definition \ref{loctan},  we may  conclude that
 \[
 {\mathsf{T}}\Omega_{\nu|N}\circ\Phi\subset  {\mathsf{T}}(\Omega\circ\Phi)_\nu
 \]
 On the other hand, in view of Example \ref{furthex},  $\Omega\circ\Phi\subset  {\mathcal{S}}(M)$ is implying  ${\mathsf{T}}(\Omega\circ\Phi)_\nu\subset {\mathsf{T}}{\mathcal{S}}(M)_\nu={\mathsf{T}}_\nu(M)$. Thus, for $f\in {\mathsf{T}}(\Omega\circ\Phi)_\nu$ we either have $f={\mathsf{0}}$, in which case $f$  belongs to ${\mathsf{T}}\Omega_{\nu|N}\circ\Phi$ by triviallity, or we have $f\not={\mathsf{0}}$  and then ${\mathsf{T}}_\nu(M)\not=\{{\mathsf{0}}\}$ is fulfilled. Also, since $\Omega\circ\Phi$ is a submanifold of states with $\nu\in \Omega\circ\Phi$, there exists a parameterized curve $\hat{\gamma}=(\varrho_t)\subset \Omega\circ\Phi\subset {\mathcal{S}}^\Phi(M)$ passing through $\nu=\varrho_0$, with a $\hat{\gamma}$-compliant unital $^*$-representation around $\nu$, and which in addition is obeying $\varrho_t^{\,\prime}|_{t=0}=f$. Thus, Corollary \ref{condexpect0}\,\eqref{condexpect0b} applies and yields that $f\circ\Phi=f$. Note that from this and $f\not={\mathsf{0}}$, together with $\Phi(x)=x$ for $x\in N$, $f|N\not={\mathsf{0}}$ can be inferred to hold. On the other hand, in view of \eqref{proex2} we have $\varrho_t|N\subset \Omega$, for each $t$. From this together with the previous  $(\varrho_t|N)^{\,\prime}|_{t=0}=f|N\not={\mathsf{0}}$ follows. Hence, the parameterized curve $\gamma=\hat{\gamma}|N\subset \Omega$ is passing through $\nu|N=\varrho_0|N$, and admits as $\gamma$-compliant $^*$-representation the restriction onto $N$ of the $\hat{\gamma}$-compliant $^*$-representation, and is exhibiting $f|N$ as tangentform at $\nu|N$. Accordingly, $f|N\in {\mathsf{T}}\Omega_{\nu|N}$ has to be fulfilled. From this we infer $f=f\circ\Phi\in {\mathsf{T}}\Omega_{\nu|N}\circ\Phi$.
 In summarizing, the following inclusion relation is fulfilled:
 \[
  {\mathsf{T}}(\Omega\circ\Phi)_\nu\subset {\mathsf{T}}\Omega_{\nu|N}\circ\Phi
  \]
As previously mentioned, since also the reverse inclusion holds, \eqref{proex1} is seen.
\end{proof}
\begin{remark}\label{proexrem}
In case of ${\mathsf{W}}^*$-algebras with $N\subset M$ and common unit the assertion of Lemma \ref{proex} remains true if $\Omega\subset {\mathcal{S}}_0(N)$ is a submanifold of normal states over $N$, and if  $\Phi$ is a normal projection map of norm one obeying $N=\Phi(M)$. But in order to see this, instead of Lemma \ref{mannmap}, \eqref{condinv} and formula \eqref{mannmap0}  in the proof one merely has to make reference to Lemma \ref{mannmapnor}, \eqref{condinv0} and formula \eqref{mannmapnor0}, respectively.
\end{remark}
Further non-trivial cases of local tangent spaces to certain submanifolds of states
can be constructed along Lemma \ref{proex}/Remark \ref{proexrem}. Here comes a first example.
\begin{example}\label{proexex}
Let $M$ be a ${\mathsf{W}}^*$-algebra with faithful normal states, and be  $\Phi$ a conditional expectation projecting onto a ${\mathsf{W}}^*$-subalgebra $N$, with $N\not={\mathbb{C}}\,{\mathsf{1}}$. Then,
\begin{equation}\label{proexex1}
{\mathsf{T}}\bigl\{\varrho\in {\mathcal S}_0^{\mathsf{faithful}}(M): \varrho\circ\Phi=\varrho\bigr\}_\nu=\bigl\{f\in {\mathsf{T}}_\nu(M): f\circ\Phi=f\bigr\}
\end{equation}
for each $\nu \in {\mathcal S}_0^{\mathsf{faithful}}(M)$ obeying $\nu\circ\Phi=\nu$.
\end{example}
\begin{proof}
By Corollary \ref{wstar2} and Corollary  \ref{imbed0} both subsets ${\mathcal S}_0^{\mathsf{faithful}}(N)$ and  $$\bigl\{\varrho\in {\mathcal S}_0^{\mathsf{faithful}}(M): \varrho\circ\Phi=\varrho\bigr\}={\mathcal S}_0^{\mathsf{faithful}}(N)\circ\Phi$$ are submanifolds of faithful normal states over $N$ and $M$, respectively. Thus, application of Lemma \ref{proex}/Remark \ref{proexrem} in respect of $\Omega={\mathcal S}_0^{\mathsf{faithful}}(N)$ and with  $\nu \in {\mathcal S}_0^{\mathsf{faithful}}(M)$ obeying $\nu\circ\Phi=\nu$ implies that
\begin{equation*}
{\mathsf{T}}\bigl({\mathcal S}_0^{\mathsf{faithful}}(N)\circ\Phi\bigr)_\nu=
{\mathsf{T}}{\mathcal S}_0^{\mathsf{faithful}}(N)_{\nu|N}\circ\Phi
\end{equation*}
Since ${\mathsf{T}}{\mathcal S}_0^{\mathsf{faithful}}(N)_{\nu|N}= {\mathsf{T}}_{\nu|N}(N)$ by Example \ref{furthex}, by Lemma \ref{condexpect}
\[
{\mathsf{T}}\bigl({\mathcal S}_0^{\mathsf{faithful}}(N)\circ\Phi\bigr)_\nu=
{\mathsf{T}}_\nu(M)\circ\Phi=\bigl\{f\in {\mathsf{T}}_\nu(M): f\circ\Phi=f\bigr\}
\]
follows. This in view of the above yields \eqref{proexex1}.
\end{proof}
For ${\mathsf{dim}}\,{\mathcal{H}}>1$, by means of  formula \eqref{proexex1} the local tangent space at a state $\nu$ of the submanifold $\Omega_{{\mathcal P}}({\mathcal H})$, which has been identified in context of Example \ref{imbed0ex},  can be calculated easily. In fact, if $\Phi$ is the conditional expectation of \eqref{imbed00ex} based on  $\mathcal{P}$, in view of \eqref{imbed1ex}   formula \eqref{proexex1} can be applied immediately.
\begin{example}\label{proexex2}
For each $\nu \in {\mathcal S}_0^{\mathsf{faithful}}({\mathsf{B}}({\mathcal{H}}))$ obeying $\nu\circ\Phi=\nu$ one has
\[
{\mathsf{T}} \Omega_{{\mathcal P}}({\mathcal H})_\nu=\bigl\{f\in {\mathsf{T}}_\nu({\mathsf{B}}({\mathcal{H}})) : f\circ\Phi=f\bigr\}
\]
with the conditional expectation $\Phi$ given by \eqref{imbed00ex}.
\end{example}
In order to consider ${\mathsf{T}}\Omega_\nu$ if $\Omega$ is a stratum, $\Omega=\Omega_M(\omega)$, by Theorem  \ref{substrat}, $M\not=\mathbb{C}\,{\mathsf{1}}$  and  $\omega\in {\mathcal{S}}^{\mathsf{mix}}(M)$ must be supposed.
Let $\nu\in \Omega$ with non-trivial $\pi$-fibre relative to a unital $^*$-representation $\{\pi,{\mathcal H}_\pi\}$, and be $\varphi\in {\mathcal{S}}_{\pi,M}(\nu)$. By
Theorem \ref{Cstern}\,\eqref{Cstern2},  to $f\in {\mathsf{T}}_\nu(M)$ there is a unique $\hat{\psi}_0\in [\pi(M)_{\mathsf{h}}\varphi]$ obeying $f(\cdot)=\langle\pi(\cdot)\hat{\psi}_0,\varphi\rangle_\pi+\langle\pi(\cdot)\varphi,\hat{\psi}_0\rangle_\pi$ and  $\|f\|_\nu=\|\hat{\psi}_0\|_\pi$.
In terms of these data the following holds:
\begin{example}\label{tanstrat}
$f\in {\mathsf{T}}\Omega_\nu$ is fulfilled if, and only if, there exists $\{\pi,{\mathcal H}_\pi\}$ with non-trivial $\pi$-fibre of $\nu$ such that $\hat{\psi}_0\in p_\pi(\varphi){\mathcal{H}}_\pi$ holds, in respect of $\varphi\in {\mathcal{S}}_{\pi,M}(\nu)$.
\end{example}
\begin{proof}
For $f\not={\mathsf{0}}$, $\hat{\psi}_0\not={\mathsf{0}}$. Thus, if  $\hat{\psi}_0\in p_\pi(\varphi){\mathcal{H}}_\pi$ is fulfilled, then $\psi=\hat{\psi}_0$ meets all the requirements of \eqref{substrat12}. Hence $f\in {\mathsf{T}}\Omega_\nu$, by Corollary \ref{extsubstrat} and Definition \ref{submann}.

Suppose $f\in {\mathsf{T}}\Omega_\nu$, with $f\not={\mathsf{0}}$. Then, according to Definition \ref{submann} there has to exist   a parameterized curve $\gamma=(\nu_t)\subset \Omega$ passing through $\nu=\nu_0$ and   satisfying  $f=\nu_t^{\,\prime}|_{t=0}$, and admitting a $\gamma$-compliant $^*$-representation $\{\pi,{\mathcal{H}}_\pi\}$ around $\nu$. In line with this, let $t\mapsto \varphi_t\in {\mathcal{S}}_{\pi,M}(\nu_t)$, with $\varphi=\varphi_0$, be a differentiable implementation of $\gamma$ around $\nu$ in respect of $\{\pi,{\mathcal{H}}_\pi\}$. Since  for all $t$ we have $\nu_t\in \Omega$ and $s((\nu_t)_{\pi})=p_{\pi}(\varphi_t)$, in view of \eqref{stratum0} and Theorem \ref{suppabsstet} we get $p_{\pi}(\varphi_t)=p_{\pi}(\varphi)$, for all $t$. Hence, $f(\cdot)=\langle\pi(\cdot)\psi,\varphi\rangle_\pi+\langle\pi(\cdot)\varphi,\psi\rangle_\pi$ holds, with  $\psi=\varphi_t^{\,\prime}|_{t=0}\in p_{\pi}(\varphi){\mathcal{H}}_\pi$. Let $\hat{\psi}_0\in [\pi(M)_{\mathsf{h}}\varphi]$ be the unique element in  $[\pi(M)_{\mathsf{h}}\varphi]=[N_{\mathsf{h}}\varphi]$, with $N=\pi(M)^{\,\prime\prime}$,  generating the same $f$ when taking $\hat{\psi}_0$ instead of $\psi$ in sense of Theorem \ref{Cstern}\,\eqref{Cstern2} in the formula of $f$. Note that for each $\xi\in  [N_{\mathsf{h}}\varphi]$ and owing to $p_\pi(\varphi)\in N$ we have that  $p_\pi(\varphi)\xi$ belongs to $[N_{\mathsf{h}}\varphi]$, too. Especially, therefore we have that  $\hat{\psi}_0, p_\pi(\varphi)\hat{\psi}_0\in [N_{\mathsf{h}}\varphi]$ holds. Note that owing to $p_\pi(\varphi) \psi=\psi$ the following is fulfilled:
\begin{eqnarray*}
\|\psi-\hat{\psi}_0\|^2_\pi &= &\|p_\pi(\varphi)(\psi-\hat{\psi}_0)\|^2_\pi+\|p_\pi(\varphi)^\perp(\psi-\hat{\psi}_0)\|^2_\pi\\
&=&\|\psi-p_\pi(\varphi)\hat{\psi}_0\|^2_\pi+\|p_\pi(\varphi)^\perp \hat{\psi}_0\|^2_\pi
\end{eqnarray*}
Hence, $\|\psi-\hat{\psi}_0\|_\pi\geq \|\psi-p_\pi(\varphi)\hat{\psi}_0\|_\pi$ follows.
Since $\xi=\hat{\psi}_0$ is the $\|\psi-\xi\|_\pi$-minimizing element for $\xi$ with range in $[N_{\mathsf{h}}\varphi]$, see  Remark \ref{bestapprox}\,\eqref{bestapprox2}, by uniqueness from the previous   $\hat{\psi}_0=p_\pi(\varphi)\hat{\psi}_0$ is obtained, that is, $f\in {\mathsf{T}}\Omega_\nu$ implies $\hat{\psi}_0\in p_\pi(\varphi) {\mathcal{H}}_\pi$.
\end{proof}
\begin{remark}\label{tanstratrem}
By standard arguments (see the proof of  Corollary \ref{tanfolg}, e.g.) the condition in Example \ref{tanstrat} with a particular set of input data $\{\pi,\varphi,\hat{\psi}_0\}$  holds  if, and only if, it is satisfied with any other set of input data referring to the same $f$.
\end{remark}
\newpage
\section{Conclusions, examples and applications}\label{conappl}
For a unital ${\mathsf C}^*$-algebra $M$,  with $M\not={\mathbb{C}}{\mathsf{1}}$, let $\gamma$ be a parameterized curve
 \begin{subequations}\label{3tang0}
 \begin{equation}\label{3tang0a}
 \gamma: I\ni t\longmapsto \nu_t \in {\mathcal S}(M)
\end{equation}
passing through a state $\nu=\nu_0$ and evolving in the state space ${\mathcal{S}}(M)$ of $M$
and admitting a $\gamma$-compliant unital $^*$-representation $\{\pi,{\mathcal H}_\pi\}$ around $\nu$. It is important to note that by the postulated existence of a curve \eqref{3tang0a} of the above type some restriction on the choice of the state $\nu$ relative to ${\mathcal{S}}(M)$ will exist. Especially, if each state $\nu$ can be an option this suggests ${\mathcal{S}}(M)$ to be a (sub)manifold of states, see Definition \ref{submann}/Theorem \ref{mannig}, and requiring $M$ to be of dimension three at least, e.g.~for a non-commutative $M$ the setting of  \eqref{3tang0a} unambiguously gets possible at any $\nu$. The generic case of \eqref{3tang0a} will be that $\nu$ can be chosen freely from the states of some (proper) submanifold $\Omega\subset {\mathcal{S}}(M) $ of states, and thus $\nu\in \gamma\subset \Omega$ is fulfilled.

Along with \eqref{3tang0a}, by assumption on $\gamma$  there is an open interval  $I_\pi$, with $0\in I_\pi\subset I$, and a differentiable at $t=0$ map into the $\pi$-fibres of states of $\gamma$ around $\nu$,
\begin{equation}\label{3tang0b}
 I_\pi\ni t\mapsto \varphi_t\in {\mathcal S}_{\pi,M}(\nu_t)
\end{equation}
 such that $\varphi_t=\varphi+t\,\psi+{\mathbf o}(t)$, with $\varphi=\varphi_0$, and with  the tangent vector $\psi=\varphi_t^{\,\prime}|_{t=0}$ \begin{equation}\label{3tang0c}
 \psi=\psi(\pi,{\mathcal H}_\pi,(\varphi_t))\in{\mathsf{T}}_\nu(M|\gamma)
 \end{equation} of the implementation $(\varphi_t)$ of $\gamma$ around $\nu$,
 with tangent form $f\in{\mathsf{T}}_\nu(M)$ satisfying
\begin{equation}\label{3tang}
   f(\cdot) =\frac{d}{d\/t}\,\nu_t\bigg|_{t=0}= \langle \pi(\cdot) \psi,\varphi\rangle_\pi+\langle \pi(\cdot)\varphi, \psi\rangle_\pi
   \end{equation}
on $M$. In case of $\gamma\subset \Omega$, with submanifold $\Omega$ of states, the tangent form $f$ in \eqref{3tang} is bound to range through the local tangent space ${\mathsf{T}}\Omega_\nu$ relative to $\Omega$,  ${\mathsf{T}}\Omega_\nu\subset {\mathsf{T}}_\nu (M)$, see Definition \ref{loctan}.  Obviously,  $f$ can be obtained as
$f(\cdot)=f_\pi(\pi(\cdot))$ with the help of the ultrastrongly continuous hermitian linear form $f_\pi$ over $\pi(M)^{\,\prime\prime}$ given by
 \begin{equation}\label{3tang1}
   f_\pi = \re f_{2\psi,\varphi}
   \end{equation}
\end{subequations}
Subsequently, the above common premises and meanings in context of $M$, $\gamma$, $\pi$, ${\mathcal H}_\pi$, $\varphi_t$, $\psi$ and $f$, $f_\pi$ will be supposed to be satisfied.

\subsection{Main results}\label{conappl1}
\subsubsection{An invariant reading in terms of the tangent form}\label{conappl2}
In consequence of Theorem \ref{Cstern}, Corollary \ref{tanfolg} and Remark \ref{bestapprox}, $\|\hat{\psi}_0\|_\pi$ is an invariant of $\gamma$. In fact,  neither does $\|\hat{\psi}_0\|_\pi$ depend on the special local implementation $(\varphi_t)$ used, nor does it actually relate to the details of the $\gamma$-compliant $^*$-representation $\{\pi,{\mathcal H}_\pi\}$ used.
\begin{lemma}\label{explbound}
Let $\gamma$ be passing through $\nu$, with tangent form  $f\in{\mathsf{T}}_\nu(M)$. Suppose  $\psi=\psi(\pi,{\mathcal H}_\pi,(\varphi_t))\in{\mathsf{T}}_\nu(M|\gamma)$, with $\gamma$-compliant $^*$-representation $\{\pi,{\mathcal H}_\pi\}$ around $\nu$. Let the best $\langle\cdot,\cdot\rangle_{\pi,\mathbb R}$-approximation of $\psi\in ({\mathcal H}_\pi)_{\mathbb R}$ within $[\pi(M)_h\varphi]$ be given by   $$\hat{\psi}_0=\hat{\psi}_0(\pi,{\mathcal H}_\pi,(\varphi_t))$$
Then, with $\|f\|_\nu$ defined in accordance with  \textup{Definition \ref{e.1}}, the following is fulfilled:
\begin{equation}\label{explbound1}
  \|\hat{\psi}_0\|_\pi =\|f\|_\nu
\end{equation}
Beyond that, in order to calculate $\|f\|_\nu$,  in formula \textup{\eqref{tangentnorm}} it is sufficient to take into account only those decompositions $\{x\}$ consisting of mutually commuting elements.
\end{lemma}
\begin{proof}
The best $\langle\cdot,\cdot\rangle_{\pi,\mathbb R}$-approximation of $\psi\in ({\mathcal H}_\pi)_{\mathbb R}$ in $[\pi(M)_h\varphi]$ is the vector $\hat{\psi}_0$ figuring in context of Theorem \ref{Cstern}\,\eqref{Cstern2}, see  Remark \ref{bestapprox}. Along the latter
the validity of \eqref{explbound1} follows. That in formula  \eqref{tangentnorm} to positive decompositions of mutually commuting elements of $M$ can be restricted is a consequence of Remark \ref{comm1} and since \eqref{idi} is fulfilled, with ${\mathfrak A}=\pi(M)$ and $f_\pi$ according to \eqref{3tang1}, and in considering the relation   $$[\pi(M)_h\varphi]=[{\pi(M)^{\,\prime\prime}}_h\varphi]$$ which follows by applying a Kaplansky density argument.
\end{proof}
\subsubsection{Length element and fundamental form in Bures geometry}\label{buresleng}
In view of \eqref{3b.3}, Lemma \ref{rinow0} and Lemma \ref{explbound}, under the common premises of this section  Theorem \ref{line0} now can be given the following form:
\begin{theorem}\label{fundform}
Let $\gamma$ be passing through $\nu$, with tangent form $f\in{\mathsf{T}}_\nu(M)$, and tangent vector    $\psi=\psi(\pi,{\mathcal H}_\pi,(\varphi_t))\in{\mathsf{T}}_\nu(M|\gamma)$ referring to the $\gamma$-compliant $^*$-representation $\{\pi,{\mathcal H}_\pi\}$ around $\nu$. The following then holds:
\begin{equation}\label{fundform1}
 \biggl(\limsup_{t\to 0}
\frac{d_B(M|\nu_{t},\nu)}{|t|}\biggr)^2=\|f\|_\nu^2+\|p_\pi(\varphi)^\perp p_\pi^{\,\prime}(\varphi)^\perp\psi\|_\pi^2
\end{equation}
\end{theorem}
\noindent The formula \eqref{fundform1} will be referred to as `Pythagorean law'. The latter can be given a formulation reading in terms of the  Bures length element at $\nu$ along $\gamma$. In fact, since
$\nu_t(\cdot)=\langle\pi(\cdot)\varphi_t,\varphi_t\rangle_\pi$ holds
with  respect to the $\gamma$-compliant $^*$-representation $\{\pi,{\mathcal H}_\pi\}$ and associated implementation $(\varphi_t)$ of $\gamma$, infinitesimally the auxiliary condition
\begin{subequations}\label{fundform20}
\begin{equation}\label{fundform2a}
{\mathsf{d}}\nu(\cdot)=\langle\pi(\cdot){\mathsf{d}}\varphi,\varphi\rangle_\pi+\langle\pi(\cdot)\varphi,{\mathsf{d}}\varphi\rangle_\pi
\end{equation}
between the differentials ${\mathsf{d}}\nu$ at $\nu$  along $\gamma$ and ${\mathsf{d}}\varphi$ at $\varphi$ along $(\varphi_t)$, respectively, is satisfied. By Lemma \ref{rinow0} and due to the Pythagorean law the square of the length element ${\mathsf{d}}s={\mathsf{d}}s[\gamma,\nu]$ along the curve $\gamma$ at $\nu$ in a covariant manner can be given as
\begin{equation}\label{fundform2}
{\mathsf{d}}s^2=\|{\mathsf{d}}\nu\|_\nu^2+\|p_\pi(\varphi)^\perp p_\pi^{\,\prime}(\varphi)^\perp {\mathsf{d}}\varphi\|_\pi^2
\end{equation}
\end{subequations}
In line with this, note that there is given another invariant
 \begin{equation}\label{inv2}
 I_\nu(M|\gamma)=\sqrt{\bigl({\mathop{\mathrm{dil}}}_t^B \gamma|_{t=0}\bigr)^2-\|f\|_\nu^2}=\|p_\pi(\varphi)^\perp p_\pi^{\,\prime}(\varphi)^\perp\psi\|_\pi
\end{equation}
which is intrinsic of the curve $\gamma$ and state $\nu$ it is passing through, but which now reads in terms of the tangent vector $\psi\in{\mathsf{T}}_\nu(M|\gamma)$ of an implementation of $\gamma$. Obviously, it proves independent of the special choice of the unital $^*$-representation $\{\pi,{\mathcal H}_\pi\}$ and the implementation $(\varphi_t)$ of the given parameterized curve $\gamma$ at $\nu$ as far as $\varphi=\varphi_0$ and $\psi=\psi(\pi,{\mathcal H}_\pi,(\varphi_t))\in{\mathsf{T}}_\nu(M|\gamma)$ are fulfilled.

For fixed $\nu$ and given $f\in{\mathsf{T}}_\nu(M)$, let us consider the equivalence class of admitted parameterized curves $\gamma$ passing through $\nu$ with tangent form $f$, that is $\gamma^{\,\prime}(0)=f$ be fulfilled for the curves of this class. The range of values taken by the invariant $I_\nu(M|\gamma)$ while letting $\gamma$ extend through the curves of this equivalence class be
\[
{\mathcal R}(\nu,f)=\bigl\{I_\nu(M|\gamma): \gamma^{\,\prime}(0)=f\bigr\}
\]
Note that unless $\nu$ is a character state, ${\mathsf{T}}_\nu(M)\not=\{{\mathsf 0}\}$ is fulfilled, by Corollary \ref{extan}. Therefore ${\mathcal R}(\nu,f)$  possesses the subsequently described property.
\begin{lemma}\label{irange}
For $f\in{\mathsf{T}}_\nu(M)\backslash \{{\mathsf 0}\}$, either ${\mathcal R}(\nu,f)=\{0\}$ or ${\mathcal R}(\nu,f)={\mathbb R}_+$ holds.
\end{lemma}
\begin{proof}
 Suppose ${\mathcal R}(\nu,f)\not=\{0\}$ for some   $f\in{\mathsf{T}}_\nu(M)\backslash \{{\mathsf 0}\}$, and be $\gamma$ a parameterized curve passing through $\nu$ with tangent form $\gamma^{\,\prime}(0)=f$ there and tangent vector $\psi=\psi(\pi,{\mathcal H}_\pi,(\varphi_t))\in{\mathsf{T}}_\nu(M|\gamma)$ with respect to an implementation $(\varphi_t)$ of $\gamma$ around $\nu$ and $\gamma$-compliant representation $\{\pi,{\mathcal H}_\pi\}$ such that
\[
I_\nu(M|\gamma)=\|p_\pi(\varphi)^\perp p_\pi^{\,\prime}(\varphi)^\perp\psi\|_\pi\not=0
\]
For given $\alpha\in {\mathbb R}_+$, let us define the vector $\phi_\alpha$ by
\[
\phi_\alpha=\hat{\psi}_0+\frac{\alpha}{\|p_\pi(\varphi)^\perp p_\pi^{\,\prime}(\varphi)^\perp\psi\|_\pi}\,p_\pi(\varphi)^\perp p_\pi^{\,\prime}(\varphi)^\perp\psi
\]
with the best $\langle\cdot,\cdot\rangle_{\pi,{\mathbb R}}$-approximation  $\hat{\psi}_0$ of $\psi$ within $[\pi(M)\varphi]$. Note that then $$\langle\phi_\alpha,\varphi\rangle_{\pi,{\mathbb R}}=0$$ with $\phi_\alpha\not={\mathsf 0}$. Accordingly, $\phi_\alpha$ cannot be a multiple of $\varphi$. Also, since by assumption $f\not={\mathsf 0}$ is fulfilled, Lemma \ref{constr} can be applied, and in the case at hand  implies that, for some interval $I\subset I_\pi$ around $t=0$, the map
\[
I\ni t\longmapsto \eta_t=\frac{\varphi+t\phi_\alpha}{\sqrt{1+t^2\|\phi_\alpha\|_\pi^2}} \in {\mathcal S}_{\pi,M}(\omega_t)
\]
is the implementation map of some parameterized curve $\gamma_\alpha: I\ni t\longmapsto \omega_t\in {\mathcal S}(M)$ of the admissible class
in the state space and passing through $\nu$ at $t=0$ with $\gamma_\alpha$-compliant $^*$-representation $\{\pi,{\mathcal H}_\pi\}$ around $\nu$. It is plain to see that $$\phi_\alpha=\phi_\alpha(\pi,{\mathcal H}_\pi,(\eta_t))\in {\mathsf{T}}_\nu(M|\gamma_\alpha)$$ is the tangent vector of the above mentioned implementation of $\gamma_\alpha$. Clearly, by construction then  $\hat{(\phi_\alpha)}_0=\hat{\psi}_0$, too.  Therefore, by assumption and Lemma \ref{explbound}, $$\|\hat{(\phi_\alpha)}_0\|_\pi=\|f\|_\nu$$ Also, we see that
\[
p_\pi(\varphi)^\perp p_\pi^{\,\prime}(\varphi)^\perp\phi_\alpha=\frac{\alpha}{\|p_\pi(\varphi)^\perp p_\pi^{\,\prime}(\varphi)^\perp\psi\|_\pi}\,p_\pi(\varphi)^\perp p_\pi^{\,\prime}(\varphi)^\perp\psi
\]
has to be fulfilled.  Hence, in view of \eqref{inv2} we are arriving at
\[
I_\nu(M|\gamma_\alpha)=\|p_\pi(\varphi)^\perp p_\pi^{\,\prime}(\varphi)^\perp\phi_\alpha\|_\pi=\alpha
\]
That is, $\alpha\in {\mathcal R}(\nu,f)$. Since $\alpha\in {\mathbb R}_+$ was arbitrarily chosen, the assertion follows.
\end{proof}
\subsubsection{The Finslerian law}\label{finsler}
Suppose the invariant \eqref{inv2} is vanishing, $I_\nu(M|\gamma)=0$, for a state $\nu$ and a curve $\gamma$ passing through this state. Refer to this as `Finslerian state $\nu$ of $\gamma$' or as `$\nu$ is a Finslerian state of $\gamma$'.  By Lemma \ref{irange}, if $\nu$ is not a character state, then for each $f\in{\mathsf{T}}_\nu(M)\backslash \{{\mathsf 0}\}$ there is a curve $\gamma$ passing through $\nu$ with tangent form $f$ and such that $\nu$ is a Finslerian state of $\gamma$. More precisely, relating this the following holds.
\begin{theorem}\label{fundform3}
The state $\nu$ is a Finslerian state of $\gamma$ if, and only if
\begin{subequations}\label{fundform5}
\begin{equation}\label{fundform4}
 \lim_{t\to 0}
\frac{d_B(M|\nu_{t},\nu)}{|t|}=\|f\|_\nu
\end{equation}
is fulfilled. The latter is equivalent to the condition that for some \textup{(}and thus for any\textup{)} tangent vector $\psi=\psi(\pi,{\mathcal H}_\pi,(\varphi_t))\in{\mathsf{T}}_\nu(M|\gamma)$ a decomposition
\begin{equation}\label{fundform6}
\psi=\theta+\theta^\prime
\end{equation}
\end{subequations}
exists, with $\theta\in [\pi(M)\,\varphi]$ and $\theta^\prime\in [\pi(M)^{\,\prime}\varphi]$.
\end{theorem}
\begin{proof}
In line with the above, suppose $I_\nu(M|\gamma)=0$. Then, by \eqref{inv2} one infers that $p_\pi(\varphi)^\perp p_\pi^{\,\prime}(\varphi)^\perp\psi=0$. From \eqref{fundform1} then
\begin{equation*}
   {\mathop{\mathrm{dil}}}_0^B \gamma=\limsup_{t\to 0}
\frac{d_B(M|\nu_{t},\nu)}{|t|}=\|f\|_\nu
\end{equation*}
follows. On the other hand, by Lemma \ref{lowerbound}, \eqref{lbound}, in any case one knows that
\begin{equation*}
{\mathop{\mathrm{dil}}}_0^B \gamma\geq \liminf_{t\to 0}
\frac{d_B(M|\nu_{t},\nu)}{|t|}  \geq   \|f\|_\nu
\end{equation*}
Hence, the limit exists and obeys $${\mathop{\mathrm{dil}}}_0^B \gamma=\lim_{t\to 0}
\frac{d_B(M|\nu_{t},\nu)}{|t|}=\|f\,\|_\nu$$ which is \eqref{fundform4}. In view of  the latter by
\eqref{fundform1} obviously  $p_\pi(\varphi)^\perp p_\pi^{\,\prime}(\varphi)^\perp\psi=0$ is implied. Thus, requiring $\nu$ to be a Finslerian state of $\gamma$ is equivalent to \eqref{fundform4}. Also note that $\theta\in [\pi(M)\,\varphi]$ and $\theta^\prime\in [\pi(M)^{\,\prime}\varphi]$ imply  $p_\pi^{\,\prime}(\varphi)^\perp\theta=0$ and $p_\pi(\varphi)^\perp\theta^\prime=0$ to hold. Hence, $p_\pi(\varphi)^\perp p_\pi^{\,\prime}(\varphi)^\perp\theta=0$ as well as $p_\pi(\varphi)^\perp p_\pi^{\,\prime}(\varphi)^\perp\theta^\prime=0$ have to hold. Thus, if a  decomposition \eqref{fundform6} exists, then $p_\pi(\varphi)^\perp p_\pi^{\,\prime}(\varphi)^\perp\psi=0$ follows. On the other hand,  $p_\pi(\varphi)^\perp p_\pi^{\,\prime}(\varphi)^\perp\psi=0$ implies  $p_\pi(\varphi) p_\pi^{\,\prime}(\varphi)^\perp\psi=p_\pi^{\,\prime}(\varphi)^\perp\psi$. Hence, we have $p_\pi^{\,\prime}(\varphi)^\perp\psi\in [\pi(M)^{\,\prime}\varphi]$. From   $p_\pi^{\,\prime}(\varphi)\psi\in [\pi(M)^{\,\prime\prime}\varphi]$ and $[\pi(M)^{\,\prime\prime}\varphi]= [\pi(M)\,\varphi]$ it then follows that
$\theta=p_\pi^{\,\prime}(\varphi)\psi$ and $\theta^{\,\prime}=p_\pi^{\,\prime}(\varphi)^\perp\psi$ may be chosen in   \eqref{fundform6}.
\end{proof}
The relation \eqref{fundform4} will be referred to as the `Finslerian law', henceforth. Thereby, sufficient conditions under which the  condition \eqref{fundform6}, which according to Theorem \ref{fundform3} is equivalent to the Finslerian law, is satisfied and which are reading in terms of a local implementation can be easily read off. The simplest case of such a condition does not make reference to the very details of $(\varphi_t)$ even, and obviously guarantees that the condition \eqref{fundform6} is fulfilled for $\psi=\psi(\pi,{\mathcal H}_\pi,(\varphi_t))$ by triviallity.
\begin{lemma}\label{trivi}
Let $\varphi\in {\mathcal S}_{\pi,M}(\nu)$, for some $\gamma$-compliant representation $\{\pi,{\mathcal H}_\pi\}$ and associated implementation $(\varphi_t)$ of $\gamma$ around $\nu$, with $\varphi=\varphi_0$. Then, the condition $${\mathcal H}_\pi=[\pi(M)\varphi]+[\pi(M)^{\,\prime}\varphi]$$ implies $\nu$ to be a Finslerian state of $\gamma$.
\end{lemma}
The following condition takes more details of the implementation into account.
\begin{lemma}\label{fundform7}
    Let $(\varphi_t)$ be a differentiable local implementation of $\gamma$ around $\nu=\nu_0$ and living on the unital $^*$-representation $\{\pi,{\mathcal H}_\pi\}$. Suppose $p_\pi(\varphi_t)\leq p_\pi(\varphi)$ or
    $p_\pi^{\,\prime}(\varphi_t)\leq p_\pi^{\,\prime}(\varphi)$ for all $t$ of a subset ${\mathcal U}\subset I_\pi$ with limit point zero.
    Then the Finslerian law \eqref{fundform4} is satisfied.
\end{lemma}
\begin{proof}
The condition $p_\pi(\varphi_t)\leq p_\pi(\varphi)$ implies $ p_\pi(\varphi)(\varphi_t-\varphi)=(\varphi_t-\varphi)$ for  all $t\in {\mathcal U}$. Since $0$ is the limit of a sequence $\{t_n\}\subset {\mathcal U}$ with $t_n\not=0$ from this $p_\pi(\varphi)\psi=\psi$ follows. This is equivalent to $\psi\in [\pi(M)^{\,\prime}\varphi]$, which is a special case of \eqref{fundform6}. Analogously, from  $p_\pi^{\,\prime}(\varphi_t)\leq p_\pi^{\,\prime}(\varphi)$ the special case $\psi\in [\pi(M)^{\,\prime\prime}\varphi]=[\pi(M)\varphi]$ of \eqref{fundform6} is implied.  Hence, by Theorem \ref{fundform3}, in either case \eqref{fundform4} follows.
\end{proof}
In some situation, criteria formulated in Lemma \ref{fundform7} even can prove necessary.
\begin{example}\label{fundform7a}
Let $\gamma\in \Gamma(M|\varrho,\nu)$ be given in accordance with Example \ref{ex2}, and consider $\{\pi,{\mathcal H}_\pi\}$ and the local implementation $(\varphi_t)$ of $\gamma$ as given by \eqref{geo0a}. Then, the Finslerian law \eqref{fundform4} is satisfied if, and only if, $p_\pi^{\,\prime}(\varphi_t)\leq p_\pi^{\,\prime}(\varphi)$, for all $t$.
\end{example}
\begin{proof}
In view of Lemma \ref{fundform7} only necessity of the condition in question remains to be shown. Suppose condition \eqref{fundform4} is satisfied. Then one has
\[
{\mathop{\mathrm{dil}}}_0^B \gamma=\|f\|_\nu
\]
According to the construction of Example \ref{ex2} by setting for $\zeta\in {\mathcal S}_{\pi,M}(\varrho)$ one has $h_{\zeta,\varphi}^\pi\geq 0$ on $\pi(M)^{\,\prime}$. From this and \eqref{geo1} one infers that, with the tangent vector
$$\psi=\psi(\pi,{\mathcal H}_\pi,(\varphi_t))\in{\mathsf{T}}_\nu(M|\gamma)$$ $h_{\psi,\varphi}^\pi$ is a hermitian linear form on $\pi(M)^{\,\prime}$. Application of
Corollary \ref{kreuz0} then yields $$p_\pi^{\,\prime}(\varphi)\psi\in [\pi(M)^{\,\prime\prime}_h\varphi]$$
Hence $\hat{\psi}_0=p_\pi^{\,\prime}(\varphi)\psi$. Remind that $\|\hat{\psi}_0\|_\pi=\|f\|_\nu$, by \eqref{explbound1}. Therefore, by \eqref{geo2a} and by supposition one has
\begin{eqnarray*}
    {\mathop{\mathrm{dil}}}_0^B \gamma=\|f\|_\nu = \|\psi\|_\pi &=&\|p_\pi^{\,\prime}(\varphi)\psi+p_\pi^{\,\prime}(\varphi)^\perp\psi\|_\pi\\
     &=& \sqrt{\|p_\pi^{\,\prime}(\varphi)\psi\|_\pi^2+\|p_\pi^{\,\prime}(\varphi)^\perp\psi\|_\pi^2} \\
     &=& \sqrt{\|\hat{\psi}_0\|_\pi^2+\|p_\pi^{\,\prime}(\varphi)^\perp\psi\|_\pi^2}\\
     &=& \sqrt{\|f\|_\nu^2+\|p_\pi^{\,\prime}(\varphi)^\perp\psi\|_\pi^2}
\end{eqnarray*}
 Hence  $p_\pi^{\,\prime}(\varphi)^\perp\psi={\mathsf{0}}$  by the previous, i.e.~$p_\pi^{\,\prime}(\varphi)\psi=\psi$. From this and $\psi=\zeta-{F}\varphi$ with $F=F(M|\nu,\varrho)$ one infers that $p_\pi^{\,\prime}(\varphi)\zeta=\zeta$. In view of \eqref{geo0a} then  $p_\pi^{\,\prime}(\varphi)\varphi_t=\varphi_t$ follows. Hence, $p_\pi^{\,\prime}(\varphi_t)\leq p_\pi^{\,\prime}(\varphi)$ has to be fulfilled.
\end{proof}
\begin{remark}\label{beispielend}
  The Finslerian law \eqref{fundform4} examined in Theorem \ref{fundform3} is an important special case of  \eqref{fundform1}, but is not the general case. By Lemma \ref{explbound}, formula \eqref{explbound1}, for $\gamma\in \Gamma(M|\nu,\varrho)$ given in Example \ref{exupper} and which is obeying $$\|f\|_\nu<{\mathop{\mathrm{dil}}}_0^B \gamma$$ this e.g.~is evident and in this case $\nu$ cannot be a Finslerian state of $\gamma$.
\end{remark}
\subsection{Curves and submanifolds respecting the Finslerian law}\label{specex}
Suppose the state space ${\mathcal S}(M)$ of the unital ${\mathsf C}^*$-algebra $M\not=\mathbb{C}\,{\mathsf{1}}$  to be equipped with the Bures metric, cf.~Remark \ref{allg2c}.
Let $\gamma:I\ni t\longmapsto \nu_t\in {\mathcal S}(M)$ be a parameterized curve in ${\mathcal S}(M)$. Suppose $\gamma$ admits a $\gamma$-compliant unital $^*$-representation around each of the states it is passing through (call them `inner states'). An important special case of $\gamma$ occurs if the Pythagorean law \eqref{fundform1} at each inner state of $\gamma$  reduces to the Finslerian law \eqref{fundform4}.
\begin{definition}\label{finslercurve}
Refer to $\gamma$ as `Finslerian curve' if at each inner $s\in I$  of $\gamma$ the law
\begin{equation}\label{finsler1a}
       \lim_{t\to s}
\frac{d_B(M|\nu_{t},\nu_s)}{|t-s|}=\|f_s\|_{\nu_s}
\end{equation}
is satisfied, with tangent form $f_s=\nu_t^{\,\prime}|_{t=s}$.
\end{definition}
More generally, in view of Definition \ref{submann}, if the Bures metric is restricted to a submanifold $\Omega\subset {\mathcal{S}}(M)$ of states, then in context of Theorem \ref{fundform} and \eqref{fundform20} it is useful to set up the following  notion and consequence thereof.
\begin{definition}\label{finslermani}
$\Omega$ is termed `Finslerian submanifold' provided
\begin{equation}\label{finslermani0}
{\mathsf{d}}s^2=\|{\mathsf{d}}\nu\|_\nu^2
\end{equation}
is fulfilled, at each $\nu\in \Omega$, with length element ${\mathsf{d}}s={\mathsf{d}}s[\gamma,\nu]$ taken along any parameterized curve $\gamma\subset \Omega$  passing through $\nu$ and admitting a $\gamma$-compliant unital $^*$-representation of $M$ around there.
\end{definition}
\begin{remark}\label{vererb}
Each submanifold  $\Omega_0\subset \Omega$ of a Finslerian submanifold $\Omega$  is Finslerian.
\end{remark}
\subsubsection{Examples of Finslerian curves}\label{fins}
Start with  a parameterized continuous curve $\gamma=(\nu_t)\subset {\mathcal S}(M)$ which is a geodesic arc, see Definition \ref{allg2c0}.
\begin{example}\label{geoa}
Geodesic arcs are Finslerian curves.
\end{example}
\begin{proof}
Let $\gamma:[0,1]\ni t\longmapsto \nu_t\in {\mathcal S}(M)$  be a geodesic arc starting at a state $\nu$ and terminating at another state $\varrho\not=\nu$, $\gamma\in {\mathcal{C}}^{\nu,\varrho}(M)$. Suppose a defining
implementation of $\gamma$ to be given by $[0,1]\ni t\longmapsto \varphi_t\in {\mathcal S}_{\pi,M}(\nu_t)$ over some unital $^*$-representation $\{\pi,{\mathcal H}_\pi\}$ as defined in accordance with \eqref{geo0a} when restricted to the unit interval (cf. Remark \ref{allg2c},\,\eqref{allg2c0}). Obviously, since geodesic arcs are simple curves, each state $\nu_s$ corresponding to parameter values with $0<s<1$ is an inner state. Then, by Corollary \ref{porter3}, $p_\pi^{\,\prime}(\varphi_t)=p_\pi^{\,\prime}(\varphi_s)$, for each $t$ with $0<t<1$. Hence, if for given $s$ a local implementation $\tilde{\varphi}_t=\varphi_{s+t}$ for all $t$ with $|t|$ sufficiently small is considered, we have  $p_\pi^{\,\prime}(\tilde{\varphi}_t)=p_\pi^{\,\prime}(\tilde{\varphi})$ for such $t$, with  $\tilde{\varphi}=\tilde{\varphi}_0$.
Applying
Lemma \ref{fundform7} to $(\tilde{\varphi}_t)$ guarantees that Theorem \ref{fundform3} can be applied to $\tilde{\nu}=\nu_s$ and $\tilde{\nu}_t=\nu_{t+s}$. Hence, the \eqref{fundform5} corresponding formula \eqref{finsler1a} holds. Since $0<s<1$ could have been chosen at will, \eqref{finsler1a} has to be true at any $s\in ]0,1[$.
\end{proof}
Let by $\alpha_t\in {\mathsf{AUT}}(M)$ a uniformly continuous one-parameter group of $^*$-automorph\-isms over $M$ be given. Consider $\gamma:{\mathbb R}\ni t\longmapsto \nu_t\in {\mathcal S}(M)$, with $\nu_t$ defined by
\begin{equation}\label{auto}
\nu_t(x)=\nu(\alpha_t(x))=\nu\circ\alpha_t(x)
\end{equation}
for any $x\in M$ and $t\in {\mathbb R}$. To be non-trivial, suppose that $\nu$ is not a fixed point of $(\alpha_t)$. Refer to the implementable curve $\gamma\subset{\mathcal S}(M)$ as Hamiltonian trajectory.
\begin{example}\label{geoh}
Hamiltonian trajectories are Finslerian curves.
\end{example}
\begin{proof}
According to \eqref{auto},
 $\nu_t=\nu\circ\alpha_t$, for each $t\in{\mathbb R}$, and each $\nu_s$ is an inner state of $\gamma$. Let us chose a  unital $^*$-representation $\{\pi,{\mathcal H}_\pi\}$ such that $\tilde{\varphi}\in {\mathcal S}_{\pi,M}(\nu)$ exists. Then, there exists $h\in \pi(M)^{\,\prime\prime}$ with $h=h^*$ such that $\pi(\alpha_t(x))=\exp (-{\mathsf{i}}ht)\pi(x) \exp ({\mathsf{i}}ht)$, for any $x\in M$, see \cite{Kadi:66, Saka:66}, and the global implementation $(\tilde{\varphi}_t)$ of the Hamiltonian flow given by $\tilde{\varphi}_t=\exp ({\mathsf{i}}h t)\,\tilde{\varphi}$ can be considered. For fixed $s$, and all $t$ from some symmetric interval around zero,  $\varphi_t=\tilde{\varphi}_{s+t}$ yields a local implementation at $\nu_s$. Since  $\tilde{\varphi}_{t+s}=\exp\bigl( {\mathsf{i}}h t \bigr)\,\tilde{\varphi}_s$ is fulfilled,  with $\varphi=\tilde{\varphi}_s$ then
$\varphi_{t}\in \pi(M)^{\,\prime\prime}\varphi\subset [\pi(M)\varphi]$ follows, see \eqref{bas3a}. Hence  $p_\pi^{\,\prime}(\varphi_t)\leq p_\pi^{\,\prime}(\varphi)$, for each $t$ (in fact equality occurs).
Thus Lemma \ref{fundform7} and  Theorem \ref{fundform3} apply and imply that formula \eqref{fundform4} extends to \eqref{finsler1a}, at any $s\in {\mathbb R}$.
\end{proof}
\subsubsection{Examples of Finslerian  submanifolds}\label{finssubman}
Let ${\mathsf{ex}}\,{\mathcal S}(M)$ be the set of extremal points (pure states) of the (affinely) convex set ${\mathcal S}(M)$. By Example \ref{puremann}, for a non-commutative unital ${\mathsf{C}}^*$-algebra $M$ the set ${\mathcal S}^{{\mathsf{ex}}}(M)$ of non-character, pure states
\begin{equation*}
{\mathcal S}^{{\mathsf{ex}}}(M)=\bigl\{ \nu\in {\mathsf{ex}}\,{\mathcal S}(M):\ \exists x,y\in M, \nu(x y)\not=\nu(x)\nu(y)\bigr\}
\end{equation*}
is a submanifold. In respect of the Bures metric in addition the following holds.
\begin{theorem}\label{geoext}
If $M$ is non-commutative, then ${\mathcal S}^{{\mathsf{ex}}}(M)$
is Finslerian.
\end{theorem}
\begin{proof}
For each $\nu\in {\mathcal S}^{{\mathsf{ex}}}(M)$, we have to show that $\nu$ proves Finslerian in respect of any parameterized curve $\gamma$ obeying $\gamma\subset {\mathcal S}^{{\mathsf{ex}}}(M)$ and passing through $\nu$ and thereby admitting a $\gamma$-compliant unital $^*$-representation $\{\pi,{\mathcal{H}}_\pi\}$ of $M$ around $\nu$. Note that since by assumption and Example \ref{puremann} the set  ${\mathcal S}^{{\mathsf{ex}}}(M)$ is a submanifold of states, then one can be assured that to each $\nu\in {\mathcal S}^{{\mathsf{ex}}}(M)$ a curve  of the mentioned specification and passing through $\nu$ really exists. In line with this, suppose $\gamma\subset {\mathcal S}^{{\mathsf{ex}}}(M)$ be such a curve, with differentiable local implementation $I_\pi\ni t\longmapsto \varphi_t\in {\mathcal S}_{\pi,M}(\nu_t)$,
with symmetric to zero interval of reals $I_\pi$, and obeying $\nu_0=\nu$. Let  $\varphi=\varphi_0$. Note that by continuity of the implementation $(\varphi_t)$ at $t=0$, we have $\lim_{t\to 0} \langle\varphi_t,\varphi\rangle_\pi=1$. Hence,  $\langle\varphi_t,\varphi\rangle_\pi\not=0$ for all $t$ sufficiently close to zero, with
$\varphi_t$ and $\varphi$ implementing the pure states $\nu_t$ and $\nu$, respectively. Application of Lemma \ref{auxpure} then yields $p_\pi^{\,\prime}(\varphi)= p_\pi^{\,\prime}(\varphi_t)$, for all those parameter values.  Applying Lemma \ref{fundform7} and  Theorem \ref{fundform3} to $(\varphi_t)$ at $t=0$ implies $\nu$ to be a Finslerian state of $\gamma$.
\end{proof}
\begin{remark}\label{fsmetrik1}
Relating Theorem \ref{geoext}, note that provided $M$ is a non-commutative ${\mathsf{W}}^*$-algebra possessing normal pure states, then the set $${\mathcal{S}}^{{\mathsf{ex}}}_0(M)={\mathcal{S}}^{{\mathsf{ex}}}(M)\cap {\mathcal{S}}_0(M)$$ is a  submanifold of normal states, see Corollary \ref{wstar2}. In this situation the arguments raised in the proof of Theorem \ref{geoext} apply  as well and show that when endowed with the Bures metric ${\mathcal{S}}^{{\mathsf{ex}}}_0(M)$ has to be    Finslerian in either case.
\end{remark}
Let $\Omega\subset {\mathcal S}(M)$ be a stratum. That is, $\Omega$ is a maximal subset of mutually absolutely continuous states, see Definition \ref{absstet} and below, and which according to Lemma \ref{setprop} is a special (affinely) convex subset of states, which is a submanifold of states iff $\Omega=\Omega_M(\varrho)$ holds, for some mixed state $\varrho\in {\mathcal{S}}(M)$, see Theorem \ref{substrat}.
\begin{theorem}\label{geostrat}
A stratum generated by a mixed state is a  Finslerian submanifold.
\end{theorem}
\begin{proof}
Let $\Omega$ be a stratum generated by a mixed state. Let $\nu\in \Omega$ be arbitrarily chosen, and be $\gamma: I\ni t\longmapsto \nu_t\in \Omega$ a parameterized curve passing through $\nu$ at $t=0$, and be $\{\pi,{\mathcal H}_\pi\}$ a unital $^*$-representation of $M$ such that $I_\pi\ni t\longmapsto \varphi_t\in {\mathcal S}_{\pi,M}(\nu_t)$,
with some symmetric around zero interval of reals $I_\pi$, is a differentiable at $t=0$ implementation of $\gamma$ around $\nu$. Since according to Theorem \ref{substrat} under the assumption at hand $\Omega$  is a submanifold, such curves really exist. Let $\varphi=\varphi_0$. Note that $\nu, \nu_t\in \Omega$ means that $\nu_t\dashv\vdash \nu$ is fulfilled,  for each $t\in I_\pi$. By Theorem \ref{suppabsstet} and \eqref{traeger} this implies $p_\pi(\varphi_t)=p_\pi(\varphi)$, for all  $t\in I_\pi$.  In view of this and Theorem \ref{fundform3}, \eqref{finslermani0} is obtained from  \eqref{fundform2} by applying Lemma \ref{fundform7} with  $(\varphi_t)$.
\end{proof}
As a simple application of Theorem \ref{geostrat}, consider the Bures metric in restriction to the  set of faithful  normal states ${\mathcal S}^{\mathsf{faithful}}_0(M)$ over a non-trivial
${\mathsf W}^*$-algebra $M$.
\begin{example}\label{faithst}
${\mathcal S}^{\mathsf{faithful}}_0(M)$ is a Finslerian submanifold.
\end{example}
\begin{proof}
If $M\not=\mathbb{C
}\,{\mathsf{1}}$ is a ${\mathsf{W}}^*$-algebra with faithful normal state $\omega$, then by Theorem \ref{wstar}\,\eqref{wstar2d} the stratum $\Omega_M(\omega)$ is a submanifold of normal states. By Corollary \ref{standform}, $s(\omega)={\mathsf{1}}$  implies each normal state of full support to belong to the $\omega$-stratum, i.e.~
\begin{equation}\label{stratbed}
    {\mathcal S}_0^{\mathsf{faithful}}(M)= \Omega_M(\omega)
\end{equation}
has to be fulfilled. In view of  Theorem \ref{geostrat} the example then obviously follows.
\end{proof}
\begin{remark}\label{vererbspec}
Over a non-trivial
${\mathsf W}^*$-algebra $M$ possessing  faithful normal states, each submanifold of ${\mathcal S}_0^{\mathsf{faithful}}(M)$ is Finslerian, see Example \ref{faithst} and Remark  \ref{vererb}.
\end{remark}
In case of $M={\mathsf B}({\mathcal H})$, for separable ${\mathcal{H}}$ with ${\mathsf{dim}}{\mathcal H}\geq 2$, when discussing submanifolds of faithful normal states, the states of ${\mathcal S}^{\mathsf{faithful}}_0({\mathsf B}({\mathcal H}))$ often will be identified with the density operators of full support over $\mathcal H$, see the procedures mentioned on in \ref{beimatrix}. But in contrast to the use in \ref{beimatrix}, subsequently simply the same abbreviating notation $\varrho$ will be used for both, the state and the density operator in question.
\begin{definition}\label{blatt}
For fixed $\mu\in {\mathcal S}^{\mathsf{faithful}}_0({\mathsf B}({\mathcal H}))$ the subset given as
\[
{\mathcal S}_0(\mu)=\bigl\{\varrho\in {\mathcal S}^{\mathsf{faithful}}_0({\mathsf B}({\mathcal H})):\ \varrho\mu=\mu \varrho\bigr\}
\]
will be referred to as $\mu$-leaf (within the density operators of full support). Thereby, in the defining relation the operator multiplication is meant.
\end{definition}
\begin{example}\label{fins1}
Each $\mu$-leaf is a Finslerian submanifold of states.
\end{example}
\begin{proof}
Since by definition ${\mathcal S}_0(\mu)\subset {\mathcal S}^{\mathsf{faithful}}_0({\mathsf B}({\mathcal H}))$ is fulfilled, in line with Remark \ref{vererbspec} the assertion follows if it can be verified that each $\mu$-leaf is a submanifold of states in terms of Definition \ref{submann}. This we are going to do now, for a density operator $\mu$ of full support chosen at will.
Remark that $\mu\in {\mathcal S}_0(\mu)$ holds. Let $\mathsf{spec}_p(\mu)=\{\mu_1,\mu_2,\mu_3,\ldots\}$ be the point spectrum of $\mu$, with $\mu_{k-1} >\mu_k$ for all values from the range of $k$ if the range has cardinality larger than one. Thus, either ${\mathsf{dim}} {\mathcal H}=\infty$ holds, and then $k$ is extending over ${{\mathbb{N}}}$, or $k$ is bounded above as $k\leq \# \mathsf{spec}_p(\mu)$, with $\# \mathsf{spec}_p(\mu)\leq {\mathsf{dim}} {\mathcal H}<\infty$ in the finite dimensional case. Also, there is a unique decomposition $\{\Delta P_k\}$ of the unity $\mathsf 1$ into mutually orthogonal orthoprojections $\Delta P_k$ such that
\begin{equation}\label{blatterz}
\mu=\sum_k \mu_k \Delta P_k
\end{equation}
holds, with $m_k(\mu)={\mathsf{dim}}\,\Delta P_k{\mathcal H}<\infty$ being the (finite) multiplicity of the $k$-th proper value of $\mu$. Refer to \eqref{blatterz} as spectral representation of $\mu$. Also, in defining $$P_k=\sum_{j\leq k} \Delta P_j$$
for each $k$, an ascendingly directed system ${\mathcal P}_\mu=\{P_k\}$ of orthoprojections to the given $\mu$ can be associated with which is obeying ${\mathsf{dim}}\,P_k{\mathcal H}<\infty$, for all subscripts $k$, and ${\mathsf{l.u.b.}}\{P_k\}={\mathsf{1}}$.
Obviously, a density operator $\varrho$ with spectral representation
$$\varrho=\sum_j \lambda_j\Delta Q_j$$
belongs to ${\mathcal S}_0(\mu)$ if, and only if,  for all $k$ and $j$ the following relations hold:
\begin{equation}\label{blatterz1}
\Delta Q_k\, \Delta P_j=\Delta P_j\, \Delta Q_k
\end{equation}
The map
$
\Phi_\mu(\cdot)=\sum_k \Delta P_k(\cdot)\Delta P_k
$
over ${\mathsf B}({\mathcal H})$ and unital subalgebra $\sum_k^\oplus {\mathsf B}(\Delta P_k {\mathcal H})$ will be considered. The latter is the fixed point algebra of the conditional expectation $\Phi_\mu$.
As a consequence of \eqref{blatterz} and \eqref{blatterz1} the following representation formulae are  obtained.
\begin{subequations}\label{blattdar}
\begin{eqnarray}\label{blattdar1}
\ \ \ \ \ \ \ \ \ &{\mathcal S}_0(\mu)&=\ \ \Phi_\mu\bigl({\mathcal S}^{\mathsf{faithful}}_0({\mathsf B}({\mathcal H}))\bigr)\\
\label{blattdar2}
&&= \Bigl\{\varrho\in {\mathcal S}^{\mathsf{faithful}}_0({\mathsf B}({\mathcal H})):\ \varrho\,\Delta P_k=\Delta P_k\, \varrho, \forall k\Bigr\}\\
\label{blattdar3}
&&= \Bigl\{\varrho\in {\mathcal S}^{\mathsf{faithful}}_0({\mathsf B}({\mathcal H})):\ \ \ \,\varrho\, P_k= P_k\, \varrho, \ \ \,\forall k\Bigr\}
\end{eqnarray}
\end{subequations}
Formula \eqref{blattdar3} amounts to
$
{\mathcal S}_0(\mu)=\Omega_{{\mathcal P}_\mu}({\mathcal H})
$,
with $\Omega_{{\mathcal P}_\mu}({\mathcal H})$ arising from $\Omega_{{\mathcal P}}({\mathcal H})$ as defined in Example \ref{Bopex} at  ${\mathcal P}={\mathcal P}_\mu$. But following Example \ref{imbed0ex}, $\Omega_{{\mathcal P}_\mu}({\mathcal H})$ is a submanifold of states. Hence, ${\mathcal S}_0(\mu)$ is defining a submanifold of states.
\end{proof}
\begin{remark}\label{geomann0}
As a consequence of \eqref{blattdar1}, ${\mathcal S}^{\mathsf{faithful}}_0({\mathsf B}({\mathcal H}))={\mathcal S}_0(\mu)$ can happen. In fact, the latter will occur, if and only if, $\mu$ is obeying $\# \mathsf{spec}_p(\mu)=1$. Due to compactness of $\mu$, the latter condition can be satisfied only if $ {\mathsf{dim}} {\mathcal H}<\infty$ holds and $\mu$ is equal to the equipartition $\mu=(1/{ {\mathsf{dim}} {\mathcal H}})\,{\mathsf 1}$ there, see also Remark \ref{Bop1c} \eqref{Bop1ca}.
\end{remark}
\subsection{Analyzing the structure of the examples}\label{fins2}
The Finslerian submanifolds previously mentioned will be analyzed more in detail. Special attention will be paid to convexity requirements which by a  submanifold $\Omega$ might be fulfilled when considered as a subset of the state space under the Bures metric, see Definiton \ref{geoconvex}.
\subsubsection{The geometry of pure states}\label{fins1a}
The Finslerian submanifold $\Omega={\mathcal{S}}^{{\mathsf{ex}}}(M)$ of all pure states of a non-commutative, unital ${\mathsf{C}}^*$-algebra $M$ will be considered, see  Theorem \ref{geoext}.
The first aim will be to clarify the geodesic structure of  ${\mathcal{S}}^{{\mathsf{ex}}}(M)$ under the Bures metric.
\begin{lemma}\label{geoextarc}
Let $M$ be non-commutative, $\varrho,\nu\in {\mathcal{S}}^{{\mathsf{ex}}}(M)$ with $\varrho\not\perp \nu$. There exists exactly one geodesic arc $\gamma \in {\mathcal{C}}^{{\nu,\varrho}}(M)$. Also, $\gamma$ is obeying $|\gamma|\subset {\mathcal{S}}^{{\mathsf{ex}}}(M)$.
\end{lemma}
\begin{proof}
For $\mu=\varrho_\nu\in M_+^*$ given in accordance with Lemma \ref{diffreteile},  ${\mathsf{0}}\leq \mu\leq \varrho$ and $\mu\perp \nu$ hold. For $\varrho$ being a pure state, necessarily $\mu=s\varrho$ follows, with $s\in [0,1]$. Thereby, note that $s\not=0$ cannot occur, for in that case $s\varrho \perp \nu$ would contradict that $\varrho\not\perp \nu$ by assumption. Hence $\varrho_\nu={\mathsf{0}}$, and by symmetry $\nu_\varrho={\mathsf{0}}$. Thus, in a trivial manner $\varrho_\nu$ and $\nu_\varrho$ are mutually disjoint. By Theorem \ref{einga} a geodesic arc $\gamma\in  \mathcal{C}^{\nu,\varrho}(M)$ then must be uniquely determined.
Let $\{\pi,{\mathcal{H}}_\pi\}$ be a unital $^*$-representation with non-trivial $\pi$-fibres for $\nu$ and $\varrho$. Then,  $\xi\in {\mathcal{S}}_{\pi,M}(\varrho)$ and $\varphi\in  {\mathcal{S}}_{\pi,M}(\nu)$ with $\langle\xi,\varphi\rangle_\pi=F(M|\nu,\varrho)$ exist. By uniqueness of  $\gamma=(\nu_t)$ and in accordance with formula  \eqref{geo0a}, in respect of $\{\pi,{\mathcal{H}}_\pi\}$ the states $\nu_t$ are implemented by $\varphi_t=t\xi+\lambda(t)\varphi\in  {\mathcal{S}}_{\pi,M}(\nu_t)$, $t\in [0,1]$. By Corollary \ref{subadd}\,\eqref{subadd2} the condition $\varrho\not\perp \nu$ implies $F(M|\nu,\varrho)\not=0$. Hence  $\langle\xi,\varphi\rangle_\pi\not=0$, and thus Lemma \ref{auxpure} can be applied and yields $p_\pi^{\,\prime}(\xi)=p_\pi^{\,\prime}(\varphi)$, and $\pi$ is irreducibly acting over $p_\pi^{\,\prime}(\varphi){\mathcal{H}}_\pi$, that is,  $$\pi(M)^{\,\prime\prime}p_\pi^{\,\prime}(\varphi)={\mathsf{B}}\bigl(p_\pi^{\,\prime}(\varphi){\mathcal{H}}_\pi\bigr)$$ is fulfilled. Also, since owing to $\nu\in {\mathcal{S}}^{{\mathsf{ex}}}(M)$ the state $\nu$ cannot be a character state, according to Lemma \ref{char1} then   ${\mathsf{dim}}\,p_\pi^{\,\prime}(\varphi){\mathcal{H}}_\pi>1$ must hold. Accordingly,  each non-zero unit vector of $p_\pi^{\,\prime}(\varphi){\mathcal{H}}_\pi$ has to be cyclic and has to be implementing a pure state which cannot be a character state either. Thus especially, each state $\nu_t$ with $t\in ]0,1[$ has to be a pure but  non-character state, too. By definition of ${\mathcal{S}}^{{\mathsf{ex}}}(M)$ this is the same as asserting that $\nu_t\in {\mathcal{S}}^{{\mathsf{ex}}}(M)$, for any $t\in [0,1]$.
\end{proof}
\begin{lemma}\label{geoextarc1}
Let $M$ be non-commutative, with states  $\varrho,\nu\in {\mathcal{S}}^{{\mathsf{ex}}}(M)$ which are mutually disjoint. Then, there exists exactly one geodesic arc $\gamma \in {\mathcal{C}}^{\nu,\varrho}(M)$. In this case $\gamma$ is obeying $|\gamma|\cap {\mathcal{S}}^{{\mathsf{ex}}}(M)=\{\nu,\varrho\}$.
\end{lemma}
\begin{proof}
Let $\{\pi,{\mathcal{H}}_\pi\}$ be a unital $^*$-representation of $M$ such that $\xi\in {\mathcal{S}}_{\pi,M}(\varrho)$ and $\varphi\in  {\mathcal{S}}_{\pi,M}(\nu)$ exist. Let $N=\pi(M)^{\,\prime\prime}$. By mutual disjointness of $\nu$ and $\varrho$, there is an orthoprojection $z\in N\cap N^{\,\prime}$ such that $z\xi=\xi$ and $z^\perp \varphi=\varphi$. Thus  $$h^\pi_{\xi,\varphi}={\mathsf{0}}$$ and from which $F(M|\nu,\varrho)=0$ follows, by Lemma \ref{bas3}\,\eqref{bas3b}. From this in view of Corollary \ref{subadd}\,\eqref{subadd2}  orthogonality $\nu \perp \varrho$ follows. Now,
for $\mu=\varrho_\nu\in M_+^*$ given in accordance with Lemma \ref{diffreteile}, ${\mathsf{0}}\leq \mu\leq \varrho$ and $\mu\perp \nu$ have to hold. Since $\varrho$ is a pure state, necessarily $\mu=s\varrho$ follows, with $s\in [0,1]$. But since  $\nu\perp \varrho$ holds, only $s=1$ is possible. That is, $\varrho_\nu=\varrho$. Also, by symmetry $\nu_\varrho=\nu$ must hold. Since $\nu$ and $\varrho$ are mutually disjoint, in view of Theorem \ref{einga} from the previous uniqueness of $\gamma\in \mathcal{C}^{\nu,\varrho}(M)$ is obtained. Accordingly, if $\gamma=(\nu_t)$, then by uniqueness
$$[0,1]\ni t\longmapsto\varphi_t=t\xi+\sqrt{1-t^2}\,\varphi$$ is implementing the geodesic arc $\gamma\in {\mathcal{C}}^{\nu,\varrho}(M)$.  Thus $\nu_t=t^2\varrho+(1-t^2)\,\nu$, for $t\in [0,1]$. From this $|\gamma|\cap {\mathcal{S}}^{{\mathsf{ex}}}(M)=\{\nu,\varrho\}$ is obvious.
\end{proof}
\begin{lemma}\label{geoextarc2}
Let $M$ be non-commutative. For states    $\varrho,\nu\in {\mathcal{S}}^{{\mathsf{ex}}}(M)$ with $\varrho\perp \nu$ but which are not mutually disjoint, there exists  $\gamma \in {\mathcal{C}}^{\nu,\varrho}(M)$ with   $|\gamma|\cap {\mathcal{S}}^{{\mathsf{ex}}}(M)=\{\nu,\varrho\}$ as well as such obeying $|\gamma|\subset {\mathcal{S}}^{{\mathsf{ex}}}(M)$. The degeneracy of the set of all geodesic arcs in ${\mathcal{C}}^{\nu,\varrho}(M)$ possessing the latter property corresponds to the cardinality of the real continuum, $\# \{\gamma \in {\mathcal{C}}^{\nu,\varrho}(M): |\gamma|\subset {\mathcal{S}}^{{\mathsf{ex}}}(M)\}=\mathfrak{c}$.
\end{lemma}
\begin{proof} Since both states are mutually orthogonal pure states, we conclude from  Lemma \ref{diffreteile} that $\varrho_\nu=\varrho$ and $\nu_\varrho=\nu$. Hence, since $\nu$ and $\varrho$ are not mutually disjoint, ${\mathcal{C}}^{\nu,\varrho}(M)$ has to be degenerated by Theorem \ref{einga}, that is, $\gamma\in  {\mathcal{C}}^{\nu,\varrho}(M)$ is not uniquely determined.

Let $\{\pi,{\mathcal{H}}_\pi\}$ be any  unital $^*$-representation of $M$ with non-trivial $\pi$-fibres of $\nu$ and $\varrho$, and suppose $\xi\in {\mathcal{S}}_{\pi,M}(\varrho)$ and $\varphi\in  {\mathcal{S}}_{\pi,M}(\nu)$ to be  chosen as to obey  $\langle\xi,\varphi\rangle_\pi=F(M|\nu,\varrho)$. Let $N=\pi(M)^{\,\prime\prime}$. As we know, over $N^{\,\prime}$ this means that   $$h^\pi_{\xi,\varphi}\geq {\mathsf{0}}$$
Note that, according to Corollary \ref{subadd}\,\eqref{subadd2}, orthogonality implies $F(M|\nu,\varrho)=0$.  From the previous $h^\pi_{\xi,\varphi}= {\mathsf{0}}$ follows. Hence $\langle z\xi,y\varphi\rangle_\pi=0$, for all $y,z\in N^{\,\prime}$, that is
\begin{subequations}\label{garc2}
\begin{equation}\label{garc2a}
p_\pi(\xi)\perp p_\pi(\varphi)
\end{equation}
is inferred to hold.
Remark that owing to $p_\pi^{\,\prime}(\varphi){\mathcal H}_\pi=[\pi(M)\varphi]$ and since the state implemented by $\varphi$ is pure, the action of $\pi$ in restriction to $[\pi(M)\varphi]$ has to be an irreducible one. Thus, on the one hand,
\begin{equation}\label{garc2b}
N p_\pi^{\,\prime}(\varphi)={\mathsf{B}}\left( p_\pi^{\,\prime}(\varphi){\mathcal H}_\pi\right)
\end{equation}
has to be fulfilled, with ${\mathsf{dim}}\, p_\pi^{\,\prime}(\varphi){\mathcal H}_\pi > 1$ since owing to $\nu\in {\mathcal{S}}^{{\mathsf{ex}}}(M)$ the state $\nu$ cannot be a character state of $M$.
On the other hand, since  $p_\pi^{\,\prime}(\varphi)[\pi(M)\xi]\subset [\pi(M)\varphi]$  is a closed subspace left invariant under the action of $\pi$,  either $p_\pi^{\,\prime}(\varphi)[\pi(M)\xi]=\{{\mathsf{0}}\}$ or $p_\pi^{\,\prime}(\varphi)[\pi(M)\xi]= [\pi(M)\varphi]$ must hold. Thus either $p_\pi^{\,\prime}(\varphi)\perp p_\pi^{\,\prime}(\xi)$ or $p_\pi^{\,\prime}(\varphi)\leq p_\pi^{\,\prime}(\xi)$ holds. Since $\varrho$ is pure too, the same conclusions remain true if the r\^{o}les of $\varrho, \xi$ and $\nu, \varphi$ are interchanged. Hence, the conclusion is that either $p_\pi^{\,\prime}(\varphi)= p_\pi^{\,\prime}(\xi)$  or $p_\pi^{\,\prime}(\varphi)\perp p_\pi^{\,\prime}(\xi)$ has to hold.
For the moment, let us presume  $p_\pi^{\,\prime}(\varphi)\perp p_\pi^{\,\prime}(\xi)$ were fulfilled along with each possible choice of $\{\pi,\xi,\varphi\}$ as above. This then implied that, to each $\gamma\in  {\mathcal{C}}^{\nu,\varrho}(M)$, there had to exist an implementation
\[
[0,1]\ni t\longmapsto \varphi_t=t\xi+\sqrt{1-t^2}\varphi
\]
relative to $\pi$, for some particular choice of $\{\pi,\xi,\varphi\}$. In view of  $p_\pi^{\,\prime}(\varphi)\perp p_\pi^{\,\prime}(\xi)$ we had $\langle \pi(x)\xi,\varphi\rangle_\pi=0$, for all $x\in M$. Hence, the states $\nu_t$ implemented by $\varphi_t$ read
$$
\nu_t=t^2\varrho+(1-t^2)\,\nu
$$
and obviously were the same irrespective of the special choice of $\{\pi,\xi,\varphi\}$ met. This contradicted the fact mentioned at the beginning of the proof and saying that  ${\mathcal{C}}^{\nu,\varrho}(M)$ has to be  degenerated. Hence,  $\{\pi,\xi,\varphi\}$ always can be chosen as to obey \begin{equation}\label{garc2c}
p_\pi^{\,\prime}(\varphi)= p_\pi^{\,\prime}(\xi)
\end{equation}
\end{subequations}
Start with such case of  $\{\pi,\xi,\varphi\}$ now. In line with this, consider the Hilbert subspace ${\mathcal{K}}=p_\pi^{\,\prime}(\varphi){\mathcal H}_\pi\subset {\mathcal H}_\pi$. Then, in view of \eqref{garc2c} we have $\xi,\varphi\in \mathcal{K}$, with $q_\xi=p_\pi(\xi)p_\pi^{\,\prime}(\varphi)$ and $q_\varphi=p_\pi(\varphi)p_\pi^{\,\prime}(\varphi)$ being the minimal orthoprojections in \eqref{garc2b} projecting onto the rays spanned by $\xi$ and $\varphi$, respectively. Clearly, $q_\xi$ and $q_\varphi$ are mutually orthogonal, by \eqref{garc2a}.
In addition, let $\{\pi_L,{\mathsf{H.S.}}(\mathcal{K})\}$ be the unital $^*$-representation of ${\mathsf{B}}(\mathcal{K})$ arising by the action of elements of ${\mathsf{B}}(\mathcal{K})$ as left multiplication operators onto the Hilbert-Schmidt operators  ${\mathsf{H.S.}}(\mathcal{K})$ over $\mathcal{K}$, with $\langle x, y\rangle_{\mathsf{H.S}}={{\mathsf{tr}}}\, x y^*$ as scalar product, for $x,y\in {\mathsf{H.S.}}(\mathcal{K})$. Note that by
$$\pi_0: M\ni x\longmapsto\pi_L(\pi(x)p_\pi^{\,\prime}(\varphi))$$
another unital $^*$-representation $\{\pi_0,{\mathsf{H.S.}}(\mathcal{K})\}$ of $M$ is given. Then, due to \eqref{garc2b} and by  definition of $\pi_0$ with the help  of $\pi_L$, elements in the commutant  $\pi_0(M)^{\,\prime}$ in their action on ${\mathsf{H.S.}}(\mathcal{K})$ can be realized through right multiplications by bounded linear operators over $\mathcal{K}$.  Accordingly, since $q_\xi\in {\mathcal{S}}_{\pi_0,M}(\varrho)$ and $q_\varphi\in {\mathcal{S}}_{\pi_0,M}(\nu)$ are  fulfilled, we see that $$h^{\pi_0}_{q_\xi,q_\varphi}={\mathsf{0}}$$ holds over $\pi_0(M)^{\,\prime}$. Thus especially, in respect of $\pi_0$, by the family
$$[0,1]\ni t\longmapsto \phi_t=t\,q_\xi+\sqrt{1-t^2}\,q_\varphi\in {\mathsf{H.S.}}(\mathcal{K})$$
states $\nu_t$ of a geodesic arc  $\gamma\in{\mathcal{C}}^{\nu,\varrho}(M)$ are  implemented, which at $x\in M$ read as
\begin{eqnarray*}
\nu_t(x) &=& \langle \pi_0(x)\phi_t,\phi_t\rangle_{\mathsf{H.S}}={{\mathsf{tr}}}\,(\pi(x)p_\pi^{\,\prime}(\varphi)\phi_t) \phi_t^*={{\mathsf{tr}}}\,(\pi(x)p_\pi^{\,\prime}(\varphi))(\phi_t \phi_t^*)\\
&=& t^2\, {{\mathsf{tr}}}\,(\pi(x)p_\pi^{\,\prime}(\varphi)) q_\xi+(1-t^2)\, {{\mathsf{tr}}}\,(\pi(x)p_\pi^{\,\prime}(\varphi)) q_\varphi\\
& =& t^2\,\langle \pi(x)\xi,\xi\rangle_\pi+(1-t^2)\,\langle \pi(x)\varphi,\varphi\rangle_\pi= t^2\varrho(x)+(1-t^2)\,\nu(x)
\end{eqnarray*}
Hence, the geodesic arc $\gamma$ implemented this way is obeying $|\gamma|\cap {\mathcal{S}}^{{\mathsf{ex}}}(M)=\{\nu,\varrho\}$.

More generally, $q_\xi u\in {\mathcal{S}}_{\pi_0,M}(\varrho)$ and  $h^{\pi_0}_{q_\xi u,q_\varphi}={\mathsf{0}}$ hold over $\pi_0(M)^{\,\prime}$, for each unitary $u\in {\mathsf{B}}(\mathcal{K})$. Let us fix a   unitary $w$ such that $w^* q_\xi w=q_\varphi$. Then, for $u=w \exp{\mathsf{i}\alpha}$, with $\alpha\in [0,2\pmb{\pi}[$, we have $u^* q_\xi u=q_\varphi$, too. Accordingly, in contrast to the above now the geodesic arc $\gamma_\alpha\in{\mathcal{C}}^{\nu,\varrho}(M)$ will be considered which relative to $\pi_0$ is implemented by
$$[0,1]\ni t\longmapsto \phi^\alpha_t=t\,\exp{\mathsf{i}\alpha}\, q_\xi w+\sqrt{1-t^2}\,q_\varphi\in {\mathsf{H.S.}}(\mathcal{K})$$
The state $\nu^\alpha_t$ implemented by $\phi^\alpha_t$ with help of the trace ${{\mathsf{tr}}}$ over $\mathcal K$ at $x\in M$  reads
\begin{equation}\label{garc3}
\nu^\alpha_t(x)={{\mathsf{tr}}}\,(\pi(x)p_\pi^{\,\prime}(\varphi))(\phi^\alpha_t {\phi^\alpha_t}^*)
\end{equation}
with the density operator $\phi^\alpha_t {\phi^\alpha_t}^*$ given over $\mathcal{K}$ by
\begin{equation*}
\phi^\alpha_t {\phi^\alpha_t}^*=t^2\,q_\xi +t\sqrt{1-t^2} \left(\exp{\mathsf{i}\alpha}\, q_\varphi w^* q_\xi+ \exp{(\mathsf{-i}\alpha)}\, q_\xi w q_\varphi\right) +(1-t^2)\,q_\varphi
\end{equation*}
In using that $ q_\xi w=w q_\varphi$, $w^* q_\xi= q_\varphi w^*$  and $q_\xi q_\varphi={\mathsf{0}}$, $q_\varphi q_\xi={\mathsf{0}}$ one easily sees that $$(\phi^\alpha_t {\phi^\alpha_t}^*)^2=\phi^\alpha_t {\phi^\alpha_t}^*$$
holds, for each $\alpha$. That is,  $\phi^\alpha_t {\phi^\alpha_t}^*\in {\mathsf{B}}(\mathcal{K})$ is an orthoprojection of rank one, for each $t$ and $\alpha$. By \eqref{garc3} from this in view of \eqref{garc2b} with ${\mathsf{dim}}\,{\mathcal{K}}>1$ it follows that $\nu_t^\alpha$ is a pure but  non-character state. Hence $|\gamma_\alpha|\subset {\mathcal{S}}^{{\mathsf{ex}}}(M)$, for each $\alpha$.

Finally, for $x=q_\varphi w^* q_\xi$  we have $x^*x=q_\xi$ and $x x^*=q_\varphi$. Thus especially
\begin{equation}\label{garc4}
x^* x\not=x x^*
\end{equation}
is fulfilled. On the other hand, in presuming $\phi^\alpha_t {\phi^\alpha_t}^*=\phi^\beta_t {\phi^\beta_t}^*$ for some reals $\alpha, \beta$ obeying $\exp{\mathsf{i}\alpha}\not=\exp{\mathsf{i}\beta}$, we had to conclude that
\[
x^*=\dfrac{\exp{\mathsf{i}\beta}-\exp{\mathsf{i}\alpha}}{\exp{(-\mathsf{i}\alpha)}-\exp{(-\mathsf{i}\beta)}}\,x
\]
which relation however is contradicting \eqref{garc4} manifestly. Thus   $\phi^\alpha_t {\phi^\alpha_t}^*\not=\phi^\beta_t {\phi^\beta_t}^*$  whenever  $\exp{\mathsf{i}\alpha}\not=\exp{\mathsf{i}\beta}$. From this the validity of the last assertion follows.
\end{proof}
\begin{corolla}\label{pureconv}
Let $M$ be a non-commutative, unital ${\mathsf{C}}^*$-algebra. Then, ${\mathcal{S}}^{{\mathsf{ex}}}(M)$ is convex in   respect of the Bures metric if, and only if, mutually disjoint states in ${\mathcal{S}}^{{\mathsf{ex}}}(M)$ do not exist.
\end{corolla}
\begin{proof}
Suppose $\nu,\varrho\in {\mathcal{S}}^{{\mathsf{ex}}}(M)$ with $\nu\not=\varrho$. Under the assumption that mutually disjoint states in ${\mathcal{S}}^{{\mathsf{ex}}}(M)$ do not exist, Lemma \ref{geoextarc} and Lemma \ref{geoextarc2} assure that $\gamma\in {\mathcal{C}}^{\nu,\varrho}(M)$ with $|\gamma|\subset {\mathcal{S}}^{{\mathsf{ex}}}(M)$ exists, in any case. On the other hand, if ${\mathcal{S}}^{{\mathsf{ex}}}(M)$ is supposed to be convex, then according to Definition \ref{geoconvex}\,\eqref{geoconvex1} for each two states  $\nu,\varrho\in {\mathcal{S}}^{{\mathsf{ex}}}(M)$ there is $\gamma\in {\mathcal{C}}^{\nu,\varrho}(M)$ obeying  $|\gamma|\subset {\mathcal{S}}^{{\mathsf{ex}}}(M)$. Hence, the assumption under which the assertion of Lemma \ref{geoextarc1} applies must be violated, for each choice of $\nu,\varrho\in {\mathcal{S}}^{{\mathsf{ex}}}(M)$. Thus, mutually disjoint states cannot occur in ${\mathcal{S}}^{{\mathsf{ex}}}(M)$.
\end{proof}
\begin{remark}\label{fsmetrik}
\begin{enumerate}
\item \label{fsmetrik2}
For $M={\mathsf{B}}({\mathcal{H}})$, with ${\mathsf{dim}}\,{\mathcal{H}}\geq 2$, ${\mathcal{H}}$ separable,  one has
\begin{eqnarray*}
{\mathcal{S}}^{{\mathsf{ex}}}_0({\mathsf{B}}({\mathcal{H}}))&=&\bigl\{\langle(\cdot)\varphi,\varphi\rangle:\, \varphi\in {\mathcal{H}},\,\|\varphi\|=1\bigr\}\\ &=&{\mathsf{ex}}\,{\mathcal{S}}({\mathsf{B}}({\mathcal{H}}))\cap {\mathcal{S}}_0({\mathsf{B}}({\mathcal{H}}))
\end{eqnarray*}
and which in addition under the Bures metric  is a convex subset of   ${\mathcal{S}}_0({\mathsf{B}}({\mathcal{H}}))$, cf.  Definition \ref{geoconvex}\,\eqref{geoconvex1}. More precisely, for $\nu, \varrho\in {\mathcal{S}}^{{\mathsf{ex}}}_0({\mathsf{B}}({\mathcal{H}}))$ consider the set  $${\mathcal{C}}^{\nu,\varrho}({\mathsf{B}}({\mathcal{H}}))^{{\mathsf{ex}}}=\bigl\{\gamma\in {\mathcal{C}}^{\nu,\varrho}({\mathsf{B}}({\mathcal{H}})): |\gamma|\subset {\mathcal{S}}^{{\mathsf{ex}}}_0({\mathsf{B}}({\mathcal{H}}))\bigr\}$$
of those geodesic arcs connecting $\nu$ and $\varrho$ within the set in question. Then
$$
\# \,{\mathcal{C}}^{\nu,\varrho}({\mathsf{B}}({\mathcal{H}}))^{{\mathsf{ex}}}=
\begin{cases}
\ \mathbf{1} & \text{ for $\nu\not\perp\varrho$;}\\ \\
\ \mathbf{c} &  \text{ for $\nu\perp\varrho$.}
\end{cases}
$$
with $\mathbf{c}$ to stand for the cardinality of the real continuum. Especially this proves that  ${\mathcal{S}}^{{\mathsf{ex}}}_0({\mathsf{B}}({\mathcal{H}}))$ is not a geodesically convex subset.
\item \label{fsmetrik3}
Suppose ${\mathcal{H}}$ to be finite dimensional,  let say  ${\mathsf{dim}}\,{\mathcal{H}}=n+1$, with  $n\in {\mathbb{N}}$. Then, in view of \eqref{fsmetrik1} and \eqref{fsmetrik2}
 the Bures metric is distinguished as a  Riemannian metric on the above set which as such has to be unique since it proves equivalent to the Fubini-Study metric \cite{Fubi:04,Stud:05} on the  complex projective space ${\mathbb{C}}\mathsf{P}^n$  of unit rays over ${\mathbb{C}}^{n+1}\simeq {\mathcal{H}}$, with respect to the standard scalar product on ${\mathbb{C}}^{n+1}$.
\end{enumerate}
\end{remark}

\subsubsection{Processing density operators of full support \textup{(}generalities\textup{)}}\label{finsstat}
Throughout  this paragraph  ${\mathcal S}^{\mathsf{faithful}}_0({\mathsf B}({\mathcal H}))$ will be identified with the affinely convex set of all density operators of full support, and then for simplicity instead of using the abbreviation $\sigma_\varrho$ for the  density matrix of full support which is uniquely  corresponding to the faithful normal state $\varrho$, subsequently simply the same abbreviating notation $\varrho$ will be used for both, the normal positive linear form and the density operator, too (this slightly differs from the usage in \ref{beimatrix}). Thus especially, if not said otherwise, problems relating a curve in ${\mathcal S}^{\mathsf{faithful}}_0({\mathsf B}({\mathcal H}))$ will be discussed in terms of the curve
\begin{equation}\label{hcurve}
  \gamma:  I\ni t\longmapsto \varrho_t\in {\mathcal S}^{\mathsf{faithful}}_0({\mathsf B}({\mathcal H}))
\end{equation}
 of the corresponding density operators $\varrho_t$ of full support.
Having in mind the latter, in view of the above and according to \eqref{finsler1a} we arrive at the following example.
\begin{example}\label{sepa}
With the Bures metric on the (affinely) convex set ${\mathcal S}_0^{\mathsf{faithful}}({\mathsf B}({\mathcal H}))$, for $\gamma$ as in \eqref{hcurve} the following is fulfilled, at each $s\in I$\,:
\begin{equation}\label{hcurve2}
       \lim_{t\to s}
\frac{d_B(\varrho_{t},\varrho_s)}{|t-s|}=\|\varrho^{\,\prime}_s\|_{\varrho_s}=\|{{{\mathsf{tr}}} }\,\varrho^{\,\prime}_s(\cdot)\|_{\varrho_s}
    \end{equation}

\end{example}
\begin{remark}\label{implHS}
\begin{enumerate}
\item\label{impHS1}
In particular,  each curve $\gamma$  possessing a differentiable local implementation as Hilbert-Schmidt operators numbers among the class of curves  \eqref{hcurve} admitted in this paragraph. That is, if each  $\varrho_t$ is generated as $\varrho_t=c_tc_t^*$ with the help of a  differentiable implementation map
\begin{equation}\label{hcurve1}
   I\ni t\longmapsto c_t\in {\mathsf{ H.S.}}({\mathcal H})
\end{equation}
 into the space of Hilbert-Schmidt operators. In fact, under this supposition the special representation $\pi$ of ${\mathsf B}({\mathcal H})$ as left multiplication operators on the Hilbert space ${\mathcal H}_\pi={\mathsf{ H.S.}}({\mathcal H})$ considered in paragraph \ref{beimatrix}, see especially eqs.~\eqref{hs0}, amounts to be a $\gamma$-compliant $^*$-representation of \eqref{hcurve}.
 \item\label{impHS2}
 It is only true for  ${{\mathsf{dim}}}\,{\mathcal H}<\infty$ that the class of curves admitted exactly equals the class of all $C^1$-curves evolving in ${\mathcal S}_0^{\mathsf{faithful}}({\mathsf B}({\mathcal H}))$, since then \begin{equation}\label{hcurve1b}
 c_t=\sqrt{\varrho_t}
 \end{equation}
 is differentiable and can be chosen in  \eqref{hcurve1}. A proof of this, together with an example showing that differentiability of \eqref{hcurve1b} within the space ${\mathsf{ H.S.}}({\mathcal H})$ for a $C^1$-curve  \eqref{hcurve} can fail in case of   ${{\mathsf{dim}}}\,{\mathcal H}=\infty$, is enclosed in Appendix \ref{app_d}, see Theorem \ref{enddiffex} and Example \ref{gegendiff}.
\end{enumerate}
\end{remark}
As to the applicability of the example, besides the formula \eqref{tangentnorm} it can be  enlightening to have an evaluation of the expression $\|\varrho^{\,\prime}_s\|_{\varrho_s}$ on the right-hand side of \eqref{hcurve2} and reading more explicitly in terms of the trace class operators $\varrho_s^{\,\prime}$ and $\varrho_s$. To this sake, let us fix $s\in I$, and remind that in view of Example \ref{furthex} by suppositions about $\gamma$ we have
\begin{subequations}\label{finT0}
 \begin{equation}\label{finT}
 f(\cdot)={{{\mathsf{tr}}} }\,\varrho^{\,\prime}_s(\cdot)\in {\mathsf{T}}_{\varrho_s}{\mathcal S}_0^{\mathsf{faithful}}({\mathsf B}({\mathcal H})) ={\mathsf{T}}_{\varrho_s}({\mathsf B}({\mathcal H}))
 \end{equation}
Remind that, provided \eqref{finT} is fulfilled, according to Corollary \ref{tanfolg} the value   $$\|f\|_{\varrho_s}=\|\varrho^{\,\prime}_s\|_{\varrho_s}$$ can be evaluated within each $^*$-representation $\pi_0$ with respect to which $\varrho_s$ is implementable as a vector. Since only normal linear forms will be concerned with in context of $\varrho_s$ and $\varrho_s^{\,\prime}$, this will be the case if for $\pi_0$ the special $^*$-representation $\pi_0=\pi$ arising by left-multiplication action of elements of ${\mathsf B}({\mathcal H})$ on Hilbert-Schmidt operators ${\mathcal H}_\pi={\mathsf{H.S.}}({\mathcal H})$ is considered. For the details see  \eqref{hs0}. By $\pi$ a ${\mathsf W}^*$-isomorphism is established, and ${\mathsf B}({\mathcal H})$ may be identified with the ${\mathsf W}^*$-algebra  $N=\pi({\mathsf B}({\mathcal H}))$ of all bounded linear operators acting by multiplication from the left
 on Hilbert-Schmidt operators over ${\mathcal H}$. For $g\in {\mathsf B}({\mathcal H})^*$ let $g_\pi=g\circ \pi^{-1}$ be the unique linear form in $N^*$ obeying $g=g_\pi\circ\pi$. Then,  for $f$ from above  $f_\pi\in{\mathsf{T}}_{(\varrho_s)_\pi}(N)$
is fulfilled,
 with the vector state $(\varrho_s)_\pi$ over $N$ implemented by $$\sqrt{\varrho_s}\in {\mathcal S}_{\pi,{\mathsf B}({\mathcal H})}(\varrho_s)$$ Accordingly, Corollary \ref{vNT} can be applied with the result that
\[
f_\pi(\cdot)=\bigl\langle(\cdot)\hat{\psi}_0,\sqrt{\varrho_s}\bigr\rangle_{\mathsf {H.S.}}+\bigl\langle(\cdot)\sqrt{\varrho_s},\hat{\psi}_0\bigr\rangle_{\mathsf {H.S.}}
\]
has to be fulfilled over $N$,  with a unique  Hilbert-Schmidt operator  $\hat{\psi}_0$ obeying
$$\hat{\psi}_0\in \bigl[{\mathsf B}({\mathcal H})_{\mathsf h}\sqrt{\varrho_s}\bigr]_{{\mathsf{H.S.}}}$$
 From this, for any $x\in {\mathsf B}({\mathcal H})$ at the argument $\pi(x)\in N$, the conclusion is that
 \begin{equation*}
 f(\cdot)={{{\mathsf{tr}}} }\,\varrho^{\,\prime}_s(\cdot)= {{{\mathsf{tr}}}}\,(\cdot)\hat{\psi}_0 \sqrt{\varrho_s}+ {{{\mathsf{tr}}}}\,(\cdot) \sqrt{\varrho_s}\hat{\psi}_0^*
 \end{equation*}
has to hold, over ${\mathsf B}({\mathcal H})$. This is equivalent to the operator identity
\begin{equation}\label{u9a0}
\varrho^{\,\prime}_s= \hat{\psi}_0 \sqrt{\varrho_s}+  \sqrt{\varrho_s}\,\hat{\psi}_0^*
 \end{equation}
which has to be satisfied with uniquely determined $\hat{\psi}_0\in \bigl[{\mathsf B}({\mathcal H})_{\mathsf h}\sqrt{\varrho_s}\bigr]_{{\mathsf{H.S.}}}$, and in respect of which according to Corollary \ref{vNT} then \eqref{hcurve2} may be extended by
 \begin{equation}\label{u9a}
    \|\varrho^{\,\prime}_s\|_{\varrho_s}=\|\hat{\psi}_0\|_{\mathsf{ H.S.}}=\sqrt{{{{\mathsf{tr}}}}\,\hat{\psi}_0\hat{\psi}_0^*}
 \end{equation}
\end{subequations}
Provided $\hat{\psi}_0$ can be evaluated in terms of $\varrho_s^{\,\prime}$ and $\varrho_s$, this then on the basis of \eqref{u9a} will yield the desired result. This we are going to do now, for ${{\mathsf{dim}}}\,{\mathcal H}=\infty$.

To prepare for the details, let  $\{\varphi_k:k\in{{\mathbb{N}}}\}$ be  a complete orthonormal system of eigenvectors of $\varrho_s$ such that $\varrho_s\,\varphi_k=\lambda_k\varphi_k$, with the completely ordered sequence $$\lambda_1\geq \lambda_2\geq \lambda_3 \geq\cdots \geq \lambda_n\geq \cdots >0$$ of eigenvalues (with each eigenvalue repeated according to its multiplicity), and be $\{e_k\}$ the system of mutually orthogonal one-dimensional orthoprojections with $e_k\,\varphi_k=\varphi_k$,  and summing up to unity, $\sum_k e_k={\mathsf 1}$. Also, for each $n\in {{\mathbb{N}}}$,
$$E_n=\sum_{k\leq n} e_k$$
be the orthoprojection projecting onto the linear subspace spanned by $\{\varphi_1,\varphi_2,\ldots,\varphi_n\}$. Due to the finite dimensionality of $E_n$, and since $E_n\sqrt{\varrho_s}=\sqrt{\varrho_s}E_n$ holds, the fact
\[
\hat{\psi}_0\in \bigl[{\mathsf B}({\mathcal H})_{\mathsf h}\sqrt{\varrho_s}\,\bigr]_{{\mathsf{H.S.}}}
\]
implies that, with uniquely determined hermitian $x_n\in E_n{\mathsf B}({\mathcal H})E_n$,
\begin{subequations}\label{nE0}
\begin{equation}\label{nE}
E_n \hat{\psi}_0 E_n=x_n\sqrt{\varrho_s}
\end{equation}
has to be fulfilled.  Owing to ${\mathsf{l.u.b.}}\, \{E_n\}={\mathsf 1}$ the following relation can be followed:
\begin{equation}\label{nE2}
 \hat{\psi}_0={\mathsf{ H.S.}}-\lim_{n\to\infty} x_n\sqrt{\varrho_s}
\end{equation}
with hermitian finite-rank operators $x_n$ which in view of the operator identity \eqref{u9a0} at each $n\in {{\mathbb{N}}}$ are the unique hermitian solutions in $E_n{\mathsf B}({\mathcal H})E_n$ of
\begin{equation}\label{nE1}
E_n\varrho^{\,\prime}_s E_n=x_n \varrho_s+\varrho_s x_n
\end{equation}
\end{subequations}
According to \cite[\sc{Theorem 4.6}]{Rose:56} and \cite{Hein:51} the solution of \eqref{nE1} reads
\begin{subequations}\label{loes}
\begin{eqnarray}\label{loes1}
  x_n &=& \sum_{j,k\leq n}\frac{e_j\,\varrho^{\,\prime}_s\, e_k}{\lambda_j+\lambda_k} \\
  \label{loes2}
   &=& \int_0^\infty \exp{\bigl(-t\,\varrho_s\bigr)} \,E_n\varrho^{\,\prime}_s E_n\, \exp{\bigl(-t\,\varrho_s\bigr)}\, d\/t
\end{eqnarray}
\end{subequations}
but which formula is subject to easy proof by direct justification, too. Especially, from \eqref{loes1} in view of \eqref{nE2} the following useful  expression for $\hat{\psi}_0$ can be inferred
\begin{equation}\label{filtern2}
  \hat{\psi}_0=\sum_{j,k\in {{\mathbb{N}}}} \frac{\sqrt{\lambda_k}}{\lambda_j+\lambda_k}\,e_j\,\varrho^{\,\prime}_s\,e_k={\mathsf{ H.S.}}-\lim_{n\to\infty} \sum_{j,k\leq n} \frac{\sqrt{\lambda_k}}{\lambda_j+\lambda_k}\,e_j\,\varrho^{\,\prime}_s\,e_k
\end{equation}
with the series convergeing in the Hilbert-Schmidt sense as mentioned. Taking together Example \ref{sepa}, formula  \eqref{hcurve2} and \eqref{u9a}, in view of \eqref{loes} and  \eqref{filtern2} the following results are seen to hold, at each $s\in I$ of an admitted curve \eqref{hcurve}.
\begin{theorem}\label{infend0}
Under the suppositions of the example the following hold:
\begin{subequations}\label{infend}
\begin{eqnarray}\label{infend1}
\biggl(\lim_{t\to s}\frac{d_B(\varrho_{t},\varrho_s)}{|t-s|}\biggr)^2 & =& \lim_{n\to\infty} {{\mathsf{tr}}}\,\varrho_s x_n^2\\
\label{infend2}
&=& \sum_{j,k\in {{\mathbb{N}}}} \frac{\lambda_k}{(\lambda_j+\lambda_k)^2}\,|\langle\varrho^{\,\prime}_s\varphi_k,\varphi_j\rangle|^2
\end{eqnarray}
\end{subequations}
where $x_n$ is the unique hermitian solution \eqref{loes} of \eqref{nE1}.
\end{theorem}
It should be obvious now how to adapt the steps of the derivation of Theorem \ref{infend0}  in favour of  the finite dimensional case, say with  ${{\mathsf{dim}}}\,{\mathcal H}=n<\infty$. But in contrast to the above, for finite dimensions according to Remark \ref{implHS}\,\eqref{impHS2} one even knows that any differentiable curve $\gamma$ of density operators of full support automatically meets all the requirements supposed to be fulfilled for curves of the class figuring in Example \ref{sepa}, and then the formula \eqref{hcurve2} can be satisfied for an arbitrary $C^1$-curve
\begin{equation*}
  \gamma:  I\ni t\longmapsto \varrho_t\in {\mathcal S}^{\mathsf{faithful}}_0({\mathsf B}({\mathcal H}))
\end{equation*}
For such $\gamma$, at each $s\in I$ the formulas of \eqref{infend} then take the form
\begin{subequations}\label{infendx}
\begin{eqnarray}\label{infend10}
\biggl(\lim_{t\to s}\frac{d_B(\varrho_{t},\varrho_s)}{|t-s|}\biggr)^2 & =&  {{\mathsf{tr}}}\,\varrho_s x^2\\
\label{infend20}
&=& \sum_{j,k\leq n} \frac{\lambda_k}{(\lambda_j+\lambda_k)^2}\,|\langle\varrho^{\,\prime}_s\varphi_k,\varphi_j\rangle|^2
\end{eqnarray}
\end{subequations}
with $x$ being the unique hermitian solution of the operator equation
 \begin{equation}\label{opeq}
\varrho^{\,\prime}_s=x\, \varrho_s+\varrho_s x
\end{equation}
within ${\mathsf B}({\mathcal H})$, and with $\varrho^{\,\prime}_s$ being the derivative of $\gamma$ at $t=s$. In formula \eqref{infend20}, $\{\lambda_1\geq \lambda_2\geq\cdots \geq \lambda_n\}$ is the ordered list of eigenvalues of $\varrho_s$ (with each eigenvalue repeated according to its multiplicity), and $\{\varphi_1,\varphi_2,\ldots,\varphi_n\}$ is an eigenbasis of $\varrho_s$.
\begin{remark}\label{uhlfin}
\begin{enumerate}
\item \label{uhlfin1}
 Relating formula \eqref{infend10}, which is to be satisfied under the auxiliary condition \eqref{opeq}, by the finite dimensional variant of Theorem \ref{infend0} especially a famous  result of Uhlmann \cite{Uhlm:93} is reproduced, and which is saying that for any curve $\gamma:\, t\longmapsto \varrho_t$ of density matrices of rank $n$ and allowing for a generation as $\varrho_t=c_t c_t^*$ by a differentiable curve $t\longmapsto c_t$ of $n\times n$-matrices, the square of the length element ${\mathrm d}\/s$ of the Bures metric along $\gamma$ at $\varrho_t$  reads as
\begin{subequations}\label{uhlfin00}
\begin{equation}\label{uhlfin0}
    {{\mathrm d}\/s}^2={{{\mathsf{tr}}}}\, \varrho_t\, {{\mathrm d}\/x}^2
    \end{equation}
with the matrix-valued $1$-form ${\mathrm d}\/x$ and which is uniquely determined through
\begin{equation}\label{uhlfin1a}
{\mathrm d}\/\varrho_t={\mathrm d}\/x\, \varrho_t+\varrho_t\,{\mathrm d}\/x
\end{equation}
\end{subequations}
The additional information coming along  with \eqref{infend10} and mentioned in context of Remark \ref{implHS}\,\eqref{impHS2} is that Uhlmann's result \eqref{uhlfin00} holds for any differentiable curve evolving within the set of density matrices of rank $n$.
\item  \label{uhlfin2}
According to \eqref{infend20} the square of the length element along $\gamma$ at $\varrho_t$ reads
\begin{equation}\label{huebner}
 {{\mathrm d}\/s}^2=\sum_{j,k\leq n} \frac{\lambda_k(t)}{\bigl(\lambda_j(t)+\lambda_k(t)\bigr)^2}\,|\langle {\mathrm d}\/\varrho_t\,\varphi_k(t),\varphi_j(t)\rangle|^2
\end{equation}
with $\{\lambda_k(t)\}$ and $\{\varphi_k(t)\}$ relating to the ordered eigenvalue sequence and corresponding basis of normalized eigenvectors of $\varrho_t$, respectively. Under the original suppositions of Uhlmann  \eqref{huebner} is due to Huebner \cite{Hueb:92,Hueb:92.1,Hueb:93}.
\item  \label{uhlfin3}
Remark that Example \ref{sepa} as well as  the fact mentioned on in Remark \ref{implHS}\,\eqref{impHS2}
extend to all $C^1$-curves evolving within each (affinely) convex subset
\begin{equation*}
    {\mathcal S}_0^\omega({\mathsf B}({\mathcal H}))= \Omega_{{\mathsf B}({\mathcal H})}(\omega) \cap {\mathcal S}_0({\mathsf B}({\mathcal H}))
\end{equation*}
with the stratum generated by an arbitrarily given normal state $\omega$ with support $s(\omega)$ obeying ${{\mathsf{dim}}}\,s(\omega)<\infty$.
\end{enumerate}
\end{remark}
\subsubsection{Processing density operators of full support \textup{(}special cases\textup{)}}\label{geomann}
Let ${\mathcal H}$ be a separable Hilbert space. Relating the subset ${\mathcal S}^{\mathsf{faithful}}_0({\mathsf B}({\mathcal H}))$ of the normal state space over the special ${\mathsf C}^*$-algebra $M={\mathsf B}({\mathcal H})$ the same assumptions as in \ref{finsstat} will be supposed.
We are going to consider the class of $\mu$-leaves as submanifolds of density operators of full support, labelled by $\mu\in {\mathcal S}^{\mathsf{faithful}}_0({\mathsf B}({\mathcal H}))$, see Definition \ref{blatt}.
\begin{theorem}\label{geomannth}
Each $\mu$-leaf ${\mathcal S}_0(\mu)$ is geodesically convex under the Bures metric.
\end{theorem}
\begin{proof}
The assertion follows along with Example \ref{fins1}.
\end{proof}
Considering the $\mu$-leaf under the Bures geometry locally at $\varrho\in {\mathcal S}_0(\mu)$  requires considering all parameterized differentiable curves $$\gamma: I\ni t\longmapsto \varrho_t\in  {\mathcal S}_0(\mu)$$ passing through $\varrho=\varrho_0$ ($I$ be an open interval containing $0$ as inner point) and for which  $\gamma$-compliant implementations exist. Since by definition the curves of this class are special cases of \eqref{hcurve}, each such $\gamma$ is Finslerian and at $\varrho$ has to obey the Finslerian law \eqref{hcurve2}. Infinitesimally, this is equivalent to the fact that the square of the length element ${\mathsf{d}}s$ along $\gamma$ at $\varrho$ can be expressed as
\begin{equation}\label{geomann1}
{\mathsf{d}}s^2=\|{\mathsf{d}}\varrho\|_\varrho^2
\end{equation}
Let $\varrho^{\,\prime}=\varrho_t^{\,\prime}|_{t=0}$. Since $\varrho_t\in {\mathcal S}_0(\mu)$ holds for each $t$, from \eqref{blattdar2} we see that $\varrho^{\,\prime}\Delta P_k=\Delta P_k \,\varrho^{\,\prime}$ holds, for all $k$. Hence, \eqref{geomann1} requires the auxiliary condition
\begin{equation}\label{geomann2}
{\mathsf{d}}\varrho\,\Delta P_k=\Delta P_k\,{\mathsf{d}}\varrho,\ \forall\,k
\end{equation}
to be considered. Also, in view of \eqref{u9a0} we can be assured that the equation
\begin{equation}\label{geomann3}
{\mathsf{d}}\varrho= {{\mathsf{d}}a} \sqrt{\varrho}+\sqrt{\varrho} \,{{\mathsf{d}}a}^*
\end{equation}
is uniquely solvable with Hilbert-Schmidt valued ${{\mathsf{d}}a}$ obeying ${{\mathsf{d}}a}\in \bigl[{\mathsf B}({\mathcal H})_{\mathsf h}\sqrt{\varrho}\bigr]_{{\mathsf{H.S.}}}$.
According to \eqref{u9a} with the help of  ${{\mathsf{d}}a}$ we then have
\begin{equation}\label{geomann4}
\|{\mathsf{d}}\varrho\|_\varrho^2=\|{{\mathsf{d}}a}\|_{\mathsf{ H.S.}}^2={{{\mathsf{tr}}}}\,{{\mathsf{d}}a}\, {{\mathsf{d}}a}^*
\end{equation}
Since \eqref{geomann2} is satisfied for ${\mathsf{d}}\varrho$ at $\varrho\in {\mathcal S}_0(\mu)$, the unique form ${{\mathsf{d}}a}$ can be resolved into an explicit expression reading in terms of ${\mathsf{d}}\varrho$ and $\varrho$. This will be done now.
\begin{definition}\label{u1form0}
A possibly unbounded, self-adjoint linear operator ${\mathsf{d}} x$ on ${\mathcal H}$, with dense domain of definition ${\mathcal D}({\mathsf{d}} x)$, will be referred to as ${\mathsf{d}}\varrho$-coordinated  $1$-form if it proves to be a solution of the (differential) operator equation
\begin{equation}\label{u1form1}
{\mathsf{d}}\varrho={\mathsf{d}} x \,\varrho+\varrho\, {\mathsf{d}} x
\end{equation}
By the claim that ${\mathsf{d}} x$  be a solution of \eqref{u1form1} it will be understood that the identity  ${\mathsf{d}}\varrho\, \psi={\mathsf{d}} x \,\varrho\, \psi+\varrho\, {\mathsf{d}} x \,\psi$ proves to be fulfilled, at any  $\psi\in{\mathcal D}({\mathsf{d}} x)$.
\end{definition}
Under the assumptions of this paragraph the ${\mathsf{d}}\varrho$-coordinated  $1$-form exists and is uniquely determined. This follows since $\varrho\in {\mathcal S}_0(\mu)$ and the condition \eqref{geomann2} together imply  the condition \eqref{ope1a} in respect of $\omega={\mathsf{d}}\varrho$ and $\varrho$ to be fulfilled. But then Lemma \ref{second} may be applied and yields the assertion.  Moreover, according to Theorem \ref{hsope} with the help of the ${\mathsf{d}}\varrho$-coordinated $1$-form we find that ${{\mathsf{d}}a}={\mathsf{d}} x\,\sqrt{\varrho}\in {\mathsf{H.S.}}({\mathcal H})$ is fulfilled. The latter together with \eqref{geomann4} and \eqref{geomann1} fit together as follows.
\begin{theorem}\label{u1form2}
Let $\gamma\subset {\mathcal S}_0(\mu)$ be a curve passing through $\varrho$ and admitting $\gamma$-compliant representations. The Bures-length element ${\mathsf{d}}s$ along $\gamma$ at $\varrho$ satisfies
\begin{equation}\label{u1form2a}
{{\mathsf{d}}s}^2={{\mathsf{tr}}}\,\bigl({\mathsf{d}} x\,\sqrt{\varrho}\bigr)\bigl({\mathsf{d}} x\,\sqrt{\varrho}\bigr)^*=\| {\mathsf{d}} x\,\sqrt{\varrho}\,\|_{{\mathsf{H.S.}}}^2
\end{equation}
with the ${\mathsf{d}}\varrho$-coordinated $1$-form ${\mathsf{d}} x$ and which is  the unique solution of \eqref{u1form1} to the differential tangent form ${\mathsf{d}} \varrho$ at $\varrho$ and pointing along $\gamma$.
\end{theorem}
\begin{remark}\label{endu1}
According to Remark \ref{geomann0} the case of Uhlmann,  which has been discussed extensively in Remark \ref{uhlfin}  \eqref{uhlfin1}, can be seen as a special case of Theorem \ref{u1form2} if the category of all $\mu$-leaves with $ {\mathsf{dim}} {\mathcal H}<\infty$ and corresponding state $\mu$ of equipartition $\mu=(1/{ {\mathsf{dim}} {\mathcal H}})\,{\mathsf 1}$ is considered there. Clearly, in those cases the ${\mathsf{d}}\varrho$-coordinated $1$-form ${\mathsf{d}} x$ is bounded and thus the arguments under the trace in \eqref{u1form2a} can be rearranged such that the formula finally reads as  ${{\mathsf{d}}s}^2={{\mathsf{tr}}}\,\varrho\,{\mathsf{d}} x^2$, see  \eqref{uhlfin0}.
\end{remark}
\newpage
\section*{\protect{\mbox{}\hfill{}}Appendices\protect{\hfill{}\mbox{}}}
\appendix
\renewcommand{\thelemma}{\Alph{section}.\arabic{lemma}}
\renewcommand{\thedefinition}{\Alph{section}.\arabic{defi}}
\renewcommand{\thecorolla}{\Alph{section}.\arabic{corolla}}
\renewcommand{\theequation}{\Alph{section}.\arabic{equation}}
\setcounter{equation}{0}
\section{Decomposing compact operators}\label{app_a0}
Let ${\mathcal{H}}$ be a Hilbert-space, with scalar product $\langle \xi,\eta\rangle$ defined between arguments $\xi,\eta\in {\mathcal{H}}$ (we suppose linearity in $\xi$, and conjugate linearity in $\eta$). Also, in the following  commutation  $x y=y x$  between $x,y\in {\mathsf B}({\mathcal H})$ will be indicated as $x\smile y$.

Suppose $\Lambda\in {\mathsf B}({\mathcal H})$ to be positive and of full support, and be $\psi\not={\mathsf{0}}$  a unit vector in the range ${\mathcal{R}}(\sqrt{\Lambda})$ of $\sqrt{\Lambda}$. Note that ${\mathcal{R}}(\sqrt{\Lambda})$ is dense in ${\mathcal{H}}$ by assumptions on $\Lambda$. Especially, since it is also of full support, as a map $\sqrt{\Lambda}$ is acting injectively from $ {\mathcal{H}}$ onto ${\mathcal{R}}(\sqrt{\Lambda})$. Therefore, there exists uniquely determined $\varphi_0\not={\mathsf{0}}$ such that
\begin{equation}\label{prorelbed}
\psi=\sqrt{\Lambda}\, \varphi_0
\end{equation}
Let $\varphi$ be the normalization of $\varphi_0$, and let $p_\psi$
and $p_\varphi$ be the minimal orthoprojections with $p_\psi \psi=\psi$ and $p_\varphi \varphi=\varphi$. It is easy to see that the following relation holds
\begin{equation}\label{prore1}
p_\psi=\|\varphi_0\|^2 \sqrt{\Lambda}\, p_\varphi \sqrt{\Lambda}
\end{equation}
For real $\lambda$, according to \eqref{prore1} we then have
\begin{equation}\label{prore2}
\Lambda -\lambda p_\psi= \sqrt{\Lambda}\,\bigl({\mathsf{1}}-\lambda\|\varphi_0\|^2 p_\varphi\bigr) \sqrt{\Lambda}
\end{equation}
\begin{lemma}\label{prore0}
Let $\Lambda\in {\mathsf B}({\mathcal H})$, with $\Lambda\geq {\mathsf{0}}$ and support $s(\Lambda)={\mathsf{1}}$. Suppose the unit vector $\psi$ and vector $\varphi_0$ to obey \textrm{\eqref{prorelbed}}, and be $\Lambda_\lambda=\Lambda -\lambda p_\psi$, with $\lambda \in {\mathbb{R}}$, and $\lambda_0=\|\varphi_0\|^{-2}$. Then
\begin{enumerate}
  \item\label{prore0a} $\Lambda_{\lambda\phantom{_0}}$ is positive if, and only if, $\lambda \leq \lambda_0$;
  \item\label{prore0b} $\Lambda_{\lambda\phantom{_0}}$ is of full support for all\phantom{, } $\lambda < \lambda_0$;
  \item\label{prore0aa} $\psi\in {\mathcal{R}}(\sqrt{\Lambda_{\lambda\phantom{_0}}})$ if, and only if, $\lambda < \lambda_0$ \textup{(}provided $\lambda \leq \lambda_0$\textup{)};
  \item\label{prore0c} $\Lambda_{\lambda_0}$ is of full support if, and only if, $\psi\not\in {\mathcal{R}}(\Lambda)$;
  \item\label{prore0d} $\psi\in {\mathcal{R}}(\Lambda)$, with $\psi=\Lambda\xi$, implies ${\mathsf{dim}}\,{\mathsf{ker}}\Lambda_{\lambda_0}=1$, with $\Lambda_{\lambda_0} \xi={\mathsf{0}}$.
\end{enumerate}
Moreover, if $\Lambda$ is of trace-class, ${{\mathsf{tr}}}\,\Lambda <\infty$, for a unit vector $\psi$ the condition $\psi\in {\mathcal{R}}(\sqrt{\Lambda})$ is necessary for assuring that $\lambda>0$ exists such that $\Lambda -\lambda p_\psi\geq {\mathsf{0}}$.
\end{lemma}
\begin{proof}
 Since ${\mathcal{R}}(\sqrt{\Lambda})$ is dense in ${\mathcal{H}}$, the validity of \eqref{prore0a} is a direct consequence of \eqref{prore2}. For $\lambda<\lambda_0$, since  $({\mathsf{1}}-\lambda\|\varphi_0\|^2 p_\varphi)$ is bounded invertible and $\sqrt{\Lambda}$ has full support,  from \eqref{prore0a} and \eqref{prore2} one infers that $\Lambda_\lambda$ is of full support. This is \eqref{prore0b}. Also, for $\lambda<\lambda_0$, due to
 $$
 \sqrt{\Lambda_\lambda}=\bigl|\bigl({\mathsf{1}}-\lambda\|\varphi_0\|^2 p_\varphi\bigr)^\frac{1}{2} \sqrt{\Lambda}\,\bigr|
$$
the polar decomposition of $\bigl({\mathsf{1}}-\lambda\|\varphi_0\|^2 p_\varphi\bigr)^\frac{1}{2} \sqrt{\Lambda}$ reads
$
\bigl({\mathsf{1}}-\lambda\|\varphi_0\|^2 p_\varphi\bigr)^\frac{1}{2} \sqrt{\Lambda}=V_\lambda \sqrt{\Lambda_\lambda}
$
with unitary $V_\lambda$. But then we have $$\sqrt{\Lambda}=\sqrt{\Lambda_\lambda}V_\lambda^*\bigl({\mathsf{1}}-\lambda\|\varphi_0\|^2 p_\varphi\bigr)^{-\frac{1}{2}}$$ By bounded invertibility of $V_\lambda^*\bigl({\mathsf{1}}-\lambda\|\varphi_0\|^2 p_\varphi\bigr)^{-\frac{1}{2}}$ from this ${\mathcal{R}}(\sqrt{\Lambda})={\mathcal{R}}(\sqrt{\Lambda_\lambda})$ follows.
Also, that $\psi\not\in {\mathcal{R}}(\sqrt{\Lambda_{\lambda_0}})$ must be fulfilled, follows since otherwise application of \eqref{prore0a} with $\Lambda_{\lambda_0}$ instead of $\Lambda$ would imply $
 \Lambda_{\lambda_0}\geq \varepsilon p_\psi$ to hold, for some $\varepsilon>0$, and from which
 $$\Lambda=\Lambda_{\lambda_0}+\lambda_0 p_\psi \geq \bigl(\lambda_0+\varepsilon\bigr) p_\psi$$ had to be followed. This contradicts the fact that $\lambda=\lambda_0$ is the largest real with $\Lambda\geq\lambda p_\psi$. Thus $\psi\not\in {\mathcal{R}}(\sqrt{\Lambda_{\lambda_0}})$ is true. In view of the previous \eqref{prore0aa} then can be taken for justified.

 Relating \eqref{prore0c}, suppose that $\psi$ according to \eqref{prorelbed} in addition is obeying $\psi\in {\mathcal{R}}(\Lambda)$. This equivalently means that $\varphi_0=\sqrt{\Lambda} \xi$ with unique $\xi\in {\mathcal{H}}\backslash \{{\mathsf{0}}\}$, and then $\Lambda_{\lambda_0}\xi=\Lambda \xi-\lambda_0\langle \xi,\psi\rangle \psi=(1-\lambda_0\langle \xi,\psi\rangle) \psi$. Note that $\langle \xi,\psi\rangle=\langle \xi,\Lambda\xi\rangle
=\|\sqrt{\Lambda} \xi\|^2=\|\varphi_0\|^2$. Hence, from the previous we infer that  $1=\lambda_0\langle \xi,\psi\rangle$ and $\Lambda_{\lambda_0}\xi={\mathsf{0}}$. That is, $\Lambda_{\lambda_0}$ cannot be of full support, then.
On the other hand, if $\Lambda_{\lambda_0}$ is not of full support, then $\Lambda_{\lambda_0}\eta={\mathsf{0}}$ is fulfilled, for some $\eta\not={\mathsf{0}}$. Hence $\Lambda \eta=\lambda_0\langle \eta,\psi\rangle) \psi$. Since $\Lambda$ is injectively acting, we have $\Lambda \eta\not ={\mathsf{0}}$. The conclusion is that $\lambda_0\langle \eta,\psi\rangle\not=0$ and $\psi=\Lambda \xi$, for $\xi= \lambda_0^{-1}\langle \eta,\psi\rangle^{-1} \eta$. That is, $\psi\in {\mathcal{R}}(\Lambda)$. Since by injectivity of $\Lambda$ for given $\psi$ the vector $\xi$ is uniquely determined, in view of the previous all the vectors $\eta\not={\mathsf{0}}$  in the kernel of $\Lambda_{\lambda_0}$ have to be proportional to this unique $\xi$. Thus, provided $\Lambda_{\lambda_0}$ is not of full support, then the kernel has to be one-dimensional, with $\Lambda_{\lambda_0}\xi={\mathsf{0}}$.
In summarizing, $\psi\in {\mathcal{R}}(\Lambda)$ occurs if, and only if, $\Lambda_{\lambda_0}$ is not of full support, which is equivalent to \eqref{prore0c}, and that \eqref{prore0d} is true.
Finally, if $\Lambda$ in addition is of trace class, for a unit vector $\psi$ the condition $\Lambda\geq \lambda p_\psi$, for some $\lambda>0$, as a special case of the Radon-Nikodym theorem \cite{Saka:65, Saka:71} requires that
$$p_\psi=\sqrt{\Lambda}\,z z^*\sqrt{\Lambda}$$
be fulfilled, for some $z\in {\mathsf{B}}({\mathcal{H}})$. Hence, $\psi=\sqrt{\Lambda}\, \varphi_0$, with $\varphi_0=z z^* \sqrt{\Lambda}\psi$.
\end{proof}

\begin{lemma}\label{minidec0}
Let $\Lambda\in {\mathsf{B}}({\mathcal{H}})$ be a positive trace-class operator of full support. There is a sequence $\{\lambda_k\}\subset {\mathbb{R}}_+\backslash \{0\}$ and sequence $\{\psi_k\}_{k\in {\mathbb{N}}}\subset {\mathcal{H}}$ of unit vectors such that
\begin{equation}\label{minidec0a}
\Lambda=\sum_{k=1}^\infty \lambda_k p_{\psi_k},\text{ with }\bigvee_{j\not=k} p_{\psi_j}\not={\mathsf{1}},\, \forall\, k\in {\mathbb{N}}
\end{equation}
if, and only if,
\begin{enumerate}
  \item\label{minidec01}
  $\psi_k\in {\mathcal{R}}(\Lambda)$,\ $\forall\,k\in {\mathbb{N}}$;
  \item\label{minidec02}
  $\{\varphi_{0k}\}$, with $\psi_k=\sqrt{\Lambda}\,\varphi_{0k}$, is a maximal system of mutually orthogonal vectors.
\end{enumerate}
Moreover, if these conditions hold, then $\lambda_k=\|\varphi_{0k}\|^{-2}$ is fulfilled, for all $k\in {\mathbb{N}}$.
\end{lemma}
\begin{proof}
Assume \eqref{minidec0a}. Then, since $\Lambda-\lambda_k p_{\psi_k}\geq {\mathsf{0}}$ holds,  necessarily $\psi_k\in {\mathcal{R}}(\sqrt{\Lambda})$ is fulfilled, by Lemma \ref{prore0}. Thus, since $\Lambda$ is of full support, the vectors $\varphi_{0k}$ featuring in \eqref{minidec02} uniquely exist, and if $\varphi_k=\varphi_{0k}/\|\varphi_{0k}\|$, then according to \eqref{prore1}
$$p_{\psi_k}=\|\varphi_{0k}\|^2 \sqrt{\Lambda}\, p_{\varphi_k} \sqrt{\Lambda}$$
holds, for all $k$. Hence
$$\Lambda=\sum_{k=1}^\infty \lambda_k p_{\psi_k}=\sqrt{\Lambda} \biggl(\sum_{k=1}^\infty \lambda_k \|\varphi_{0k}\|^2  p_{\varphi_k}\biggr)\sqrt{\Lambda}$$
is obtained. That is,
\begin{equation}\label{1sum}
{\mathsf{1}}=\sum_{k=1}^\infty  \lambda_k \|\varphi_{0k}\|^2 p_{\varphi_k}
\end{equation}
has to be fulfilled.
On the other hand, since by assumption on $\{\psi_j\}$ and \eqref{minidec0a}  the positive operators
$\Lambda-\lambda_k p_{\psi_k}$ cannot be of full support, by Lemma \ref{prore0}\,\eqref{prore0b} necessarily $\lambda_k=\lambda_{0k}=\|\varphi_{0k}\|^{-2}$ follows. But then in view of Lemma \ref{prore0}\,\eqref{prore0c} even $\psi_k\in {\mathcal{R}}(\Lambda)$ has to be fulfilled, that is \eqref{minidec01} is seen, and \eqref{1sum} finally amounts to ${\mathsf{1}}=\sum_{k=1}^\infty  p_{\varphi_k}$. From this for each $j$ the relation  $1=\sum_{k=1}^\infty |\langle \varphi_k,\varphi_j\rangle |^2$ follows. Hence, $0=\sum_{k\not=j} |\langle \varphi_k,\varphi_j\rangle |^2$, for each $j\in {\mathbb{N}}$. That is, the normalized system $\{\varphi_k\}$ is an orthonormal system of vectors. Hence, $p_{\varphi_k}$ are mutually orthogonal minimal orthoprojections summing up to unity, which means that $\{\varphi_k\}$ is even complete in ${\mathcal{H}}$. This is the same as \eqref{minidec2}.

Now, suppose \eqref{minidec1} and \eqref{minidec2} are fulfilled, with positive trace-class operator $\Lambda$  of full support. Thus in view of \eqref{minidec1} and owing to ${\mathcal{R}}(\Lambda)\subset {\mathcal{R}}(\sqrt{\Lambda})$  the settings around $\varphi_{0k}$ in  \eqref{minidec2} make sense. Let $\varphi_k=\varphi_{0k}/\|\varphi_{0k}\|$, for all $k$. From \eqref{minidec1} and \eqref{minidec2} for all $k$ we get $$p_{\psi_k}=\|\varphi_{0k}\|^2 \sqrt{\Lambda}\, p_{\varphi_k} \sqrt{\Lambda}$$ On the other hand, since according to  \eqref{minidec2}  $\{\varphi_k\}$ is a complete orthonormal system in ${\mathcal{H}}$, then ${\mathsf{1}}=\sum_j p_{\varphi_j}$ holds. Taking together these relations is showing that
\begin{equation}\label{2sum}
\sum_{j=1}^\infty \|\varphi_{0j}\|^{-2} p_{\psi_j}=\sqrt{\Lambda}\biggl(\sum_j p_{\varphi_j}\biggr)\sqrt{\Lambda}=\Lambda
\end{equation}
Thus, by Lemma \ref{prore0}, $\lambda_{0k}=\|\varphi_{0k}\|^{-2}$ is the maximal $\lambda$ such that $\Lambda -\lambda p_{\psi_k}\geq {\mathsf{0}}$. But then \eqref{minidec1} and Lemma \ref{prore0}\,\eqref{prore0c} together imply that $\Lambda -\|\varphi_{0k}\|^{-2} p_{\psi_k}$ cannot be of full support. By \eqref{2sum}
the latter support equals $\bigvee_{j\not=k} p_{\psi_j}$. Hence $\bigvee_{j\not=k} p_{\psi_j}\not={\mathsf{1}}$, for each $k\in {\mathbb{N}}$. Thus, \eqref{minidec1} and \eqref{minidec2} together imply that \eqref{minidec0a} holds, with $\lambda_k=\|\varphi_{0k}\|^{-2}$, for each $k$.
\end{proof}
\begin{lemma}\label{hilfcomm}
Let $\Lambda$ be a positive compact operator of full support over a separable Hilbert space ${\mathcal{H}}$ with ${\mathsf{dim}}\,{\mathcal{H}}=\infty$. There exists a one-dimensional orthoprojection $p$ as follows:
\begin{enumerate}
  \item\label{hilfcomm1}
    Whenever $\{P_k\}$ is a strictly ascending, directed system of orthoprojections of finite rank with ${\mathsf{l.u.b.}} \{P_k\}={\mathsf{1}}$ and $\Lambda \smile P_k$, for all $k\in {\mathbb{N}}$, then there is $j$ with $p\not\smile P_j $.
  \item\label{hilfcomm2}
    There is a real $\beta<0$ such that for all $\lambda\in {{\mathbb{R}}}$ with $\lambda >\beta$ the compact operator $\Lambda+\lambda p$ is positive and has full support.
\end{enumerate}
\end{lemma}
\begin{proof}
The spectral representation
of $\Lambda$ reads $\Lambda=\sum_k \lambda_k q_k$, where the $\lambda_k$'s are the strictly positive different proper values  and the $q_k$'s are the corresponding spectral projections. By assumptions about $\Lambda$, $\sum_k q_k={\mathsf{1}}$ and ${\mathsf{dim}}\,q_k {\mathcal{H}}<\infty$, for each $k$.  Remind that $\Lambda \smile c$ for $c\in {\mathsf{B}}({\mathcal{H}})$ if and only if $q_k\smile c$, for each $k$.   Let $\varphi_0\in {\mathcal{H}}$ be chosen with $q_k\varphi_0\not=0$, for all $k\in {\mathbb{N}}$, and such that
$\psi=\Lambda\,\varphi_0$ is a unit vector. Then, $q_k\psi\not={\mathsf{0}}$, for all $k\in {\mathbb{N}}$.
Let $p_\psi$ be the one-dimensional orthoprojection with $p_\psi \psi=\psi$.
To see \eqref{hilfcomm1}, let $\{P_k\}$ be as specified. We are going to show that $p=p_\psi$ can be taken.  First, note that $\Lambda \smile P_k$, for all $k$, is equivalent to $P_k \smile q_j $, for all $j,k\in {\mathbb{N}}$. Then, $P_k q_j$ is an orthoprojection of finite rank, and since $P_k$ is of finite rank, the set $I_k=\{j\in  {\mathbb{N}}: P_k q_j\not={\mathsf{0}}\}$ is a finite set of subscripts obeying
$ P_k\leq \sum_{j\in I_k} q_j=Q_k$.
Accordingly, $$\textstyle \|P_k \psi\|^2\leq \|Q_k \psi\|^2=\sum_{j\in I_k} \|q_j \psi\|^2<\sum_{j=1}^{\infty} \|q_j \psi\|^2=\|\psi\|^2=1$$ has to hold, for each $k$. On the other hand, ${\mathsf{l.u.b.}} \{P_k\}={\mathsf{1}}$ implies that $\lim_k \|P_k \psi\|^2=1$ must hold. Hence, there exists $j\in {\mathbb{N}}$ such that $0<\|P_j \psi\|^2 < 1$.  Especially this means that $\psi$ cannot be a proper vector of $P_j$. Hence, from $P_j p_\psi \psi=P_j\psi$
and $p_\psi P_j \psi=\langle P_j\psi,\psi\rangle \psi=\|P_j \psi\|^2 \psi $ we infer that
$(p_\psi P_j-P_j p_\psi)\psi=\|P_j \psi\|^2 \psi-P_j\psi\not={\mathsf{0}}$, that is, $p=p_\psi$ can be chosen in \eqref{hilfcomm1}.
Finally, owing to $\psi\in {\mathcal{R}}(\Lambda)$ and since $\Lambda$ is positive with full support, Lemma \ref{prore0} \eqref{prore0a}-\eqref{prore0b} implies
$\Lambda_\lambda=\Lambda-\lambda p_\psi > {\mathsf{0}}$, for all reals $\lambda$ with  $$\lambda<\|\sqrt{\Lambda}\,\varphi_0\|^{-2}=\langle \Lambda\varphi_0,\varphi_0\rangle$$
Hence $\Lambda+\lambda p_\psi$ is positive and of full support, for all reals $\lambda$ with $\lambda > -\langle \Lambda\varphi_0,\varphi_0\rangle^{-1}$.
Thus, from the previous we infer that if  $p=p_\psi$ is chosen, then
$\Lambda+\lambda p > {\mathsf{0}}$ has to be fulfilled for all reals $\lambda> \beta$, with $\beta=-1/\langle \Lambda\,\varphi_0,\varphi_0\rangle$. This proves \eqref{hilfcomm2}.
\end{proof}

\section{Additional remarks on structure of the operator $a_0$}\label{app_a1}
Under the suppositions of paragraph \ref{beimatrix4}, some further information around the structure of the symmetric positive linear operator $a_0$ as defined in \eqref{a0def} will be appended. As known by Lemma \ref{supp5full}, for given density operators of full support $\sigma_\nu$ and $\sigma_\varrho$ over ${\mathcal{H}}$ the hypothesis of the intermediate faithfulness will hold if, and only if, $a_0$ will prove to be essentially self-adjoint. Remind that whereas from  \ref{beimatrix4} and \ref{beimatrix5} a few sufficient conditions on $\sigma_\nu$ and $\sigma_\varrho$ providing essential self-adjoint $a_0$ are known, it is not clear whether or not counterexamples exist with non essential self-adjoint $a_0$. This problem is left open. Add some further insights into the structure of $a_0$ which might prove useful in that context.
Along with $a_0$ consider the linear operator $x_0$ given by
\begin{subequations}\label{a0map0}
\begin{equation}\label{arcbx}
x_0=\bigl|\sqrt{\sigma_\varrho}\sqrt{\sigma_\nu}\,\bigr|^{\frac{1}{2}}\sqrt{\sigma_\nu}^{-1}
\end{equation}
Since $\sqrt{\sigma_\nu}$ as a map is injective, the inverse $\sqrt{\sigma_\nu}^{-1}$ by standard results is densely defined on the range of  $\sqrt{\sigma_\nu}$, and is positive and self-adjoint there. Hence, also $x_0$ has ${\mathcal{D}}(x_0)=\sqrt{\sigma_\nu}{\mathcal{H}}$ as its natural domains of definition, and is acting there as
\begin{equation}\label{a0map1}
  x_0:\ {\mathcal{D}}(x_0)\ni \sqrt{\sigma_\nu}\varphi\longmapsto \,\bigl|\sqrt{\sigma_\varrho}\sqrt{\sigma_\nu}\,\bigr|^{\frac{1}{2}}\,\varphi\phantom{u^*\ }
\end{equation}
\end{subequations}
for each $\varphi\in {\mathcal{H}}$.
The operators $x_0$ of \eqref{a0map0} and $a_0$ are intimately related.
\begin{lemma}\label{a0x0}
The following properties hold:
\begin{enumerate}
  \item \label{a0x0a}
  $x_0$ is closable and obeying ${\mathcal{D}}(x_0^*x_0)={\mathcal{D}}(x_0)$;
  \item \label{a0x0b}
  $a_0=x_0^*x_0$;
  \item \label{a0x0c}
  $a_0$ and $x_0$ are injectively acting over ${\mathcal{D}}(x_0)={\mathcal{D}}(a_0)$.
\end{enumerate}
\end{lemma}
\begin{proof}
Let $\{\sqrt{\sigma_\nu}\varphi_n\}\subset {\mathcal{D}}(x_0)$ with  $\lim_{n\to\infty} \sqrt{\sigma_\nu}\varphi_n={\mathsf{0}}$ and such that $\{x_0\sqrt{\sigma_\nu}\varphi_n\}$ converges,    $\lim_{n\to\infty} x_0\sqrt{\sigma_\nu}\varphi_n=\eta$. The closure $\bar{x}_0$ of $x_0$  exists if, and only if, $\eta={\mathsf{0}}$ in each such situation.
To see this, note that by assumption and \eqref{arcbx}, $$\lim_{n\to\infty} \bigl|\sqrt{\sigma_\varrho}\sqrt{\sigma_\nu}\,\bigr|^{\frac{1}{2}}\varphi_n=
\eta$$
But then
$\lim_{n\to\infty} \bigl|\sqrt{\sigma_\varrho}\sqrt{\sigma_\nu}\,\bigr|\varphi_n= \bigl|\sqrt{\sigma_\varrho}\sqrt{\sigma_\nu}\,\bigr|^{\frac{1}{2}}\eta$.
Especially, from this in view of \eqref{a0x01} we see that
$\bigl|\sqrt{\sigma_\varrho}\sqrt{\sigma_\nu}\,\bigr|\varphi_n=u\sqrt{\sigma_\varrho}\bigl(\sqrt{\sigma_\nu}\varphi_n\bigr)$ holds, for each $n\in {\mathbb{N}}$, and thus $$\bigl|\sqrt{\sigma_\varrho}\sqrt{\sigma_\nu}\,\bigr|^{\frac{1}{2}}\eta=\lim_{n\to\infty} \bigl|\sqrt{\sigma_\varrho}\sqrt{\sigma_\nu}\,\bigr|\varphi_n=u\sqrt{\sigma_\varrho}\lim_{n\to\infty} \sqrt{\sigma_\nu}\varphi_n={\mathsf{0}}$$ follows.
Thus, since as a map $\bigl|\sqrt{\sigma_\varrho}\sqrt{\sigma_\nu}\,\bigr|^{\frac{1}{2}}$ is injective,   $\eta={\mathsf{0}}$ follows.
On the other hand, $\psi=\sqrt{\sigma_\nu}\varphi$ with $\varphi\in {\mathcal{H}}$ by \eqref{a0map1} and \eqref{a0x01} implies  $$\bigl|\sqrt{\sigma_\varrho}\sqrt{\sigma_\nu}\,\bigr|^{\frac{1}{2}} x_0 \psi=\bigl|\sqrt{\sigma_\varrho}\sqrt{\sigma_\nu}\,\bigr| \varphi=\sqrt{\sigma_\nu}\sqrt{\sigma_\varrho}u^*\varphi$$
Hence we have
\begin{equation}\label{a0x01a}
\begin{split}
\bigl|\sqrt{\sigma_\varrho}\sqrt{\sigma_\nu}\,\bigr|^{\frac{1}{2}}{\mathcal{R}}(x_0)
={\mathcal{R}}\bigl(\bigl|\sqrt{\sigma_\varrho}\sqrt{\sigma_\nu}\,\bigr|^{\frac{1}{2}} x_0\bigr)=
{\mathcal{R}}\bigl(\bigl|\sqrt{\sigma_\varrho}\sqrt{\sigma_\nu}\,\bigr|\bigr)
\subset\\
\subset {\mathcal{R}}\bigl(\sqrt{\sigma_\nu}\bigr)={\mathcal{D}}(x_0)
\end{split}
\end{equation}
By standard results, as a consequence of  \eqref{arcbx} we see that
\begin{subequations}\label{x0adjoint}
\begin{eqnarray}
\label{x0adjoint1}
  x_0^* &=& \sqrt{\sigma_\nu}^{-1}\bigl|\sqrt{\sigma_\varrho}\sqrt{\sigma_\nu}\,\bigr|^{\frac{1}{2}} \\
  \label{x0adjoint2}
  {\mathcal{D}}(x_0^*) &=& \bigl\{\varphi\in {\mathcal{H}}: \bigl|\sqrt{\sigma_\varrho}\sqrt{\sigma_\nu}\,\bigr|^{\frac{1}{2}}\varphi\in {\mathcal{D}}(x_0)\bigr\}
\end{eqnarray}
\end{subequations}
Thus, in view of \eqref{x0adjoint2}, \eqref{a0x01a} implies the following inclusion relation to be true
\begin{equation}\label{grund}
{\mathcal{R}}(x_0)\subset {\mathcal{D}}(x_0^*)
\end{equation}
and therefore
${\mathcal{D}}(x_0^*x_0)=\{\varphi\in {\mathcal{D}}(x_0)  : x_0\varphi\in {\mathcal{D}}(x_0^*)\}={\mathcal{D}}(x_0)$ is seen. Thus \eqref{a0x0a} is true.
To see \eqref{a0x0b}, note that since owing to ${\mathcal{D}}(x_0)={\mathcal{R}}(\sqrt{\sigma_\nu})$,  $\psi=\sqrt{\sigma_\nu}\varphi$ is running through all of $ {\mathcal{D}}(x_0)$ if $\varphi$ is running through all of ${\mathcal{H}}$, in making use of  \eqref{grund}, \eqref{a0x01} and \eqref{a0mapneu}  we may conclude as follows
\begin{eqnarray*}
  x_0^*x_0 \psi &=& x_0^*(x_0\psi)=\sqrt{\sigma_\nu}^{-1}\bigl|\sqrt{\sigma_\varrho}\sqrt{\sigma_\nu}\,\bigr|^{\frac{1}{2}}\bigl(\bigl|\sqrt{\sigma_\varrho}\sqrt{\sigma_\nu}\,\bigr|^{\frac{1}{2}}\varphi\bigr)=\sqrt{\sigma_\nu}^{-1}
\bigl(\bigl|\sqrt{\sigma_\varrho}\sqrt{\sigma_\nu}\,\bigr|\varphi\bigr)\\ &=& \sqrt{\sigma_\nu}^{-1}
\bigl(\sqrt{\sigma_\nu}\sqrt{\sigma_\varrho}u^*\varphi\bigr)=\sqrt{\sigma_\varrho}u^*\varphi
   = a_0 \psi
\end{eqnarray*}
for all $\psi\in {\mathcal{D}}(x_0)={\mathcal{D}}(a_0)$. That is, $x_0^*x_0=a_0$ is fulfilled. Thus \eqref{a0x0b} holds.

Finally, to see \eqref{a0x0c}, assume $a_0\xi={\mathsf{0}}$, with $\xi=\sqrt{\sigma_\nu}\varphi$. In view of \eqref{essa3neu} this equivalently means that $\sqrt{\sigma_\varrho}u^*\varphi={\mathsf{0}}$. Since as a map $\sqrt{\sigma_\varrho}$
is acting injectively, $u^*\varphi={\mathsf{0}}$ follows. Thus, by unitarity of $u$,  $\varphi={\mathsf{0}}$. Hence, $\xi={\mathsf{0}}$. Similarly, since as a map also $$\bigl|\sqrt{\sigma_\varrho}\sqrt{\sigma_\nu}\,\bigr|^{\frac{1}{2}}$$ is injective, in view of \eqref{a0map1} we infer that $x_0\xi={\mathsf{0}}$ implies $\varphi={\mathsf{0}}$, and therefore $\xi={\mathsf{0}}$. That is, neither $a_0$ nor $x_0$ admit $0$ as proper value, which is equivalent with injectivity of the mappings \eqref{a0mapneu} and \eqref{a0map1}.
\end{proof}
\noindent Note that $a_0$ as a positive symmetric operator has to be closable, and therefore at least one positive self-adjoint extension of $a_0$ has to exist. A known procedure providing such a self-adjoint extension $a_F\supset a_0$ is by Friedrich's extension, which for $a_0$ reads
\begin{subequations}\label{fried0}
\begin{eqnarray}
\label{fried0a}
  {\mathcal{D}}(a_F) &=& \{\xi\in {\mathcal{D}}(a_0^*): \exists \{\xi_n\}\subset {\mathcal{D}}(a_0), \xi_n\to \xi, \\ \nonumber
  & & \phantom{\{\xi\in {\mathcal{D}}(a_0^*): \exists \{\xi_n\}\subset {\mathcal{D}}(a_0)}\langle a_0 (\xi_n-\xi_m),\xi_n-\xi_m\rangle \to 0\} \\
  \label{fried0b}
  a_F &=& a_0^*| {\mathcal{D}}(a_F)
\end{eqnarray}
\end{subequations}
Note that according to the formula for $a_0$ given by Lemma \ref{a0x0}\,\eqref{a0x0b} and by \eqref{grund} for each Cauchy-sequence $\{\xi_n\}\subset {\mathcal{D}}(a_0)$ with $\lim_{n\to\infty}\xi_n=\xi$ and $\langle a_0 (\xi_n-\xi_m),\xi_n-\xi_m\rangle \to 0$ the latter equivalently means that $\|x_0 \xi_n -x_0\xi_m\|\to 0$, that is, $\{x_0\xi_n\}$ is a Cauchy-sequence. Hence, we have $\xi\in {\mathcal{D}}(\bar{x}_0)$ with $\lim_{n\to\infty} x_0\xi_n=\bar{x}_0\xi$. Therefore, \eqref{fried0a} is the same as
\begin{equation}\label{fried1}
 {\mathcal{D}}(a_F)={\mathcal{D}}(a_0^*)\cap {\mathcal{D}}(\bar{x}_0)
\end{equation}
Now, in view of \eqref{a0x01},  $u\sqrt{\sigma_\varrho}\psi=\bigl|\sqrt{\sigma_\varrho}\sqrt{\sigma_\nu}\,\bigr|^{\frac{1}{2}} x_0 \psi$ holds, for each $\psi\in {\mathcal{D}}(x_0)$. From this for $\xi\in {\mathcal{D}}(\bar{x}_0)$, with $\{\psi_n\}\subset {\mathcal{D}}(x_0)$ and $\psi_n\to \xi$ and $x_0\psi_n\to \bar{x}_0\xi$, the relation
$$u \sqrt{\sigma_\varrho}\xi=\bigl|\sqrt{\sigma_\varrho}\sqrt{\sigma_\nu}\,\bigr|^{\frac{1}{2}} \bar{x}_0 \xi$$
follows. Accordingly, \eqref{essa4bneu} and \eqref{fried1} yield that
\begin{equation}\label{af}
{\mathcal{D}}(a_F)=\bigl\{\xi\in {\mathcal{D}}(\bar{x}_0): \bigl|\sqrt{\sigma_\varrho}\sqrt{\sigma_\nu}\,\bigr|^{\frac{1}{2}} \bar{x}_0 \xi \in {\mathcal{D}}(a_0)  \bigr\}
\end{equation}
On the other hand, since $x_0$ is densely defined and closable, $x_0^*=\bar{x}_0^*$ holds. Accordingly, $x_0^*\bar{x}_0=\bar{x}_0^*\bar{x}_0$ implies that $x_0^*\bar{x}_0$ is a positive self-adjoint linear operator, the domain of definition of which according to \eqref{x0adjoint2} is given by
\begin{equation}\label{af0}
{\mathcal{D}}(x_0^*\bar{x}_0)=\bigl\{\xi\in {\mathcal{D}}(\bar{x}_0): \bigl|\sqrt{\sigma_\varrho}\sqrt{\sigma_\nu}\,\bigr|^{\frac{1}{2}} \bar{x}_0\xi\in {\mathcal{D}}(x_0)\bigr\}
\end{equation}
Thus, in comparing \eqref{af} with \eqref{af0} and keeping in mind that ${\mathcal{D}}(a_0)={\mathcal{D}}(x_0)$ holds, the conclusion is ${\mathcal{D}}(a_F)={\mathcal{D}}(x_0^*\bar{x}_0)$. But then, since from $a_0=x_0^*x_0\subset x_0^*\bar{x}_0$ and self-adjointness of $x_0^*\bar{x}_0$, $a_0^*\supset x_0^*\bar{x}_0$ follows, from the latter in restriction to ${\mathcal{D}}(a_F)$ the relation $a_F=a_0^*|{\mathcal{D}}(a_F)= x_0^*\bar{x}_0$ is obtained. This together with some other information
relating the closures $\bar{x}_0$ and $\bar{a}_0$ is summarized in the following.
\begin{corolla}\label{closa0}
Let $\sigma_\nu, \sigma_\varrho$ be of full rank, and $a_0, x_0$ defined as in \eqref{essa3neu},  \eqref{a0map0}. Then,
\begin{enumerate}
  \item\label{closa03}
   $a_F=x_0^*\bar{x}_0$ is the Friedrich's extension of $a_0$;
  \item\label{closa02}
  $a_0$ is essentially self-adjoint if, and only if, $\bar{a}_0=x_0^*\bar{x}_0$;
  \item\label{closa01}
  $a_0$ is self-adjoint if, and only if, ${\mathcal{D}}(x_0^*\bar{x}_0)={\mathcal{D}}(x_0)$.
\end{enumerate}
The condition ${\mathcal{D}}(x_0^*\bar{x}_0)={\mathcal{D}}(x_0)$ is equivalent to injectivity of the map $\bar{x}_0: {\mathcal{D}}(\bar{x}_0)\rightarrow {\mathcal{R}}(\bar{x}_0)$, with the range obeying ${\mathcal{R}}(\bar{x}_0)\cap {\mathcal{D}}(x_0^*)={\mathcal{R}}(x_0)$.
\end{corolla}
\begin{proof}
The validity of \eqref{closa03} can be taken for granted by the previously given derivation.
Also, if $\bar{a}_0=x_0^*\bar{x}_0$ is supposed, then $\bar{a}_0$ is self-adjoint, that is, $a_0$ is essentially self-adjoint. On the other hand, by \eqref{closa03} $a_F$ is  a positive self-adjoint linear operator which is extending $a_0=x_0^*x_0$, see Lemma \ref{a0x0}\,\eqref{a0x0b}. That is $a_0\subset a_F$ holds. From this $\bar{a}_0\subset a_F$ follows. Hence, if $\bar{a}_0$  is supposed to be self-adjoint, then since each self-adjoint linear operator is maximally symmetric, $\bar{a}_0=a_F$ follows,  that is $\bar{a}_0=x_0^*\bar{x}_0$. Thus \eqref{closa02} is true.

Relating \eqref{closa01}, note that if $a_0$ is self-adjoint, then $a_0$ is closed and thus $a_0=x_0^*\bar{x}_0$ by \eqref{closa02}. Hence, ${\mathcal{D}}(x_0^*\bar{x}_0)={\mathcal{D}}(a_0)={\mathcal{D}}(x_0)$ follows.
On the other hand, if ${\mathcal{D}}(x_0^*\bar{x}_0)={\mathcal{D}}(x_0)$ is supposed, this obviously implies that
$x_0^*x_0=x_0^*\bar{x}_0$ holds. Hence, application of Lemma \ref{a0x0}\,\eqref{a0x0b} yields that $a_0=x_0^*\bar{x}_0=\bar{x}_0^*\bar{x}_0$. Thus $a_0$ is self-adjoint.

Suppose ${\mathcal{D}}(x_0^*\bar{x}_0)={\mathcal{D}}(x_0)$. By \eqref{closa01}, if $\bar{x}_0\xi={\mathsf{0}}$ is supposed to hold, with $\xi\in {\mathcal{D}}(\bar{x}_0)$, then also $\xi\in {\mathcal{D}}(x_0)$ with $a_0\xi={\mathsf{0}}$. Thus, since according to Lemma \ref{a0x0}\,\eqref{a0x0c}, $a_0$ is acting injectively, $\xi={\mathsf{0}}$ must hold.
That is, $\bar{x}_0$ maps injectively.
If as a map $\bar{x}_0$ is injective, then
${\mathcal{R}}(\bar{x}_0) = {\mathcal{R}}(x_0)\cup \{\bar{x}_0\varphi: \varphi\in {\mathcal{D}}(\bar{x}_0)\backslash {\mathcal{D}}(x_0)\}$ and
  ${\mathcal{R}}(x_0)\cap \{\bar{x}_0\varphi: \varphi\in {\mathcal{D}}(\bar{x}_0)\backslash {\mathcal{D}}(x_0)\}=\emptyset$ hold.
Note that $\{\bar{x}_0\varphi: \varphi\in {\mathcal{D}}(\bar{x}_0)\backslash {\mathcal{D}}(x_0)\}\cap {\mathcal{D}}(x_0^*)\not=\emptyset$ would contradict
${\mathcal{D}}(x_0^*\bar{x}_0)={\mathcal{D}}(x_0)$. Hence, the conclusion is ${\mathcal{R}}(\bar{x}_0)\cap {\mathcal{D}}(x_0^*)={\mathcal{R}}(x_0)\cap {\mathcal{D}}(x_0^*)$, from which  by \eqref{grund}
${\mathcal{R}}(\bar{x}_0)\cap {\mathcal{D}}(x_0^*)={\mathcal{R}}(x_0)$ can be followed.
On the other hand, if $\bar{x}_0$ is supposed to be injective with ${\mathcal{R}}(\bar{x}_0)\cap {\mathcal{D}}(x_0^*)={\mathcal{R}}(x_0)$, from the latter in view of the above consequence of injectivity of $\bar{x}_0$ one has to conclude that $\{\bar{x}_0\varphi: \varphi\in {\mathcal{D}}(\bar{x}_0)\backslash {\mathcal{D}}(x_0)\}\cap {\mathcal{D}}(x_0^*)=\emptyset$ must hold. In view of  \eqref{grund} again this is the same as saying that
${\mathcal{D}}(x_0^*\bar{x}_0)=\{\varphi\in {\mathcal{D}}(\bar{x}_0): \bar{x}_0\varphi\in {\mathcal{D}}(x_0^*)\}= \{\varphi\in {\mathcal{D}}(x_0): \bar{x}_0\varphi\in {\mathcal{D}}(x_0^*)\}={\mathcal{D}}(x_0)$.
\end{proof}
In summary, relating essential self-adjointness of $a_0$,  we may state the following.
\begin{corolla}\label{essa}
For faithful normal states $\nu$,\,$\varrho$ the following are  equivalent:
\begin{enumerate}
  \item\label{essa1} $a_0$ is essentially self-adjoint;
  \item\label{essa3} $a_0^*=x_0^*\bar{x}_0$;
  \item\label{essa2} ${\mathcal{D}}(a_0^*)\subset{\mathcal{D}}(\bar{x}_0)$.
\end{enumerate}
\end{corolla}
\begin{proof}
For a proof we may content with showing that the implications \eqref{essa1}\,$\Rightarrow$ \eqref{essa3}, \eqref{essa3}\,$\Rightarrow$ \eqref{essa2}, \eqref{essa2}\,$\Rightarrow$ \eqref{essa1} hold true. Note that
essential self-adjointness of $a_0$ means that $a_0^*$ is self-adjoint. In line with this, \eqref{essa1} implies $a_0^*\supset a_F$, with both $a_0^*$ and $a_F$ being self-adjoint. This can be possible only if equality occurs,
$a_0^*= a_F$. From this in view of Corollary \ref{closa0}\,\eqref{closa03} then \eqref{essa3} follows.
Obviously, the inclusion relation  \eqref{essa2} is a direct consequence of \eqref{essa3}.
But, if  \eqref{essa2} is fulfilled, then ${\mathcal{D}}(a_0^*)={\mathcal{D}}(a_F)$, by  \eqref{fried1}. Since owing to $a_0\subset a_F$ also $a_0^*\supset a_F$ follows, from the latter together with  ${\mathcal{D}}(a_0^*)={\mathcal{D}}(a_F)$ we have to conclude that $a_0^*=a_F$. Thus $a_0^*$ is self-adjoint, which is the same as \eqref{essa1}, again.
\end{proof}
\section{Auxiliary estimates}\label{app_a}
In all what follows, a vector $\vec{\alpha}=(\alpha_1,
\,\ldots\,,\alpha_n)\in {{{\mathbb{R}}}}^n,\ n\in{{\mathbb{N}}}$, is
termed $n$-dimensional probability vector if
$\alpha_k\geq 0$, for all $k\leq n$, and $\sum_{k=1}^n
\alpha_k=1$ holds.
\begin{lemma}\label{vektoren}
 For each two $n$-dimensional probability vectors $\vec{\xi}=(\xi_1,\,\ldots\,,\xi_n)$ and $\vec{\eta}=(\eta_1,\,\ldots\,,\eta_n)$ the estimate
\begin{equation}\label{4.1}
\sum_j \sqrt{\xi_j\,\eta_j}\leq 1- {\frac{1}{8}}\,\sum_{j\in J}\frac{\left(\xi_j-\eta_j\right)^2}{\max\{
\xi_j,\eta_j\}}
\end{equation}
holds, with the subset $J$ given by $J=\{\,j\,:\,\eta_j\not=0\}$.
\end{lemma}
\begin{proof}
Let $\alpha,\,\beta\in  {\mathbb R}$ with $\alpha\geq \beta>0$. Then
\begin{equation}\label{auxi}
\sqrt{\alpha\,\beta}\leq \frac{1}{2}\,\alpha + \frac{1}{2}\,\beta -\frac{1}{8}\,\frac{(\alpha-\beta)^2}{\max\{\alpha,\beta\}}\,.
\end{equation}
In fact, by assumption $$0\leq \frac{\alpha-\beta}{\alpha}<1$$ from which by calculus (consider the power series expansion of $\sqrt{1-x}$ at $x=0$)
$$\sqrt{\beta}=\sqrt{\alpha}\,\sqrt{1-\left(\frac{\alpha-\beta}{\alpha}\right)}= \sqrt{\alpha}\,\left(1-{\frac{1}{2}}\,\left(\frac{\alpha-\beta}{\alpha}\right)-
{\frac{1}{8}}\,\left(\frac{\alpha-\beta}{\alpha}\right)^2-\delta_+\right)$$
follows, with
$
\delta_+={\frac{1}{16}}\,\left(\frac{\alpha-\beta}{\alpha}\right)^3+{\frac{5}{128}}\,\left(\frac{\alpha-\beta}{\alpha}\right)^4
+\,\cdots \ \geq 0
$.
Thus $$\sqrt{\beta}\leq \sqrt{\alpha}\,\left(1-{\frac{1}{2}}\,\left(\frac{\alpha-\beta}{\alpha}\right)-{\frac{1}{8}}\,\left(\frac{\alpha-\beta}{\alpha}\right)^2\right)$$ Multiplying this by $\sqrt{\alpha}$ yields an estimate
$$\sqrt{\alpha\,\beta}\leq \alpha-{\frac{1}{2}}\,\left(\alpha-\beta\right) -{\frac{1}{8}}\,\frac{\left(\alpha-\beta\right)^2}{\alpha}$$ Due to $\alpha=\max\{
\alpha,\beta\}$ and $\alpha-{\frac{1}{2}}\,\left(\alpha-\beta\right)=\frac{1}{2}\,\alpha + \frac{1}{2}\,\beta$ from this finally \eqref{auxi} is seen.

Now, with fixed but arbitrarily chosen $n\in {{\mathbb{N}}}$, let $\vec{\xi}=(\xi_1,\,\ldots\,,\xi_n)$ and $\vec{\eta}=(\eta_1,\,\ldots\,,\eta_n)$ be  $n$-dimensional probability vectors. Define a system $\{I_+, I_-, J_0, J_1\}$ of mutually disjoint subsets $I_+$, $I_-$, $J_0$, $J_1$ of subscripts by $I_+=\{\,j\,:\,\xi_j>\eta_j>0\,\}$, $I_-=\{\,j\,:\,\eta_j\geq\xi_j>0\,\}$,  $J_0=\{\,j\,:\,\xi_j\geq 0, \eta_j=0\,\}$, $J_1=\{\,j\,:\,\xi_j=0, \eta_j\not=0\,\}$. Note that the members of the system $\{I_+, I_-, J_0, J_1\}$ then form a decomposition of the set of all subscripts:
\begin{subequations}\label{unii}
\begin{equation}\label{uni_a}
   \{1,2,\ldots,n\}=I_+\cup \,I_-\cup \,J_0\,\cup\, J_1
\end{equation}
Also, by definition of these subsets and since not all $\eta_j$ can vanish simultaneously, for  $J=\{\,j\,:\,\eta_j\not=0\}$ one concludes that
\begin{equation}\label{uni_b}
   J=I_+\cup \,I_-\cup \,J_1\not=\emptyset\,.
\end{equation}
\end{subequations}
Suppose $J_0\not=\emptyset$. Then, because of $\eta_j=0$ for $j\in J_0$, in such case $\sqrt{\xi_j\,\eta_j}=0$. Hence
\[
\sqrt{\xi_j\,\eta_j}\leq \frac{1}{2}\,\xi_j + \frac{1}{2}\,\eta_j\,,\text{ for }j\in J_0
\]
then is obviously valid. Adding up these estimates in case of $J_0\not=\emptyset$ will yield that
\begin{subequations}\label{in}
\begin{equation}\label{in_a}
\sum_{j\in J_0} \sqrt{\xi_j\,\eta_j}\leq \frac{1}{2}\sum_{j\in J_0}\xi_j + \frac{1}{2}\sum_{j\in J_0}\eta_j\,
\end{equation}

Now, suppose $J_1\not=\emptyset$. Then, due to $\xi_j=0$ for $j\in J_1$, in this case $\sqrt{\xi_j\,\eta_j}=0$ and $\max\{
\xi_j,\eta_j\}=\eta_j$ hold. On the other hand, in this case
\[
\frac{1}{2}\,\xi_j + \frac{1}{2}\,\eta_j-{\frac{1}{8}}\,\frac{\left(\xi_j-\eta_j\right)^2}{\max\{
\xi_j,\eta_j\}}=
\frac{1}{2}\,\eta_j-{\frac{1}{8}}\,\frac{\left(\eta_j\right)^2}{\eta_j}=\frac{3}{8}\,\eta_j\,.
\]
Hence, whenever $J_1\not=\emptyset$ is supposed,
\begin{equation}\label{in_b}
\sqrt{\xi_j\,\eta_j}\leq \frac{1}{2}\,\xi_j + \frac{1}{2}\,\eta_j-{\frac{1}{8}}\,\frac{\left(\xi_j-\eta_j\right)^2}{\max\{
\xi_j,\eta_j\}}\,,\text{ for }j\in J_1
\end{equation}
can be followed. Note that in all cases if $J_1\not=\emptyset$ upon additively combining \eqref{in_a} and the respective inequalities of \eqref{in_b} for $j\in J_0\,\cup\,J_1$  the conclusion is
\begin{equation}\label{in_ab}
\sum_{j\in J_0\,\cup\,J_1}\sqrt{\xi_j\,\eta_j}\leq \frac{1}{2}\sum_{j\in J_0\,\cup\,J_1}\xi_j + \frac{1}{2}\sum_{j\in J_0\,\cup\,J_1}\eta_j-{\frac{1}{8}}\sum_{j\in J_1}\frac{\left(\xi_j-\eta_j\right)^2}{\max\{
\xi_j,\eta_j\}}\,.
\end{equation}
\end{subequations}
Now, suppose $I_+\not=\emptyset$ (resp.~$I_-\not=\emptyset$) and $j\in I_+$ (resp.~$j\in I_-$). Then, in choosing $\alpha=\xi_j$ and $\beta=\eta_j$ (resp.~$\alpha=\eta_j$ and $\beta=\xi_j$) within \eqref{auxi}, the estimate
\begin{equation}\label{in_c}
\sqrt{\xi_j\,\eta_j}\leq \frac{1}{2}\,\xi_j + \frac{1}{2}\,\eta_j-{\frac{1}{8}}\,\frac{\left(\xi_j-\eta_j\right)^2}{\max\{
\xi_j,\eta_j\}}\,,\text{ for }j\in I_+\text{ (resp. $j\in I_-$)}
\end{equation}
is seen to be fulfilled. Hence, upon combining the respective inequalities of \eqref{in_c}, in case of $I_+\cup\,I_-\not=\emptyset$ the conclusion is
\begin{equation}\label{in_d}
\sum_{j\in I_+\cup\,I_- }\sqrt{\xi_j\,\eta_j}\leq \frac{1}{2}\sum_{j\in I_+\cup\,I_- }\xi_j + \frac{1}{2}\sum_{j\in I_+\cup\,I_- }\eta_j-{\frac{1}{8}}\sum_{j\in I_+\cup\,I_- }\frac{\left(\xi_j-\eta_j\right)^2}{\max\{
\xi_j,\eta_j\}}\,.
\end{equation}
For a proof of \eqref{4.1}, suppose $\vec{\xi}$ and $\vec{\eta}$ to be fixed but arbitrarily chosen $n$-dimensional probability vectors. Two alternative cases can occur, $I_+\cup\, I_-=\emptyset$ or $I_+\cup\, I_-\not=\emptyset$.

In case of $I_+\cup\, I_-=\emptyset$, according \eqref{uni_a} and \eqref{uni_b} one has $\{1,2,\ldots, n\}=J_0\,\cup\,J_1$ and $J=J_1\not=\emptyset$. Hence, \eqref{in_ab} in this case reads as
\[
\sum_{j}\sqrt{\xi_j\,\eta_j}\leq \frac{1}{2}\sum_{j}\xi_j + \frac{1}{2}\sum_{j}\eta_j-{\frac{1}{8}}\sum_{j\in J}\frac{\left(\xi_j-\eta_j\right)^2}{\max\{
\xi_j,\eta_j\}}
\]
which owing to $\sum_{j}\xi_j=\sum_{j}\eta_j=1$ is \eqref{4.1}.

Suppose $I_+\cup\, I_-\not=\emptyset$. Then, four alternative subcases can occur: $J_0=J_1=\emptyset$, or $J_0\not=\emptyset$ and $J_1=\emptyset$, or
$J_0=\emptyset$ and $J_1\not=\emptyset$, or $J_0\not=\emptyset$ and $J_1\not=\emptyset$.\\
For $J_0=J_1=\emptyset$ by \eqref{uni_a} and by Definition of $J$ one has $\{1,2,\ldots, n\}=I_+\cup\,I_-=J$.
Hence, \eqref{in_d} applies and then owing to $\sum_{j}\xi_j=\sum_{j}\eta_j=1$ turns into \eqref{4.1}.\\
For $J_0\not=\emptyset$ and $J_1=\emptyset$, one has $J_0=J_0\,\cup\,J_1$ and according to \eqref{uni_a} and \eqref{uni_b}, $\{1,2,\ldots, n\}=I_+\cup\,I_-\cup J_0$ and $J=I_+\cup\,I_-$ have to be fulfilled. Because of $\sum_{j}\xi_j=\sum_{j}\eta_j=1$ the estimate \eqref{4.1} then arises upon additively combining \eqref{in_a} with \eqref{in_d}.
Finally, if the the last two alternatives occur, i.e.~if $I_+\cup\, I_-\not=\emptyset$ and $J_1\not=\emptyset$ are fulfilled, according to \eqref{uni_a} and \eqref{uni_b} the estimate \eqref{4.1} can be followed upon additively combining \eqref{in_ab} with \eqref{in_d} and by taking into account that $\sum_{j}\xi_j=\sum_{j}\eta_j=1$ is fulfilled. Since by $n\in {{\mathbb{N}}}$ and by $\vec{\xi}, \vec{\eta}\in {\mathbb R}^n$ all cases which might be occurring with two probability vectors  have been covered by the previous, \eqref{4.1} now can be taken for verified.
\end{proof}
Subsequently, let $\tau$ vary through the ascendingly directed net of all
finite partitions of the unit interval. That is, each $\tau$ is given as  $\tau=\{t_0,t_1,\ldots,t_n,t_{n+1}\}$,  with $n=n_\tau\in {{\mathbb{N}}}\cup \{0\}$ and  $0=t_0<t_1<\cdots< t_n<t_{n+1}=1$, and if $\tau^\prime$ is another finite partition of $[0,1]$, then the upward (preordering) direction $\tau\leq \tau^\prime$ is to mean that $\tau\subset \tau^\prime$ is fulfilled.
\begin{lemma}\label{kugel}
Let $\Phi: [0,1]\ni t\longmapsto \varphi_t \in {\mathcal H}$ be a continuous map into unit vectors in a Hilbert space  ${\mathcal H}$, with $|\langle\varphi_0,\varphi_1\rangle|<1$.
Then, the following estimate holds:
\begin{equation}\label{ungleich}
\arcsin{\sqrt{1-|\langle\varphi_0,\varphi_1\rangle|^2}}\leq
\lim_\tau {\sum_{j=0}^{n_\tau}
{\|\varphi_{t_j}-\varphi_{t_{j+1}}\|}}
\end{equation}
Thereby, if equality happens in \textup{\eqref{ungleich}}, then $\varphi_t$ for each $t$ is a linear combination of $\varphi_0$ and $\varphi_1$, $\varphi_t = \alpha(t) \varphi_0+\beta(t)\varphi_1$, with the coefficients $\alpha$ and $\beta$ given by
\begin{subequations}\label{xtcoeff}
\begin{equation}\label{xtcoeff1}
\alpha(t)=\frac{\langle\varphi_t,\varphi_0\rangle-\langle\varphi_t,\varphi_1\rangle\langle\varphi_1,\varphi_0\rangle}{1-|\langle\varphi_1,\varphi_0\rangle|^2}
\end{equation}
\begin{equation}\label{xtcoeff2}
\beta(t)=\frac{\langle\varphi_t,\varphi_1\rangle-\langle\varphi_t,\varphi_0\rangle\langle\varphi_0,\varphi_1\rangle}{1-|\langle\varphi_1,\varphi_0\rangle|^2}
\end{equation}
\end{subequations}
\end{lemma}
\begin{proof} Let $P=|\langle\varphi_0,\varphi_1\rangle|^2$. By supposition $0\leq P< 1$. Consider the map
$$
{\gamma}\,:\,[0,1]\ni t\,\longmapsto\,
\vec{r}_t=(\sqrt{1-t^2 (1-P)},t\,\sqrt{1-P},0)\in {{{\mathbb{R}}}}_+^{3}\,.
$$
It is easily seen that $\vec{r}_t\in {\mathbb{S}}^2$, where
${\mathbb{S}}^2$ is the unit
sphere around ${\mathsf 0}$ in ${{{\mathbb{R}}}}^{3}$. If we let $t$ vary monotoneously from $0$ to $1$, the vector
$\vec{r}_t$ runs bijectively through the points of the Euclidean geodesic
connecting $(1,0,0)$ with $(\sqrt{P},\sqrt{1-P},0)$ in the
positive octant of ${\mathbb{S}}^2$. By Euclidian geometry
the Euclidean length of the geodesic $\gamma$ then is
\begin{equation}\label{geos2}
l_2({\gamma})=\arcsin{\sqrt{1-P}}
\end{equation}
Now, let
$\varphi_t=x(t)\varphi_0+y(t)\eta^\prime+\eta_t$,
where $\eta^\prime$  is the unique unit vector
$\eta^\prime\,{\perp}\,\varphi_0$
obeying
\begin{equation}\label{kuzer}
\varphi_1=\langle \varphi_1,\varphi_0\rangle
\varphi_0+\sqrt{1-P}\,\eta^\prime
\end{equation}
and with a continuous
vector valued family $\{\eta_t\}$ such that
$\eta_t\,{\perp}\,\varphi_0$
and $\eta_t\,{\perp}\,\eta^\prime$, for all $t\in [0,1]$. It is clear that such a
decomposition always exists, with $x(t)=\langle\varphi_t,\varphi_0\rangle$ and $y(t)=\langle\varphi_t,\eta^\prime\rangle$. Let us look on the continuous curve ${\gamma}^\prime$
on ${\mathbb{S}}^2$ defined by
$${\gamma}^\prime\,:\,[0,1]\ni t\,\longmapsto\,\vec{r}^{\,\prime}_t=
(|x(t)|,|y(t)|,\|\eta_t\|)\in {\mathbb{S}}^2$$
From \eqref{kuzer} it follows that
${\gamma}^\prime$ is a continuous curve connecting the points
$(1,0,0)$ and
$(\sqrt{P},\sqrt{1-P},0)$ on ${\mathbb{S}}^2$.
Therefore, the Euclidean length of ${\gamma}^\prime$ has to be larger than the
length of the geodesic ${\gamma}$ between the same two points.
This means
\begin{equation}\label{8.0d}
l_2({\gamma}^{\,\prime})\geq l_2({\gamma})\,.
\end{equation}
Note that for any two points $t,s$ by standard estimates and by making use of the Cauchy-Schwarz inequality we have
\begin{eqnarray*}
\|\varphi_t-\varphi_s\|^2 & = &
|x(t)-x(s)|^2+|y(t)-y(s)|^2+\|\eta_t-\eta_s\|^2\\
& \geq &
(|x(t)|-|x(s)|)^2+(|y(t)|-|y(s)|)^2+|\|\eta_t\|-\|\eta_s\||^2\\
& = & \|\vec{r}^{\,\prime}_t-\vec{r}^{\,\prime}_s\|^2\,.
\end{eqnarray*}
From this we arrive at
$$\lim_\tau{\sum_{j=0}^{n_\tau}
{\|\varphi_{t_j}-\varphi_{t_{j+1}}\|}}\geq l_2({\gamma}^{\,\prime})$$
Combining the latter with \eqref{8.0d} and \eqref{geos2} finally yields (\ref{ungleich}). Relating the remaining assertion, note that if $|x(t)|^2+|y(t)|^2<1$ holds for some $0<t<1$, then by continuity this means that $\|\eta_s\|\not=0$ for all $s$ taken from some open interval around $t$. Hence,  the continuous curve $\gamma^\prime$ connecting  $(1,0,0)$ with $(\sqrt{P},\sqrt{1-P},0)$ in ${\mathbb{S}}^2$ will measurably differ from the points of the unique geodesic $\gamma$ connecting  the same two points. Hence, instead of  \eqref{8.0d} then  $l_2({\gamma}^{\,\prime})> l_2({\gamma})$ were to hold, which then implied strict inequality to occur in \eqref{ungleich}. Hence, equality in \eqref{ungleich} requires $\|\eta_t\|=0$, for any $t$. From this
$$\varphi_t = x(t)\varphi_0+y(t)\eta^\prime=\langle\varphi_t,\varphi_0\rangle\varphi_0+\langle\varphi_t,\eta^\prime\rangle\eta^\prime$$ follows. Resolving formula \eqref{kuzer} for $\eta^\prime$ and inserting
\[
\eta^\prime=\frac{\varphi_1-\langle\varphi_1,\varphi_0\rangle\varphi_0}{\sqrt{1-P}}
\]
into the latter equality will result in formulae \eqref{xtcoeff} for the coefficients $\alpha$ and $\beta$.
\end{proof}

\section{Unbounded linear operators affiliated to a $vN$-algebra}\label{app_b}
\noindent
Recall some basic facts from
unbounded operator theory. For these and subsequently
used basic definitions and notions refer to
\cite[V.\,\S\,3.,VI.\,\S\S\,1--3]{Kato:66}.
Let ${\mathcal H}$ be a Hilbert space, and be $R_0$ a densely
defined, symmetric positive linear operator acting from ${\mathcal H}$
into itself, with linear domain of definition ${\mathcal
D}(R_0)
\subset{\mathcal H}$. Recall that these settings mean that for the closure $[{\mathcal D}(R_0)]$ of ${\mathcal D}(R_0)$ the relations
$[{\mathcal D}(R_0)]={\mathcal H}$, $\langle
R_0\varphi,\psi\rangle=\langle
\varphi,R_0\psi\rangle$, and $\langle
R_0\psi,\psi\rangle\geq 0$, for all $\psi,\varphi\in
{\mathcal D}(R_0)$, respectively, have to be fulfilled.
Then, a positive symmetric sesquilinear form ${\mathfrak r}_0$
with the dense form domain ${\mathcal D}({\mathfrak r}_0)=
{\mathcal D}(R_0)$
is given by defining
${\mathfrak r}_0[\varphi,\psi]=\langle R_0\varphi,\psi\rangle$,
for $\varphi,\,\psi\in {\mathcal D}({\mathfrak r}_0)$.
It is known that this form is closable, see
\cite[\sc{Theorem} 1.27.]{Kato:66}.
In line with this, let
${\mathfrak r}=\bar{{\mathfrak r}}_0$ be the closure of this form, and
be
${\mathcal D}({\mathfrak r})$ the form domain of this closure. It is known that a unique
positive self-adjoint extension $R$ of $R_0$ exists the domain ${\mathcal D}(R)$ of which
obeys
${\mathcal D}(R)\subset {\mathcal D}({\mathfrak r})$, see
\cite[\sc{Theorems} 2.1., 2.11.]{Kato:66}.
This unique
extension $R$ of $R_0$ is the so-called {\em{Friedrichs extension}},
and is known to be the unique positive self-adjoint linear operator
obeying ${\mathcal D}(R)\subset {\mathcal D}({\mathfrak r})$ and
${\mathfrak r}[\varphi,\psi]=
\langle R\varphi,\psi\rangle$,
for $\varphi\in{\mathcal D}(R)$ and $\psi\in {\mathcal D}({\mathfrak r})$.
Relating Friedrichs extensions and elements of the set ${\mathsf
U}({\mathcal H})$ of all unitary operators over ${\mathcal H}$, there is the
following useful fact.
\begin{lemma}\label{fried}
Let $R_0,\,T_0$ be densely defined, positive symmetric linear operators on
${\mathcal H}$. Suppose, for some $u\in {\mathsf
U}({\mathcal H})$, $T_0=u^*R_0u$ and ${\mathcal D}(T_0)=u^*{\mathcal D}
(R_0)$ are fulfilled.
Then, the respective Friedrichs extensions $R,\,T$ of $R_0,\,T_0$ obey
$T=u^*Ru$, too.
\end{lemma}
\begin{proof}
Let the forms ${\mathfrak t}_0,{\mathfrak t}$ with respect to
$T_0$ have the
analogous meaning as ${\mathfrak r}_0,{\mathfrak r}$ defined above have with respect
to $R_0$. Then, by assumption one has
$
{\mathfrak t}_0[\varphi,\psi]={\mathfrak r}_0[u\varphi,u\psi]$, for all $\varphi,\psi\in
{\mathcal D}({\mathfrak t}_0)={\mathcal D}(T_0)$. From the structure
of
the operation of taking the closure (smallest closed extension)
of a (closable) form, see \cite[\sc{Theorem} 1.17.]{Kato:66}, together with the assumptions on
the respective domains and because the transforming operator
$u$ between them is unitary, it is easily seen that the previous type of
relation among forms has
to remain valid for the
respective closures and associated form domains too, that is
${\mathcal D}({\mathfrak t})=u^*{\mathcal D}({\mathfrak r})$ and
${\mathfrak t}[\varphi,\psi]={\mathfrak r}[u\varphi,u\psi]$ have to hold, for all
$\varphi,\psi\in{\mathcal D}({\mathfrak t})$. Let $S=u^*Ru$, with
${\mathcal D}(S)=u^*{\mathcal D}(R)$. Also
this operator is positive, self-adjoint and extends $T_0$, $T_0\subset S$. Since $R$ by
assumption is the Friedrichs extension of $R_0$, one has ${\mathcal D}(R)
\subset
{\mathcal D}({\mathfrak r})$. Hence, ${\mathcal D}(S)=u^*{\mathcal D}(R)\subset
u^*{\mathcal D}({\mathfrak r})$. From this in view of the above
${\mathcal D}(S)\subset {\mathcal D}({\mathfrak t})$ can be seen. That is, the positive
self-adjoint operator $S$ extends $T_0$ with domain of definition contained in the
form domain ${\mathcal D}({\mathfrak t})$. By the above mentioned uniqueness of the
Friedrichs extension one therefore has to conclude that $S=T$ holds.
\end{proof}
Let ${  N}$ be a $vN$-algebra acting on the Hilbert space ${\mathcal H}$, with commutant $vN$-algebra ${  N}^{\,\prime}$. Recall the notion of a  closed linear operator affiliated to a $vN$-algebra. In line with this, let $A$ be a closed linear operator on ${\mathcal H}$, with domain of definition ${\mathcal D}(A)$. Let ${  N}^{\,\prime}{\mathcal D}(A)=\{x\psi\,:\,x\in {  N}^{\,\prime},\,\psi\in{\mathcal D}(A) \}$. Then
\begin{definition}\label{affop}
$A$ is said to be
affiliated with ${  N}$, and then the notation $A\mathop{\eta}{  N}$ will be
used, if ${  N}^{\,\prime}{\mathcal D}(A)
\subset
{\mathcal D}(A)$ holds and $xA\subset Ax$ is fulfilled, for each $x\in {  N}^{\,\prime}$.
\end{definition}
In argueing by uniqueness of the spectral decomposition of a
positive self-adjoint linear operator $A$ it is easily seen that in this case
$A\mathop{\eta}{  N}$ holds if, and only if, all the spectral projections
of $A$ are in ${  N}$. In view of this and with the help of the
polar decomposition, which accordingly exists for
unbounded (closed) linear operator $A$, but with $|A|$ being
unbounded, see
\cite[VI.\,\S\,2.,\,7.]{Kato:66}, the general case of a closed linear
operator $A$ affiliated with ${  N}$ might be characterized as follows.
Let $A=U|A|$ be the polar decomposition of the
densely defined closed linear operator $A$,
and be $|A|=\int_0^\infty \lambda\, E(d\/\lambda)$ the spectral decomposition of the module
$|A|$ of $A$. Then $A\mathop{\eta}{  N}$ if, and only if,
$U, E(\lambda)\in {  N}$,
for each $\lambda\in {{{\mathbb{R}}}}_+$. In context of affiliated operators and the set ${\mathcal U}({  N}^{\,\prime})={  N}^{\,\prime}\cap {\mathsf
U}({\mathcal H})$ of the unitary operators of ${  N}^{\,\prime}$ the following
fact often proves useful.
\begin{lemma}\label{affi}
Let $A_0$ be a densely defined, positive
symmetric
linear operator on ${\mathcal H}$. Suppose the domain of definition of $A_0$ is
invariant under the action of each unitary element
$u\in {\mathcal U}({  N}^{\,\prime})$ and $A_0=u^*A_0u$ holds, for all these $u$.
Then, the Friedrichs extension $A$ of $A_0$ is affiliated with ${  N}$.
\end{lemma}
\begin{proof}
Note that since ${\mathcal U}({  N}^{\,\prime})$ is
a group
from the assumption
$u^*{\mathcal D}(A_0)\subset{\mathcal D}(A_0)$ for each $u\in
{\mathcal U}({  N}^{\,\prime})$ especially  $u^*{\mathcal D}(A_0)=
{\mathcal D}(A_0)$ follows. Hence, in defining $T_0=R_0=A_0$ we have
a situation where the premises of Lemma \ref{fried} are fulfilled. The
application of the latter then yields $A=u^*Au$, for each
$u\in {\mathcal U}({  N}^{\,\prime})$. This is the same as
$u^*{\mathcal D}(A)\subset {\mathcal D}(A)$ and $uA\subset Au$, for all
$u\in {\mathcal U}({  N}^{\,\prime})$.
Since the $vN$-algebra
${  N}^{\,\prime}$ can be linearily generated by their unitary elements,
from this in view of the above $x{\mathcal D}(A)\subset {\mathcal D}(A)$ as well as
$xA\subset Ax$ can be inferred, for each $x\in {  N}^{\,\prime}$.
\end{proof}

\section{Hermitian linear forms implemented by vectors\hfill{}}\label{app_c}
\noindent
Let $N$ be a $vN$-algebra acting on the Hilbert space ${\mathcal H}$, with commutant $vN$-algebra $N^{\,\prime}$.
In the following, for given vector $\varphi\in {\mathcal H}$, let $p(\varphi)$ and $p^{\,\prime}(\varphi)$ be the orthoprojections projecting onto the closed linear subspaces $[N^{\,\prime}\varphi]\subset{\mathcal H}$ and  $[N \varphi]\subset{\mathcal H}$, respectively. It is known that $p(\varphi)\in N$ and $p^{\,\prime}(\varphi)\in N^{\,\prime}$. Notice that to each two vectors
$\psi,\varphi\in
{\mathcal H}$ two linear forms $f_{\psi,\varphi}\in N_*$ and
$h_{\psi,\varphi}\in N^{\,\prime}_*$  can be
associated through the settings
\begin{equation}\label{impvec}
f_{\psi,\varphi}(\cdot)=\langle
(\cdot)|_N\psi,\varphi\rangle,\  h_{\psi,\varphi}(\cdot)=\langle
(\cdot)|_{N^{\,\prime}}\psi,\varphi\rangle.
\end{equation}
If positivity (or even weaker: hermiticity) of at least one of the linear
forms of \eqref{impvec} happens, this fact will have interesting consequences for the geometry of the
pair of vectors $\psi,\varphi$ in respect of the
action of $N$ (resp.~$N^{\,\prime}$) on ${\mathcal H}$, and which subsequently
will play a r\^ole. The prototype of such kind of results is in the following.
\begin{lemma}\label{affipos}
The following conditions are mutually equivalent\,\textup{:}
\begin{enumerate}
\item\label{affipos.1}
$h_{\psi,\varphi}\geq 0$ \textup{(}resp.~$f_{\psi,\varphi}\geq 0$\textup{)};
\item\label{affipos.2}
$p(\varphi)\psi=A\varphi$ \textup{(}resp.~$p^{\,\prime}(\varphi)\psi=A\varphi$\textup{)}, with a positive self-adjoint
$A\mathop{\eta}{  N}$ \textup{(}resp.~$A\mathop{\eta}{  N}^{\,\prime}$\textup{)};
\item\label{affipos.3}
$p(\varphi)\psi\in [{  N}_+\varphi]$
\textup{(}resp.~$p^{\,\prime}(\varphi)\psi\in [{  N}^{\,\prime}_+\varphi]$\textup{)}.
\end{enumerate}
The operator $A$ in \eqref{affipos.2} can be chosen to obey $A p(\varphi)^\perp=
p(\varphi)^\perp$ \textup{(}resp.~$A
p^{\,\prime}(\varphi)^\perp=p^{\,\prime}(\varphi)^\perp$\textup{)} and $A
p(\varphi)=p(p(\varphi)\psi)A$ \textup{(}resp.~$A
p^{\,\prime}(\varphi)=p^{\,\prime}(p^{\,\prime}(\varphi)\psi)A$\textup{)}, and to be invertible if
$p(\varphi)=p(\psi)$
\textup{(}resp.~$p^{\,\prime}(\varphi)=p^{\,\prime}(\psi)$\textup{)} is fulfilled.
\end{lemma}
\begin{proof}
By symmetry between $N$ and $N^{\,\prime}$ with respect to the $^{\,\prime}$-operation
and by the definitions of $p(\varphi)$ and $p^{\,\prime}(\varphi)$ and symmetry
of the definitions within \eqref{impvec}, we may be content with
treating the assertions only in case of $h_{\psi,\varphi}$ explicitly.
We are going to verify that the net of implications
(\ref{affipos.2})
$\Rightarrow$(\ref{affipos.3})
$\Rightarrow$(\ref{affipos.1})
$\Rightarrow$(\ref{affipos.2}) is true, in this case. Assume $A$ in accordance with
(\ref{affipos.2}) to be given. Let $A=\int_0^\infty \lambda\, E(d\/\lambda)$ be the
spectral representation of the positive self-adjoint operator $A$.
As has been mentioned in {\sc{Appendix}} \ref{app_b}, since $A$ is affiliated with ${  N}$, each of the
spectral projections of the corresponding spectral decomposition belongs to
${  N}$. Hence, for each $n\in {{\mathbb{N}}}$, the operator
$A_n=\int_0^n \lambda\, E(d\/\lambda)$ obeys $A_n\in{  N}_+$, and
since by assumption $\varphi\in {\mathcal D}(A)$ holds,
$p(\varphi)\psi=A\varphi=\lim_{n\to\infty} A_n\varphi$ follows.
This is (\ref{affipos.3}).
Suppose (\ref{affipos.3}), that is, there is a sequence
$\{A_n\}\subset{  N}_+$ with
$p(\varphi)\psi=A\varphi=\lim_{n\to\infty} A_n\varphi$.
Owing to $p(\varphi),A_n\in{  N}_+$ for each $y\in {  N}^{\,\prime}$
the following conclusion holds\,:
\begin{multline*}
h_{\psi,\varphi}(y^*y)=
\langle y^*y\psi,\varphi\rangle=
\langle y^*yp(\varphi)\psi,\varphi\rangle=
\lim_{n\to\infty}\langle y^*yA_n\varphi,\varphi\rangle\\
=\lim_{n\to\infty}\langle\sqrt{A_n} y^*y\sqrt{A_n}\varphi,\varphi\rangle=
\lim_{n\to\infty}\|y\sqrt{A_n}\varphi\|^2.
\end{multline*}
Hence, $h_{\psi,\varphi}$
is positive, and thus (\ref{affipos.1}) holds.
Suppose (\ref{affipos.1}) is fulfilled. Let
$s(h_{\psi,\varphi})\in
{  N}^{\,\prime}$ be the
support projection of the positive linear form $h_{\psi,\varphi}$. Obviously one
has $s(h_{\psi,\varphi})\leq p^{\,\prime}(\varphi)$ and
$s(h_{\psi,\varphi})\leq p^{\,\prime}(\psi)$.
Let $z= p^{\,\prime}(\psi)-s(h_{\psi,\varphi})$. This is an orthoprojection in
${  N}^{\,\prime}$, which is orthogonal to $s(h_{\psi,\varphi})$, and
thus for each $x,y\in {  N}^{\,\prime}$  the
relation
$\langle xz\psi,y\varphi\rangle=h_{\psi,\varphi}(y^*xz)=
h_{\psi,\varphi}(y^*xzs(h_{\psi,\varphi}))=0$ follows.
From this one concludes that
$$[{  N}^{\,\prime}z\psi]\subset p(\varphi){\mathcal H}^\perp$$ is
fulfilled. Hence, $p(z\psi)\leq p(\varphi)^\perp$. Let a dense linear submanifold
${\mathcal D}$ of ${\mathcal H}$ be defined as
${\mathcal D}={  N}^{\,\prime}\varphi+p(\varphi){\mathcal H}^\perp$.
Firstly, we are
going to show that for $x\in {  N}^{\,\prime}$ and $\delta\in
p(\varphi){\mathcal H}^\perp$ obeying $x\varphi+\delta=0$ the relation
$p(\varphi)x\psi+\delta=0$ is implied. In fact, since
$x\varphi\in p(\varphi){\mathcal H}$ holds $x\varphi+\delta=0$ implies
$\delta=0$ and $x\varphi=0$. The latter implies $xp^{\,\prime}(\varphi)=0$, from which in view of
the above especially $xs(h_{\psi,\varphi})=0$ follows. Accordingly
 $p(\varphi)x\psi=p(\varphi)xp^{\,\prime}(\psi)\psi=
p(\varphi)xs(h_{\psi,\varphi})\psi+p(\varphi)xz\psi=p(\varphi)xz\psi$ holds.
Since in view of above $$xz\psi\in{  N}^{\,\prime}z\psi\subset
p(\varphi){\mathcal H}^\perp$$ has to be fulfilled, from the previous $p(\varphi)x\psi=0$
is obtained. Thus, the implication in question is true.
From the above proven implication one especially concludes
that by the prescription $A_0\,:\,{\mathcal D}\ni x\varphi+\delta\,\longmapsto\,
p(\varphi)x\psi+\delta$, for all $x\in {  N}^{\,\prime}$ and $\delta\in
p(\varphi){\mathcal H}^\perp$, a linear operator acting from the dense domain
${\mathcal D}(A_0)={\mathcal D}$ into ${\mathcal H}$ is given. For
$x,y\in {  N}^{\,\prime}$ and $\delta,\delta'\in
p(\varphi){\mathcal H}^\perp$ by definition of $A_0$ one infers that
$$\langle A_0(x\varphi+\delta),y\varphi+\delta'\rangle=
\langle p(\varphi)x\psi+\delta, y\varphi+\delta'
\rangle=h_{\psi,\varphi}(y^*x)+\langle\delta,\delta'\rangle$$ is fulfilled.
Since by assumption $h_{\psi,\varphi}\geq 0$ holds, the just considered relation
especially yields
$\langle A_0\varPsi,\varPsi\rangle\geq 0$, for each
$\varPsi\in{\mathcal D}(A_0)$. Thus, $A_0$ is seen to be a densely defined,
positive symmetric linear operator.
Note that for each unitary
$u\in {\mathcal U}({  N}^{\,\prime})$ obviously
$u{  N}^{\,\prime}\varphi\subset {  N}^{\,\prime}\varphi$ as well as
$u p(\varphi){\mathcal H}\subset p(\varphi){\mathcal H}$ hold true. From the latter
owing to
normality of $u$ then $u p(\varphi){\mathcal H}^\perp\subset
p(\varphi){\mathcal H}^\perp$ follows. Hence, in view of the definition of
${\mathcal D}$
one infers that $u {\mathcal D}(A_0)\subset {\mathcal D}(A_0)$,
for each $u\in {\mathcal U}({  N}^{\,\prime})$. Clearly, since we have
to do with a group, then even $u {\mathcal D}(A_0)=
{\mathcal D}(A_0)$ has to be fulfilled.
Also, from the just derived and for each $x\in{  N}^{\,\prime}$,
$\delta\in p(\varphi){\mathcal H}^\perp$
and $u\in {\mathcal U}({  N}^{\,\prime})$ one concludes from
$u\delta\in p(\varphi){\mathcal H}^\perp$ that
$
A_0u\{x\varphi+\delta\}=A_0\{ux\varphi+u\delta\}=p(\varphi)ux\psi+
u\delta=u\{p(\varphi)x\psi+\delta\}=uA_0\{x\varphi+\delta\}
$ is fulfilled.
Thus, in summarizing from all that,
the densely defined, positive, symmetric linear operator $A_0$ obeys
$u {\mathcal D}(A_0)=
{\mathcal D}(A_0)$ and $A_0=u^*A_0u$, for each unitary
$u\in {\mathcal U}({  N}^{\,\prime})$. By Definition \ref{affop}
and Lemma \ref{affi} the Friedrichs extension $A$ of $A_0$ exists, and is a
positive self-adjoint
operator which is affiliated with ${  N}$. Since $p(\varphi)\psi=A_0\varphi$
holds by construction, (\ref{affipos.2}) is seen to be fulfilled with the Friedrichs
extension $A$ of $A_0$, which owing to $A_0 p(\varphi)^\perp=p(\varphi)^\perp$ obeys
$A p(\varphi)^\perp=p(\varphi)^\perp$. Note that for this
special $A$ owing to
linearity of ${\mathcal D}(A)$ and
$p(\varphi)^\perp{\mathcal
H}\subset {\mathcal D}(A_0)\subset {\mathcal D}(A)$ for $\eta\in
{\mathcal H}$ one finds $\eta\in {\mathcal D}(A)$ if, and only if,
$p(\varphi)\eta\in {\mathcal D}(A)$ is fulfilled. Thus, in this case
${\mathcal D}(A p(\varphi))={\mathcal D}(A)$ has to hold. Having in
mind this, since $A$ is self-adjoint and affiliated with $N$, for each
$\eta\in {\mathcal D}(A)$ and all $\chi\in
{\mathcal D}(A_0)$, with $\chi=z\varphi+p(\varphi)^\perp\chi$, $z\in
N^{\,\prime}$, one has $\langle A p(\varphi)\eta,\chi\rangle=
\langle p(\varphi)\eta,A_0\chi\rangle=\langle
p(\varphi)\eta,z p(\varphi)\psi +
p(\varphi)^\perp\chi\rangle=\langle\eta,p(\varphi) z\psi\rangle=
\langle\eta,A_0 z\varphi\rangle=\langle\eta,A z\varphi\rangle=\langle
A\eta,z\varphi\rangle=\langle A\eta,p(\varphi)\chi\rangle=\langle
p(\varphi)A\eta,\chi\rangle$. Hence, by denseness of ${\mathcal D}(A_0)$
within ${\mathcal H}$, from this $A p(\varphi)\supset p(\varphi)A$
follows. In view of ${\mathcal D}(A p(\varphi))={\mathcal D}(A)$ then
even $A p(\varphi)=p(\varphi)A$ is seen. Note that from
\eqref{affipos.2} for the range ${\mathcal R}(A p(\varphi))$
especially ${\mathcal R}(A p(\varphi))\subset p(p(\varphi)\psi){\mathcal
H}$ follows. From this and the previous owing to $p(p(\varphi)\psi)\leq
p(\varphi)$ then finally $A p(\varphi)=p(p(\varphi)\psi) A$ is obtained.
The validity of the latter for $p(\psi)=p(\varphi)$
is obvious by standard conclusions since by symmetry of the problem
in $\psi,\varphi$ in this case the range ${\mathcal R}(A_0)$ of $A_0$
is dense within
${\mathcal H}$. Hence, the range ${\mathcal R}(A)$ of $A$ will be dense.
\end{proof}
\begin{corolla}\label{hermform}
For $\varphi,\psi\in {\mathcal H}$, $h_{\psi,\varphi}=h_{\psi,\varphi}^*$
\textup{(}resp.~$f_{\psi,\varphi}=f_{\psi,\varphi}^*$\textup{)} if, and only if,
$$p(\varphi)\psi\in
[N_{\mathrm{h}}\varphi]\ (\text{resp.~}p^{\,\prime}(\varphi)\psi\in
[N^{\,\prime}_{\mathrm{h}}\varphi])$$
in which case
$p(\varphi)\psi=p^{\,\prime}(\varphi)p(\varphi)\psi$ \textup{(}resp.~
$p^{\,\prime}(\varphi)\psi=p(\varphi)p^{\,\prime}(\varphi)\psi$\textup{)} is fulfilled.
\end{corolla}
\begin{proof}
We content with treating the case with
$h_{\psi,\varphi}$ (the other
case can be dealt with analogously).
Let $h=h_{\psi,\varphi}=h_{\psi,\varphi}^*$, and be $q\in N^{\,\prime}$ an
orthoprojection such that both $$h((\cdot)q)\geq 0,\, -h((\cdot)q^\perp)\geq 0$$
are fulfilled on the $vN$-algebra $N^{\,\prime}$ (for $q$ the support projection of the positive part $h_+$ in the canonical decomposition of $h$ can be chosen). Hence, $h_{q\psi,\varphi}\geq 0$ and
$h_{(-q^\perp\psi),\varphi}\geq 0$. Twice applying Lemma \ref{affipos}
\eqref{affipos.3}
then yields
$$p(\varphi)\psi=p(\varphi)(q\psi+q^\perp\psi)\in
[N_+\varphi]-[N_+\varphi]\subset [N_{\mathrm{h}}\varphi]$$
On the other hand, owing to $N_{\mathrm{h}}=N_+-N_+$ for each $a\in N_{\mathrm{h}}$
one has $$h_{a\varphi,\varphi}=\langle(\cdot)|_{N^{\,\prime}}\sqrt{a_+}\varphi,
\sqrt{a_+}\varphi\rangle-\langle(\cdot)|_{N^{\,\prime}}\sqrt{a_-}\varphi,
\sqrt{a_-}\varphi\rangle$$ which obviously is hermitian. Hence, since $p(\varphi)\in N$ holds,
in case of $p(\varphi)\psi\in
[N_{\mathrm{h}}\varphi]$ the linear form $h_{\psi,\varphi}=h_{p(\varphi)\psi,\varphi}$
will be a limit of hermitian linear forms, and therefore has to be hermitian.
Also, if the premises are fulfilled, from the
previous in view of the definition of $p^{\,\prime}(\varphi)$ the validity of
$p(\varphi)\psi=p^{\,\prime}(\varphi)p(\varphi)\psi$ is obvious. This proves the assertion for $h_{\psi,\varphi}$, and analogously one proceeds in case of $f_{\psi,\varphi}$,
with $N$ and $p^{\,\prime}(\varphi)$ instead
of $N^{\,\prime}$ and $p(\varphi)$, respectively.
\end{proof}
\begin{corolla}\label{kreuz0}
$\varphi,\psi\in {\mathcal H},\ h_{\psi,\varphi}=h_{\psi,\varphi}^*\ \ \Longrightarrow\ \
p^{\,\prime}(\varphi)\psi\in [N_{\mathrm{h}}\varphi]\,.$
\end{corolla}
\begin{proof}
By Corollary \ref{hermform} one has $p(\varphi)\psi\in[N_{\mathrm{h}}\varphi]$ and
$p(\varphi)\psi=p(\varphi)p^{\,\prime}(\varphi)\psi$.
Thus we have $p(\varphi)p^{\,\prime}(\varphi)\psi\in
[N_{\mathrm{h}}\varphi]$. Let us consider $p(\varphi)^\perp p^{\,\prime}(\varphi)\psi$.
Then $p^{\,\prime}(p(\varphi)^\perp p^{\,\prime}(\varphi)\psi)\leq p^{\,\prime}(\varphi)$, and therefore
there exists a sequence $\{z_n\}\subset N$ obeying $p(\varphi)^\perp p^{\,\prime}(\varphi)\psi=
\lim_{n\to\infty} z_n\varphi$. Define $y_n=p(\varphi)^\perp z_n+z_n^*p(\varphi)^\perp$,
for each $n\in {{\mathbb{N}}}$. Then $\{y_n\}\subset N_{\mathrm{h}}$, and one has
$$p(\varphi)^\perp p^{\,\prime}(\varphi)\psi=\lim_{n\to\infty} z_n\varphi=\lim_{n\to\infty}
p(\varphi)^\perp z_n\varphi=\lim_{n\to\infty} y_n\varphi$$ Hence,
$p(\varphi)^\perp p^{\,\prime}(\varphi)\psi\in[N_{\mathrm{h}}\varphi]$. Taking together
this with the above yields $p^{\,\prime}(\varphi)\psi=p(\varphi)p^{\,\prime}(\varphi)\psi+
p(\varphi)^\perp p^{\,\prime}(\varphi)\psi\in[N_{\mathrm{h}}\varphi]$.
\end{proof}
\section{Hilbert-Schmidt implementations\hfill{}}\label{app_d}
\noindent
Suppose ${\mathcal H}$ is a separable Hilbert space. Throughout
$\gamma:  I\ni t\longmapsto \varrho_t\in {\mathcal S}_0({\mathsf B}({\mathcal H}))$ be a parameterized $C^1$-curve in the set of density operators over ${\mathcal H}$. Then, the implementation
$
I\ni t\longmapsto \sqrt{\varrho_t}\in {\mathsf{ H.S.}}({\mathcal H})
$
of $\gamma$ will be under consideration.

Note that continuity of $t\mapsto \varrho_t$ as a curve of density operators implies continuity of $t\mapsto \sqrt{\varrho_t}$ in respect of the Hilbert-Schmidt topology. This is due to the following.
\begin{lemma}\label{wurzeleins}
Suppose $\varrho, \sigma\geq {\mathsf 0}$, with ${\mathsf tr}\,\varrho<\infty$ and ${\mathsf tr}\,\sigma<\infty$. Then
\begin{equation}\label{aux00}
  \bigl\|\sqrt{\varrho}-\sqrt{\sigma}\,\bigr\|_{\mathsf{ H.S.}}^2\leq \|\varrho-\sigma\|_1
\end{equation}
\end{lemma}
\begin{proof}
Remark that by supposition both $\varrho$ and $\sigma$ are compact linear operators over ${\mathcal H}$. To be non-trivial suppose $\varrho\not=\sigma$, which by uniqueness of the square root then implies $\sqrt{\varrho}\not=\sqrt{\sigma}$. Also, under the given premises the hermitian linear operator $$C=\sqrt{\varrho}-\sqrt{\sigma}$$ has to be compact, too. Let $P_+$ and $P_-$ be the mutual orthogonal support orthoprojections of the positive and negative parts $C_+\geq {\mathsf 0}$ and $C_-\geq {\mathsf 0}$ of the canonical decomposition $C=C_+-C_-$. In case this makes sense, let $\{\varphi^+_k\}$ and $\{\varphi^-_k\}$ be complete orthonormal systems of eigenvectors within $P_+{\mathcal H}$ and $P_-{\mathcal H}$ of $C_+$ and $C_-$, respectively, with eigenvalue sequences $\{\lambda^+_k\}\subset {\mathbb{R}}_+\backslash \{0\}$ and $\{\lambda^-_k\}\subset {\mathbb{R}}_+\backslash \{0\}$. In case of $C_+\not={\mathsf 0}$, the $\lambda^+_k$'s are the non-zero positive eigenvalues of $C$ , whereas in case of $C_-\not={\mathsf 0}$  the $\lambda^-_k$'s are the moduli of the non-zero negative eigenvalues of $C$ . With the help of the obvious operator identity $\varrho-\sigma=\sqrt{\varrho}\,C+C\sqrt{\sigma}$ and the fact, that in case of $C_+\not={\mathsf 0}$ from $C\,\varphi^+_k=\lambda^+_k\,\varphi^+_k$
\begin{equation*}
    \langle \sqrt{\varrho}\,\varphi^+_k,\varphi^+_k\rangle=\lambda^+_k+\langle \sqrt{\sigma}\,\varphi^+_k,\varphi^+_k\rangle
\end{equation*}
is obtained,
we may conclude as follows:
\begin{eqnarray*}
  {\mathsf{tr}}\,(\varrho-\sigma)P_+ &=& \sum_k \langle (\varrho-\sigma)\,\varphi^+_k,\varphi^+_k\rangle \\
  &=& \sum_k \langle \sqrt{\varrho}\,C\,\varphi^+_k,\varphi^+_k\rangle+\sum_k\langle C\,\sqrt{\sigma}\,\varphi^+_k,\varphi^+_k\rangle\\
   &=& \sum_k \lambda^+_k\Bigl(\langle \sqrt{\varrho}\,\varphi^+_k,\varphi^+_k\rangle+\langle \sqrt{\sigma}\,\varphi^+_k,\varphi^+_k\rangle \Bigr)\\
   &=& \sum_k \lambda^+_k\Bigl(\lambda^+_k+2\langle \sqrt{\sigma}\,\varphi^+_k,\varphi^+_k\rangle \Bigr)\\
   &\geq & \sum_k {\lambda^+_k}^2 ={\mathsf{tr}}\,C_+^2= \|C_+\|_{\mathsf{ H.S.}}^2
\end{eqnarray*}
Note that in case of $C_-={\mathsf 0}$, owing to  $P_- ={\mathsf 0}$ this is the same as $${\mathsf{tr}}\,(\varrho-\sigma)(P_+-P_-)\geq \|C\|_{\mathsf{ H.S.}}^2$$
For $C_-\not={\mathsf 0}$, from the same operator identity and since from $C\,\varphi^-_k=-\lambda^-_k\,\varphi^-_k$
\begin{equation*}
    \langle \sqrt{\sigma}\,\varphi^-_k,\varphi^-_k\rangle=\lambda^-_k+\langle \sqrt{\varrho}\,\varphi^-_k,\varphi^-_k\rangle
\end{equation*}
follows, the conclusion now can be led as follows:
\begin{eqnarray*}
  -{\mathsf{tr}}\,(\varrho-\sigma)P_- &=& -\sum_k \langle (\varrho-\sigma)\,\varphi^-_k,\varphi^-_k\rangle\\
  & = &-\sum_k \langle \sqrt{\varrho}\,C\,\varphi^-_k,\varphi^-_k\rangle-\sum_k \langle C\,\sqrt{\sigma}\,\varphi^-_k,\varphi^-_k\rangle\\
   &=& \sum_k \lambda^-_k\Bigl(\langle \sqrt{\varrho}\,\varphi^-_k,\varphi^-_k\rangle+\langle \sqrt{\sigma}\,\varphi^-_k,\varphi^-_k\rangle \Bigr)\\
   &=& \sum_k \lambda^-_k\Bigl(\lambda^-_k+2\langle \sqrt{\varrho}\,\varphi^+_k,\varphi^+_k\rangle \Bigr)\\
   &\geq & \sum_k {(\lambda^-_k)}^2 ={\mathsf{tr}}\,C_-^2= \|C_-\|_{\mathsf{ H.S.}}^2
\end{eqnarray*}
Note that in case of $C_+={\mathsf 0}$, owing to  $P_+ ={\mathsf 0}$ this is the same as $${\mathsf{tr}}\,(\varrho-\sigma)(P_+-P_-)\geq \|C\|_{\mathsf{ H.S.}}^2$$
In case that both $C_+$ and $C_-$ are non-zero, by taking together both the above derived groups of estimates we can conclude that
\begin{equation*}
    {\mathsf{tr}}\,(\varrho-\sigma)(P_+-P_-)\geq {\mathsf{tr}}\,(C_+^2+C_-^2)={\mathsf{tr}}\,(C_+ -C_-)^2={\mathsf{tr}}\,C^2=\|C\|_{\mathsf{ H.S.}}^2
\end{equation*}
Thus, in view of the above this estimate is seen to be generally valid.
Remark that from $\|\varrho-\sigma\|_1=\sup_{\|Y\|\leq 1} |{\mathsf{tr}}\,(\varrho-\sigma)Y |$ and since both $\varrho-\sigma$ and $X=P_+-P_-$ are hermitian, and since $\|X\|\leq 1$ holds, one can conclude that
$${\mathsf{tr}}\,(\varrho-\sigma)(P_+-P_-)\leq \sup_{\|Y\|\leq 1} |{\mathsf{tr}}\,(\varrho-\sigma)Y |=\|\varrho-\sigma\|_1$$
From this together with the previous finally the validity of \eqref{aux00} follows.
\end{proof}
In the following, an arbitrary parameterized $C^1$-curve  $\gamma:  I\ni t\longmapsto \varrho_t\in {\mathcal S}_0({\mathsf B}({\mathcal H}))$ will be considered around a
given fixed parameter value $s\in I$. The result \eqref{aux00} can be applied with $\varrho=\varrho_t$ and $\sigma=\varrho_s$ and then is showing that $t\mapsto \sqrt{\varrho_t}$ will be continuous with respect to the Hilbert-Schmidt norm. Thus
\begin{equation}\label{trajekt}
I\ni t\longmapsto \sqrt{\varrho_t}\in {\mathsf{ H.S.}}({\mathcal H})
\end{equation}
is a continuous implementation of $\gamma$ in the space of Hilbert-Schmidt operators over ${\mathcal H}$ in the sense used in this paper. Since the Hilbert-Schmidt topology is stronger than the operator norm topology, for each $s\in I$ in a uniform sense one has
\begin{equation}\label{wurzcont}
    \lim_{r\to s} \sqrt{\varrho_r}=\sqrt{\varrho_s}
\end{equation}
Since the parametrization $t\mapsto \varrho_t$ is differentiable, the question arises under which conditions  \eqref{trajekt} gets differentiable. This question will be considered now.

To start with, for $r\in I\backslash \{s\}$, define
\begin{subequations}\label{diffquot}
\begin{eqnarray}\label{diffquot1}
\Delta_r\varrho_s &=& (r-s)^{-1}(\varrho_r-\varrho_s)\\
\label{diffquot2}
\Delta_r\sqrt{\varrho_s} &=& (r-s)^{-1}(\sqrt{\varrho_r}-\sqrt{\varrho_s})
\end{eqnarray}
\end{subequations}
Let $s(\varrho_s)$ be the support orthoprojections of $\varrho_s$, and $\{\varphi_k\}$ be a complete orthonormal system of eigenvectors of $\varrho_s$ within $s(\varrho_s)\,{\mathcal H}$. Also, we let $\{e_k\}$ be the corresponding  system of mutually orthogonal one-dimensional orthoprojections with $e_k\,\varphi_k=\varphi_k$, for all $k$, and which are summing up to $s(\varrho_s)$, $\sum_k e_k=s(\varrho_s)$. Let $\{\lambda_k(s)\}$ be the decreasingly ordered sequence of non-zero eigenvalues of $\varrho_s$, that is  $\varrho_s\,\varphi_k=\lambda_k(s)\varphi_k$ be fulfilled, with each eigenvalue repeated according to its multiplicity. Then, for each set $J$  of subscripts obeying $\# J<\infty$ and $J\subset \{1,2,3,\ldots, {\mathsf {dim}}\,s(\varrho_s)\,{\mathcal H}\}$ we let the orthoprojection $E(J)$ be defined by
$
E(J)=\sum_{k\in J} e_k
$. In the special case of  $J(n)=\{1,2,3,\ldots,n\}$  the abbreviation $E_n=E(J(n))$ will be used.
\begin{lemma}\label{diffex}
 For each finite subset $J\subset {{\mathbb{N}}}$ the map $I\ni t\longmapsto \sqrt{\varrho_t}\,E(J)\in {\mathsf{ H.S.}}({\mathcal H})$ for $r\in I\backslash \{s\}$
 has difference quotient $\Delta_r\sqrt{\varrho_s}\,E(J)$ given by
 \begin{subequations}\label{wurzab}
 \begin{equation}\label{diffq}
 \Delta_r\sqrt{\varrho_s}\,E(J)\phantom{\,\biggl|_{t=s}} = \int_0^\infty \exp(-t\,\sqrt{\varrho_r})\,\Delta_r\varrho_s\,\exp(-t\,\sqrt{\varrho_s})\,E(J)\,d\/t
 \end{equation}
 and is differentiable, with  derivative at $t=s$ reading as
 \begin{equation}\label{diffformu}
     \frac{d}{d\/t}\,\sqrt{\varrho_t}\,E(J)\,\biggl|_{t=s} =\int_0^\infty \exp(-t\,\sqrt{\varrho_s})\,\,\,\,\varrho_s^{\,\prime}\,\,\,\exp(-t\,\sqrt{\varrho_s})\,E(J)\,d\/t
 \end{equation}
 \end{subequations}
\end{lemma}
\begin{proof}
With the expressions defined in \eqref{diffquot} the following identity is satisfied
\begin{equation}\label{opeq1}
 \Delta_r\varrho_s=\sqrt{\varrho_r}\, \bigl(\Delta_r\sqrt{\varrho_s}\bigr)+\bigl(\Delta_r\sqrt{\varrho_s}\bigr)\,\sqrt{\varrho_s}
\end{equation}
Now, if \eqref{opeq1} is multiplied by $e_k$ from the right, then the following is obtained
\begin{eqnarray*}
  \Delta_r\varrho_s\,e_k &=& \sqrt{\varrho_r}\, \bigl(\Delta_r\sqrt{\varrho_s}\,e_k\bigr)+\bigl(\Delta_r\sqrt{\varrho_s}\,e_k\bigr)\,\sqrt{\lambda_k(s)} \\
   &=& \Bigl( \sqrt{\varrho_r}+\sqrt{\lambda_k(s)}\,{\mathsf 1}\Bigr) \bigl(\Delta_r\sqrt{\varrho_s}\,e_k\bigr)
\end{eqnarray*}
Since $\sqrt{\varrho_r}$ is positive semi-definite and $-\sqrt{\lambda_k(s)}$ belongs to the resolvent set of each of $\sqrt{\varrho_r}$, the above relation can be resolved for $\Delta_r\sqrt{\varrho_s}\,e_k$. The result is
\begin{subequations}\label{grunds0}
\begin{eqnarray}\label{grunds0a}
    \Delta_r\sqrt{\varrho_s}\,e_k&=&\Bigl( \sqrt{\varrho_r}+\sqrt{\lambda_k(s)}\,{\mathsf 1}\Bigr)^{-1} \Delta_r \varrho_s\,e_k\\
 \nonumber
   &=& \int_0^\infty \exp(-t\,\bigl(\sqrt{\varrho_r}+\sqrt{\lambda_k(s)}\,{\mathsf 1}\bigr))\,\Delta_r \varrho_s\,e_k \,d\/t\\
    \nonumber
      &=& \int_0^\infty \exp(-t\,\sqrt{\varrho_r})\,\Delta_r \varrho_s \exp(-t\,\sqrt{\lambda_k(s)})\,e_k \,d\/t\\
   \label{grunds0b}
   &=& \int_0^\infty \exp(-t\,\sqrt{\varrho_r})\,\Delta_r \varrho_s\exp(-t\,\sqrt{\varrho_s})\,e_k\,d\/t
\end{eqnarray}
\end{subequations}
where equivalence of \eqref{grunds0a} with \eqref{grunds0b} follows by functional calculus along the given intermediate steps.
Now, by assumption $\lim_{r\to s} \Delta_r\varrho=\varrho_s^{\,\prime}$
exists and the map $$I\ni t \mapsto \sqrt{\varrho_t}$$ according to  \eqref{wurzcont} is continuous. With the help of the latter and since the map $a\mapsto a^{-1}$ is continuous over the invertible elements $a\in {\mathsf B}({\mathcal H})$, from \eqref{grunds0a} we conclude that the limit of $\Delta_r\sqrt{\varrho_s}\,e_k$ as $r\to s$ exists and obviously for each $k$ reads
\begin{subequations}\label{grunds1}
\begin{eqnarray}\label{grunds2}
   \lim_{r\to s}\Delta_r\sqrt{\varrho_s}\,e_k &=& \Bigl( \sqrt{\varrho_s}+\sqrt{\lambda_k(s)}\,{\mathsf 1}\Bigr)^{-1} \varrho_s^{\,\prime}\,e_k \\
   \label{grunds3}
   &=& \int_0^\infty \exp(-t\,\sqrt{\varrho_s})\,\varrho_s^{\,\prime}\,\exp(-t\,\sqrt{\varrho_s})\,e_k\,d\/t
\end{eqnarray}
\end{subequations}
Thereby,  equivalence of \eqref{grunds2} with \eqref{grunds3} follows by a reasoning similar to the one used in context of \eqref{grunds0}.
Finally, in summing up within  \eqref{grunds0b} and \eqref{grunds3} over all $k\in J$ both the formulae of \eqref{wurzab} are obtained.
\end{proof}
For ${{\mathsf{dim}}}\,{\mathcal H}=m<\infty$, one has  $s(\varrho_s)=E(J)$, with the special choice  $J=J(k)$ and   $k={{\mathsf{dim}}}\,s(\varrho_s){\mathcal H}$, and then  \eqref{diffformu} reads
\begin{equation}\label{suppspec0}
    \frac{d}{d\/t}\,\sqrt{\varrho_t}\,s(\varrho_s)\,\biggl|_{t=s}=\int_0^\infty \exp(-t\,\sqrt{\varrho_s})\,\varrho_s^{\,\prime}\,\exp(-t\,\sqrt{\varrho_s})\,s(\varrho_s)\,d\/t
\end{equation}
Thus, in finite dimensional cases the right hand side integral exists.  If $\varrho_t$ evolves in the set density operators of full support, this has an interesting consequence.
\begin{theorem}\label{enddiffex}
Suppose ${{\mathsf{dim}}}\,{\mathcal H}<\infty$. For each  parameterized $C^1$-curve of density operators of full support $\gamma:  I\ni t\longmapsto \varrho_t\in {\mathcal S}_0^{\mathsf{faithful}}({\mathsf B}({\mathcal H}))$ the implementation $$I\ni t\longmapsto \sqrt{\varrho_t}\in {{\mathsf{H.S.}}}({\mathcal H})$$ is differentiable at each $t=s$, with the derivative
\begin{equation}\label{suppspec}
   \sqrt{\varrho_s}^{\,\,\prime}= \frac{d}{d\/t}\,\sqrt{\varrho_t}\,\biggl|_{t=s}=\int_0^\infty \exp(-t\,\sqrt{\varrho_s})\,\varrho_s^{\,\prime}\,\exp(-t\,\sqrt{\varrho_s})\,d\/t
\end{equation}
\end{theorem}
Let us see what happens  if   ${{\mathsf{dim}}}\,{\mathcal H}=\infty$ is fulfilled, under the same suppositions about $\gamma$ as in the previous theorem. Then, for each  increasingly directed sequence $\{J_k\}$ of finite subsets $J_k\subset {{\mathbb{N}}}$ with   $\bigcup_k J_k={{\mathbb{N}}}$ we have $\mathsf{l.u.b}\,\{E(J_k)\}= {\mathsf 1}$.  Thus, for $\Delta_r\sqrt{\varrho_s}$ the following expression is obtained from formulae \eqref{diffq}:
\begin{equation}\label{aufsupp1}
\Delta_r\sqrt{\varrho_s}= \lim_{k\to\infty}\int_0^\infty \exp(-t\,\sqrt{\varrho_r})\,\Delta_r\varrho_s\,\exp(-t\,\sqrt{\varrho_s})\,E(J_k)\,d\/t
\end{equation}
 The limit is existing with respect to the  Hilbert-Schmidt norm. Hence, the special implementation $t\mapsto \sqrt{\varrho_t}$ of $\gamma$ in this case possesses a derivative if, and only if, the following double limit exists (and the value of which then yields the derivative):
 \begin{equation}\label{deri}
\sqrt{\varrho_s}^{\,\,\prime}=\lim_{r\to s} \lim_{k\to\infty}  \int_0^\infty   \exp(-t\,\sqrt{\varrho_r})\,\Delta_r\varrho_s\,\exp(-t\,\sqrt{\varrho_s})\,E(J_k)\,d\/t
 \end{equation}
On the other hand, if the derivative exists, i.e. if the limit  $\sqrt{\varrho_s}^{\,\,\prime}= \lim_{r\to s} \Delta_r \sqrt{\varrho_s}$ exists, then owing to ${\mathsf{l.u.b}}\,\{E(J_k)\}= {\mathsf 1}$ once more again, from \eqref{diffformu} we infer that
$$\sqrt{\varrho_s}^{\,\,\prime}=\lim_{k\to\infty} \sqrt{\varrho_s}^{\,\,\prime}E(J_k)=\lim_{k\to\infty} \frac{d}{d\/t}\,\sqrt{\varrho_t}\,E(J_k)\,\biggl|_{t=s}=\lim_{r\to s} \Delta_r \sqrt{\varrho_s}$$
Hence, in view of \eqref{deri}, \eqref{diffformu} the existence of $\sqrt{\varrho_s}^{\,\,\prime}$ is equivalent to the condition
\begin{equation}\label{ness}
\begin{split}
\lim_{r\to s} \lim_{k\to\infty} & \int_0^\infty   \exp(-t\,\sqrt{\varrho_r})\,\Delta_r\varrho_s\,\exp(-t\,\sqrt{\varrho_s})\,E(J_k)\,d\/t  = \\
& = \lim_{k\to\infty} \int_0^\infty \exp(-t\,\sqrt{\varrho_s})\,\,\varrho_s^{\,\prime}\,\exp(-t\,\sqrt{\varrho_s})\,E(J_k)\,d\/t=\\
&\phantom{\lim_{k\to\infty}} =\lim_{k\to\infty} \lim_{r\to s}  \int_0^\infty  \exp(-t\,\sqrt{\varrho_r})\,\Delta_r\varrho_s\,\exp(-t\,\sqrt{\varrho_s})\,E(J_k)\,d\/t
\end{split}
\end{equation}
Thus, existence of the double limit in \eqref{deri} implies the limiting procedure to be even invariant against reversing the order of the partial limits $\lim_{r\to s}$ and $\lim_{k\to\infty}$, and according to \eqref{ness} the derivative equally well may be calculated by the formula
\begin{equation}\label{deri1}
 \sqrt{\varrho_s}^{\,\,\prime}=\lim_{k\to\infty} \int_0^\infty \exp(-t\,\sqrt{\varrho_s})\,\,\varrho_s^{\,\prime}\,\exp(-t\,\sqrt{\varrho_s})\,E(J_k)\,d\/t
\end{equation}
We now will take the following special case of $\varrho_t$. Suppose there is a strictly ascending sequence $\{P_k\}$ of orthoprojections of finite rank obeying   ${\mathsf{l.u.b.}} \{P_k\}={\mathsf 1}$ and
\begin{subequations}\label{comm0}
\begin{equation}\label{comm}
\forall\,k : \varrho_t P_k=P_k \varrho_t, \,\forall\,t\in I
\end{equation}
When  defining $\Delta P_1=P_1$, and $\Delta P_k=P_k-P_{k-1}$ for $k>1$, then
$\{\Delta P_k\}$ is a decomposition of the unity into mutually orthogonal orthoprojections of finite rank which are obeying
\begin{equation}\label{comm1a}
\forall\,k : \varrho_t\, \Delta P_k=\Delta P_k\, \varrho_t, \,\forall\,t\in I
\end{equation}
\end{subequations}
Clearly, we then may think of $\{e_k\}$ to be chosen in accordance with the decomposition $\{\Delta P_k\}$. That is, there exists a  decomposition $\{S_k\}$ of ${{\mathbb{N}}}$ into mutually disjoint non-void finite subsets $S_k$ such that $\Delta P_k=E(S_k)$, and $$P_n=E(J_n)=\sum_{k\leq n} \Delta P_k$$ for $J_n=\bigcup_{k\leq n} S_k$. Especially, provided that \eqref{comm} is fulfilled,  formula \eqref{deri1} takes the form $\sqrt{\varrho_s}^{\,\prime}=y(s)$, with the Hilbert-Schmidt operator
\[
y(s)=\sum_{k=1}^\infty\int_0^\infty \exp(-t\,\sqrt{\varrho_s})\,\varrho_s^{\,\prime}\,\exp(-t\,\sqrt{\varrho_s})\,\Delta P_k\,d\/t
\]
Since the orthoprojections $\Delta P_k$ are commuting with each single factor of the expression under the integral, we arrive at a necessary condition for the derivative to exist:
\begin{equation}\label{hscond}
\begin{split}
\sum_{k=1}^\infty\bigl\|\frac{d}{d\/t}\,\sqrt{\varrho_t}\,\Delta P_k\,&\biggl|_{t=s}\bigr\|_{\mathsf{H.S}}^2=\\
&\sum_{k=1}^\infty\bigl\|\int_0^\infty\exp(-t\,\sqrt{\varrho_s})\,\varrho_s^{\,\prime}\,\exp(-t\,\sqrt{\varrho_s})\,\Delta P_k\,d\/t\,\bigr\|_{\mathsf{H.S}}^2< \infty
\end{split}
\end{equation}
This can be usefully applied in order to construct counter-examples  showing that the assertion of Theorem \ref{enddiffex}  cannot be extended beyond the context of ${{\mathsf{dim}}}\,{\mathcal H}<\infty$.
\protect{
\begin{example}\label{gegendiff}
Let ${\mathcal H}={\bigoplus_{k=1}^\infty} {\mathbb C}^2$. For $t\in {\mathbb{R}}$ define $\varrho_t$ by
$
\varrho_t=\bigoplus_{k=1}^\infty \varrho_t^{(k)}
$
with $$\varrho_t^{(k)}= \frac{3}{2^{2k+1}}\,\begin{pmatrix}
1+\frac{1}{2}\,\sin{2^k\,t} & 0\\
0 & 1-\frac{1}{2}\,\sin{2^k\,t}
\end{pmatrix}$$
Then, $t\mapsto \varrho_t$ is differentiable, but the ${\mathsf{H.S.}}$- implementation $t\mapsto \sqrt{\varrho_t}$ is not.
\end{example}
}
\begin{proof}
It is easily seen that $t\mapsto \varrho_t\in {\mathcal S}_0^{\mathsf{faithful}}({\mathsf B}({\mathcal H}))$ is fulfilled, with derivative
\[
\varrho_t^{\,\prime}=\bigoplus_{k=1}^\infty \frac{3\,\cos{2^k\,t}}{2^{k+2}}\,\begin{pmatrix}
1 & \phantom{-}0\\
0 & -1
\end{pmatrix}\in {\mathsf{H.S}}({\mathcal H})
\]
Also, if $\Delta P_k$ is the orthoprojection projecting onto the $k^{th}$ component within  ${\mathcal H}={\bigoplus_{k=1}^\infty} {\mathbb C}^2$, then for $t\mapsto \varrho_t$ the  assumptions according to \eqref{comm0} in respect of the complete orthogonal  decomposition  $\{\Delta P_k\}$ of the unity into projections of rank two are fulfilled. Especially, a little calculation is showing that, for each subscript $k$, at $t=0$ we have
\[
\frac{d}{d\/t}\,\sqrt{\varrho_t}\,\Delta P_k\,\biggl|_{t=0}=\frac{1}{4}\,\sqrt{\frac{3}{2}}\,\begin{pmatrix}
1 & \phantom{-}0\\
0 & -1
\end{pmatrix}\
\]
Thus, $\bigl\|\frac{d}{d\/t}\,\sqrt{\varrho_t}\,\Delta P_k\,\biggl|_{t=0}\bigr\|_{\mathsf{H.S}}^2= \frac{3}{16}$ is independent from $k$, and therefore the necessary condition \eqref{hscond} for $\sqrt{\varrho_t}$ to possess a derivative at $t=0$ is violated.
\end{proof}
\begin{remark}\label{fini}
\begin{enumerate}
\item \label{fini1}
If $\gamma$ does not evolve in a set of density operators of constant support exclusively, then even in the finite dimensional case  differentiability of $t\mapsto \sqrt{\varrho_t}$ will fail. As an example consider the differentiable map
\[
\gamma: ]-1,1[\ni t\longmapsto \varrho_t=\begin{pmatrix} t^2 & 0\\
0 & 1-t^2
\end{pmatrix}
\]
into the set of $2\times 2$-density matrices, and which corresponds to the case of ${\mathcal H}={\mathbb C}^2$. Obviously, at $t=0$ the density matrix $\varrho_t$ has rank one instead of rank two at $t\not=0$ else. But then, since
\[
\sqrt{\varrho_t}=\begin{pmatrix} |t| & 0\\
0 & \sqrt{1-t^2}
\end{pmatrix}
\]
holds
$\sqrt{\varrho_t}$ obviously fails to possess a derivative at $t=0$.
\item \label{fini2} Suppose $\varrho$ is an $n\times\, n$-density matrix of full rank. Then, for each hermitian $n\times\, n$-matrix $b$ with ${{{\mathsf{tr}}}}\,b=0$, for some symmetric around $0$ interval $I\subset {\mathbb{R}}$, $$\gamma: I\ni t\mapsto \varrho_t=\varrho+t\,b$$ is a differentiable curve in the  $n\times\, n$-density matrices of full rank, with  $(\varrho_t)^{\,\prime}|_{t=0}=b$,   and then  accordingly the derivative of $\sqrt{\varrho_t}$ at $t=0$ exists and reads as  $$\Bigl(\sqrt{\varrho+t\,b}\Bigr)^{\,\prime}\bigl|_{t=0}=\int_0^\infty \exp(-t\,\sqrt{\varrho})\,b\,\exp(-t\,\sqrt{\varrho})\,d\/t$$
by Theorem \ref{enddiffex} and formula \eqref{suppspec}.
\end{enumerate}
\end{remark}
\section{Solutions of certain operator equations\hfill{}}\label{app_e}
\noindent
In the following, ${\mathcal H}$ be a separable Hilbert space. Then, within the bounded linear operators ${\mathsf B}({\mathcal H})$, let a decomposition $\{\Delta P_k\}$ of the unity into mutually orthogonal orthoprojections $\Delta P_k$ be fixed, with ${{\mathsf{dim}}}\,\Delta P_k{\mathcal H}<\infty$, for all $k\in {{\mathbb{N}}}$. Let $\varrho,  \omega$ be a density operator of full support and a self-adjoint linear operator of trace-class  over ${\mathcal H}$, respectively, and for which
\begin{equation}\label{ope1a}
\varrho\, \Delta P_k=\Delta P_k\varrho,\ \omega\, \Delta P_k=\Delta P_k\omega
\end{equation}
holds, for each $k$. In the following, we are considering two operator equations for which $\omega$ and $\varrho$ satisfying \eqref{ope1a} serve as input data. Whereas the first of these reads
\begin{subequations}\label{ope1}
\begin{equation}\label{ope1b}
\omega= a\sqrt{\varrho}+\sqrt{\varrho}\,a^*
\end{equation}
and is asking for a solution $a\in {\mathsf{H.S.}}({\mathcal H})$ which is satisfying the auxiliary condition \begin{equation}\label{ope1aa}
a\in \bigl[{\mathsf B}({\mathcal H})_{\mathsf h}\sqrt{\varrho}\,\bigr]_{{\mathsf{H.S.}}}
\end{equation}
\end{subequations}
the second type of operator equation  we will have in focus reads
\begin{subequations}\label{twoope}
\begin{equation}\label{twoope1}
\omega= x\varrho+\varrho x
\end{equation}
and is asking for solutions $x$ in the densely defined, self-adjoint linear operators on ${\mathcal H}$. To be a well-defined, the domain of definition ${\mathcal D}(x)$ has to obey the condition
\begin{equation}\label{twoope2}
\varrho\,{\mathcal D}(x)\subset {\mathcal D}(x)
\end{equation}
\end{subequations}
Start with the uniqueness question for solutions of problem \eqref{ope1}.
\begin{lemma}\label{unope1}
Provided a solution of \eqref{ope1} exists, then this is the unique solution.
\end{lemma}
\begin{proof}
Provided solutions exist, by uniqueness $c={\mathsf 0}$ has to be the only solution of
\begin{equation}\label{kerope1}
c\sqrt{\varrho}+\sqrt{\varrho}\, c^*={\mathsf 0},\ c\in \bigl[{\mathsf B}({\mathcal H})_{\mathsf h}\sqrt{\varrho}\,\bigr]_{{\mathsf{H.S.}}}
\end{equation}
This we are going to verify now.
Let $\{\varphi_k\}$ be an eigenbasis of $\varrho$. By faithfulness of $\varrho$ this is a complete orthonormal system of eigenvectors in ${\mathcal H}$. Let $\{\lambda_k\}$ be the corresponding sequence of  eigenvalues (each value repeated in the list as often as according to its multiplicity). Let $E_n$ be the orthoprojection on the linear span of  $\{\varphi_1,\varphi_2,\ldots,\varphi_n\}$.
For $c$ obeying \eqref{kerope1} and $n\in {{\mathbb{N}}}$, consider
$
c_n=E_n c E_n
$
Then, since $E_n$ is of finite rank, we have $$c_n=b_n\sqrt{\varrho}$$ for some $b_n\in E_n{\mathsf B}({\mathcal H})_{\mathsf h} E_n $, and for which according to \eqref{kerope1} the equation
\[
b_n\varrho+\varrho\, b_n={\mathsf 0}
\]
 holds. From this, we get $(\lambda_j+\lambda_k)(b_n)_{jk}=0$, for the $jk$-matrix element $(b_n)_{jk}$ of $b_n$ taken in  respect of the eigenbasis of $\varrho$. Thus, since the eigenvalues are strictly positive $(b_n)_{jk}=0$ follows, for all $j,k\in {{\mathbb{N}}}$. Hence, $b_n={\mathsf 0}$ has to be fulfilled. This conclusion applies for any $n$. Thus, in view of the above, $c_n={\mathsf 0}$ for all $n\in {{\mathbb{N}}}$. From this $c={\mathsf{H.S.}}-\lim_{n\to\infty} c_n={\mathsf 0}$ is seen.
\end{proof}
Prior to going in for the existence of a solution, consider first the existence and uniqueness problem for the second type of operator equation  \eqref{twoope}.
\begin{lemma}\label{second}
There exists a uniquely determined  solution of \eqref{twoope}.
\end{lemma}
\begin{proof}
To find a solution, for each $k\in {{\mathbb{N}}}$ consider the operator equation
\begin{equation}\label{partope}
\omega\,\Delta P_k=x_k \varrho+\varrho\,x_k
\end{equation}
and asking for a self-adjoint solution $x_k\in \Delta P_k {\mathsf B}({\mathcal H}) \Delta P_k$. We remark that since $\varrho$ is positive of full support,  ${\mathsf{dim}} \Delta P_k {\mathcal H}<\infty$ is fulfilled and \eqref{ope1a} holds, the previous equation is uniquely solvable in $\Delta P_k {\mathsf B}({\mathcal H}) \Delta P_k$, and this solution reads
\begin{equation}\label{ope1d}
x_k=\int_0^\infty \exp{(-\varrho\,t)}\,\omega \exp{(-\varrho\,t)}\,\Delta P_k\ d\/t
\end{equation}
We  are going to construct a solution $x$ of  the required specification \eqref{twoope} out of all these $x_k$. To start with, let a linear operator $x_0$ over ${\mathcal H}$ be given by setting that
\[
x_0\psi=\sum_k x_k\psi
\]
for each $\psi$ with $\Delta P_k \psi\not={\mathsf 0}$ for at most finitely many subscripts $k\in {{\mathbb{N}}}$. The set of all vectors of that kind is a   dense linear domain of definition for $x_0$, i.e.~we set
\begin{equation}\label{saope20}
{\mathcal D}(x_0)={\mathsf{Lin}}\Bigl( \bigcup_k \Delta P_k {\mathcal H}\Bigr)
\end{equation}By definition of $x_0$ one easily infers that the adjoint operator $x_0^*$ reads
\[
x_0^*\psi=\sum_k x_k\psi
\]
with domain of definition
$
{\mathcal{D}}(x_0^*)=\{\psi\in {\mathcal H}: \sum_k \| x_k \psi\|^2<\infty\}
$. Hence, $x_0\subset x_0^*$, i.e. $x_0$ is symmetric.  Note that by definition of $x_0$, and since $x_k$ for each $k$ is a bounded symmetric linear operator with support in $\Delta P_k$, the relations $\Delta P_k{\mathcal H}=(x_k\pm {\mathsf{i}}{\mathsf 1}) \Delta P_k{\mathcal H}$ have to hold. In view of the definition of $x_0$ and ${\mathcal D}(x_0)$, from this ${\mathcal D}(x_0)=(x_0\pm {\mathsf{i}}{\mathsf 1}){\mathcal D}(x_0)$ follows. That is, both $(x_0+ {\mathsf{i}}{\mathsf 1}){\mathcal D}(x_0)$ and $(x_0- {\mathsf{i}}{\mathsf 1}){\mathcal D}(x_0)$ are dense subsets of ${\mathcal H}$. Hence, $x_0$ has to be essentially self-adjoint.
Thus $x_0^*=x_0^{**}$ holds, and by the setting
\begin{subequations}\label{saope2}
\begin{equation}\label{saope2a}
x\psi=\sum_k x_k\psi
\end{equation}
for each $\psi$ belonging to the  domain of definition
\begin{equation}\label{saope2b}
{\mathcal{D}}(x)=\Bigl\{\psi\in {\mathcal H}: \sum_k \| x_k \psi\|^2<\infty\Bigr\}
\end{equation}
\end{subequations}
owing to $x=x_0^*$ the unique self-adjoint extension $x$  of $x_0$ is given. Now, with the  eigenbasis  $\{\varphi_k\}$ of $\varrho$, since \eqref{ope1a} is fulfilled with ${{\mathsf{dim}}}\,\Delta P_k{\mathcal H}<\infty$ for all $k$, we equally well have that ${\mathcal D}(x_0)=\mathsf{Lin}(\{\varphi_k\})$. From this especially
\begin{equation}\label{invope}
{\mathcal D}(x_0)=\varrho \,{\mathcal D}(x_0)
\end{equation}
can be followed. Thus, for any $\psi\in {\mathcal D}(x_0)$, according to \eqref{partope} we infer that
$$
x_0\varrho \psi+\varrho\,x_0\psi= \sum_k (x_k\varrho+\varrho\,x_k)\,\psi
=\sum_k \omega\,\Delta P_k\psi=\omega\,\psi
$$
From this and  $x=x_0^*$, for $\psi$ varying in ${\mathcal D}(x_0)$ and each $\varphi\in {\mathcal D}(x)={\mathcal D}(x_0^*)$, we will show that $ {\mathcal D}(x_0)\ni\psi\mapsto \langle x_0\psi,\varrho\varphi\rangle$ is bounded. In fact, we may conclude as follows:
\begin{eqnarray*}
\langle x_0\psi,\varrho\varphi\rangle&=& \langle \varrho\,x_0\psi,\varphi\rangle \\
&=& \langle\omega\,\psi,\varphi\rangle-\langle x_0\varrho\psi,\varphi\rangle\\
& = & \langle\omega\,\psi,\varphi\rangle-\langle \psi,\varrho x\varphi\rangle\\
&=& \langle\psi,(\omega-\varrho\,x)\varphi\rangle
\end{eqnarray*}
Hence, for each $\varphi\in {\mathcal D}(x)$ we have $\varrho\varphi\in {\mathcal D}(x_0^*)={\mathcal D}(x)$, that is,
\begin{equation}\label{incope}
\varrho\,{\mathcal D}(x)\subset {\mathcal D}(x)
\end{equation}
Accordingly,  since then $\langle x_0\psi,\varrho\varphi\rangle=\langle\psi,x\varrho\varphi\rangle$ holds, the conclusion is
$$
\langle\psi,x\varrho\varphi\rangle=\langle\psi,(\omega-\varrho\,x)\varphi\rangle
$$
for all $\psi\in {\mathcal D}(x_0)$ and each $\varphi\in {\mathcal D}(x)$. By denseness of ${\mathcal D}(x_0)$ in ${\mathcal H}$ from this
$\omega\varphi=x\varrho\varphi+\varrho x\varphi$ follows, for each $\varphi\in {\mathcal D}(x)$. Thus, in view of this and since \eqref{incope} holds, $x$ is a self-adjoint solution of the operator equation $\omega=x\varrho+\varrho x$.

We are going to show  uniqueness of this solution. Suppose $y=y^*$ satisfies $\omega=y\varrho+\varrho y$. Then, owing to ${\mathcal D}(x_0)\subset {\mathcal D}(x)\cap {\mathcal D}(y)$, with $z=x-y$ we have
\[
(\lambda_k {\mathsf 1}+\varrho) z \varphi_k=z\varrho \varphi_k+\varrho z \varphi_k={\mathsf 0}
\]
Note that owing to $\varrho\geq {\mathsf 0}$ and since each eigenvalue $\lambda_k$ is strictly positive, $(\lambda_k {\mathsf 1}+\varrho)$ has bounded inverse. Accordingly, $z\varphi_k={\mathsf 0}$ follows, for each $k$. Hence,
$$
y|{\mathcal D}(x_0)=x|{\mathcal D}(x_0)=x_0
$$
has to be fulfilled. That is, $y$ is a self-adjoint extension of $x_0$. By essential self-adjointness of $x_0$ then $y=x$ follows. This proves uniqueness.
\end{proof}
Now, we are ready to state the intimate relationship which exists between  the solvableness of \eqref{ope1} and the structure of the solution of \eqref{twoope}.
\begin{corolla}\label{saope}
Suppose \eqref{ope1} admits the solution $a$.  If $x$ is the solution of  \eqref{twoope}, then the following are satisfied:
\begin{subequations}\label{saope0}
\begin{eqnarray}\label{saope0c}
\sqrt{\varrho}\,{\mathcal H} &\subset & {\mathcal D}(x)\\
\label{saope0a}
a &=& x \sqrt{\varrho}
\end{eqnarray}
\end{subequations}
\end{corolla}
\begin{proof}
Suppose a solution $a$ of \eqref{ope1} to exist. By assumption \eqref{ope1a} from \eqref{ope1b}
\begin{equation}\label{ope1c}
\omega\,\Delta P_k=\Delta P_k a \Delta P_k \sqrt{\varrho}+\sqrt{\varrho}\,(\Delta P_k a \Delta P_k)^*
\end{equation}
follows, for each $k$. Due to $a\in \bigl[{\mathsf B}({\mathcal H})_{\mathsf h}\sqrt{\varrho}\,\bigr]_{{\mathsf{H.S.}}}$, $\{z_n\}\subset {\mathsf B}({\mathcal H})_{\mathsf h}$ exists with  $$a={\mathsf{H.S.}}-\lim_{n\to\infty} z_n\sqrt{\varrho}$$
Let $\sqrt{\varrho}_k^{\,-1}$ be the unique positive linear operator with support $\Delta P_k$ and obeying
$$\Delta P_k=\sqrt{\varrho} \sqrt{\varrho}_k^{\,-1}=\sqrt{\varrho}_k^{\,-1}  \sqrt{\varrho}
$$
Then, from the above obviously the following limiting relation is obtained to hold $$\Delta P_k a \Delta P_k \sqrt{\varrho}_k^{\,-1}={\mathsf{H.S.}}-\lim_{n\to\infty} \Delta P_k z_n \Delta P_k
$$
Hence, if $y_k\in \Delta P_k {\mathsf B}({\mathcal H})_{\mathsf h} \Delta P_k$ is defined by  $y_k=\lim_{n\to\infty} \Delta P_k z_n \Delta P_k$, we infer that
\begin{equation}\label{saope1}
\Delta P_k a \Delta P_k= y_k \sqrt{\varrho}
\end{equation}
is fulfilled, for each $k$. In view of \eqref{ope1c} and \eqref{saope1}, one infers that by each of the bounded self-adjoint linear operators $y_k\in \Delta P_k {\mathsf B}({\mathcal H}) \Delta P_k$ obviously the equation
\[
\omega\,\Delta P_k=y_k \varrho+\varrho y_k
\]
is satisfied. In comparing the latter equation to the one considered on occassion of the proof of Lemma \ref{second}, we conclude that $y_k$ has to equal the unique solution of \eqref{partope}. That is, $y_k=x_k$ has to be fulfilled, with the partial operators $x_k$ defined by formula  \eqref{ope1d} and used in the construction of the unique self-adjoint solution $x$ of \eqref{twoope}.
Now, consider the normal conditional expectation mapping $\Phi$ defined as   $$\Phi(\cdot)=\sum_k \Delta P_k (\cdot) \Delta P_k
$$
and projecting from ${\mathsf B}({\mathcal H})$ onto the unital subalgebra ${\mathfrak{A}}_\Phi=\sum_k^\oplus {\mathsf B}(\Delta P_k{\mathcal H})$ of fixed-points of $\Phi$.
From \eqref{saope1} and due to $y_k=x_k$, we conclude that for each $\psi \in {\mathcal H}$,
\[
\Phi(a)\psi=\lim_{n\to\infty}\sum_{k\leq n} \Delta P_k a \Delta P_k\psi =\lim_{n\to\infty} \sum_{k\leq n} x_k \sqrt{\varrho}\,\psi
\]
holds. Hence, two facts can be stated: firstly, comparison with \eqref{saope2} shows that    $$\sqrt{\varrho}\,\psi\in {{\mathcal{D}}}(x)$$ must hold, for all  $\psi \in {\mathcal H}$, and thus  \eqref{saope0c} is seen. Secondly, we see that
\begin{equation}\label{saope3}
\Phi(a)=x \sqrt{\varrho}
\end{equation}
has to be satisfied.  From this  $\Phi(a) \in \bigl[{\mathsf B}({\mathcal H})_{\mathsf h}\sqrt{\varrho}\,\bigr]_{{\mathsf{H.S.}}}$ can be followed. In fact, for $$P_n=\sum_{k\leq n} \Delta P_k$$
we have ${\mathsf{l.u.b.}}\,\{P_n\} ={\mathsf 1}$. Since a conditional expectation $\Phi$ is a unital, completely positive linear mapping  \cite{storm:74}, according to \cite{choi:72} the Cauchy-Schwarz inequality  $$\Phi(z)\Phi(z)^*\leq \Phi(z z^*)$$ holds, for any operator $z$. Especially, we infer that
\[
{{\mathsf{tr}}}\,\Phi(a)\Phi(a^*)P_n^\perp\leq {{\mathsf{tr}}}\,\Phi(a a^*)P_n^\perp={{\mathsf{tr}}}\,a a^* P_n^\perp
\]
is fulfilled, and since the operator $a a^*$ is a positive trace-class operator, we have
\[
\lim_{n\to\infty} {{\mathsf{tr}}}\,a a^* P_n^\perp=0
\]
Therefore, from ${{\mathsf{tr}}}\,\Phi(a)\Phi(a^*)P_n^\perp=\| \Phi(a)- P_n\Phi(a)\|_{{\mathsf{H.S.}}}^2$ in view of the previous relations and  since $P_n$ is of  finite rank and $\Phi(a)$ is bounded, we have $P_n \Phi(a)\in {\mathsf{H.S.}}({\mathcal H})$ and thus
\[
\Phi(a)={\mathsf{H.S.}}-\lim_{n\to\infty} P_n \Phi(a)
\]
In addition, by \eqref{saope3} and by construction of $x$, according to \eqref{saope2a} we have
\[
P_n  \Phi(a) = \Bigl(\sum_{k\leq n} x_k\Bigr)\sqrt{\varrho} \in {\mathsf B}({\mathcal H})_{\mathsf{h}}\sqrt{\varrho}
\]
In summarizing, from the previous the following fact now can be inferred to hold: $$\Phi(a) \in \bigl[{\mathsf B}({\mathcal H})_{\mathsf h}\sqrt{\varrho}\,\bigr]_{{\mathsf{H.S.}}}$$ On the other hand, by the conditional expectation properties of $\Phi$, and since due to \eqref{ope1a} both  $\sqrt{\varrho}$ and $\omega$ belong to $\mathfrak{A}_\Phi$, application of  $\Phi$ to  \eqref{ope1b} shows that
\[
\omega=\Phi(a)\sqrt{\varrho}+\sqrt{\varrho}\,\Phi(a)^*
\]
is satisfied.
Thus, provided a solution $a$ of \eqref{ope1} exists, then  $\Phi(a)$ in view of the just proven will yield a solution of \eqref{ope1}, too. By the uniqueness assertion of Lemma \ref{unope1} we have $a=\Phi(a)$. Hence, the solution $a$ of \eqref{ope1} is as asserted by \eqref{saope0a}.
\end{proof}
There exists a weak converse of the assertion in the following form.
\begin{lemma}\label{invsaope}
Under the condition that \eqref{ope1a} holds, suppose the condition $$\sqrt{\varrho}\,{\mathcal H}\subset {\mathcal D}(x)$$ to be  fulfilled, with the unique self-adjoint solution $x$  of \eqref{twoope}. Then, $a=x\,\sqrt{\varrho}$ is a bounded solution of the operator equation \eqref{ope1b}.
\end{lemma}
\begin{proof}
Note that  since $\sqrt{\varrho}$ is bounded, by standard facts it follows that $$b=\sqrt{\varrho}\, x$$ is densely defined with ${\mathcal D}(b)={\mathcal D}(x)$, and with adjoint reading as $b^*=x\sqrt{\varrho}$. Thus, as an adjoint and since \eqref{saope0c} has been supposed, $b^*$ is a closed, everywhere defined on ${\mathcal H}$ linear operator.  Hence, by the closed graph theorem, $a=x\sqrt{\varrho}=b^*$ is bounded,  and then $a^*=\sqrt{\varrho}\,x=b$ uniquely extends from ${\mathcal D}(x)$ to a bounded linear operator on all of ${\mathcal H}$. Formula \eqref{twoope1} implies that the following relation is satisfied:
$$
\bigl(a\sqrt{\varrho}+\sqrt{\varrho}\, a^*\bigr)\bigl|{\mathcal D}(x)=\bigl(x\varrho+\varrho x\bigr)\bigl|{\mathcal D}(x)=\omega\bigl|{\mathcal D}(x)
$$
Due to the above-stated boundedness of $a$ and $a^*$, by continuity the conclusion is that the operator equation \eqref{ope1b} is satisfied by the bounded operator $a=x\sqrt{\varrho}$.
\end{proof}
We may summarize all that into the following result.
\begin{theorem}\label{hsope}
If \eqref{ope1a} holds, then \eqref{ope1} admits a solution if, and only if,
\begin{equation}\label{hscond1}
x\sqrt{\varrho}\in {\mathsf{H.S.}}({\mathcal H})
\end{equation}
is fulfilled, with the densely defined, uniquely determined  self-adjoint linear operator $x$ obeying \eqref{twoope}, and in which case $a=x\sqrt{\varrho}$ is the unique solution of \eqref{ope1}.
\end{theorem}
\begin{proof}
By Lemma \ref{unope1} and Corollary \ref{saope}, if $a$ is a solution of \eqref{ope1}, then this solution is uniquely determined, is of Hilbert-Schmidt class and reads  $a=x\sqrt{\varrho}$. Thus especially
condition \eqref{hscond1} is fulfilled. On the other hand, suppose \eqref{hscond1} to be fulfilled. Then,
$a=x\sqrt{\varrho}$ is of Hilbert-Schmidt class. Thus, since $a a^*$ is of trace-class, we infer that
\[
\lim_{n\to\infty} {{\mathsf{tr}}}\,a a^* P_n^\perp=0
\]
is fulfilled, with $P_n=\sum_{k\leq n} \Delta P_k$.
Therefore, from ${{\mathsf{tr}}}\,a a^* P_n^\perp=\| a- P_n a\|_{{\mathsf{H.S.}}}^2$ and owing to $P_n a\in {\mathsf{H.S.}}({\mathcal H})$, for each $n\in {{\mathbb{N}}}$, we get
\[
a={\mathsf{H.S.}}-\lim_{n\to\infty} P_n a
\]
Now, by \eqref{saope3} and by construction of $x$ in accordance with \eqref{saope2a} we have
\[
P_n  a = \Bigl(\sum_{k\leq n} x_k\Bigr)\sqrt{\varrho} \in {\mathsf B}({\mathcal H})_{\mathsf{h}}\sqrt{\varrho}
\]
Hence, we can be assured that $a$ is satisfying the condition
$$a \in \bigl[{\mathsf B}({\mathcal H})_{\mathsf h}\sqrt{\varrho}\,\bigr]_{{\mathsf{H.S.}}}$$
Thus \eqref{ope1aa} holds.
Also, due to \eqref{hscond1},  boundedness of $a$ implies  $\sqrt{\varrho}\, {\mathcal H}\subset {\mathcal D}(x)$. Thus, Lemma \ref{invsaope} applies and is showing  that  $a$ solves the equation \eqref{ope1b}.
Hence, $a$ is a solution of \eqref{ope1} which by Lemma \ref{unope1} is the only one of that specification.
\end{proof}
Note that if \eqref{ope1a} is fulfilled with $\varrho$ and $\omega$, then according to Lemma \ref{second} the operator equation \eqref{twoope} always is solvable, with its solution $x$ being uniquely determined. The following example is showing that - provided ${\mathsf{dim}}\,{\mathcal H}=\infty$ holds - within the set of all  $\varrho$ and $\omega$ which are distinguished by  \eqref{ope1a}  combinations of $\varrho$ and $\omega$ can be chosen in such a way that condition \eqref{hscond1} actually may be violated. In those cases then one can be assured that a solution of \eqref{ope1} cannot exist.

We are going to construct an example where  $a=x\sqrt{\varrho}$ indeed is bounded  but fails to be compact.
To fix a context, let ${\mathcal K}$ be an infinite dimensional separable  Hilbert space and consider ${\mathsf B}({\mathcal H})={\mathsf M}_2({\mathsf B}({\mathcal K}))$
over the separable Hilbert space ${\mathcal H}={\mathcal K}\oplus {\mathcal K}$. This is to be understood in the usual sense of identifying linear operators over a direct sum of two examples of one Hilbert space with the $2\times 2$ operator-valued matrices with values in ${\mathsf B}({\mathcal K})$. Fix a density operator $\mu$ of full support over ${\mathcal K}$, and in terms of the above suppose
\[
\varrho=\frac{1}{2\,{{\mathsf{tr}}}\,\mu^2}\,
\begin{pmatrix}
\mu^2 & {\mathsf 0}  \\
{\mathsf 0} & \mu^2
\end{pmatrix},\ \omega=\frac{1}{{{\mathsf{tr}}}\,\mu^2}\,
\begin{pmatrix}
\mu & {\mathsf 0}  \\
{\mathsf 0} & -\mu
\end{pmatrix}
\]
Then, $\varrho$ is of full support over ${\mathcal H}$, and $\varrho$  and $\omega$ are mutually commuting, self-adjoint compact linear operators admitting a common diagonalization. Especially, as $\Delta P_k$ the one-dimensional orthoprojections projecting onto the eigenspaces generated by the eigenvectors of a common eigenbasis $\{\varphi_k\}$  of $\varrho$ will be chosen, that is $\Delta P_k \varphi_k=\varphi_k$ is supposed to hold. Thereby, $\{\varphi_{2m-1}\}$ and $\{\varphi_{2m}\}$ respectively, are  supposed to be clones of a fixed eigenbasis of $\mu$ and belonging to the first or the second of the ${\mathcal K}$-component exposed in the decomposition of ${\mathcal H}$, respectively. Also, for each $m\in {{\mathbb{N}}}$, both $\varphi_{2m-1}$ and $\varphi_{2m}$ are assumed to belong to the same eigenvalue $\mu_m^2/(2\,{{\mathsf{tr}}}\,\mu^2)$, with $\{\mu_m\}$ being a complete list of the  eigenvalues of $\mu$ (with each value repeated according to its multiplicity in the spectrum of $\mu$). In respect of $\{\Delta P_k\}$ both $\varrho$ and $\omega$ are satisfying the condition \eqref{ope1a} in a trivial way. Relative to these settings the following then is easily seen (the details are left for the reader).
\begin{example}\label{exsaope}
The unique solution $x$ of \eqref{twoope} reads
\[
x=
\begin{pmatrix}
\mu^{-1} & {\mathsf 0}  \\
{\mathsf 0} & \mu^{-1}
\end{pmatrix}, \text{ with }{\mathcal D}(x)=\Bigl\{\psi\in {\mathcal H}: \sum_{m=1}^\infty \frac{1}{\mu_m^2}\,\bigl(|\langle\psi,\varphi_{2m-1}\rangle|^2+|\langle\psi,\varphi_{2m}\rangle|^2\bigr)\Bigr\}
\]
and is satisfying $\sqrt{\varrho}\,{\mathcal H}\subset {\mathcal D}(x)$. Hence, the solution $a=x\sqrt{\varrho}$ of \eqref{ope1b} then reads
\[
a=\frac{1}{\sqrt{2\,{{\mathsf{tr}}}\,\mu^2}}\,
\begin{pmatrix}
{\mathsf 1} & \phantom{-}{\mathsf 0}  \\
{\mathsf 0} & -{\mathsf 1}
\end{pmatrix}
\]
Obviously, this is not a compact operator. Thus,  the operator equation \eqref{ope1b} for the given combination of $\varrho$ and $\omega$ cannot possess a solution satisfying \eqref{ope1aa}.
\end{example}
\newpage
\bibliographystyle{abbrv}

\end{document}